\documentclass[11pt, letterpaper]{article}
\usepackage{fullpage}
\usepackage{amsthm}
\usepackage{amsmath,amssymb,amsfonts,nicefrac}
\usepackage{xspace}
\usepackage{color}
\usepackage{url}
\usepackage{hyperref}
\usepackage{bm}
\usepackage{bbm}
\usepackage{times}

\usepackage{enumitem}

\newtheorem{thm}{Theorem}[section]
\newtheorem{conj}[thm]{Hypothesis}
\newtheorem{observation}[thm]{Observation}

\newtheorem{lemma}[thm]{Lemma}
\newtheorem{corollary}[thm]{Corollary}
\newtheorem{claim}[thm]{Claim}
\newtheorem{statement}[thm]{Statement}

\newtheorem{definition}[thm]{Definition}
\newtheorem{remark}[thm]{Remark}

\newtheorem{fact}[thm]{Fact}

\newcommand\E{\mathop{\mathbb{E}}}
\newcommand\card[1]{\left| {#1} \right|}
\newcommand\sett[2]{\left\{ \left. #1 \;\right\vert #2 \right\}}

\newcommand\set[1]{{\left\{ #1 \right\}}}
\newcommand\Prob[2]{{\Pr_{#1}\left[ {#2} \right]}}

\newcommand\cProb[3]{{\Pr_{#1}\left[ \left. #3 \;\right\vert #2 \right]}}

\newcommand\Expect[2]{{\mathop{\mathbb{E}}_{#1}\left[ {#2} \right]}}

\newcommand\cExpect[3]{{\mathbb{E}_{#1}\left[ \left. #3 \;\right\vert #2 \right]}}
\newcommand\norm[1]{\| #1 \|}

\newcommand\skipi{{\vskip 10pt}}

\newcommand\spn{{\sf Span}}

\newcommand\inner[2]{\langle{#1},{#2}\rangle}
\newcommand\eps{\varepsilon}

\renewcommand\geq{\geqslant}
\renewcommand\leq{\leqslant}

\newcommand{\x}{x}
\newcommand{\y}{y}
\newcommand{\z}{z}

\newcommand\T{\mathrm{T}}

\newcommand{\rom}[1]{\uppercase\expandafter{\romannumeral #1\relax}}
\newcommand{\nenm}{{\sf Non}\text{-}{\sf Embed}\text{-}{\sf Non}\text{-}{\sf Modest}}
\newcommand{\effnon}{{\sf eff}\text{-}{\sf nedeg}}
\newcommand{\Cons}{{\sf Cons}}
\newcommand{\newCons}{{\sf newCons}}

\def\ggg{\gtrsim}
\def\lll{\lesssim}

\title{On Approximability of Satisfiable k-CSPs: \rom{4}}
\author{Amey Bhangale\thanks{Department of Computer Science and Engineering, University of California, Riverside.}
	\and
	Subhash Khot\thanks{Department of Computer Science, Courant Institute of Mathematical Sciences, New York University. Supported by
		the NSF Award CCF-1422159, NSF CCF award 2130816, and the Simons Investigator Award.}
	\and
	Dor Minzer\thanks{Department of Mathematics, Massachusetts Institute of Technology. Supported by a Sloan Research
   Fellowship, NSF CCF award 2227876 and NSF CAREER award 2239160.}}
\date{\vspace{-5ex}}
\begin{document}
\maketitle
\begin{abstract}
 We prove a stability result for general $3$-wise correlations over distributions satisfying mild connectivity properties. More concretely, we
 show that if $\Sigma,\Gamma$ and $\Phi$ are alphabets of constant size, and $\mu$ is a distribution over $\Sigma\times\Gamma\times\Phi$
 satisfying:
 (1) the probability of each atom is at least $\Omega(1)$,
 (2) $\mu$ is pairwise connected,
 and
 (3) $\mu$ has no Abelian embeddings into $(\mathbb{Z},+)$, then the following holds.
 Any triplets of $1$-bounded functions $f\colon \Sigma^n\to\mathbb{C}$, $g\colon \Gamma^n\to\mathbb{C}$, $h\colon \Phi^n\to\mathbb{C}$
 satisfying
 \[
 \card{\Expect{(x,y,z)\sim \mu^{\otimes n}}{f(x)g(y)h(z)}}\geq \eps
 \]
 must arise from an Abelian group associated with the distribution $\mu$. More specifically, we show that there is an Abelian group $(H,+)$
 of constant size such that for any such $f,g$ and $h$, the function $f$ (and similarly $g$ and $h$) is correlated with a function of the form $\tilde{f}(x) = \chi(\sigma(x_1),\ldots,\sigma(x_n)) L (x)$,
 where $\sigma\colon \Sigma \to H$ is some map, $\chi\in \hat{H}^{\otimes n}$ is a character, and $L\colon \Sigma^n\to\mathbb{C}$ is a low-degree function with bounded $2$-norm.

 En route we prove a few additional results that may be of independent interest, such as an improved direct product theorem, as well as a result
 we refer to as a ``restriction inverse theorem'' about the structure of functions that, under random restrictions, with noticeable probability
 have significant correlation with a product function.

 In companion papers, we show applications of our results to the fields of Probabilistically Checkable Proofs,
 as well as various areas in discrete mathematics such as extremal combinatorics and additive combinatorics.
\end{abstract}
\section{Introduction}
\subsection{Studying $3$-wise Correlations with Respect to a Distribution}
Let $\Sigma$, $\Gamma$ and $\Phi$ be alphabets of constant size,
suppose $\mu$ is a distribution over $\Sigma\times\Gamma\times \Phi$, and let
$f\colon\Sigma^n\to\mathbb{C}$, $g\colon\Gamma^n\to\mathbb{C}$, $h\colon\Phi^n\to\mathbb{C}$ be $1$-bounded functions.
What sort of triplets of functions $f,g$ and $h$ have a significant $3$-wise correlation
with respect to $\mu$? In other words, what can be said about the functions $f$, $g$ and $h$ in the case that
\begin{equation}\label{eq:1}
\card{\Expect{(x,y,z)\sim \mu^{\otimes n}}{f(x)g(y)h(z)}}\geq \eps,
\end{equation}
where $\eps>0$ is thought of as a small constant?
In~\cite{Mossel}, it is shown that if $\mu$ is connected, then this can only be the case if each one of $f$, $g$ and $h$ is correlated
with a low-degree function. Here, we say that a distribution $\mu$ over $\Sigma_1\times\Sigma_2\times \Sigma_3$ is connected if
for any partition of $\{1,2,3\}$ into two sets $I\cup J$, the bi-partite graph between ${\sf supp}(\mu_I)$ and
${\sf supp}(\mu_J)$ whose edges are all $(a,b)$ if $(a,b)\in {\sf supp}(\mu)$, is connected ($\mu_I$ is the marginal distribution
of $\mu$ on the coordinates of $I$). In~\cite{BKMcsp1,BKMcsp2}, a
strengthening of this result is proved, and it is shown that it suffices that the distribution $\mu$ does not admit any \emph{Abelian embeddings}.
\begin{definition}
  An Abelian embedding of a distribution $\mu$ over $\Sigma\times\Gamma\times \Phi$ consists of an Abelian group $(H,+)$
  and $3$ maps $\sigma\colon \Sigma\to H$, $\gamma\colon \Gamma\to H$ and $\phi\colon \Phi\to H$ such that $\sigma(x) + \gamma(y) + \phi(z) = 0$
  for all $(x,y,z)\in {\sf supp}(\mu)$. We say that the embedding $(\sigma,\gamma,\phi)$ is non-trivial if at least one of the maps
  is not constant.
\end{definition}
\begin{definition}
  We say a distribution $\mu$ admits an Abelian embedding if it has a non-trivial Abelian embedding.
\end{definition}
In this language, the main result of~\cite{BKMcsp1,BKMcsp2} asserts that if $\mu$ does not admit an Abelian embedding and the probability
of each atom in $\mu$ is at least $\alpha$ thought of as a constant, then each one of $f$, $g$ and $h$ must be correlated with a low-degree
function. As it can easily be seen, this result is strictly stronger than the corresponding result in~\cite{Mossel} since any distribution
$\mu$ which is connected does not admit an Abelian embedding. Moreover, as explained in~\cite{BKMcsp1,BKMcsp2} this result
is an if and only if, in the sense that in the presence of Abelian embedding one could design $1$-bounded functions $f$, $g$ and $h$
for which~\eqref{eq:1} holds while at least one of the functions $f,g$ and $h$ only has $o(1)$-correlation with any low-degree function.

The main goal of this paper is to extend this understanding beyond the realm of distributions which do not have Abelian embeddings
and prove structural results on functions $f,g$ and $h$ satisfying~\eqref{eq:1} in this more general setting. At a high level, one
would like to say that such functions $f$, $g$ and $h$ could only arise as a result of using Abelian embeddings, using low-degree
functions, or both. To prove such result however, we must focus our attention on distributions $\mu$ satisfying a very mild form
of connectivity, which we refer to as pairwise connectedness.
\begin{definition}
  Let $\Sigma_1,\Sigma_2,\Sigma_3$ be finite alphabets, and let $P\subseteq \Sigma_1\times \Sigma_2\times \Sigma_3$.
  For a pair of distinct coordinates $i,j\in \{1,2,3\}$, we say $P$ is $\{i,j\}$ connected if the bipartite graph
  $G = (\Sigma_i\cup \Sigma_j, E_{i,j})$, where $E_{i,j}$ is the set of label pairs that appear in some element of $P$,
  is connected.

  We say $P$ is pairwise connected if it is pairwise connected for any two distinct $i,j\in\{1,2,3\}$.

  We say a distribution $\mu$ is pairwise connected if ${\sf supp}(\mu)$ is pairwise connected.
\end{definition}
At a high level, the notion of pairwise connectedness stems from the fact that if ${\sf supp}(\mu)$ is not pairwise connected, then
there are examples of functions satisfying~\eqref{eq:1} without any useful structure for our purposes. Indeed, if ${\sf supp}(\mu)$
is not pairwise connected -- without loss of generality it is not $\{1,2\}$-connected, then we may find a non-trivial partition
$\Sigma = \Sigma'\cup\Sigma''$ and $\Gamma = \Gamma' \cup \Gamma''$ so that in the support of $\mu$ there can only be pairs
from $\Sigma'\times \Gamma'$ and $\Sigma''\times \Gamma''$ on the first two coordinates. In this case, we may pick
any pair of functions $s,s'\colon \{1,2\}^n\to\mathbb{C}$ such that $s(a)s'(a) = 1$ for all $a\in \{1,2\}^n$ (for example, one can
take $s$ whose absolute value is always $1$, and $s'$ to be its conjugate) and construct $f,g,h$ as follows.
For $f$, we set $f(x) = s(x')$ where for each $i$, $x'_i = 1$ if $x_i\in \Sigma'$ and otherwise $x'_i = 2$.
For $g$, we similarly set $g(y) = s'(y')$ where for each $i$, $y'_i = 1$ if $y_i\in \Gamma'$ and otherwise $y'_i = 2$.
For $h$, we take $h\equiv 1$. Thus, for any $(x,y,z)\in {\sf supp}(\mu)$ we have that
\[
f(x)g(y)h(z) = s(x')s'(y') = 1,
\]
as we have that $x'=y'$ by construction.

Henceforth, we will focus our attention on distributions $\mu$ which are pairwise connected. With this in mind, as explained earlier
there are two ways of constructing functions $f$, $g$ and $h$ satisfying~\eqref{eq:1}:
\begin{enumerate}
  \item If ${\sf supp}(\mu)$ admits a linear embedding, say for simplicity a cyclic group $(H,+) = (\mathbb{Z}_p,+)$
and maps $\sigma\colon\Sigma\to H$, $\gamma\colon\Gamma\to H$ and $\phi\colon \Phi\to H$ not all constant such that
$\sigma(x)+\gamma(y)+\phi(z) = 0$, then one can take
\begin{align*}
&f(x_1,\ldots,x_n) = e^{\frac{2\pi i}{\card{H}}\left(\sigma(x_1)+\ldots+\sigma(x_n)\right)},
\qquad
g(y_1,\ldots,y_n) = e^{\frac{2\pi i}{\card{H}}\left(\gamma(y_1)+\ldots+\gamma(y_n)\right)},\\
&\qquad\qquad\qquad\qquad\qquad h(z_1,\ldots,z_n) = e^{\frac{2\pi i}{\card{H}}\left(\phi(z_1)+\ldots+\phi(z_n)\right)},
\end{align*}
and note that $f(x)g(y)h(z) = 1$ pointwise, hence~\eqref{eq:1} holds. More generally, for a general Abelian group $(H,+)$
one can pick non-trivial characters $\chi_1,\ldots,\chi_n\in \hat{H}$, define
\begin{align*}
f(x_1,\ldots,x_n) = \prod\limits_{j=1}^{n}\chi_j(\sigma(x_j)),
~~
g(y_1,\ldots,y_n) = \prod\limits_{j=1}^{n}\chi_j(\gamma(y_j)),
~~
h(z_1,\ldots,z_n) = \prod\limits_{j=1}^{n}\chi_j(\phi(z_j)),
\end{align*}
and note again that $f(x)g(y)h(z) = 1$ pointwise hence~\eqref{eq:1} holds.
  \item In general, it may also be the case that for a distribution $\mu$, low-degree functions also satisfy~\eqref{eq:1}.
  Indeed, in that case one may try to find univariate $1$-bounded functions $u\colon \Sigma\to\mathbb{C}$, $v\colon \Gamma\to\mathbb{C}$
  and $w\colon \Phi\to\mathbb{C}$ for which $\card{\Expect{(x,y,z)\sim \mu}{u(x)v(y)w(z)}}\geq \Omega(1)$, and then tensorize them to get
  \begin{align*}
f(x_1,\ldots,x_n) = \prod\limits_{j=1}^{d}u(x_j),
~~
g(y_1,\ldots,y_n) = \prod\limits_{j=1}^{d}v(y_j),
~~
h(z_1,\ldots,z_n) = \prod\limits_{j=1}^{d}w(z_j),
\end{align*}
which get value of $2^{-O(d)}$ in~\eqref{eq:1}.
\end{enumerate}

\subsection{Main Results}
With the above discussion in mind, one is tempted to conjecture that if $\mu$ is pairwise connected, then the only possible examples of triplets
of functions $f,g$ and $h$ satisfying~\eqref{eq:1} must come from the above template.
\subsubsection{The Stability Result}
The main result of this paper is a stability result that
formalizes this intuition, saying that under some mild assumptions on the distribution $\mu$,
if $f,g$ and $h$ are $1$-bounded functions achieving significant $3$-wise correlation as in~\eqref{eq:1},
then $f$ (and similarly $g$ and $h$) must be correlated with a product of an embedding type function as in the first recipe, with a low-degree
function as in the second recipe. The mild assumptions on $\mu$ correspond to it being pairwise connected (which is necessary, otherwise the
statement is simply false), and for technical reasons we also need an additional assumption, namely that $\mu$ cannot be embedded in the
Abelian group $(\mathbb{Z},+)$. We remark though that this additional assumption is, as far as we know, not necessary, but removing it seems
to require more ideas. With this in mind, a precise formulation of our main result is:
\begin{thm}\label{thm:main_stab_3}
  For all $m\in\mathbb{N}$, $\alpha>0$ and $\eps>0$, there are $d\in\mathbb{N}$ and $\delta>0$ such that the following holds.
  Suppose that $\mu$ is a distribution over $\Sigma\times\Gamma\times \Phi$ such that
  \begin{enumerate}
    \item The probability of each atom in $\mu$ is at least $\alpha$.
    \item The size of each one of $\Sigma,\Gamma,\Phi$ is at most $m$.
    \item The distribution $\mu$ is pairwise connected.
    \item $\mu$ does not admit an Abelian embedding into $(\mathbb{Z},+)$.
  \end{enumerate}
  Then, if $f\colon\Sigma^n\to \mathbb{C}$, $g\colon\Gamma^n\to \mathbb{C}$ and $h\colon\Phi^n\to \mathbb{C}$
  are $1$-bounded functions such that
  \[
  \card{\Expect{(x,y,z)\sim \mu^{\otimes n}}{f(x)g(y)h(z)}}\geq \eps,
  \]
  then there are $1$-bounded functions $u_1,\ldots,u_n\colon \Sigma\to \mathbb{C}$ and a function $L\colon \Sigma^n\to \mathbb{C}$ of degree at most $d$
  and $2$-norm at most $1$ such that
  \[
  \card{\Expect{x\sim \mu_x^{\otimes n}}{f(x)\cdot L(x)\prod\limits_{i=1}^{n}u_i(x_i)}}\geq \delta.
  \]

  Furthermore, there is $r\in\mathbb{N}$ depending only on $m$ and an Abelian embedding $(\sigma,\gamma,\phi)$
  of $\mu$ into an Abelian group $(H,+)$ of size at most $r$
  such that for all $i$, $u_i(x_i) = \chi_i(\sigma(x_i))$ where $\chi_i\in \widehat{H}$ is a character of $H$.

  Quantitatively, we have that
  \[
  d = {\sf poly}_{\alpha,m}(1/\eps),
  \qquad\qquad\qquad
  \delta = 2^{-{\sf poly}_{\alpha,m}(1/\eps)}.
  \]
\end{thm}
The proof of Theorem~\ref{thm:main_stab_3} is quite long, and in Section~\ref{sec:pf_overview} we give an overview of the steps we take
in the proof. Some of the steps require ingredients that may be of independent interest, and which we explain next.

\subsubsection{The Restriction Inverse Theorem}
The proof of Theorem~\ref{thm:main_stab_3} uses a result which we refer to as the Restriction Inverse Theorem and present next.
\paragraph{Restrictions and Random Restrictions.}
Restrictions and random restrictions are vital to our argument to go through, and the Restriction Inverse Theorem can be thought of
as a statement about them of independent interest. Given a function $f\colon (\Sigma^n,\mu^{\otimes n})\to\mathbb{C}$, a set of coordinates $I\subseteq [n]$
and $\tilde{x}\in \Sigma^{\overline{I}}$, the restricted function $f_{\overline{I}\rightarrow \tilde{x}}$ is a function from
$\Sigma^{I}$ to $\mathbb{C}$ defined as
\[
f_{\overline{I}\rightarrow \tilde{x}}(x') = f(x_I = x', x_{\overline{I}} = \tilde{x}),
\]
where $(x_I = x', x_{\overline{I}} = \tilde{x})$ is the point in $\Sigma^n$ whose $I$-coordinates are set according to $x'$, and whose $\overline{I}$-coordinates
are set according to $\tilde{x}$.

Random restrictions are restrictions in which either $I$, $\tilde{x}$ or both are chosen randomly. A typical setting we use is one where we have a parameter
$\rho>0$, and we pick $I\subseteq_{\rho} [n]$, by which we mean that we include each $i\in [n]$ in $I$ with probability $\rho$; we then choose $\tilde{x}\sim \mu^{\overline{I}}$. For
the purposes of this paper it is necessary to consider other (less standard) settings of random restrictions, but we will limit ourselves to this more typical
setting for the purposes of this introduction; we refer the reader to Section~\ref{sec:random_restrictions} for a discussion on the other settings we use.

\paragraph{Product functions.}
A function
$f\colon \Sigma^n\to\mathbb{C}$ is called a product function if there are $1$-bounded functions $f_1,\ldots,f_n\colon \Sigma\to\mathbb{C}$ such
that
\[
f(x_1\ldots,x_n) = \prod\limits_{i=1}^{n}f_i(x_i).
\]
It is clear that if $f$ is a product function, then any restriction of it is still a product function. Thus, with probability $1$, taking a random restriction of
$f$ yields a function which has perfect correlation with a product function. The Restriction Inverse Theorem is a statement about a deduction in the reverse direction:
suppose $f$ is a function that after random restriction it has a significant correlation with a product function. Is it necessarily the case that $f$ itself is correlated
with a product function?

As is usually the case with inverse-type questions, there are multiple regimes of parameters one may consider, and for us the most relevant regime is the so-called $1\%$ regime.
In this case, we have a parameter $\rho>0$ (which is small but bounded away from $0$) and a function $f\colon \Sigma^n\to\mathbb{C}$ such that
\begin{equation}\label{eq:intro_rest_inverse}
\Prob{I\subseteq_{\rho} [n], \tilde{x}\sim \mu^{\overline {I}}}{\exists p\colon\Sigma^{I}\to\mathbb{C}\text{ a product function such that }
\card{\inner{f_{\overline{I}\rightarrow \tilde{x}}}{p}}\geq \eps}
\geq \eps,
\end{equation}
and we wish to deduce a structural result about $f$. As discussed, such situations may arise when $f$ is a product function -- or more generally when it is correlated
with a product function. However, if $f$ is a low-degree function (or even if it is just correlated with a low-degree function), a random restriction of $f$ will be correlated
with a constant function with noticeable probability, and hence with a product function. The Restriction Inverse Theorem essentially says that these are the only
two ways that~\eqref{eq:intro_rest_inverse} can come about:
\begin{thm}[The Restriction Inverse Theorem, Informal]\label{thm:rest_inverse_intro}
  For all $\eps,\rho,\alpha>0$ and $m\in\mathbb{N}$ there are $d\in\mathbb{N}$ and $\delta>0$ such that the following holds.
  Suppose $\Sigma$ is a finite alphabet of size at most $m$, $\mu$ is a distribution over $\Sigma$ in which the probability of each atom is at least $\alpha$,
  and $f\colon(\Sigma^n,\mu^{\otimes n})\to\mathbb{C}$ is a $1$-bounded function satisfying~\eqref{eq:intro_rest_inverse}. Then there is a product function
  $p\colon\Sigma^n\to\mathbb{C}$ and a function $L\colon \Sigma^n\to\mathbb{C}$ of degree at most $d$ and $\norm{L}_2\leq 1$ such that
  \[
  \card{\inner{f}{pL}}\geq \delta.
  \]
\end{thm}
We refer the reader to Section~\ref{sec:rest_inverse} for a more formal and general version of the Restriction Inverse Theorem. We remark that among other things,
we also give explicit dependency of $d$ and $\delta$ on $\eps$ and $\rho$. These quantitative aspects are important if one wishes to get decent quantitative bounds
in Theorem~\ref{thm:main_stab_3}, and we think they are also interesting in their own right.

\subsubsection{The Direct Product Theorem}
The proof of Theorem~\ref{thm:rest_inverse_intro} (and thus, in turn, of Theorem~\ref{thm:main_stab_3}) hinges on a direct product testing result, which may also be
of independent interest. The problem of direct product testing has its roots in the field of probabilistic checkable proofs and in particular in hardness amplification.
In this setting, one wishes to encode a function $f\colon [n]\to [R]$ (where $n$ is thought of as very large) by local pieces that, on the one hand allows for local
access to values of $f$. On the other hand, the encoding should be testable, in the sense that there is a test that only looks at a handful of locations of the encoding and
determines whether it is an encoding of an actual function $f\colon [n]\to [R]$, or whether it is far from the encoding of any such function.

Our application calls for a particular direct product tester that has been extensively studied in the
literature~\cite{DinurR06,DinurG08,ImpagliazzoKW12,IJKW,DinurS14,DinurFH19,BKMcsp3}.
In this tester, the function $f$ is encoded via its table of restrictions to sub-cubes of certain dimension.
Namely, given a parameter $k\in\mathbb{N}$ (which for us will be equal to $\rho n$, where $\rho$ should be thought of as a very small constant),
the direct product encoding of $f$ is the mapping $F\colon \binom{[n]}{k}\to [R]^k$ defined by
\[
F[A] = f|_{A}
\]
for all $A\subseteq [n]$ of size $k$.

The test we associate with this encoding is determined by two parameters, $\alpha,\beta\in (0,1)$ that also should be
thought of as small constants. Given a supposed table of restrictions $G\colon \binom{[n]}{k}\to [R]^k$, the test, which we call
${\sf DP}(\rho,\alpha,\beta)$, proceeds in the following way:
\begin{enumerate}
  \item Sample $C\subseteq [n]$ of size $\alpha k$ and sample $A,B\in\binom{[n]}{k}$ independently containing $C$.
  \item Sample $T\subseteq [n]$ of size $\beta n$.
  \item Query $G[A]$, $G[B]$ and check that $G[A]|_{C\cap T} = G[B]|_{C\cap T}$.
\end{enumerate}
In other words, the tester selected two sets $A,B$ that intersect on a sizable number of elements (at least $\alpha k$), then  a random subset of their shared
elements and checks that the local assignments $G[A]$ and $G[B]$ agree on this random subset of shared elements.

Note that this test is complete, in the sense that if $G$ is a legitimate direct product encoding, then it passes the test with probability $1$. Thus, as is usually the case,
the interesting aspect of this test is the soundness, which is equivalent to the following question. Suppose that the tester accepts a table $G\colon \binom{[n]}{k}\to [R]^k$
with probability at least $s$; is it necessarily the case that $G$ is somewhat close to a legitimate direct product testing codewords?

In the so-called $99\%$ regime, where the probability $s = 1-\eps$ is thought of close to $1$, this problem is completely understood, and in~\cite{DinurS14,DinurFH19} it is shown
that in this case there is a function $f\colon [n]\to[R]$ such that for at least $1-O(\eps)$ fraction of $A\in \binom{[n]}{k}$ it holds that $G[A] = f|_{A}$.

For us, the most so-called $1\%$ regime is more relevant, wherein the probability $s = \eps$ is thought of as close to $0$. In this case, one can no longer expect a strong conclusion
as in the $99\%$ regime. Instead, naturally one would expect that in this case, there would have to be a function $f\colon [n]\to[R]$ such that for at least
$\delta = \delta(\rho,\alpha,\beta,\eps)>0$ fraction of $A\in \binom{[n]}{k}$ it holds that $G[A] = f|_{A}$, but this is also too much to expect. Indeed, to see
that take any $g\colon [n]\to [R]$, and for each $A$ take $G[A]$ uniformly from $[R]^k$ with probability $1-\eps$, and otherwise take it to be a string in
$[R]^k$ of Hamming distance $r = \Theta(\log(1/\eps))$ from $g|_A$. Using Chernoff's bound, one can prove that with high probability there is no function $f\colon[n]\to[R]$
satisfying the natural conclusion one expects, yet the tester passes with probability at least $\eps^2 (1-\beta)^{2r} = {\sf poly}(\eps)$. The reason for that is
that looking at two locations $A, B$ queried by the tester, with probability $\eps^2$ both of them get assigned strings close to $g|_A$ and $g|_B$ respectively, in which
case with probability at least $(1-\beta)^{2r}$ the subset $T$ excludes all coordinates on which either $G[A]$ and $g|_{A}$, or $G[B]$ and $g|_{B}$, disagree on.

Due to a rather versatile set of examples, results in the $1\%$ regime are often more challenging to prove. Indeed, earlier results by~\cite{DinurG08,DinurS14} managed
to show that in this case there is a function $f\colon [n]\to [R]$ such that for at least $\delta = \delta(\rho,\alpha,\beta,\eps,\eta)>0$ fraction
of $A\in \binom{[n]}{k}$ it holds that $\Delta(G[A], f|_{A})\leq \eta k$. Here and throughout, $\Delta(x,y)$ represents the Hamming distance between strings $x$ and
$y$. The main drawback of this result is that the distance between $G[A]$ and $f|_{A}$ is linear in $k$, which is not good enough for our purposes. Indeed,
for our application we need a result that gets a Hamming distance which is a constant $r = r(\rho,\alpha,\beta,\eps)\in\mathbb{N}$ as opposed to a constant fraction.

In~\cite{BKMcsp3}, such result was proved for a more specialized version of this test in the case of $\beta = 1$ and $R = 2$. Therein, both of the parameters $\alpha$ and
$\rho$ are thought of as constant, and it is proved that there are $r = r(\eps,\alpha,\rho)\in\mathbb{N}$ and $\delta = \delta(\eps,\alpha,\rho)>0$ such that if
$G$ passes the test ${\sf DP}(\rho,\alpha,\beta=1)$ with probability at least $\eps$, then there is a function $f\colon [n]\to \{0,1\}$ such that for at least
$\delta$ fraction of $A\in \binom{[n]}{k}$ it holds that $\Delta(G[A], f|_{A})\leq r$. Besides being a natural question of interest, the motivation of this result
therein was to establish an earlier, less general version of the Restriction Inverse Theorem, Theorem~\ref{thm:rest_inverse_intro} herein.

In this paper, we are once again in a situation that our proof of a restriction inverse theorem requires a direct product testing result, and the relevant test
for us is the test ${\sf DP}(\rho,\alpha,\beta)$ above. Moreover, as herein we are concerned with getting good quantitative bounds, we no longer think of the
parameters $\rho,\alpha,\beta$ as constants and thus try to get reasonable dependencies of $r$ and $\delta$ on these parameters. For the purposes of this introductory
section however, we do not mention these quantitative aspects and defer the interested reader to Section~\ref{sec:dp}. Thus, without a concern for these quantitative
aspects our result reads:
\begin{thm}[The Direct Product Testing Theorem, Informal]\label{thm:dp_intro}
  For all $\eps,\rho,\alpha,\beta>0$ there are $r\in\mathbb{N}$ and $\delta>0$ such that the following holds for all $R\in\mathbb{N}$.
  For $k = \rho n$, if $G\colon \binom{[n]}{k}\to [R]^k$ is a function that passes the test ${\sf DP}(\rho,\alpha,\beta)$ with probability
  at least $\eps$, then there is a function $f\colon [n]\to [R]$ such that
  \[
  \Prob{A\in\binom{[n]}{k}}{\Delta(f|_A, G[A])\leq r}\geq \delta.
  \]
\end{thm}

\subsection{Applications and Motivations}
In this section, we discuss some applications and motivating fields and type of problems Theorem~\ref{thm:main_stab_3} (and possible extensions of it)
are likely to be related to. For some of them, we already have initial leads (and pursue them in subsequent papers as the current paper is already long
enough as is), while for others the connection is more speculative.

\subsubsection{Hardness of Approximation}
Recall Mossel's result~\cite{Mossel}, asserting that in the case that $\mu$ is a connected distribution only the low-degree part of functions contributes to~\eqref{eq:1}.
For low-degree functions, one has the invariance principle of~\cite{MOO}, and thus the combination of these two results can be seen as transforming expectations as in~\eqref{eq:1}
to expectations over Gaussian space. This result has a few notable striking consequences in the field of hardness of approximation. Most notably, Raghavendra~\cite{Rag08} uses
precisely such ideas to show the relationship between dictatorship tests and Gaussian rounding scheme to semi-definite relaxations.

In this light, the result proved in this paper shows that only functions that are ``characters times low-degree functions'' can contribute to~\eqref{eq:1}, and this suggests
that an invariance principle that extends the invariance principle of~\cite{MOO} should exist. Indeed, in a future work~\cite{BKMinv} we are exploring this direction and will
prove a more general such invariance principle, and discuss its relation to rounding schemes that combine semi-definite programming relaxations as well as linear programming
relaxations. We believe such invariance principles will be crucial in the journey of understanding the approximability of satisfiable constraint satisfaction problems.

\subsubsection{Higher Arity Predicates}
The original motivation behind the question considered in this paper is the non-Abelian embedding hypothesis of~\cite{BKMcsp1}, which is the following statement.
Suppose $k\geq 3$ is an integer, $\Sigma_1,\ldots,\Sigma_k$ finite alphabets and $\mu$ is a distribution over $\Sigma_1\times\ldots\times\Sigma_k$ in which
the probability of each atom is at least $\alpha > 0$. We say $\mu$ admits an Abelian embedding if there is an Abelian group $(H,+)$ and maps
$\sigma_i\colon\Sigma_i\to H$ for $i=1,\ldots,k$ such that $\sum\limits_{i=1}^{k}\sigma_i(x_i) = 0$ for all $(x_1,\ldots,x_k)\in {\sf supp}(\mu)$.
We say $\mu$ admits a non-trivial Abelian embedding if at least one of the maps $\sigma_i$ is non-constant.
\begin{conj}\label{hyp:bkm1}
  In the above setting, if $\mu$ admits no non-trivial Abelian embeddings, then for all $\eps>0$ there is $\delta>0$ such that
  if $f_i\colon \Sigma_i^n\to\mathbb{C}$ are $1$-bounded functions with ${\sf Stab}_{1/2}(f_i;\mu_i^{\otimes n})\leq \delta$ for
  at least one of the $i$'s, then
  \[
  \card{\Expect{(x_1,\ldots,x_k)\sim\mu^{\otimes n}}{\prod\limits_{i=1}^{k}f_i(x_i)}}\leq \delta.
  \]
\end{conj}
In~\cite{BKMcsp1} a special case of this hypothesis is proved for a class of $k=3$-ary distributions, and in~\cite{BKMcsp2} this hypothesis is proved in general for all
$k=3$-ary distributions. In these terms, the current paper does not signify any further progress towards establishing Hypothesis~\ref{hyp:bkm1} beyond the case of $3$-ary
predicates, however we believe that the stability version proved herein will be crucial towards making further progress in this direction.

\subsubsection{Gowers' Norms}
Theorem~\ref{thm:main_stab_3} can be seen as an analog of the $U_2$-inverse theorem for Gowers uniformity norms~\cite{Gowers} for general distributions.
In the context of Gowers uniformity norms, the $U_2$-inverse theorem is a simple Fourier analytic computation only involving Fourier coefficients.
Interestingly, at a point in our argument we too have to carry out such a computation (this is, however, a small part of the proof). It is tempting
to speculate, and we have initial leads for this fact, that there should be higher order analogs of Gowers inverse theorems in the much more general
setting of Theorem~\ref{thm:main_stab_3}.

If true, such statements could be very useful to make progress on multiple problems in extremal combinatorics,
and in particular in Szemer{\'e}di-type theorems~\cite{szemeredi1975sets}. This is so because it appears they are strong enough to facilitate density increment arguments.
Indeed, as we explain next, in a companion paper we have used Theorem~\ref{thm:main_stab_3} to give effective
bounds for the problem of finding restricted $3$-arithmetic progressions in dense sets in $\mathbb{F}_p^n$, for a prime $p$.

\subsubsection{Extremal Combinatorics}
A set $A\subseteq \mathbb{F}_p^n$ is called somewhat restricted $3$-AP free if it does not contain an arithmetic progression
$x,x+a,x+2a$ where $x\in\mathbb{F}_p^n$ and $a\in\{0,1,2\}^n\setminus\{\vec{0}\}$. In a companion paper~\cite{BKLMroth},
we use Theorem~\ref{thm:main_stab_3} to give effective bounds on the density of restricted $3$-AP sets:
\begin{thm}\label{thm:3AP}
  There are absolute constants $C>0$ and $1\leq k\leq 10$ such that if $A\subseteq\mathbb{F}_p^n$ is a restricted $3$-AP set, then
  \[
  \mu(A) = \frac{\card{A}}{p^n}\leq \frac{C}{\log^{(k)} n},
  \]
  where $\log^{(k)} n$ is the $k$-fold iterated logarithm function.
\end{thm}
Previously, the best known bound was $O(1/\log^{*} n)$, achieved by appealing to a quantitative version of the density Hales-Jewett theorem~\cite{polymath2012new}.
Theorem~\ref{thm:3AP} makes progress on a question of Green~\cite{Green} and on a question of Haszla, Holenstein and Mossel~\cite{Mossel}.

\subsubsection{Multi-Player Parallel Repetition Theorems}
Parallel repetition is a basic building block in the area of interactive protocols and in particular in applications in the field of hardness of approximation.
In the setting of $k$-player games, we have a basic game $\Psi$ involving a verifier and $k$ players. The game consists of a set of questions $X$ that are supposed
to get labels from a finite alphabet $\Sigma$, a predicate $P\colon X^k\times \Sigma^k\to\{0,1\}$ that gives $k$-challenges and answers to them dictates
whether these answers are deemed satisfactory or not, and a distribution $\mu$ over $k$-tuples of challenges. In the basic game $\Psi$, the verifier
samples a challenge $(x_1,\ldots,x_k)\sim \mu$, sends the question $x_i$ to the $i$th player, receives an answer $a_i\in\Sigma$ from player $i$, and then
accepts if and only if $P(x_1,\ldots,x_k,a_1,\ldots,a_k) = 1$. The value of the game, ${\sf val}(\Psi)$, is defined to be the maximum probability the verifier
accepts under the best strategy for the players.

The $t$-fold repeated game, $\Psi^{\otimes t}$, is a game in which the verifier samples $t$ sets of challenges, say $(x_{1,j},\ldots,x_{k,j})\sim \mu$ for $j=1,\ldots,t$
independently, sends $(x_{i,1},\ldots,x_{i,t})$ to player $i$, receives from them answers $(a_{i,1},\ldots,a_{i,t})$ and accepts if and only if
$P(x_{1,j},\ldots,x_{k,j},a_{1,j},\ldots,a_{k,j}) = 1$ for all $j=1,\ldots,t$. In words, the game is repeated for $t$-times, but in parallel, and the verifier
confirms that each one of the executions of the basic game was accepting. It is clear that ${\sf val}(\Psi^{\otimes t})\geq {\sf val}(\Psi)^{t}$, and the main
question of interest in parallel repetition theorems is regarding the rate of decay of ${\sf val}(\Psi^{\otimes t})$ as a function of $t$; in particular is this
decay exponential?

For $2$-player games, i.e.\ for the case that $k=2$, this problem is by now well understood, and it is known that the value of $\Psi^{\otimes t}$ is
indeed exponentially decaying in $t$ (however not in the most obvious manner); see~\cite{Raz,Holenstein,Rao,BravermanGarg,DinurSteurer}. The techniques
that go into these proofs are either information theoretical, or analytical. In a sense, the analytical proofs are based on the well-known fact that the
eigenvalues of a matrix tensorize when one tensorizes the matrix, as it turns out that, in a sense, the value of a game can be vaguely viewed as eigenvalues
of a matrix.

For $k\geq 3$, the situation is much more complicated, and the only known bound for general games is due to Verbitsky~\cite{Verbitsky} and gives rather weak bounds
(as, once again, it relies on the density Hales-Jewett theorem).

Recently, the work of~\cite{DHVY} identified a class of games referred to as ``connected games'' for which the information theoretic techniques from the setting of
$2$-player games still work, which sparked renewed interest in multi-player parallel repetition theorems. We remark that the notion of ``connectedness'' therein is
very much similar to the notion of connectedness of distribution in our setting (which is much stronger than pairwise connectedness). This motivated a recent
line of works~\cite{HR,GMRZ,GHMRZ,GHMRZ2} that studied parallel repetition of $3$-player games over binary questions. This line of work started with studying
a game known as the GHZ game (which is well known in the physics literature and is a bottleneck to the techniques of~\cite{DHVY}), proving polynomial decay for
it, and using this as a stepping stone to prove polynomial decay parallel repetition theorems for more general classes of games.

We believe that the notion of Abelian embeddability should have a fundamental connection to the problem of parallel repetition in multiplayer games.
In a sense, this question too is about ``tensorization'' of some value, but in this time one has to deal with $k$-dimensional tensors as opposed to
matrices. Some evidence to that has been given in~\cite{GHZgameBKM}, wherein the authors give a very simple proof for the fact that the value of the
GHZ game is exponentially vanishing with $t$ (as opposed to just polynomial) which is inspired by Abelian embeddability. In a sense, the proof proceeds
by identifying that the GHZ game actually entails within it a $(\mathbb{Z}_4,+)$-type additive structure. Then, using this fact along with powerful theorems
from additive combinatorics, the authors give a structural result on the set of strategies for the players that perform well, which are then analyzed directly.

While being speculative, we believe that such connection should indeed exist, and in it the quantitative aspects of Theorem~\ref{thm:main_stab_3} should be
highly relevant. At the current state, the quantitative bounds we get are not very good, but we believe that with more effort these could be improved to
results that would be able to show $2^{-t^{\Omega(1)}}$ rate of decay in parallel repetition.

\subsection{High Level Overview of the Proof of Theorem~\ref{thm:main_stab_3}}\label{sec:pf_overview}
In this section we give a high level overview of the proof of Theorem~\ref{thm:main_stab_3}. As such, we often omit details, make simplifying assumptions
and appeal to intuition in order to concentrate on the main ideas. We also point out the sections relevant to each part of the argument.

At its core, our argument relies on the following intuition: if $\mu$ does not admit any Abelian embedding, then Theorem~\ref{thm:main_stab_3} is just
equivalent to the main result of~\cite{BKMcsp1,BKMcsp2}. Thus, one idea is to try to identify all Abelian embeddings of $\mu$, define partial basis
for $L_2(\Sigma^n;\mu_x^{\otimes n})$, $L_2(\Gamma^n;\mu_y^{\otimes n})$ and $L_2(\Phi^n;\mu_z^{\otimes n})$ based on these Abelian embeddings and
then show that for $f,g$ and $h$ to satisfy~\eqref{eq:1}, it must be the case that they correlated with a function from the span of this partial basis.
The intuition is completing the partial bases into complete bases, once we ``peel off'' these embeddings based functions the rest of the functions in
the bases are ``oblivious'' to the fact that $\mu$ admits Abelian embeddings. So, once we ``peel off'' these embedding functions, the situation is
analogous to the case that $\mu$ does not have any Abelian embeddings, in which case the result of~\cite{BKMcsp1,BKMcsp2} kick in.

Much of the effort in our proof goes into formalizing this rough idea, and once one is able to do that the rest of the proof is more streamline (but
still requires a significant effort). Below, we give step by step description of the way we formalize this intuition.
\subsubsection{Step 1: Master Embedding}
  The first issue is a distribution $\mu$ may have multiple linear embeddings, and they may interact in a non-trivial way. Indeed, given an Abelian
  embedding $(\sigma,\gamma,\phi)$ of $\mu$ into $H$, one can define a partial basis by composing characters of $H$ with the embedding functions.
  But how do we know that different embeddings give us different basis elements? How do we combine these partial bases into something that makes
  sense and is convenient to work with?

  Our first step is to identify that one may define a single embedding, which we refer to as the master embedding, that encapsulates within it
  all of the Abelian embeddings of $\mu$. Indeed, we show that if $\mu$ does not have any $(\mathbb{Z},+)$ embedding, then
  there is a size $M$ such that any Abelian embedding of $\mu$ ``comes from'' an Abelian embedding of $\mu$ into an Abelian group of size at most $M$. Hence, to include all Abelian embeddings it suffices to only look
  into embeddings of $\mu$ into Abelian groups of size at most $M$, and as there are only finitely many such embeddings we can just tensorize them.
  That is, letting $\sigma_i\colon \Sigma\to H_i$
  be all possible $\sigma$'s in linear embeddings of $\mu$, where $(H_i,+)$ are Abelian groups, the master embedding of $x$ is
  $\sigma_{{\sf master}}\colon \Sigma\to\prod\limits_{i=1}^{R} H_i$ defined by
  \[
  \sigma_{{\sf master}}(x) = (\sigma_1(x),\ldots,\sigma_R(x)),
  \]
  and similarly one may define $\gamma_{{\sf master}}$ and $\phi_{{\sf master}}$. With the master embeddings in hand we now have a sensible way
  of defining a partial basis for functions in $x$, $y$ and $z$ by considering compositions of characters from $H$ with the master embeddings.

  At the present state, this partial basis is not very convenient. For example, it may well be the case that there are distinct $\chi,\chi'\in H$
  such that $\chi\circ\sigma_{{\sf master}} = \chi'\circ \sigma_{{\sf master}}$. Indeed, this would be the case if the image of $\sigma_{{\sf master}}$
  was a strict subgroup of $H$. More generally, linear dependencies within $\{\chi\circ\sigma_{{\sf master}}\}_{\chi\in\hat{H}}$ already start appearing
  as soon as the image of $\sigma_{{\sf master}}$ is not the entire group $H$, and this presents issues which we wish to avoid.

  \skipi

  This part of the argument appears in Section~\ref{sec:master_embed}.

\subsubsection{Step 2: Saturating the Master Embeddings}
Our goal is therefore to arrange for the master embeddings $\sigma_{{\sf master}}$, $\gamma_{{\sf master}}$ and $\phi_{{\sf master}}$ to be \emph{saturated},
meaning that the image of each one of them is the entire group $H$. To do so, we must change the distribution $\mu$ into a distribution $\mu'$ such
that (a) on $\mu'$ the master embeddings are saturated, (b) there is a good enough relationship between $3$-wise correlations over $\mu$ and $3$-wise
correlations over $\mu'$, and (c) we can deduce the conclusion of Theorem~\ref{thm:main_stab_3} on $\mu$ from the conclusion of Theorem~\ref{thm:main_stab_3}
on $\mu'$.

This transformation is achieved via the path trick, introduced in~\cite{BKMcsp1}, which is ultimately just an application of the Cauchy-Schwarz inequality.
The path trick is used in our arguments extensively, and often time the structure we need is quite subtle thereby requiring a very careful application of
the path-trick. Nevertheless, below we explain at a high level the intuition behind the path trick and what it achieves.

Given a distribution $\mu$, the path trick distribution (of length $2t+1$) with respect to $x$ can be described as the following distribution $\mu'$:
\begin{enumerate}
  \item Sample $(x_1,y_1,z_1)\sim \mu$.
  \item Make a step from $y$: sample $(x_1',y_1',z_1')\sim \mu$ conditioned on $y_1' = y_1$.
  \item Make a step from $z$: sample $(x_2,y_2,z_2)$ conditioned on $z_2 = z_1'$.
  \item Repeat make a step from $y$/ make a step from $z$ for $t$ times.
\end{enumerate}
Thus, the sequences $(y_1,y_1',y_2,y_2',\ldots,y_t,y_t',y_{t+1})$ of $y$'s and
$(z_1,z_1',z_2,z_2',\ldots,z_t,z_t',z_{t+1})$ $z$'s are generated (where $z_{i+1} = z_i'$ and $y_i' = y_i$),
as well as a sequence $\vec{x} = (x_1,x_1',\ldots,x_t,x_t',x_{t+1})$ of $x$'s.
The output of the distribution $\nu$ is $(\vec{x}, y_{t+1}, z_1)$, and it is thought of as a $3$-ary distribution
over $\Sigma'\times\Gamma\times\Phi$ where $\Sigma'\subseteq \Sigma^{2t+1}$ is the set of feasible tuples of
$x$ in the procedure.

We refer to this procedure as the path trick since one may consider the bi-partite graph $G = (\Gamma\cup \Phi, E)$ whose
edges are $(y,z)\in \Gamma\times \Phi$ that are in the support of $\mu|_{\Gamma\times \Phi}$. Thus, thinking of the $x$'es
as labeling the edges of $G$, namely labeling an edge $(y,z)$ by $x$ if $(x,y,z)\in{\sf supp}(\mu)$, one gets that the above
procedure generates a random path of length $2t+1$ in the graph and record the labels of the edges that it traversed on.

Moving from the distribution $\mu$ to $\mu'$ has several benefits that have been used in our earlier papers:
\begin{enumerate}
  \item \textbf{Improving connectivity:} if $\mu$ is $\{2,3\}$-connected, then for large enough $t$ the support of $\mu'$ on the last two coordinates is full. Indeed,
  taking the random path view of the path trick, it is clear that as the graph $G$ is connected, for sufficiently large $t$ the same graph corresponding to
  $\mu'$ would be a complete bipartite clique.
  \item \textbf{Preserving properties of $\mu$:} he distribution $\mu'$ preserves much of the properties of the distribution $\mu$.
  In particular, if $\mu$ is pairwise connected then so is $\mu'$, and if $\mu$ does not admit any Abelian embeddings, then so does $\mu'$.
  \item \textbf{The $3$-wise correlations relations:} $3$-wise correlations of functions over $\mu$ can be upper bounded by $3$-wise correlations of functions related to the original
  functions over $\mu'$. Indeed, assume for simplicity that the functions are real valued. If $f$, $g$ and $h$ achieve large $3$-wise correlation in $\mu$,
  then for $(x,y,z)\sim \mu^{\otimes n}$ one has that the values $h(z)$ and $f(x)g(y)$ are correlated, so looking at the above path we get that
  $h(z_{i+1}) \approx f(x_{i+1})g(y_{i+1})$ and $g(y_{i+1})\approx f(x_{i+1}')h(z_{i+1}')$ and combining these we get that
  \[
  h(z_{t+1}) \approx f(x_{t+1}')h(z_{t+1}') \approx  f(x_{t+1}')f(x_{t})g(y_t)
  \approx\ldots\approx f(x_{t+1}')f(x_{t})\cdots f(x_{2})f(x_{1}') g(y_1),
  \]
  and hence we expected $g, h$ and $F(\vec{x}) = f(x_{t+1}')f(x_{t})\cdots f(x_{2})f(x_{1}')$ to achieve a significant correlation in $\mu'$. Indeed, this can be proved
  via an appropriate application of the Cauchy-Schwarz inequality.
\end{enumerate}

For the purposes of this paper we need additional properties of the path trick transformations, which we explain next.
\begin{enumerate}
  \item \textbf{Abelian ebmeddings of $\mu$ lift to Abelian embeddings of $\mu'$:}
  not only does the path trick preserve lack of Abelian embeddings, but in fact if $\mu$ does admit Abelian embeddings,
then $\mu'$ does not introduce any new ones. To be more precise, suppose
that $\sigma\colon \Sigma\to H$, $\gamma\colon \Gamma\to H$ and $\phi\colon \Phi\to H$ are Abelian embeddings of $\mu$.
Then, these embeddings give rise to an Abelian embedding $\sigma_t\colon \Sigma'\to H$ with $\gamma$ and $\phi$ of $\nu$, as follows:
\begin{equation}\label{eq:intro_lift_embed}
\sigma_t(\vec{x}) = \sum\limits_{i=1}^{t}\sigma(x_i) - \sigma(x_i') + \sigma(x_{t+1}).
\end{equation}
With the notation above, we have that
$\sigma(x_i) + \gamma(y_i) + \phi(z_i) = 0$, $\sigma(x_i') + \gamma(y_i') + \phi(z_i') = 0$, and doing a proper addition/ substraction one gets that
\[
\sigma_t(\vec{x})+\gamma(y_{t+1}) + \phi(z_1) = 0,
\]
hence $(\sigma_t,\gamma,\phi)$ form an Abelian embedding of $\mu'$ into $H$.
  \item \textbf{The only Abelian embeddings of $\mu'$ are lifts of Abelian embeddings of $\mu$:} all Abelian embeddings of $\mu'$ are precisely of this form.
  Namely, for any Abelian embedding $(\sigma',\gamma,\phi)$ of $\mu'$ into an Abelian group $(H,+)$ there is an Abelian embedding $(\sigma,\gamma,\phi)$ of
  $\mu$ into $(H,+)$ where $\sigma$ satisfies a relation as in~\eqref{eq:intro_lift_embed} where $\sigma'$ plays the role of $\sigma_t$.
  (see Lemma~\ref{lem:reverse_embed}).This result has a few important consequences, and in particular it says that the path trick preserves master embeddings.
  Namely, if we start with a master embedding of $\mu$, apply the path trick and the above transformation corresponding to it on the embeddings,
  then we will get the master embedding of $\mu'$.

  \item \textbf{Saturating the embeddings:} it can be easily observed that if $(\sigma_{{\sf master}},\gamma_{{\sf master}}, \phi_{{\sf master}})$
  is a master embedding of $\mu$ (or for this purpose, any embedding of $\mu$), then after the path trick we get the embedding
  $(\sigma_{{\sf master}, t},\gamma_{{\sf master}}, \phi_{{\sf master}})$ that clearly satisfies that
  ${\sf Image}(\sigma_{{\sf master}, t})\subseteq {\sf Image}(\sigma_{{\sf master}})$; this follows by looking
  at trivial paths that traverse the same edge back and fourth and use the same label of $x$ all of the time.
  Moreover, it is clear that if ${\sf Image}(\sigma_{{\sf master}})$ was a sub-group of $H$ then we would have that
  ${\sf Image}(\sigma_{{\sf master}, t}) = {\sf Image}(\sigma_{\sf master})$. It stands to reason that unless ${\sf Image}(\sigma_{{\sf master}})$
  is indeed a subgroup, then for large enough $t$ we would have that ${\sf Image}(\sigma_{{\sf master}, t})\subsetneq {\sf Image}(\sigma_{{\sf master}})$,
  in which case we enlarged the image of $\sigma_{{\sf master}}$ via the path trick.

  Indeed, something along these lines is true. Namely, we show that by combination of path tricks along different directions (not only $x$) one can indeed
  always enlarge the image of an embedding so long as it is not a subgroup.\footnote{In our formal proof this has to be done rather carefully as we wish
  to preserve the property that the alphabet of $x$ is always a power of the original alphabet $\Sigma$.}
\end{enumerate}
In conclusion, using the path trick multiple times we can pass to a new distribution $\nu$ on which the embeddings are all saturated, $3$-wise correlations
over $\mu$ are upper bounded by $3$-wise correlations over $\nu$, and $\nu$ has improved connectivity -- say that its support on the last two coordinates is
full. It can be easily shown that in that case, the images of all of the components must be the same sub-group, and without loss of generality we assume it
is the group $H$ itself.

Note that in particular, the above properties mean that if $(\sigma_{{\sf master}},\gamma_{{\sf master}}, \phi_{{\sf master}})$ is a saturated master embedding
of $\nu$, then the distribution of $(\sigma_{{\sf master}}(x),\gamma_{{\sf master}}(y), \phi_{{\sf master}}(z))$ where $(x,y,z)\sim \nu$ has a full support
on
\[
\sett{(a,b,c)\in H^3}{a+b+c=0},
\]
which intuitively says that by moving from $\mu$ to $\nu$ we have ``exposed'' all of the Abelian structure in the distribution $\mu$.

\skipi

This part of the argument appears in Section~\ref{sec:master_embed}.
\subsubsection{Step 3: Setting up a Basis Consisting of Embedding and Non-embedding Functions}
Fix distributions $\mu$ over $\Sigma\times \Gamma\times \Phi$ and $\nu$ over $\Sigma'\times\Gamma'\times\Phi'$ as we have done so far, and suppose
that (a power of) the $3$-wise correlation of $f,g$ and $h$ over $\mu$ is upper bounded by the $3$-wise correlation of $F, G$ and $H$ over $\nu$.

Now that we have saturated the master embeddings in $\nu$ we can set up a partial for functions in $x\in {\Sigma'}$ as
basis as before $\tilde{\chi}(x) = \chi(\sigma_{{\sf master}}(x))$ for all $\chi\in \hat{H}$ as before and get that now these functions
are indeed linearly independent We can then complete it to a basis of $L_2(\Sigma'; \nu_x)$ by adding to it the functions $\psi_1,\ldots,\psi_{s}$ that
are orthogonal to all functions in ${\sf Span}(B_1)$, so that writing $B_1 = \sett{\tilde{\chi}}{\chi\in \hat{H}}$ and $B_2 = \set{\psi_1,\ldots,\psi_s}$ we have a basis
$B_1\cup B_2$ for $L_2(\Sigma'; \nu_x)$.
Tensorizing, we get that $\set{v_{\vec{b}}}_{\vec{b}\in (B_1\cup B_2)^{\otimes n}}$
where $v_{\vec{b}}\colon \Sigma'^{n}\to \mathbb{C}$ is defined by $v_{\vec{b}}(x) = \prod\limits_{i=1}^{n} v_{b_i}(x_i)$, is a basis for $L_2({\Sigma'}^n; \nu_x^{\otimes n})$.
Thus, we can write
\[
F(x_1,\ldots,x_n) = \sum\limits_{\alpha\in (B_1\cup B_2)^{n}}\widehat{F}(\alpha) v_{\alpha}(x).
\]
We can define analogous bases for $L_2(\Gamma'; \nu_y)$ and $L_2(\Phi'; \nu_z)$. Now, each one of the functions $F$, $G$ and $H$ has an ``embedding part'', which is the
parts of the monomials that use functions from $B_1$, and ``non-embedding parts'', which are monomials using functions from $B_2$. Intuitively, it should be the case
that the more mass the functions have on the non-embedding parts, the smaller the $3$-wise correlations would be; this is because that for uni-variate functions
$u\colon \Sigma'\to\mathbb{C}$, $v\colon \Gamma'\to\mathbb{C}$ and $w\colon \Phi'\to\mathbb{C}$ of $2$-norm $1$, to achieve perfect $3$-wise correlation it must be
the case that $u(x) = v(y)w(z)$ in the support of $\nu$, in which case $u,v$ and $w$ behave like an embedding function. We remark that there is a serious leap in
this last step, which causes complications in later points of the argument. Later on, we refer to this gap as the Horn-SAT obstruction, and we will explain how
it arises and how to overcome it later on.

In light of the above, it makes sense to define two notions of degrees for our partial basis. The first of which is the embedding degree of a monomial $v_{\vec{b}}$,
${\sf embeddeg}(v_{\vec{b}})$, which is the number of components $v_{b_i}$ that come from the partial embedding basis $B_1$. The second of which is
the non-embedding degree of a monomial $v_{\vec{b}}$, ${\sf nedeg}(v_{\vec{b}})$, which is the number of components of $v_{b_i}$ that come from $B_2$.

\skipi

This part of the argument appears in Section~\ref{sec:prep}.
\subsubsection{Step 4: Analyzing the Contribution of High Non-embedded Degree Components}
The above discussion suggests that the parts of $F$, $G$ and $H$ of high non-embedding degree should contribute very little to their $3$-wise correlation
according to $\nu$. Formally showing this, however, is quite tricky and this is where a considerable amount of effort in this paper is devoted to.
Our argument here builds on an argument from~\cite{BKMcsp2} and it is the main subject in
Sections~\ref{sec:prep},~\ref{sec:max_merg_mot},~\ref{sec:base_case},~\ref{sec:reudce_to_homogenous},~\ref{sec:reduce_to_near_linear},~\ref{sec:prove_near_lin}.

To give some intuition for the argument we make several simplifying assumptions (some of which can be ensured, while other are not necessary).
Assume that the marginal distribution of $\nu$ over $y,z$ is uniform, and that
the distribution of $(\sigma_{{\sf master}}(x),\gamma_{{\sf master}}(y), \phi_{{\sf master}}(z))$ where $(x,y,z)\sim \nu$ is uniform over
on
\[
\sett{(a,b,c)\in H^3}{a+b+c=0}.
\]
Further assume that the functions $F$, $G$ and $H$ are embedding homogenous and non-embedding homogenous functions, by which we mean that the embedding degree
of each monomial of $F$ is the same, and also the non-embedding degree of each monomial in $F$ is the same; the same goes for $G$ and $H$. Our argument here
will be inductive on the number of coordinates $n$, and we show that the $3$-wise correlation of functions $F$, $G$ and $H$ as above can be upper bounded
by either the $3$-wise correlation of $n-1$ variate functions of the same non-embedding degree, or by $(1-\Omega(1))$ times the $3$-wise correlation of $n-1$ variate
functions with non-embedding degree smaller by $1$. Thus, iterating we would ultimately get a bound of $(1-\Omega(1))^{{\sf nedeg}(F)}$ on the $3$-wise correlations,
which is small if the non-embedding degree of $F$ is high.

In fact, we have two separate inductive arguments depending on if $n$ is much larger than the non-embedding degree of $F$, or if it is comparable to it; we refer
to this last case as the ``near linear non-embedding degree case'', and we now elaborate on each one of these cases.

\paragraph{The case that $n$ is much larger than the non-embedding degree of $F$.} In this case there is a variable, say the $n$th variable,
such that in almost all of the mass of $F$ lies on monomials in which the component of $x_n$ is an embedding function.
Using the homogeneity of $F$ we can use find a decomposition of $F$ as
\[
\sum\limits_{t\in T} \psi_t F_t(x_1,\ldots,x_{n-1})F_t'(x_n)
\]
where each $F_t'$ is either from $B_1$ or from $B_2$, and $\set{F_t}, \set{F_t'}$ form orthonormal sets and $\sum\limits_{t\in T}\card{\psi_t}^2 = 1$.
Similarly, we can find analogous decompositions form $G$ and $H$ as
\[
\sum\limits_{r\in R} \kappa_r G_r(y_1,\ldots,y_{n-1})G_r'(y_n),
\qquad \qquad
\sum\limits_{s\in S} \rho_s H_s(z_1,\ldots,z_{n-1})H_s'(z_n).
\]
Moreover, if $F_t'$ is a function from $B_1$ then $F_t$ has the same non-embedding degree as $F$, and if $F_t'$ is from $B_2$ then $F_t$ has one smaller non-embedding degree.
The same goes for $G$ and $H$, so to simplify presentation we consider the specialized case where
\begin{align*}
&F(x) = \psi_1 F_1(x_1,\ldots,x_{n-1})F_1'(x_n)+\psi_2 F_2(x_1,\ldots,x_{n-1})F_2'(x_n),\\
&G(y) = \kappa_1 G_1(y_1,\ldots,y_{n-1})G_1'(y_n) + \kappa_2 G_2(y_1,\ldots,y_{n-1})G_2'(y_n),\\
&H(y) = \rho_1 H_1(z_1,\ldots,z_{n-1})G_1'(z_n) + \rho_2 H_2(z_1,\ldots,z_{n-1})H_2'(z_n),
\end{align*}
where $F_1', G_1'$ and $H_1'$ are embedding functions and $F_2', G_2'$ and $H_2'$ are non-embedding functions. Thus, the coefficient $\psi_2$ is related to the
mass $x_n$ has on non-embedding functions and by choice is therefore small, and similarly we can expect it to be the case that $\kappa_2$ and $\rho_2$ are also
small (which is true, but requires some preparatory work). Thus, the $3$-wise correlation of $F$, $G$ and $H$ according to $\nu$ can be written as
\begin{align}\label{eq:intro_inductive}
\Expect{\nu^{n}}{F G H}
&=\psi_1\kappa_1\rho_1 \Expect{\nu^{\otimes n-1}}{F_1 G_1 H_1} \Expect{\nu}{F_1'G_1'H_1'}
+\psi_1\kappa_1\rho_2 \Expect{\nu^{\otimes n-1}}{F_1 G_1 H_2} \Expect{\nu}{F_1'G_1'H_2'}\notag\\
&+\psi_1\kappa_2\rho_1 \Expect{\nu^{\otimes n-1}}{F_1 G_2 H_1} \Expect{\nu}{F_1'G_2'H_1'}
+\psi_1\kappa_2\rho_2 \Expect{\nu^{\otimes n-1}}{F_1 G_2 H_2} \Expect{\nu}{F_1'G_2'H_2'}\notag\\
&+\psi_2\kappa_1\rho_1 \Expect{\nu^{\otimes n-1}}{F_2 G_1 H_1} \Expect{\nu}{F_2'G_1'H_1'}
+\psi_2\kappa_1\rho_2 \Expect{\nu^{\otimes n-1}}{F_2 G_1 H_2} \Expect{\nu}{F_2'G_1'H_2'}\notag\\
&+\psi_2\kappa_2\rho_1 \Expect{\nu^{\otimes n-1}}{F_2 G_2 H_1} \Expect{\nu}{F_2'G_2'H_1'}
+\psi_2\kappa_2\rho_2 \Expect{\nu^{\otimes n-1}}{F_2 G_2 H_2} \Expect{\nu}{F_2'G_2'H_2'}.
\end{align}
It turns out that terms the only term involving $F_1'$ that does not vanish is $\Expect{\nu}{F_1'G_1'H_1'}$. Indeed, as $F_1'$ is a function from $B_1$
we may write it as a product of a function on $y$ with a function on $z$, and thus expectations such as $\Expect{\nu}{F_1'G_1'H_2'}$ can be written as
expectation of product of a function of $y$ and a function of $z$. Using independence, this product can be further written as the product of two expectations
where at least one of these expectations is $0$.

Thus, if the terms involving $F_2$ were not existent, then we would get the upper bound
\[
\card{\Expect{\nu^{n}}{F G H} }
\leq
\card{\psi_1\kappa_1\rho_1} \card{\Expect{\nu^{\otimes n-1}}{F_1 G_1 H_1}} \card{\Expect{\nu}{F_1'G_1'H_1'}}
\leq \card{\Expect{\nu^{\otimes n-1}}{F_1 G_1 H_1}},
\]
and we have decreased the number of variables $n$ by $1$ (while keeping the non-embedding degree.
In a sense, as $\psi_2$ is small this term indeed should constitute the majority of the contribution to
$\Expect{\nu^{\otimes n}}{F G H}$, but we cannot just ignore the contribution from the other terms.

A naive attempt at bounding the other term (and using the Cauchy-Scharz in a favorable way) can show that
$\card{\Expect{\nu^{\otimes n}}{F G H}}$ is at most the maximum of $\card{\Expect{\nu^{\otimes n-1}}{F_i G_j H_k}}$ over $i,j$ and $k$,
however this bound is not good enough for us; indeed, if this maximum is achieved at anywhere other than $i=j=k=1$
then the non-embedding degrees decrease, and in that case we must gain a factor of $(1-\Omega(1))$ for our argument to
go through.

The key to improve upon this naive attempt lies in what we refer to as the ``additive base case''. The additive base case is a statement about univariate functions
that helps us to control the contribution form terms involving $F_2'$ in a favorable way. Stated simply, the additive base case we use is the statement that if
$F'$ is a univariate non-embedding function, and $G'$, $H'$ are any univariate functions, then
\[
\card{\Expect{\nu}{F'(G'+H')}}\leq (1-\Omega(1))\norm{F'}_2\norm{G'+H'}.
\]
The intuition for this inequality is that otherwise, the value of $F'$ would be very close to the value of $\overline{G'} + \overline{H'}$,
but this is only possible for embedding functions.

The point of the additive base case is that except for $\Expect{\nu}{F_2'G_2'H_2'}$,
the contribution of the terms not involving $F_1'$ in~\eqref{eq:intro_inductive} may be re-casted as an expectation of the form dealt with in
the additive base case. Indeed, if $F_1', G_1'$ and $H_1'$ were the simplest of embedding functions -- namely constant functions -- then this
is rather clear, as these terms can be written as
\[
\Expect{\nu}{F_2' (a G_2 + b H_2)}
\]
for some coefficients $a$ and $b$. In the case $F_1'$, $G_1'$ and $H_1'$ are not the constant functions more effort is needed, and in particular
one needs to guarantee that they ``come'' from the same character of $H$. Namely, that there is $\chi\in \hat{H}$ such that $F_1'$, $G_1'$
and $H_1'$ are multiples of $\chi\circ \sigma_{{\sf master}}$, $\chi\circ \gamma_{{\sf master}}$ and $\chi\circ \phi_{{\sf master}}$ respectively.
As we are able to guarantee this fact, the contribution of these terms can still be associated with the additive base case, as essentially
$G_1' = \overline{F_1'}\overline{H_1'}$ and $H_1' = \overline{F_1'} \overline{G_1'}$. Hence, that contribution can be re-written as
\[
\Expect{\nu}{F_1'F_2' (a G_2 + b H_2)}
\]
for some coefficients $a$ and $b$, and this is still an expectation of the form handled by the additive base case.

Making an effective enough use of the additive base case, one can use~\ref{eq:intro_inductive} to either reduce $n$ by $1$ and keep the non-embedding degrees
the same, or else reduce both $n$ and the non-embedding degrees by $1$ and then gain a factor of $1-\Omega(1)$.

\skipi

This part of the argument is presented in Section~\ref{sec:reduce_to_near_linear}.
\paragraph{The near linear non-embedding degree case.}
Once the non-embedding degree of $F$, $G$ and $H$ is comparable to $n$, the above argument no-longer works, and the reason is that the
last term in~\eqref{eq:intro_inductive} is no longer negligible, and at the same time we do not know how to give an effective upper bound on
it using only the additive base case.
Thus we must have a new base case that handles this last term, and intuitively one would expect the following base case to hold.
Suppose that $\tilde{F}\colon \Sigma'\to\mathbb{C}$ is a function from $B_2$, and $\tilde{G}\colon \Gamma'\to\mathbb{C}$,
$\tilde{H}\colon \Phi'\to\mathbb{C}$ are any functions of $2$-norm $1$. Then
\[
\card{\Expect{(x,y,z)\sim \nu}{\tilde{F}(x)\tilde{G}(y)\tilde{H}(z)}}\leq 1 - \Omega(1).
\]
The reason we expect this to be true is that otherwise, by compactness we would be able to find $3$ such functions satisfying $\tilde{F}(x) = \tilde{G}(y)\tilde{H}(z)$
on the support of $\nu$. Thus, the logs of these functions form an Abelian embedding so $\log(\tilde{F})$ must be an embedding function, in which case $\tilde{F}$ is also
an embedding function in contradiction.

As stated, this argument is not quite correct, since it may be the case that the function $F'$ gets the value $0$ sometimes, in which case we cannot take
logs and get away with it. This obstruction has already appeared in~\cite{BKMcsp2} wherein it was referred to as the ``Horn-SAT obstruction'', and here too
we have to face it. In fact, as in our scenario we need to maintain many more properties of the distribution $\nu$ (compared to what was necessary in~\cite{BKMcsp2}),
more care is needed to handle the Horn-SAT obstruction. Ultimately, the Horn-SAT obstruction is dealt with by stating a more complicated base case statement
which we can guarantee to hold for the distribution $\nu$ while being useful enough to make our argument go through. For the simplicity of presentation however,
we ignore this obstruction for now and explain the argument in the case the ideal base case holds.

Equipped with the ideal base case, we can give effective enough bounds on the last term in~\eqref{eq:intro_inductive}. In particular, if all of the contribution
came from it, we would have been able to conclude that
\[
\card{\Expect{\nu^{\otimes n}}{F G H}}\leq (1-\Omega(1))\card{\Expect{\nu^{\otimes n-1}}{F_2 G_2 H_2}},
\]
and iterating would finish the proof. One again however, there are other terms in~\eqref{eq:intro_inductive} that need to be accounted for (the other terms
involving $F_2$). To do so, ideally we would have liked (as in the additive base case) to re-arrange these terms so as to view their total contribution as
an instantiation of the ideal base case, but this is not possible. Using a similar (but more complicated) argument we can still show that it is in fact
the case that
\[
\card{\Expect{\nu^{\otimes n}}{F G H}}\leq (1-\Omega(1))\card{\Expect{\nu^{\otimes n-1}}{\tilde{F_2} \tilde{G_2} \tilde{H_2}}},
\]
where $\tilde{F_2}, \tilde{G_2}$ and $\tilde{H_2}$ are functions of non-embedding degree at most $1$ less of $F_2$, $G_2$ and $H_2$.

\skipi
This part of the argument is presented in Section~\ref{sec:prove_near_lin}.

\paragraph{Overcoming the Horn-SAT Obstruction.}
The bulk of Sections~\ref{sec:prep},~\ref{sec:max_merg_mot},~\ref{sec:base_case},~\ref{sec:reudce_to_homogenous} is devoted to gaining additional properties of
$\nu$, as well as other crucial reductions (for example, to allow us to assume homogeneity of $F, G$ and $H$ as above). One of the key properties achieved
in this section is the so-called relaxed base case, which is a replacement for the ideal base case above that we are able to ensure.

A triplet of functions $u\colon \Sigma'\to\mathbb{C}$, $v\colon \Gamma'\to\mathbb{C}$ and $w\colon \Phi'\to\mathbb{C}$ is called a Horn-SAT embedding
if $u(x) = v(y)w(z)$ in the support of $\nu$. If $u$ never vanishes, a Horn-SAT embedding can be transformed into an Abelian embedding, and thus
(simply put) the Horn-SAT obstruction is really about the possible $0$-patterns non-embedding functions may have. By careful manipulations
(once again utilizing the path trick) we are able to find a set $\Sigma_{{\sf modest}}\subseteq \Sigma'$ of size at least $2$ such no Horn-SAT embedding
can vanish on. Thus, we get that if $u$ doesn't vanish on $\Sigma_{{\sf modest}}$ then it can never be a part of a Horn-SAT embedding. Therefore,
it is natural to expect that if $u$ has variance at least $\tau$ on $\Sigma_{{\sf modest}}$, then
\[
\card{\Expect{\nu}{u v w}}\leq 1-\theta(\tau)
\]
where $\theta(\tau)>0$ is some function of $\tau$. This turns out to be true and useful, but there are many subtleties. For once, we need additional properties
from $\Sigma_{{\sf modest}}$ to make this relaxed base case useful, and most important we need the symbols in $\Sigma_{{\sf modest}}$ to be mapped to the same
group element in $H$ by the master embedding. Secondly (and this has already appeared in~\cite{BKMcsp2}) we need a decent dependency between $\tau$ and $\theta(\tau)$.

\subsubsection{Step 5: Reducing to Functions over $H$}
We now return our functions $F$, $G$ and $H$, armed with the knowledge that the contribution of high non-embedding degree parts if small.
Thus, taking their parts of small non-embedding degree $F'$, $G'$ and $H'$, we are able to conclude that
$\card{\Expect{\nu^{\otimes n}}{F G H} - \Expect{\nu^{\otimes n}}{F'G'H'}}\leq \frac{\eps}{100}$, and so the $3$-wise correlation of $F'$, $G'$ and $H'$ according
to $\nu$ is still significant.

We remark that as in our actual argument we will need the functions $F'$, $G'$ and $H'$ to be bounded, so harsh truncations
as we described do not fit the bill. Thus, we use a softer notion of truncations given by the \emph{non-embedding noise operator}.
For $\rho\in [0,1]$ consider the Markov chain $\T_{\text{non-embed}, \rho}$ on $\Sigma'$ that on $x$, with probability $\rho$ takes $x' = x$, and otherwise
samples $x'\sim \nu|_x$ conditioned on $\sigma_{{\sf master}}(x') = \sigma_{{\sf master}}(x)$. When $\rho$ is not specified, that is, when we write
$\T_{\text{non-embed}}$, we mean that $\rho$ is taken to be $0$. Given such Markov chain one may consider the corresponding averaging operator on
$L_2(\Sigma',\nu_x)$ given as
\[
\T_{\text{non-embed}, \rho} f(x) = \Expect{x'\sim\T_{\text{non-embed}, \rho} x}{f(x')},
\]
and tensorize it to get an averaging operator $\T_{\text{non-embed}, \rho}^{\otimes n} \colon L_2({\Sigma'}^{n},\nu_x^{\otimes n})\to L_2({\Sigma'}^{n},\nu_x^{\otimes n})$.
Simiarly, we can get averaging operators on $L_2({\Gamma'}^{n},\nu_y^{\otimes n})$ and $L_2({\Phi'}^{n},\nu_z^{\otimes n})$, which by abuse of notation we also
denote by $\T_{\text{non-embed}, \rho}$. These averaging operators can be shown to essentially kill monomials of high non-embedding degree, hence serve as a replacement
for harsh truncation arguments as above.

With these operators in hand, we can replace the harsh truncations above by
$F' = \T_{\text{non-embed}, \rho}^{\otimes n} F$, $G' = \T_{\text{non-embed}, \rho}^{\otimes n} G$ and $H' = \T_{\text{non-embed}, \rho}^{\otimes n} H$ (for suitably
chosen $\rho$) and effectively
be in the same situation as before, wherein we have functions $F', G'$ and $H'$ that have almost all of their mass on monomials with small non-embedding degrees,
and also that $\card{\Expect{\nu^{\otimes n}}{F G H} - \Expect{\nu^{\otimes n}}{F'G'H'}}\leq \frac{\eps}{100}$

\skipi
We wish to transform the functions $F'$, $G'$ and $H'$ into related bounded functions with non-embedding degree $0$ for which the $3$-wise correlation over $\nu$
is still significant. For that, we use a combination of random restrictions (so as the mass of $F'$, $G'$, and $H'$ of small but not $0$ non-embedding degree would
almost all collapse to level $0$), followed by averaging (to get rid of all monomials of positive non-embedding degree). Thus, we get
functions $F''' = \T_{\text{non-embed}}(F'')$, $G''' = \T_{\text{non-embed}}(G'')$ and $H''' = \T_{\text{non-embed}}(H'')$ where $F''$, $G''$ and $H''$ are random
restrictions of $F', G'$ and $H'$, so that with noticeable probability we have that
\[
\Expect{(x,y,z)\sim \nu^{\otimes n'''}}{F'''(x)G'''(y)H'''(z)}\geq \frac{\eps}{2}
\]
where $n'''$ is the number of coordinates left alive after the random restriction.
Now the functions $F'''$, $G'''$ and $H'''$ can be viewed as functions defined over $H^{n'''}$,
so the above expectation should be amendable to standard tools from discrete Fourier analysis.

This part of the argument appears in Section~\ref{sec:the_hastad_argument}.
\subsubsection{Step 6: Applying the Linearity Testing Argument}
Indeed, we re-cast the functions $F'''$, $G'''$ and $H'''$ above as $F^{\sharp}\colon H^{n'''}\to \mathbb{C}$,
$G^{\sharp}\colon H^{n'''}\to \mathbb{C}$ and $H^{\sharp}\colon H^{n'''}\to \mathbb{C}$ defined in the natural way (for example,
$F^{\sharp}(a) = F(x)$ for $x$ such that $\sigma_{{\sf master}}(x_i) = a_i$ for each coordinate, where we know that the specific choice of
$x$ doesn't matter). Thus, from the distribution $\nu$ we get a corresponding distribution $\nu^{\sharp}$ over $\sett{(a,b,c)\in H^3}{a+b+c = 0}$
whose support is full, and
\[
\card{\Expect{(a,b,c)\sim (\nu^{\sharp})^{\otimes n'''}}{F^{\sharp}(a)G^{\sharp}(b)H^{\sharp}(c)}}\geq \frac{\eps}{2}.
\]
We now use random restrictions again, but for a different reason. Namely, we use random restrictions to
shift from the distribution $\nu^{\sharp}$ to the uniform over $\sett{(a,b,c)\in H^3}{a+b+c = 0}$, and
get from $F^{\sharp}$, $G^{\sharp}$ and $H^{\sharp}$ restrictions ${F^{\sharp}}'$, ${G^{\sharp}}'$ and ${H^{\sharp}}'$ so that with noticeable probability
\[
\card{\Expect{a,b\in H^{n''''}}{{F^{\sharp}}'(a){G^{\sharp}}'(b){H^{\sharp}}'(-a-b)}}\geq \frac{\eps}{4},
\]
where $n''''$ is the number of coordinates left alive.
In this case, a straightforward, classical Fourier analytic computation can be applied to relate the left hand side to the Fourier coefficients of
${F^{\sharp}}'$, ${G^{\sharp}}'$ and ${H^{\sharp}}'$  so that we get
\[
\card{\sum\limits_{\chi\in \hat{H}^{\otimes n''''}}\widehat{{F^{\sharp}}'}(\chi)\widehat{{G^{\sharp}}'}(\chi)\widehat{{H^{\sharp}}'}(\chi)}
\geq \frac{\eps}{4},
\]
from which one can quickly conclude that there is $\chi$ such that $\card{\widehat{{F^{\sharp}}'}(\chi)}\geq \frac{\eps}{4}$. In words, after a sequence of
random restrictions, averaging and further random restriction, the function $F$ is correlated with a function of the form $\chi\circ \sigma_{{\sf master}}$.
This is the type of result we are after, except that we wish to have such result for $F$ and not for $F$ after this sequence of operations.

\skipi
This part of the argument appears in Section~\ref{sec:the_hastad_argument}.
\subsubsection{Step 7: Going back to $F$ via the Restriction Inverse Theorem}
We have thus concluded that after random restriction, $F^{\sharp}$ is correlated with a function of the form $\chi\circ \sigma_{{\sf master}}$
where $\chi\in\hat{H}^{n''''}$, and we wish to unravel the steps we took to get from $F$ to $F^{\sharp}$ and conclude a structural result about $F$.

Noting that $\chi\circ \sigma_{{\sf master}}$ is a product function, this is precisely a situation in which the
restriction inverse theorem kicks in, and using a modified version of Theorem~\ref{thm:rest_inverse_intro} we are able to conclude that
$F^{\sharp}$ is correlated with a function of the form $L\circ \sigma_{{\sf master}}\cdot \chi\circ\sigma_{{\sf master}}$ where
$\chi\in\hat{H}^{n'''}$ and $L$ is a low-degree
function of $2$-norm at most $1$. Thus, the same conclusion holds for $F'''$ (as it is essentially the same function as $F^{\sharp}$.

Recalling that $F''' = \mathrm{T}_{\text{non-embed}, 0}F''$, we get that
\[
\card{\inner{\mathrm{T}_{\text{non-embed}, 0}F''}{L\circ \sigma_{{\sf master}}\cdot \chi\circ\sigma_{{\sf master}}}}\geq \frac{\eps}{4},
\]
but on the other hand we also have that
\begin{align*}
\card{\inner{\mathrm{T}_{\text{non-embed}, 0}F''}{L\circ \sigma_{{\sf master}}\cdot \chi\circ\sigma_{{\sf master}}}}
&=
\card{\inner{F''}{\mathrm{T}_{\text{non-embed}, 0}^{*}(L\circ \sigma_{{\sf master}}\cdot \chi\circ\sigma_{{\sf master}})}}\\
&=
\card{\inner{F''}{\mathrm{T}_{\text{non-embed}, 0}(L\circ \sigma_{{\sf master}}\cdot \chi\circ\sigma_{{\sf master}})}}\\
&=
\card{\inner{F''}{L\circ \sigma_{{\sf master}}\cdot \chi\circ\sigma_{{\sf master}}}},
\end{align*}
where we used the fact that $\mathrm{T}_{\text{non-embed}, 0}$ is self-adjoint. Hence, we conclude that $F''$ is correlated with
$L\circ \sigma_{{\sf master}}\cdot \chi\circ\sigma_{{\sf master}}$.

We now wish to unravel the last step of random restriction (that goes from $F$ to $F''$), and for that we once again want to appeal to
the restriction inverse theorem. However, the correlations we are talking about now are not quite about correlations with product functions.
Amusingly, to circumvent this issue we apply more random restrictions. Intuitively, after a suitably chosen random restriction, the function
$L$ becomes close to constant, hence one expects the fact that $F''$ is correlated with $L\circ \sigma_{{\sf master}}\cdot \chi\circ\sigma_{{\sf master}}$
to convert to the fact that a random restriction of $F''$ is correlated with a restriction of $\chi\circ\sigma_{{\sf master}}$ (which is a product
function), and we show that this is indeed the case. Thus, we conclude that a random restriction of $F''$ is correlated with a function of the
form $\chi\circ\sigma_{{\sf master}}$. Noting that a random restriction of $F''$ is (overall) a random restriction of $F$ (with different parameters),
we are thus able to conclude from the restriction inverse theorem that $F$ itself is correlated with a function of the form
$L\cdot \chi\circ\sigma_{{\sf master}}$.

\skipi
This part of the argument appears in Section~\ref{sec:the_hastad_argument}.

\subsubsection{Step 8: Going back to $f$ via Properties of the Master Embedding}
The last step in the proof of Theorem~\ref{thm:main_stab_3} is to use the structural result obtained for the function $F$ to deduce a similar structural result for $f$.
For that, we recall that (ignoring complex conjugates) the value of $F(x_1,\ldots,x_{s})$ is $f(x_1)\cdots f(x_s)$, and, ignoring the low-degree part $L$ for now,
we know that $F$ is correlated with $\chi\circ \sigma_{{\sf master}}$ for some $\chi\in\hat{H}^{n}$. Recalling the relation~\ref{eq:intro_lift_embed}, one quickly gets
from it that
\[
\chi\circ \sigma_{{\sf master}}(x_1,\ldots,x_{s})
=\prod\limits_{i=1}^{s}\chi\circ\sigma_{{\sf master}}(x_i),
\]
where, by abuse of notation, $\sigma_{{\sf master}}$ on the right hand side is the master embedding of $\mu$ (which is the original distribution, prior to
any application of the path trick). Hence, the correlation between $F$ and $\chi\circ \sigma_{{\sf master}}$ translates to the fact that
\[
\card{\Expect{(x_1,\ldots,x_s)\sim \nu_x^{\otimes n}}{\prod\limits_{i=1}^{s}f(x_i)\chi\circ\sigma_{{\sf master}}(x_i)}}\geq \eps' = \eps'(\eps)>0
\]
As discussed earlier, in~\cite{Mossel} it is shown that if $\nu_x$ is a connected distribution, a correlation such as in the above can be noticeable
only if $f\cdot \chi\circ\sigma_{{\sf master}}$ is correlated with a low-degree function. Thus, the proof would be concluded if we are able to
ensure connectivity of $\nu_x$, which we are indeed able to. This requires some care in some of our earlier steps, and most notably in the way we
apply the path trick. In fact, we are able to guarantee that the support of $\nu_x$ is full, that is, $\Sigma^{s}$.

Bringing the low-degree part $L$ back, essentially the same argument works except that we need to apply a suitable random restriction beforehand to get
rid of the low-degree part. Thus, the previous argument gives that with noticeable probability, a random restriction of $f\cdot \chi\circ\sigma_{{\sf master}}$
is correlated with a low-degree function. Hence, after more random restrictions, we conclude that with noticeable probability a random restriction of
$f\cdot \chi\circ\sigma_{{\sf master}}$ is correlated with a constant function. Re-phrasing, this means that with noticeable probability a random restriction
of $f$ is correlated with a function of the form $\chi\circ\sigma_{{\sf master}}$, and a final invocation of the restriction inverse theorem finishes the proof.
\skipi

This part of the argument appears in Section~\ref{sec:unraveling}.

\section{Preliminaries}
\paragraph{Notations.} We denote $[n] = \set{1,\ldots,n}$.
For a vector $x\in \Sigma^n$ and a subset $I\subseteq [n]$ of coordinates, we denote by
$x_I$ the vector in $\Sigma^{I}$ which results by dropping from $x$ all coordinates outside $I$. We denote by
$x_{-I}$ the vector in $\Sigma^{n-\card{I}}$ resulting from dropping from $x$ all coordinates from $I$;
if $I = \{i\}$ we often simplify the notation and write it as $x_{-i}$.  For $I\subseteq[n]$, $a\in \Sigma^{I}$
and $b\in \Sigma^{n-\card{I}}$ we denote by $(x_I = a, x_{-I} = b)$ the point in $\Sigma^{n}$ whose $I$-coordinates
are filled according to $a$, and whose $\overline{I}$-coordinates are filled according to $b$. For two strings $x,y\in \Sigma^n$
we denote by $\Delta(x,y)$ the Hamming distance between $x$ and $y$, that is, the number of coordinates $i\in [n]$ such that $x_i\neq y_i$.

We denote $A\lll B$ to refer to the fact that $A\leq C\cdot B$ for some absolute constant $C>0$,
and $A\ggg B$ to refer to the fact that $A\geq c \cdot B$ for some absolute constant $c>0$.
If this constant depends on some parameter, say $m$, the corresponding notation is $A\lll_m B$. We will also
use standard big-$O$ notations: we denote $A = O(B)$ if $A\lll B$, $A = \Omega(B)$ if $A\ggg B$; if
there is dependency of the hidden constant on some auxiliary parameter, say $m$, we denote $A = O_m(B)$
and $A = \Omega_m(B)$.

We denote by ${\bm i}$ the complex root of $-1$, and by $\overline{a}$ the complex conjugate of the number $a\in\mathbb{C}$.
For a matrix $M$, we denote by $M^{*}$ the conjugate transpose matrix of $M$.

\subsection{Product Spaces}
Let $(\Sigma^n,\mathcal{D}^{\otimes n})$ be a probability space. We often work with the space
$L_2(\Sigma^n; \mathcal{D}^{\otimes n})$ of complex valued functions with finite values. We think of this space
as an inner product space, where the inner product of $f,f'\colon \Sigma^n\to\mathbb{C}$ is defined by
\[
\inner{f}{f'}_{L_2(\Sigma^n; \mathcal{D}^{\otimes n})} = \Expect{x\sim\mathcal{D}^{\otimes n}}{f(x)\overline{f'(x)}}.
\]
Often times, when the measure $\mathcal{D}^{\otimes n}$ is clear from context, we will omit the $L_2(\Sigma^n; \mathcal{D}^{\otimes n})$
subscript and denote the inner product between $f$ and $f'$ by $\inner{f}{f'}$.
\subsection{The Degree Decomposition and the Efron-Stein Decomposition}
Given an inner product space, one may associate with it orthogonal decompositions of $L_2(\Sigma^n; \mathcal{D}^{\otimes n})$.
In this section we present two such decompositions, the degree decomposition and its refinement the Efron-Stein decomposition.
We will only present the basic notions and facts we need about them, and refer the reader to~\cite{ODonnell} to a more comprehensive
treatment.

\subsubsection{Juntas, Degrees and the Degree Decomposition}
To define the notion of degrees, it is most convenient to start with the notion of juntas, which are functions that depend only on
few of their input coordinates.
\begin{definition}
  For $D\subseteq[n]$, a function $f\colon \Sigma^n\to\mathbb{C}$ is called a $D$-junta if there exists
  $f'\colon \Sigma^{D}\to\mathbb{C}$ such that $f(x) = f'(x_D)$ for all $x\in\Sigma^n$.

  For an integer $0\leq d\leq n$, a function $f\colon \Sigma^n\to\mathbb{C}$ is called a $d$-junta if there
  exists a set $D\subseteq [n]$ of size $d$ such that $f$ is a $D$-junta.
\end{definition}

Equipped with the notion of juntas, we may define the degree decomposition in the following way:
\begin{definition}
    For an inner product space $L_2(\Sigma^n; \mathcal{D}^{\otimes n})$ as above and an integer $0\leq d\leq n$, we define
    the space $V_{\leq d}(\Sigma^n; \mathcal{D}^{\otimes n})\subseteq L_2(\Sigma^n; \mathcal{D}^{\otimes n})$ as the space spanned
    by all $d$-junta. We also define
    \[
    V_{=d}(\Sigma^n; \mathcal{D}^{\otimes n}) = V_{\leq d}(\Sigma^n; \mathcal{D}^{\otimes n})\cap V_{\leq d-1}(\Sigma^n; \mathcal{D}^{\otimes n})^{\perp}.
    \]
\end{definition}
It is clear by definition that the spaces $V_{=d}(\Sigma^n; \mathcal{D}^{\otimes n})$ are mutually orthogonal and
\[
L_2(\Sigma^n; \mathcal{D}^{\otimes n}) = \bigoplus_{d=0}^{n}V_{=d}(\Sigma^n; \mathcal{D}^{\otimes n}),
\]
so any function $f\colon\Sigma^n\to\mathbb{C}$ can be uniquely written as $f(x) = \sum\limits_{d=0}^{n} f^{=d}(x)$ where $f^{=d}\in V_{=d}(\Sigma^n; \mathcal{D}^{\otimes n})$
is called the degree $d$ component of $f$. With these notations, Plancherel's equality states that for any pair of functions $f,g\colon (\Sigma^n,\mathcal{D}^{\otimes n})\to\mathbb{C}$ one has that
\[
\inner{f}{g} = \sum\limits_{d=0}^{n}\inner{f^{=d}}{g^{=d}}.
\]
Parseval's equality is the specialized statement where $f=g$, in which case one get that $\norm{f}_2^2 = \sum\limits_{d=0}^{n}\norm{f^{=d}}_2^2$.
We often refer to the function $f^{\leq d}(x) = \sum\limits_{i=0}^{d} f^{=i}(x)$ as the part of $f$ of level at most $d$, and
refer to the quantity $\norm{f^{\leq d}}_2^2$ as the level $d$ weight of $f$:
\begin{definition}
  The weight of $f$ on level up to $d$ is defined as $W_{\leq d}[f] = \norm{f^{\leq d}}_2^2 = \sum\limits_{i=0}^{d}\norm{f^{=i}}_2^2$.
\end{definition}

\subsubsection{The Efron-Stein Decomposition}
The Efron-Stein decomposition is a refinement of the degree decomposition, which we will make use of a handful of times.
\begin{definition}
For an integer $0\leq d\leq n$ and $S\subseteq [n]$ of size $d$, we define
\[
V_{=S}(\Sigma^n; \mathcal{D}^{\otimes n}) = V_{=d}(\Sigma^n; \mathcal{D}^{\otimes n})\cap \spn\left(\sett{f\colon \Sigma^n\to\mathbb{C}}{f\text{ is an $S$-junta}}\right).
\]
\end{definition}
It can be shown that the spaces $V_{=S}(\Sigma^n; \mathcal{D}^{\otimes n})$ are mutually orthogonal, thus
\[
V_{=d}(\Sigma^n; \mathcal{D}^{\otimes n}) = \bigoplus_{\card{S} = d}V_{=S}(\Sigma^n; \mathcal{D}^{\otimes n}).
\]
In particular, given any function $f\colon\Sigma^n\to\mathbb{C}$ and $0\leq d\leq n$, we may further decompose the degree $d$ component
of $f$, namely $f^{=d}$, and uniquely write it as $f^{=d}(x) = \sum\limits_{\card{S} = d} f^{=S}(x)$ where $f^{=S} \in V_{=S}(\Sigma^n; \mathcal{D}^{\otimes n})$.
Thus, we get the Efron-Stein decomposition of $f$:
$f(x) = \sum\limits_{S\subseteq [n]}{f^{=S}(x)}$, where $f^{=S}\in V_{=S}(\Sigma^n; \mathcal{D}^{\otimes n})$.
Once again, with these notations Plancherel's equality states that for any pair of functions $f,g\colon (\Sigma^n,\mathcal{D}^{\otimes n})\to\mathbb{C}$, one has that
\[
\inner{f}{g} = \sum\limits_{S\subseteq [n]}\inner{f^{=S}}{g^{=S}},
\]
and Parseval's equality is the specialized statement where $f=g$, in which case we get $\norm{f}_2^2 = \sum\limits_{S\subseteq [n]}\norm{f^{=S}}_2^2$.

\subsection{Random Restrictions}\label{sec:random_restrictions}
In this section we define the notions of restrictions and random restrictions of functions, which are used in this paper extensively.
We use two types of random restrictions. The first type is very common in the area
of analysis of Boolean functions; one selects a set of coordinates randomly, fixes them according to the marginal distribution
there and thinks of the rest of the coordinates as variables. The underlying measure of the input space stays the same.
The second type is a much less common type of restrictions, and it is crucial for our arguments. In this second type, we still chooses a random set of variables and a fixing
for them, but not according to these the marginal distribution on this coordinates. Rather, the fixing for this set of variables is chosen according to a different
measure, and to balance this out the underlying measure of the rest of the coordinates changes. Below is a more formal description.

\subsubsection{Restrictions that Preserve the Underlying Measure}
For a finite alphabet $\Sigma$ and a probability measure $\mu$ over it,
a function $f\colon (\Sigma^n, \mu^{\otimes n})\to\mathbb{C}$, a set of coordinates $I\subseteq [n]$ and a partial input
$z\in \Sigma^I$, the restricted function $f_{I\rightarrow z}\colon \Sigma^{[n]\setminus I}\to\mathbb{R}$ is defined as
\[
f_{I\rightarrow z}(y) = f(x_I = z, x_{\bar{I}} = y).
\]
A random restriction of a function $f\colon (\Sigma^n, \mu^{\otimes n})\to\mathbb{C}$ refers to a restriction in which either
(or both) $I$ and $z$ are chosen randomly. Typically, random restrictions are associated with a parameter $\rho\in (0,1)$:
we first choose $I\subseteq[n]$ by including each element $i\in [n]$ independently with probability $\rho$, then choose
$z\sim \mu^{\overline{I}}$ and then consider the function $f_{\overline{I}\rightarrow z}$ as a function from $(\Sigma^{I}, \mu^{I})$
to $\mathbb{C}$. we often denote by $I\subset_{\rho} [n]$ the distribution of $I$ which is sampled in such a way.

\subsubsection{Restrictions that Do Not Preserve the Underlying Measure}
An important utility of restrictions for us will be that they allow us to change the underlying measure of our probability
space. Suppose that the measure $\mu$ can be written as $\mu = \rho\mathcal{D}_1 + (1-\rho)\mathcal{D}_2$, where $\mathcal{D}_1$ and
$\mathcal{D}_2$ are distributions and $\rho\in (0,1)$. In such situations we will often consider the following random restriction process:
choose $I\subseteq_{\rho} [n]$, choose $z\sim \mathcal{D}_2^{\overline{I}}$, and consider the function
$f_{\overline{I}\rightarrow z}$ as a function from $(\Sigma^{I}, \mathcal{D}_1^{I})$
to $\mathbb{C}$. Note that under these random choices, choosing $y\sim \mathcal{D}_1^{I}$,
the distribution of the point $(x_I = y, x_{\bar{I}} = z)$ is still $\mu$, hence this restriction process makes sense. In particular, the
expected average, as well as the expected $2$-norm squared of $f_{\overline{I}\rightarrow z}$ over the choice of $z$ are the average
and the $2$-norm squared of $f$.

Such random restrictions are used extensively in the paper. An example case where this can be useful is the case that in
the distribution $\mu$ the probability of each atom is at least $\alpha$, and we wish to switch from it to the uniform
distribution over $\Sigma$. In that case, we may write $\mu = \frac{\alpha}{2}U + \left(1-\frac{\alpha}{2}\right)\mu'$ where
$U$ is the uniform distribution over $\Sigma$ and $\mu'$ is some distribution. Following the above procedure for random
restrictions, we may thus change the underlying measure of our space from $\mu$ to $U$ by approximately fixing $1-\frac{\alpha}{2}$
randomly chosen fraction of the coordinates according to $\mu'$.

\subsection{Markov Chains}
Given a probability space $(\Sigma,\mu)$, we will often consider Markov chains over $\Sigma$ that have $\mu$ as a stationary distribution.
We often denote these Markov Chain by $\mathrm{T}$, and abusing notations we will also think of $\mathrm{T}$ as an averaging operator from
$L_2(\Sigma,\mu)$ to $L_2(\Sigma,\mu)$ defined as
\[
\mathrm{T} f(x) = \Expect{y\sim \mathrm{T}x}{f(y)}.
\]
We say a Markov chain $\mathrm{T}$ is connected if the graph, whose vertices are $\Sigma$ and the edges are $(a,b)$ if there is
a transition from $a$ to $b$ in $\mathrm{T}$, is connected. We need a few well known basic properties of Markov chains that we record
below.
\begin{fact}\label{fact:MC_eval}
If $\mathrm{T}$ is connected and the probability of each atom is at least $\alpha$, then $\lambda_2(\mathrm{T})\leq 1-\Omega(\alpha^2)$.
\end{fact}
\begin{proof}
  The proof is by an application of Cheeger's inequality on the graph associated with $\mathrm{T}$, and we refer the reader
  to~\cite{Mossel} for a formal proof.
\end{proof}
Given an averaging operator $\mathrm{T}$ acting on univariate functions, we often think of its $n$-fold tensor $\mathrm{T}^{\otimes n}$.
Again, we will think of $\mathrm{T}^{\otimes n}$ both as a Markov chain over $\Sigma^n$ (on which, on each coordinate the Markov chain $\mathrm{T}$ is applied
independently), as well as an averaging operator acting on $L_2(\Sigma^n,\mu^{\otimes n})$.
In the case that $\mu$ is a stationary distribution of $\mathrm{T}$, it is easily shown that the spaces $V_{=S}$ are invariant under
$\mathrm{T}^{\otimes n}$, and for each $g\in V_{=S}$ it holds that $\norm{\mathrm{T}^{\otimes n} g}_2\leq \lambda_2(\mathrm{T})^{\card{S}}\norm{g}_2$.

\skipi
The following lemma asserts that if $f,g\colon\Sigma^n\to\mathbb{C}$ are $1$-bounded and $\inner{f}{\mathrm{T}^{\otimes n}g}$ is
significant, then $f$ and $g$ must have significant mass of the low levels.
\begin{lemma}\label{lem:noticeable_to_lowdegwt}
  Suppose $(\Sigma,\mu)$ is a finite domain and $\mathrm{T}$ is a connected Markov chain with stationary distribution $\mu$, in
  which the probability of each atom is at least $\alpha$. Then for all $\eps>0$ there is $d = O_{\alpha}(\log(1/\eps))$ such that
  if $f,g\colon (\Sigma^n,\mu^{\otimes n})\to \mathbb{C}$ are $1$-bounded and $\card{\inner{f}{\mathrm{T}^{\otimes n}g}}\geq \eps$, then
  $W_{\leq d}[f]\geq \frac{\eps^2}{4}$.
\end{lemma}
\begin{proof}
  Decomposing $f,g$ according to the Efron-Stein decomposition of $(\Sigma^n,\mu^{\otimes n})$ as $f = \sum\limits_{S} f^{=S}$
  and $g = \sum\limits_{S} g^{=S}$, we get that
  \[
  \inner{f}{\mathrm{T}^{\otimes n} g}
  =\sum\limits_{S,Q}\inner{f^{=S}}{\mathrm{T}^{\otimes n} g^{=Q}}
  =\sum\limits_{S}\inner{f^{=S}}{\mathrm{T}^{\otimes n} g^{=S}}
  =\sum\limits_{S\neq \emptyset}\inner{f^{=S}}{\mathrm{T}^{\otimes n} g^{=S}}.
  \]
  The contribution from $\card{S}\leq d$ is at most
  \begin{align*}
  \sum\limits_{\card{S}\leq d}\card{\inner{f^{=S}}{\mathrm{T}^{\otimes n} g^{=S}}}
  \leq
  \sum\limits_{\card{S}\leq d}\norm{f^{=S}}_2\norm{\mathrm{T}^{\otimes n} g^{=S}}_2
  \leq
  \sum\limits_{\card{S}\leq d}\norm{f^{=S}}_2\norm{g^{=S}}_2
  &\leq\sqrt{W_{\leq d}[f]W_{\leq d}[g]}\\
  &\leq \sqrt{W_{\leq d}[f]}.
  \end{align*}
  For $\card{S}>d$, we have by Fact~\ref{fact:MC_eval} that
  \[
  \norm{\mathrm{T}^{\otimes n} g^{=S}}_2
  \leq \lambda_2(\mathrm{T})^{\card{S}}\norm{g^{=S}}_2
  \leq (1-\Omega_{\alpha}(1))^d\norm{g^{=S}}_2
  \leq \frac{\eps}{2}\norm{g^{=S}}_2
  \]
  for $d$ chosen suitably as in the statement. Thus, using Cauchy-Schwarz the contribution from $\card{S} > d$ is at most
  \begin{align*}
  \sum\limits_{\card{S}>d}\card{\inner{f^{=S}}{\mathrm{T}^{\otimes n} g^{=S}}}
  \leq
  \sum\limits_{\card{S}>d}\norm{f^{=S}}_2\norm{\mathrm{T}^{\otimes n} g^{=S}}_2
  &\leq
  \frac{\eps}{2}\sum\limits_{\card{S}>d}\norm{f^{=S}}_2\norm{g^{=S}}_2\\
  &\leq \frac{\eps}{2}\sqrt{\sum\limits_{\card{S}>d}\norm{f^{=S}}_2^2}\sqrt{\sum\limits_{\card{S}>d}\norm{g^{=S}}_2^2},
  \end{align*}
  which is at most $\frac{\eps}{2}$ by Parseval. Combining, we get that $\sqrt{W_{\leq d}[f]}+\frac{\eps}{2}\geq \eps$, and the statement follows by re-arranging.
\end{proof}
\subsection{Some Markov Chain Lemmas}
In this section, we collect a few basic results regarding Markov chains that we will make use of repeatedly.
\subsubsection{The Eigenvalues of a Markov Chain}
The first statement gives a description of the eigenvalues of a Markov chain (both upper and
lower bounds) as a function of the minimum probability of each atom and the minimum probability that the Markov chain stays in the same state.
It will help us to analyze noise-type operators that are similar to the standard noise operator on product spaces (but are not quite the same).
\begin{lemma}\label{lem:mossel_MC}
  For all $\alpha>0$ and $m\in\mathbb{N}$ there are $c>0$ and $C>0$ such that the following holds.
  Let $\Sigma$ be an alphabet of size at most $m$, let $\mu$ be a distribution over $\Sigma$ in which the probability of each atom
  is at least $\alpha$, and let $\mathrm{T}$ be a Markov chain over $\Sigma$ in which $\mu$ is a stationary distribution. Let
  $k$ be the number of connected components in $\mathrm{T}$, and let $\lambda_1(\mathrm{T})\geq \ldots\lambda_{m}(\mathrm{T})$ be
  the eigenvalues of $\mathrm{T}$.
  \begin{enumerate}
    \item We have $\lambda_1(\mathrm{T}) = \ldots = \lambda_k(\mathrm{T}) = 1$.
    \item If the probability of each transition in $\mathrm{T}$ is at least $\xi$,
    then $\lambda_{i}(\mathrm{T})\leq 1-c\xi$ for all $i\geq k+1$.
    \item If for all $x\in\Sigma$ we have that $\Prob{x'\sim\mathrm{T} x}{x'=x}\geq 1-\xi$, then
    $\lambda_i(\mathrm{T})\geq 1-C\xi$.
  \end{enumerate}
\end{lemma}
\begin{proof}
  For the first item, write $\Sigma = \Sigma_1\cup\ldots\cup \Sigma_k$ where $\Sigma_1,\ldots,\Sigma_k$ are the connected components
  of $\mathrm{T}$, and take $g_i$ which is $1$ on $\Sigma_i$ and $0$ on the rest. Note that the $g_i$'s are linearly
  independent and all have eigenvalues $1$ in $\mathrm{T}$.

  For the second item, let $g$ be an eigenvector of $\mathrm{T}$ perpendicular to $g_1,\ldots,g_k$ with eigenvalue $\lambda$, and normalize it so that
  $\norm{g}_2^2 = 1$. Then there is $i$ such that
  \[
  \sum\limits_{x\in \Sigma_i}\mu(x)\card{g(x)}^2\geq \frac{1}{k},
  \]
  so either for $h = {\sf Re}(g)$ or $h = {\sf Im}(g)$ we have that
  $\sum\limits_{x\in \Sigma_i}\mu(x)h(x)^2\geq \frac{1}{2k}$. Therefore there must be $x\in\Sigma_i$
  such that either $h(x)\geq \sqrt{1/2k}$ or $h(x)\leq -\sqrt{1/2k}$; without loss of generality
  we assume the former. As $\inner{g}{g_i} = 0$ we get $\sum\limits_{x\in\Sigma_i}\mu(x) g(x) = 0$
  hence $\sum\limits_{x\in\Sigma_i}\mu(x) h(x) = 0$, and in particular there is $x'\in\Sigma_i$ such that $h(x)\leq 0$.
  Summarizing, $\card{h(x) - h(x')}\geq 1/\sqrt{2k}$. As $\Sigma_i$ is a connected component of $\mathrm{T}$, there is a path $x=x_0\rightarrow\ldots\rightarrow x_{\ell} = x'$
  between $x$ and $x'$, and so
  \[
  \frac{1}{\sqrt{2k}}\leq \card{h(x) - h(x')}\leq \sum\limits_{i=0}^{\ell-1}\card{h(x_i) - h(x_{i+1})},
  \]
  so there is $i$ such that $\card{h(x_i) - h(x_{i+1})}\geq \frac{1}{\sqrt{2k}\ell}\geq \frac{1}{m^{3/2}\sqrt{2}}$.
  Therefore $\card{g(x_i) - g(x_{i+1})}\geq \frac{1}{m^{3/2}\sqrt{2}}$, and it follows that
  \[
  1-\lambda
  =1-\inner{g}{\mathrm{T}g}
  = \frac{1}{2}\Expect{\substack{z\sim \mu\\ z'\sim \mathrm{T}z}}{\card{g(z)-g(z')}^2}
  \geq \frac{\xi}{2}\card{g(x_i) - g(x_{i+1})}^2
  \geq \frac{\xi}{4m^3},
  \]
  finishing the proof of the second item.

  For the third item, taking $g$ as before we have that
  \[
  1-\lambda
  =1-\inner{g}{\mathrm{T}g}
  = \frac{1}{2}\Expect{\substack{z\sim \mu\\ z'\sim \mathrm{T}z}}{\card{g(z)-g(z')}^2}
  \leq \frac{1}{2}\Prob{\substack{z\sim\mu\\ z'\sim\mathrm{T}z}}{z\neq z'} 4\norm{g}_{\infty}^2,
  \]
  which is at most $\xi \frac{2}{\alpha}\norm{g}_2^2 \leq \frac{2}{\alpha} \xi$,
  and the third item is proved.
\end{proof}

\subsubsection{Markov Chains and Random Restrictions}
We will often measure various notion of degrees via Markov Chains based notions (as opposed to precise degrees), so as to preserve
boundedness of functions. As such, we often want to make assertions of the form: ``if a function $f$ has high degree, then a random restriction of $f$ also
has high degree'' (where again, the notion of degree is not necessarily the standard notion). Such statements are quite straightforward when dealing
with the standard notion of degree, but less so with Markov Chain based notions. For our purposes, the following lemma will play the role of such statement for our softer notion of high-degreeness.

In the statement below one should think of the quantity $\inner{f}{\mathrm{T}_{1-c\beta\xi, \mu^{\otimes n}}f}$ as small, and of the fact that it is small
as saying that if we write $f$ in basis of eigenvectors of the operator $\mathrm{T}_{1-c\beta\xi, \mu^{\otimes n}}$, then most of the $L_2$-mass of $f$
will lie on monomials involving many eigenvectors whose eigenvalue is not $1$; the number of such eigenvectors will need to $\Omega\left(\frac{1}{\beta\xi}\right)$.
Morally, the lemma says that after a random restriction leaving $\beta$ fraction of the coordinates alive, the degree with respect to such eigenvectors is
at least $\Omega\left(\frac{1}{\xi}\right)$, but some care is needed as the operator with respect to which we measure degree, changes.

\begin{lemma}\label{lem:op_comparison_lemma}
  For all $\alpha,\beta>0$ there is $c>0$ such that the following holds.
  Let $\Sigma$ be an alphabet of size at most $m$, let $\mu, \nu_1, \nu_2$ be distributions over $\Sigma$ in which the probability of
  each atom is at least $\alpha$ and $\mu = \beta \nu_1 + (1-\beta)\nu_2$.  Let $G = (\Sigma, E)$ be a graph on $\Sigma$, and consider
  the Markov chains $\mathrm{T}_{1-\xi, \mu}$ and $\mathrm{T}_{1-\xi, \nu_1}$ defined as: for
  $\mathrm{T}_{1-\xi, \mu}$, on $x\in \Sigma$, a sample $x'\sim \mathrm{T}_{1-\xi, \mu}x$ is generated
  by taking $x'=x$ with probability $1-\xi$, and otherwise take $x'\sim \mu$ conditioned on $x'$ being a neighbour
  of $x$ in $G$. We similarly define $\mathrm{T}_{1-\xi, \nu_1}$. Then for all $f\colon\Sigma^n\to\mathbb{C}$,
  \[
  \Expect{\substack{I\subseteq_{\beta} [n]\\ z\sim \nu_2^{\overline{I}}}}
  {\inner{f_{\overline{I}\rightarrow z}}{\mathrm{T}_{1-\xi, \nu_1^I}f_{\overline{I}\rightarrow z}}_{\nu_1^{I}}}
  \leq \inner{f}{\mathrm{T}_{1-c\beta\xi, \mu^{\otimes n}}f}.
  \]
\end{lemma}
\begin{proof}
  Expanding, the left hand side is equal to $\Expect{(x,x')\sim \mathcal{D}}{f(x)\overline{f(x')}}$,
  where the distribution $\mathcal{D}$ is defined as: for each $i$ independently, with probability $1-\beta$
  we take $x_i = x_i'$ sampled according to $\nu_2$, and otherwise we sample $x_i\sim \nu_1$ and then $x_i'\sim \mathrm{T}_{1-\xi, \nu} x_i$.
  Noting that marginally, $x$ and $x'$ are distributed according to $\mu$, we may view $\mathcal{D}$ as a reversible Markov chain
  $\mathrm{T}^{\otimes n}$ whose stationary distribution is $\mu$, so that the left hand side is equal to
  $\inner{f}{\mathrm{T}^{\otimes n} f}$. We show that it is at most the right hand side, and to do so we examine the eigenvalues and eigenvectors
  of these operators. We begin by remarking that the two operators are symmetric and positive semi-definite
  (as $\inner{f}{\mathrm{T}_{1-\xi, \mu^{\otimes n}}f}=\inner{\mathrm{T}_{\sqrt{1-\xi}, \mu^{\otimes n}} f}{\mathrm{T}_{\sqrt{1-\xi}, \mu^{\otimes n}}f}\geq 0$),
  hence these eigenvalues are non-negative.

  Let $k$ be the number of connected components in $G$, and note that if $g$ is a function which is constant on the connected components of $G$, then
  both operators act on it as the identity. We choose a basis for $L_2(\Sigma)$, say $g_1,\ldots,g_k,g_{k+1},\ldots,g_m$, where $g_1,\ldots, g_k$ are
  constant on all of the connected components of $G$ and $g_{k+1},\ldots,g_m$ are perpendicular to $g_1,\ldots,g_k$ with respect to $\inner{\cdot}{\cdot}_{\mu}$.
  Then we have that $\mathrm{T} g_i = \mathrm{T}_{1-\xi, \mu} g_i$ for $i=1,\ldots,k$. Also, the space $\spn(g_{k+1},\ldots,g_m)$ is invariant
  under both $\mathrm{T}$ and $\mathrm{T}_{1-\xi, \mu}$, and by Lemma~\ref{lem:mossel_MC} all eigenvalues of
  $\mathrm{T}$ are at most $1-s\beta\xi$ for some $s(m,\alpha)>0$ (as the probability of each transition is at least $\xi\beta$), and all
  eigenvalues of $\mathrm{T}_{1-c\xi\beta,\mu}$ are at least $1-cC\xi\beta > 1-s\beta\xi$.
  Thus, for every $g\in{\sf Span}(g_{k+1},\ldots,g_m)$ we have that $\inner{g}{\mathrm{T}g}\leq\inner{g}{\mathrm{T}_{1-c\xi\beta,\mu} g}$, and
  it follows that this inequality holds for every $g\colon \Sigma\to\mathbb{C}$.

  For multi-variate functions,
  writing $V = \spn(g_{1},\ldots,g_k)$ and $V' = \spn(g_{k+1},\ldots,g_m)$,
  we may decompose
  \[
  L_2(\Sigma;\mu^{\otimes n}) = \oplus_{S\subseteq [n]} \spn(V^{\otimes S}\otimes V'^{\otimes [n]\setminus S})
  \]
  and thus write any
  $g\colon \Sigma^n\to\mathbb{C}$ as $g = \sum\limits_S c_S g_S$ where $g_S\in \spn(V^{\otimes S} \otimes V'^{\otimes [n]\setminus S})$ has
  $2$-norm equal to $1$.

  \paragraph{A Computation for $\mathrm{T}$.}
  For $\mathrm{T}$ we now get that:
  \[
    \inner{g}{\mathrm{T}^{\otimes n} g}
    =
    \sum\limits_{S}\card{c_S}^2\inner{g_S}{\mathrm{T}^{\otimes n} g_S}
    \leq\sum\limits_{S}\card{c_S}^2\norm{g_S}_2\norm{\mathrm{T}^{\otimes n} g_S}
    \leq\sum\limits_{S}(1-s\beta\xi)^{n-\card{S}}\card{c_S}^2\norm{g_S}_2^2.
  \]
  In the last inequality, we used the fact that $\norm{\mathrm{T}^{\otimes n} g_S}\leq (1-s\beta\xi)^{n-\card{S}}c_S^2\norm{g_S}_2$,
  which may be observed as follows. Further decomposing $V$ and $V'$ into eigenspaces of $\mathrm{T}$, we may write
  $g_S = \sum\limits_{a\in A_S} c_a' v_a$ where $v_a\in V^{\otimes S} \otimes V'^{\otimes [n]\setminus S}$ are orthogonal unit vectors.
  Moreover, $v_a$ is an eigenvector of $\mathrm{T}$ with eigenvalue $0\leq \lambda_a\leq (1-s\beta\xi)^{2(n-\card{S})}$, and so
  \[
  \norm{\mathrm{T}^{\otimes n} g_S}_2^2
  =\sum\limits_{a\in A_S} \card{c_a'}^2 \lambda_a^2
  \leq \sum\limits_{a\in A_S} \card{c_a'}^2(1-s\beta\xi)^{2(n-\card{S})}
  =(1-s\beta\xi)^{2(n-\card{S})} \norm{g_S}_2^2.
  \]

  \paragraph{A Computation for $\mathrm{T}'$.}
  Similarly, for $\mathrm{T}_{1-\xi, \mu^{\otimes n}}$, we further decompose $V$ and $V'$ into eigenspaces of $\mathrm{T}_{1-\xi, \mu^{\otimes n}}$
  and write $g_S = \sum\limits_{a\in A_S} c_a' v_a$ where the $v_a\in V^{\otimes S}\otimes V'^{\otimes [n]\setminus S}$'s are orthogonal, have $2$-norm equal to
  $1$ and are each tensor of eigenvectors of $\mathrm{T}_{1-\xi, \mu}$. Thus, $\mathrm{T}_{1-\xi, \mu^{\otimes n}} v_a = \lambda_a v_a$ for $\lambda_a\geq (1-cC\xi\beta)^{n-\card{S}}$ and so
  \begin{align*}
    \inner{g}{\mathrm{T}_{1-\xi, \mu^{\otimes n}} g}
    =
    \sum\limits_{S}\card{c_S}^2\inner{g_S}{\mathrm{T}_{1-\xi, \mu^{\otimes n}} g_S}
    &=\sum\limits_{S,a\in A_S} \card{c_S}^2c_a^2\lambda_a \norm{v_a}_2^2\\
    &\geq \sum\limits_{S} \card{c_S}^2(1-cC\xi\beta)^{n-\card{S}} \sum\limits_{a\in A_S} \card{c_a'}^2,
  \end{align*}
  which is at least $\sum\limits_{S}(1-s\beta\xi)^{n-\card{S}}c_S^2\norm{g_S}_2^2\geq \inner{g}{\mathrm{T}^{\otimes n} g}$,
  as required.
\end{proof}

\subsubsection{Comparing Two Markov Chains}
The next lemma is tailored to handle the following case. Suppose that $\mathrm{T}$ and $\mathrm{T}'$ are two Markov chains on $\Sigma$
that both have $\mu$ as stationary distribution, and whose corresponding averaging operators $\mathrm{T},\mathrm{T}'\colon L_2(\Sigma,\mu)\to L_2(\Sigma,\mu)$
are both positive semi-definite. Intuitively, if the Markov chain $\mathrm{T}$ is richer than $\mathrm{T}$', then the averaging operator $\mathrm{T}$ does
more averaging than $\mathrm{T}'$. Hence, if we know that for some function $f\colon (\Sigma^n,\mu^{\otimes n})\to\mathbb{C}$ of $2$-norm equal to $1$ it holds
that $\norm{\mathrm{T}'^{\otimes n} f}_2$ is small, then $\norm{\mathrm{T}^{\otimes n} f}_2$ should also be small. More formally:
\begin{lemma}\label{lem:compare_averaging_ops}
  For all $m\in\mathbb{N}$ and $\alpha,\beta>0$ there is $s>0$ such that the following holds.
  Suppose that $\mu$ is a distribution over $\Sigma$ in which the probability of each atom is at least $\alpha$ and $\card{\Sigma}\leq m$,
  and let $\mathrm{T}, \mathrm{T}'$ be reversible Markov chains on $\Sigma$ with $\mu$ as stationary distribution.
  Further suppose that there is $\delta>0$ such that the following holds:
  \begin{enumerate}
    \item For all $x\in\Sigma$, $\Prob{x'\sim \mathrm{T}x}{x' = x}\geq 1-\delta$ and $\Prob{x'\sim \mathrm{T}'x}{x' = x}\geq 1-\delta$.
    \item For any distinct $x,y\in \Sigma$ such that $\mathrm{T}'(x,y)>0$ it holds that $\mathrm{T}(x,y)>0$.
    \item The probability of each transition in $\mathrm{T}$, $\mathrm{T}'$ is at least $\beta \delta$.
  \end{enumerate}
  Then for all $f\colon (\Sigma^n,\mu^{\otimes n})\to\mathbb{C}$ of $2$-norm at most $1$ it holds that
  \[
  \norm{\mathrm{T}^{\otimes n} f}_2
  \leq
  \norm{\mathrm{T}'^{\otimes n} f}_2^{s}.
  \]
\end{lemma}
\begin{proof}
  Suppose for simplicity of notation that $\card{\Sigma} = m$.
  Let $C_1\cup\ldots\cup C_k$ be a partition of $\Sigma$ into the connected components of $\mathrm{T}$, and let
  $V \subseteq L_2(\Sigma)$ be the sub-space of functions that are constant on each $C_i$; we take $V' = V^{\perp}$,
  so that $L_2(\Sigma^n,\mu^{\otimes n}) = \bigoplus_{S\subseteq [n]}\spn(V_S)$ where $V_S = V^{\otimes S}\otimes V'^{\otimes[n]\setminus S}$.

  We note that as $\mathrm{T}$ is self adjoint (by reversibility) and $V$ is an invariant space of $\mathrm{T}$, it follows that $V'$ is
  also an invariant space of $\mathrm{T}$.  We thus write $f = \sum\limits_{S\subseteq [n]} c_S f_S$ where $f_S\in \spn(V_S)$ has $2$-norm
  equal to $1$, and get that $\mathrm{T} f_S$ is also in $\spn(V_S)$. Similarly, $\mathrm{T}' f_S$ is in $\spn(V_S)$, as
  $\mathrm{T}'$ preserves $V$ and as each connected component of $\mathrm{T}'$ is contained in a connected component of $\mathrm{T}$. Thus, we get that
  \[
    \norm{\mathrm{T}^{\otimes n} f}_2^2
    =
    \sum\limits_{S} \card{c_S}^2\norm{\mathrm{T}^{\otimes n} f_S}_2^2,
    \qquad\qquad
    \norm{\mathrm{T}'^{\otimes n} f}_2^2
    =
    \sum\limits_{S} \card{c_S}^2\norm{\mathrm{T}'^{\otimes n} f_S}_2^2.
  \]
  We now argue that
  $\norm{\mathrm{T}^{\otimes n} f_S}_2^2\leq (1-c(\alpha,\beta,m)\delta)^{2(n-\card{S})}$ and
  $\norm{\mathrm{T}'^{\otimes n} f_S}_2^2\geq (1-C(\alpha,\beta,m)\delta)^{2(n-\card{S})}$ where $c,C>0$.
  Indeed, for the first inequality we further decompose $V$ and $V'$ into eigenspaces of $\mathrm{T}$
  (as in Lemma~\ref{lem:op_comparison_lemma}) and proceed with the same computation as there; we use
  Lemma~\ref{lem:mossel_MC} to upper bound the eigenvalues of $\mathrm{T}$. For the second
  inequality, we decompose $V$ and $V'$ into eigenspaces of $\mathrm{T}'$, and proceed with the same computation as there; we use
  Lemma~\ref{lem:mossel_MC} to lower bound the eigenvalues of $\mathrm{T}'$. Thus, we conclude that there is a constant $A(\alpha,\beta,m)$
  such that
  \[
  \norm{\mathrm{T}^{\otimes n} f}_2^2
  \leq
  \sum\limits_{S} \card{c_S}^2(1-c(\alpha,\beta,m)\delta)^{2(n-\card{S})},
  \qquad
  \norm{\mathrm{T}'^{\otimes n} f}_2^2
  \geq
  \sum\limits_{S} \card{c_S}^2(1-c(\alpha,\beta,m)\delta)^{2A(n-\card{S})}.
  \]
  By H\"{o}lder's inequality, we get that $\norm{\mathrm{T}^{\otimes n} f}_2^{2A}$ can be upper bounded as
  \begin{align*}
    \left(\sum\limits_{S} \card{c_S}^2(1-c(\alpha,\beta,m)\delta)^{2(n-\card{S})}\right)^{A}
    &=
    \left(\sum\limits_{S} \card{c_S}^{2(A-1)/A}\cdot \card{c_S}^{2/A}(1-c(\alpha,\beta,m)\delta)^{2(n-\card{S})}\right)^{A}\\
    &\leq
    \left(\sum\limits_{S} \card{c_S}^2\right)^{A-1}
    \sum\limits_{S} \card{c_S}^{2}(1-c(\alpha,\beta,m)\delta)^{2A(n-\card{S})}\\
    &\leq \norm{f}_2^{2(A-1)}\norm{\mathrm{T}'^{\otimes n} f}_2^2,
  \end{align*}
  where we used Parseval. As $\norm{f}_2\leq 1$, taking $2A$-th root finishes the proof.
\end{proof}

\section{On the Master Embedding, Path Trick, and Their Interaction}\label{sec:master_embed}
In this section, we present master embeddings as well as the path trick from~\cite{BKMcsp1,BKMcsp2}, and prove some properties of them that will be crucial to us.
The path trick is an idea which was used in previous works in this series, and for multiple reasons. In~\cite{BKMcsp1} it was used to enrich the distribution
$\mu$ so as to gain pairwise independence, and in~\cite{BKMcsp2} it was used to gain a more limited form of pairwise independence as well as for
overcoming a certain technical challenge referred to as the ``Horn-SAT obstruction'' therein.\footnote{This obstruction will also appear in the present work and here
too we will make an essential use of the path trick to resolve it; this will be the topic of discussion in Section~\ref{sec:base_case}.} One important property
of the path trick used in both of these works, is that if a distribution $\mu$ does not admit Abelian embeddings, then applying the path trick on
it results in a distribution that also does not admit Abelian embeddings

In the current work, however, we have to deal with distribution admitting Abelian embeddings, and thus we have to dig deeper.
In particular, we have to study the interaction between the path trick and  Abelian embeddings of $\mu$, which is the primary topic
of this section. To do that, we first
define the master embedding, which is an embedding of $\mu$ that ``encapsules'' within it all Abelian embeddings on $\mu$.
We then
show two properties of the path trick and master embeddings:
\begin{enumerate}
  \item \textbf{The Master is Preserved Under Path Tricks.}
  We show that the path trick never ``introduces'' new Abelian embeddings. By that, we mean that from any
Abelian embedding of the original distribution $\mu$ one can construct an Abelian embedding of distribution after the path trick,
and furthermore (and this is the important part) there are no other Abelian embeddings. This means that the a master embedding for
$\mu$ remains a master embedding for the distribution after the path trick (after an appropriate transformation).
  \item \textbf{The Master is Saturated after Path Tricks.} For reasons that were discussed in the introduction
  and will be further discussed below, it is desirable for us that the image of our master embeddings will be a whole group.
  A-priori, there is no reason this will be the case, and indeed it is most often not. We show that, by applying
  the path trick in a certain way, we are able to enrich the predicate so that the image of the master embeddings become complete
  Abelian groups.
\end{enumerate}

\subsection{Defining the Master Embedding}
Let $\Sigma$, $\Gamma$, and $\Phi$ be finite alphabets and let $\mu$ be a distribution over $\Sigma\times\Gamma\times \Phi$. In this section,
we wish to show that there is a single Abelian group and a single embedding of $\mu$ that within it ``encapsulates'' all Abelian embeddings of $\mu$,
which we often refer to as a master embedding of $\mu$.

To start getting some intuition, note that every Abelian embedding of $\mu$ given by $\sigma\colon \Sigma\to (H,+)$, $\gamma\colon\Gamma\to(H,+)$ and $\phi\colon \Phi\to(H,+)$,
where $(H,+)$ is a finite Abelian group can be thought about as partitions of each one of the alphabets $\Sigma$, $\Gamma$ and $\Phi$, and labeling each
part of each partition by a group element. Thus, as the number of these partitions is a finite number depending only on the alphabet sizes, it makes sense that we would not need
to look into too large of groups to find proper labelings of that partition by group elements so as to get an Abelian embedding (if such one exists).
This requires some care, and towards this end we define the notion of equivalent embeddings.
\begin{definition}
  Let $\mu$ be a distribution over $\Sigma\times \Gamma\times \Phi$. We say an embedding
  $\sigma \colon\Sigma\to (G,+)$, $\gamma\colon \Gamma\to (G,+)$ and $\phi\colon \Phi \to (G,+)$
  is a linear reduction of
  $\sigma' \colon\Sigma\to (H,+)$, $\gamma'\colon \Gamma\to (H,+)$ and $\phi'\colon \Phi \to (H,+)$
  if $G$ is a subgroup of $H$, and there are injective linear maps
  $m_1, m_2, m_3\colon G \to H$,
  such that $\sigma'(x) = m_1(\sigma(x))$, $\gamma'(y) = m_2(\gamma(y))$ and $\phi'(z) = m_3(\phi(z))$.
\end{definition}
As a concrete example for the notion of linear refinements, we note that if $\sigma$, $\gamma$ and $\phi$
form an embedding of $\mu$ into $(\mathbb{Z}_p,+)$, then $p \sigma$, $p\gamma$ and $p\phi$ form
an embedding of $\mu$ into $(\mathbb{Z}_{p^2},+)$. In essence though, these two embeddings
are ``the same'', thus to encapsulate all of the Abelian embeddings of $\mu$ it suffices
to only take into account one of them. Indeed, one can observe that $(\sigma,\gamma,\phi)$ is a linear
reduction of $(p\sigma,p\gamma,p\phi)$.

\skipi

Below, we consider a distribution $\mu$ that does not admit non-trivial Abelian embeddings into $(\mathbb{Z},+)$
and gradually construct a master embedding of it.
Towards this end, we consider the basic building block of all Abelian groups, namely cyclic groups,
and show that modulo equivalences, $\mu$ can admit only $O_{m}(1)$ many Abelian embeddings,
where $m$ is an upper bound on the alphabet sizes.

\subsubsection{Embeddings into Cyclic Groups}
We start by showing that if $\mu$ does not admit any $(\mathbb{Z},+)$ embedding, then $\mu$ cannot have Abelian
embedding into arbitrarily large Abelian groups of prime order.
\begin{lemma}\label{lem:bound_p_in_embed}
  Let $m\in\mathbb{N}$, and suppose that $\Sigma$, $\Gamma$ and $\Phi$ are finite alphabets of size at most $m$.
  Then there exists $r\in\mathbb{N}$, such that for any distribution $\mu$ over $\Sigma\times\Gamma\times \Phi$,
  if $\mu$ does not admit an Abelian embedding into $(\mathbb{Z},+)$, then $\mu$ does not admit any Abelian
  embedding to $(\mathbb{F}_p,+)$ for $p>r$.
\end{lemma}
\begin{proof}
  Write
  $\Sigma = \{x_1,\ldots,x_{m_1}\}$,
  $\Gamma = \{y_1,\ldots,y_{m_2}\}$
  and
  $\Phi = \{z_1,\ldots,z_{m_3}\}$,
  and associate with each symbol $\Sigma$ a variable $V_{x_i}$ and similarly for symbols in $\Gamma$ and $\Phi$.
  We will think of these variables as representing the values of an embedding of $\mu$, and so the conditions
  that they form an embedding can be written as the system of linear equations $V_x + V_y + V_z = 0$
  for all $(x,y,z)\in {\sf supp}(\mu)$.   Thus, the fact that there are no embedding of $\mu$ into $(\mathbb{Z},+)$
  is equivalent to the fact that over integers, the only solution to this system of equations are trivial constant
  (say, assigning all $x$-variables the value $17$, all $y$-variables the value $-6$, and all $z$-variables the value $-11$).
  We pick some $(x^{\star}, y^{\star}, z^{\star})\in {\sf supp}(\mu)$ and assign $V_{x^{\star}}$, $V_{y^{\star}}$ and
  $V_{z^{\star}}$ the value $0$, each solution of the original system corresponds (over any Abelian group) corresponds
  to a shift of a solution of the new system. Thus, the new system only has the trivial all $0$ solution over $(\mathbb{Z},+)$.

  We write this system in matrix form as $M V = 0$, and bring it to diagonal form without performing divisions. Namely, each time we pick a unpivoted equation
  from our system, then a variable from it, say $V_i$. We multiply all equations by appropriate constants so that the coefficients of $V_i$ in
  each equation is the same (or $0$ if $V_i$ doesn't appear in that equation), then subtract the chosen equation from all equations in which $V_i$ appears,
  and then declare the equation as pivoted. In the end of the process, we will end up with a system of equations of the form $c_i V_i = 0$ for $c_i\in\mathbb{Z}$,
  where all $c_i$ are bounded by a universal constant $C(m)$ (as the coefficients grow by at most a constant factor depending only on $m$ in each step,
  and the number of steps is at most $m^3$). We write this system as $M'V = 0$, and now the fact that there are no solutions over $\mathbb{Z}$ means that
  the rank of $M'$ is full, namely $3m$.

  Thus, for primes $p > C(m)$, we note that an $\mathbb{F}_p$ solution of the system is also a solution over integers. Indeed, for the equation $c_i V_i = 0\pmod{p}$
  to hold for $V_i\in\mathbb{F}_p$, it must be the case that either $V_i = 0$, or else $c_i$ must be divisible by $p$. It follows that there are no non-trivial
  embeddings over $(\mathbb{F}_p,+)$.
\end{proof}

Next, we discuss Abelian embeddings into cyclic groups of prime power order. Here, the situation is slightly trickier, as if $\mu$ has an embedding
into $(\mathbb{F}_p,+)$, then by multiplying it by $p^{k-1}$ one automatically gets an embedding into $(\mathbb{Z}_{p^{k}},+)$. In the following lemma,
we show that there is $k_0 = k_0(m)\in\mathbb{N}$ such that to ``exhaust'' all embeddings into groups of the form $(\mathbb{Z}_{p^k},+)$, it suffices
to consider $k\leq k_0$, in the sense that an embedding into $(\mathbb{Z}_{p^k},+)$ for $k>k_0$ is equivalent to an embedding into
$(\mathbb{Z}_{p^{k_0}},+)$.
\begin{lemma}\label{lem:bound_k_in_embed}
  Let $m\in\mathbb{N}$, and suppose that $\Sigma$, $\Gamma$ and $\Phi$ are finite alphabets of size at most $m$.
  Then there exists $p_0\in\mathbb{N}$ and $k_0\in\mathbb{N}$,
  such that for any distribution $\mu$ over $\Sigma\times\Gamma\times \Phi$ that doesn't admit any $(\mathbb{Z},+)$ embedding, if
  $\mu$ has a non-trivial Abelian embedding $(\sigma,\gamma,\phi)$ into $(\mathbb{Z}_{p^k},+)$,
  then $p\leq p_0$ and there is $k'\leq k_0$ and an Abelian embedding $(\sigma',\gamma',\phi')$ into $(\mathbb{Z}_{p^{k'}},+)$,
  such that $(\sigma',\gamma',\phi')$ is a linear reduction of $(\sigma,\gamma,\phi)$.
\end{lemma}
\begin{proof}
  We write a system of linear equations and bring it to a diagonal form as in the proof of Lemma~\ref{lem:bound_p_in_embed},
  and let $C(m)\in\mathbb{N}$ be an upper bound on the size of all of the coefficients there. We prove the statement for $p_0 = C(m)$
  and $k_0 = \lceil \log C(m)\rceil$. If $k\leq k_0$ we are done, so assume otherwise.

  By applying appropriate affine shift, we can assume that $\sigma(x^{\star}) = \gamma(y^{\star}) = \phi(z^{\star}) = 0$.
  Then the fact that $p\leq p_0$ for some $p_0 = O_{m}(1)$ follows from the argument in Lemma~\ref{lem:bound_p_in_embed}, and we next argue about $k$.
  Consider an equation $c_i V_i = 0 \pmod{p^k}$ therein, and assume that $k>k_0$. Then $p^{k}$ cannot divide $c_i$, so we may write $c_i = p^{a_i} c_i'$
  where $a_i<k_0$ and $c_i'$ is relatively prime to $p$. This means that $V_i = 0 \pmod{p^{k-a_i}}$, and so $p$ divides $V_i$. This means that
  all values of $\sigma,\gamma,\phi$ are divisible by $p^{k-k_0}$, hence we may look at $\sigma' = \sigma/p^{k-k_0}$, $\gamma' = \gamma/p^{k-k_0}$
  and $\phi' = \phi/p^{k-k_0}$ and get that $(\sigma',\gamma',\phi')$ form an embedding of $\mu$ into $(\mathbb{Z}_{p^{k_0}}, +)$, as required.
\end{proof}

\subsubsection{Finding Small Equivalent Embeddings on Finite Abelian Groups}
With Lemmas~\ref{lem:bound_p_in_embed},~\ref{lem:bound_k_in_embed} in hand, we can now address general Abelian embeddings,
and show that any Abelian embedding is equivalent to
an Abelian embedding into a group of bounded size.
\begin{lemma}\label{lem:exhaust_embed}
  For all $m\in \mathbb{N}$ there is $r$ such that for alphabets $\Sigma$, $\Gamma$, $\Phi$ of size at most $m$,
  if $\mu$ is a distribution over $\Sigma\times\Gamma\times \Phi$ that does not admit non-trivial Abelian embeddings
  over $(\mathbb{Z},+)$, and $(\sigma,\gamma,\phi)$ is an Abelian embedding of $\mu$,
  then there is an Abelian embedding $(\sigma',\gamma',\phi')$ into a group of size at most $r$
  which is a linear reduction of $(\sigma,\gamma,\phi)$.
\end{lemma}
\begin{proof}
  Let $\sigma$, $\gamma$ and $\phi$ be an Abelian embedding of $\mu$ into an Abelian group $(H,+)$.
  By the fundamental theorem of finite Abelian groups there are primes $p_1,\ldots,p_k$ and integers $r_1,\ldots,r_k\geq 1$
  such that
  \[
  H = \prod\limits_{i=1}^{k}\mathbb{Z}_{p_i^{r_i}},
  \]
  and with this identification we can
  write $\sigma\colon \Sigma \to H$ as $\sigma = (\sigma_1,\ldots,\sigma_k)$ where $\sigma_i\colon\Sigma\to \mathbb{Z}_{p_i^{r_i}}$,
  and similarly write $\gamma = (\gamma_1,\ldots,\gamma_k)$ and $\phi = (\phi_1,\ldots,\phi_k)$. We assume each $(\sigma_i,\gamma_i,\phi_i)$
  is non-trivial, otherwise we may drop it altogether.

  Note that each $(\sigma_i,\gamma_i,\phi_i)$ forms a cyclic Abelian embedding of $\mu$, so by Lemma~\ref{lem:bound_p_in_embed} is follows that $p_i\leq O_m(1)$ for all $i$.
  By Lemma~\ref{lem:bound_k_in_embed} we get that $(\sigma_i,\gamma_i,\phi_i)$ is equivalent to an Abelian embedding $(\sigma_i',\gamma_i',\phi_i')$ into $\mathbb{Z}_{p_i^{r_i'}}$ for $r_i'\leq O_m(1)$.
  Thus $(\sigma,\gamma,\phi)$ is equivalent to $(\sigma',\gamma',\phi')$ where
  $\sigma' = (\sigma_1',\ldots,\sigma_k')$, $\gamma' = (\gamma_1',\ldots,\gamma_k')$ and $\phi' = (\phi_1',\ldots,\phi_k')$.
  Now, $(\sigma',\gamma',\phi')$ is an embedding of $\mu$ into $H' = \prod\limits_{i\leq k}\mathbb{Z}_{p_i^{r_i'}}$. In other words,
  we managed to reduce each $p_i^{r_i}$ to $p_i'$ that is bounded and get an equivalent embedding, and to simplify notation we
  drop the primes and assume the embeddings $\sigma, \phi$ and $\phi$ as well as $H$ are of this form to begin with.

  Next, we clean up redundancies. We say a coordinate $i$ is $\Sigma$-redundant if the partition of $\Sigma$ induced by
  $\sigma = (\sigma_1,\ldots,\sigma_k)$ is the same as the partition induced by $\sigma_{-i} = (\sigma_1,\ldots,\sigma_{i-1},\sigma_{i+1},\ldots,\sigma_k)$;
  similarly we define the notions of $\Gamma$-redundant and $\Phi$-redundant coordinates. Note that if $i$ is $\Sigma$-redundant, $\Gamma$-redundant
  and $\Phi$-redundant, then $(\sigma,\gamma,\phi)$ is equivalent to $(\sigma_{-i},\gamma_{-i},\phi_{-i})$. Also note that there are at most $\card{\Sigma}$
  non $\Sigma$-redundant coordinates, as well as at most $\card{\Gamma}$ non $\Gamma$-redundant coordinates and at most $\card{\Phi}$ non $\Phi$-redundant coordinates.
  Thus, we can eliminate all but at most $k' = \card{\Sigma}+\card{\Gamma}+\card{\Phi}$ of the coordinates using this process, and thus get an embedding
  into a group of size at most $O_m(1)^{k'}$, which is a size depending only on the alphabet size, as required.
\end{proof}

\subsubsection{The Master Embedding}
In this section we formally define the notion of master embeddings, which is crucial for our arguments.
We begin by formally defining the notion of master embeddings.

\begin{definition}
  Let $\Sigma$, $\Gamma$ and $\Phi$ be finite alphabets, let $\mu$ be a distribution over $\Sigma\times \Gamma\times \Phi$,
   let $(H,+) = \prod_{i=1}^{s} (H_i,+)$ be an Abelian group and let $\sigma\colon \Sigma\to H$, $\gamma\colon \Gamma\to H$ and $\phi\colon \Phi\to H$ be an Abelian embedding
  of $\mu$. We say that $(\sigma,\gamma,\phi)$ is a master embedding if any Abelian embedding $(\sigma',\gamma',\phi')$ of $\mu$
  there is an $i\in \{1,\ldots,s\}$, such that $(\sigma_i,\gamma_i,\phi_i)$ is a linear refinement of $(\sigma',\gamma',\phi')$.
\end{definition}

Using Lemma~\ref{lem:exhaust_embed}, we can construct a master embedding for $\mu$ as follows. Choose $r$ as in the lemma therein and consider all embeddings $(\sigma_i,\gamma_i,\phi_i)$
of $\mu$ into Abelian groups of size at most $r$ for $i=1,\ldots,s$; and note that $r$ and $s$ are some finite numbers depending only on the alphabet sizes.
We can thus define the master embedding as follows:
\begin{definition}\label{def:master_embed}
  For all finite alphabets $\Sigma$, $\Gamma$ and $\Phi$  take $r\in\mathbb{N}$ from Lemma~\ref{lem:exhaust_embed} and
  $k$ to be the number of Abelian groups of size at most $r$.
  For a distribution $\mu$ over $\Sigma\times \Gamma\times \Phi$, there are at most $k$ Abelian embeddings of $\mu$
  into Abelian groups of size at most $r$, and letting $(\sigma_1,\gamma_1,\phi_1),\ldots,(\sigma_s,\gamma_s,\phi_s)$
  into $(H_1,+),\ldots,(H_s,+)$ be an enumeration of all of these embeddings, we define the embedding
  \[
  \sigma_{{\sf master}}(x) = (\sigma_1(x),\ldots,\sigma_s(x)),
  ~~
  \gamma_{{\sf master}}(y) = (\gamma_1(y),\ldots,\gamma_s(y)),
  ~~
  \phi_{{\sf master}}(z) = (\phi_1(z),\ldots,\phi_s(z)),
  \]
  into $(H,+) = (\prod\limits_{i=1}^{s}H_i,+)$.
\end{definition}

In the following lemma, we prove that $(\sigma_{{\sf master}}, \gamma_{{\sf master}}, \phi_{{\sf master}})$ is a master embedding of $\mu$.
We remark though that it is not necessarily unique.
\begin{lemma}\label{lem:master_captures_all}
 We have that $(\sigma_{{\sf master}}, \gamma_{{\sf master}}, \phi_{{\sf master}})$ is a master embedding of $\mu$.
\end{lemma}
\begin{proof}
  By Lemma~\ref{lem:exhaust_embed}, for any Abelian embedding $(\sigma,\gamma,\phi)$ of $\mu$ there is an Abelian embedding $(\sigma',\gamma',\phi')$
  into an Abelian group of size at most $r$ such that $(\sigma',\gamma',\phi')$ is a linear refinement of $(\sigma,\gamma,\phi)$,
  and the result follows from the definition of $(\sigma_{{\sf master}}, \gamma_{{\sf master}}, \phi_{{\sf master}})$ (as it includes all such embeddings).
\end{proof}
For technical reasons, it will be convenient for us to assume that for each $(\sigma_i,\gamma_i,\phi_i)$ in the definition of the master embedding, we have
that $0$ is in the image of each one of them; this can easily be arranged by a proper affine shift. We assume henceforth
that all Abelian embeddings we are dealing with have $0$ in their image.

\subsection{The Path Trick}
In this section, we formally define the path trick from~\cite{BKMcsp1,BKMcsp2}, and recall some basic properties of it from these works. We then begin discussing
the interactions between the path trick and Abelian embeddings, and prove that, in a sense, the path trick preserves the structure of Abelian embeddings.

\subsubsection{The Definition of the Path Trick Distribution}
Suppose $\Sigma$, $\Gamma$ and $\Phi$ are finite alphabets, and $\mu$ is a distribution over $\Sigma\times \Gamma \times \Phi$. Below, we think
of $(x,y,z)\sim \mu$. The path trick of length $\ell = 2^t-1$ with respect to $x$ is a distribution $\mu_{\ell}$ over $\Sigma'\times \Gamma\times \Phi$
where $\Sigma'\subseteq \Sigma^{\ell}$, and there are several equivalent ways of defining it. Below we present the two ways we use: one of them
will be more intuitive to think about (namely, the path definition), whereas the other one will be more convenient to work with when we apply Cauchy-Schwarz.

\paragraph{The Path Definition.}
For the path definition, consider the bipartite graph $G = (\Gamma\cup \Phi, E)$ wherein the edges of the graph
are between $(y,z)$ for which there is an $x\in\Sigma$ such that $(x,y,z)\in {\sf supp}(\mu)$. The edges of the graph
are labeled by the $x$ that produced them, as well as weighted according to the distribution $\mu$; we allow for parallel edges.

With this in mind, a sample from the path trick distribution $\mu_{\ell}$ is generated as follows:
sample a starting point $y\sim \mu_y$ and then proceed by taking a random walk of length $\ell$, collecting
all the labels $x$ on edges encountered during this walk. Thus, in the end one has the labels $x$, $\vec{x} = (x_1,\ldots,x_{\ell})$ of edges
encountered, the starting point of the walk $y\in \Gamma$, and the endpoint of the walk $z\in \Phi$. The output of the process
is $(\vec{x}, y, z)$.

\paragraph{The Inductive Definition.}
For the inductive definition it is more convenient to define it for $\ell$'s that are power of $2$ rather than powers of $2$ minus $1$.
We define $\nu_2$ as the following distribution over $\Sigma'\times \Gamma'\times \Phi$ where $\Sigma'\subseteq \Sigma^2$, $\Gamma'\subseteq \Gamma^2$.
\begin{enumerate}
  \item Sample $z\sim \mu_z$.
  \item Sample $(x_1,y_1,z_1)$ and $(x_2,y_2,z_2)$ according to $\mu$ conditioned on $z_1 = z$ and $z_2 = z$.
  \item Output $(x_1,x_2)$, $(y_1,y_2)$ and $z$.
\end{enumerate}
Note that the path $y_1\rightarrow z\rightarrow y_2$ with the labels $x_1$ and $x_2$ corresponds to a path as sampled in the previous definition (of length $2$).
Once $\nu_{\ell}$ has been defined, we define $\nu_{2\ell}$ as:
\begin{enumerate}
  \item Sample $y\sim \mu_y$.
  \item Sample $(x_1,\ldots,x_{\ell}), (y_1,\ldots,y_{\ell})$ and $(z_1,\ldots,z_{\ell-1})$ according to $\nu_{\ell}$ conditioned on $y_1 = y$.
  \item Independently sample $(x_1',\ldots,x_{\ell}'), (y_1',\ldots,y_{\ell}')$ and $(z_1',\ldots,z_{\ell-1}')$ according to $\nu_{\ell}$ conditioned on $y_1' = y$.
  \item Output $(x_{\ell},x_{\ell-1},\ldots,x_1,x_1',x_2',\ldots,x_{\ell}')$, $(y_{\ell},\ldots,y_1, y_2',y_3',\ldots,y_{\ell}')$
  and $(z_{\ell-1},\ldots,z_1,z_1',\ldots,z_{\ell-1}')$.
\end{enumerate}
In words, we sample $y\sim \mu_y$, take two independent walks of length $\ell$ starting from it, and then concatenate them.
The distribution $\tilde{\mu}_{\ell}$ is naturally derived from the distribution $\nu_{\ell}$ by just picking the end-points of the $y$ and $z$,
which in the notations above are $y_{\ell}$ and $z_{\ell-1}'$.

The following lemma from~\cite{BKMcsp2} shows that the two ways of describing the path trick distribution are equivalent.
For the sake of completeness we give a sketch of the proof.
\begin{lemma}
  For all $\ell = 2^{t}-1$, it holds that $\tilde{\mu}_{\ell} = \mu_{\ell}$.
\end{lemma}
\begin{proof}
  We prove by induction that the distribution $\tilde{\mu}_{\ell}$ is over $(\vec{x}, y, z)$ wherein $y\sim \mu_y$, and conditioned on
  that $\vec{x}$ and $z$ are a result of a random walk as in the path definition. Indeed, for $\ell=1$ and $\ell=2$ this is clear
  by inspection.

  Assume the statement for $\ell$ and prove for $2\ell$; note that the process above may be thought of as first sampling a
  midpoint $y$, generate from it two paths of length $\ell$ independently, concatenate them and then chop off the final
  step in the path (so as that one ends up in the $\Phi$-part as opposed to the $\Gamma$-part). Note that for every length $\ell$,
  the distribution over paths from $\Gamma$ to $\Gamma$ of length $\ell$ is invariant under path reversals, and
  that the distribution over paths of length $r_1+r_2$ is the concatenation of a distribution of a path of length $r_1$ with
  a path of length $r_2$ conditioned on the endpoint of the first path being the starting point of the second path. Thus,
  the distribution we end up with can be viewed as ${\sf reverse}(p_1)\circ p_2$ where $p_1,p_2$ are path of length $\ell$
  from $\Gamma$ to $\Gamma$ conditioned on them starting at the first point. Thus, the distribution of ${\sf reverse}(p_1)$
  and $p_2$ are of length $\ell$ conditioned on $p_2$ starting where ${\sf reverse}(p_1)$ ended.
\end{proof}

\subsubsection{The Path Trick Distribution and $3$-wise Correlations}
One important feature of the path trick distribution that was already explored in~\cite{BKMcsp1,BKMcsp2} is that it allows one to reduce upper bound
powers of the $3$-wise correlations as in Theorem~\ref{thm:main_stab_3} with respect to $\mu$ by $3$-wise correlations with respect to $\mu_{\ell}$.
Below is a formal statement.
\begin{lemma}\label{lem:from_mu_to_path}
  For all $t\geq 0$, and all $1$-bounded functions
  $f\colon \Sigma^{n}\to \mathbb{C}$,
  $g\colon \Gamma^n\to\mathbb{C}$
  and $h\colon \Phi^n\to \mathbb{C}$
  we have that
  \[
  \card{\Expect{(x,y,z)\sim \mu}{f(x)g(y)h(z)}}^{2^{t}}
  \leq
  \card{\Expect{(\vec{x},y,z)\sim \mu_{2^t}}{F(\vec{x})g(y)h'(z)}},
  \]
  where $F(\vec{x}) = f(x_{2^t-1})\prod\limits_{j=1}^{2^{t-1}-1} f(x_{2j-1})\overline{f(x_{2j})} $
  and $h'(z) = \cExpect{(x',y',z')\sim \mu}{z'=z}{\overline{f(x')}\overline{g(y')}}$.
\end{lemma}
\begin{proof}
  First, we note that
  \[
  \card{\Expect{(x,y,z)\sim \mu}{f(x)g(y)h(z)}}^{2}
  =\card{\Expect{z}{h(z)\overline{h'}(z)}}^2
  \leq \norm{h}_2^2\norm{h'}_2^2
  \leq \norm{h'}_2^2
  =\card{\Expect{(x,y,z)\sim \mu}{f(x)g(y)h'(z)}},
  \]
  which proves the statement for $t=1$, in which case the function $F$ is just $f$.
  We now proceed by a sequence of Cauchy-Schwarz inequalities. To illustrate, note that
  \begin{align*}
  \card{\Expect{(x,y,z)\sim \mu}{f(x)g(y)h'(z)}}^2
  &=\card{\Expect{z'}{h'(z')\cExpect{(x,y,z)\sim \mu}{z=z'}{f(x)g(y)}}}^2\\
  &\leq \Expect{(x_1,x_2), (y_1,y_2), z \sim \nu_2}{f(x_1)\overline{f(x_2)} g(y_1) \overline{g(y_2)}}.
  \end{align*}
  Thus,
  \begin{align*}
  \card{\Expect{(x,y,z)\sim \mu}{f(x)g(y)h'(z)}}^4
  &\leq \Expect{(x_1,x_2), (y_1,y_2), z \sim \nu_2}{f(x_1)\overline{f(x_2)} g(y_1) \overline{g(y_2)}}^2\\
  &=\Expect{y_2'\sim \mu_y}{\overline{g(y_2)}\cExpect{(x_1,x_2), (y_1,y_2), z \sim \nu_2}{y_2 = y_2'}{f(x_1)\overline{f(x_2)} g(y_1)}}^2\\
  &\leq \Expect{y_2'\sim \mu_y}{\card{\cExpect{(x_1,x_2), (y_1,y_2), z \sim \nu_2}{y_2 = y_2'}{f(x_1)\overline{f(x_2)} g(y_1)}}^2}\\
  &=\Expect{\vec{x},\vec{y}\sim \nu_4}{F(x_1,x_2,x_3,x_4) g(y_1)\overline{g(y_3)}},
  \end{align*}
  and inductively we get that
  \[
  \card{\Expect{(x,y,z)\sim \mu}{f(x)g(y)h'(z)}}^{2^{t}}
  \leq
  \Expect{\vec{x},\vec{y}, \vec{z}\sim \nu_{2^t}}{F'(\vec{x}) g(y_1)\overline{g(y_{2^{t-1}+1})}}
  \]
  where $F'(\vec{x}) = \prod\limits_{j=1}^{2^{t-1}} f(x_{2j-1})\overline{f(x_{2j})}$ is given as in the statement of the lemma. To finish the proof, note that
  \begin{align*}
  \Expect{\vec{x},\vec{y}, \vec{z}\sim \nu_{2^t}}{F'(\vec{x}) g(y_1)\overline{g(y_{2^t})}}
  &=
  \Expect{\vec{x},\vec{y}, \vec{z}\sim \nu_{2^{t}}}{F(\vec{x}) g(y_1)\cExpect{x_{2^t}, y_{2^{t-1}}}{z_{2^{t-1}}}{\overline{f(x_{2^t})}\overline{g(y_{2^{t-1}+1})}}}\\
  &=
  \Expect{\vec{x},y, z\sim \mu_{2^{t}}}{F(\vec{x}) g(y_1)h'(z_{2^{t-1}})},
  \end{align*}
  where in the last transition we used the definition of $h'$.
\end{proof}

Some remarks regarding Lemma~\ref{lem:from_mu_to_path} that are essential to make an effective use of it, are in order. Lemma~\ref{lem:from_mu_to_path}
allows us to relate upper bound correlations with respect to $\mu$ with correlations with respect to $\mu_{\ell}$, but this comes at the expense of several complications:
\begin{enumerate}
  \item \textbf{The role of $y$ and $z$.}
  Note that in the above lemma, the role played by the $y$ and $z$ function in the premise are interchangeable, but this is not the case in the
  conclusion of the lemma. Namely, in the above formulation we kept the $y$ function to be the same but replaced the $z$-function from $h$ to
  $h'$. It is possible however, and this will be important for us, that we might as well have kept the $h$ function in place and changed the function
  $g$ into a function $g'$ analogously to the above.
  \item \textbf{Applying the path-trick with respect to other directions.} We could also apply the path trick with respect to $y$ (or $z$) instead
  of on $x$ as in the above formulation, in which case one gets an analogous statement to the above. It is always the case that the variable we apply
  the path trick on results in the function on that variable becoming a more complicated ``product version'' of the previous function.
  As for the other two functions, one of them stays put, whereas on the last one we have essentially no control
  over and it may change altogether. As explained above, there is flexibility for us in the choice of which function stays and on which we lose control
  over, and we will utilize this.
  \item
  \textbf{The relationship between $f$ and $F$.}
  The function on $x$, $f$, becomes a more complicated function $F$. In our argument for the proof of Theorem~\ref{thm:main_stab_3} we will
  make use of this lemma several times, and on each invocation the $x$-function will get more and more complicated, until eventually we will reach
  a distribution $\mu'$ on which we will do a direct analysis. At that point, we will be able to conclude a structural result on the $x$-function
  in that correlation, and from that deduce a structural result regarding $f$. In that respect, it is important to keep in mind the relationship
  between the $x$-functions; at each invocation of the path trick we will either keep the $x$ function the same (this is the role played by
  the $y$-function in the formulation above), or else it will be a multiplication of the previous function on several inputs that are correlated
  in some way.
\end{enumerate}

\subsubsection{The Structure of Embeddings after the Path Trick}
In~\cite{BKMcsp1,BKMcsp2} it is shown that if $\mu$ does not admit an Abelian embedding, then the path trick distribution $\mu_{\ell}$ also does not
admit any Abelian embedding. In this section, we extend this connection by observing that, in a sense, even in the presence of Abelian embeddings,
the path trick distribution has the ``same'' Abelian embeddings as the original distribution. More precisely:
\begin{lemma}\label{lem:reverse_embed}
  Let $\Sigma$, $\Gamma$ and $\Phi$ be finite alphabets, let $\mu$ be a distribution over $\Sigma\times \Gamma\times \Phi$, let $\ell\in \mathbb{N}$ be odd
  and let $\mu_{\ell}$ be the path trick distribution applied on $\mu$ with respect to $x$, which is a distribution over $\Sigma'\times \Gamma\times \Phi$
  where $\Sigma'\subseteq \Sigma^{\ell}$.
  If $\sigma\colon \Sigma'\to (H,+)$, $\gamma\colon \Gamma\to (H,+)$ and $\phi\colon \Phi\to (H,+)$ is an Abelian embedding
  of $\mu_{\ell}$ in $(H,+)$, then there exists an $\sigma'\colon \Sigma\to (H,+)$ such that
  \begin{enumerate}
    \item $(\sigma', \gamma,\phi)$ form an embedding of $\mu$ in $(H,+)$.
    \item For all $(x_1,\ldots,x_{\ell})\in \Sigma'$ we have that
    \[
    \sigma(x_1,\ldots,x_{\ell})
    =\sum\limits_{j=1}^{\ell}(-1)^{j-1}\sigma'(x_j).
    \]
  \end{enumerate}
\end{lemma}
\begin{proof}
  Consider a sample $(\vec{x}, \vec{y}, \vec{z})$ as in $\nu_{\ell}$; that is, we consider a sample according to $\mu_{\ell}$
  but record all the $y$'s and $z$'s generated in the process, so that
  $(x_{2j-1}, y_j, z_j)$ and $(x_{2j}, y_{j+1}, z_j)$ are in ${\sf supp}(\mu)$ for all $j$.

  Let $v_j\in \Sigma^{\ell}$ be the vector whose coordinates are all equal to $x_{2j-1}$,
  and $u_j\in \Sigma^{\ell}$ be the vector whose coordinates are all $x_{2j}$.
  Note that as $(v_j,y_j,z_j)\in {\sf supp}(\mu_{\ell})$ (as there is a path from $y_j$ to $z_j$ of length $\ell$ that just goes
  back and forth on the edge labeled by $x_{2j-1}$) and $(u_j, y_{j+1}, z_j)\in {\sf supp}(\mu_{\ell})$ (for similar reasons),
  it follows that
  \begin{equation}\label{eq:3}
  \sigma(v_j) + \gamma(y_j) + \phi(z_j) = 0,
  \qquad
  \sigma(u_j) + \gamma(y_{j+1}) + \phi(z_j) = 0
  \end{equation}
  for all $j$. Denote $\sigma'(x) = \sigma(x,x,\ldots,x)$, and note that $(\sigma',\gamma,\phi)$ form an Abelian embedding of $\mu$ into $(H,+)$.
  In the rest of the proof, we show that the formula in the second item of the statement holds.
  For that, we look at~\eqref{eq:3} for all $j = 1,\ldots,(\ell+1)/2$, multiply the left equations by $1$ and the second one by $-1$ and add up to get that
  \[
  \sum\limits_{j=1}^{\ell}(-1)^{j-1}\sigma'(x_j)
  +\sigma(y) + \sigma(z) = 0,
  \]
  where $y = y_1$ is the starting point of the walk and $z = z_{(\ell+1)/2}$ is the end point of the walk.
  We also have that $\sigma(\vec{x}) + \gamma(y) + \phi(z) = 0$ as $(\sigma,\gamma,\phi)$ form an embedding of
  $\mu_{\ell}$, so we get that
  \[
  \sigma(\vec{x})
  =\sum\limits_{j=1}^{\ell}(-1)^{j-1}\sigma'(x_j),
  \]
  finishing the proof.
\end{proof}

Lemma~\ref{lem:reverse_embed} has a few important consequences.
Later on in Section~\ref{sec:unraveling}, we will use it in order to convert structural results for $F$ to structural results for $f$.
\footnote{Indeed, the above formula suggests that if $F$ is correlated with an embedding function on $\mu_{\ell}$,
then the embedding function can be broken to a product function over the $x_i$'s just the same way as $F$ can be, and one thus expects to be able to argue
that $f$ itself is correlated with an embedding function.
}For now though, we shall use Lemma~\ref{lem:reverse_embed} in a different way in order to achieve the ``master embedding is saturated'' property as explained earlier.
Towards this end, we introduce the following convenient notation.
\begin{definition}\label{def:master}
  For $\sigma\colon \Sigma\to (H,+)$, we define $\sigma^{\sharp\ell}\colon \Sigma^{\ell}\to(H,+)$ by
  \[
  \sigma^{\sharp\ell}(x_1,\ldots,x_{\ell}) = \sum\limits_{j=1}^{\ell}(-1)^{j+1}\sigma(x_j).
  \]
\end{definition}
Often times, $\sigma$ will be part of some embedding of $\mu$, in which case it will be more natural to view the domain of $\sigma^{\sharp\ell}$
as $\Sigma'\subseteq \Sigma^{\ell}$ where $\Sigma'$ is the support of the first coordinate of $\mu_{\ell}$.

\subsection{Saturating the Master Embedding via the Path Trick}
Recalling the definition of the master embedding from Definition~\ref{def:master_embed}, there is no reason that its image would be the entire group $(H,+)$
(and in fact, typically it would not be). The goal of this section is to use the path trick to move from the distribution $\mu$ to a related distribution
$\mu'$ such that the master embedding in $\mu'$ has as image a group. We moreover assert that it ``suffices'' to prove Theorem~\ref{thm:main_stab_3} for $\mu'$,
in the sense that then we would be able to deduce it for $\mu$.
For the sake of this section, we will focus on the first point -- namely that the master embeddings have full images, and
ignore the second point for now; this last deduction is covered in Section~\ref{sec:unraveling}.

\subsubsection{The Evolution of the Master Embeddings under Path Tricks}
The following lemma explains the way the master embeddings of a distribution evolve after an application of the path trick.
\begin{lemma}\label{lem:master_stays}
  Let $\Sigma$, $\Gamma$ and $\Phi$ be finite alphabets, let $\mu$ be a distribution over $\Sigma\times \Gamma\times \Phi$, let $\ell\in \mathbb{N}$ be odd
  and let $\mu_{\ell}$ be the path trick distribution applied on $\mu$ with respect to $x$, which is a distribution over $\Sigma'\times \Gamma\times \Phi$
  where $\Sigma'\subseteq \Sigma^{\ell}$.

  If $\sigma_{{\sf master}},\gamma_{{\sf master}},\phi_{{\sf master}}$ is a master embedding with respect to $\mu$, then $\sigma_{{\sf master}}^{\sharp\ell}$, $\gamma_{{\sf master}}$, $\phi_{{\sf master}}$ is a master embedding for $\mu_{\ell}$.
\end{lemma}
\begin{proof}
Let $(\sigma',\gamma',\phi')$ be an Abelian embedding of $\mu_{\ell}$. By Lemma~\ref{lem:reverse_embed} we can find
an embedding $(\sigma'',\gamma',\phi')$ of $\mu$ where $\sigma' = \sigma''^{\sharp \ell}$ on the domain of $\sigma'$.
By Lemma~\ref{lem:master_captures_all}, there is an $i$ such that
$(\sigma_{{\sf master},i},\gamma_{{\sf master},i},\phi_{{\sf master},i})$ is a linear refinement of
$(\sigma'',\gamma',\phi')$, and it follows that $(\sigma_{{\sf master},i}^{\sharp \ell},\gamma_{{\sf master},i},\phi_{{\sf master},i})$
is a linear refinement of $(\sigma''^{\sharp \ell},\gamma',\phi') = (\sigma',\gamma',\phi')$, as required.
\end{proof}
In words, Lemma~\ref{lem:master_stays} says that once we have a master embedding for a distribution and we apply the natural operations on it,
then it is also a master embedding for $\mu_{\ell}$. Hence, we are not losing anything with respect to the master embedding while performing
path tricks.

\subsubsection{The Path Trick Preserves Pairwise Connectedness}
The following simple lemma shows that path tricks also preserve pairwise connectedness, and in fact improve pairwise connectedness with respect to
$2$ of the coordinates.
\begin{lemma}\label{lem:path_keeps_connectedness}
  Let $\Sigma$, $\Gamma$ and $\Phi$ be finite alphabets, let $\mu$ be a distribution over $\Sigma\times \Gamma\times \Phi$, let $\ell\in \mathbb{N}$ be odd
  and let $\mu_{\ell}$ be the path trick distribution applied on $\mu$ with respect to $x$, which is a distribution over $\Sigma'\times \Gamma\times \Phi$
  where $\Sigma'\subseteq \Sigma^{\ell}$.

  If $\mu$ is pairwise connected, then $\mu_{\ell}$ is pairwise connected. Furthermore, for large enough $\ell$ depending only on
  the alphabet sizes, the graph between $\Gamma$ and $\Phi$ becomes complete.
\end{lemma}
\begin{proof}
  First, note that for every $x\in \Sigma$ and $y\in \Gamma$ such that there is $z\in \Phi$ for which $(x,y,z)$ is in the support of $\mu$,
  it holds that $(\vec{x},y,z)$ is in the support of $\mu_{\ell}$, where $\vec{x} = (x,\ldots,x)$. Thus, by connectedness of $\mu$ it follows
  that there is a path between any $\vec{x}$ and $y\in\Gamma$ in the graph of $\mu_{\ell}$, and as any $\vec{x}'\in \Sigma'$ is connected to
  some $y\in \Gamma$, it follows that the $\Sigma'$, $\Gamma$ graph in $\mu_{\ell}$ is connected. Analogously, we have that the $\Sigma'$,
  $\Phi$ graph is also connected.

  As for the $\Gamma, \Phi$ graph, note that it is connected in $\mu$, hence there is an odd number $w$ such that between any $y\in \Gamma$ and
  $z\in\Phi$ there is a path of length at most $w$ (of odd length). Note that in that case there will also be a path of length exactly $w$
  (as one can always do steps that go from $z$ to some neighbour of it and back), hence we get that for $\ell = w$ the graph between
  $\Gamma$ and $\Phi$ in $\mu_{\ell}$ is complete.
\end{proof}

\subsubsection{The Path Trick Helps in Saturating the Master Embeddings}
It is easy to observe that if $\ell$ is odd, then the image of $\sigma^{\sharp\ell}$ always contains the image of $\sigma$
(by considering inputs of the form $(x,\ldots,x)$). Additionally, the image of the other two components remains the same (and in particular does
not shrink). In light of the formula in Definition~\ref{def:master}, intuition suggests that as long as the image of $\sigma$ is not a subgroup,
the image of $\sigma^{\sharp\ell}$ would have to be larger (as we are considering signed sums). In this context, once the image becomes a sub-group it can
never further increase. A natural hypothesis therefore would be that by applying enough path tricks, the image of $\sigma$ would eventually have to
become a sub-group, in which point we will refer to $\sigma$ as \emph{saturated}.

Strictly speaking, this does not have to be the case if one applies the path trick in the naive way. Nevertheless, if one is willing to apply
alternating path tricks on all coordinates, then eventually $\sigma$ does become saturated. To show that, we first establish the following lemma,
asserting that if the image of $\sigma$ is not a sub-group, then one may enlarge it by applying the path trick $3$ times:
\begin{lemma}\label{lem:path_trick_saturates}
  Let $\Sigma$, $\Gamma$ and $\Phi$ be finite alphabets, let $\mu$ be a pairwise connected distribution over $\Sigma\times \Gamma\times \Phi$.
  Then there are constants $\ell_1,\ell_2$ and $\ell_3=3$ depending only on the alphabet sizes such that if $(\sigma,\gamma,\phi)$
  is a master embedding for $\mu$ into $(H,+)$, wherein each one of ${\sf Image}(\sigma)$, ${\sf Image}(\gamma)$, ${\sf Image}(\phi)$
  contains $0$, then the following holds. Consider the distributions:
  \begin{enumerate}
    \item $\mu'$ which is the result of the application of the path trick on $\mu$ with respect to $z$ of length $\ell_1$,
    producing the master embedding $(\sigma',\gamma',\phi') = (\sigma,\gamma,\phi^{\sharp \ell_1})$.
    \item $\mu''$ which is the result of the application of the path trick on $\mu'$ with respect to $y$ of length $\ell_2$,
    producing the master embedding $(\sigma'',\gamma'',\phi'') = (\sigma',\gamma^{\sharp \ell_2},\phi')$.
    \item $\mu'''$ which is a result of the application of the path trick on $\mu''$ with respect to $x$ of length $\ell_3$,
    producing the master embedding $(\sigma''',\gamma''',\phi''') = (\sigma''^{\sharp\ell_3}, \gamma'', \phi'')$.
  \end{enumerate}
  then if ${\sf Image}(\sigma)$ is not a subgroup of $H$, then ${\sf Image}(\sigma)\subsetneq {\sf Image}(\sigma''')$. Furthermore,
  the support of $\mu'''$ on its first coordinate (namely, its $x$-coordinate) is full, that is, $\Sigma^{\ell_3}$.
\end{lemma}
\begin{proof}
  Assume that ${\sf Image}(\sigma)$ is not a subgroup of $H$; then it is not closed under addition, so there are $x_1,x_2\in \Sigma$ such that
  $\sigma(x_1) + \sigma(x_2)\not \in {\sf Image}(\sigma)$. We will show that after we choose $\ell_1$ and $\ell_2$ appropriately large, we could
  take $\ell_3  = 3$ and get that the support on $\Sigma^3$ is full, hence taking $x_1$, $x_2$ and some $x^{\star}\in\Sigma$ for which $\sigma(x^{\star}) = 0$,
  we would get that $\sigma'''(x_1,x^{\star},x_2) = \sigma(x_1) - \sigma(x^{\star}) + \sigma(x_2)$ is in ${\sf Image}(\sigma''')$ and not in ${\sf Image}(\sigma)$.

  Denote the alphabets in question as: $\mu'$ is a distribution over $\Sigma\times\Gamma\times \Phi'$,
  $\mu''$ is a distribution over $\Sigma\times\Gamma'\times\Phi'$ and $\mu''$ is a distribution over
  $\Sigma'\times\Gamma'\times\Phi'$.
  By Lemma~\ref{lem:path_keeps_connectedness}, we may take $\ell_1$ large enough so that the support of $\mu'$ on $\Sigma\times \Gamma$ is full,
  and fixing $\ell_1$ we may choose $\ell_2$ large enough so that the support of $\mu''$ on $\Sigma\times \Phi'$ is full. We then pick $\ell_3 = 3$.

  We show that any $(x,x',x'')\in \Sigma^3$ is in $\Sigma'$. Pick some $y\in \Gamma$ and look at $\vec{y}\in \Gamma^{\ell_2}$ which has all of its coordinates
  equal to $y$, and note that it is in $\Gamma'$. Since the support of $\mu'$ on $\Sigma\times \Gamma$ is full, we get that there are some $\vec{z},\vec{z}'\in \Phi'$
  such that $(x,y,\vec{z})$ and $(x',y, \vec{z}')$ are in the support of $\mu'$. It follows that $(x,\vec{y},\vec{z})$ and $(x',\vec{y},\vec{z}')$ are both in
  the support of $\mu''$, so in the graph of $\mu''$ we get a path
  from $\vec{z}$ to $\vec{y}$ to $\vec{z}'$ labeled by $x,x'$. As the support of $\Sigma\times \Phi'$ in $\mu''$ is full, we may continue this path
  by an edge labeled by $x''$ to get to some $\vec{y}'$. Overall, we get a path from $\vec{z}$ to $\vec{y}$ to $\vec{z}'$ to $\vec{y}'$ labeled by
  $x,x',x''$, and this path means that $(x,x',x'')$ is in the support of $\mu'''$.
\end{proof}

Using Lemma~\ref{lem:path_trick_saturates} iteratively, as long as our master embeddings do not have images which are sub-groups, we may enlarge them
via consecutive applications of the path trick (while keeping all of the properties assumed for our original distribution, such as probability of atoms
being $\Omega(1)$, pairwise connectedness and so on). It will be important for us to do this is in a more careful manner, and maintain the fact that the
alphabet of $x$ is $\Sigma^{T}$ for some $T\in\mathbb{N}$. With respect to that, iterating Lemma~\ref{lem:path_trick_saturates} directly gives saturation
with respect to the embedding of $x$:
\begin{lemma}\label{lem:path_trick_saturates_x}
  Let $\Sigma$, $\Gamma$ and $\Phi$ be finite alphabets, let $\mu$ be a pairwise connected distribution over $\Sigma\times \Gamma\times \Phi$.
  Then there are constants $T$ and $\ell$ depending only on the alphabet sizes such that if $(\sigma,\gamma,\phi)$
  is a master embedding for $\mu$ into $(H,+)$, wherein each one of ${\sf Image}(\sigma)$, ${\sf Image}(\gamma)$, ${\sf Image}(\phi)$
  contains $0$, then repeating the process in Lemma~\ref{lem:path_trick_saturates} $T$ times one gets a distribution $\mu'$ over $\Sigma'\times\Gamma'\times\Phi'$
  times where
  \begin{enumerate}
    \item $\Sigma' = \Sigma^{\ell}$.
    \item Letting $(\sigma',\gamma',\gamma')$ be the induced master embedding of $\mu'$, we have that ${\sf Image}(\sigma')$ is a subgroup of $H$.
  \end{enumerate}
\end{lemma}
\begin{proof}
  As long as ${\sf Image}(\sigma)$ is not a sub-group, applying the process in Lemma~\ref{lem:path_trick_saturates} enlarges it, and the new alphabet
  of $x$ is a power of the older alphabet of $x$. Hence repeating the process $T = \card{H}$ times gives eventually $\mu'$ on which $\sigma'$ is saturated.
\end{proof}

The next lemma is an analogous statement for $y$ and $z$, and we show that the same procedure -- applied sufficiently many times  -- also works.
The argument is a bit more subtle, as we cannot afford ourselves to apply long path tricks on $x$; such operation may not preserve the fact that
the alphabet of $x$ would remain a power of $\Sigma$. We remark that the role of $y$ and $z$ is symmetric, and thus while we formulate the statement in terms of $z$ (as well
as in Lemma~\ref{lem:path_trick_saturates}), by flipping the roles of $y$ and $z$ one gets an analogous statement for $y$, which we shall also use.
\begin{lemma}\label{lem:path_trick_saturates_z}
  Let $\Sigma$, $\Gamma$ and $\Phi$ be finite alphabets, let $\mu$ be a pairwise connected distribution over $\Sigma\times \Gamma\times \Phi$.
  Then there are constants $T,T'$ and $\ell$ depending only on the alphabet sizes such that if $(\sigma,\gamma,\phi)$
  is a master embedding for $\mu$ into $(H,+)$, wherein each one of ${\sf Image}(\sigma)$, ${\sf Image}(\gamma)$, ${\sf Image}(\phi)$
  contains $0$, then the following holds. Consider the distributions:
  \begin{enumerate}
    \item Apply the transformation in Lemma~\ref{lem:path_trick_saturates} $T$ times to get a distribution $\mu'$.
    Let $(\sigma',\gamma',\phi')$ be the induced master embedding.
    \item Apply the path trick on $\mu'$ with respect to $y$ for $T'$ times to get a distribution $\mu''$.
    \item Apply the path trick on $\mu'$ with respect to $z$ for $3$ steps to get a distribution $\mu'''$.
    Let $(\sigma''',\gamma''',\phi''')$ be the induced master embedding.
  \end{enumerate}
  If ${\sf Image}(\phi)$ is not a subgroup of $H$, then ${\sf Image}(\phi)\subsetneq {\sf Image}(\phi''')$.
  Furthermore, the support of $\mu'''$ on its first coordinate (namely, its $x$-coordinate) is full, that is, $\Sigma^{\ell}$.
\end{lemma}
\begin{proof}
  Consider the transformation in Lemma~\ref{lem:path_trick_saturates}: given a distribution $\mu$ over $\Sigma\times \Gamma\times \Phi$,
  it outputs a distribution $\nu$ over $\Sigma^{3}\times \Gamma'\times \Phi'$ where $\Gamma'\subseteq \Gamma^{\ell_2}$ and $\Phi'\subseteq \Phi^{\ell_3}$
  for some $\ell_2,\ell_3$ depending only on the alphabet sizes. Let $\Gamma_{{\sf original}}\subseteq \Gamma^{\ell_2}$ and $\Phi_{{\sf original}}\subseteq \Phi^{\ell_3}$
  be copies of $\Gamma$ and $\Phi$ in $\nu$, namely
  \[
  \Gamma_{{\sf original}} = \sett{(y,\ldots,y)}{y\in\Gamma},
  \qquad
  \Phi_{{\sf original}} = \sett{(z,\ldots,z)}{z\in\Phi},
  \]
  and note that $\Gamma_{{\sf original}}\subseteq \Gamma'$ and $\Phi_{{\sf original}}\subseteq \Phi'$. Note that if there is a path of length at most $3$ between
  $y$ and $z$ in $\mu$, then there is an edge between $(y,\ldots,y)$ and $(z,\ldots,z)$ in $\nu$. Indeed, if there was a path $y$ to $z'$ to $y'$ to $z$ labeled by
  $x_1,x_2,x_3$, then $((x_1,x_2,x_3), (y,\ldots,y), (z,\ldots,z))$ would be in ${\sf supp}(\nu)$. Repeating this transformation twice get that if there was a path from
  $y$ to $z$ of length at most $4$, then we would have an edge between $(y,\ldots,y)$ and $(z,\ldots,z)$ in the new distribution.

  Thus, as $\mu$ is pairwise connected we can apply the transformation
  in Lemma~\ref{lem:path_trick_saturates} $T$ times, where $T$ only depends on the alphabet sizes, to get a distribution $\mu'$ whose alphabet on $x$
  is $\Sigma^{\ell}$ for some $\ell$ and its support on $\Gamma_{{\sf original}}\times \Phi_{{\sf original}}$ is full. Let $(\sigma',\gamma',\phi')$
  be the induced master embedding as in Definition~\ref{def:master}. We then apply the path trick on
  $\mu'$ with respect to $y$ for $T'$ steps so that the support of the distribution over $x$ and $z$ is full (thanks to Lemma~\ref{lem:path_keeps_connectedness}), and
  we think of $\mu''$ as a distribution over $\Sigma''\times\Gamma''\times \Phi''$. We remark that the support of $\mu''$ on $\Gamma_{{\sf original}}\times \Phi_{{\sf original}}$
  is still full (this is preserved under the last application of the path trick), and also that $\Sigma'' = \Sigma^{3T}$.

  We now apply the path trick on $\mu''$ for length $3$ on $z$ to get the distribution $\mu'''$ over $\Sigma'''\times\Gamma'''\times\Phi'''$, where
  $\Sigma''' = \Sigma'' = \Sigma^{3T}$, and we have an induced master embedding of $\mu'''$, denoted by $(\sigma''',\gamma''',\phi''')$ as in Definition~\ref{def:master}.
  As in Lemma~\ref{lem:path_trick_saturates} we have that ${\sf Image}(\phi)\subseteq {\sf Image}(\phi''')$, and we next argue that if ${\sf Image}(\phi)$ is not a sub-group,
  then this is a strict containment.

  If ${\sf Image}(\phi)$ is not a sub-group of $H$,
  then we may find $z_1,z_3\in \Phi$ such that $\phi(z_1)+\phi(z_3)\not\in {\sf Image}(\phi)$, and we pick $z_2\in\Phi$ such that $\phi(z_2) = 0$.
  Below we show that $((z_1,\ldots,z_1),(z_2,\ldots,z_2),(z_3,\ldots,z_3))$ is in the support of $\mu'''$ on $\Phi'''$, and we now argue that this would
  give the strict containment. Indeed, by Definition~\ref{def:master} we have that $\phi'''((z_1,\ldots,z_1),(z_2,\ldots,z_2),(z_3,\ldots,z_3))$ is
  equal to
  \[
  \phi'(z_1,\ldots,z_1)-\phi'(z_2,\ldots,z_2) + \phi'(z_3,\ldots,z_3)
  \]
  and as $\phi'(z_1,\ldots,z_1) = \phi(z_1)$, $\phi'(z_2,\ldots,z_2) = \phi(z_2) = 0$ and $\phi'(z_3,\ldots,z_3) = \phi(z_3)$, we get that
  ${\sf Image}(\phi'')\subsetneq{\sf Image}(\phi)$ as desired.

  To show that $((z_1,\ldots,z_1),(z_2,\ldots,z_2),(z_3,\ldots,z_3))$ is in the support of $\mu'''$ on $\Phi'''$, first pick any
  $y\in \Gamma$. Since the support of $\mu''$ on the set $\Gamma_{{\sf original}}\times \Phi_{{\sf original}}$ is full,
  we get that there are $x_1,x_2\in \Sigma''$ such that $(x_1,(y,\ldots,y),(z_1,\ldots,z_1))$
  and $(x_2,(y,\ldots,y),(z_2,\ldots,z_2))$ are in ${\sf supp}(\mu'')$. As in $\mu''$ we have that the support on $\Sigma''\times \Phi''$ is full, it follows
  that there is $\vec{y}$ such that $(x_2,\vec{y},(z_3,\ldots,z_3))$ is in  ${\sf supp}(\mu'')$. We note that now in $\mu''$ we have a path from
  $x_1$ to $(y,\ldots,y)$ to $x_2$ to $\vec{y}$, whose labels are $(z_1,\ldots,z_1)$, $(z_2,\ldots,z_2)$ and $(z_3,\ldots,z_3)$, and therefore
  $((z_1,\ldots,z_1),(z_2,\ldots,z_2),(z_3,\ldots,z_3))$ is in the support of $\mu'''$ on $\Phi'''$ as required.
\end{proof}

We now combine Lemmas~\ref{lem:path_trick_saturates_x},~\ref{lem:path_trick_saturates_z} to get our overall transformation that saturates the master embeddings.
\begin{lemma}\label{lem:iterate_path_trick_saturates}
  Let $\Sigma$, $\Gamma$ and $\Phi$ be finite alphabets of size at most $m$, and let $\mu$ be a pairwise connected distribution in which the probability
  of each atom is at least $\alpha>0$, and let $(\sigma_{{\sf master}}, \gamma_{{\sf master}}, \phi_{{\sf master}})$ be a master embedding for $\mu$
  into an Abelian group $(H,+)$.
  Then, there are $\alpha' = \alpha'(m,\alpha)>0$, $\ell = \ell(m)\in\mathbb{N}$ and a distribution $\mu'$ over
  $\Sigma'\times\Gamma'\times \Phi'$ which results from $\mu$ by a sequence of applications of the path trick, such that:
  \begin{enumerate}
    \item $\mu'$ is pairwise connected and the probability of each atom is at least $\alpha'$.
    \item The number of applications of the path trick is at most some $T\in\mathbb{N}$ depending only on $m$.
    Furthermore, the new alphabet sizes $\card{\Sigma'}$, $\card{\Gamma'}$ and $\card{\Phi'}$ are bounded by some function of $m$.
    \item There is a master embedding $(\sigma_{{\sf master}}', \gamma_{{\sf master}}', \phi_{{\sf master}}')$
    of $\mu'$ into $(H,+)$ such that:
    \begin{enumerate}
    \item The image of each one of $\sigma_{{\sf master}}'$, $\gamma_{{\sf master}}'$ and $\phi_{{\sf master}'}$ is $H$.
      \item This master embedding is given by a master embedding of $\mu$ by the transformations described in Definition~\ref{def:master}
        following the applications of the path trick.
      \item The support of $(\sigma_{{\sf master}}'(x), \gamma_{{\sf master}}'(y), \phi_{{\sf master}}'(z))$ where $(x,y,z)\sim \mu'$
      is full, namely it is
      \[
      \sett{(a,b,c)\in H}{a+b+c = 0}.
      \]
      \item The support of $\mu'$ on the first coordinate is $\Sigma' = \Sigma^{\ell}$.
    \end{enumerate}

  \end{enumerate}
\end{lemma}
\begin{proof}
  We proceed by an iterative process. Starting with the distribution $\mu$ and a master embedding for it into an Abelian group $(H,+)$,
  so long as the image of one of the embedding's component is not a sub-group, we apply either Lemma~\ref{lem:path_trick_saturates_x} or
  Lemma~\ref{lem:path_trick_saturates_z} to enlarge it
  (while clearly not decreasing the size of the image of the other two components), so eventually we get to a distribution $\nu$
  and master embeddings $\sigma'$, $\gamma'$ and $\phi'$ such that ${\sf Image}(\sigma') = H_1$, ${\sf Image}(\gamma') = H_2$ and
  ${\sf Image}(\phi') = H_3$ where $H_1,H_2$ and $H_3$ are subgroups of $H$. We then take the distribution $\mu'$ which is a result
  of applying the path trick of $\nu$ with respect to $z$ for $T = T(m)\in\mathbb{N}$ times for sufficiently large $T$, so that by
  Lemma~\ref{lem:path_keeps_connectedness} the distribution of $\mu'$ over the first two coordinates is full. Let the alphabets
  of $\mu'$ be $\Sigma'$, $\Gamma'$ and $\Phi'$; then by Lemmas~\ref{lem:path_trick_saturates_x},~\ref{lem:path_trick_saturates_z}
  we see that $\Sigma' = \Sigma^{\ell}$ for some $\ell = \ell(m)\in\mathbb{N}$.

  We now argue that $H_1 = H_2 = H_3$. By Lemma~\ref{lem:path_keeps_connectedness} we may apply the path trick on $\mu'$ to get $\mu''$
  in which the support on $\Sigma'\times \Gamma'$ is full, and we argue that this means that $H_1 + H_2 \subseteq H_3$.
  Indeed, we could pick any $h_1\in H_1$, $h_2\in H_2$ and find $x\in\Sigma'$ and $y\in\Gamma'$ such that
  $\sigma_{\sf master}'(x) = -h_1$ and $\gamma_{\sf master}'(y) = -h_2$, and thus find $\vec{z}\in \Phi''$ such that $(x,y,\vec{z})$
  is in the support of $\mu''$, so
  \[
  -h_1-h_2 + \phi_{\sf master}'^{\sharp}(\vec{z})
  =\sigma_{\sf master}'(x) + \gamma_{\sf master}'(y) + \phi_{\sf master}'^{\sharp}(\vec{z})
  =0
  \]
  by definition of embeddings, so $\phi_{\sf master}'^{\sharp}(\vec{z}) = h_1+h_2$, hence $h_1+h_2$ are in the image of $\phi'^{\sharp}(\vec{z})$,
  which is the same as the image of $\phi'$ (as it is already a subgroup), so $h_1 + h_2 \in H_3$.

  Thus, $H_1 + H_2 \subseteq H_3$ and analogously $H_1+H_3\subseteq H_2$ and $H_2+H_3\subseteq H_1$,
  and it follows that $H_1 = H_2 = H_3$. Thus, we can view $(\sigma_{{\sf master}}', \gamma_{{\sf master}}', \phi_{{\sf master}}')$
  as an embedding into $(H_1,+)$, so now it is a master embedding satisfying the third bullet. The first two bullets
  are clear, as the number of path trick applications is some constant depending only on the alphabet sizes (and the size of
  $H$, which also only depends on the alphabet sizes of $\mu$).
\end{proof}

\subsection{Conclusion of Section~\ref{sec:master_embed}}
Using Lemma~\ref{lem:iterate_path_trick_saturates} and Lemma~\ref{lem:from_mu_to_path} together, one gets that
an expectation as in~\eqref{eq:1} over $\mu$ can be upper bounded by (some power bounded away from $0$) of a similar looking expectation
over $\mu'$ in which the master embeddings are saturated; by applying a few more path tricks and using Lemma~\ref{lem:path_keeps_connectedness},
we can also ensure further connectedness properties of $\mu'$ (which we will need in the future).
Thus, we have gained further important properties of our $\mu$, at the expense of:
\begin{enumerate}
  \item The $y$-function and $z$-function may become completely different as a result of these operations (but they remain bounded).
  \item The $x$-function becomes more complicated. Indeed, our $x$-alphabet will be some $\Sigma' \subseteq \Sigma^{\ell}$, and our
  $x$-function will be given as
  \[
  F(x_1,\ldots,x_{\ell}) = \prod\limits_{i=1}^{\ell} f_{i}(x_i),
  \]
  where each $f_i$ is either the function $f$ or its complex conjugate $\overline{f}$.
\end{enumerate}
Thus, to prove Theorem~\ref{thm:main_stab_3}, it suffices to
(1) prove a version of that theorem under the additional assumptions we gained on $\mu'$;
(2) prove that a structural result for $F$ as in Theorem~\ref{thm:main_stab_3} implies a similar
structural result for $f$.
The majority of our effort will be to establish the first step: this part of the argument is contained in
Sections~\ref{sec:prep},~\ref{sec:max_merg_mot},~\ref{sec:base_case},~\ref{sec:reudce_to_homogenous},~\ref{sec:reduce_to_near_linear},~\ref{sec:prove_near_lin}
and~\ref{sec:the_hastad_argument}.
The second step will be a relatively easy consequence: this part of the argument is contained in Section~\ref{sec:unraveling}.
Thus, we arrive at the following statement which is the same as the statement of Theorem~\ref{thm:main_stab_3}, except that we have additional
assumptions on the distribution $\mu$:
\begin{thm}\label{thm:main_stab_3_saturated}
  For all $m\in\mathbb{N}$, $\alpha>0$ and $\eps>0$, there exists $d\in\mathbb{N}$ and $\eps'>0$ such that the following holds.
  Suppose that $\mu$ is a distribution over $\Sigma\times\Gamma\times \Phi$ such that:
  \begin{enumerate}
    \item The probability of each atom is at least $\alpha$.
    \item The size of each one of $\Sigma,\Gamma,\Phi$ is at most $m$.
    \item ${\sf supp}(\mu)$ is pairwise connected.
    \item There is a master embedding $(\sigma,\gamma,\phi)$ of $\mu$ into an Abelian group $(H,+)$ which is saturated,
    and the distribution of $(\sigma(x),\gamma(y),\phi(z))$ where $(x,y,z)\sim \mu$ has full support on $\{(a,b,c)\in H^3~|~a+b+c = 0\}$.
  \end{enumerate}
  Then, if $f\colon\Sigma^n\to \mathbb{C}$, $g\colon\Gamma^n\to \mathbb{C}$ and $h\colon\Phi^n\to \mathbb{C}$
  are $1$-bounded functions such that
  \[
  \card{\Expect{(x,y,z)\sim \mu^{\otimes n}}{f(x)g(y)h(z)}}\geq \eps,
  \]
  then there are $1$-bounded functions $u_1,\ldots,u_n\colon \Sigma\to \mathbb{C}$ and a function $L\colon \Sigma^n\to \mathbb{C}$ of degree at most $d$
  and $2$-norm at most $1$ such that
  \[
  \card{\Expect{x\sim \mu_x^{\otimes n}}{f(x)\cdot L(x)\prod\limits_{i=1}^{n}u_i(x_i)}}\geq \eps'.
  \]

  Furthermore, there are $\chi_i\in\hat{H}$ such that for all $i$, $u_i(x_i) = \chi_i(\sigma(x_i))$.
  Quantitatively, we have $d = {\sf poly}_{m,\alpha}\left(\frac{1}{\eps}\right)$ and $\eps' = 2^{-{\sf poly}_{m,\alpha}\left(\frac{1}{\eps}\right)}$.
\end{thm}

\subsection{Embeddings Into the Infinite Cyclic Group}
So far we have discussed embeddings of a distribution into finite Abelian groups, however it also makes sense
to consider embeddings into infinite groups. Specifically, we will need to consider embeddings of a distribution into the
infinite cyclic group $([0,1),\pmod{1})$. Using approximation arguments (and more specifically, Dirichlet's Approximation Theorem),
we show in the following lemma that any embedding of a distribution into $([0,1),\pmod{1})$
is \emph{equivalent} to an embedding into a finite Abelian group, hence there is nothing particularly special about them.
For this, we first define the notion of equivalence.

\begin{definition}
  Let $\mu$ be a distribution over $\Sigma\times\Gamma\times \Phi$, and let
  $(\sigma,\gamma,\phi)$ and $(\sigma',\gamma',\phi')$ be Abelian embeddings
  of $\mu$.
  We say $(\sigma,\gamma,\phi)$ is equivalent to $(\sigma',\gamma',\phi')$ if
  there are bijective maps
  $m_1\colon {\sf Image}(\sigma)\to {\sf Image}(\sigma')$,
  $m_2\colon {\sf Image}(\gamma)\to {\sf Image}(\gamma')$
  and
  $m_3\colon {\sf Image}(\phi)\to {\sf Image}(\phi')$
  such that $\sigma'(x) = m_1(\sigma(x))$,
  $\gamma'(y) = m_2(\gamma(y))$ and
  $\phi'(z) = m_3(\phi(z))$ for all
  $x\in\Sigma$, $y\in \Gamma$, $z\in \Phi$.
\end{definition}

\begin{lemma}\label{lem:turn_infinite_to_finite}
  Let $\Sigma$, $\Gamma$ and $\Phi$ be finite alphabets and let $\mu$ be a distribution over $\Sigma\times\Gamma\times \Phi$.
  If $\sigma\colon \Sigma\to [0,1)$, $\gamma\colon \Gamma\to[0,1)$ and $\phi\colon \Phi\to [0,1)$
  is an embedding of $\mu$ into $([0,1),+\pmod{1})$, then
  $(\sigma,\gamma,\phi)$ is equivalent to an Abelian embedding of $\mu$ into a finite Abelian group.
\end{lemma}
\begin{proof}
Consider the set of numbers $S = {\sf Image}(\sigma)\cup {\sf Image}(\phi) \cup {\sf Image}(\gamma)$,
  let $r = \card{\Sigma} +\card{\Phi} + \card{\Gamma}$ and let $N = N(r)\in\mathbb{N}$ to be determined.
  Then $\card{S}\leq r$, so by Dirichlet's approximation theorem we may find integers $p_i, q$ such that for each $s_i\in S$ we have that
  $\card{s_i - \frac{p_i}{q}}\leq \frac{1}{q N^{1/r}}$. Let
  \[
  \alpha = \min_{x,x' \sigma(x)\neq \sigma(x')}\min_{z\in\mathbb{Z}}\card{z+\sigma(x)-\sigma(x')}.
  \]
  We choose
  $N = \left(\frac{3}{\alpha}\right)^r$, define $\sigma'$ by $\sigma'(x) = \frac{p_i}{q}\pmod{1}$ if $\sigma(x) = s_i$, and similarly define $\phi',\gamma'$.
  \begin{enumerate}
    \item First, we show that $\sigma', \gamma', \phi'$ is an embedding. Fix $(x,y,z)\in{\sf supp}(\mu)$; then we have
    \[
    \sigma'(x) + \phi'(y) + \gamma'(z) =
    \sigma(x)+\phi(y) + \gamma(z) + \Delta,
    \]
    where $\card{\Delta}\leq \frac{3}{q N^{1/r}}$. Noting that $\sigma(x) + \phi(y) + \gamma(z)$ is an integer (as it is $0$ mod $1$), it
    follows that $\sigma'(x) + \phi'(y) + \gamma'(z)$ is very close to an integer, up to $\frac{3}{q N^{1/r}} < \frac{1}{q}$. On the other hand, by definition
    of $\sigma',\phi',\gamma'$, it is a number of the form $P/q$ for some integer $P$, hence it can either be an integer or at least $\frac{1}{q}$ far from all
    integers. It follows that it is an integer, so $\sigma'(x) + \phi'(y) + \gamma'(z) = 0\pmod{1}$.

    \item Second, we argue that $(\sigma',\gamma',\phi')$ is equivalent to $(\sigma,\gamma,\phi)$. For that, we have to argue that
    $\sigma(x)\neq \sigma(x')$ if and only if $\sigma'(x)\neq \sigma'(x')$. If $\sigma(x) = \sigma(x')$ then it is clear that
    $\sigma'(x) = \sigma'(x')$ by definition. If $\sigma(x)\neq \sigma(x')$, then by the definition of $\alpha$ we get that
    $\sigma(x) - \sigma(x')$ is at least $\alpha$-far from all integers, and as $\card{\sigma(x) - \sigma'(x')}\leq \frac{\alpha}{3}$,
    it follows that $\sigma'(x) - \sigma'(x')$ is at least $\alpha/3$ far from all integers, and in particular from $0$,
    so $\sigma'(x)\neq \sigma'(x')$.
  \end{enumerate}
  In conclusion, we get that $(\sigma,\gamma,\phi)$ and $(\sigma',\gamma',\phi')$, and noting that, after multiplying by $q$, the latter is an embedding
  into $(\mathbb{Z}_q,+)$, the proof is concluded.
\end{proof}

\section{Non-embedding Degrees and Partial Bases}\label{sec:prep}
In this section we make progress towards the proof of Theorem~\ref{thm:main_stab_3_saturated}, and state Theorem~\ref{thm:nonembed_deg_must_be_small} which is a related
by weaker form. The proof of Theorem~\ref{thm:nonembed_deg_must_be_small} then spans
Sections~\ref{sec:base_case},~\ref{sec:reudce_to_homogenous},~\ref{sec:reduce_to_near_linear},~\ref{sec:prove_near_lin},
and the derivation of Theorem~\ref{thm:main_stab_3_saturated} from Theorem~\ref{thm:nonembed_deg_must_be_small} is done in Section~\ref{sec:the_hastad_argument}.
\subsection{A Motivating Case}\label{sec:motivating_all_are_H}
Let $\mu$ be a distribution over $\Sigma\times \Gamma\times \Phi$ as in Theorem~\ref{thm:main_stab_3_saturated}, and let $\sigma$, $\gamma$, $\phi$
be a saturated master embedding of $\mu$ into $(H,+)$. To motivate the discussion, below we begin by considering
a motivating example in which the master embeddings partition the alphabets into singletons.

Namely, suppose that for each $h\in H$ each one of $\sigma^{-1}(h)$, $\gamma^{-1}(h)$ and $\phi^{-1}(h)$ has size exactly $1$. In that case the
master embeddings form an identification between our alphabets and the group $H$, hence what we really have in our hands is $3$ functions,
$f^{\sharp}\colon H^n\to\mathbb{C}$, $g^{\sharp} \colon H^{n}\to\mathbb{C}$ and $h^{\sharp}\colon H^n\to\mathbb{C}$ defined as
\begin{align*}
&f^{\sharp}(h_1,\ldots,h_n) = f(\sigma^{-1}(h_1),\ldots,\sigma^{-1}(h_n)),
\qquad
g^{\sharp}(h_1,\ldots,h_n) = g(\gamma^{-1}(h_1),\ldots,\gamma^{-1}(h_n)),\\
&\qquad\qquad\qquad\qquad\qquad\qquad
h^{\sharp}(h_1,\ldots,h_n) = h(\phi^{-1}(h_1),\ldots,\phi^{-1}(h_n)).
\end{align*}
Thus, considering the distribution $\nu$ over $H^3$ which is the distribution of $(\sigma(x),\gamma(y),\phi(z))$ where $(x,y,z)\sim \mu$,
we get that
\[
\card{\Expect{(x^{\sharp},y^{\sharp},z^{\sharp})\sim \nu^{\otimes n}}{f^{\sharp}(x^{\sharp})g^{\sharp}(y^{\sharp})h^{\sharp}(z^{\sharp})}}
=
\card{\Expect{(x,y,z)\sim \mu^{\otimes n}}{f(x)g(y)h(z)}}\geq \eps.
\]
Thus, we have transformed our question into an equivalent question over Abelian groups. As $\mu$ is pairwise connected, $\nu$ is also
pairwise connected, and by definition of the master embedding it follows that its support is contained in $S = \{(x^{\sharp},y^{\sharp},z^{\sharp})\in H^3~|~x^{\sharp}+y^{\sharp}+z^{\sharp}=0\}$. Combining the pairwise connectedness and the fact that the master embedding is saturated, it follows that
the support of $\nu$ is precisely $S$. Using other ideas (based on random restrictions) we can ensure that the distribution $\nu$ is actually uniform
over $S$, in which case we have reduced the problem to a well-known Fourier analytic computation (which appears in many places, such as Roth's theorem~\cite{Roth,Meshulam}
as well as in theoretical computer science~\cite{BLR,Has01}). In particular, one can show that there is a Fourier character $\chi\in \hat{H}^{n}$ such
that $\card{\widehat{f^{\sharp}}(\chi)}\geq \eps$, and translating this back into information about the function $f$ one gets the conclusion of
Theorem~\ref{thm:main_stab_3_saturated} with the low-degree part $L$ being the constant $1$ function. We remark that even in this simplistic argument,
the presence of the low-degree function ultimately comes from the step in which we switched from the distribution $\nu$ to the uniform distribution over $S$.
Nevertheless, we encourage the reader to ignore this point for now.

Our goal in this, and in the several subsequent sections will be to show that while in general, it need not be the case that $\sigma$, $\gamma$ and $\phi$
completely partition their respective alphabets, the only functions $f$, $g$ and $h$ for which the expectation in Theorem~\ref{thm:main_stab_3_saturated} may be
have a special property. Specifically, we show that such the function $f$ ``hardly distinguish'' between two input symbols $x$ and $x'$ that are mapped to the
same group element by the master embedding component $\sigma$ (and similarly for $g$ and $h$). Towards this end, in this section
we first define a partial basis for the set of functions composed of functions that only depend on the values of the master embeddings, and then
complete them to bases. We then define the notions of ``embedding degree'' and ``non-embedding degree''. These are notions that capture how well
does our function $f$ distinguish between inputs that are mapped to the same group element by the master embedding.
With these notions, we show that $f,g,h$ for which the expectation in Theorem~\ref{thm:main_stab_3_saturated}
is large, must have small non-embedding degree.

\subsection{Setting Up a Partial Basis via Saturated Embeddings, and Non-embedding Degrees}\label{sec:partial_basis}
Let $\mu$ be a distribution over $\Sigma\times \Gamma\times \Phi$ as in Theorem~\ref{thm:main_stab_3_saturated}
and let $\sigma$, $\gamma$, $\phi$ be saturated master embeddings into $(H,+)$. In this section, we explain how
to use these embeddings to define useful partial bases for the spaces of functions we are dealing with, as well
as how to define the notion the related notion of non-embedding degree.  For the sake of concreteness, we shall
phrase everything in the language of functions of $x$ and the alphabet $\Sigma$, however everything holds for
the other two variables and alphabets as well.

Consider the space $L_2(\Sigma,\mu_x)$, and note that we may set up a partial basis for them using characters over $H$
and the master embeddings.
\begin{definition}
  Given a finite set $\Sigma$, an Abelian group $(H,+)$ and $\sigma\colon \Sigma\to H$,
  for each $\chi\in\hat{H}$, we define $\chi_{\sigma}\colon \Sigma\to \mathbb{C}$ by
  $\chi_{\sigma}(x) = \chi(\sigma(x))$.
\end{definition}
We note that in our setting, the set $\{\chi_{\sigma}\}_{\chi\in \hat{H}}$ is a linearly independent set. Indeed,
to observe note that as $\sigma$ is saturated, it is enough to show that there is a distribution $\mathcal{D}$ over $\Sigma$
in which $\chi_{\sigma}$ is an orthonormal set, and we consider a distribution $\mathcal{D}$ over $\Sigma$ such that
$\sigma(x)$ is distributed uniformly in $H$ when $x\sim \mathcal{D}$ (this is clearly possible). In that case, for all
$\chi,\chi'\in \hat{H}$ we have that $\inner{\chi_{\sigma}}{\chi'_{\sigma}}_{\mathcal{D}} = \inner{\chi}{\chi'} = 1_{\chi = \chi'}$,
as required.
\begin{definition}
  Given a finite set $\Sigma$, an Abelian group $(H,+)$ and $\sigma\colon \Sigma\to H$, we define
  ${\sf Embed}_{\sigma}(\mu) \subseteq \set{f\colon \Sigma\to\mathbb{C}}$ by
  \[
    {\sf Embed}_{\sigma}(\mu) = {\sf Span}\left(\sett{\chi_{\sigma}}{\chi\in\hat{H}}\right).
  \]
\end{definition}

Thus, in our setting we have ${\sf Embed}_{\sigma}(\mu)\subseteq L_2(\Sigma; \mu_x)$, ${\sf Embed}_{\gamma}(\mu) \subseteq L_2(\Gamma; \mu_y)$
and ${\sf Embed}_{\phi}(\mu) \subseteq L_2(\Phi; \mu_z)$. These spaces capture the space of embedding functions; an $x$ function $f$ is called
an embedding function if $f(x)$ only depends on $\sigma_{\sf master}(x)$.
\begin{claim}\label{claim:embedding_fns}
  Suppose that $\mu$ is a distribution over $\Sigma\times\Gamma\times \Phi$ in which the master embeddings are saturated,
  and let $f\colon \Sigma\to\mathbb{C}$, $g\colon\Gamma\to\mathbb{C}$ and $h\colon \Phi\to\mathbb{C}$ be functions such that
  \[
        f(x) + g(y) + h(z) = 0
  \]
  for all $(x,y,z)\in {\sf supp}(\mu)$. Then $f\in {\sf Embed}_{\sigma}(\mu)$.
\end{claim}
\begin{proof}
  Taking real and imaginary parts separately, it suffices to prove the statement for real valued functions.
  Multiplying $f,g,h$ by small enough constant, we may assume that $\card{f(x)},\card{g(y)},\card{h(z)}\leq 1$
  for all $x,y,z$, hence $f,g,h$ form an embedding of $\mu$ into $([-1,1),+\pmod{2})$. By Lemma~\ref{lem:turn_infinite_to_finite},
  it follows that there are $m_1,m_2,m_3$ injectives such that $f' = m_1(f)$, $g' = m_2(g)$ and $h' = m_3(h)$ is an embedding
  of $\mu$ into a finite Abelian group. By the definition of the master embedding if follows that $f'$ is constant on each part of the
  partition on $\Sigma$ induced by the master embedding, and so $f'(x) = f'(x')$ if $\sigma_{{\sf master}}(x) = \sigma_{{\sf master}}(x')$.
  Since $m_1$ is injective, it follows that the same is true for $f$, and so $f\in {\sf Embed}_{\sigma}(\mu)$.
\end{proof}

The above motivating example is just the case that the $3$ containments ${\sf Embed}_{\sigma}(\mu)\subseteq L_2(\Sigma; \mu_x)$, ${\sf Embed}_{\gamma}(\mu) \subseteq L_2(\Gamma; \mu_y)$
and ${\sf Embed}_{\phi}(\mu) \subseteq L_2(\Phi; \mu_z)$
are in fact equalities. This need not be necessarily the case for us, hence we may need to complete these sets to get all of $ L_2(\Sigma; \mu_x)$.
More precisely, let $B_1\subseteq {\sf Embed}_{\sigma}(\mu)$ be an orthonormal basis for ${\sf Embed}_{\sigma}(\mu)$ (with respect to the inner
product in $L_2(\Sigma,\mu_x)$). We complete it to an orthonormal basis for $L_2(\Sigma,\mu_x)$ by adding the set $B_2\subseteq L_2(\Sigma; \mu_x)$.
\begin{definition}
  A monomial in $L_2(\Sigma; \mu_x)$ is one of the basis functions from $B_1\cup B_2$. A monomial
  in $L_2(\Sigma^n,\mu_x^{\otimes n})$ is $\prod\limits_{i=1}^{n} u_i$ where $u_i\in B_1\cup B_2$
  for all $i$.
\end{definition}

\begin{definition}\label{def:degrees_early}
  The non-embedding degree of $u$, ${\sf nedeg}(u)$, is the number of $i$'s for which $u_i\in B_2$.
\end{definition}

With these notions, we may write any $f\colon \Sigma^n\to\mathbb{C}$ as
\[
f(x) = \sum\limits_{\vec{u} = (u_1,\ldots,u_n)\in (B_1\cup B_2)^{n}}\widehat{f}(\vec{u})\prod\limits_{i=1}^{n}u_i(x_i),
\]
where $\widehat{f}(\vec{u}) = \inner{f}{\vec{u}}$. Each monomial of $f$ has its non-embedding degree, and we will want to define
a notion of non-embedding degree which captures the mass of $f$ on low non-embedding degree monomials and is convenient to work with, and towards
this end we define the non-embedding noise stability of a function.

\subsection{The Non-embedding Noise Stability of a Function}
In this section, we define the notion of non-embedding stability, which will be a crucial tool for us to measure the degree of a function with respect to ``non-embedding
functions''. We also state a few basic properties of it that will be used later on in our arguments.

For a parameter $\xi>0$ and a distribution $\mathcal{D}$ over $\Sigma$, consider the Markov chain $\mathrm{T}_{\text{non-embed}, 1-\xi, \mathcal{D}, \mu}$ on $\Sigma$ that
on $x\in \Sigma$, with probability $1-\xi$ stays in $x$, and otherwise samples $x'\sim \mathcal{D}$ conditioned on $\sigma(x') = \sigma(x)$. It will most often be the case for
us that $\mathcal{D} = \mu_x$, however this operator depends on $\mu$ as a whole (as it depends on the master embedding of it), we chose to include both in the notations.
There will be some rare exceptions though,
in which case we will make the notations explicit. Otherwise, to simplify notations we will often drop $\mathcal{D}$ from the notation (with the understanding that
it is just $\mu_x$).

Observe that $\mu_x$ is a stationary distribution for $\mathrm{T}_{\text{non-embed}, 1-\xi, \mu}$. Thus we can think of
$\mathrm{T}_{\text{non-embed}, 1-\xi, \mu}$ as an operator acting on $L_2(\Sigma; \mu_x)$ as
\[
\mathrm{T}_{\text{non-embed}, 1-\xi,\mu} f(x) = \Expect{x'\sim \mathrm{T}_{\text{non-embed}, 1-\xi} x}{f(x')}.
\]
With this in mind, we may define the non-embedding stability of $f$ as follows:
\begin{definition}\label{def:nestab}
  The non-embedding $\xi$-noise stability of $f\colon (\Sigma^n; \mathcal{D}^{\otimes n})\to \mathbb{C}$ is defined as
  \[
  {\sf NEStab}_{1-\xi, \mu}(f;\mathcal{D}^{\otimes n}) = \inner{f}{\mathrm{T}_{\text{non-embed}, 1-\xi,\mathcal{D},\mu}^{\otimes n} f}.
  \]
\end{definition}
We note that the operator $\mathrm{T}_{\text{non-embed}, 1-\xi,\mu}$ depends on the distribution $\mu$ itself and not only on its marginal on $x$,
as it is defined using the master embedding of $\mu$; the same goes for the non-embedding noise stability of a function.
Nevertheless, and to simplify notations we will often omit $\mu$ from notations when it is clear, and denote the operator
by $\mathrm{T}_{\text{non-embed}, 1-\xi}$ and the corresponding notion of noise stability by ${\sf NEStab}_{1-\xi}(f;\mu_x^{\otimes n})$.

\subsubsection{Diagonalizing the Non-embedding Stability Operator}
The basis functions $B_1,B_2$ defined earlier are eigenfunctions of the operator $\mathrm{T}_{\text{non-embed}, 1-\xi,\mu}$, and the following
fact gives us their eigenvalues:
\begin{fact}\label{fact:soft_nonbembed_op}
  Suppose that $u\colon \Sigma^n\to\mathbb{C}$ is a monomial of non-embedding degree equal to $d$. Then
  \[
  \mathrm{T}_{\text{non-embed}, 1-\xi,\mu}^{\otimes u} = (1-\xi)^{d} u.
  \]
\end{fact}
\begin{proof}
  It suffices to show that in the $1$-dimensional case, for $u\in B_1\cup B_2$, if $u\in B_1$ then
  $\mathrm{T}_{\text{non-embed}, 1-\xi,\mu} u = u$, and if $u\in B_2$ then $\mathrm{T}_{\text{non-embed}, 1-\xi,\mu} u =(1-\xi) u$.

  For $u\in B_1$ this is clear, since for every $x\in \Sigma$ and every $x'\in {\sf supp}(\mathrm{T}_{\text{non-embed}, 1-\xi,\mu} x)$
  it holds that $\sigma(x) = \sigma(x')$, and so $u(x) = u(x')$.

  Fix $u\in B_2$ and fix $x\in\Sigma$. We have that
  \[
  \mathrm{T}_{\text{non-embed}, 1-\xi,\mu} u (x)
  =(1-\xi)u(x) + \xi\cExpect{x'\sim \mu_x}{\sigma(x) = \sigma(x')}{u(x')}.
  \]
  Let $s = \sigma(x)$, and let $p_s = \Prob{x'}{\sigma(x') = s}$. Then the expectation on the right hand side is equal to
  \[
  \cExpect{x'\sim \mu_x}{\sigma(x')=s}{u(x')}
  =p_s^{-1} \Expect{x'\sim \mu_x}{u(x')1_{\sigma(x') = s}}
  =p_s^{-1} \inner{u}{1_{\sigma(\cdot) = s}}_{\mu_x}.
  \]
  Note that for all $s\in H$, the function $1_{\sigma(\cdot) = s}$ is in the span of $B_1$, and as $u\in B_2$ it is orthogonal to it,
  and so the last expression is $0$. We conclude that $\mathrm{T}_{\text{non-embed}, 1-\xi,\mu} u (x)  = (1-\xi)u(x)$.
\end{proof}

\subsubsection{Changing Noise Rates in Non-embedding Stability}
A basic property of the non-embedding noise stability is that it decreases as a result for increasing the noise rate:
\begin{claim}\label{claim:increase_noise_decrease_stab}
  Suppose that $0\leq \rho_1\leq \rho_2\leq 1$. Then for every $f\colon (\Sigma^n; \mu_x^{\otimes n})\to\mathbb{C}$
  we have that
  \[
  {\sf NEStab}_{\rho_1}(f;\mu_x^{\otimes n})
  \leq
  {\sf NEStab}_{\rho_2}(f;\mu_x^{\otimes n}).
  \]
\end{claim}
\begin{proof}
  Writing $f(x) = \sum\limits_{u\in (B_1\cup B_2)^n}\widehat{f}(u)\prod\limits_{i=1}^{n}u_i(x_i)$, we have by Fact~\ref{fact:soft_nonbembed_op} that
  \begin{align*}
    {\sf NEStab}_{\rho_1}(f;\mu_x^{\otimes n})
    =\inner{f}{\mathrm{T}_{\text{non-embed}, \rho_1} f}
    &=\inner{\sum\limits_{u\in (B_1\cup B_2)^n}\widehat{f}(u)u}{\sum\limits_{u\in (B_1\cup B_2)^n}\widehat{f}(u)\rho_1^{\text{non-embed-deg}(u)}u}\\
    &=\sum\limits_{u\in (B_1\cup B_2)^n}\rho_1^{\text{non-embed-deg}(u)}\card{\widehat{f}(u)}^2,
  \end{align*}
  and the result follows as all $\card{\widehat{f}(u)}^2$ are non-negative.
\end{proof}

\subsubsection{Random Restrictions and Non-embedding Stability}
The following claim is an instantiation of Lemma~\ref{lem:op_comparison_lemma} (and in fact our primary application for that lemma), asserting that
if we have a function $f$ that has small noise non-embedding stability, then in expectation after random restrictions it still has small noise non-embedding
stability.
\begin{claim}\label{claim:rr_nestab}
  For all $m\in\mathbb{N}$, $\alpha>$ there is $c>0$ such that the following holds.
  Let $\nu$ be a distribution over $\Sigma\times\Gamma\times \Phi$ and let $\mathcal{D}$, $\mathcal{D}'$ and $\mathcal{D}''$
  be distributions over $\Sigma$ such that:
  \begin{enumerate}
    \item The probability of each atom in $\mathcal{D},\mathcal{D}',\mathcal{D}''$ is at least $\alpha$.
    \item $\mathcal{D} = \beta\mathcal{D}' + (1-\beta)\mathcal{D}''$.
  \end{enumerate}
  Then, for all $f\colon(\Sigma^n,\mathcal{D}^{\otimes n})\to\mathbb{C}$ we have
  \[
  \Expect{\substack{J\subseteq_{\beta}[n]\\ x''\sim \mathcal{D}''{\overline{J}}}}{{\sf NEStab}_{1-\delta, \nu}(f_{\overline{J}\rightarrow x''}; \mathcal{D}'^{J})}
  \leq {\sf NEStab}_{1-c\beta\delta,\nu}(f; \mathcal{D}).
  \]
\end{claim}
\begin{proof}
  This is an immediate consequence of Lemma~\ref{lem:op_comparison_lemma}. In the notation therein, the vertex set of the graph $G$ is
  $\Sigma$, and $x,x'$ are adjacent if $\sigma(x) = \sigma(x')$ where $\sigma$ is the master embedding of $x$ for the distribution $\nu$.
  The left hand side and the right hand side in the above claim are precisely the left hand side and the right hand in Lemma~\ref{lem:op_comparison_lemma}.
\end{proof}

\subsection{Non-Embedding Influences of a Function}\label{sec:nonembed_inf}
We will need the notion of non-embedding influences of a function defined as follows.
\begin{definition}\label{def:non_embed_inf}
  Let $\mu$ be a distribution over $\Sigma\times \Gamma\times \Phi$, and let $\sigma\colon \Sigma\to H$, $\gamma\colon \Gamma\to H$
  and $\phi\colon \Phi\to H$ be a master embedding. For a function $f\colon \Gamma^n\to\mathbb{C}$ and a coordinate $j\in [n]$, we
  define the non-embedding influence of $f$ to be
  \[
  I_{j, \text{non-embed}}[f] = \cExpect{x'\sim \mu_x^{\otimes n}, a, b\sim \mu_x}{\sigma(a) = \sigma(b)}{\card{f(x_{-j} = x', x_j = a) - f(x_{-j} = x', x_j = b)}^2}.
  \]
\end{definition}

Non-embedding influences are defined analogously for functions over $y$ and $z$.
\begin{definition}
  In the setting of Definition~\ref{def:non_embed_inf}, the total non-embedding influence of $f$ is
  \[
  I_{\text{non-embed}}[f] = \sum\limits_{j=1}^{n} I_{j, \text{non-embed}}[f].
  \]
\end{definition}

We have the following easy fact.
\begin{fact}\label{fact:nonembed_inf_formula}
  In the setting of Definition~\ref{def:non_embed_inf}, we have
  \begin{enumerate}
    \item $I_{j, \text{non-embed}}[f] = 2\sum\limits_{\chi: \chi_{j}\not\in B_{1}}\card{\widehat{f}(\chi)}^2$.
    \item $I_{\text{non-embed}}[f] = \sum\limits_{\chi}{\sf nedeg}(\chi)\card{\widehat{f}(\chi)}^2$.
  \end{enumerate}
\end{fact}
\begin{proof}
  For the first bullet, we note that
  \begin{align*}
  f(x_{-j} = x', x_j = a) - f(x_{-j} = x', x_j = b)
  &=\sum\limits_{\chi}\widehat{f}(\chi)(\chi(x',a)-\chi(x',b))\\
  &=\sum\limits_{\chi}\widehat{f}(\chi)(\chi_{j}(b) - \chi_{j}(a))\prod\limits_{i\neq j}\chi_i(x'_i).
  \end{align*}
  For $\chi_j\in B_{1}$, $\chi_{j}(b) - \chi_{j}(a)=0$ whenever $\sigma(a) = \sigma(b)$, hence such terms give no contribution to the non-embedding
  influence. We thus get
  \begin{align*}
    &\cExpect{x', a, b}{\sigma(a) = \sigma(b)}{\card{f(x_{-j} = x', x_j = a) - f(x_{-j} = x', x_j = b)}^2}\\
    &=\sum\limits_{\substack{\chi,\chi'\\\chi_j,\chi_j'\not\in B_{{\sf embed}}}}\widehat{f}(\chi)\overline{\widehat{f}(\chi')}\left(\prod\limits_{i\neq j} 1_{\chi_i = \chi_i'}\right)
    \cExpect{a, b}{\sigma(a) = \sigma(b)}{(\chi_{j}(b) - \chi_{j}(a))\overline{(\chi_{j}'(b) - \chi_{j}'(a))}}
  \end{align*}
  We claim that if $\chi_{j}\neq \chi_{j}'$, then the expectation is $0$. Indeed, expanding we get terms such as
  $\chi_{j}(b)\overline{\chi_{j}'(a)}$, and we have that their expectation can be written as
  \[
  \sum\limits_{h\in H}\Prob{a\sim \mu_x}{\sigma(a) = h}\cExpect{b}{\sigma(b) = h}{\chi_j(b)}\cExpect{a}{\sigma(a) = h}{\overline{\chi_j'(a)}},
  \]
  which is $0$ as $\cExpect{b}{\sigma(b) = h}{\chi_j(b)}$ is proportional to the inner product between $\chi_j$ (a function orthogonal to embedding functions)
  and $1_{\sigma(b) = h}$ (an embedding function). Other terms are $\chi_{j}(b)\overline{\chi_{j}'(b)}$, which give the inner product between $\chi_j$ and $\chi_j'$
  which is $0$.

  We conclude that to give non-zero contribution we must have that $\chi_j = \chi_j'$ and so the non-embedding influence of $j$ is
  equal to
  \[
  \sum\limits_{\chi:\chi_j\not\in B_{{\sf embed}}}\card{\widehat{f}(\chi)}^2
    \cExpect{a, b}{\sigma(a) = \sigma(b)}{\card{\chi_{j}(b) - \chi_{j}(a)}^2}.
  \]
  Expanding the square and repeating the above computation, we get that the expectation is $2$, and the claim is proved.

  The second bullet follows immediately by summing up the first bullet over all $j=1,\ldots,n$.
\end{proof}
\subsection{A Non-Embedding Stability Formulation of Theorem~\ref{thm:main_stab_3_saturated}}
With these notions, we can now state the result asserting that the only functions $f,g,h$ for which expectations as in Theorem~\ref{thm:main_stab_3_saturated}
may be large, are functions for which the non-embedding stability is significant. This result by itself is very much in the spirit of Theorem~\ref{thm:main_stab_3_saturated},
but morally it is strictly weaker. Later, in Section~\ref{sec:the_hastad_argument}, we will show how Theorem~\ref{thm:main_stab_3_saturated} is implied
by Theorem~\ref{thm:nonembed_deg_must_be_small} below. Another difference in the formulation is that below the dependency between the parameters is more explicit,
and we do so as it is necessary for our proof to go through.

\begin{thm}\label{thm:nonembed_deg_must_be_small}
  For all $m\in\mathbb{N}$, $\alpha>0$ there are $M\in\mathbb{N}$, $\delta_0>0$ and $\eta>0$ such that the following holds for all $0< \delta\leq \delta_0$.
  Suppose that $\mu$ is a distribution over $\Sigma\times \Gamma\times \Phi$ satisfying:
  \begin{enumerate}
    \item The probability of each atom is at least $\alpha$.
    \item The size of each one of $\Sigma,\Gamma,\Phi$ is at most $m$.
    \item ${\sf supp}(\mu)$ is pairwise connected and the support of $\mu_{y,z}$ is full.
    \item There are master embeddings $\sigma,\gamma,\phi$ for $\mu$ into an Abelian group $(H,+)$ that are saturated,
    and the distribution of $(\sigma(x),\gamma(y),\phi(z))$ where $(x,y,z)\sim \mu$ has full support on $\{(a,b,c)\in H^3~|~a+b+c = 0\}$.
  \end{enumerate}
  Then, if $f\colon\Sigma^n\to \mathbb{C}$, $g\colon\Gamma^n\to \mathbb{C}$ and $h\colon\Phi^n\to \mathbb{C}$ are $1$-bounded functions such that
  ${\sf NEStab}_{1-\delta}(g;\mu_y^{\otimes n})\leq \delta$,
  then
  \[
  \card{\Expect{(x,y,z)\sim \mu^{\otimes n}}{f(x)g(y)h(z)}}\leq M\delta^{\eta}.
  \]
\end{thm}

\section{Maximality, Merging Symbols and Some Motivating Examples}\label{sec:max_merg_mot}
We do not know how to prove Theorem~\ref{thm:nonembed_deg_must_be_small} directly, and our argument instead proceeds by further reducing this statement
to a similar looking statement in which the distribution $\mu$ has additional useful properties.
In this section we present two important ideas/ properties that are crucial in this reduction, which are called ``merging
symbols'' and ``maximality''. We then give a few examples of arguments that could be carried out using these notions, often
making additional assumptions on $\mu$ (which we are not going to have in our formal argument in
Sections~\ref{sec:base_case},~\ref{sec:reudce_to_homogenous}). We do this so as to demonstrate typical scenarios in which these ideas are useful while
avoiding gory technicalities. As such, the language in this section will be informal at times, and we often appeal to intuition instead of making precise arguments.

Once we have explained these concepts, we will turn our attention into discussing the so-called ``base case'' of Theorem~\ref{thm:nonembed_deg_must_be_small}.
By that, we mean a specialized statement in the setting of Theorem~\ref{thm:nonembed_deg_must_be_small} for uni-variate functions, in which one manages to
prove that the expectation is consideration is significantly smaller than $1$ (for functions whose $2$-norm is at most $1$). We will discuss the ``ideal base case'',
which is a hypothetical scenario that we are not actually able to ensure; nevertheless, if such scenario were to hold, our argument would greatly simplify,
and intuitively Theorem~\ref{thm:nonembed_deg_must_be_small} would follows from the base case form some tensorization argument. Once again, our focus in
this section will be to explain how ``merges'' and ``maximality'' facilitate such arguments.

Finally, after exploring the ``ideal base case'' scenario we will explain the issue that may arise,
which we refer to as the ``Horn-SAT'' obstruction. We will explain the high level idea of how this issue is dealt with via
what we call the ``relaxed base case'' and the intuition to why this relaxed base case should suffice for the purpose of proving
Theorem~\ref{thm:nonembed_deg_must_be_small}. Once again, our focus here will be in explaining how the concept of ``merges'' and ``maximality''
fit together with the relaxed base case.


\subsection{Merging Symbols}
The first operation we discuss is the merge operation. Suppose that we have a distribution $\mu$ as in Theorem~\ref{thm:nonembed_deg_must_be_small}
in which there are distinct symbols $x$ and $x'$ such that there are common $y\in \Gamma$ and $z\in \Phi$ for which $(x,y,z)$ and $(x',y,z)$ are both
in the support of $\mu$. Intuitively, this means that in coordinates wherein the $g$ function gets $y$ and the $h$ functions gets $z$, some
non-trivial averaging of the $f$ function still occurs (as both $x$ and $x'$ are still possible). Naturally, averaging a function decreases its
$2$-norm, and we expect there to be a constant fraction of the coordinates in which even after fixing $y$ and $z$ there is still uncertainty
whether $x$ or $x'$ occur in the corresponding coordinate in the function $f$. It follows that if the value of the function $f$ ``heavily distinguishes''
between the symbols $x$ and $x'$ (in the sense that its value changes drastically if we change some coordinates in which $x$ occurs to be $x'$)
then the expectation in Theorem~\ref{thm:nonembed_deg_must_be_small} is small based solely on the fact that some non-trivial averaging occurs over $f$.

Following this line of reasoning leads one to speculate that one may assume that the function $f$ does not distinguish between the two
symbols $x$ and $x'$ (as otherwise the statement is trivial), in which case one may as well treat them as the same symbol.
The goal of the merge operation is to precisely capture this idea, and we formally present it below.
\begin{definition}
  Let $\Sigma$, $\Gamma$, $\Phi$ be finite alphabets and let $P\subseteq \Sigma\times \Gamma\times \Phi$. We say that $y,z$ imply $x$ in $P$
  if for all $y\in\Gamma$, $z\in\Phi$ there is at most a single $x\in\Sigma$ such that $(x,y,z)\in P$.
  We say that $y,z$ imply $x$ in a distribution $\mu$ over $\Sigma\times\Gamma\times \Phi$ if $y,z$ imply $x$ in ${\sf supp}(\mu)$.
\end{definition}
  If the value of any two coordinates implies the third, we say a distribution $\mu$ is fully merged:
\begin{definition}\label{def:fully_merged}
  Let $\Sigma$, $\Gamma$, $\Phi$ be finite alphabets and let $\mu$ be a distribution over $\Sigma\times \Gamma\times \Phi$. We say
  that $\mu$ is fully merged if (1) $y,z$ imply $x$, (2) $x,z$ imply $y$, and (3) $x,y$ imply $z$.
\end{definition}

Fix finite alphabets $\Sigma$, $\Gamma$, $\Phi$ and a distribution $\mu$ over $\Sigma\times\Gamma\times\Phi$.
To define the merge operation more precisely consider the graph
$G_{x,\mu,{\sf merge}}$ over $\Sigma$, wherein $x$ and $x'$ are adjacent if there are common $y\in\Gamma$ and $z\in \Phi$ such that
$(x,y,z)$ and $(x',y,z)$ are both in ${\sf supp}(\mu)$. Note that $G_{x,\mu,{\sf merge}}$ has $\card{\Sigma}$ connected components
if and only if $y,z$ imply $x$ (in which case the merge operation will do nothing), and by the above logic one should think of each
connected component of $G_{x,\mu,{\sf merge}}$ as a single symbol. More formally, given $\mu$ we may define the distribution $\mu'$
as follows:
\begin{definition}\label{def:merge}
  Let $\Sigma$, $\Gamma$, $\Phi$ be finite alphabets, let $\mu$ be a distribution over $\Sigma\times\Gamma\times\Phi$ in which $y,z$ does not
  imply $x$, and let $G_{x,\mu,{\sf merge}}$ be the graph above.
  We define the $x$-merged distribution $\mu'$ over $\Sigma'\times \Gamma\times \Phi$ where $\Sigma' \subsetneq \Sigma$ as: let the connected
  components of $G_{x,\mu,{\sf merge}}$ be $C_1,\ldots,C_{\ell}$ and choose a distinguished element $x_i^{\star}$ from each connected component.
  To sample according to $\mu'$, we sample $(x,y,z)\sim \mu$, take $i$ to be the connected components of $x$, and then output $(x_i^{\star}, y, z)$.
\end{definition}
In words, the merge distribution should be thought of as combining each connected component of $G_{x,\mu,{\sf merge}}$ into a single alphabet symbol.
\begin{remark}
A few remarks are in order.
  \begin{enumerate}
    \item In Definition~\ref{def:merge} we have defined the $x$-merge operation, and we will also use the $y$-merge operation and $z$-merge
    operation that are defined in an analogous way.
    \item The distributions $\mu$ and $\mu'$ are very closely related, and it is easy to observe that if $\mu$ satisfies all of the conditions of
        Theorem~\ref{thm:nonembed_deg_must_be_small},then $\mu'$ also satisfies all of the conditions of that theorem.
        We omit the straightforward proof, but remark that a master embedding of $\mu$ is translated to a master embedding of $\mu'$ in
        the obvious way (restriction).
    \item The merge operation has already made an appearance in~\cite{BKMcsp1,BKMcsp2}, however in that context the notions of degree/ noise stability are different.
    Here, our notions are more intricate and thus working with the merge operation requires a bit more care.
  \end{enumerate}
\end{remark}

\subsubsection{Simplifying Distributions via Merges}
The following two lemmas allow us to reduce the proof of Theorem~\ref{thm:nonembed_deg_must_be_small}
for a distribution $\mu$, to the proof of the same statement for a simpler distribution $\mu'$ which is a merge of $\mu$.
The first lemma handles $x$-merges, and analogously $z$-merges (as $x$ and $z$ are symmetric in the statement of
Theorem~\ref{thm:nonembed_deg_must_be_small}).
\begin{lemma}\label{lem:merge}
  Let $\Sigma$, $\Gamma$ and $\Phi$ be finite alphabets, and let $\mu$ be a distribution over $\Sigma\times\Gamma\times \Phi$
  for which the support of $\mu_{y,z}$ is full,
  and let $\mu'$ be the $x$-merged distribution coming from $\mu$.
  If the conclusion of Theorem~\ref{thm:nonembed_deg_must_be_small} holds for $\mu'$, then it also holds for $\mu$.
\end{lemma}
\begin{proof}
  Deferred to Section~\ref{sec:merge}.
\end{proof}

The second lemma handles $y$-merges; the proof is similar to the proof of Lemma~\ref{lem:merge} but some additional care is needed.
\begin{lemma}\label{lem:merge2}
  Let $\Sigma$, $\Gamma$ and $\Phi$ be finite alphabets, and let $\mu$ be a distribution over $\Sigma\times\Gamma\times \Phi$
  for which the support of $\mu_{y,z}$ is full,
  and let $\mu'$ be the $y$-merged distribution coming from $\mu$.
  If the conclusion of Theorem~\ref{thm:nonembed_deg_must_be_small} holds for $\mu'$, then it also holds for $\mu$.
\end{lemma}
\begin{proof}
  Deferred to Section~\ref{sec:merge2}.
\end{proof}

Lemmas~\ref{lem:merge},~\ref{lem:merge2} will be important for us later on, when we reduce the statement of Theorem~\ref{thm:nonembed_deg_must_be_small}
to a similar statement in which the distribution $\mu$ is fully merged. We note that for symbols $x,x'\in\Sigma$ that are mapped to different group elements
via the master embedding, that is, $\sigma(x)\neq \sigma(x')$, we could never identify $x$ and $x'$ by a merge. The reason is that if this was possible,
then there would be distinct tuples $(x,y,z)$ and $(x',y,z)$ in the support of $\mu$, so that by the definition of the master embeddings we have
\[
\sigma(x) + \gamma(y) + \phi(z) = 0,
\qquad
\sigma(x') + \gamma(y) + \phi(z) = 0,
\]
and it follows that it must be the case that $\sigma(x) = \sigma(x')$. Thus, merges will never decrease the alphabet sizes below $\card{H}$,
and at some point it will be important for us to consider how large are the alphabet sizes in comparison to the size of $H$.
\subsubsection{What are Merges Good For?}
As explained, the idea of merges will be crucial for us and to illustrate it we consider a special case that nevertheless illustrates
some important concepts. Besides the fact that if $x$ and $x'$ are merge-able then $\sigma(x) = \sigma(x')$,
we are not aware of any other clear obstructions to merges. Thus, a special case that one may consider is what happens when the merge operation successfully
reduces some of the alphabet sizes all the way down to $\card{H}$.

Suppose that after merging, the alphabet $\Sigma$ and $\Phi$ both have size exactly $\card{H}$ so that each one of them could be identified with the Abelian group $H$.
In that case, the partial basis we have set for $x$-functions and $z$-functions
is not partial but rather a full basis, and one expects that multiplying $f(x)h(z)$ one would get an embedding function over $y$. As the non-embedding
stability of $g$ is small, it has small mass on monomials which are embedding functions, and thus it should be the case that the correlation of
$g(y)$ and $f(x)h(z)$ is small. Indeed, arguments along these lines can be made -- and are indeed are crucial in Section~\ref{sec:base_case}.
To illustrate such ideas, below we show an argument along these lines under several additional assumptions on the distribution $\mu$.
\begin{lemma}\label{lem:demonstrate_merge}
  Suppose that $\mu$ is a distribution as in Theorem~\ref{thm:nonembed_deg_must_be_small}, and further suppose that $\card{\Sigma} = \card{\Phi} = H$,
  that the distribution of $(\sigma(x),\gamma(y),\phi(z))$ is uniform over $\sett{(a,b,c)\in H^3}{a+b+c = 0}$ and
  that $\mu_{y,z}$ is uniform, where $(x,y,z)\sim \mu$. Then
  the conclusion of Theorem~\ref{thm:nonembed_deg_must_be_small} holds.
\end{lemma}
\begin{proof}
  Relabeling $\Sigma$ and $\Phi$, we may assume that they are both equal to $H$ and that the master embeddings $\sigma$ and $\phi$ are the identity.
  Define $G(y) = \cExpect{(x',y',z')\sim\mu^{\otimes n}}{y' = y}{f(x')h(z')}$, and expand
  \[
  f(x) = \sum\limits_{\chi\in\hat{H}^{\otimes n}}\widehat{f}(\chi)\chi(x),
  \qquad
  h(z) = \sum\limits_{\chi'\in \hat{H}^{\otimes n}}\widehat{h}(\chi')\chi'(z).
  \]
  Then
  \[
  G(y) =
  \sum\limits_{\chi,\chi'\in\hat{H}^{\otimes n}}\widehat{f}(\chi)\widehat{h}(\chi')
  \cExpect{(x',y',z')\sim\mu^{\otimes n}}{y' = y}{\chi(x')\chi'(z')}.
  \]
  Note that conditioned on $y$, the distribution over $x'$ and $z'$ is uniform such that $x' + \gamma(y) + z' = 0$.
  Thus, the expectation is $0$ if $\chi\neq \chi'$ and otherwise is equal to
  $\chi(x'+z') = \chi(-\gamma(y))$. Thus,
  \[
  G(y) =
  \sum\limits_{\chi\in\hat{H}^{\otimes n}}\widehat{f}(\chi)\widehat{h}(\chi)\chi(-\gamma(y)).
  \]
  It follows that $G$ is an embedding function, that is, it is from ${\sf Embed}_{\gamma}(\mu^{\otimes n})$, and hence it cannot have large correlation
  with $g$. More precisely, the left hand side of Theorem~\ref{thm:nonembed_deg_must_be_small} is equal to
  \[
    \inner{G}{\overline{g}}
    =\inner{\T_{\text{non-embed}, \sqrt{1-\delta}}G}{\overline{g}}
    =\inner{G}{\T_{\text{non-embed}, \sqrt{1-\delta}}^{*}\overline{g}}
    =\inner{G}{\T_{\text{non-embed}, \sqrt{1-\delta}}\overline{g}}.
  \]
  Using Cauchy-Schwarz, this is at most $\norm{G}_2\norm{\T_{\text{non-embed}, \sqrt{1-\delta}}\overline{g}}_2 \leq {\sf NEStab}_{1-\delta}(\overline{g})^{1/2}\leq \sqrt{\delta}$.
\end{proof}
Lemma~\ref{lem:demonstrate_merge} should be thought of as saying that if we managed to reduce the alphabet sizes of $x$ and $z$ all the way to their minimal point,
which is the size of $H$, then we managed to prove Theorem~\ref{thm:nonembed_deg_must_be_small}. There are several ways to think about this; one way is as a sort of
base case of an inductive statement, in which we are trying to prove Theorem~\ref{thm:nonembed_deg_must_be_small} by induction on several parameters, one of which
are the alphabet sizes of $x$ and $z$ (or rather their sum).

Our presentation below will be somewhat different (but morally equivalent), and we will use the idea of Lemma~\ref{lem:demonstrate_merge} to
be able to assume that the size of the alphabet of $x$ exceeds $H$, so that (by the pigeonhole principle)
there are two distinct $x$ symbols that are mapped to the same group element in $H$. Instead of directly trying to argue inductively based on this parameter,
we will use these distinct elements, and the fact that $\mu$ is fully merged, to relate expectations as in Theorem~\ref{thm:nonembed_deg_must_be_small} to
expectations over different distributions $\mu'$ that have richer supports. In other words, in some scenarios we will be able to use Lemma~\ref{lem:demonstrate_merge}
to enrich the support of the distribution $\mu$. Intuitively, this marks significant progress since if we were able to enrich the support of $\mu$ indefinitely,
we would eventually reach a distribution rich enough so that the statement of Theorem~\ref{thm:nonembed_deg_must_be_small} becomes obvious.
To capture this idea (and avoiding explicit iterative arguments of this nature), in the next section we define the notion of maximality of a distribution.

\subsection{Maximality of Distributions}
The next concept we describe is maximality of distributions. There are several equivalent ways of thinking about it, and ultimately it is just a form of induction.
However, one good way to think about it is by assuming that there is a distribution $\mu$ which serves as a counter example for Theorem~\ref{thm:nonembed_deg_must_be_small},
and then trying to come up with a distribution $\nu$ which is a ``maximal'' counter example for the statement, in the sense that it is still a counter example to the statement
but adding any additional atoms to the support of $\nu$ would lead to a distribution for which the statement is true. More formally:
\begin{definition}\label{def:maximal}
  Let $\Sigma$ , $\Gamma$ and $\Phi$, and let $\mu$ be a distribution over $\Sigma\times \Gamma\times \Phi$.
  We say that $\mu$ is maximal if for all distributions $\nu$ over $\Sigma\times\Gamma\times \Phi$ such that
  ${\sf supp}(\mu)\subsetneq {\sf supp}(\nu)$, the statement in Theorem~\ref{thm:nonembed_deg_must_be_small} holds.
\end{definition}

The notions of maximality and merges often work in synergy together. Merges allow us to argue that certain atoms do not exist in
$\mu$. For example, if we know about distinct atoms $(x,y,z)$ and $(x',y',z')$ that are both in the support of $\mu$ where $x\neq x'$, and
we know that in $\mu$ the value of $y,z$ implies $x$ (because we already applied some merges to guarantee that), then we automatically
can conclude that $(x',y,z)$ is not in the support of $\mu$. Thus, whenever we are facing similar looking expectation to
Theorem~\ref{thm:nonembed_deg_must_be_small} which involves a distribution $\nu$ that contains the support of $\mu$ and additionally
$(x',y,z)$, we can appeal to the maximality of $\mu$ and upper bound it as in Theorem~\ref{thm:nonembed_deg_must_be_small}.

\subsubsection{What is Maximality Good For?}
Below, we give a concrete example of such synergy, and once again the argument below is informal in nature.\footnote{Similar instantiations of this synergy appear numerous times in Sections~\ref{sec:base_case},~\ref{sec:reudce_to_homogenous},
and we refer the reader to there for a more formal presentation.} In Theorem~\ref{thm:nonembed_deg_must_be_small} the assumption
that the non-embedding stability of $g$ is small really amounts to saying that $g$ has almost all of its mass on monomials
whose non-embedding degree is at least $\Theta(1/\delta)$. Intuitively, one expects that then the part of $f$ of non-embedding degree
significantly less -- say $O(1/\delta^{0.99})$, would not contribute much to the expectation. Ignoring the $h$ function for a
moment, this would be clear as then there is simply a mis-match of degrees and one could appeal to the orthogonality of non-embedding
functions to embedding functions. Now taking the $h$ function into account once again, there doesn't seem to be a way for it to compensate
for the large mis-match of the non-embedding degrees coming from the $f$ function and those coming from the $g$ function.

Indeed, in the following lemma we show that such expectations must indeed be small, and to do so we appeal to the notion of maximality;
once again, we make some additional simplifying assumptions.
\begin{lemma}\label{lem:demonstrate_maximality}
  Suppose that $\mu$ is a distribution as in Theorem~\ref{thm:nonembed_deg_must_be_small}, and further suppose that $\mu$ is maximal,
  fully merged and $\card{\Sigma} >  \card{H}$. Then for all $\xi>0$ there is $\eta = \eta(m,\alpha,\xi)>0$ such that
  \[
        \card{\Expect{(x,y,z)\sim \mu^{\otimes n}}{\T_{{\sf non-embed}, 1- \delta^{1-\xi}}^{\otimes n} f(x)g(y)h(z)}}\leq M\delta^{\eta}.
  \]
\end{lemma}
\begin{proof}
 We may re-interpret the expectation in consideration as
 $\Expect{(x,y,z)\sim \mu'^{\otimes n}}{f(x)g(y)h(z)}$
 wherein the distribution $\mu'$ is defined by taking $(x,y,z)\sim \mu$,
 then $x'\sim \T_{\text{non-embed}, 1-\delta^{1-\xi}} x$, $y' = y$ and $z' = z$
 and outputting $(x',y',z')$. The distribution $\mu'$ has small atoms with probability $\rho = \alpha\delta^{1-\xi}$,
 and to circumvent that we use random restrictions.

 We may write
 $\mu' = \frac{\rho}{2}\mu'' + \left(1-\frac{\rho}{2}\right)\mu'''$ wherein $\mu''$, $\mu'''$ have
 the same support as $\mu'$ and in $\mu''$ the probability of each atom is at least $\alpha' = \alpha'(\alpha) > 0$,
 and now use random restrictions to switch from the probability measure $\mu'$ to the probability measure $\mu''$.
 Sample $J\subseteq_{\rho/2} [n]$ and choose the value of coordinates in $\overline{J}$ according to $\mu'''^{\overline{J}}$. That is, sample $(\tilde{x},\tilde{y},\tilde{z})\sim \mu'''^{\overline{J}}$
 so that the expectation in hand can be seen as an averaging over these restrictions of a similar looking expectation over $\mu''$:
 \[
 \Expect{(x,y,z)\sim {\mu'}^{\otimes n}}{f(x)g(y)h(z)}
 =
 \Expect{\substack{J\subseteq_{\rho/2}[n]\\ (\tilde{x},\tilde{y},\tilde{z})\sim {\mu'''}^{\overline{J}}}}{\Expect{(x',y',z')\sim {\mu''}^{J}}{f'(x')g'(y')h'(z')}}.
 \]
 As most of the mass of $g$ lies on monomials of non-embedding degree at least $\Theta(1/\delta)$, one expects most mass of $g'$
 to lie on non-embedding degree at least $\Theta(\rho/\delta) = \Theta(\delta^{-\xi})$; this is indeed true and can be argued by
 appealing to Claim~\ref{claim:rr_nestab}, but we do not elaborate on it for now. Hence the inner expectation is an expectation
 of the same type as in Theorem~\ref{thm:nonembed_deg_must_be_small}, except that it is over the distribution $\mu''$.

 We are now going to appeal to the maximality of $\mu$ to argue that this inner expectation must be small, and for that we must argue
 that the support of $\mu''$ strictly contains the support of $\mu$. By definition, the support of $\mu$ is contained in the support of $\mu'$,
 and as the support of $\mu''$ is the same as that of $\mu'$, we get that ${\sf supp}(\mu)\subseteq {\sf supp}(\mu')$. Next, we argue
 that this is a strict containment. As $\card{\Sigma} > \card{H}$, by the pigeonhole principle there are distinct $a,a'\in \Sigma$ such that $\sigma(a) = \sigma(a')$, and
 we take $b'\in\Gamma$ and $c'\in\Phi$ such that $(a',b',c')\in {\sf supp}(\mu)$. By definition of $\mu'$ it follows that $(a,b',c')$ is in ${\sf supp}(\mu')$, and
 hence it is in the support of $\mu''$. Note that $(a,b',c')\not\in {\sf supp}(\mu)$, as otherwise the symbols $a$ and $a'$ could be merged, but by assumption the distribution
 $\mu$ is already fully merged.

 We may thus appeal to the maximality of $\mu$ and get the inner expectation is upper bounded by $M\delta^{\xi\eta}$ where
 $\eta>0$ and $M$ only depend on $\alpha$ and $m$, concluding the proof.
\end{proof}

The main benefit of Lemma~\ref{lem:demonstrate_maximality} is that it allows us to assume that not only the $g$ function in Theorem~\ref{thm:nonembed_deg_must_be_small}
has large non-embedding degree, but rather also the $f$ and $h$ functions. Indeed, by Lemma~\ref{lem:demonstrate_maximality} it follows that the expectation
in Theorem~\ref{thm:nonembed_deg_must_be_small} is very close to
\[
\card{\Expect{(x,y,z)\sim \mu^{\otimes n}}{(I-\T_{{\sf non-embed}, 1- \delta^{1-\xi}}^{\otimes n}) f(x)g(y)h(z)}},
\]
and now the new $f$ function, which is $(I-\T_{{\sf non-embed}, 1- \delta^{1-\xi}}^{\otimes n}) f$, can be seen to have high non-embedding degree (while importantly also
remaining bounded). As $f$ and $g$ now have high non-embedding degree, one could make a direct argument saying that the part of $h$ of small non-embedding degree
also has negligible contribution to the above expectation, and hence one is reduced to upper bounding an expectation of the form
\[
\card{\Expect{(x,y,z)\sim \mu^{\otimes n}}{(I-\T_{{\sf non-embed}, 1- \delta^{1-\xi}}^{\otimes n}) f(x)g(y) (I-\T_{{\sf non-embed}, 1- \delta^{1-\xi}}^{\otimes n})h(z)}}.
\]
\paragraph{Proof by tensorization: motivating the base case.}
We will attempt to prove this statement by a tensorization argument, inducting on $n$. Ideally, we would have liked to argue that the worst-case
high-degree non-embedding functions for the above expectation are simply product functions $\prod\limits_{i=1}^{1/\delta} u_i(x_i)$,
$\prod\limits_{i=1}^{1/\delta} v_i(y_i)$ and $\prod\limits_{i=1}^{1/\delta} w_i(z_i)$, in which case the task of proving an upper bound as above reduces to a problem
about univariate functions. There is one important distinction, however, which is that while the $f,g$ and $h$ functions are bounded (in $\ell_{\infty}$-norm),
the functions $u_i,v_i$ and $w_i$ need not be bounded in $\ell_{\infty}$, and instead we will only be able to guarantee $\ell_2$ boundedness. Thus, while
the ``base case'' of the tensorization statement seems much easier than the original statement, we have to prove it for a more general class of functions.

Below, we begin the discussion regarding this base case statement, and present the ``ideal base case'' scenario in which this logic is much simpler, however
which we are not able to guarantee. Nevertheless, exploring the ideal base case scenario carries with it a lot of useful intuition, and in particular further
relations between the utilization of a base case statement and the concepts of ``merges'' and ``maximality''.

\subsection{An Ideal Base Case Scenario in Theorem~\ref{thm:nonembed_deg_must_be_small}}
With the notions of maximality and merging in hand, we can now present an ideal setting in which case the
intuition behind the proof of Theorem~\ref{thm:nonembed_deg_must_be_small} is relatively simple.

Consider a distribution $\mu$ as in Theorem~\ref{thm:nonembed_deg_must_be_small}, and suppose that it is fully merged
as well as maximal. Further suppose that the marginal distribution $\mu_{y,z}$ is uniform over $\Gamma\times \Phi$.
In this case, one notes that for $f\colon \Sigma\to\mathbb{C}$, $g\colon \Gamma\to\mathbb{C}$ and $h\colon \Phi\to\mathbb{C}$
of $2$-norm equal to $1$, one has that
\[
\card{\Expect{(x,y,z)\sim \mu}{f(x)g(y)h(z)}}
\leq
\sqrt{\Expect{(x,y,z)\sim \mu}{\card{f(x)}^2}}
\sqrt{\Expect{(x,y,z)\sim \mu}{\card{g(y)}^2\card{h(z)}^2}}
=\norm{f}_2\norm{g}_2\norm{h}_2
\leq 1,
\]
where we used Cauchy-Schwarz and the fact that $\mu_{y,z}$ is uniform.

Inspecting equality cases for Cauchy-Schwarz, one notes that equality holds if and only if
there is a constant $\theta\in\mathbb{C}$ with absolute value $1$ such that
$f(x) = \theta g(y)h(z)$, and for simplicity we assume that $\theta = 1$. Thus,
equality holds if and only if the triplet $f$, $g$ and $h$ satisfy an
embedding-like equation, except that it is not clear which group one should take.
Taking the principle branch of the logarithm, one gets that
$\log(f(x)) = \log(g(y)) + \log(h(z))\pmod{2\pi {\bf i}}$, and now this is
indeed an equation over an Abelian group (albeit infinite, but we already
saw in Lemma~\ref{lem:turn_infinite_to_finite} that one can convert such
embeddings into finite Abelian group embeddings).
Thus, it follows that the logs form an embedding of $\mu$ into an Abelian group,
and by Claim~\ref{claim:embedding_fns} one may conclude that each one of them
is an embedding function, hence $f\in {\sf Embed}_{\sigma}(\mu)$,
$g\in {\sf Embed}_{\gamma}(\mu)$ and $h\in {\sf Embed}_{\phi}(\mu)$.

In words, we have argued that if univariate functions achieve perfect value of
$\card{\Expect{(x,y,z)\sim \mu}{f(x)g(y)h(z)}}$, then they are embedding functions.
This motivates the following statement, which we refer to as the ``ideal base case'':
\begin{statement}\label{statement:idea_base_case}
  Let $\mu$ be a distribution as above. Then for all $\tau>0$ there is $\lambda>0$
  such that if $f\colon \Sigma\to\mathbb{C}$, $g\colon \Gamma\to\mathbb{C}$ and $h\colon \Phi\to\mathbb{C}$
  are functions with $2$-norm equal to $1$ and $\norm{{\sf Proj}_{{\sf Embed}_{\sigma}(\mu)}(f)}_2\leq 1-\tau$,
  then
  \[
    \card{\Expect{(x,y,z)\sim \mu}{f(x)g(y)h(z)}}\leq 1-\lambda.
  \]
\end{statement}
In words, Statement~\ref{statement:idea_base_case} asserts that if $f$ is somewhat far from all embedding functions,
then the value of $\card{\Expect{(x,y,z)\sim \mu}{f(x)g(y)h(z)}}$ must be bounded away from $1$. Indeed, in light
of the above analysis this is something which is natural to expect; we examined the equality case, and by compactness
type argument it follows that near equality cases can be characterized as ``near embedding functions''.
\footnote{We remark that in our actual argument we are going to need a decent quantitative dependency between
the parameters $\tau$ and $\lambda$, typically a polynomial dependency. Thus, we will not be able to directly
use compactness arguments and we will have to unravel them.}

We take a moment to clarify that Statement~\ref{statement:idea_base_case} as stated is false in general.
The issue in the above logic is that $f$, $g$ and $h$ may take the value $0$ sometime, in which case we cannot apply
the log function, and this turns out to be a rather serious obstacle referred to as the Horn-SAT obstruction.
To resolve this issue we introduce the so-called \emph{Relaxed Base Case}, which we give intuition to in Section~\ref{sec:relaxed_base_case_introduce}
and which is the primary topic of Section~\ref{sec:base_case}.

Having said that, considering the class of distributions $\mu$ satisfying Statement~\ref{statement:idea_base_case} is helpful,
and we now work under the assumption that it holds.
In that case, given $\mu$ and $f$, $g$ and $h$ as in Theorem~\ref{thm:nonembed_deg_must_be_small}, we know that $g$ has almost all of its $\ell_2$ mass on monomials of non-embedding
degree at least $1/\delta$. Using Lemma~\ref{lem:demonstrate_maximality} we may also truncate the low-degree non-embedding degrees of both
$f$ and $h$, so that eventually we need to upper bound an expectation of the form (for simplicity of notation we ignored the $\xi>0$ therein)
\begin{equation}\label{eq:motivate_relaxed}
    \card{\Expect{(x,y,z)\sim \mu}{(I - \T_{\text{non-embed}, 1-\delta}^{\otimes n})f(x)g(y)(I-\T_{\text{non-embed}, 1-\delta}^{\otimes n})h(z)}}.
\end{equation}
In conclusion, we are now reduced to working with the functions $f' = (I - \T_{\text{non-embed}, 1-\delta}^{\otimes n})f$,
$g$ and $h' = (I - \T_{\text{non-embed}, 1-\delta}^{\otimes n})h$, which all have almost all of their mass on monomials with
non-embedding degree at least $1/\delta$.

As the base case gives some a gain of $1-\lambda$ over the trivial bound when we have a
univariate function with some non-embedding components, we expect to make this gain $1/\delta$ times,
once for each non-embedding component in $f'$, $g$ and $h'$; in total, this would yield a
bound of $(1-\lambda)^{1/\delta}$, which is satisfactory for us (and even much better than what
we're shooting for).

If Statement~\ref{statement:idea_base_case} was true, this argument would not be too
far from the truth, and in fact can be made rigorous to work. Alas, as we said it could be the
case that there are non-embedding functions $f$, $g$ and $h$ such that $f(x) = g(y) h(z)$, but
then it is necessarily the case that the function $f$ must vanish somewhere; we refer to such
illegitimate-looking embeddings as Horn-SAT embeddings, and to the existence of which as the
Horn-SAT obstruction.

To bypass the Horn-SAT obstruction we must study the possible vanishing patterns of the function $f$.
We do not know how to argue about this for the distribution $\mu$ itself, and hence we have to once
again move to a closely related distribution $\mu'$ (which is obtained from $\mu$ from a combination
of more path tricks and merges), in which we are able to assert non-trivial information about the
$0$-sets of Horn-SAT embeddings. We defer the precise description of this reduction to Section~\ref{sec:base_case},
and in the next section we give some high level overview of the relaxed base case we are able to guarantee,
how to work with it and the way that it fits in together with the notions of ``merges'' and ``maximality''.

\subsection{On the Relaxed Base Case Scenario in Theorem~\ref{thm:nonembed_deg_must_be_small}}\label{sec:relaxed_base_case_introduce}
\subsubsection{A Naive Relaxed Base Case}
Let $\mu$ be a distribution over $\Sigma\times \Gamma\times \Phi$ as above, and assume that $\card{\Sigma} > \card{H}$.\footnote{We remark
that to justifying this assumption is precisely where results in the spirit of Lemma~\ref{lem:demonstrate_merge} come in handy.}
After a suitable transformation of the distribution $\mu$ into a distribution $\tilde{\mu}$
over $\tilde{\Sigma}\times \tilde{\Gamma}\times \tilde{\Phi}$, we are (morally) able to make the following guarantee:
\begin{statement}\label{statement:idea_relaxed_base_case}
  There exists $\Sigma'\subseteq \tilde{\Sigma}$ of size larger than $\card{H}$,
  such that if $f\colon \tilde{\Sigma}\to\mathbb{C}$, $g\colon \tilde{\Gamma}\to\mathbb{C}$ and $h\colon \tilde{\Phi}\to\mathbb{C}$
  are functions such that $f(x) = g(y)h(z)$ in the support of $\tilde{\mu}$, then $f|_{\Sigma'} \equiv 0$.
\end{statement}
In words, Statement~\ref{statement:idea_relaxed_base_case} tells us that in $\tilde{\mu}$, all Horn-SAT embeddings must be
$0$ on $\Sigma'$. This motivates to attempt to formulate an analog of Statement~\ref{statement:idea_base_case} that instead
of assuming non-trivial projection outside the subspace of embedding functions, assumes some variance of $\Sigma'$. And indeed,
such statement can be proved to be true, but as is it is not very useful for us.

To be more specific, if we took that route and tried to write down an analog of~\eqref{eq:motivate_relaxed}, we would have to define
a notion of degree that corresponds to not being $0$ on $\Sigma'$ and attempt to reduce ourselves to the case where this new notion of degree
for the function $f$ is large (so that we will be able to assert that we are avoiding the Horn-SAT obstruction on many coordinates, hence gaining
some $1-\lambda$ factor). At that point it is important though to keep the function $f$ bounded, and hence to execute this
logic we would need to define an averaging operator corresponding to the Markov chain that mixes inside the set $\Sigma'$
and stays put on elements in $\tilde{\Sigma}\setminus \Sigma'$; this is so that functions that are not constant on $\Sigma'$ would have their $2$-norm
decreased as a result of applying this averaging operator, so as to truncate of the part of $f$ that has low degree with respect to the new notion.

This averaging operator however is incompatible with non-embedding degrees and the non-embedding averaging operator.
The reason is that symbols in $\Sigma'$ may be mapped to different group elements in $H$, in which case embedding functions
also have variance on $\Sigma'$ and hence get their $2$-norm decreased by this averaging operator. We have no hope
of gaining any $1-\lambda$ factor from embedding functions, meaning that while identifying a property of Horn-SAT embeddings
that we can ensure not to happen, we would re-introduce embedding functions into the mix and thus still not have a proper base case to induct on.
There are other manifestations of this issue down the line if one pursues this
direction, but ultimately they all boil down to the fact that the above averaging operator is not necessary a ``sub-averaging
operator'' of our non-embedding operator $\T_{\text{non-embed},1-\delta}$. By that, we mean that there are functions which
the proposed averaging operator contracts, whereas $\T_{\text{non-embed},1-\delta}$ keeps in place.

\subsubsection{The Relaxed Base Case and Effective Non-embedding degrees}
To resolve this issue, we take a subset of $\Sigma'$ on which the master embedding is constant. As $\card{\Sigma'}>\card{H}$, by the pigeonhole
principle there are distinct $x,x'\in \Sigma'$ such that $\sigma(x) = \sigma(x')$, so that we can take $\Sigma_{{\sf modest}} = \{x,x'\}\subseteq \Sigma'$.
We then have the following relaxed form of our ideal base case from above:
\begin{statement}\label{statement:relaxed_base_case}
  For all $\tau>0$ there is $\lambda>0$ such that if $f\colon \tilde{\Sigma}\to\mathbb{C}$, $g\colon \tilde{\Gamma}\to\mathbb{C}$ and $h\colon \tilde{\Phi}\to\mathbb{C}$
  are functions with $2$-norm equal to $1$ such that $\Expect{x,x'\in\Sigma_{{\sf modest}}}{\card{f(x) - f(x')}^2}\geq \tau$,
  then
   \[
    \card{\Expect{(x,y,z)\sim \tilde{\mu}}{f(x)g(y)h(z)}}\leq 1-\lambda.
  \]
\end{statement}
In words, for univariate functions, if our function $f$ has a little bit of variance on $\Sigma_{{\sf modest}}$, then we immediately get a gain of
$1-\lambda$ over the trivial bound. We remark that the property
of having variance on $\Sigma_{{\sf modest}}$ immediately prohibits $f$ from being an embedding function (as any embedding function is constant
on $\Sigma_{{\sf modest}}$), as well as from being part of a Horn-SAT embedding (as any Horn-SAT embedding must vanish on $\Sigma'$ and hence
on $\Sigma_{{\sf modest}}\subseteq \Sigma'$).

This motivates defining a certain notion of degree, which we refer to as ``effective non-embedding degree''; we often abbreviate this and just say effective degree
instead.
The effective non-embedding degree of a monomial over $x$ is the number of coordinates on which the corresponding component has variance over $\Sigma_{{\sf modest}}$,
and intuitively this measures the number of times we will gain a factor of $1-\lambda$ by appealing to Statement~\ref{statement:relaxed_base_case}.
To make this definition more precise, we have to refine the basis we constructed consisting of embedding functions and non-embedding functions,
and set-up an orthonormal basis of $L_2(\tilde{\Sigma}, \tilde{\mu}_x)$ composed on: (1) embedding functions, (2) non-embedding functions that
are constant on $\Sigma_{{\sf modest}}$, and (3) non-embedding functions that have variance on $\Sigma_{{\sf modest}}$. Then, the effective (non-embedding) degree
of a monomial is the number of components in it of functions not constant on $\Sigma_{{\sf modest}}$.

\subsubsection{Working with the Relaxed Base Case}
Taking inspiration
from the above discussion, one is tempted to argue that just like in~\eqref{eq:motivate_relaxed} we managed to argue that the non-embedding degree of $f$ can
be assumed to be large, we should also manage to assume that the effective degree of $f$ is large. This is indeed possible, and to do so we identify
a proper Markov chain that captures effective degree; we refer to this Markov chain as ``the modest Markov chain'', and it is defined as follows.
On $a\in \tilde{\Sigma}\setminus \Sigma_{{\sf modest}}$ the chain stays in place, and on $a\in \Sigma_{{\sf modest}}$ the chain
re-samples a symbol from $\Sigma_{{\sf modest}}$ according to the marginal distribution of $\tilde{\mu}_x$ on $\Sigma_{{\sf modest}}$.

With the modest Markov chain in hand we can define a corresponding averaging operator,
$\mathrm{E}_{\text{non-embed}, 1-\delta}$ from $L_2(\tilde{\Sigma}, \tilde{\mu}_x)$ to
$L_2(\tilde{\Sigma}, \tilde{\mu}_x)$ defined as $\mathrm{E}_{\text{non-embed}, 1-\delta}f(x) = \Expect{x'\sim \mathrm{E}_{\text{non-embed}, 1-\delta} x}{f(x')}$,
where in $x'\sim \mathrm{E}_{\text{non-embed}, 1-\delta} x$ we take $x' = x$ with probability $1-\delta$, and otherwise we sample $x'$ according to the
modest Markov chain on $x$. This Markov chain can be shown to precisely capture the notion of effective degrees, and hence our task now is to justify that
we can assume that $f$ has high effective degree, in the sense that we can reduce the task of proving that $\card{\Expect{(x,y,z)\sim \tilde{\mu}^{\otimes n}}{f(x)g(y)h(z)}}$
is small to an analog of~\eqref{eq:motivate_relaxed} of the form:
\begin{equation}\label{eq:motivate_relaxed2}
    \card{\Expect{(x,y,z)\sim \mu}{(I - \mathrm{E}_{\text{non-embed}, 1-\delta}^{\otimes n})f(x)g(y)(I-\T_{\text{non-embed}, 1-\delta}^{\otimes n})h(z)}},
\end{equation}
is small. Towards this end we must argue that the contribution from the part of $f$ of small effective degree is small, and this is once again where maximality
and merges come into play.
\begin{lemma}\label{lem:demonstrate_maximality2}
  Suppose that $\tilde{\mu}$ is a distribution as above, and further suppose that $\tilde{\mu}$ is maximal
  and fully merged. Then for all $\xi>0$ there is $\eta = \eta(m,\alpha,\xi)>0$ such that
  \[
        \card{\Expect{(x,y,z)\sim \mu^{\otimes n}}{\mathrm{E}_{{\sf non-embed}, 1- \delta^{1-\xi}}^{\otimes n} f(x)g(y)h(z)}}\leq M\delta^{\eta}.
  \]
\end{lemma}
\begin{proof}
  The proof is almost identical to the proof of Lemma~\ref{lem:demonstrate_maximality}, and we only sketch it. Re-interpreting this expectation as
  $\Expect{(x,y,z)\sim {\mu'}^{\otimes n}}{f(x)g(y)h(z)}$ where the distribution $\mu'$ is the distribution in
  which we first sample $(x',y',z')\sim \tilde{\mu}$, then $x\sim \mathrm{E}_{{\sf non-embed}, 1- \delta^{1-\xi}}^{\otimes n}x'$,
  take $y=y'$, $z=z'$ and output $(x,y,z)$. As in Lemma~\ref{lem:demonstrate_maximality},
  the support of $\mu'$ strictly contains the support of $\tilde{\mu}$  and hence it makes sense to try to appeal to the maximality of $\tilde{\mu}$.
  The only issue is that in $\tilde{\mu}$ there are atoms with small probability $\rho = \alpha \delta^{1-\xi}$, and to bypass that we write
  $\mu' = \frac{\rho}{2}\mu'' + \left(1-\frac{\rho}{2}\right)\mu'''$ where $\mu''$ and $\mu'''$ have the same supports as $\mu'$ and in $\mu''$
  the probability of each atom is at least $\alpha'(\alpha)>0$, and then use random restrictions as in Lemma~\ref{lem:demonstrate_maximality}.
\end{proof}

\subsection{Some Additional Remarks on Combining These Ingredients}\label{sec:additional_remarks_max}
Throughout this section we have proved some useful lemmas regarding the interaction of expectations as in Theorem~\ref{thm:nonembed_deg_must_be_small}
and the notions of merges, maximality and how the relaxed base case fits in. Ideally, we would have liked to have a distribution $\mu$ that possess
all of the properties that we needed (on top of the ones assumed in Theorem~\ref{thm:nonembed_deg_must_be_small}):
(1) the alphabet of $x$ has size larger than $\card{H}$,
(2) $\mu_{y,z}$ is uniform,
(3) $\mu$ is fully merged,
(4) $\mu$ admits a relaxed base case statement as in Statement~\ref{statement:relaxed_base_case},
(5) $\mu$ is maximal.

We are not going to be able to ensure that all of these properties simultaneously occur for $\mu$. Instead, we will
argue that the distribution $\mu$ ``contains within it'' some other distribution $\nu$ (possibly on different alphabets)
on which some of these properties hold. More specifically, we are not going to be able to guarantee that $\mu$ is maximal,
and instead we will be able to argue that ``within it'' there is a maximal distribution. By that, we mean that are $\Sigma'$ and $\Phi'$,
a distribution $\nu$ over $\Sigma'\times \Gamma\times \Phi'$ which is maximal and maps $a\colon \Sigma'\to\Sigma$ and $c\colon \Phi'\to\Phi$ such that:
\begin{enumerate}
  \item \textbf{Containment:} $\sett{(a(x),y,c(z))}{(x,y,z)\in {\sf supp}(\nu)}\subseteq {\sf supp}(\mu)$;
  \item \textbf{Alignment of Master Embeddings:}
  taking $(\sigma,\gamma,\phi)$ to be a master embedding of $\mu$, we have that $(\sigma\circ a,\gamma, \phi\circ c)$
  is a master embedding of $\nu$.
\end{enumerate}
Intuitively, the reason that this is useful is that, after suitable random restrictions, we can relate expectations with respect to $\mu$ to
expectations with respect to $\nu$. Indeed, letting $\mu'$ be the condition distribution of $\mu$ on
$\sett{(a(x),y,c(z))}{(x,y,z)\in {\sf supp}(\nu)}$ and writing $\mu = \beta \mu' + (1-\beta)\mu''$ for some distribution $\mu''$ and $\beta = \beta(\alpha)>0$,
we can switch from the distribution $\mu$ to the distribution $\mu'$ (as in the proof of Lemma~\ref{lem:demonstrate_maximality}), and functions over this
domain can be lifted to functions over the domain of $\nu$:
\[
f'(x') = f(a(x'_1),\ldots,a(x'_n)),
\qquad
g'(y') = g(y')
\qquad
h'(z') = h(c(z'_1),\ldots,c(z'_n)).
\]
We thus managed to reduce the problem of bounding some expectation with respect to $\mu$ to the task of bounding some expectation with respect to
$\nu$. Often times, this line of reasoning (on top of arguments as above) will allow us to appeal to the maximality of $\nu$ (and get a result which
is qualitatively the same as if we could assume that $\mu$ itself is maximal).

\section{Arranging for a Base Case for Theorem~\ref{thm:nonembed_deg_must_be_small}}\label{sec:base_case}
In this section, we begin the proof of Theorem~\ref{thm:nonembed_deg_must_be_small}.
As explained earlier, the core of our argument will ultimately be by induction on $n$, and as such we are going to need a base case statement for functions over a single variable.
Our inductive process though will be unable to preserve $1$-boundedness and will only be able to give us $L_2$-bound guarantees. Therefore, the base case
we are looking for has to address functions with a bounded $L_2$-norm.

The most naive attempts at arriving at such base case lead one to a difficulty referred to as the ``Horn-SAT'' obstruction,
which refers to the possibility of a existence of a triplet of functions $f,g$ and $h$ that satisfy $f(x) = g(y)h(z)$
on the support on $\mu$ but that do not necessarily yield an Abelian embedding; this difficulty arises due to the fact that $f$ may be $0$ on some inputs.

In this section, our goal is to state a result that implies Theorem~\ref{thm:nonembed_deg_must_be_small}, and which is more amendable to a proof
by induction along the lines of~\cite{BKMcsp2}. To do so, we will first have to go through some reductions and present intermediate statements
which imply Theorem~\ref{thm:nonembed_deg_must_be_small}; roughly speaking, these statement will all be similar to the statement of Theorem~\ref{thm:nonembed_deg_must_be_small}
with additional assumptions on the distribution $\mu$.

\subsection{Further Preprocessing of the Distribution $\mu$: Pushing Counter-examples to the Extreme}

At a high level, the goal of our preprocessing step is to arrive at a distribution $\mu'$ (which may be different from $\mu$) such
that if Theorem~\ref{thm:nonembed_deg_must_be_small} is false for $\mu$, then it is also false for $\mu'$, but moreover
$\mu'$ is the ``richest'' distribution on which the statement remains false. That is, if we consider any distribution $\nu$
satisfying the conditions of Theorem~\ref{thm:nonembed_deg_must_be_small} whose support strictly contains the support of
$\mu'$, then the conclusion of that theorem is satisfied for $\nu$.
To be more precise, assuming the statement is false for a distribution $\mu$, the distribution $\mu'$ will be a result of applying a sequence
of the following two operations, so long as it is possible.
\begin{enumerate}
  \item \textbf{Merging Symbols:} if there are distinct symbols $x,x'\in\Sigma$ for which there are $y\in\Gamma$ and $z\in\Phi$ such that
  $(x,y,z)$ and $(x',y,z)$ are both in ${\sf supp}(\mu)$, then we can define a distribution $\mu'$ over $\Sigma'\times \Gamma\times \Phi$ where $\Sigma'\subsetneq\Sigma$,
  such that Theorem~\ref{thm:nonembed_deg_must_be_small} holds for $\mu$ if and only if it holds for $\mu'$. This is done via Lemma~\ref{lem:merge}, and by symmetry of
  the roles of $x$ and $z$ the same goes for symbols in $\Phi$.
  \item \textbf{Enlarging the support:} looking at $\mu$ for which the statement is false,
  we ask ourselves whether there are additional atoms from $\Sigma\times\Gamma\times \Phi$ that can be inserted to ${\sf supp}(\mu)$ so that the statement remains false.
  If so, we pass from $\mu$ to another distribution over $\Sigma\times\Gamma\times \Phi$ for which the statement is still false and whose support strictly contains
  the support of $\mu$.
\end{enumerate}
Repeating the above steps so long that it is possible (noting that it eventually terminates as we are either reducing the alphabet sizes or enlarging the size of
${\sf supp}(\mu)$ at each step), we reach a distribution $\mu$ on which Theorem~\ref{thm:nonembed_deg_must_be_small} is false and is extremal
in these regards. More precisely:

\begin{lemma}\label{lem:pushing_counter_to_extreme}
  For all $m\in\mathbb{N}$ there is $m'\in\mathbb{N}$ such that the following holds.
  Suppose that $\Sigma$, $\Gamma$ and $\Phi$ are alphabets of size at most $m$ and $\mu$ is a distribution over $\Sigma\times\Gamma\times\Phi$
  for which Theorem~\ref{thm:nonembed_deg_must_be_small} fails.
  Then there exist alphabets $\Sigma'$, $\Gamma'$ and $\Phi'$ of size at most $m'$ and a distribution $\mu'$ over $\Sigma'\times\Gamma'\times\Phi'$ such that
  \begin{enumerate}
    \item The distribution $\mu'$ is fully merged as per Definition~\ref{def:fully_merged}.
    \item The distribution $\mu'$ is maximal as per Definition~\ref{def:maximal}.
    \item The conclusion of Theorem~\ref{thm:nonembed_deg_must_be_small} fails for $\mu'$ but it satisfies all of its conditions.
  \end{enumerate}
\end{lemma}
\begin{proof}
  If $\mu$ is not merged, we perform a merge operation and use either Lemma~\ref{lem:merge} or Lemma~\ref{lem:merge2}, noting that all of the conditions of Theorem~\ref{thm:nonembed_deg_must_be_small}
  continue to hold for the new distribution. If $\mu$ is not maximal, we keep greedily form a sequence of distributions by adding elements form $\Sigma\times\Gamma\times\Phi$
  to the support of the distribution so long as Theorem~\ref{thm:nonembed_deg_must_be_small} fails for it, and eventually reach a maximal distribution.
  We then iterate. Note that this process terminates, as each invocation of merge decreases the size of the alphabet $\Sigma$ by at least $1$.

  A delicate point to notice, is that if we add an element to $\mu$ to form a distribution $\nu$,
  then all of the conditions of Theorem~\ref{thm:nonembed_deg_must_be_small}
  continue to hold. Indeed, if $\sigma,\gamma,\phi$ is an Abelian embedding of $\nu$, then it is also an Abelian embedding of $\mu$ hence
  it is a coordinate of the master embedding of $\mu$. Thus we can keep a subset of coordinates of the master embeddings of $\mu$ (which correspond
  to embeddings of $\nu$) and have that these form a master embedding of $\nu$; for simplicity say $\sigma_{{\sf master}, \mu} = (\sigma_1,\ldots,\sigma_r)$
  and $\sigma_{{\sf master}, \nu} = (\sigma_1,\ldots,\sigma_{r'})$ where $r'<r$ and $H = H_1\times\ldots\times H_r$.

  Note that $\sigma_{{\sf master}, \nu}$ is saturated on $H_{\leq r'} = H_1\times\ldots\times H_{r'}$. We define
  $\gamma_{{\sf master}, \nu}$ and $\phi_{{\sf master}, \nu}$ similarly and notice that they are also saturated on $H_{\leq r'}$.
  The distribution of $(\sigma_{{\sf master}, \nu}(x), \gamma_{{\sf master},\nu}(y), \phi_{{\sf master},\nu}(z))$
  where $(x,y,z)\sim \nu$ has full support on $\sett{(a,b,c)\in H_{\leq r'}^{3}}{a+b+c = 0}$.
\end{proof}
\begin{remark}
Another important feature of a distribution $\nu$ that results from a distribution $\mu$ by adding elements to its support, is that
the non-embedding noise operator of $\nu$ is ``weaker'' than that of $\mu$. Intuitively, this follows as the master embedding of $\mu$ is a refinement
of the master embedding of $\nu$, hence the operator of $\nu$ does more averaging. We will use this fact in the future (formalized appropriately)
to argue that if a function $g$ has small noise stability with respect to the non-embedding noise operator of $\mu$, then it also has small stability
with respect to the non-embedding noise operator of $\nu$.
\end{remark}

In Section~\ref{sec:motivating_all_are_H} we have seen that if each one of the alphabets $\Sigma$, $\Phi$ is equal to $H$,
then the proof Theorem~\ref{thm:nonembed_deg_must_be_small} is a rather simple, requiring a change of distribution and a basic Fourier
analytic computation. As explained in Section~\ref{sec:max_merg_mot} some of our arguments require either the size of $\Sigma$ or of
$\Phi$ to be strictly larger than the size of the group $H$, and thus we have to separately deal with the case that both have the same
size as $H$. This is a slightly more general case than the case handled in Section~\ref{sec:motivating_all_are_H}, and in the following lemma
we show that argument to the argument therein works:
\begin{lemma}\label{lem:very_base_case}
  Suppose that $\mu$ is a distribution as in Theorem~\ref{thm:nonembed_deg_must_be_small},
  that ${\sf NEStab}_{1-\delta}(g)\leq \delta$ and that $\card{\Sigma} = \card{\Phi} = \card{H}$. Then
  the conclusion of Theorem~\ref{thm:nonembed_deg_must_be_small} holds.
\end{lemma}
\begin{proof}
We first argue that is suffices to prove Theorem~\ref{thm:nonembed_deg_must_be_small} in our case under the additional assumption
that the distribution of $\mu_{y,z}$ is uniform, and then show a direct argument for that case.

\paragraph{Reduction to the case the master embeddings are uniform.}
Consider the distribution $\nu$ over the support of $\mu$ which is defined by first sampling
$(y,z)\in \Gamma\times \Phi$ uniformly $(x,y,z)\sim \mu$ conditioned on $y$ and $z$.
As the support of this distribution is the same as of $\mu$, we may write
$\mu = \alpha'\nu + (1-\alpha')\nu'$ for some $\alpha' = \alpha/2$ where $\nu'$ is some distribution. We now think of generating
$(x,y,z)\sim \mu^{\otimes n}$ as first choosing $I\subseteq_{\alpha'} [n]$,
sampling $(x',y',z')\sim \nu^{I}$, $(x'',y'',z'')\sim\nu'^{\overline{I}}$ and taking $x = (x',x'')$, $y = (y',y'')$, $z = (z',z'')$.
Then
\begin{equation}\label{eq:very_base_case_1}
\card{\Expect{(x,y,z)\sim\mu^{\otimes n}}{f(x)g(y)h(z)}}
\leq \Expect{I, (x'',y'',z'')}{\card{\Expect{(x',y',z')\sim\nu^{I}}{f_{\overline{I}\rightarrow x''}(x')g_{\overline{I}\rightarrow y''}(y')h_{\overline{I}\rightarrow z''}(z')}}}
\end{equation}
Note that in $\nu,\nu'$ the probability of each atom is at least $\alpha/2$, and we choose $c(\alpha/2,\alpha')$ from Lemma~\ref{lem:op_comparison_lemma}.
Let $c' = \alpha c/2$ and let $E_1$ be the event that ${\sf NEStab}_{1-c'^{-1}\delta}(g_{\overline{I}\rightarrow y''}; \nu_y)\geq \sqrt{\delta}$.
Then by Lemma~\ref{lem:op_comparison_lemma} we have that
\[
\Expect{I, y''}{{\sf NEStab}_{1-c'^{-1}\delta}(g_{\overline{I}\rightarrow y''}; \nu_y^{I})}
\leq
{\sf NEStab}_{1-\delta}(g; \mu_y^{\otimes n})\leq \delta,
\]
so by Markov's inequality it follows that $\Prob{}{E_1}\leq \sqrt{\delta}$. Thus, looking at~\eqref{eq:very_base_case_1} we get that
\begin{align*}
\card{\Expect{(x,y,z)\sim\mu^{\otimes n}}{f(x)g(y)h(z)}}
&\leq
\Expect{I, (x'',y'',z'')}{1_{E_1}\card{\Expect{(x',y',z')\sim\nu^{I}}{f_{\overline{I}\rightarrow x''}(x')g_{\overline{I}\rightarrow y''}(y')h_{\overline{I}\rightarrow z''}(z')}}}\\
&+
\Expect{I, (x'',y'',z'')}{1_{\overline{E_1}}\card{\Expect{(x',y',z')\sim\nu^{I}}{f_{\overline{I}\rightarrow x''}(x')g_{\overline{I}\rightarrow y''}(y')h_{\overline{I}\rightarrow z''}(z')}}}\\
&\leq \sqrt{\delta}+
\Expect{I, (x'',y'',z'')}{1_{\overline{E_1}}\card{\Expect{(x',y',z')\sim\nu^{I}}{f_{\overline{I}\rightarrow x''}(x')g_{\overline{I}\rightarrow y''}(y')h_{\overline{I}\rightarrow z''}(z')}}}.
\end{align*}
Whenever $\overline{E_1}$ holds, we may take $\delta' = \sqrt{\delta}$ and get that
${\sf NEStab}_{1-\delta'}(g_{\overline{I}\rightarrow y''}; \nu_y)\leq \delta'$ (provided that $\delta$ is small enough).
Moreover, in $\nu$ we now have the additional property that $\nu_{y,z}$ is uniform, hence the last expectation is at most
$M\delta'^{\eta}\leq M\delta^{\eta/2}$ where $\eta, M$ depend only on $\alpha$ and the alphabet sizes, hence we get that
$\card{\Expect{(x,y,z)\sim\mu^{\otimes n}}{f(x)g(y)h(z)}}\leq M'\delta^{\eta'}$ for $M' = 2M$ and $\eta' = \eta/2$.

\paragraph{The argument in the case the master embeddings are uniform.}
We now assume that $\mu_{y,z}$ is uniform.
we may re-label the symbols in $\Sigma$ to be group elements (and similarly $\Phi$), so that the master
embeddings of $x$ (and similarly $z$) will be the identity. We now note that conditioned on $y$,
the distribution over $x,z$ is uniform over all pairs such that $x+\gamma(y) + z = 0$.
Indeed, all elements $(x,y,z)\in {\sf supp}(\mu)$ satisfy that
$x + \gamma(y) + z = 0$ over $H$, and as $\mu_{y,z}$ is uniform we get that conditioned on $y$, $z$ is uniform
and the result follows.

Define $\tilde{g}(y) = \cExpect{(x',y',z')\sim \mu}{y'=y}{f(x')h(z')}$ and note that $\tilde{g}$
is in ${\sf Embed}_{\gamma}(\mu)$. Indeed, write $f = \sum\limits_{\chi\in\hat{H}^{\otimes n}}\widehat{f}(\chi)\chi(x)$ and
$h = \sum\limits_{\chi'\in\hat{H}^{\otimes n}}{\widehat{h}(\chi')\chi'(z)}$. Multiplying out, one observes that for $\chi, \chi'$ we have
\[
\cExpect{(x',y',z')\sim \mu}{y'=y}{\chi(x')\chi'(z')}
=\chi(-\gamma(y))\cExpect{(x',y',z')\sim \mu}{y'=y}{(\chi'\overline{\chi})(z')}
=\chi(-\gamma(y))1_{\chi'=\overline{\chi}},
\]
where we used the fact that conditioned on $y'$, the distribution over $(x',z')$ is uniform over $x'+\gamma(y') + z' = 0$.
Thus,
\[
\card{\Expect{(x,y,z)\sim \mu^{\otimes n}}{f(x)g(y)h(z)}}
=\inner{g}{\overline{\tilde{g}}}
=\inner{g}{\mathrm{T}_{\sqrt{1-\delta},\text{non-embed}}\overline{\tilde{g}}}
=\inner{\mathrm{T}_{\sqrt{1-\delta},\text{non-embed}} g}{\overline{g}},
\]
and by Cauchy-Schwarz this is at most $\norm{\mathrm{T}_{\sqrt{1-\delta},\text{non-embed}} g}_{2} = \sqrt{{\sf NEStab}_{1-\delta}(g)}\leq \sqrt{\delta}$.
\end{proof}

With Lemmas~\ref{lem:very_base_case},~\ref{lem:pushing_counter_to_extreme} in hand, it quickly follows that
it suffices to prove Theorem~\ref{thm:nonembed_deg_must_be_small} in the case that none of these results apply,
and hence it suffices to prove the following variant of Theorem~\ref{thm:nonembed_deg_must_be_small}:
\begin{thm}\label{thm:nonembed_deg_must_be_small_rephrase_maximal}
  For all $m\in\mathbb{N}$, $\alpha>0$ there are $M\in\mathbb{N}$, $\delta_0>0$ and $\eta>0$ such that the following holds for all $0< \delta\leq \delta_0$.
  Suppose that $\mu$ is a distribution over $\Sigma\times \Gamma\times \Phi$ satisfying:
  \begin{enumerate}
    \item The probability of each atom is at least $\alpha$.
    \item The size of each one of $\Sigma,\Gamma,\Phi$ is at most $m$.
    \item ${\sf supp}(\mu)$ is pairwise connected.
    \item There are master embeddings $\sigma,\gamma,\phi$ for $\mu$ into an Abelian group $(H,+)$ that are saturated,
    and the distribution of $(\sigma(x),\gamma(y),\phi(z))$ where $(x,y,z)\sim \mu$ has full support on $\{(a,b,c)\in H^3~|~a+b+c = 0\}$.
    \item The distribution $\mu$ is maximal as per Definition~\ref{def:maximal}.
    \item $\card{\Sigma}>\card{H}$ or $\card{\Phi}>\card{H}$.
    \item $\mu$ is fully merged.
  \end{enumerate}
  Then, if $f\colon\Sigma^n\to \mathbb{C}$, $g\colon\Gamma^n\to \mathbb{C}$ and $h\colon\Phi^n\to \mathbb{C}$ are $1$-bounded functions such
  that
  ${\sf NEStab}_{1-\delta}(g;\mu_y^{\otimes n})\leq \delta$
  then
  \[
  \card{\Expect{(x,y,z)\sim \mu^{\otimes n}}{f(x)g(y)h(z)}}\leq M\delta^{\eta}.
  \]
\end{thm}

\subsubsection{Theorem~\ref{thm:nonembed_deg_must_be_small_rephrase_maximal} implies Theorem~\ref{thm:nonembed_deg_must_be_small}}\label{sec:rephrase_implies}
By Lemma~\ref{lem:from_mu_to_path}, the conclusion of Theorem~\ref{thm:nonembed_deg_must_be_small}
for $\mu$ would follow if it is true for the distribution $\mu_{2^t}$ after the path trick, where $t$ is
taken to be sufficiently large. It is clear that $\mu_{2t}$ is a distribution over $\Sigma'\times \Gamma\times \Phi$
where $\Sigma'\subseteq \Sigma^{2^t-1}$ and that for sufficiently large $t\geq t_0(\alpha,m)$ we have that $\mu_{2t}$
satisfies all of the conditions of Theorem~\ref{thm:nonembed_deg_must_be_small}
(the fact that the support of the marginal distribution over $y,z$ is full follows
from Lemma~\ref{lem:master_stays}).
By Lemma~\ref{lem:pushing_counter_to_extreme}, Theorem~\ref{thm:nonembed_deg_must_be_small} for $\mu_{2t}$
is equivalent to the same statement
under the additional assumptions of maximality and that the distribution is fully merged.
If $\card{\Sigma'} = \card{\Phi} = \card{H}$ then the validity of the statement follows from Lemma~\ref{lem:very_base_case}.
Otherwise, we also have that either $\card{\Sigma'}>\card{H}$ or $\card{\Phi}>\card{H}$ and the validity follows from Theorem~\ref{thm:nonembed_deg_must_be_small_rephrase_maximal}.
\qed

\subsection{Restating Theorem~\ref{thm:nonembed_deg_must_be_small_rephrase_maximal} via Role Symmetry of $x,z$}
Noting that the roles of $x$ and $z$ are symmetric in Theorem~\ref{thm:nonembed_deg_must_be_small_rephrase_maximal} we have that it is equivalent to
the following statement, in which the same condition has been replaced by $\card{\Sigma}>\card{H}$.
\begin{thm}\label{thm:nonembed_deg_must_be_small_rephrase_maximal2}
  For all $m\in\mathbb{N}$, $\alpha>0$ there are $M\in\mathbb{N}$, $\delta_0>0$ and $\eta>0$ such that the following holds for all $0<\delta\leq \delta_0$.
  Suppose that $\mu$ is a distribution over $\Sigma\times \Gamma\times \Phi$ satisfying:
  \begin{enumerate}
    \item The probability of each atom is at least $\alpha$.
    \item The size of each one of $\Sigma,\Gamma,\Phi$ is at most $m$.
    \item ${\sf supp}(\mu)$ is pairwise connected.
    \item There are master embeddings $\sigma,\gamma,\phi$ for $\mu$ into an Abelian group $(H,+)$ that are saturated,
    and the distribution of $(\sigma(x),\gamma(y),\phi(z))$ where $(x,y,z)\sim \mu$ has full support on $\{(a,b,c)\in H^3~|~a+b+c = 0\}$.
    \item The distribution $\mu$ is maximal as per Definition~\ref{def:maximal}.
    \item $\card{\Sigma}>\card{H}$.
    \item In $\mu$, the value of any two coordinates implies the value of the third.
  \end{enumerate}
  Then, if $f\colon\Sigma^n\to \mathbb{C}$, $g\colon\Gamma^n\to \mathbb{C}$ and $h\colon\Phi^n\to \mathbb{C}$ are $1$-bounded functions such
  that
  ${\sf NEStab}_{1-\delta}(g;\mu_y^{\otimes n})\leq \delta$
  then
  \[
  \card{\Expect{(x,y,z)\sim \mu^{\otimes n}}{f(x)g(y)h(z)}}\leq M\delta^{\eta}.
  \]
\end{thm}

With Theorem~\ref{thm:nonembed_deg_must_be_small_rephrase_maximal2} in hand and following the logic of Section~\ref{sec:max_merg_mot},
we should now strive to achieve a relaxed base case for distributions as in Theorem~\ref{thm:nonembed_deg_must_be_small_rephrase_maximal2}.
This is the content of the next section.

\subsection{Establishing a Relaxed Base Case}
With the statement of Theorem~\ref{thm:nonembed_deg_must_be_small_rephrase_maximal2} in hand, we are now ready to begin addressing the so-called Horn-SAT
obstruction. The bulk of our proof will eventually be an inductive proof for a statement similar in spirit to Theorem~\ref{thm:nonembed_deg_must_be_small_rephrase_maximal2},
and to facilitate this induction we must have a base case. The goal of this section is to design such base case.

We begin by describing an ideal scenario, also discussed in Section~\ref{sec:max_merg_mot}, in which a simple base case
statement holds for the distribution $\mu$; if this could be achieved for all distributions $\mu$ as in Theorem~\ref{thm:nonembed_deg_must_be_small_rephrase_maximal2}
our argument would simplify considerably. Alas, there are distributions $\mu$ for which this ideal base case fails, and we call such obstructions as Horn-SAT embeddings.
After explaining what Horn-SAT embeddings are, we will turn to the question of how to overcome the Horn-SAT obstruction. To do so we will switch from the distribution $\mu$
to a related distribution $\mu'$ (using the path trick and merges), for which we have a relaxed form of this ideal base case (and is sufficient for our purposes).
\subsubsection{The Ideal Base Case}
Ideally, we would have liked to have a base case statement as follows.
For all $\tau>0$, there is $\lambda<1$ such that if $f\colon \Sigma\to\mathbb{C}$, $g\colon \Gamma\to\mathbb{C}$ and $h\colon \Phi\to\mathbb{C}$
satisfy that ${\sf NEStab}_{1-\xi}(g;\mu_y)\leq (1-\tau)\norm{g}_2^2$, then
\[
\card{\Expect{(x,y,z)\sim \mu}{f(x)g(y)h(z)}}\leq \lambda\norm{f}_2\norm{g}_2\norm{h}_2.
\]
There are several issues with this statement.
\begin{enumerate}
  \item First, at the present scenario it is not even clear that the statement should be true for $\lambda = 1$,
    let alone $\lambda<1$. The reason is that trivially, we can only bound the left hand side using H\"{o}lder's inequality by $\norm{f}_3\norm{g}_3\norm{h}_3$
    (or some other product of $q_1,q_2,q_3$ norms where $\frac{1}{q_1}+\frac{1}{q_2} + \frac{1}{q_3} = 1$), which may be much larger than
    $\norm{f}_2\norm{g}_2\norm{h}_2$. To address this issue, we will ensure that in $(x,y,z)\sim \mu$ we have that $y$ and $z$ are independent,
    in which case we may now use Cauchy-Schwarz to always argue that
    \[
    \card{\Expect{(x,y,z)\sim \mu}{f(x)g(y)h(z)}}
    \leq
    \sqrt{\Expect{(x,y,z)\sim \mu}{\card{f(x)}^2}}
    \sqrt{\Expect{(x,y,z)\sim \mu}{\card{g(y)}^2\card{h(z)}^2}}
    =\norm{f}_2\norm{g}_2\norm{h}_2,
    \]
    so we can at least take $\lambda = 1$ in the ideal base case.

  \item Second, even with the above transformation, it turns out that one still may not be able to take $\lambda<1$.
  Naively, one is tempted argue that it is possible to take some $\lambda<1$ by considering the equality cases of Cauchy-Schwarz and argue as follows:
  if we cannot take $\lambda < 1$, then it means that there are functions $f,g,h$ of $2$-norm $1$ for which the above Cauchy-Schwarz is tight,
  hence $f(x) = g(y) h(z)$ in the support of $\mu$, and such $f,g,h$ constitute an Abelian embedding, so $g\in {\sf Embed}_{\gamma}(\mu)$ and in particular
  ${\sf NEStab}_{1-\xi}(g;\mu_y) = \norm{g}_2^2$.

  Taking a closer inspection though reveals that there are issues if $f$ sometimes gets the value $0$; indeed, if $f$ is always non-zero this
  argument goes through, but we do not know how to ensure that. This issue was referred to as the ``Horn-SAT Obstruction'' in~\cite{BKMcsp2},
  and here too we have to circumvent it.
\end{enumerate}

\subsubsection{The Relaxed Base Case}
As the ideal base case may fail, we resort to a more relaxed form of it, referred to as the ``relaxed base case''; below is a formal definition.
\begin{definition}\label{def:relaxed_base}
  Let $\Sigma,\Gamma$ and $\Phi$ be finite alphabets of size at most $m$, let $\mu$ be a distribution over $\Sigma\times \Gamma\times \Phi$ in
  which the probability of each atom is at least $\alpha$, and let
  $(\sigma,\gamma,\phi)$ be master embeddings of $\mu$ into an Abelian group $(H,+)$. We say that $\mu$ satisfies the relaxed base
  case if there is $\Sigma_{{\sf modest}}\subseteq \Sigma$ of size at least $2$ such that the following holds:
  \begin{enumerate}
    \item For all $x,x'\in\Sigma_{{\sf modest}}$ it holds that $\sigma(x) = \sigma(x')$.
    \item There are $c(\alpha,m),C(\alpha,m)>0$ such that the following holds. For all $\tau>0$, for all
    $f\colon \Sigma\to\mathbb{C}$, $g\colon\Gamma\to\mathbb{C}$ and $h\colon \Phi\to\mathbb{C}$ such
    that $\Expect{x,x'\in\Sigma_{{\sf modest}}}{\card{f(x) - f(x')}^2}\geq \tau$ it holds that
    \[
    \card{\Expect{(x,y,z)\sim \mu^{\otimes n}}{f(x)g(y)h(z)}}\leq (1-c\tau^{C})\norm{f}_2\norm{g}_2\norm{h}_2.
    \]
  \end{enumerate}
\end{definition}

We now state a variant of Theorem~\ref{thm:nonembed_deg_must_be_small_rephrase_maximal} for distributions that have the relaxed base case,
and then show it implies Theorem~\ref{thm:nonembed_deg_must_be_small_rephrase_maximal}. In comparison to Theorem~\ref{thm:nonembed_deg_must_be_small_rephrase_maximal2},
below we have the additional assumption that the distribution of $(\sigma(x),\gamma(y),\phi(z))$ is uniform over elements $(a,b,c)\in H^3$ such that $a+b+c = 0$,
that $\mu_{y,z}$ is uniform and that the relaxed base case holds.
\begin{thm}\label{thm:nonembed_deg_must_be_small_rephrase_maximal_relaxed}
  For all $m\in\mathbb{N}$, $\alpha>0$ there are $K, M\in\mathbb{N}$, $\delta_0>0$ and $\eta>0$ such that the following holds for all $0<\delta\leq \delta_0$.
  Suppose that $\nu$ is a distribution over $\Sigma\times \Gamma\times \Phi$ satisfying:
  \begin{enumerate}
    \item The probability of each atom is at least $\alpha$.
    \item The size of each one of $\Sigma,\Gamma,\Phi$ is at most $m$.
    \item ${\sf supp}(\nu)$ is pairwise connected.
    \item There are master embeddings $\sigma,\gamma,\phi$ for $\mu$ into an Abelian group $(H,+)$ that are saturated,
    and the distribution of $(\sigma(x),\gamma(y),\phi(z))$ where $(x,y,z)\sim \nu$ is uniform over $\{(a,b,c)\in H^3~|~a+b+c = 0\}$.
    \item The support of $\nu$ on $\Gamma\times \Phi$ is full, and the marginal distribution of $\nu$ on $y,z$ is uniform.
    \item There are $\Sigma_{{\sf modest}}\subseteq \Sigma'\subseteq \Sigma$, $\Gamma'\subseteq \Gamma$ and $\Phi'\subseteq \Phi$
    such that
    \begin{enumerate}
      \item \textbf{Relaxed base case:} the distribution $\nu$ satisfies the relaxed base case with $\Sigma_{{\sf modest}}$.
      \item \textbf{Containing a maximal, fully merged distribution:}
      There is an alphabet $\Sigma''$ of size at most $K$,
      a map $a\colon \Sigma''\to \Sigma'$
      and a distribution $\tilde{\nu}$ supported on $\Sigma''\times \Gamma'\times \Phi'$ such that:
      \begin{enumerate}
        \item $a$ is surjective.
        \item $\sett{(a(x),y,z)}{(x,y,z)\in {\sf supp}(\tilde{\nu})} \subseteq {\sf supp}(\nu)$.
        \item $\tilde{\nu}$ is maximal as per Definition~\ref{def:maximal}.
        \item In $\tilde{\nu}$, the value of any two coordinates implies the third.
        \item $(\sigma\circ a, \gamma, \phi)$ is a master embedding of $\tilde{\nu}$, and it is saturated.
        \item There are distinct $v,v'\in \Gamma'$ such that $\gamma(v) = \gamma(v')$.
      \end{enumerate}
      \item \textbf{Full support on restriction of the first two coordinates:} $\Sigma'\times \Gamma\subseteq {\sf supp}(\nu_{x,y})$. In
      words, for all $x\in \Sigma'$ and $y\in \Gamma$ there is $z\in \Phi$ such that $(x,y,z)\in {\sf supp}(\nu)$.
    \end{enumerate}
    \item In $\nu$, the value of $y,z$ implies the value of $x$.
  \end{enumerate}
  Then, if $f\colon\Sigma^n\to \mathbb{C}$, $g\colon\Gamma^n\to \mathbb{C}$ and $h\colon\Phi^n\to \mathbb{C}$ are $1$-bounded functions such
  that
  ${\sf NEStab}_{1-\delta}(g;\nu_y^{\otimes n})\leq \delta$
  then
  \[
  \card{\Expect{(x,y,z)\sim \nu^{\otimes n}}{f(x)g(y)h(z)}}\leq M\delta^{\eta}.
  \]
\end{thm}

At first reading, we encourage the reader to think of condition 6b above as saying that $\mu$ itself is maximal and fully merged.
We do not know how to ensure that however, at least without losing some other properties of $\mu$ which are necessary for us (or without introducing
further complications). Instead, as explained in Section~\ref{sec:additional_remarks_max}, we just say that within $\mu$ we can find a distribution
$\nu$ which satisfies these additional properties; for our purposes, this is just as good as the distribution $\mu$ itself having these properties.

\subsection{Achieving the Relaxed Base Case: Theorem~\ref{thm:nonembed_deg_must_be_small_rephrase_maximal_relaxed} implies Theorem~\ref{thm:nonembed_deg_must_be_small_rephrase_maximal2}}
\label{sec:relaxed_implies}
In this section we show that Theorem~\ref{thm:nonembed_deg_must_be_small_rephrase_maximal_relaxed} implies Theorem~\ref{thm:nonembed_deg_must_be_small_rephrase_maximal2}.
For that, we start with a distribution $\mu$ as in Theorem~\ref{thm:nonembed_deg_must_be_small_rephrase_maximal2} and construct from it a distribution $\nu$ as in
Theorem~\ref{thm:nonembed_deg_must_be_small_rephrase_maximal_relaxed}, such that correlations with respect to $\mu$ can be upper bounded by a similar looking
correlations over $\nu$. We begin by noting that we may assume that in the notation of Theorem~\ref{thm:nonembed_deg_must_be_small_rephrase_maximal2},
there are distinct $v,v'\in \Gamma$ such that $\gamma(v) = \gamma(v')$. Indeed, otherwise the statement holds vacuously, as then
$\norm{g}_2^2 = {\sf NEStab}_{1-\delta}(g; \mu_y^{\otimes n})\leq \delta$ and the expectation would clearly be bounded by above by $\norm{g}_2\leq \sqrt{\delta}$.
\subsubsection{The Construction of $\nu$, $\Sigma_{{\sf modest}}$, $\Sigma'$ and $\Phi'$}
Take $\ell_1$ large enough with respect to $m$, and then $\ell_2$ large enough with respect to $\ell_1, m$.
\begin{enumerate}
  \item Take $\mu'$ to be the path trick distribution with respect to $z$ on $\mu$ for length $\ell_1$, and choose $\ell_1$ large
  enough so that the support of $\mu'$ on $x,y$ is full (using Lemma~\ref{lem:path_keeps_connectedness}). Then $\mu'$ is a distribution
  over $\Sigma\times \Gamma\times \tilde{\Phi}$ where $\tilde{\Phi}\subseteq \Phi^{\ell_1}$.

  \item Take $\mu''$ to be the path trick distribution with respect to
  $x$ on $\mu'$ for length $\ell_2$, and choose $\ell_2$ to be large enough so that $\mu''$ has full support on $y,z$ (again, using Lemma~\ref{lem:path_keeps_connectedness}).
  Then $\mu''$ is a distribution over $\tilde{\Sigma}\times \Gamma\times \tilde{\Phi}$ where $\tilde{\Sigma}\subseteq \Sigma^{\ell_2}$.

  \item Take $\nu'$ to be the $x$-merge of $\mu''$, so that it is a distribution over $\tilde{\Sigma}_{{\sf fin}}\times \Gamma\times \tilde{\Phi}$
  where $\tilde{\Sigma}_{{\sf fin}}\subseteq \tilde{\Sigma}$. For a symbol $x\in \tilde{\Sigma}$, we denote by $a(x)\in \tilde{\Sigma}_{\sf fin}$ the distinguished element
  from the connected component of $x$ as per Definition~\ref{def:merge}.

  \item Duplicate symbols: let $\sigma$, $\gamma$ and $\phi$ be a saturated master embedding for $\nu'$ into an Abelian group $H$.
  For each $a\in H$, let $d_3(a)$ be the number of $z\in \tilde{\Phi}$ such that $\phi(z) = a$ (and note that $d_3(a)\geq 1$ for all $a\in H$ as $\phi$ is saturated).
  Similarly, let $d_2(a)$ be the number of $y\in \Gamma$ such that $\gamma(y) = a$.

  We are going to duplicate each symbol $y\in \Gamma$, and $z\in \tilde{\Phi}$ symbol multiple times. That is, define
  \[
  \Gamma_{{\sf dup}} = \sett{(y,i)}{y\in \Gamma, 1\leq i\leq \frac{\prod\limits_{t\in H}d_2(t)}{d_2(\gamma(y))}},
  \qquad
  \tilde{\Phi}_{{\sf dup}} = \sett{(z,j)}{z\in \tilde{\Phi}, 1\leq j\leq \frac{\prod\limits_{t\in H}d_3(t)}{d_3(\phi(z))}}.
  \]

  \item Take $\nu$ to be the following distribution: sample $(b,c)\in H^2$ uniformly,
  sample $y\in \Gamma$ and $z\in \tilde{\Phi}$ uniformly such that $\gamma(y) = b$ and $\phi(z) = c$,
  sample $1\leq i\leq \frac{\prod\limits_{t\in H}d_2(t)}{d_2(b)}$ and $1\leq j\leq \frac{\prod\limits_{t\in H}d_3(t)}{d_3(c)}$ uniformly
  and then sample $(x',y',z')\sim \nu'$ conditioned on
  $y' = y$ and $z' = z$. Output $(x',(y,i),(z,j))$.
\end{enumerate}
We take $\Sigma' = \sett{a(x,\ldots,x)}{x\in \Sigma}$, $\Gamma' = \sett{(y,1)}{y\in \Gamma}$ and $\Phi' = \sett{((z,\ldots,z),1)}{z\in \Phi}$.
As $\card{\Sigma}>\card{H}$, by the pigeonhole principle there are distinct $x^{\star},{x^{\star}}'\in \Sigma$ such that $\sigma(x^{\star}) = \sigma({x^{\star}}')$,
and we take
\[
\Sigma_{{\sf modest}} = \set{a(x^{\star},\ldots,x^{\star}), a({x^{\star}}',\ldots,{x^{\star}}')}.
\]
There are two cases, depending on the size of
$\Sigma_{{\sf modest}}$:
\begin{enumerate}
  \item If $\card{\Sigma_{{\sf modest}}} = 1$, we show that the validity of Theorem~\ref{thm:nonembed_deg_must_be_small_rephrase_maximal} for
  $\mu$ follows from its maximality. The idea here is that if $\Sigma_{{\sf modest}}$ has size $1$, it means that the symbols $x$ and $x'$
  get ``mixed up'' when we look at the distribution $\nu$; however $\nu$ contains inside it a copy of $\mu$, hence we are able to relate
  correlations over $\nu$ to correlations over distribution $\mu$ of random restrictions, and this mix-up means that we will actually look
  at a distribution richer than $\mu$.
  \item If $\card{\Sigma_{{\sf modest}}} = 2$, we show that the relaxed base case for $\nu$ holds, and then Theorem~\ref{thm:nonembed_deg_must_be_small_rephrase_maximal_relaxed}
  quickly implies Theorem~\ref{thm:nonembed_deg_must_be_small_rephrase_maximal2}.
\end{enumerate}

The two cases are addressed int the following two lemmas:
\begin{lemma}\label{lem:ugly_modest_case}
  If $\card{\Sigma_{{\sf modest}}} = 1$, then $\mu$ satisfies the conclusion of
  Theorem~\ref{thm:nonembed_deg_must_be_small_rephrase_maximal2}.
\end{lemma}
\begin{proof}
  Deferred to Section~\ref{sec:ugly_modest_case}.
\end{proof}
\begin{lemma}\label{lem:nu_satisfies}
  If $\card{\Sigma_{{\sf modest}}}>1$, then
  the distribution $\nu$ satisfies the relaxed base case and, furthermore, with $\Sigma_{{\sf modest}}, \Sigma'$ and $\Phi'$ defined as above,
  it satisfies the conditions in Theorem~\ref{thm:nonembed_deg_must_be_small_rephrase_maximal_relaxed}.
\end{lemma}
\begin{proof}
  We discuss the master embedding $(\sigma',\gamma',\phi')$ of $\nu$:
  \begin{enumerate}
    \item Let $(\sigma,\gamma,\phi)$ be a master embedding of $\mu$. In Lemma~\ref{lem:master_stays} we have seen how master embedding evolve under the path trick,
    hence we get a master embedding for $\mu''$; abusing notation we denote it also by $(\sigma,\gamma,\phi)$. Note that if two symbols $\vec{x}$ and
    $\vec{x}'$ are to be merged in $\nu'$, then they are mapped to the same group element by $\sigma$. Indeed, if there are $y\in \Gamma$ and $z\in\tilde{\Phi}$
    such that $(\vec{x},y,z)$ and $(\vec{x}',y,z)$ are both in the support of $\mu''$, then
    \[
    \sigma(\vec{x}) + \gamma(y) + \phi(z) = 0 = \sigma(\vec{x}') + \gamma(y) + \phi(z),
    \]
    hence $\sigma(\vec{x}) = \sigma(\vec{x}')$. It follows that $\sigma$ is constant on each connected component of $G_{x,\mu'',{\sf merge}}$, hence we may
    define $\sigma(a(\vec{x})) = \sigma(\vec{x})$ unambiguously and get an embedding for $\nu'$.

    Lastly, a master embedding for $\nu$ follows. The map $\sigma$ stays the same, and $\gamma$ and $\phi$ simply ignore $i$ and $j$, that is,
    \[
    \gamma'(y,i) = \gamma(y),
    \qquad
    \phi'(z,j) = \phi(z).
    \]
    We note that $(\sigma,\gamma',\phi')$ is a master embedding of $\nu$.
    To see that, first note there is a $1$-to-$1$ correspondence between embeddding of $\nu$ and embedding of $\nu'$:
    indeed, note that for all $(y,i)$ and $(y,i')$ there are $x$ an $(z,j)$ such that $(x,(y,i),(z,j))$ and $(x,(y,i'),(z,j))$ are in the support
    of $\nu'$ it follows that any embedding of the second symbol must ignore $i$, thereby essentially be an embedding of $\nu'$. Thus, any embedding of
    $\nu$ must ignore $i$ and $j$, hence is also an embedding of $\nu'$, and thereby must be equivalent to some coordinate of $(\sigma,\gamma,\phi)$
    (as it is a master embedding of $\nu'$).
    \item For any $b,c\in H$, sampling $(x,(y,i),(z,j))\sim \nu$ we get that the probability that $\gamma'(y,i) = b$ and $\phi'(z,j) = c$ is $1/\card{H}^2$ by definition,
    and then $\sigma(x)$ must be $-b-c$. Hence, the distribution of $(\sigma(x),\gamma(y,i),\phi(z,j))$ is uniform over
    $\sett{(a,b,c)\in H^3}{a+b+c = 0}$.
    \item We note that it follows that the master embedding of $\nu$ maps the symbols in $\Sigma_{{\sf modest}}$ to the same group elements. Indeed,
    in $\mu$ we have that $\sigma(x^{\star}) = \sigma({x^{\star}}')$, and following the evolution of the master embedding above (using Lemmas~\ref{lem:reverse_embed},~\ref{lem:master_stays} again) it follows that the master embedding of $\nu$ maps
    $a(x^{\star},\ldots,x^{\star})$ and $a({x^{\star}}',\ldots,{x^{\star}}')$ to the same group element.
    \item Finally, we note that as above we found distinct $v,v'\in\Gamma$ such that $\gamma(v) = \gamma(v')$, it follows that $(v,1)$ and $(v',1)$ are distinct elements
    in $\Gamma'$ such that $\gamma'(v,1) = \gamma'(v',1)$.
  \end{enumerate}
  Next, we note that the marginal distribution of $\nu$ on $\Gamma_{{\sf dup}}\times \tilde{\Phi}_{{\sf dup}}$ is uniform. Indeed, by definition it is clear this
  distribution is a product distribution, so it suffices to argue that its marginal on each one of $\Gamma_{{\sf dup}}$ and $\tilde{\Phi}_{{\sf dup}}$ is uniform.
  To see that, note that for all $t = (t',i')\in \Gamma_{{\sf dup}}$
  \[
        \Prob{(x,(y,i),(z,j))\sim \nu}{(y,i) = t}
        =\frac{1}{\card{H}}\cdot \frac{1}{d_2(t')}\cdot\frac{d_2(t')}{\prod\limits_{a\in H}d_2(a)}
        =\frac{1}{\card{H}}\cdot\frac{1}{\prod\limits_{a\in H}d_2(a)},
  \]
  wherein the first factor counts the probability we chose $b=\gamma(t')$ in the process, the second factor counts the probability we chose a specific pre-image
  $y$ of $t'$, and the third factor counts the probability that $i = i'$. This probability is independent of $t$, hence the marginal of $\nu$ on $\Gamma_{{\sf dup}}$
  is uniform. The same argument shows that the marginal of $\nu$ on $\tilde{\Phi}_{{\sf dup}}$ is uniform.

  We now argue that $\Sigma'\times \Gamma_{{\sf dup}}\subseteq {\sf supp}(\nu_{x,y})$. Indeed, by construction we have that
  $\Sigma\times \Gamma \subseteq {\sf supp}(\mu'_{x,y})$ and hence $\sett{(x,\ldots,x)}{x\in \Sigma}\times \Gamma \subseteq {\sf supp}(\mu''_{x,y})$.
  Thus, after the $x$-merge we get that $\Sigma'\times \Gamma\subseteq {\sf supp}(\nu'_{x,y})$, and after duplicating symbols we get that
  $\Sigma'\times \Gamma_{{\sf dup}}\subseteq {\sf supp}(\nu_{x,y})$.

  We defer the proof of the rest of properties of $\nu$  to Section~\ref{sec:modest_case_main}.
\end{proof}

To finish this section, we quickly explain how to derive Theorem~\ref{thm:nonembed_deg_must_be_small_rephrase_maximal2} from Theorem~\ref{thm:nonembed_deg_must_be_small_rephrase_maximal_relaxed}.
\subsubsection{Proof that Theorem~\ref{thm:nonembed_deg_must_be_small_rephrase_maximal_relaxed} implies Theorem~\ref{thm:nonembed_deg_must_be_small_rephrase_maximal2}}
With the above set-up above, if $\card{\Sigma_{{\sf modest}}} = 1$
then we are done by Lemma~\ref{lem:ugly_modest_case}, so assume otherwise. Thus, $\card{\Sigma_{{\sf modest}}} > 1$, and
by Lemma~\ref{lem:nu_satisfies} it follows that $\nu$ satisfies all of the conditions of Theorem~\ref{thm:nonembed_deg_must_be_small_rephrase_maximal_relaxed}.
By Lemma~\ref{lem:from_mu_to_path} we may upper bound
\[
\card{\Expect{(x,y,z)\sim \mu^{\otimes n}}{f(x)g(y)h(z)}}^{\ell}
\leq
\card{\Expect{(x,y,z)\sim{\mu''}^{\otimes n}}{F(x)g(y)H(z)}},
\]
where $F$ and $H$ are some bounded functions, $\ell = O(1)$  and the marginal of ${\mu''}_y$ is the same as $\mu_y$. Hence, it suffices to
establish the conclusion of Theorem~\ref{thm:nonembed_deg_must_be_small_rephrase_maximal2} for $\mu''$. Using Lemma~\ref{lem:merge}, it
suffices to prove the statement for the $x$-merge of $\mu''$, namely for $\nu'$.

For the distribution $\nu$, by Theorem~\ref{thm:nonembed_deg_must_be_small_rephrase_maximal_relaxed}
we know that it satisfies the conclusion of Theorem~\ref{thm:nonembed_deg_must_be_small_rephrase_maximal2}. Consider the distribution
$\nu_{{\sf pre}}$ over $\tilde{\Sigma}_{{\sf final}}\times \Gamma \times \tilde{\Phi}$ defined as: sample $(x,(y,i),(z,j))\sim \nu$
and output $(x,y,z)$. Given functions $f\colon \tilde{\Sigma}_{{\sf final}}^n\to\mathbb{C}$,
$g\colon \Gamma^n\to\mathbb{C}$ and $h\colon \tilde{\Phi}^n\to\mathbb{C}$ we may define
$f'\colon \tilde{\Sigma}_{{\sf final}}^n\to\mathbb{C}$,
$g'\colon \Gamma_{\sf dup}^n\to\mathbb{C}$ and
$h'\colon \tilde{\Phi}_{\sf dup}^n\to\mathbb{C}$
by $f' = f$ and
\[
g'((y_1,i_1),\ldots,(y_n,i_n)) = g(y_1,\ldots,y_n),
\qquad
h'((z_1,j_1),\ldots,(z_n,j_n)) = h(z_1,\ldots,z_n).
\]
Then
\[
\Expect{(x,y,z)\sim \nu_{{\sf pre}^{\otimes n}}}{f(x)g(y)h(z)}
=\Expect{(x',y',z')\sim \nu^{\otimes n}}{f'(x)g'(y)h'(z)}
\]
and the master embeddings of $\nu$ and $\nu_{{\sf pre}}$ are essentially the same, so ${\sf NEStab}_{1-\delta}(g') = {\sf NEStab}_{1-\delta}(g)$.
Thus, as $\nu$ satisfies the conclusion of Theorem~\ref{thm:nonembed_deg_must_be_small_rephrase_maximal2}, it follows that $\nu_{\sf pre}$ also
satisfies the conclusion of Theorem~\ref{thm:nonembed_deg_must_be_small_rephrase_maximal2}.

We now note that $\nu'$ and $\nu'_{{\sf pre}}$ are two distributions over $\tilde{\Sigma}_{{\sf final}}\times \Gamma \times \tilde{\Phi}$ with
the same support and in which the probability of each atom is at least $\alpha' = \alpha'(\alpha,m)>0$, hence the conclusion of
Theorem~\ref{thm:nonembed_deg_must_be_small_rephrase_maximal2} is equivalent for them. Indeed, we may write
$\nu' = \beta\nu'_{{\sf pre}} + (1-\beta)\mathcal{D}$ where $\beta = \beta(\alpha,m)>0$ and $\mathcal{D}$ is some distribution.
Thus, given functions $f\colon \tilde{\Sigma}_{{\sf final}}^n\to\mathbb{C}$,
$g\colon \Gamma^{n}\to\mathbb{C}$ and $h\colon \tilde{\Phi}^n\to\mathbb{C}$ that are $1$-bounded
and ${\sf NEStab}_{1-\delta}(g; \nu'^{\otimes n})\leq \delta$ we may choose $J\subseteq_{\beta} [n]$, sample
$(\tilde{x},\tilde{y},\tilde{z})\sim \mathcal{D}^{\overline {J}}$, set
\[
\tilde{f} = f_{\overline{J}\rightarrow \tilde{x}},
\qquad
\tilde{g} = g_{\overline{J}\rightarrow \tilde{y}},
\qquad
\tilde{h} = h_{\overline{J}\rightarrow \tilde{z}},
\]
and get that
\[
\Expect{(x,y,z)\sim {\nu'}^{\otimes n}}{f(x)g(y)h(z)}
=\Expect{J, (\tilde{x},\tilde{y},\tilde{z})\sim \mathcal{D}^{\overline{J}}}{\Expect{(x',y',z')\sim {\nu_{{\sf pre}}'}^{J}}{\tilde{f}(x')\tilde{g}(y')\tilde{h}(z')}}.
\]
We denote the inner expectation by $\phi_{\nu_{\sf pre}}(\tilde{f},\tilde{g},\tilde{h})$, and
let $E$ be the event that ${\sf NEStab}_{1-\sqrt{\delta}}(\tilde{g}; \nu_{{\sf pre}}'^{I})\leq \sqrt{\delta}$.
By Claims~\ref{claim:rr_nestab},~\ref{claim:increase_noise_decrease_stab} it follows that
\[
\Expect{\substack{J\subseteq_{\beta}[n]\\ \tilde{y}\sim \mathcal{D}_y^{\overline{J}}}}{{\sf NEStab}_{1-\sqrt{\delta}}(\tilde{g}; \nu_{{\sf pre}}')}
\leq {\sf NEStab}_{1-c(m,\alpha)\beta\sqrt{\delta}}(g; \nu')
\leq {\sf NEStab}_{1-\delta}(g; \nu')
\leq \delta.
\]
Thus, by Markov's inequality we have that $\Prob{}{\overline{E}}\leq \sqrt{\delta}$, and so
\begin{align*}
\card{\Expect{J, (\tilde{x},\tilde{y},\tilde{z})\sim \mathcal{D}^{\overline{J}}}{\phi_{\nu_{\sf pre}}(\tilde{f},\tilde{g},\tilde{h})}}
&\leq
\card{\Expect{J, (\tilde{x},\tilde{y},\tilde{z})\sim \mathcal{D}^{\overline{J}}}{1_E\phi_{\nu_{\sf pre}}(\tilde{f},\tilde{g},\tilde{h})}}
+
\card{\Expect{J, (\tilde{x},\tilde{y},\tilde{z})\sim \mathcal{D}^{\overline{J}}}{1_{\overline{E}}\phi_{\nu_{\sf pre}}(\tilde{f},\tilde{g},\tilde{h})}}\\
&\leq
M\delta^{s/2}
+\Prob{}{\overline{E}}
\leq (M+1)\delta^{s/2},
\end{align*}
where we used the fact that if $E$ holds, then we may bound $\card{\phi_{\nu_{\sf pre}}(\tilde{f},\tilde{g},\tilde{h})}\leq M\delta^{s}$ for some
$M,s>0$ depending only on $m$ and $\alpha$, and also we used the upper bound on $\Prob{}{\overline{E}}$.
\qed

\subsection{The case that $\card{\Sigma_{{\sf modest}}} = 2$: Proof of Lemma~\ref{lem:nu_satisfies}}\label{sec:modest_case_main}
\subsubsection{On the Supports of Horn-SAT Embeddings}\label{sec:satisfies_relaxed}
A triplet of functions $f\colon \tilde{\Sigma}_{{\sf fin}}\to\mathbb{C}$, $g\colon \Gamma_{{\sf dup}}\to\mathbb{C}$ and $h\colon \tilde{\Phi}_{{\sf dup}}\to\mathbb{C}$
is called a Horn-SAT embedding if $f(\vec{x}) = g(y) h(\vec{z})$; we say it is non-trivial if $f \not\in {\sf Embed}_{\sigma}(\nu)$. Consider the collection of subsets
\[
\mathcal{F} = \sett{F}{\exists (f,g,h)\text{ a non-trivial Horn-SAT embedding of $\nu$, }F = {\sf supp}(f)}.
\]
Let $F_{{\sf union}} = \bigcup_{F'\in\mathcal{F}} F'$.
\begin{claim}\label{claim:find_sigma'}
  For $x\in\Sigma$ we have that $a(x,\ldots,x)\not\in F_{{\sf union}}$.
\end{claim}
\begin{proof}
  Assume towards contradiction otherwise, and let $x^{\star}$ be such $x$.
  Then there is $F'\in \mathcal{F}$ containing $a(x^{\star},\ldots,x^{\star})$, hence there is a non-trivial Horn-SAT
  embedding $f$, $g$ and $h$ such that
  $f(a(x^{\star},\ldots,x^{\star}))\neq 0$.

  Define $f'\colon \Sigma\to\mathbb{C}$ by $f'(x) = f(a(x,\ldots,x))$, and note that for all $(x,y,z)\in {\sf supp}(\mu')$ we get that
  \[
  f'(x) = g(y)h(\vec{z}),
  \]
  and also that $f'(x^{\star})\neq 0$. In $\mu'$ the support of the distribution on $\Sigma\times \Gamma$ is full, and hence we conclude that $g(y)\neq 0$
  for all $y$. Thus, for all $(x,y,\vec{z})\in {\sf supp}(\mu')$ we get that $f'(x) = 0$ if and only if $h(\vec{z}) = 0$, and since $\mu'$ is pairwise connected
  it follows that the functions $f'$ and $h$ must either be trivially $0$, or else never $0$. As $f'(x^{\star})\neq 0$, we get that $f'$ and $h$ never
  vanish.

  Thus, $g$ and $h$ never vanish, and as $f,g,h$ is a Horn-SAT embedding it follows that $f$ also never vanishes.
  Consider the principal branch of the complex logarithm function $\log$, and define
  $f''(x) = \log(f(x))$, $g''(y) = \log(g(y))$ and $h''(z) = \log(h(z))$; by Claim~\ref{claim:embedding_fns}
  we get that $f''\in {\sf Embed}_{\sigma}(\nu)$, and as $\log$ is injective it follows that $f\in {\sf Embed}_{\sigma}(\nu)$,
  in contradiction to the fact that $(f,g,h)$ is a non-trivial Horn-SAT embedding.
\end{proof}

We have the following immediate corollary, asserting that Horn-SAT embeddings must be constant on $\Sigma_{{\sf modest}}$ (thereby giving
some sense that the relaxed base case holds for $\nu$).
\begin{corollary}\label{cor:Horn_vanish_sigma'}
  If $f\colon \tilde{\Sigma}_{{\sf fin}}\to\mathbb{C}$, $g\colon \Gamma_{{\sf dup}}\to\mathbb{C}$ and $h\colon \tilde{\Phi}_{{\sf dup}}\to\mathbb{C}$
  form an Horn-SAT embedding, then $f|_{\Sigma_{{\sf modest}}}$ is constant.
\end{corollary}
\begin{proof}
  Suppose towards contradiction that $f|_{\Sigma_{{\sf modest}}}$ is not constant. Note that the master embedding of
  $\nu$ assigns that same group element to both members of $\Sigma_{{\sf modest}}$, hence we get that $f\not\in{\sf Embed}_{\sigma}(\nu)$.
  Thus, $f,g,h$ is a non-trivial Horn-SAT embedding, hence by Claim~\ref{claim:find_sigma'} $f$ must vanish on $\Sigma'$, and as
  $\Sigma_{{\sf modest}}\subseteq \Sigma'$ it follows that $f$ must vanish on $\Sigma_{{\sf modest}}$, and contradiction.
\end{proof}
\subsubsection{Proof of Lemma~\ref{lem:nu_satisfies}: Compactness}
We begin by establishing a weak form of the relaxed base case via a compactness argument. As such, this argument does
not produce quantitative bounds (which are crucial for our application), however it serves as a good warm-up for the
actual argument proving the relaxed base case (which is very similar in spirit).

Namely, we show that for all $\tau>0$ there is $\tau'>0$ such that
if $f\colon \tilde{\Sigma}_{{\sf final}}\to\mathbb{C}$, $g\colon\Gamma_{\sf dup}\to\mathbb{C}$ and $h\colon \tilde{\Phi}_{\sf dup}\to\mathbb{C}$ are functions such that
$\Expect{x,x'\in\Sigma_{{\sf modest}}}{\card{f(x)-f(x')}^2}\geq \tau$, then
\[
\card{\Expect{(\vec{x},y,z)\sim \nu}{f(x)g(y)h(z)}}\leq (1-\tau')\norm{f}_2\norm{g}_2\norm{h}_2.
\]
Indeed, otherwise we could find a sequence $f_m, g_m, h_m$ with $2$-norm equal to $1$ such that the left hand side approaches $1$,
and by limiting we could find $f,g,h$ as above, in which $\card{\Expect{(\vec{x},y,z)\sim \nu}{f(x)g(y)h(z)}} = 1$.
By Cauchy-Schwarz we have
\[
1
=\card{\Expect{(x,y,z)\sim \nu}{f(x)g(y)h(z)}}
\leq
\sqrt{\Expect{(x,y,z)\sim \nu}{\card{f(x)}^2}}
\sqrt{\Expect{(x,y,z)\sim \nu}{\card{g(y)h(z)}^2}}
=\norm{f}_2\norm{g}_2\norm{h}_2 = 1,
\]
hence Cauchy-Schwarz is an equality and so $f(x) = \theta g(y)h(z)$ in the support of $\nu$, where $\theta$ has absolute value equal to
$1$. Multiplying $g$ by $\theta$ we get that $f(x) = g(y)h(z)$, hence this triplet forms an Horn-SAT embedding. By limiting we have
that $\E_{x,x'\in\Sigma_{{\sf modest}}}\left[\card{f(x) - f(x')}^2\right]\geq \tau$, and this contradicts Corollary~\ref{cor:Horn_vanish_sigma'}.

\subsubsection{Proof of Lemma~\ref{lem:nu_satisfies}: Unraveling Compactness}\label{sec:relaxed_base_quan}
In this section we show that $\nu$ satisfies the relaxed base case.
Fix $f,g,h$ and $\tau>0$ as in the definition of the relaxed base case;
by normalizing, we may assume that the $2$-norms of each one of $f,g$ and $h$ is equal to $1$.
We will assume that $\tau\leq c$ where $c = c(m,\alpha) > 0$ is a constant to be determined, as otherwise we may lower $\tau$.

Assume towards contradiction that $\card{\Expect{(x,y,z)\sim \nu}{f(x)g(y)h(z)}}\geq 1-\tau^{50m}$.
Thus there is a complex number $\theta$ of absolute value $1$ such that
$\theta\Expect{(x,y,z)\sim \nu}{f(x)g(y)h(z)}\geq 1-\tau^{50m}$, and to simplify notation
we multiply $f$ by $\theta$ so that this inequality becomes $\Expect{(x,y,z)\sim \nu}{f(x)g(y)h(z)}\geq 1-\tau^{50m}$;
we note we may multiply $f$ by a constant without loss of generality, as this does not affect any of the assumptions on $f$.
\begin{claim}\label{claim:relaxed1}
  $\Expect{x\in \Sigma_{{\sf modest}}}{\card{f(x)}^2}\geq \frac{\tau}{4}$.
\end{claim}
\begin{proof}
  By $\card{a-b}^2\leq 2\card{a}^2 + 2 \card{b}^2$ for all $a,b\in\mathbb{C}$ we get
  \[
  \Expect{x\in \Sigma_{{\sf modest}}}{\card{f(x)}^2}
  \geq
  \frac{1}{4}\Expect{x,x'\in \Sigma'}{\card{f(x)-f(x')}^2}
  \geq \frac{\tau}{4}.
  \qedhere
  \]
\end{proof}

\begin{claim}\label{claim:relaxed2}
  For all $(x,y,z)\in{\sf supp}(\nu)$ we have $\card{\overline{f(x)} - g(y)h(z)}\leq \tau^{20m}$.
\end{claim}
\begin{proof}
  We have
  \begin{align*}
    \Expect{(x,y,z)\sim\nu}{\card{\overline{f(x)} - g(y)h(z)}^2}
    &=2-
    \Expect{(x,y,z)\sim \nu}{f(x)g(y)h(z)}
    -\Expect{(x,y,z)\sim \nu}{\overline{f(x)g(y)h(z)}}\\
    &=2-2\Expect{(x,y,z)\sim \nu}{f(x)g(y)h(z)}\\
    &\leq 2\tau^{50m},
  \end{align*}
  and as the probability of each atom in $\nu$ is least $\alpha'(\alpha,m)>0$ we get that
  $\card{\overline{f(x)} - g(y)h(z)}^2\leq \alpha'^{-1}\tau^{50m}\leq \tau^{40m}$ for $c$ sufficiently
  small, giving the claim.
\end{proof}

\begin{claim}\label{claim:relaxed3}
  For all $y\in\Gamma_{{\sf dup}}$ we have $\card{g(y)}\ggg_{\alpha, m} \sqrt{\tau}$.
\end{claim}
\begin{proof}
   Assume this is not the case, so that there is $y^{\star}\in \Gamma_{{\sf dup}}$ such that $\card{g(y^{\star})}\leq \eta(\alpha,m)\sqrt{\tau'}$ where $\eta(\alpha,m)>0$
   is a small function of $\alpha$ and $m$.
   By Claim~\ref{claim:relaxed1} there is $\vec{x} = a(x,\ldots,x)\in \Sigma_{{\sf modest}}$ such that $\card{f(\vec{x})}\geq \sqrt{\tau}/2$ (where $x\in \Sigma$).
   As the support of $\mu'$ on $\Sigma\times \Gamma$ is full, there is $\vec{z}\in \tilde{\Phi}$ such that $(x,y^{\star},z)$ is in ${\sf supp}(\mu')$,
   hence $(\vec{x},y^{\star},\vec{z})$ is in the support of $\nu$, so by Claim~\ref{claim:relaxed2}
   \[
   \sqrt{\tau}/2
   \leq \card{f(\vec{x})}
   \leq \card{g(y^{\star})}\card{h(\vec{z},1)} + \tau^{20m}
   \leq \eta \sqrt{\tau} (\alpha')^{-1/2} + \tau^{20m}
   \]
   where we used the fact that $\card{h(\vec{z})}\leq (\alpha')^{-1/2}\norm{h}_2^2 = (\alpha')^{-1/2}$, where $\alpha' = \alpha'(\alpha,m)>0$
    is a lower bound on the probability of any atom in $\nu$. This is a contradiction for a small enough $\eta$.
\end{proof}

Consider the interval $(0,1)$, and in it define the intervals $I_j = [\tau^{3(j+1)}, \tau^{3j})$ for $j=0,1,\ldots$. We say an interval
	$I_j$ is free if it doesn't contain any point from either ${\sf Image}(f)$ or ${\sf Image}(h)$. Note that as the intervals $I_j$ are disjoint
	and each one of these sets has size at most $m$, we may find $j\in \set{0,1,\ldots,2m}$ such that $I_j$ is free, and we fix such $j$ henceforth.
	\begin{claim}\label{claim:relaxed4}
		We have that $\card{f(\vec{x})}\geq \tau^{3j+1.5}$ for all $\vec{x}\in \Sigma_{{\sf modest}}$ and
		$\card{h(\vec{z})}\geq \tau^{3j+1.5}$ for all $\vec{z}\in \tilde{\Phi}_{\sf dup}$.
	\end{claim}
	\begin{proof}
		Assume otherwise, and define
		\[
		A(x) = 1_{\card{f(x)}\geq \tau^{3j+1.5}},
		\qquad
		B(z) = 1_{\card{h(z)}\geq \tau^{3j+1.5}}.
		\]
		Then by our assumption, at least one of $A,B$ is not identically $1$, say $A$ without loss of generality.
		Note that since $\norm{f}_{2}=1$, there is some $x$ such that $\card{f(x)}\geq 1>\tau$, so $a$ is also not constantly
		$0$, hence $A$ is not constant. We next show that
		$A(x) + B(z) = 0\pmod{2}$ for all $(x,y,z)\in {\sf supp}(\nu)$, hence conclude that $\nu$ is not pairwise connected
        in contradiction to the fact that $\mu$ is pairwise connected and Lemma~\ref{lem:path_keeps_connectedness}.
		
		\paragraph{The case that $A(x)=1$.} Suppose that $(x,y,z)\in {\sf supp}(\nu)$ are such that
		$A(x) = 1$. It follows from Claim~\ref{claim:relaxed2} that
		\[
		\card{g(y)h(z)}\geq \card{f(x)} - \tau^{10m}\geq \tau^{3j+1.5} - \tau^{20m},
		\]
		and as $\card{g(y)}\leq \alpha^{-1}\norm{g}_{2,\tilde{\nu}_y}=\alpha^{-1}$, it follows that
		\[
		\card{h(z)}\geq \frac{\tau^{3j+1.5} - \tau^{20m}}{1/\alpha}\geq \tau^{3j+2},
		\]
		where we used the fact that $\tau<c$ is small enough, and $j\leq 2m$. Thus, as $I_j$ is free, it follows that
		$\card{h(z)}\geq \tau^{3j}$, and so $B(z) = 1$.
		
		\paragraph{The case that $A(x)=0$.} Suppose that $(x,y,z)\in {\sf supp}(\tilde{\nu})$ are such that
		$A(x) = 0$. It follows from Claim~\ref{claim:relaxed2} that
		\[
		\card{g(y)h(z)}\leq \card{f(x)} + \tau^{20m}\leq \tau^{3j+1.5} + \tau^{20m},
		\]
		and as $\card{g(y)}\geq \tau$ by Claim~\ref{claim:relaxed3}, it follows that
		\[
		\card{h(z)}\leq \frac{\tau^{3j+1.5} + \tau^{20m}}{\tau}\leq \tau^{3j+0.4},
		\]
		where we used the fact that $\tau<c$ is small enough, and $j\leq 2m$. Thus, as $I_j$ is free, it follows that
		$\card{h(z)}< \tau^{3(j+1)}$, and so $B(z) = 0$.
	\end{proof}

	Claim~\ref{claim:relaxed4} implies that the functions $f,g,h$ never get the value $0$, and also that in the support of $\nu$,
    \[
        \card{\frac{\overline{f(x)}}{g(y)h(z)} - 1}
        =\frac{\card{\overline{f(x)}-g(y)h(z)}}{\card{g(y)h(z)}}
        \leq \frac{\tau^{20m}}{\Omega_{\alpha,m}(\sqrt{\tau})\tau^{3j+1.5}}
        \leq \tau^{10m}.
    \]
    This motivates the definition of an approximate Abelian embedding by taking the principal branch of the logarithm:
	\[
	f'(x) = \log(\overline{f(x)}),
	\qquad
	g'(y) = \log(g(y)),
	\qquad
	h'(z) = \log(h(z)).
	\]
    Define $d(a,b) = \min_{k\in\mathbb{Z}}\card{a-b-2\pi{\bf i}k}$ for $a,b\in\mathbb{C}$.
    Then in the support of $\nu$ we have:
	\begin{equation}\label{eq:unravel_compact3}
	d(f'(x), g'(y) + h'(z))
	=d\left(\log\left(\frac{\overline{f(x)}}{g(y)h(z)}\right),0\right)
	\leq
    \card{\frac{\overline{f(x)}}{g(y)h(z)} - 1}
    \leq \tau^{10m}.
	\end{equation}
    We now use the logic of Claim~\ref{claim:embedding_fns} (and the underlying application of Dirichlet's Approximation Theorem)
    to slightly change $f',g'$ and $h'$ to get a proper Abeling embedding.
	
    We work separately with the real and imaginary part of $f',g',h'$.
    Looking at $S = {\sf Image}({\sf Re}(f'))\cup {\sf Image}({\sf Re}(g')) \cup {\sf Image}({\sf Re}(h'))$,
	we have that $\card{S}\leq 3m$. We take $N=(\alpha\tau)^{-9m}$; from Dirichlet's Approximation Theorem it follows that there are
	$\sigma_{\sf Re}\colon \Sigma_{{\sf fin}}\to \mathbb{Z}$, $\gamma_{\sf Re}\colon \Gamma\to\mathbb{Z}$
	and $\phi_{\sf Re}\colon \Phi\to\mathbb{Z}$ such that for some integer $1\leq q\leq N$ we have
	\begin{align*}
	&\card{{\sf Re}(f'(x)) - \frac{\sigma_{\sf Re}(x)}{q}}\leq \frac{1}{q N^{1/\card{S}}},
	\qquad
	\card{{\sf Re}(g'(y)) - \frac{\gamma_{\sf Re}(y)}{q}}\leq \frac{1}{q N^{1/\card{S}}},\\
	&\qquad\qquad\qquad\qquad
	\card{{\sf Re}(h'(z)) - \frac{\phi_{\sf Re}(z)}{q}}\leq \frac{1}{q N^{1/\card{S}}}.
	\end{align*}
    Similarly we find $\sigma_{\sf Im}\colon\Sigma_{{\sf fin}}\to \mathbb{Z}$, $\gamma_{\sf Im}\colon \Gamma\to\mathbb{Z}$ and $\phi_{\sf Im}\colon \Phi\to\mathbb{Z}$ such that
	\begin{align*}
	&\card{{\sf Im}(f'(x)) - \frac{\sigma_{\sf Im}(x)}{q}}\leq \frac{1}{q N^{1/\card{S}}},
	\qquad
	\card{{\sf Im}(g'(y)) - \frac{\gamma_{\sf Im}(y)}{q}}\leq \frac{1}{q N^{1/\card{S}}},\\
	&\qquad\qquad\qquad\qquad\card{{\sf Im}(h'(x)) - \frac{\phi_{\sf Im}(z)}{q}}\leq \frac{1}{q N^{1/\card{S}}}.
	\end{align*}
    Define
    \[
    \tilde{\sigma}(x) = \sigma_{{\sf Re}}(x) + 2\pi{\bf i}\sigma_{{\sf Im}}(x),
    \qquad
    \tilde{\gamma}(y) = \gamma_{{\sf Re}}(y) + 2\pi{\bf i}\gamma_{{\sf Im}}(y),
    \qquad
    \tilde{\phi}(z) = \phi_{{\sf Re}}(yz) + 2\pi{\bf i}\phi_{{\sf Im}}(z).
    \]

	\begin{claim}\label{claim:relaxed_base_case_embed_2}
		$\tilde{\sigma}(x) - \tilde{\gamma}(y) - \tilde{\phi}(z) = 0\pmod{2\pi{\bf i}}$ for all $(x,y,z)\in {\sf supp}(\nu)$.
	\end{claim}
	\begin{proof}
		From the choice of $\tilde{\sigma},\tilde{\gamma},\tilde{\phi}$ and~\eqref{eq:unravel_compact3} it follows that for all $(x,y,z)\in {\sf supp}(\nu)$ it holds that
		\[
		d\left(\frac{\tilde{\sigma}(x)}{q}, \frac{\tilde{\gamma}(y) + \tilde{\phi}(z)}{q}\right)
		\leq \frac{6}{q N^{1/\card{S}}} + d\left(f'(x), g'(y) + h'(z)\right)
        \leq \frac{6}{q N^{1/\card{S}}} + \tau^{10m}
		\]
        where we used the triangle inequality and~\eqref{eq:unravel_compact3}. Multiplying by $q$ we get that
        $d\left(\sigma(x), \gamma(y) + \phi(z)\right)\leq \frac{6}{N^{1/\card{S}}} + q\tau^{10m} < 1$,
        where we used the fact that $q\leq N$, the choice of $N$ and the fact that $\tau$ is sufficiently small.
        As $\tilde{\sigma}(x) - \tilde{\gamma}(y) - \tilde{\phi}(z)$ is a complex number of the form $a+b{\bf i}$ for $a,b\in \mathbb{Z}$,
        it must be the case that $a=0$ and $b$ is a multiple of $2\pi$, and the claim follows.
	\end{proof}
	
    Looking at the real parts of $\tilde{\sigma},\tilde{\gamma}$ and $\tilde{\phi}$, we get from Claim~\ref{claim:embedding_fns} that
    ${\sf Re}(\tilde{\sigma})\in {\sf Embed}_{\sigma}(\nu)$
    and looking at their imaginery part we get that they form an embedding into an infinite cyclic group, hence by Lemma~\ref{lem:turn_infinite_to_finite}
    they are equivalent to an embedding into a finite Abelian group, and again by Claim~\ref{claim:embedding_fns} we get that ${\sf Im}(\tilde{\sigma})\in {\sf Embed}_{\sigma}(\nu)$.
    It follows that $\tilde{\sigma}$ is in ${\sf Embed}_{\sigma}(\nu)$, hence $e^{\tilde{\sigma}}$ is in ${\sf Embed}_{\sigma}(\nu)$.
    Hence $e^{\tilde{\sigma}}$ is constant on $\Sigma_{{\sf modest}}$ and so
    \begin{align*}
    \Expect{x,x'\in\Sigma_{{\sf modest}}}{\card{f(x) - f(x')}^2}
    &=
    \Expect{x,x'\in\Sigma_{{\sf modest}}}{\card{(f(x) - e^{\tilde{\sigma}(x)/q}) - (f(x')-e^{\tilde{\sigma}(x')/q})}^2}\\
    &\leq
    4\Expect{x\in\Sigma_{{\sf modest}}}{\card{f(x) - e^{\tilde{\sigma}(x)/q}}^2}\\
    &=
    4\Expect{x\in\Sigma_{{\sf modest}}}{\card{f(x)}^2\card{e^{\tilde{\sigma}(x)/q - f'(x)}-1}^2}
    \end{align*}
    For each $x$, $\frac{\tilde{\sigma}(x)}{q} - f'(x)$ is a complex number $a+b{\bf i}$ for $a,b\in \mathbb{R}$ which are at most
    $\frac{1}{qN^{1/\card{S}}}$ in absolute value, hence
    \[
    \card{e^{\tilde{\sigma}(x)/q - f'(x)}-1}
    \leq \card{e^{a}{\sf cos}(b) - 1} + \card{e^a{\sf sin}(a)}
    \leq O\left(\frac{1}{qN^{1/\card{S}}}\right)
    \leq \tau^{8}.
    \]
    Thus, we get that
    \[
    \Expect{x,x'\in\Sigma_{{\sf modest}}}{\card{f(x) - f(x')}^2}
    \leq
    4\tau^{16} \Expect{x\in\Sigma_{{\sf modest}}}{\card{f(x)}^2}
    \leq 4\tau^{16}(\alpha')^{-1}\norm{f}_2^2
    \leq \tau^{15},
    \]
    and contradiction to the assumption that $\Expect{x,x'\in\Sigma_{{\sf modest}}}{\card{f(x) - f(x')}^2}\geq \tau$.

    \subsubsection{Proof of Lemma~\ref{lem:nu_satisfies}: $\nu$ Satisfies the Other Conditions of Theorem~\ref{thm:nonembed_deg_must_be_small_rephrase_maximal_relaxed}}
    \label{sec:the_other_conditions_of_nu}
    The first, second, third and seventh properties of $\nu$ are all direct by the construction.
    We have already argued about the fourth and fifth properties of $\nu$ in the beginning of the proof of Lemma~\ref{lem:nu_satisfies}.
    Property 6a was established in Section~\ref{sec:relaxed_base_quan}, and now we only need to verify property 6b.
    For that we describe a distribution $\tilde{\nu}$ satisfying the properties therein.

    Take $\Sigma'' = \sett{(x,\ldots,x)}{x\in \Sigma}$, recall that $\Gamma' = \sett{(y,1)}{y\in\Gamma}$ and let $\tilde{\nu}$ be the distribution
    of
    \[
    ((x,\ldots,x), (y,1), ((z,\ldots,z),1))
    \]
    where $(x,y,z)\sim \mu$. The map $a$ that we take is the same as the map $a$ we defined to
    make the $x$-merge. Property 6b(i) is clear by the definition of $\Sigma'$ and $\Sigma''$, and property 6b(ii) is clear
    by the properties of the path trick. Property 6b(iv) follows as in $\mu$ the value of any two coordinates implies the last one,
    and property 6b(v) is immediate by the construction. Property 6b(vi) was verified above in the body of the proof of Lemma~\ref{lem:nu_satisfies}.
    We now argue about maximality, property 6b(iii).

    Suppose we have
    $f\colon \Sigma''^{n}\to\mathbb{C}$, $g\colon {\Gamma'}^n\to\mathbb{C}$ and $h\colon \Phi'^n\to \mathbb{C}$ that are $1$-bounded, and define
    $f'\colon \Sigma^n\to\mathbb{C}$, $g'\colon \Gamma^{n}\to\mathbb{C}$ and $h'\colon \Phi^n\to\mathbb{C}$ by
    \begin{align*}
    &f'(x_1,\ldots,x_n) = f((x_1,\ldots,x_1),\ldots, (x_n,\ldots,x_n)),
    \qquad
    g'(y_1,\ldots,y_n) = g((y_1,1),\ldots,(y_n,1)),\\
    &\qquad\qquad\qquad\qquad
    h'(z_1,\ldots,z_n) = h(((z_1,\ldots,z_1),1),\ldots,((z_n,\ldots,z_n),1)).
    \end{align*}
    Then
    \[
    \Expect{(\vec{x}, y, \vec{z})\sim \tilde{\nu}^{\otimes n}}{f(\vec{x})g(y)h(\vec{z})}
    =\Expect{(x,y,z)\sim \mu^{\otimes n}}{f'(x)g'(y)h'(z)}.
    \]
    The point is now that if we take a distribution $\mathcal{D}$ over $\Sigma''\times \Gamma' \times \Phi'$ whose support strictly contains
    the support of $\tilde{\nu}$, then we could repeat a similar reasoning above and relate our expectation to an expectation with respect to a
    distribution $\tilde{\mathcal{D}}$ whose support strictly contains $\mu$. Indeed, let $\mathcal{D}$ be such distribution and define $\tilde{\mathcal{D}}$ by taking
    $(\vec{x}, (y,1), (\vec{z},1))\sim \mathcal{D}$ and outputting $(x,y,z)$. Clearly the support of $\tilde{\mathcal{D}}$ contains the support
    of $\tilde{\mu}$, and hence the support of $\mu$. If it was the case that ${\sf supp}(\tilde{\mathcal{D}}) = {\sf supp}(\mu)$, then
    we would get that
    \[
    {\sf supp}(\mathcal{D})\subseteq \sett{(\vec{x}, (y,1), (\vec{z},1))}{\exists x\in\Sigma, z\in\Phi, \vec{x} = (x,\ldots,x), \vec{z} = (z,\ldots,z), (x,y,z)\in {\sf supp}(\mu)},
    \]
    but this is contained in the support of $\tilde{\nu}$ in contradiction. Thus we get that taking $f'$, $g'$ and $h'$ as above we have that
    \[
    \Expect{(\vec{x}, y, \vec{z})\sim \mathcal{D}^{\otimes n}}{f(\vec{x})g(y)h(\vec{z})}
    =\Expect{(x,y,z)\sim \tilde{\mathcal{D}}^{\otimes n}}{f'(x)g'(y)h'(z)},
    \]
    and $\tilde{\mathcal{D}}$ is a distribution whose support strictly contains ${\sf supp}(\mu)$.
    If ${\sf NEStab}_{1-\delta}(g; \mathcal{D})\leq \delta$, then by Lemma~\ref{lem:compare_averaging_ops}
    we get that ${\sf NEStab}_{1-\delta}(g'; \tilde{\mathcal{D}})\leq \delta^s$, where $s=s(m,\alpha)>0$
    (as the master embedding of $y$ in $\mathcal{D}$ is a refinement of the master embedding of $y$ in $\tilde{\mathcal{D}}$,
    we get that the corresponding averaging operators satisfy the properties of the lemma).
    Thus, by Claim~\ref{claim:increase_noise_decrease_stab} we get that
    ${\sf NEStab}_{1-\delta^s}(g; \tilde{\mathcal{D}})\leq \delta^s$, and as $\mu$ is maximal
    we get that
    \[
        \card{\Expect{(x,y,z)\sim \tilde{\mathcal{D}}^{\otimes n}}{f'(x)g(y)h'(z)}}\leq M\delta^{s\eta}
    \]
    for some $M,\eta$ depending only on $m,\alpha$, concluding the proof.\qed
\subsection{The Case that $\card{\Sigma_{{\sf modest}}} = 1$: Proof of Lemma~\ref{lem:ugly_modest_case}}\label{sec:ugly_modest_case}
Let $f,g,h$ be functions as in Theorem~\ref{thm:nonembed_deg_must_be_small_rephrase_maximal2}. Using Lemma~\ref{lem:from_mu_to_path} repeatedly,
we get that there is $\ell = O_{m,\alpha}(1)$ such that
\[
\card{\Expect{(x,y,z)\sim\mu^{\otimes n}}{f(x)g(y)h(z)}}^{\ell}
\leq \card{\Expect{(x,y,z)\sim \mu''^{\otimes n}}{F(x)g(y)H(z)}},
\]
where $F\colon \tilde{\Sigma}^n\to \mathbb{C}$ and $H\colon \tilde{\Phi}^{n}\to\mathbb{C}$ are some $1$-bounded functions (and $g$ remains the same). Our goal is to show that
this quantity is at most $M\delta^{\eta}$ where $M\in\mathbb{N}$ and $\eta>0$ depend only on $m,\alpha$. By Lemma~\ref{lem:merge}, this would follow if we
show the statement for the $x$-merge distribution of $\mu''$, which is nothing but $\nu'$, hence it suffices to show that there are $M$ and $\eta>0$
such that for $F\colon\tilde{\Sigma}_{{\sf final}}\to[-1,1]$, $g\colon \Gamma^n\to[-1,1]$ and $h\colon \tilde{\Phi}^n\to[-1,1]$ such that
${\sf NEStab}_{1-\delta}[g;\nu_y^{\otimes n}]\leq \delta$ it holds that
\[
\card{\Expect{(\tilde{x},y,\tilde{z})\sim {\nu'}^{\otimes n}}{F(\tilde{x})g(y)H(\tilde{z})}}\leq M\delta^{\eta}.
\]
We focus on this task henceforth. Recall $\Sigma' = \sett{a(x,\ldots,x)}{x\in \Sigma}$ and define
$\Phi'_{{\sf pre}} = \sett{(z,\ldots,z)}{z\in \Phi}$ as well as $\Sigma'_{{\sf pre}} = \sett{(x,\ldots,x)}{x\in \Sigma}$.
Let $\tilde{\nu}$ be the distribution of $(\tilde{x},y,\tilde{z})\sim \nu$ conditioned on $\tilde{x}\in \Sigma'$ and $\tilde{z}\in \Phi'$,
and write $\nu' = \beta\tilde{\nu} + (1-\beta)\nu''$ where $\beta > 0$ depends only on $\alpha$ and $\nu''$ is some distribution. Choose
$J\subseteq_{\beta} [n]$, $(\tilde{x}',y',\tilde{z}')\sim \nu''^{\overline{J}}$
and define $\tilde{F}\colon \Sigma'^{J}\to\mathbb{C}$, $\tilde{g}\colon \Gamma^{J}\to\mathbb{C}$ and $\tilde{H}\colon \tilde{\Phi}^{J}\to\mathbb{C}$ by
\[
\tilde{F} = F_{\overline{J}\rightarrow \tilde{x}'},
\qquad
\tilde{g} = g_{\overline{J}\rightarrow y'},
\qquad
\tilde{H} = H_{\overline{J}\rightarrow \tilde{z}'}.
\]
Denote $\phi_{\nu'}(F,g,H) = \Expect{(\tilde{x},y,\tilde{z})\sim {\nu'}^{\otimes n}}{F(\tilde{x})g(y)H(\tilde{z})}$ and analogously define
$\phi_{\tilde{\nu}}(\tilde{F},\tilde{g},\tilde{H})$. Then
\begin{equation}\label{eq:handle_ugly_modest_case2}
\phi_{\nu'}(F,g,H)
=\hspace{-2ex}\Expect{(\tilde{x}',y',\tilde{z}')\sim \nu''^{\overline{J}}}{\phi_{\tilde{\nu}}(\tilde{F}, \tilde{g},\tilde{H})}
=\hspace{-2ex}
\underbrace{\Expect{(\tilde{x}',y',\tilde{z}')\sim \nu''^{\overline{J}}}{1_E\phi_{\tilde{\nu}}(\tilde{F}, \tilde{g},\tilde{H})}}_{(\rom{1})}
+
\underbrace{\Expect{(\tilde{x}',y',\tilde{z}')\sim \nu''^{\overline{J}}}{1_{\overline{E}}\phi_{\tilde{\nu}}(\tilde{F}, \tilde{g},\tilde{H})}}_{(\rom{2})},
\end{equation}
where $E$ is the event that ${\sf NEStab}_{1-c^{-1}\beta^{-1}\delta,\nu'}(\tilde{g};\tilde{\nu}_y^{J})\leq \sqrt{\delta}$; here,
$c>0$ is from Claim~\ref{claim:rr_nestab} (with $\mathcal{D} = {\nu'}_y$, $\mathcal{D}' = \tilde{\nu}_y$ and $\mathcal{D}'' = {\nu''}_y$ in the notations therein).
By Claim~\ref{claim:rr_nestab} we have
\[
\Expect{J,y'}{{\sf NEStab}_{1-c^{-1}\beta^{-1}\delta,\nu'}(\tilde{g};{\tilde{\nu}}_y^{J})}
\leq
{\sf NEStab}_{1-\delta, \nu'}(g;{\nu'}_y^{\otimes n})
\leq
\delta,
\]
so by Markov's inequality $\Prob{}{\overline{E}}\leq \sqrt{\delta}$. Hence, as $\tilde{F},\tilde{g}$ and $\tilde{H}$ are $1$-bounded,
we have that $\card{(\rom{2})}\leq \Prob{}{\overline{E}}\leq \sqrt{\delta}$. In the rest of the argument we bound $(\rom{1})$.

\paragraph{Bounding $(\rom{1})$.} Fix $J$ and $\tilde{x}',y',\tilde{z}'$ so that the event $E$ holds;
we show that then $\card{\phi_{\tilde{\nu}}(\tilde{F}, \tilde{g},\tilde{H})}\leq M\delta^{\eta}$, where $M\in\mathbb{N}$ and $\eta>0$
only depend on $m$ and $\alpha$. The main idea is to re-interpret $\phi_{\tilde{\nu}}(\tilde{F}, \tilde{g},\tilde{H})$ as an expectation
with respect to a distribution over $\Sigma\times \Gamma\times \Phi$ whose support strictly contains the support of $\mu$, and
then use the maximality of $\mu$.

Without loss of generality, we assume for notational convenience that $J = \{1,\ldots,n'\}$.
Consider the functions $F^{\sharp}\colon \Sigma^{J}\to\mathbb{C}$ and $H^{\sharp}\colon \Phi^{J}\to\mathbb{C}$
defined as
\[
F^{\sharp}(x) = \tilde{F}(a(x_1,\ldots,x_1),\ldots,a(x_{n'},\ldots,x_{n'})),
\qquad
H^{\sharp}(z) = \tilde{H}((z_1,\ldots,z_1),\ldots,(z_{n'},\ldots,z_{n'})).
\]
Consider the distribution $\mathcal{D}$ over $\Sigma\times \Gamma\times \Phi$ that results from sampling $(X, y, Z)$ according to
$\tilde{\nu}$, writing $Z = (z,\ldots,z)$ and sampling $(x,\ldots,x)\sim \mu''|_{\tilde{\Sigma}}$ conditioned on $a(x,\ldots,x) = X$,
and then outputting $(x,y,z)$. Then:
\begin{equation}\label{eq:handle_ugly_modest_case1}
\phi_{\tilde{\nu}}(\tilde{F}, \tilde{g},\tilde{H})
=\Expect{(x,y,z)\sim\mathcal{D}^{J}}{F^{\sharp}(x)\tilde{g}(y)H^{\sharp}(z)}.
\end{equation}
It is easily seen that the probability of each atom in $\mathcal{D}$ is at least $\alpha'>0$ that depends only on $m$ and $\alpha$.
We argue that ${\sf supp}(\mu)$ is strictly contained in ${\sf supp}(\mathcal{D})$. First, containment is clear by construction, and we next observe
the strict containment. In $\mu$ we have that $y,z$ implies $x$, and in the construction of $\nu$ we first took distinct $x^{\star}, {x^{\star}}'\in \Sigma$
(satisfying some other property that is not important for now), and took $\Sigma_{{\sf modest}} = \set{a(x^{\star},\ldots,x^{\star}),a({x^{\star}}',\ldots,{x^{\star}}')}$.
Thus, there are $y,y'\in\Gamma$ and $z,z'\in\Phi$ such that $(x^{\star},y,z),({x^{\star}}',y',z')\in {\sf supp}(\mu)$, and it must be the case that $(x^{\star},y',z')$ is not in
${\sf supp}(\mu)$ (otherwise, this would mean that the value of the $y$ and $z$ coordinate does not imply the value of the $x$ coordinate in $\mu$).
We argue that $(x^{\star},y',z')$ is in ${\sf supp}(\mathcal{D})$. Indeed, as by assumption $\card{\Sigma_{{\sf modest}}}=1$ we have
that $a(x^{\star},\ldots,x^{\star}) = a({x^{\star}}',\ldots,{x^{\star}}')$, and we have that for $X = a({x^{\star}}',\ldots,{x^{\star}}')$,
$Z = (z',\ldots,z')$ it holds that $(X,y',Z)$ is in ${\sf supp}(\nu)$, hence by definition of $\mathcal{D}$ and the fact that $a(x^{\star},\ldots,x^{\star}) = X$
we get that $(x^{\star}, y', z')$ is in ${\sf supp}(\mathcal{D})$.

In conclusion, we get that the expectation on the right hand side of~\eqref{eq:handle_ugly_modest_case1} is an expectation with respect to a distribution $\mathcal{D}$
whose support strictly contains the support of $\mu$ and hence we can use the maximality of $\mu$. One subtle point is that the ``high degreeness'' of $g$ is not phrased
in quite the appropriate language; as the event $E$ holds we know that ${\sf NEStab}_{1-c^{-1}\beta^{-1}\delta,\nu'}(\tilde{g};\tilde{\nu}_y^{J})\leq \sqrt{\delta}$ and we
need to conclude from this high-degreeness with respect to non-embeddability in $\mathcal{D}$.

As $\tilde{\nu}_y = \mathcal{D}_y$ we have
${\sf NEStab}_{1-c^{-1}\beta^{-1}\delta,\nu'}(\tilde{g};\tilde{\nu}_y^{J}) = {\sf NEStab}_{1-c^{-1}\beta^{-1}\delta,\nu'}(\tilde{g};\mathcal{D}_y^{J})$. Next, note
that any Abelian embedding of $\nu'$ can be used to define an Abelian embedding of $\mathcal{D}$. This is done by mapping $x$ to $a(x,\ldots,x)$ and then applying the embedding of $\nu'$
on the first coordinate, applying the embedding of $\nu'$ on the second coordinate on $y$, and mapping $z$ to $(z,\ldots,z)$ and then applying the embedding of
$\nu'$ on the third coordinate. Thus, the partition of $\Gamma$ defined by master embedding of $\nu'$ is a refinement of the partition of $\Gamma$ defined by the master embedding of $\mathcal{D}$.
Applying Lemma~\ref{lem:compare_averaging_ops} we conclude that
\[
{\sf NEStab}_{1-c^{-1}\beta^{-1}\delta,\mathcal{D}}(\tilde{g};\mathcal{D}_y^{J})\leq {\sf NEStab}_{1-c^{-1}\beta^{-1}\delta,\nu'}(\tilde{g};\mathcal{D}_y^{J})^s
\leq \delta^{s/2}.
\]
where $s = s(m,\alpha)>0$. By Claim~\ref{claim:increase_noise_decrease_stab} we get that
${\sf NEStab}_{1-\delta^{s/2},\mathcal{D}}(\tilde{g};\mathcal{D}_y^{J})\leq \delta^{s/2}$
provided that $\delta_0$ is small enough, hence by maximality of $\mu$ it follows that
\[
\card{\Expect{(x,y,z)\sim\mathcal{D}^{J}}{F^{\sharp}(x)\tilde{g}(y)H^{\sharp}(z)}}\leq M'\delta^{s\eta'/2}
\]
where $M'$ and $\eta'>0$ depends only on $\alpha$ and $m$. Plugging this into~\eqref{eq:handle_ugly_modest_case1} we have that
$\card{\phi_{\nu'}(\tilde{F}, \tilde{g},\tilde{H})}\leq \delta^{s\eta'/2}$, and so $\card{(\rom{1})}\leq \delta^{s\eta'/2}$.

\paragraph{Combining the bounds on $(\rom{1}), (\rom{2})$.}
Plugging the bounds on $(\rom{1})$, $(\rom{2})$ into~\eqref{eq:handle_ugly_modest_case2} yields that
\[
\card{\phi_{\nu}(F,g,H)}\leq M'\delta^{s\eta'/2} + \sqrt{\delta}\leq M''\delta^{\eta''}
\]
for $M'' = M'+1$ and $\eta'' = s\eta'/2$, and we are done.\qed

\section{Reducing to the Homogenous Statement}\label{sec:reudce_to_homogenous}
In this section, our goal is to reduce Theorem~\ref{thm:nonembed_deg_must_be_small_rephrase_maximal_relaxed} into a
result that relaxes the assumption that our functions are $1$-bounded to the assumption that they are bounded $2$-norm
(which is therefore more amendable to a proof by induction). See Theorem~\ref{thm:nonembed_homogenous} for a precise statement.

\subsection{Degree, Non-embedding degree and Effective Non-embedding Degree}
Recall that in Section~\ref{sec:partial_basis} we showed that given a distribution $\mu$ that has saturated master embeddings,
one may construct a basis $B_1\cup B_2$ for $L_2(\Sigma; \mu_x)$ in which $B_1$ consists of embedding functions
and $B_2$ consists of functions that are orthogonal to all embeddings functions. The goal of this section is to refine this
further so as to be more compatible with the relaxed base case.

\begin{definition}
  Let $\mu$ be a distribution over $\Sigma\times \Gamma\times \Phi$ as in Theorem~\ref{thm:nonembed_deg_must_be_small_rephrase_maximal_relaxed}.
  We define an orthonormal basis $B_{{\sf embed}}\cup B_{{\sf non-embed}}\cup B_{{\sf modest}}$ for $L_2(\Sigma;\mu_x)$, as follows:
  \begin{enumerate}
    \item Consider the space of embedding functions, ${\sf Embed}_{\sigma}(\mu) = {\sf span}\left(\sett{\chi\circ \sigma}{\chi\in\widehat{H}}\right)$,
    and pick an orthonormal basis $B_{{\sf embed}} = \sett{\chi\circ \sigma}{\chi\in\widehat{H}}$ for it.

    \item Consider the orthogonal space to ${\sf Embed}_{\sigma}(\mu)$, namely ${\sf Embed}_{\sigma}(\mu)^{\perp}$, and consider the subspace of
    it of functions that are constant on $\Sigma_{{\sf modest}}$:
    \[
        \nenm_{\sigma}(\mu) = \sett{f\colon \Sigma\to\mathbb{C}}{f\in {\sf Embed}_{\sigma}(\mu)^{\perp}, f|_{\Sigma_{{\sf modest}}}\text{ is constant}}.
    \]
    Pick $B_{{\sf non-embed}}$ an orthonormal basis of $\nenm_{\sigma}(\mu)$.
    \item Consider the space $\nenm_{\sigma}(\mu)^{\perp}\cap {\sf Embed}_{\sigma}(\mu)^{\perp}$, and take $B_{{\sf modest}}$ to be an orthonormal basis
    for it.
  \end{enumerate}
\end{definition}

With the basis $B_{{\sf embed}}\cup B_{{\sf non-embed}}\cup B_{{\sf modest}}$ in hand, we can now construct an orthonormal basis for
$L_2(\Sigma^{n}, \mu_x^{\otimes n})$ by tensorizing. Namely, we take $(B_{{\sf embed}}\cup B_{{\sf non-embed}}\cup B_{{\sf modest}})^{\otimes n}$,
as an orthonormal basis of $L_2(\Sigma^{n}, \mu_x^{\otimes n})$, and thus we may write any $f\colon \Sigma^n\to\mathbb{C}$ as
\[
f(x) = \sum\limits_{\chi\in (B_{{\sf embed}}\cup B_{{\sf non-embed}}\cup B_{{\sf modest}})^{\otimes n}}{\widehat{f}(\chi)\prod\limits_{i=1}^{n} \chi_i(x_i)},
\qquad
\text{where }
\widehat{f}(\chi) = \inner{f}{\chi}.
\]
\begin{definition}
  A function $\chi$ of the form $\chi(x) = \prod\limits_{i=1}^{n}\chi_i(x_i)$ where $\chi_i \in B_{{\sf embed}}\cup B_{{\sf non-embed}}\cup B_{{\sf modest}}$
  for all $i$ is called a monomial.
\end{definition}

There are also several important notions of degree that may be associated with monomials, which extend Definition~\ref{def:degrees_early}.
\begin{definition}
  Let $\chi = \prod\limits_{i=1}^{n} \chi_i$ be a function in $(B_{{\sf embed}}\cup B_{{\sf non-embed}}\cup B_{{\sf modest}})^{\otimes n}$.
  \begin{enumerate}
    \item For a character $a\in \hat{H}$, the $a$-embedding degree of $\chi$, denoted by ${\sf embeddeg}_a(\chi)$ is the number of coordinates $i$ on which $\chi_i = a$.
    \item The embedding degree of $\chi$, denoted by ${\sf embeddeg}(\chi)$, is the number of coordinates $i$ on which $\chi_{i}\in B_{{\sf embed}}$.
    In other words, ${\sf embeddeg}(\chi) = \sum\limits_{a}{\sf embeddeg}_a(\chi)$.
    \item The non-embedding degree of $\chi$, denoted by ${\sf nedeg}(\chi)$, is the number of $i$ on which $\chi_i\not\in B_{{\sf embed}}$, that is,
    $\card{\sett{i\in [n]}{\chi_i\not\in B_{{\sf embed}}}}$.
    \item The effective non-embeding degree of $\chi$, denoted by $\effnon(\chi)$, is the number of $i$ on which $\chi_{i}\in B_{{\sf modest}}$, that is,
    $\card{\sett{i\in [n]}{\chi_i\in B_{{\sf modest}}}}$.
  \end{enumerate}
\end{definition}
We note that clearly, for every monomial $\chi$ it holds that $\effnon(\chi)\leq {\sf nedeg}(\chi)\leq  n$,
and that ${\sf nedeg}(\chi) + {\sf embeddeg}(\chi) = n$.

\subsubsection{The Modest Markov Chain}
Recalling the non-embedding
noise operator $\mathrm{T}_{\text{non-embed}, 1-\delta}^{\otimes n}\colon L_2(\Sigma^n, \mu_x^{\otimes n})\to L_2(\Sigma^n, \mu_x^{\otimes n})$, we have by
Fact~\ref{fact:soft_nonbembed_op} that $\mathrm{T}_{{\sf non-embed}, 1-\delta}^{\otimes n}\chi = (1-\delta)^{{\sf nedeg}(\chi)}\chi$. Next, we design the
\emph{effective non-embedding} noise operator, which will be helpful for us later on. Towards this end, we define the modest Markov chain associated with
$\mu$ on $x$.

\begin{definition}
   The modest Markov chain, on $\Sigma$, denoted by ${\sf Modest}$, is the Markov chain that on $x\in \Sigma$ takes $x' = x$ if $x\not\in \Sigma_{{\sf modest}}$,
   and otherwise, if $x\in \Sigma_{{\sf modest}}$, samples $x'\sim \mu_x$ conditioned on $x'\in \Sigma_{{\sf modest}}$.
\end{definition}

\subsubsection{The Effective Noise Operator and Effective Non-embedding Degree}
The modest Markov chain will be useful for us to define several notions. The first of which is the effective noise operator, and we
first define the $\mathrm{E}_{\text{non-embed}, 1-\delta}$.
\begin{definition}
  For $\delta>0$, the Markov chain $\mathrm{E}_{\text{non-embed}, 1-\delta}$ on $\Sigma$ is the Markov chain that on $x\in \Sigma$, takes
  $x' = x$ with probability $1-\delta$, and otherwise samples a neighbour $x'$ of $x$ according to the modest Markov chain.
\end{definition}
Clearly, $\mu_x$ is a stationary distribution of $\mathrm{E}_{\text{non-embed}, 1-\delta}$, and as usual we associated this Markov chain
an averaging operator acting. Abusing notation, we denote it by $\mathrm{E}_{\text{non-embed}, 1-\delta}\colon L_2(\Sigma; \mu_x)\to L_2(\Sigma; \mu_x)$
and define it as
\[
    \mathrm{E}_{\text{non-embed}, 1-\delta} f(x) = \Expect{x'\sim \mathrm{E}_{\text{non-embed}, 1-\delta} x}{f(x')}.
\]
The following lemma gives the most basic properties of $\mathrm{E}_{\text{non-embed}, 1-\delta}$. In words, any function in
$B_{{\sf embed}}\cup B_{{\sf non-embed}}$ is an eigenfunction of it with eigenvalue $1$ and any function in $B_{{\sf modest}}$
is an eigenfunction with eigenvalue $1-\delta$.
\begin{lemma}\label{lem:effective_noise}
  For all $\delta>0$ we have that for any $\chi\in B_{{\sf embed}}\cup B_{{\sf non-embed}}\cup B_{{\sf modest}}$,
  \begin{enumerate}
    \item If $\chi\in B_{{\sf modest}}$, then $\mathrm{E}_{\text{non-embed}, 1-\delta} \chi = (1-\delta)\chi$.
    \item Else, $\mathrm{E}_{\text{non-embed}, 1-\delta} \chi = \chi$.
  \end{enumerate}
  Consequently, for a monomial $\chi\colon \Sigma^n\to\mathbb{C}$ we have
  $\mathrm{E}_{\text{non-embed}, 1-\delta}^{\otimes n}\chi = (1-\delta)^{\effnon(\chi)}\chi$.
\end{lemma}
\begin{proof}
  For the first item, note that for all $x\in \Sigma$,
  \[
  \mathrm{E}_{\text{non-embed}, 1-\delta} \chi(x)
  =(1-\delta)\chi(x) + \delta\Expect{x'\sim  {\sf Modest}~x}{\chi(x')},
  \]
  and we show the last expectation is equal to $0$. We may write this expectation as a sum
  $\sum\limits_{x'\in \Sigma}p(x') \chi(x')$ where $p(x')$ it the probability that a set on the modest Markov chain
  reaches $x'$ from $x$. Clearly, we may write this sum as $\inner{\chi}{\frac{p}{\mu_x}}_{\mu_x}$.
  We note that $p/\mu_x $ is constant on $\Sigma_{{\sf modest}}$. Indeed, if $a,b\in\Sigma_{{\sf modest}}$, then
  $p(a) = \frac{\mu_x(a)}{\mu_x(\Sigma_{{\sf modest}})}$ and $p(b) = \frac{\mu_x(b)}{\mu_x(\Sigma_{{\sf modest}})}$.
  Thus, we have that $p/\mu_x$ is in the span of $B_{{\sf embed}}\cup B_{{\sf non-embed}}$ (as these contain all
  functions that are constant on $\Sigma_{{\sf modest}}$), hence $p/\mu_x$ is perpendicular to $B_{{\sf modest}}$
  and so $\inner{\chi}{\frac{p}{\mu_x}}_{\mu_x}=0$.

  For the second item, note that any such $\chi$ is constant on $\Sigma_{{\sf modest}}$ and the Markov chain
  $\mathrm{E}_{\text{non-embed}, 1-\delta}$ stays at any $x\not\in \Sigma_{{\sf modest}}$, and otherwise stays
  inside $\Sigma_{{\sf modest}}$.
\end{proof}

\subsubsection{Modest Influences}
The modest Markov chain will also be useful for us to define the notion of modest influence of a coordinate as well as the modest total influence.
\begin{definition}
  In the setting of Theorem~\ref{thm:nonembed_deg_must_be_small_rephrase_maximal_relaxed}, for a function
  $f\colon (\Sigma^n, \mu_x^{\otimes n})\to\mathbb{C}$ and a coordinate $i\in [n]$, we define the modest
  influence of $f$ at coordinate $i$ as
  \[
    I_{i, \text{modest}}[f] = \Expect{x'\sim\mu_x^{n-1}, a\sim \mu_x, b\sim {{\sf Modest}}~a}{\card{f(x_{-i} = x', x_i = a) - f(x_{-i} = x', x_i = b)}^2}.
  \]
  The modest total influence of $f$ is $I_{\text{modest}}[f] = \sum\limits_{i=1}^{n} I_{i, \text{modest}}[f]$.
\end{definition}

As usual, using Pareval's equality we get an analytical formula for influences in terms of the orthogonal decomposition of $f$.
\begin{fact}\label{fact:modest_inf_formula}
  In the setting of Theorem~\ref{thm:nonembed_deg_must_be_small_rephrase_maximal_relaxed}, for a function
  $f\colon (\Sigma^n, \mu_x^{\otimes n})\to\mathbb{C}$ and a coordinate $i\in [n]$, we have
  \[
    I_{i, \text{modest}}[f] = 2\sum\limits_{\chi: \chi_{i}\in B_{{\sf modest}}}\card{\widehat{f}(\chi)}^2.
  \]
  Subsequently, the modest total influence of $f$ is
  $I_{\text{modest}}[f] = 2\sum\limits_{\chi}\effnon(\chi)\card{\widehat{f}(\chi)}^2$.
\end{fact}

\subsection{Embedding Homogenous functions and Effectively Homogenous functions}
Equipped with the notions of embedding degree and effective embedding degree, we may define homogeneity with respect to them.

\begin{definition}
  We say $f$ is completely embedding homogenous of degree $D$ if there are
  integers $\{D_a\}_{a\in\hat{H}}$ that sum up to $D$ such that for all
  $\chi\in(B_{{\sf embed}}\cup B_{{\sf non-embed}}\cup B_{{\sf modest}})^{\otimes n}$
  such that $\widehat{f}(\chi)\neq 0$, we have that ${\sf embeddeg}_a(\chi) = D_a$.

  We say that $f$ is completely embedding homogenous if it is completely embedding homogenous of degree $D$ for some $D$.
\end{definition}

\begin{definition}
  A function $f\colon \Sigma^n\to\mathbb{C}$ is called non-embedding homogenous of non-embedding degree $d$ if for all
  $\chi\in(B_{{\sf embed}}\cup B_{{\sf non-embed}}\cup B_{{\sf modest}})^{\otimes n}$ such that $\widehat{f}(\chi)\neq 0$ it holds that ${\sf nedeg}(\chi) = d$.

  We say $f$ is non-embedding homogenous if it is non-embedding homogenous of non-embedding degree $d$ for some $d$.
\end{definition}

\begin{definition}
   We say a function $f\colon \Sigma^n\to\mathbb{C}$ has effective non-embedding degree at least $d$
   if for all $\chi\in(B_{{\sf embed}}\cup B_{{\sf non-embed}}\cup B_{{\sf modest}})^{\otimes n}$ such that $\widehat{f}(\chi)\neq 0$ it holds that $\effnon(\chi) \geq d$.
\end{definition}


With the notions of homogenous functions, we may now formulate a version of Theorem~\ref{thm:nonembed_deg_must_be_small_rephrase_maximal_relaxed} for homogenous functions of
bounded $L_2$ norm. We first define the collections of these homogenous functions:
\begin{definition}
Let $n\in\mathbb{N}$.
\begin{enumerate}
  \item For $d, d'$ we define the class $\mathcal{F}'_{n,d,d'}$ to be the collection of functions $f\colon \Sigma^n\to\mathbb{C}$
  that are completely embedding homogenous, non-embedding homogenous of degree $d$ and have effective degree at least $d'$.
  \item We define the class $\mathcal{G}_{n}$ to be the class of functions $g\colon\Gamma^n\to\mathbb{C}$ that are completely embedding homogenous
  and non-embedding homogenous.
  \item We define the class $\mathcal{H}_{n}$ to be the class of all functions $h\colon \Phi^n\to\mathbb{C}$ that are completely embedding homogenous
  and non-embedding homogenous.
\end{enumerate}
\end{definition}
We can now define the parameter $\beta$:
\begin{definition}
  For integers $n\geq d\geq d'$, finite alphabets $\Sigma$, $\Gamma$, $\Phi$ and a distribution $\mu$ over $\Sigma\times \Gamma\times \Phi$
  as in Theorem~\ref{thm:nonembed_deg_must_be_small_rephrase_maximal_relaxed}, we define
  \[
        \beta_{n,d,d'}'[\mu] =
        \sup\limits_{\substack{f\in \mathcal{F}'_{n,d,d'}\\ g\in \mathcal{G}_{n}\\ h\in \mathcal{H}_{n}}}\frac{\card{\Expect{(x,y,z)\sim \mu^{\otimes n}}{f(x)g(y)h(z)}}}{\norm{f}_2\norm{g}_2\norm{h}_2}.
  \]
\end{definition}
When the distribution $\mu$ is clear from context, we often drop it from the notation and denote the parameter simply by $\beta_{n,d,d'}'$.
We are now ready to formulate the homogenous version of Theorem~\ref{thm:nonembed_deg_must_be_small_rephrase_maximal_relaxed}.
\begin{thm}\label{thm:nonembed_homogenous'}
  For all $\alpha>0$ and $m\in\mathbb{N}$ there are $\xi>0$ and $c>0$ such that the following holds. Let $\Sigma$, $\Gamma$ and $\Phi$ be alphabets of size at most
  $m$, and let $\mu$ be a distribution over $\Sigma\times\Gamma\times \Phi$ as in Theorem~\ref{thm:nonembed_deg_must_be_small_rephrase_maximal_relaxed}.
  Suppose that $d,d'\in\mathbb{N}$ satisfy that $d'\geq d^{1-\xi}$. Then
  \[
        \beta_{n,d,d'}'[\mu]\leq (1+c)^{-d'^{1-2\xi}}.
  \]
\end{thm}
We have the following claim asserting that Theorem~\ref{thm:nonembed_homogenous'} implies Theorem~\ref{thm:nonembed_deg_must_be_small_rephrase_maximal_relaxed}.
The proof uses ``soft-truncation'' and ``truncation'' type argument, and is deferred to Section~\ref{sec:truncate}.
\begin{claim}\label{claim:beta_thm_implies}
  Theorem~\ref{thm:nonembed_homogenous'} implies Theorem~\ref{thm:nonembed_deg_must_be_small_rephrase_maximal_relaxed}.
\end{claim}
\begin{proof}
  Deferred to Section~\ref{sec:truncate}.
\end{proof}

\subsection{Reformulating Theorem~\ref{thm:nonembed_homogenous'}: Functions that are Constant on Connected Components}
As for the proof of Theorem~\ref{thm:nonembed_homogenous'}, for technical reasons it will be more convenient for us to view the parameter $\beta_{n,d_1,d_1'}'$
in a different but equivalent way, and prove an analogous statement for it. Fix a distribution $\mu$ as in Theorem~\ref{thm:nonembed_homogenous'}; the distribution
$\mu_{y,z}$ is uniform and hence is very nice to work with, but the distribution $\mu_x$ may be more complicated, hence it will be more convenient for us to switch
to a statement that is only concerned with $\mu_{y,z}$. For that, we are going to use the fact that as in $\mu$ it holds that the value of $y,z$ implies $x$, there is
a natural identification between functions over $x$, and a certain class of functions over $(y,z)$.

Given a function $f\colon \Sigma^n\to\mathbb{C}$, we may define the function
$F \colon \Gamma^n\times \Phi^n\to\mathbb{C}$ by
\[
F(y,z) = f(x)
\]
where $x\in \Sigma^n$ is the unique point such that $(x_i,y_i,z_i)$ is in the support of $\mu$ for all $i$ (recall that this $x$ is unique as in $\mu$, it holds
that $y,z$ implies $x$). We view this transformation as a mapping $W\colon L_2(\Sigma^n;\mu_x^{\otimes n})\to L_2(\Gamma^n\times \Phi^n; \mu_{y,z}^{\otimes n})$.
A function in the image of $W$ is not an arbitrary function, and we refer to such functions as functions that are constant on connected components:
\begin{definition}\label{def:const_on_ccs}
  We say a function $F\colon \Gamma^n\times \Phi^n\to\mathbb{C}$ is constant on connected components if there is $f\colon \Sigma^n\to\mathbb{C}$
  such that $F = W f$.
\end{definition}
The reason for this terminology is that  we consider a graph $(V,E)$ whose vertex set is $V = \Gamma\times \Phi$,
and the vertices $(y,z)$ and $(y',z')$ are adjacent if there is an $x$ such that $(x,y,z)$ and $(x,y',z')$ are both in $\Gamma\times \Phi$;
then a function $F\colon \Gamma\times \Phi\to\mathbb{C}$ that is constant on connected components as per the above definition
precisely corresponds to a function that is constant on each connected component of this graph:
\begin{lemma}\label{lem:constant_on_ccs_equiv}
  The following conditions are equivalent for $F\colon \Gamma^n\times \Phi^n\to\mathbb{C}$:
  \begin{enumerate}
    \item $F$ is constant on connected components as per Definition~\ref{def:const_on_ccs}.
    \item $F$ is constant on each connected component of $(V,E)^{\otimes n}$.
    \item $F \in {\sf span}\left(\sett{W\chi}{\chi\in (B_{{\sf embed}}\cup B_{{\sf non-embed}}\cup B_{{\sf modest}})^{\otimes n}}\right)$.
  \end{enumerate}
\end{lemma}
\begin{proof}
  We show that the first item is equivalent to the second item and that the first item is equivalent to the third item.

  \paragraph{The first item is equivalent to the second item.}
  Given $F$ as the second item, for $x\in\Sigma^n$ we take $y\in\Gamma^n$ and $z\in\Phi^n$
  such that $(x_i,y_i,z_i)$ is in the support of $\mu$ for all $i$, and define $f(x) = F(y,z)$. We note that this
  is well defined, as for such potential $y,z$ the value of $F(y,z)$ is the same (as they are all in the same connected
  component of $(V,E)^{\otimes n}$. Thus, $F = W f$ and so $F$ is constant on connected components as per Definition~\ref{def:const_on_ccs}.

  Given $F$ as in the first item, we have that $F = Wf$ for some $f\colon \Sigma^n\to\mathbb{C}$. If $(y,z)$ and $(y',z')$ are adjacent in
  the graph $(V,E)^{\otimes n}$, then there is $x\in \Sigma^n$ such that $(x_i,y_i,z_i)$ and $(x_i,y_i',z_i')$ are in the support for $\mu$
  for all $i$, hence $F(y,z) = f(x)$ and $F(y',z') = f(x)$ by the definition of $W$. It follows that the value of $F$ is the same vertices
  that are adjacent in $(V,E)^{\otimes n}$, hence $F$ is constant on the connected components of $(V,E)^{\otimes n}$.

  \paragraph{The first item is equivalent to the third item.}
  Taking $F$ as in the first item, writing $F = W f$ and expanding $f$ according to the basis
  $(B_{{\sf embed}}\cup B_{{\sf non-embed}}\cup B_{{\sf modest}})^{\otimes n}$
  and using the linearity of $W$, the third item follows. In the reverse, given
  \[
  F = \sum\limits_{\chi} a_{\chi} W\chi,
  \]
  we can write $F = W f$ for $f = \sum\limits_{\chi} a_{\chi}\chi$.
\end{proof}

We note that the set $\sett{W\chi}{\chi\in B_{{\sf embed}}\cup B_{{\sf non-embed}}\cup B_{{\sf modest}}}$ is an orthonormal set
in $L_2(\Gamma\times\Phi,\mu_{y,z})$, as
\[
\inner{W\chi}{W \chi'}_{\mu_{y,z}} = \inner{\chi}{\chi'}_{\mu_x} = 1_{\chi = \chi'},
\]
hence $\sett{W\chi}{\chi\in B_{{\sf embed}}\cup B_{{\sf non-embed}}\cup B_{{\sf modest}}}^{\otimes n}$ is an orthonormal basis
to the space of functions $F\in L_2(\Gamma^n\times\Phi^n; \mu_{y,z}^{\otimes n})$ that are constant on connected components.
Thus, to translate Theorem~\ref{thm:nonembed_homogenous'} to the language of functions that are constant on connected components,
it remains to discuss the analog of degrees, non-embedding degrees and effective non-embedding degrees. These are all very natural
analogs of the notions we have already seen.

\begin{definition}
  A function $F\colon \Gamma^n\times \Phi^n\to\mathbb{C}$ which is constant
  on connected components is said to be completely embedding homogenous of degree $D$
  if there are integers $\{D_a\}_{a\in\hat{H}}$ summing to $D$ such that
  for all $\chi\in(B_{{\sf embed}}\cup B_{{\sf non-embed}}\cup B_{{\sf modest}})^{\otimes n}$
  for which $\widehat{F}(W\chi)\neq 0$ it holds that ${\sf embeddeg}_a(\chi) = D_a$.

  We say $F$ is completely embedding homogenous if it is completely embedding homogenous of degree $D$ for some $D$.
\end{definition}

\begin{definition}
  A function $F\colon \Gamma^n\times \Phi^n\to\mathbb{C}$ which is constant on connected components is called
  non-embedding homogenous of non-embedding degree $d$ if for all $\chi\in(B_{{\sf embed}}\cup B_{{\sf non-embed}}\cup B_{{\sf modest}})^{\otimes n}$
  such that $\widehat{F}(W\chi)\neq 0$ it holds that ${\sf nedeg}(\chi) = d$.
\end{definition}

\begin{definition}
   We say a function $F\colon \Gamma^n\times \Phi^n\to\mathbb{C}$ which is constant on connected components has effective non-embedding degree at least $d$
   if for all $\chi\in(B_{{\sf embed}}\cup B_{{\sf non-embed}}\cup B_{{\sf modest}})^{\otimes n}$ such that $\widehat{F}(W\chi)\neq 0$ it holds that $\effnon(\chi) \geq d$.
\end{definition}

Lastly, we define the analog of the class $\mathcal{F}'_{n,d,d'}$.
\begin{definition}
For $n\geq d\geq d'$ we define the class $\mathcal{F}_{n,d,d'}$ to be the collection of functions $F\colon \Gamma^n\times \Phi^n\to\mathbb{C}$
that are constant on connected components, are completely embedding homogenous, are non-embedding homogenous of degree $d$ and effective non-embedding degree at
least $d'$.
\end{definition}

With all of this, we now define the analog of the parameter $\beta_{n,d,d'}'[\mu]$ for functions that are constant on connected components.
\begin{definition}
  For integers $n\geq d\geq d'$, finite alphabets $\Sigma$, $\Gamma$, $\Phi$ and a distribution $\mu$ over
  $\Sigma\times\Gamma\times \Phi$ as in Theorem~\ref{thm:nonembed_deg_must_be_small_rephrase_maximal_relaxed}, we define
  \[
         \beta_{n,d,d'}[\mu] =
        \sup\limits_{\substack{F\in \mathcal{F}_{n,d,d'}\\ g\in \mathcal{G}_{n}\\ h\in \mathcal{H}_{n}}}
        \frac{\card{\Expect{(x,y,z)\sim \mu^{\otimes n}}{F(y,z)g(y)h(z)}}}{\norm{F}_2\norm{g}_2\norm{h}_2}.
  \]
\end{definition}
Often the distribution $\mu$ will be clear from context, in which case we may omit it from the notation and simply write $\beta_{n,d,d'}$.
The following result is an equivalent formulation of Theorem~\ref{thm:nonembed_homogenous'} in the language of functions that are constant on connected components;
its proof is given in Sections~\ref{sec:reduce_to_near_linear} and~\ref{sec:prove_near_lin}.
\begin{thm}\label{thm:nonembed_homogenous}
  For all $\alpha>0$ and $m\in\mathbb{N}$ there are $\xi>0$ and $c>0$ such that the following holds.
  Let $\Sigma$, $\Gamma$ and $\Phi$ be alphabets of size at most
  $m$, and let $\mu$ be a distribution over $\Sigma\times\Gamma\times \Phi$ as in Theorem~\ref{thm:nonembed_deg_must_be_small_rephrase_maximal_relaxed}.
  Suppose that $n\geq d\geq d'$ satisfy that $d'\geq d^{1-\xi}$. Then
  \[
        \beta_{n,d,d'}[\mu]\leq (1+c)^{-d'^{1-2\xi}}.
  \]
\end{thm}
\subsubsection{Theorem~\ref{thm:nonembed_homogenous} implies Theorem~\ref{thm:nonembed_homogenous'}}\label{sec:nonembed_hom_implies}
We first show that Theorem~\ref{thm:nonembed_homogenous} implies Theorem~\ref{thm:nonembed_homogenous'}.\footnote{The reverse direction is also true and the proof proceeds via similar lines, but is omitted since it is unnecessary for us.}
Let $f,g,h$ be as in the definition of $\beta_{n,d,d'}[\mu]$. Taking $F = W f$, we see that $\norm{F}_2 = \norm{f}_2$ and
\[
\Expect{(x,y,z)\sim \mu^{\otimes n}}{f(x)g(y)h(z)} = \Expect{(x,y,z)\sim \mu^{\otimes n}}{F(y,z)g(y)h(z)}.
\]
Moreover, by Lemma~\ref{lem:constant_on_ccs_equiv} and the definitions it follows that $F$ is completely embedding homogenous,
non-embedding homogenous of degree $d$ and effective non-embedding of degree at least $d'$. It follows that the supremum defining
$\beta_{n,d,d'}'[\mu]$ is at most
the supremum defining $\beta_{n,d,d'}[\mu]$, hence by Theorem~\ref{thm:nonembed_homogenous} we get that
$\beta_{n,d,d'}'[\mu]\leq \beta_{n,d,d'}[\mu] \leq (1+c)^{-d'^{1-2\xi}}$.\qedhere

\subsection{Influences for Functions Constant on Connected Components}
As the map $W$ is a $1$-to-$1$ correspondence between functions over $x$, and functions over $y,z$ that are constant on connected components,
we get analogs of the various notions of influences (non-embedding and modest) for functions that are constant on connected components.
Below we define them formally.
\begin{definition}
  In the setting of Theorem~\ref{thm:nonembed_deg_must_be_small_rephrase_maximal_relaxed}, let $F\colon \Gamma^n\times\Phi^n\to\mathbb{C}$
  be a function which is constant on connected components, and let $f\colon \Sigma^n\to\mathbb{C}$ be the unique function such that
  $F = Wf$. For a coordinate $i\in [n]$:
  \begin{enumerate}
    \item The non-embedding influence of $i$ on $F$, denoted by $I_{i,\text{non-embed}}[F]$, is defined to be $I_{i, \text{non-embed}}[f]$.
    The total non-embedding influence of $F$ is $I_{\text{non-embed}}[F] = \sum\limits_{i=1}^{n}I_{i,\text{non-embed}}[F]$.
    \item The modest influence of $i$ on $F$, denoted by $I_{i,\text{modest}}[F]$, is defined to be $I_{i, \text{modest}}[f]$.
    The total modest influence of $F$ is $I_{\text{modest}}[F] = \sum\limits_{i=1}^{n}I_{i, \text{modest}}[F]$.
  \end{enumerate}
\end{definition}
\section{Proof of Theorem~\ref{thm:nonembed_homogenous}: Reducing to Near Linear Degrees}\label{sec:reduce_to_near_linear}
In this section we begin the proof of Theorem~\ref{thm:nonembed_homogenous} and show the following lemma.
\begin{lemma}\label{lem:reduce_beta}
  For all $m\in\mathbb{N}$ and $\alpha>0$ there are $\eps>0$ and $L>0$ such that the following holds. Let $n\geq d\geq d'$ be integers
  and let $\mu$ be as in Theorem~\ref{thm:nonembed_homogenous}. If $n\geq L d$, then
  \[
  \beta_{n,d,d'}\leq \max\left(\beta_{n-1,d,d'}, (1-\eps)\beta_{n-1,d-1,d'-1}\right).
  \]
\end{lemma}
In words, Lemma~\ref{lem:reduce_beta} says that when we look at $\beta$, we can either decrease $n$ and leave the non-embedding degree $d$,
and the lower bound on the effective degree $d'$ to be the same and get a bound by a similar parameter $\beta$;
else we drop $d$ and $d'$ by $1$ but to compensate for that we gain a factor of $(1-\eps)$.

After proving
Lemma~\ref{lem:reduce_beta} we are going to iterate it and conclude that either we are done (by gaining sufficiently many factors of $1-\eps$), or
else we have reduced $n$ all the way down to $L d$, in which case the number of variables is linear in the non-embedding
degree of the $x$-function and thus nearly linear in the effective non-embedding degree.
The case that $n\approx L d$ will be the subject of discussion in Section~\ref{sec:prove_near_lin}.

\skipi
The rest of this section is devoted to the proof of Lemma~\ref{lem:reduce_beta}, and our argument closely follows~\cite[Section 5]{BKMcsp2}, with a
few additional complication due to the greater generality in our case. In a sense, one could think of the argument in~\cite[Section 5]{BKMcsp2} as
addressing the case that there is only the trivial embedding functions corresponding to the all $1$ character, hence the constant all $1$ function plays
a special role therein. In our context embedding functions play the role of the constant functions in the
setting of~\cite{BKMcsp2}, and non-embedding functions play the role of functions with average $0$ in the setting of~\cite{BKMcsp2}.
The presence of multiple Abelian embeddings however makes the argument here a bit more tricky, and this is ultimately the reason we needed
the notion of completely embedding homogenous functions.

\subsection{Preliminaries: the Additive Base Case and Some Simple Orthogonalities}
In our argument we are going to need to upper bound the absolute value of expectations of the form
\[
\Expect{(x,y,z)\sim \mu}{f(x)(g(y)+h(z))}
\]
for univariate functions. We may of course apply Cauchy-Schwarz to bound this, but importantly
the following claim shows that if $f$ is perpendicular to all embedding functions,
then we can gain a factor of $(1-c)$ over the trivial bound given by Cauchy-Schwarz.
\begin{claim}\label{claim:additive_base_case}
  Let $\Sigma$, $\Gamma$ and $\Phi$ be alphabets of size at most $m$, and let
  $\mu$ be a distribution as in Theorem~\ref{thm:nonembed_deg_must_be_small_rephrase_maximal_relaxed}.
  Then there exists $c = c(m,\alpha)>0$ such that for all functions $f\colon\Sigma\to\mathbb{C}$, $g\colon \Gamma\to\mathbb{C}$ and
  $h\colon \Phi\to\mathbb{C}$ such that $f\in {\sf Embed}_{\sigma}(\mu)^{\perp}$
  we have that
  \[
    \card{\Expect{(x,y,z)\sim \mu}{f(x)(g(y) + h(z))}}\leq (1-c)\norm{f}_2\norm{g+h}_2.
  \]
\end{claim}
\begin{proof}
  Assume this is not the case. Thus, we may find a sequence of functions $(f_m,g_m,h_m)$ and measures $\mu_m$
  such that $\norm{f_m}_2 = 1$, $\norm{g_m+h_m}_2 = 1$, $f_m\in {\sf Embed}_{\sigma}(\mu_m)^{\perp}$, and
  \[
  \card{\Expect{(x,y,z)\sim \mu_m}{f_m(x)(g_m(y) + h_m(z))}}\geq 1-\frac{1}{m}.
  \]
  Passing to subsequences, we may assume that ${\sf supp}(\mu_m)$ is the same for all $m$,
  and passing to further subsequences we may assume that $f_m$ converges to a function $f$,
  $g_m$ converges to a function $g$ and $h_m$ converges to a function $h$ and $\mu_m$ converges to a measure $\mu$.
  We get that $f\in {\sf Embed}_{\sigma}(\mu)^{\perp}$,
  $\norm{f}_2 = \norm{g+h}_2 = 1$ and
  \[
  \card{\Expect{(x,y,z)\sim \mu}{f(x)(g(y) + h(z))}}\geq 1.
  \]
  On the other hand, by Cauchy-Schwarz we get that
  \[
  \card{\Expect{(x,y,z)\sim \mu}{f(x)(g(y) + h(z))}}\leq \norm{f}_2\norm{g+h}_2=1,
  \]
  hence Cauchy-Schwarz is tight. Therefore there is $\theta\in \mathbb{C}$ of
  absolute value $1$ such that $f(x) = \theta(\overline{g(y) + h(z)})$ for all $(x,y,z)\in \mathbb{C}$.
  Claim~\ref{claim:embedding_fns} implies now that $f\in {\sf Embed}_{\sigma}(\mu)$, and this
  means that $f$ must be identically $0$ in contradiction to the fact that $\norm{f}_2=1$.
\end{proof}

Next, we have the following claim asserts that embedding functions in $x$ are perpendicular to
non-embedding functions in either $y$ or $z$.
\begin{claim}\label{claim:ortho_of_embed_nonembed}
  Let $\Sigma$, $\Gamma$ and $\Phi$ be alphabets of size at most $m$, and let
  $\mu$ be a distribution as in Theorem~\ref{thm:nonembed_deg_must_be_small_rephrase_maximal_relaxed}.
  \begin{enumerate}
    \item If $f\in {\sf Embed}_{\sigma}(\mu)$, $g\in {\sf Embed}_{\gamma}(\mu)^{\perp}$
    and $h$ is any function, then $\Expect{(x,y,z)\sim \mu}{f(x)g(y)h(z)} = 0$.
    \item If $f\in {\sf Embed}_{\sigma}(\mu)$ and $h\in {\sf Embed}_{\phi}(\mu)^{\perp}$
    and $g$ is any function, then $\Expect{(x,y,z)\sim \mu}{f(x)g(y)h(z)} = 0$.
    \item If $f\in {\sf Embed}_{\sigma}(\mu)$, $g\in {\sf Embed}_{\gamma}(\mu)$
    and $h\in {\sf Embed}_{\phi}(\mu)$ each have average $0$, then
    \[
        \Expect{(x,y,z)\sim \mu}{f(x)g(y)} = \Expect{(x,y,z)\sim \mu}{f(x)h(z)}
        =\Expect{(x,y,z)\sim \mu}{g(y)h(z)} = 0.
    \]
  \end{enumerate}
\end{claim}
\begin{proof}
  For the first bullet, it suffices to prove this for any $f$ which is a basis element of ${\sf Embed}_{\sigma}(\mu)$.
  Thus, let $\chi\in \widehat{H}$ and consider $f(x) = \chi(\sigma(x))$. Then $f(x) = \chi(-\gamma(y) - \phi(z))$, so
  \begin{align*}
  \Expect{(x,y,z)\sim \mu}{f(x)g(y)h(z)}
  &=
  \Expect{(x,y,z)\sim \mu}{\chi(-\gamma(y))g(y)\chi(-\phi(z))h(z)}\\
  &=
  \Expect{(x,y,z)\sim \mu}{\chi(-\gamma(y))g(y)}
  \Expect{(x,y,z)\sim \mu}{\chi(-\phi(z))h(z)},
  \end{align*}
  where in the last transition we used the fact that $\mu_{y,z}$ is uniform.
  As $g$ is perpendicular to all embedding functions in $y$ we get that
  $\Expect{(x,y,z)\sim \mu}{\chi(-\gamma(y))g(y)} = 0$.

  The second bullet is identical.

  For the third bullet, we argue that $\Expect{(x,y,z)\sim \mu}{f(x)g(y)} = 0$ and the other two are identical.
  Note that $f$ and $g$ are functions only of $\sigma(x)$ and $\gamma(y)$. As $\sigma(x)$ and $\gamma(y)$ are independent
  under $(x,y,z)\sim \mu$, we get that the values $f(x)$ and $g(y)$ are independent and so
  $\Expect{(x,y,z)\sim \mu}{f(x)g(y)} = \Expect{(x,y,z)\sim \mu}{f(x)}\Expect{(x,y,z)\sim \mu}{g(y)} = 0$.
\end{proof}

\subsection{The Singular-Value Decomposition}
Next, our argument requires an appropriate singular value decomposition with respect to our notions of embedding degrees and non-embedding degrees,
and we have the following lemma. Let $I,J$ be a partition of $[n]$ where $\card{I} = n-1$ and $\card{J} = 1$.
    \begin{claim}\label{claim:svd}
        Let $\mu$ be as in Theorem~\ref{thm:nonembed_homogenous}.
        If $g\colon\Gamma^n\to\mathbb{C}$ is a completely embedding homogenous, non-embedding homogenous and $\norm{g}_2=1$, then we may write
		\[
		g(y) = \sum\limits_{\chi\in \hat{H}} \kappa_{\chi} g_{\chi}(y_I)g_{\chi}'(y_J) +
                \sum\limits_{r=1}^{s} \kappa_r g_r(y_I) g_r'(y_J),
		\]
        where letting $P = \sett{1\leq r\leq s}{\kappa_r\neq 0}\cup \sett{\chi\in\hat{H}}{\kappa_{\chi}\neq 0}$, we have
        $\kappa_r\geq 0$ for $r\in P$ and:
		\begin{enumerate}
            \item For $\chi\in\hat{H}$, $g_{\chi}' \in {\sf span}(\{\chi\circ \gamma\})$ and $g_{\chi}$ is completely homogenous embedding function
            and non-embedding homogenous of degree $d$.
			\item For $1\leq r\leq s$, $g_r$ is completely homogenous embedding function and non-embedding homogenous of degree $d-1$.
            Also, $g_r'$ is in ${\sf Embed}_{\gamma}(\mu)^{\perp}$.
            \item The set $\set{g_r\colon \Gamma^{I}\to\mathbb{C}}_{r\in P}$ is orthonormal.
            \item The set $\set{g_r'\colon \Gamma^{J}\to\mathbb{C}}_{r\in P}$ is orthonormal.
			\item $\sum\limits_{r\in P} \kappa_r^2 = 1$.
		\end{enumerate}
	\end{claim}
	\begin{proof}
		The proof is deferred to Section~\ref{sec:svd_apx}.
	\end{proof}

	We next state an SVD decomposition statement that addresses the function $F$. The idea is the same as in Claim~\ref{claim:svd}, and
    the only difference is that we have an additional notion of effective non-embedding degree.
	\begin{claim}\label{claim:svd_effective}
		Suppose $F\colon\Gamma^n\times \Phi^n\to\mathbb{C}$ is a function which is constant on connected components,
        is completely embedding homogenous, is non-embedding homogenous of degree $d$, and has effective degree at least $d'$. If $\norm{F}_2 = 1$,
        then we may write
		\[
		F(y,z) =
        \sum\limits_{\chi\in\hat{H}} \psi_{\chi} F_{\chi}(y_I,z_I) F_{\chi}'(y_J,z_J)
        +\sum\limits_{t=1}^{s} \psi_t F_t(y_I,z_I) F_t'(y_J,z_J),
		\]
		and $P = \sett{1\leq t\leq s}{\psi_t\neq 0}\cup\sett{\chi\in \hat{H}}{\psi_{\chi}\neq 0}$, where $\psi_t\geq 0$ for $t\in P$ as well as:
		\begin{enumerate}
            \item For $\chi\in\hat{H}$, $F_{\chi}' \in {\sf span}(\{W\circ \chi\circ \sigma\})$
            and $F_{\chi}\colon \Gamma^{I}\times \Phi^{I}\to\mathbb{C}$ is a completely embedding homogenous,
            non-embedding homogenous of degree $d$ and effective non-embedding degree at least $d'$.
			\item For $1\leq t\leq s$, $F_t$ is completely embedding homogenous,
            non-embedding homogenous of degree $d-1$ and the effective non-embedding degree is at least $d'-1$.
            Also, $F_t'$ is orthogonal to all embedding functions.
			\item The functions $F_t$ and $F_t'$ are constant on connected components for all $t$.
            \item The set $\set{F_t\colon \Gamma^{I}\times\Phi^{I}\to\mathbb{C}}_{t\in P}$ is orthonormal.
			\item The set $\set{F_t'\colon \Gamma^{J}\times \Phi^J\to\mathbb{C}}_{t\in P}$ is orthonormal.
			\item $\sum\limits_{t\in P}\psi_t^2 = 1$.
		\end{enumerate}
	\end{claim}
	\begin{proof}
        The proof is deferred to Section~\ref{sec:svd_apx}.
	\end{proof}

    We will need the following claim, stating a connection between the coefficients $\psi_t$ in the SVD decomposition and our notions of
    non-embedding influence.
    \begin{claim}\label{claim:coefs_svd_effective}
		Let $F\colon\Gamma^n\times \Phi^n\to\mathbb{C}$ be a function which is constant on connected components.
        Suppose that $F$ is completely embedding homogenous, non-embedding homogenous of degree $d$, the effective degree is at least
        $d'$, and $\norm{F}_2=1$. Write
		\[
		F(y,z) =
        \sum\limits_{\chi\in\hat{H}} \psi_{\chi} F_{\chi}(y_I,z_I) F_{\chi}'(y_J,z_J)
        +\sum\limits_{t\in T} \psi_t F_t(y_I,z_I) F_t'(y_J,z_J).
		\]
		as in Claim~\ref{claim:svd_effective}. Then if $j$ is the unique variable in the set $J$ in the partition $[n] = I\cup J$, then
		\[
		\sum\limits_{t\in T}\psi_t^2 = \frac{1}{2}I_{j,\text{non-embed}}[F].
		\]
	\end{claim}
	\begin{proof}
        Write $F = Wf$ for some $f\colon\Sigma^n\to\mathbb{C}$ so that $I_{j, \text{non-embed}}[F] = I_{j, \text{non-embed}}[f]$,
        and note that from the decomposition of $F$ we get an analogous decomposition of $f$ where $F_t = Wf_t$, $F_t' = Wf_t$ for all $t$.

        Consider $I_{j, \text{non-embed}}[F]$; sampling $x\sim \mu_x^{n-1}$ and $a,b\sim \mu_x$ conditioned on $\sigma(a) = \sigma(b)$
        we get that
        \[
			I_{j, \text{non-embed}}[f]
			= \sum\limits_{t\in T\cup \hat{H}} \kappa_{t}^2\cExpect{a, b\sim\mu_x}{\sigma(a) = \sigma(b)}{(f_{t}'(a) - g_{t}'(b))\overline{(f_{t}'(a) - f_{t}'(b))}}.
		\]
        For $t\in \hat{H}$ we have that $f_t'(a) = f_t'(b)$ and so the expectation is $0$.
        For $t\in T$, computing the expectation by expanding it we get it is equal to
        \[
        2\norm{f_t'}_2^2
        -
        2\cExpect{a, b\sim \mu}{\sigma(a) = \sigma(b)}{f_{t}'(a)\overline{f_{t}'(b)}}
        =2-2\cExpect{a, b\sim \mu}{\sigma(a) = \sigma(b)}{f_{t}'(a)\overline{f_{t}'(b)}},
        \]
        and we argue that last expectation is $0$. Indeed, let $\mathcal{D}$ be the distribution over the group $H$ of $\sigma(a)$ where we sample $a\sim \mu_x$,
        we have that the expectation is equal to
        \[
        \Expect{w\sim \mathcal{D}}{\card{\cExpect{a}{\sigma(a) = w}{f_t'(a)}}^2}
        =\Expect{w\sim \mathcal{D}}{\frac{\card{\inner{f_t'}{1_{\sigma(\cdot) = w}}}^2}{\card{\Expect{a}{1_{\sigma(a) = w}}}^2}}
        =\Expect{w\sim \mathcal{D}}{\frac{0}{\card{\Expect{a}{1_{\sigma(a) = w}}}^2}}
        =0,
        \]
        where we used the fact that $f_t'$ is orthogonal to the embedding function $1_{\sigma(\cdot) = w}$.
	\end{proof}

\subsection{Proof of Lemma~\ref{lem:reduce_beta}}
    In this section, we give the formal proof of Lemma~\ref{lem:reduce_beta}. Let $F, g$ and $h$ be functions of $2$-norm equal to $1$
    achieving the value $\beta_{n,d,d'}$. Using
	the singular-value decomposition we may write $F$, $g$ and $h$ as a sum of functions satisfying some orthogonality
	properties (see Claim~\ref{claim:svd} and Claim~\ref{claim:svd_effective} for precise statements):

	\begin{align}\label{eq:svd_main}
		&F(y,z) =
        \sum\limits_{\chi\in \hat{H}} \psi_{\chi} F_{\chi}(y_I,z_I) F_{\chi}'(y_J,z_J)
        +\sum\limits_{t\in T} \psi_t F_t(y_I,z_I) F_t'(y_J,z_J),\notag\\
		&g(\y) = \sum\limits_{\pi\in \hat{H}}\kappa_{\pi} g_{\pi}(y_I)g_{\pi}'(y_J) +
                \sum\limits_{r\in R} \kappa_r g_r(y_I) g_r'(y_J), \notag\\
		&h(z) = \sum\limits_{\lambda\in\hat{H}}\rho_{\lambda} h_{\lambda}(z_I) h_{\lambda}'(z_I)
                + \sum\limits_{s\in S} \rho_s h_s(y_I) h_s'(y_J).
	\end{align}
	Here, each one of the sets $\{F_t\}_{t\in T}\cup\{F_{\chi}\}_{\chi\in\hat{H}}$,
    $\{F_t'\}_{t\in T}\cup\{F_{\chi}'\}_{\chi\in\hat{H}}$,
    $\{g_r\}_{r\in R}\cup \{g_{\pi}\}_{\pi\in\hat{H}}$,
    $\{g_r'\}_{r\in R}\cup \{g_{\pi}'\}_{\pi\in\hat{H}}$,
    $\{h_s\}_{s\in S}\cup\{h_{\lambda}\}_{\lambda\in\hat{H}}$
    and
    $\{h_s'\}_{s\in S}\cup\{h_{\lambda}'\}_{\lambda\in\hat{H}}$ is orthonormal.
    By multiplying each one of the coefficients $\psi_t, \kappa_r$ and $\rho_s$ by an appropriate complex number of absolute value $1$,
    we will assume henceforth that $F_{\chi}' = W\circ \chi\circ \sigma$,
    $g_{\pi}' = \pi\circ \gamma$
    and $h_{\lambda}' = \lambda\circ \phi$.
    Thus, we have that
    \begin{align*}
	\Expect{(x,y,z)\sim\mu^{\otimes n}}{F(y,z)g(y)h(z)}
	=
	\sum\limits_{\substack{r\in R\cup\hat{H}\\ s\in S\cup\hat{H} \\ t\in T\cup\hat{H}}}\psi_t\kappa_r\rho_s
	\hspace{-2ex}\Expect{(x,y,z)\sim\mu^{\otimes n}}{F_t(y,z)g_r(y)h_s(z)}
	\hspace{-2ex}\Expect{(x,y,z)\sim\mu^{\otimes n}}{F_t'(y,z)g_r'(y)h_s'(z)}.
	\end{align*}
	It will be convenient for us to denote
    \[
    \widehat{F_t}(r,s) = \Expect{(x,y,z)\sim\mu^{\otimes (n-1)}}{F_t(y,z)g_r(y)h_s(z)},
    \qquad\qquad
	\widehat{F_t'}(r,s) = \Expect{(x,y,z)\sim\mu}{F_t'(y,z)g_r'(y)h_s'(z)}.
    \]
    This is justified because $(\overline{g_r} \overline{h_s})_{r\in R, s\in S}$
	form an orthonormal set in $L_2(y_I, z_I;\mu_{y,z}^{\otimes I})$, and it can be completed to an orthonormal basis,
	in which case the coefficient $\widehat{F_t}(r,s)$ appear in front of $\overline{g_r}\overline{h_s}$ in the representation of $F_t$.
	Thus, we get that
	\begin{equation}\label{eq:project_eq}
		\Expect{(x,y,z)\sim\mu^{\otimes n}}{F(y,z)g(y)h(z)}
		=
		\sum\limits_{\substack{r\in R\cup\hat{H}\\ s\in S\cup\hat{H} \\ t\in T\cup\hat{H}}}\psi_t\kappa_r\rho_s
		\widehat{F_t}(r,s)\widehat{F_t'}(r,s),
	\end{equation}

	\subsubsection{The parameters and choosing the partition}
	We will use several parameters throughout this section, obeying the following relations:
	\begin{equation}\label{eq:hier}
		0 \ll \eps\ll \eta\ll \tau \ll c\ll m^{-1},\alpha\leq 1.
	\end{equation}

    We need to choose the partition $I,J$ so that the mass of $F$ on the non-embedding components is small, namely
    \begin{equation}\label{eq:small_nonembed_inf}
    \sum\limits_{t \in T}\card{\psi_t}^2 \leq \eta,
    \end{equation}
    and next show that as long as $n$ is much larger than $d$ this is possible. Indeed,
    if $n\geq \frac{10}{\eta}d$, then choosing the partition randomly we have by
    Claim~\ref{claim:coefs_svd_effective} and Fact~\ref{fact:nonembed_inf_formula} that
    \[
    \Expect{I,J}{\sum\limits_{t\in T}\card{\psi_t}^2} =
    \frac{1}{2}\Expect{J=\{j\}}{I_{j,\text{non-embed}}[F]}\leq \frac{d}{n}\leq \frac{\eta}{10},
    \]
	so by Markov's inequality we get that
	\[
	\Prob{I,J}{\eqref{eq:small_nonembed_inf}\text{ fails}}\leq \frac{1}{10} < 1.
	\]
	We may thus find a partition $[n]=I\cup J$ with $\card{J} = 1$ such that~\eqref{eq:small_nonembed_inf} holds.
    We fix this partition henceforth.
    	
	\subsubsection{The main inductive argument}
	We are going to need to split the expression from~\eqref{eq:project_eq} into various sums and towards this end we start with a few observations
    that asserts that various terms in~\eqref{eq:project_eq} are $0$.
    \begin{observation}
    With the setup above:
	\begin{enumerate}
        \item Consider $\chi,\pi,\lambda\in \hat{H}$, and note that $\widehat{F}_{\chi}'(\pi,\lambda) = 1$ if $\chi = \pi = \lambda$,
        and $\widehat{F}_{\chi}'(\pi,\lambda) = 0$ otherwise. Indeed,
        \[
        \widehat{F}_{\chi}'(\pi,\lambda)
        =\Expect{(x,y,z)\sim \mu}{F_{\chi}'(y,z)g_{\pi}'(y)h_{\phi}'(z)}
        =\Expect{(x,y,z)\sim \mu}{\chi(\sigma(x))\pi(\gamma(y))\lambda(\phi(z))}.
        \]
        As $(\sigma(x),\gamma(y),\phi(z))$ is distributed uniformly among all $(a,b,c)\in H^3$ such that $a+b+c = 0$,
        we get that expectation is $0$ if not all of $\chi,\pi,\lambda$ are equal, and $1$ if they are equal.
        \item For $\chi\in \hat{H}$ and $r,s$ such that either $r\in R$ or $s\in S$ we have that $\widehat{F}_{\chi}'(r,s) = 0$.
        This is thanks to Claim~\ref{claim:ortho_of_embed_nonembed}.
        \item For $t\in T$, we have that $\widehat{F}_t'(\pi,\lambda) = 0$ if $\pi = \lambda$. Indeed, write $F_t' = W\circ f_t'$, then
        \begin{align*}
        \widehat{F}_{\chi}'(\pi,\lambda)
        =\Expect{(x,y,z)\sim \mu}{f_t'(x)\pi(\gamma(y))\lambda(\phi(z))}
        &=\Expect{(x,y,z)\sim \mu}{f_t'(x)\pi(\gamma(y)+\phi(z))}\\
        &=\Expect{(x,y,z)\sim \mu}{f_t'(x)\pi(-\sigma(x))}\\
        &=\inner{f_t'}{\pi\circ \sigma},
        \end{align*}
        which is $0$ as $f_t'$ is orthogonal to embedding functions.
	\end{enumerate}
\end{observation}
With these observations in hand, we expand out the sum in~\eqref{eq:project_eq} and drop off terms which are $0$ to get that
	\begin{align}\label{eq:large_large_0}
		\beta_{n,d,d'}
		&=
        \sum\limits_{\chi\in\hat{H}}\psi_{\chi}\kappa_{\chi}\rho_{\chi} \widehat{F}_{\chi}(\chi,\chi)
        +
        \sum\limits_{t\in T, \pi,\lambda\in\hat{H}, \pi\neq \lambda}
        \psi_t\kappa_{\pi}\rho_{\lambda} \widehat{F}_{t}(\pi,\lambda)\widehat{F}_{t}'(\pi,\lambda)\notag\\
        &+
        \sum\limits_{t\in T, r\in R,\lambda\in\hat{H}}
        \psi_t\kappa_r\rho_{\lambda} \widehat{F}_{t}(r,\lambda)\widehat{F}_{t}'(r,\lambda)
        +\sum\limits_{t\in T, \pi\in\hat{H}, s\in S}
        \psi_t\kappa_{\pi}\rho_s \widehat{F}_{t}(\pi,s)\widehat{F}_{t}'(\pi,s)\notag\\
        &+
        \sum\limits_{t\in T, r\in R, s\in S}
        \psi_t\kappa_r\rho_s \widehat{F}_{t}(r,s)\widehat{F}_{t}'(r,s).
	\end{align}

    \subsubsection{The Embedding Masses of $F$, $g$ and $h$ Must Align}
    We start off with handling the case where the mass of $F$, $g$ and $h$ is mis-aligned on the embedding part of the decomposition.
    By the choice of $I,J$ we know that most of the mass of $F$ lies on components corresponding to
    embedding functions, and as we noted above among these only the diagonal terms $\chi = \pi = \lambda$ survive. Hence, it makes sense that unless there
    is a single $\chi$ on which most of the embedding mass of $F$, $g$ and $h$ lies, we will already be able to use~\eqref{eq:large_large_0} to argue that
    $\beta_{n,d,d'}$ must be considerably smaller than
    $\beta_{n-1,d-1,d'-1}$. This is the content of the following lemma.
    \begin{lemma}\label{lem:diagonal_small_beta}
        If $\sum\limits_{\chi\in\hat{H}}\card{\psi_{\chi}\kappa_{\chi}\rho_{\chi}}\leq 1-\tau$, then
        $\beta_{n,d,d'}\leq (1-\eps)\max(\beta_{n-1,d-1,d'-1},\beta_{n-1,d,d'})$ (and hence we are done).
    \end{lemma}
    \begin{proof}
        Using~\eqref{eq:large_large_0} we get
        \[
        \beta_{n,d,d'}
        \leq
        (1-\tau) \beta_{n-1,d,d'}
        +O_{m}\left(\sqrt{\eta} \beta_{n-1,d-1,d'-1}\right),
        \]
        where we used the fact that for $t\in T$ we have that $\card{\psi_t}\leq \sqrt{\sum\limits_{t\in T}\card{\psi_t}^2}\leq \sqrt{\eta}$,
        as well as
        \[
        \card{\widehat{F}_{\chi}(\chi,\chi)}\leq \beta_{n-1,d,d'},
        \qquad \card{\widehat{F}_{t}(r,s)}\leq \beta_{n-1,d-1,d'-1}
        \]
        for all $t\in T$, $r\in R\cup \hat{H}$, $s\in S\cup \hat{H}$.
        We get that
        \[
            \frac{\beta_{n,d,d'}}{\max(\beta_{n-1,d-1,d'-1},\beta_{n-1,d,d'})}
            \leq 1-\tau + O_m(\sqrt{\eta})
            \leq 1-\tau/2
            \leq 1-\eps,
        \]
        and we are done.
    \end{proof}

    By Lemma~\ref{lem:diagonal_small_beta}, we assume henceforth that
    $\sum\limits_{\chi\in\hat{H}}\card{\psi_{\chi}\kappa_{\chi}\rho_{\chi}}> 1-\tau$, and we note that
    \begin{align*}
    \sum\limits_{\chi\in\hat{H}}\card{\psi_{\chi}\kappa_{\chi}\rho_{\chi}}
    &\leq
    \max_{\chi}\card{\psi_{\chi}\kappa_{\chi}\rho_{\chi}}^{1/3}\sum\limits_{\chi\in\hat{H}}\card{\psi_{\chi}\kappa_{\chi}\rho_{\chi}}^{2/3}\\
    &\leq
    \max_{\chi}\card{\psi_{\chi}\kappa_{\chi}\rho_{\chi}}^{1/3}
    \left(\sum\limits_{\chi\in\hat{H}}\card{\psi_{\chi}}^2\right)^{1/3}
    \left(\sum\limits_{\chi\in\hat{H}}\card{\kappa_{\chi}}^2\right)^{1/3}
    \left(\sum\limits_{\chi\in\hat{H}}\card{\rho_{\chi}}^2\right)^{1/3}\\
    &\leq \max_{\chi}\card{\psi_{\chi}\kappa_{\chi}\rho_{\chi}}^{1/3},
    \end{align*}
    where we used H\"{o}lder's inequality. It follows that there is $\chi$ such that
    $\card{\psi_{\chi}\kappa_{\chi}\rho_{\chi}}^{1/3}\geq 1-\tau$, and we call it $\chi^{\star}$.
    In particular it follows that
    \begin{equation}\label{eq:beta_arg_most_mass_embed}
      \card{\psi_{\chi^{\star}}},\card{\kappa_{\chi^{\star}}}, \card{\rho_{\chi^{\star}}}\geq 1-3\tau.
    \end{equation}

    \subsubsection{Upper Bounding Terms in~\eqref{eq:large_large_0}}
    We now proceed to upper bounding terms in~\eqref{eq:large_large_0}. The first sum is handled by the following lemma:
    \begin{lemma}\label{lem:bound_main_term_beta}
      It holds that:
        \[
        \card{\sum\limits_{\chi\in\hat{H}}\psi_{\chi}\kappa_{\chi}\rho_{\chi} \widehat{F}_{\chi}(\chi,\chi)}\leq
        \beta_{n-1,d,d'}\sqrt{\sum\limits_{\chi\in\hat{H}}\card{\psi_{\chi}}^2}\sqrt{\sum\limits_{\chi\in\hat{H}}\card{\kappa_{\chi}}^2\card{\rho_{\chi}}^2}.
        \]
    \end{lemma}
    \begin{proof}
        By definition$\card{\widehat{F}_{\chi}(\chi,\chi)}\leq \beta_{n-1,d,d'}$ for all $\chi\in\hat{H}$,
        and the result follows from the triangle inequality and Cauchy-Schwarz.
    \end{proof}

    The next lemma handles summands from the second, third and fourth sums in which $\chi^{\star}$ does not appear.
    \begin{lemma}\label{lem:bound_small_beta}
      It holds that:

        \[
        \card{\sum\limits_{\substack{t\in T, r\in R\cup\hat{H}, s\in S\cup\hat{H}\\ r,s\neq \chi^{\star}}} \psi_t\kappa_r\rho_s \widehat{F}_{t}(r,s)\widehat{F}_{t}'(r,s)}\lll_{m}
        \beta_{n-1,d-1,d'-1}\sqrt{\sum\limits_{t\in T} \card{\psi_t}^2}\sqrt{1-\card{\kappa_{\chi^{\star}}}^2}\sqrt{1-\card{\rho_{\chi^{\star}}}^2}.
        \]
    \end{lemma}
    \begin{proof}
        For all $t\in T$ it holds that $\card{\psi_t}\leq \sqrt{\sum\limits_{t\in T} \card{\psi_t}^2}$, for $r\neq \chi^{\star}$ we have
        $\card{\lambda_r}\leq \sqrt{\sum\limits_{r'\neq \chi^{\star}}\card{\kappa_{r'}}^2} = \sqrt{1-\card{\kappa_{\chi^{\star}}}^2}$,
        and similarly for $s\neq \chi^{\star}$ we have $\card{\rho_s}\leq \sqrt{1-\card{\rho_{\chi^{\star}}}^2}$. Also, for all such $t,r,s$
        we have that $\card{\widehat{F}_{t}(r,s)}\leq \beta_{n-1,d-1,d'-1}$ and $\card{\widehat{F}_{t}'(r,s)}\leq 1$, and as the number of summands
        is at most $m^3$ the bound follows.
    \end{proof}

    The following lemma handles the rest of the summands from the second, third and fourth sums, namely those in which $\chi^{\star}$ does appear
    (we recall that if both $r$ and $s$ are equal to $\chi^{\star}$, then such summands are $0$, hence we are dealing with summands in which exactly
    one of $r$ and $s$ is equal to $\chi^{\star}$).
    \begin{lemma}\label{lem:bound_additivebase_beta}
      It holds that:
        \begin{align*}
        &\card{\sum\limits_{\substack{t\in T\\ r\in R\cup\hat{H}, s\in S\cup\hat{H}\\ r = \chi^{\star}\text{ or } s = \chi^{\star}}}
        \psi_t\kappa_r\rho_s \widehat{F}_{t}(r,s)\widehat{F}_{t}'(r,s)}\\
        &\qquad\qquad\qquad\leq
        (1-c)\beta_{n-1,d-1,d'-1}\sqrt{\sum\limits_{t\in T} \card{\psi_t}^2}
        \sqrt{\card{\kappa_{\chi^{\star}}}^2\sum\limits_{s\neq \chi^{\star}}\card{\rho_s}^2
      +\card{\rho_{\chi^{\star}}}^2\sum\limits_{r\neq \chi^{\star}}\card{\kappa_r}^2}.
        \end{align*}
    \end{lemma}
    \begin{proof}
      We write the sum in discussion as
      \begin{equation}\label{eq:apply_additive_base_case}
      \sum\limits_{t\in T, r\neq \chi^{\star}}
        \psi_{t}\kappa_{r}\rho_{\chi^{\star}} \widehat{F}_{t}(r,\chi^{\star})\widehat{F}_{t}'(r,\chi^{\star})
      +
      \sum\limits_{t\in T, s\neq \chi^{\star}}
        \psi_{t}\kappa_{\chi^{\star}}\rho_{s} \widehat{F}_{t}(\chi^{\star},s)\widehat{F}_{t}'(\chi^{\star},s).
      \end{equation}
      Fix $t\in T$, write $F_t' = W f_t'$ and consider the corresponding terms. We note that
      \[
      \widehat{F}_{t}'(r,\chi^{\star})
      =\Expect{(x,y,z)\sim \mu}{f_t'(x) g_r'(y) \chi^{\star}(\phi(z))}
      =\Expect{(x,y,z)\sim \mu}{f_t'(x)\chi^{\star}(-\sigma(x)) g_r'(y) \chi^{\star}(-\gamma(y))}
      \]
      and similarly
      \[
      \widehat{F}_{t}'(\chi^{\star},s)
      =\Expect{(x,y,z)\sim \mu}{f_t'(x) \chi^{\star}(\phi(y))h_s'(z)}
      =\Expect{(x,y,z)\sim \mu}{f_t'(x)\chi^{\star}(-\sigma(x)) h_s'(z) \chi^{\star}(-\phi(z))}.
      \]
      Thus, we define
      \[
      \tilde{g}(y) = \sum\limits_{r\neq \chi^{\star}} \widehat{F}_t(r,\chi^{\star})\kappa_r\rho_{\chi^{\star}} g_r'(y) \chi^{\star}(-\gamma(y)),
      \qquad
      \tilde{h}(z) = \sum\limits_{s\neq \chi^{\star}} \widehat{F}_t(\chi^{\star}, s)\kappa_{\chi^{\star}}\rho_{s} h_s'(z) \chi^{\star}(-\phi(z)),
      \]
      and get that the contribution of $t$ to~\eqref{eq:apply_additive_base_case} is equal to
      \[
      \Expect{(x,y,z)\sim\mu}{f_t'(x)\chi^{\star}(-\sigma(x))(\tilde{g}(y) + \tilde{h}(z))}.
      \]
      We note that $\tilde{g}$, $\tilde{h}$ are orthogonal, and
      \[
      \norm{\tilde{g}}_2^2
      =
      \sum\limits_{r,r'\neq \chi^{\star}} \widehat{F}_t(r,\chi^{\star})\overline{\widehat{F}_t(r',\chi^{\star})}\card{\rho_{\chi^{\star}}}^2\kappa_r\overline{\kappa_{r'}}
      \inner{g_r'}{g_{r'}'}
      =\card{\rho_{\chi^{\star}}}^2\sum\limits_{r\neq \chi^{\star}}\card{\widehat{F}_t(r,\chi^{\star})}^2\card{\kappa_r}^2,
      \]
      and similarly for $\tilde{h}$. Also, $f_t'(x)\chi^{\star}(-\sigma(x))$ is perpendicular
      to all embedding functions. Applying Claim~\ref{claim:additive_base_case} we conclude that
      \begin{align*}
      &\card{\Expect{(x,y,z)\sim\mu}{f_t'(x)\chi^{\star}(-\sigma(x))(\tilde{g}(y) + \tilde{h}(z))}}\\
      &\qquad\qquad\qquad\leq
      (1-c)\sqrt{\card{\kappa_{\chi^{\star}}}^2\sum\limits_{s\neq \chi^{\star}}\card{\rho_s}^2\card{\widehat{F}_t(\chi^{\star},s)}^2
      +\card{\rho_{\chi^{\star}}}^2\sum\limits_{r\neq \chi^{\star}}\card{\kappa_r}^2\card{\widehat{F}_t(r,\chi^{\star})}^2}.
      \end{align*}
      Plugging this into~\eqref{eq:apply_additive_base_case} gives that
      \begin{align}\label{eq:apply_additive_base_case2}
      \card{\eqref{eq:apply_additive_base_case}}
      &\leq
      (1-c)\sum\limits_{t\in T}\card{\psi_t}\sqrt{\card{\kappa_{\chi^{\star}}}^2\sum\limits_{s\neq \chi^{\star}}\card{\rho_s}^2\card{\widehat{F}_t(\chi^{\star},s)}^2
      +\card{\rho_{\chi^{\star}}}^2\sum\limits_{r\neq \chi^{\star}}\card{\kappa_r}^2\card{\widehat{F}_t(r,\chi^{\star})}^2} \notag\\
      &\leq
      (1-c)\sqrt{\sum\limits_{t\in T}\card{\psi_t}^2}
      \sqrt{\card{\kappa_{\chi^{\star}}}^2\sum\limits_{s\neq \chi^{\star}}\card{\rho_s}^2\sum\limits_{t\in T}\card{\widehat{F}_t(\chi^{\star},s)}^2
      +\card{\rho_{\chi^{\star}}}^2\sum\limits_{r\neq \chi^{\star}}\card{\kappa_r}^2\sum\limits_{t\in T}\card{\widehat{F}_t(r,\chi^{\star})}^2},
      \end{align}
      where we used Cauchy-Schwarz. For fixed $s\in S\cup\hat{H}$, note that letting
      $\tilde{F} = \frac{\sum\limits_{t\in T}\overline{\widehat{F}_t(\chi^{\star},s)} F_t}{\sqrt{\sum\limits_{t\in T}\card{\widehat{F}_t(\chi^{\star},s)}^2}}$
      we have that $\tilde{F}$ has $2$-norm equal to $1$, is completely embedding homogenous and non-embedding
      homogenous of degree $d-1$. Thus,
      \[
      \sum\limits_{t\in T}\card{\widehat{F}_t(\chi^{\star},s)}^2
      =
      \card{\Expect{(x,y,z)\sim \mu^{\otimes n-1}}{\tilde{F}(y,z)g_{\chi^{\star}}(y)h_s(z)}}^2
      \leq \beta_{n-1,d-1,d'-1}^2.
      \]
      Similarly, $\sum\limits_{t\in T}\card{\widehat{F}_t(r,\chi^{\star})}^2\leq \beta_{n-1,d-1,d'-1}^2$ for all $r\in R\cup\hat{H}$,
      and plugging this into~\eqref{eq:apply_additive_base_case2} yields that
      \[
      \card{\eqref{eq:apply_additive_base_case}}
      \leq(1-c)\beta_{n-1,d-1,d'-1}\sqrt{\sum\limits_{t\in T}\card{\psi_t}^2}
      \sqrt{\card{\kappa_{\chi^{\star}}}^2\sum\limits_{s\neq \chi^{\star}}\card{\rho_s}^2
      +\card{\rho_{\chi^{\star}}}^2\sum\limits_{r\neq \chi^{\star}}\card{\kappa_r}^2}.
      \]
    \end{proof}

    Plugging Lemmas~\ref{lem:bound_main_term_beta}~\ref{lem:bound_small_beta},~\ref{lem:bound_additivebase_beta} into~\eqref{eq:large_large_0} we get that
    \begin{align*}
    \beta_{n,d,d'}
    &\leq
    \beta_{n-1,d,d'}\sqrt{\sum\limits_{\chi\in \hat{H}}\card{\psi_{\chi}}^2}
    \sqrt{\sum\limits_{\chi\in\hat{H}}\card{\kappa_{\chi}}^2\card{\rho_{\chi}}^2}\\
    &+
    (1-c)\beta_{n-1,d-1,d'-1}\sqrt{\sum\limits_{t\in T} \card{\psi_t}^2}
        \sqrt{\card{\kappa_{\chi^{\star}}}^2\sum\limits_{s\neq \chi^{\star}}\card{\rho_s}^2
      +\card{\rho_{\chi^{\star}}}^2\sum\limits_{r\neq \chi^{\star}}\card{\kappa_r}^2}\\
    &+
    C_m\beta_{n-1,d,d'}\sqrt{\sum\limits_{t\in T} \card{\psi_t}^2}\sqrt{1-\card{\kappa_{\chi^{\star}}}^2}\sqrt{1-\card{\rho_{\chi^{\star}}}^2}.
    \end{align*}
    Denoting $\beta'' = \max((1-c/2)\beta_{n-1,d-1,d'-1},\beta_{n-1,d,d'})$ we get by Cauchy-Schwarz that
    \begin{align*}
    & \beta_{n,d,d'}\\
    &\leq \beta''\sqrt{\sum\limits_{\chi\in \hat{H}}\card{\psi_{\chi}}^2+\sum\limits_{t\in T} \card{\psi_t}^2}
    \sqrt{\sum\limits_{\chi\in\hat{H}}\card{\kappa_{\chi}}^2\card{\rho_{\chi}}^2+
    (1-c/2)\card{\kappa_{\chi^{\star}}}^2\sum\limits_{s\neq \chi^{\star}}\card{\rho_s}^2
    +(1-c/2)\card{\rho_{\chi^{\star}}}^2\sum\limits_{r\neq \chi^{\star}}\card{\kappa_r}^2}\\
    &+
    C_m\beta''\sqrt{\sum\limits_{t\in T} \card{\psi_t}^2}\sqrt{1-\card{\kappa_{\chi^{\star}}}^2}\sqrt{1-\card{\rho_{\chi^{\star}}}^2}.
    \end{align*}
    We have that
    \[
    \sum\limits_{\chi\in\hat{H}}\card{\kappa_{\chi}}^2\card{\rho_{\chi}}^2+
    \card{\kappa_{\chi^{\star}}}^2\sum\limits_{s\neq \chi^{\star}}\card{\rho_s}^2
    +\card{\rho_{\chi^{\star}}}^2\sum\limits_{r\neq \chi^{\star}}\card{\kappa_r}^2
    \leq
    \sum\limits_{r,s}\card{\kappa_r}^2\card{\rho_s}^2
    =1,
    \]
    so we conclude that
    \begin{align*}
    \beta_{n,d,d'}
    &\leq \beta''
    \sqrt{1- \frac{c}{2}\left(\card{\kappa_{\chi^{\star}}}^2\sum\limits_{s\neq \chi^{\star}}\card{\rho_s}^2+\card{\rho_{\chi^{\star}}}^2
    \sum\limits_{r\neq \chi^{\star}}\card{\kappa_r}^2\right)}\\
    &+
    C_m\beta''\sqrt{\sum\limits_{t\in T} \card{\psi_t}^2}\sqrt{1-\card{\kappa_{\chi^{\star}}}^2}\sqrt{1-\card{\rho_{\chi^{\star}}}^2}.
    \end{align*}
    Using $\sum\limits_{s\neq \chi^{\star}}\card{\rho_s}^2 = 1-\card{\rho_{\chi^{\star}}}^2$ and
    $\sum\limits_{r\neq \chi^{\star}}\card{\kappa_r}^2 = 1-\card{\kappa_{\chi^{\star}}}^2$ we get
    \begin{align*}
    \beta_{n,d,d'}
    &\leq \beta''
    \sqrt{1- \frac{c}{4}\left(\card{\kappa_{\chi^{\star}}}^2+\card{\rho_{\chi^{\star}}}^2-2\card{\kappa_{\chi^{\star}}}^2\card{\rho_{\chi^{\star}}}^2\right)}\\
    &+
    C_m\beta''\sqrt{\sum\limits_{t\in T} \card{\psi_t}^2}\sqrt{1-\card{\kappa_{\chi^{\star}}}^2}\sqrt{1-\card{\rho_{\chi^{\star}}}^2}.
    \end{align*}
    As $\sqrt{1-a}\leq 1-a/2$ and $\sqrt{ab}\leq \frac{1}{2}(a+b)$ for non-negative $a,b$ and $\sum\limits_{t\in T} \card{\psi_t}^2\leq \eta$, we get that
    \[
    \frac{\beta_{n,d,d'}}{\beta''}
    \leq
    \left(1- \frac{c}{8}\left(\card{\kappa_{\chi^{\star}}}^2+\card{\rho_{\chi^{\star}}}^2-2\card{\kappa_{\chi^{\star}}}^2\card{\rho_{\chi^{\star}}}^2\right)\right)
    +C_m\sqrt{\eta}\left(2-\card{\kappa_{\chi^{\star}}}^2  - \card{\rho_{\chi^{\star}}}^2\right).
    \]
    We write the right hand side as $1-\frac{c}{8} P(\card{\kappa_{\chi^{\star}}}^2, \card{\rho_{\chi^{\star}}}^2)$, where
    $P\colon [0,1]^2\to\mathbb{R}$ is defined as
    \[
    P(a,b) = a+b-2ab - \frac{8C_m}{c}\sqrt{\eta}(2-a-b).
    \]
    \begin{claim}\label{claim:bound_calculus}
      Given~\eqref{eq:hier}, we have that $P(a,b)\geq 0$ for all $a,b\geq 1-6\tau$.
    \end{claim}
    \begin{proof}
      Looking at partial derivatives, we see that $\frac{\partial P}{\partial a} = 1-2b + \frac{8C_m}{c}\sqrt{\eta}$,
      and similarly $\frac{\partial P}{\partial b} = 1-2a + \frac{8C_m}{c}\sqrt{\eta}$.
      It follows that $P$ is decreasing in both variables in the range that $a,b \geq 1/2 + \frac{4C_m}{c}\sqrt{\eta}$, hence
      for all $a,b\geq 1-6\tau$ we have that $P(a,b)\geq P(1,1) = 0$.
    \end{proof}
    Combining Claim~\ref{claim:bound_calculus} and~\eqref{eq:beta_arg_most_mass_embed}, it follows that
    $\frac{\beta_{n,d,d'}}{\beta''}\leq 1-\frac{c}{8} P(\card{\kappa_{\chi^{\star}}}^2, \card{\rho_{\chi^{\star}}}^2)\leq 1$, concluding the proof.
    \qed

\subsection{Conclusion of the Reduction to Near Linear Degree: Iterating Lemma~\ref{lem:reduce_beta}}
To conclude this section, we iterate Lemma~\ref{lem:reduce_beta} and get that we can either get an exponential
upper bound on our parameter $\beta$, or else we can reduce to the case the number of variables $n$ is proportional
to the non-embedding degree $d$, and furthermore the non-embedding degree and effective non-embedding degree stay
roughly the same. Formally:
\begin{lemma}\label{lem:beta_bound_iterated}
  For all $m\in\mathbb{N}$ and $\alpha>0$ there are $L>0$ and $\eps>0$ such that the following holds for all $\xi>0$.
  Let $\mu$ be as in Theorem~\ref{thm:nonembed_homogenous}, and let $n\geq d\geq d'$
  be integers such that $d'\geq d^{1-\xi}$. Then either:
  \begin{enumerate}
    \item $\beta_{n,d,d'}[\mu]\leq (1-\eps)^{d'/2}$.
    \item Else, there are integers $e\geq d/2$, $e'\geq d'/2$ and $n'\leq 2L e$.
    such that$\beta_{n,d,d'}[\mu]\leq \beta_{n',e,e'}$.
  \end{enumerate}
\end{lemma}
\begin{proof}
  We apply Lemma~\ref{lem:reduce_beta} so long as possible. When we can no longer do that, let $s$ be the number of times
  it was the case that the upper bound was $(1-\eps)\beta_{n-1,d-1,d'-1}$. If $s\geq d'/2$, then we clearly got a factor of $1-\eps$
  $s$ times, and using the trivial bound of $1$ on all $\beta$'s it follows that the first item holds.

  Else, $s\leq d'/2$ and hence the non-embedding degree parameter $e$ in the end has dropped by at most $s$,
  so $e\geq d - s\geq d/2$. Also, the effective non-embedding degree lower bound parameter $e'$ also drops
  by at most $s$, hence we get that $e'\geq d'-s\geq d'/2$. Since we can no longer apply Lemma~\ref{lem:reduce_beta}
  it follows that the number of variables we reached to is $n'\leq L\cdot d\leq 2L e$, and also that
  \[
  \beta_{n,d,d'}[\mu]
        \leq
        (1-\eps)^s\beta_{n',e,e'}[\mu]
        \leq
        \beta_{n',e,e'}[\mu].
        \qedhere
  \]
\end{proof}
\section{Proof of Theorem~\ref{thm:nonembed_homogenous}: Proof for Near Linear Degrees}\label{sec:prove_near_lin}
In this section we start off where Section~\ref{sec:reduce_to_near_linear} ended, and handle the case
that the number of variables $n$ is as in Lemma~\ref{lem:beta_bound_iterated}. In that case, we present a different inductive
argument, similar to the one in~\cite[Section 6]{BKMcsp2}, that establishes an exponential upper bound on $\beta$. This
is the part in our argument in which we use the relaxed base case.

\subsection{The Parameter $\delta$}
In contrast to Section~\ref{sec:reduce_to_near_linear}, the argument we present herein will not be able to preserve homogeneity and instead
will reduce degrees in a controlled manner. To facilitate that, we define a modification of the class $\mathcal{F}_{n,d,d'}$ as follows:
\begin{definition}
  We define the class $\mathcal{F}_{n,d'}$ as the class of functions $F\colon \Gamma^n\times\Phi^n\to\mathbb{C}$ that are constant on connected
  components and have effective non-embedding degree at least $d'$.
\end{definition}
We now define a variant of the parameter $\beta$ that is central to this section.
\begin{definition}
    For integers $n\geq d'$, finite alphabets $\Sigma$, $\Gamma$, $\Phi$ and a distribution $\mu$ over $\Sigma\times\Gamma\times \Phi$ as in
    Theorem~\ref{thm:nonembed_deg_must_be_small_rephrase_maximal_relaxed}, we define
  \[
         \delta_{n,d'}[\mu] =
        \sup\limits_{\substack{F\in \mathcal{F}_{n,d'}\\ g\colon \Gamma^n\to\mathbb{C}\\ h\colon\Phi^n\to\mathbb{C}}}
        \frac{\card{\Expect{(x,y,z)\sim \mu^{\otimes n}}{F(y,z)g(y)h(z)}}}{\norm{F}_2\norm{g}_2\norm{h}_2}.
  \]
\end{definition}
  When the distribution $\mu$ is clear from context, we will omit it from the notation and simply write $\delta_{n,d'}$.
  With the parameter $\delta_{n,d'}$ in hand, we can now state the main inductive statement of this section.
  \begin{lemma}\label{lem:delta_inductive}
  For all $m\in\mathbb{N}$ and $\alpha>0$ there is $C>0$ such that the following holds.
  For integers $n\geq d'$ such that $n\leq d' A$ and $A\geq 2$, finite alphabets $\Sigma$, $\Gamma$, $\Phi$ and a distribution $\mu$ over $\Sigma\times\Gamma\times \Phi$ as in
  Theorem~\ref{thm:nonembed_deg_must_be_small_rephrase_maximal_relaxed}, we have that
  \[
         \delta_{n,d'}[\mu] \leq \left(1-\frac{1}{A^C}\right)\delta_{n-1,d'-1}[\mu].
  \]
\end{lemma}

The rest of this section is devoted to the proof of Lemma~\ref{lem:delta_inductive}. We first establish a singular-value decomposition stated similar
to Claim~\ref{claim:svd_effective} but for the case that $F$ is not homogenous, and relate this decomposition with the modest influences similarly to
Claim~\ref{claim:coefs_svd_effective}.

\subsection{The SVD Decomposition and Relation to Modest Influences}
We are going to need the following two singular-value decomposition statements, and as usual we are going to have a partition
$I,J$ of $[n]$ into $\card{I} = n-1$ and $\card{J}=1$. The first statement addresses the functions $g$ and $h$:
\begin{claim}\label{claim:svd_effective_nonhomo_g}
		Suppose $g\colon\Gamma^n\to\mathbb{C}$ has $2$-norm equal to $1$. Then we may write
		\[
		g(y) = \sum\limits_{r\in R} \kappa_r g_r(y_I) g_r'(y_J),
		\]
		where each $\kappa_r\geq 0$ and
		\begin{enumerate}
			\item For $r\in R$, $F_r\colon \Gamma^{I}\to\mathbb{C}$ is an orthonormal set of functions.
            \item For $r\in R$, $F_r'\colon \Gamma^{J}\to\mathbb{C}$ is an orthonormal set of functions.
			\item $\sum\limits_{r\in R}\kappa_r^2 = 1$.
		\end{enumerate}
\end{claim}
\begin{proof}
  The proof is similar to the proof of Claim~\ref{claim:svd} and is simpler (no discussion of invariant spaces is needed), and we omit the details.
\end{proof}

The second singular-value decomposition statement addresses the function $F$:
\begin{claim}\label{claim:svd_effective_nonhomo}
		Suppose $F\colon\Gamma^n\times \Phi^n\to\mathbb{C}$ has $2$-norm equal to $1$, is constant on connected components and has effective non-embedding degree at least $d'$.
        Then we may write
		\[
		F(y,z) = \sum\limits_{t\in T} \psi_t F_t(y_I,z_I) F_t'(y_J,z_J),
		\]
		where each $\psi_t\geq 0$ and
		\begin{enumerate}
			\item For $t\in T$, $F_t\colon \Gamma^{I}\times\Phi^{I}\to\mathbb{C}$ is an orthonormal set of functions.
            \item For $t\in T$, $F_t'\colon \Gamma^{J}\times \Phi^J\to\mathbb{C}$ is an orthonormal set of functions.
			\item For $t\in T$, $F_t$ has effective non-embedding degree at least $d'-1$.
			\item For $t\in T$, the functions $F_t$ and $F_t'$ are constant on connected components.
			\item $\sum\limits_{t\in T}\psi_t^2 = 1$.
		\end{enumerate}
\end{claim}
\begin{proof}
  The proof is similar to the proof of Claim~\ref{claim:svd_effective} and is simpler (no discussion of invariant spaces is needed), and we omit the details.
\end{proof}

The next statement relates the SVD decomposition of $F$ and the variance the $F_t'$ have on $\Sigma_{{\sf modest}}$ with the notion of
modest influences. For that, we define the variance a function on $\Sigma_{{\sf modest}}$:
\begin{definition}
  For a function $f\colon \Sigma \to\mathbb{C}$ we denote ${\sf var}_{\Sigma_{{\sf modest}}}(f) = \Expect{a,b\in \Sigma_{{\sf modest}}}{\card{f(a) - f(b)}^2}$.
\end{definition}

With this definition, the following claim asserts that if a function has large modest influence, then the corresponding parts $F_t'$ have significant
variance on $\Sigma_{{\sf modest}}$. More precisely:
\begin{claim}\label{claim:large_modestinf_sig_var}
        If $J = \{j\}$, then letting $f_t'\colon \Sigma\to\mathbb{C}$ be such that $F_t' = W f_t'$ we have that
		\[
        \sum\limits_{t}\psi_t^2{\sf var}_{\Sigma_{{\sf modest}}}(f_t')\ggg_{m,\alpha} I_{j, {\sf modest}}[F].
        \]
	\end{claim}
	\begin{proof}
        Write $F = W f$ for $f\colon\Sigma^n\to\mathbb{C}$, so that from the SVD decomposition of $F$ we get a similar SVD decomposition for
        $f$ with $f_t, f_t'$ satisfying $F_t = Wf_t$, $F_t' = Wf_t'$. Consider $I_{j, {\sf modest}}[F] = I_{j, {\sf modest}}[f]$,
        and consider the distribution over $(a,b)$ where $a\sim \mu_x$, and $b\sim {\sf Modest}~a$. Then by definition
		\begin{align*}
			I_{j, {\sf modest}}[f]
			&= \Expect{x\sim\mu_x^{\otimes n-1}, a, b}{\card{f(x_I,a) - f(x_I,b)}^2}\\
			&= \Expect{x\sim\mu_x^{\otimes n-1}, a, b}{\card{\sum\limits_{t}\psi_t f_t(x_I)(f_t'(a) - f_t'(b))}^2}\\
			&= \sum\limits_{t_1,t_2} \psi_{t_1}\psi_{t_2} \inner{f_{t_1}}{f_{t_2}} \Expect{a, b}{(f_{t_1}'(a) - f_{t_1}'(b))\overline{(f_{t_2}'(a) - f_{t_2}'(b))}}.
		\end{align*}
		For $t_1\neq t_2$ we have $\inner{f_{t_1}}{f_{t_2}} = 0$, so the last sum is equal to
		\[
		\sum\limits_{t}\psi_t^2 \Expect{a, b}{\card{f_{t}'(a) - f_{t}'(b)}^2} \lll_{\alpha,m}
		\sum\limits_{t}\psi_t^2 {\sf var}_{\Sigma_{{\sf modest}}}(f_t')
        \qedhere.
		\]
	\end{proof}
\subsection{Proof of Lemma~\ref{lem:delta_inductive}}
We are going to have the following hierarchy of parameters in this section (we recall that $A$ is in Lemma~\ref{lem:delta_inductive}):
\begin{align}\label{eq:hier2}
    &0 \ll R_4^{-1}\ll R_3^{-1}\ll R_2^{-1}\ll R_1^{-1}\ll c \ll \alpha, m^{-1}\leq 1,\notag\\
    &0 < \eta = \frac{1}{A^{R_4}}< \eps = \frac{1}{A^{R_3}}< \zeta = \frac{1}{A^{R_2}}< \tau = \frac{1}{A^{R_1}}\leq 1.
\end{align}
Fix $n$ and $d'$ as in Lemma~\ref{lem:delta_inductive}, and let $F\in\mathcal{F}_{n,d'}$ and $g,h$ be functions of $2$-norm $1$
so that
\[
\delta_{n,d'} = \Expect{(x,y,z)\sim \mu^{\otimes n}}{F(x)g(y)h(z)}.
\]
As $I_{{\sf modest}}[F]\geq 2d'$ by Fact~\ref{fact:modest_inf_formula}, it follows that there is $j\in[n]$ such that
\begin{equation}\label{eq:choose_j_in_delta}
I_{j,{\sf modest}}[F]\geq \frac{I_{{\sf modest}}[F]}{n}\geq \frac{2d'}{n}\geq \frac{1}{A},
\end{equation}
and we fix such $j$ henceforth. We choose the partition $I = [n]\setminus\{j\}$ and $J = \{j\}$, and then use the SVD decompositions
for $g$ and $h$ from Claim~\ref{claim:svd_effective_nonhomo_g} and for $F$ as in Claim~\ref{claim:svd_effective_nonhomo}
to write
\[
g(y) = \sum\limits_{r\in R}\kappa_r g_r(y_I)g_r'(y_J),
\qquad
h(z) = \sum\limits_{s\in S}\rho_s h_s(y_I)h_s'(z_J),
\qquad
F(y,z) = \sum\limits_{t\in T}\psi_t F_t(y_I,z_J) F_t'(y_I, z_J).
\]
Thus, using the notations of the previous section we have that
\begin{equation}\label{eq:delta_1}
\delta_{n,d'} = \sum\limits_{t\in T, r\in R, s\in S}\psi_t \kappa_r \rho_s \widehat{F}_t(r,s)\widehat{F}_t'(r,s).
\end{equation}

\subsubsection{Separating Out Coefficients Bounded Away From $0$ and Coefficients Close to $0$}
The role of the parameters $\eta$ and $\tau$ above will be that none of the coefficients
$\kappa_r, \rho_s$ and $\psi_t$ will be in the interval $[\eta,\tau)$. Formally, we show
\begin{claim}\label{claim:separate_out}
    We may find parameters as in~\eqref{eq:hier2} such that
    none of the coefficients $\kappa_r, \rho_s$ and $\psi_t$ are in the interval $[\eta,\tau)$.
\end{claim}
\begin{proof}
  We start with some parameters as in~\eqref{eq:hier2}.
  As long as there is a coefficient in the interval $[\eta,\tau)$,
  we take a new set of $R$'s by taking $R_1' = R_4$, $R_4'\gg R_3'\gg R_2'\gg R_1'$ and take the collection $\{R_i'\}_{i=1}^{4}$.
  Note that as $R_1' > R_4$ the intervals we consider at each step are disjoint, and therefore
  after at most $3m$ iterations we will get that none of the coefficients lie in the interval $[\eta,\tau)$, and
  we are done.
\end{proof}
We assume henceforth that $\eta\ll \tau$ that none of $\kappa_r, \rho_s, \psi_t$ are in $[\eta,\tau)$.
Define
\[
T' = \sett{t\in T}{\psi_t\geq \tau},
\qquad
R' = \sett{r\in R}{\kappa_r\geq \tau},
\qquad
S' = \sett{s\in S}{\rho_s\geq \tau}.
\]
Using~\eqref{eq:delta_1} we write
\begin{equation}\label{eq:delta_2}
\delta_{n,d'}
=
\underbrace{\sum\limits_{t\in T', r\in R', s\in S'}\psi_t \kappa_r \rho_s \widehat{F}_t(r,s)\widehat{F}_t'(r,s)}_{(\rom{1})}
+
\underbrace{\sum\limits_{\substack{t\in T, r\in R, s\in S\\t\not\in T'\text{ or }r\not\in R'\text{ or }s\not\in S'}}\psi_t \kappa_r \rho_s \widehat{F}_t(r,s)\widehat{F}_t'(r,s)}_{(\rom{2})}.
\end{equation}

The majority of our effort will go into bounding the contribution of $(\rom{1})$ (which we defer to the rest of this section),
and we first bound the contribution of $(\rom{2})$ by the following claim.
\begin{claim}\label{claim:bound_err_term_gamma}
    It holds that $\card{(\rom{2})}\lll_{m} \eta \delta_{n-1,d'-1}$.
\end{claim}
\begin{proof}
  By the triangle inequality, it suffices to upper bound each summand. For each $r,s,t$ in the sum we have that
  $\card{\psi_t \kappa_r \rho_s}\leq \eta$, as at least one of them is at most $\eta$ in absolute value, and the
  other two are at most $1$. Also, $\card{\widehat{F}_t(r,s)}\leq \delta_{n-1,d'-1}$ and
  $\card{\widehat{F}_t'(r,s)}\leq 1$, and the claim follows since the number of terms in the sum
  is at most $m^3$.
\end{proof}

\subsubsection{A Naive Cauchy-Schwarz Bound on $(\rom{1})$}
We now present a naive Cauchy-Schwarz based argument showing that $\card{(\rom{1})}\leq \delta_{n-1,d'-1}$ which can be used to recover the trivial bound
$\delta_{n,d'}\leq \delta_{n-1,d'-1}$. The upshot of this argument is that by inspecting near equality cases, the argument will allow us to say that unless
near equalities in the pursuing applications of Cauchy-Schwarz argument hold, we will successfully have shown that $\delta_{n,d'}\leq (1-\eps)\delta_{n-1,d'-1}$.
Hence, we will be able to assume in the rest of the argument that all of the Cauchy-Schwarz applications were
nearly tight.
To be more precise, by Cauchy-Schwarz
\begin{align*}
\card{(\rom{1})}
&\leq
\sqrt{\sum\limits_{t\in T', r\in R', s\in S'}\kappa_r^2\rho_s^2\card{\widehat{F}_t(r,s)}^2}
\sqrt{\sum\limits_{t\in T', r\in R', s\in S'}\psi_t^2\card{\widehat{F}_t'(r,s)}^2}\\
&=\sqrt{\sum\limits_{r\in R', s\in S'}\kappa_r^2\rho_s^2\sum\limits_{t\in T'}\card{\widehat{F}_t(r,s)}^2}
\sqrt{\sum\limits_{t\in T'}\psi_t^2\sum\limits_{r\in R', s\in S'}\card{\widehat{F}_t'(r,s)}^2}.
\end{align*}
We note that
$\sum\limits_{t\in T'}\psi_t^2\sum\limits_{r\in R', s\in S'}\card{\widehat{F}_t'(r,s)}^2\leq \sum\limits_{t\in T'}\psi_t^2$
as $\sum\limits_{r\in R', s\in S'}\card{\widehat{F}_t'(r,s)}^2\leq 1$, and that
\[
\sum\limits_{t\in T'}\card{\widehat{F}_t(r,s)}^2
=\Expect{(x,y,z)\sim \mu^{\otimes(n-1)}}{\tilde{F}(x)g_r(y)h_s(z)}^2
\]
where $\tilde{F}(x) = \frac{\sum\limits_{t\in T'}\overline{\widehat{F}_t(r,s)}F_t}{\sqrt{\sum\limits_{t\in T'}\card{\widehat{F}_t(r,s)}^2}}$,
which by the definition is at most $\delta_{n-1,d'-1}^2$. Combining, we get that
\[
\card{(\rom{1})}
\leq
\sqrt{\sum\limits_{r\in R', s\in S'}\kappa_r^2\rho_s^2\delta_{n-1,d_1'-1}^2}
\sqrt{\sum\limits_{t\in T'}\psi_t^2}
\leq \delta_{n-1,d'-1}.
\]
Inspecting the proof, we see that if we had that
\begin{align*}
&\card{\sum\limits_{t\in T', r\in R', s\in S'}\psi_t \kappa_r \rho_s \widehat{F}_t(r,s)\widehat{F}_t'(r,s)}\\
&\qquad\leq
(1-\zeta)
\sqrt{\sum\limits_{t\in T', r\in R', s\in S'}\kappa_r^2\rho_s^2\card{\widehat{F}_t(r,s)}^2}
\sqrt{\sum\limits_{t\in T', r\in R', s\in S'}\psi_t^2\card{\widehat{F}_t'(r,s)}^2}
\end{align*}
then we would get that $\card{(\rom{1})}\leq (1-\zeta)\gamma_{n-1,d_1'-1}$ and combining with Claim~\ref{claim:bound_err_term_gamma} we get that
$\delta_{n,d_1'}\leq (1-\eps)\delta_{n-1,d'-1}$ and the proof would be concluded. We henceforth assume that
\begin{align}\label{eq:assume_delta_1}
&\card{\sum\limits_{t\in T', r\in R', s\in S'}\psi_t \kappa_r \rho_s \widehat{F}_t(r,s)\widehat{F}_t'(r,s)}\notag\\
&\qquad>
(1-\zeta)
\sqrt{\sum\limits_{t\in T', r\in R', s\in S'}\kappa_r^2\rho_s^2\card{\widehat{F}_t(r,s)}^2}
\sqrt{\sum\limits_{t\in T', r\in R', s\in S'}\psi_t^2\card{\widehat{F}_t'(r,s)}^2}.
\end{align}
Also, we see that if we had that $\sum\limits_{t\in T'}\psi_t^2\sum\limits_{r\in R', s\in S'}\card{\widehat{F}_t'(r,s)}^2\leq (1-\zeta)\sum\limits_{t\in T'}\psi_t^2$,
then we would get that $\card{(\rom{1})}\leq \sqrt{1-\zeta}\delta_{n-1,d'-1}$ and combining with Claim~\ref{claim:bound_err_term_gamma} we get that
$\delta_{n,d'}\leq (1-\eps)\delta_{n-1,d'-1}$ and the proof would be concluded. We henceforth assume that
\begin{equation}\label{eq:assume_delta_2}
  \sum\limits_{t\in T'}\psi_t^2\sum\limits_{r\in R', s\in S'}\card{\widehat{F}_t'(r,s)}^2> (1-\zeta)\sum\limits_{t\in T'}\psi_t^2.
\end{equation}
Lastly, if
$\sum\limits_{t\in T', r\in R', s\in S'}\kappa_r^2\rho_s^2\card{\widehat{F}_t(r,s)}^2\leq (1-\zeta)\delta_{n-1,d'-1}^2\sum\limits_{r\in R', s\in S'}\kappa_r^2\rho_s^2$
then once again we will be able to conclude that $\delta_{n,d_1'}\leq (1-\eps)\delta_{n-1,d'-1}$, and thus henceforth we have that
\begin{equation}\label{eq:assume_delta_3}
  \sum\limits_{t\in T', r\in R', s\in S'}\kappa_r^2\rho_s^2\card{\widehat{F}_t(r,s)}^2 > (1-\zeta)\delta_{n-1,d'-1}^2\sum\limits_{r\in R', s\in S'}\kappa_r^2\rho_s^2.
\end{equation}

\subsubsection{The Random Pertubation Argument}
Define $M_1 = \sum\limits_{r\in R'} \kappa_r^2$ and $M_2 = \sum\limits_{s\in S'} \rho_s^2$. Below, we will consider complex-valued variables
$\pi_{r}$ and $\theta_s$ satisfying that
\begin{equation}\label{eq:delta_3}
\sum\limits_{r\in R'} \card{\pi_r}^2\kappa_r^2 = M_1,
\qquad
\sum\limits_{s\in S'} \card{\theta_s}^2\rho_s^2 = M_2.
\end{equation}
Eventually, we choose $\pi_r, \theta_s$ according to a distribution over the points
satisfying these equalities. More precisely, we choose the distribution of $\pi_r$'s so that
the vector $(\kappa_r\pi_r)_{r\in R'}$ is distributed uniformly over vectors in $\mathbb{C}^{R'}$
with $2$-norm equal to $\sqrt{M_1}$; similarly, we choose the distribution of $\theta_s$ so that
the vector $(\theta_s\pi_s)_{s\in R'}$ is distributed uniformly over vectors in $\mathbb{C}^{S'}$
with $2$-norm equal to $\sqrt{M_2}$.

The point of these constraints is so that we can define the functions
\[
\tilde{g}(y) = \sum\limits_{r\in R'}\kappa_r\pi_r g_r(y),
\qquad
\tilde{h}(z) = \sum\limits_{s\in S'}\rho_s\theta_s h_s(z),
\]
and have that their $2$-norms squared are equal to $M_1$ and $M_2$ respectively.

The functions $\tilde{g}$, and $\tilde{h}$ allow us to consider the lower order, $n-1$ dimensional problems
corresponding to $\tilde{F}$ and $\tilde{g}$, $\tilde{h}$ where
$\tilde{F} = \frac{\sum\limits_{t}\inner{F_t}{\tilde{g}\tilde{h}} F_t}{\sqrt{\sum\limits_{t}\card{\inner{F_t}{\tilde{g}\tilde{h}}}^2}}$.
Indeed, we define the function
\[
p(\{\pi_r\}_{r\in R'}, \{\theta\}_{s\in S'})
=\card{\Expect{(x,y,z)\sim \mu^{\otimes(n-1)}}{\tilde{F}(x)\tilde{g}(y)\tilde{h}(z)}}^2
=\sum\limits_{t}\card{\inner{F_t}{\overline{\tilde{g}\tilde{h}}}}^2.
\]
By the first expression for $p$ and definition, we know that
$\card{p(\{\pi_r\}_{r\in R'}, \{\theta\}_{s\in S'})}\leq \delta_{n-1,d'-1}^2$.
In our argument we will analyze the expectation of $p(\cdot)$ and show  it is close to $\delta_{n-1,d'-1}^2$,
from which it follows that the value of $p(\cdot)$ is roughly constant. On the other end, we will analyze
the variance of $p(\cdot)$ and lower bound it, and combining these two facts will finish the proof. We now do each
one of these steps in detail.

First, we establish the point-wise upper bound on $p$:
\begin{claim}\label{claim:pointwise_bd_p}
  For all inputs satisfying~\eqref{eq:delta_3} we have that
  $p(\{\pi_r\}_{r\in R'}, \{\theta\}_{s\in S'})\leq \delta_{n-1,d'-1}^2$.
\end{claim}
\begin{proof}
  This is immediate by definition of $\delta_{n-1,d'-1}$.
\end{proof}

To compute the expectation and variance of $p$ it will be useful for us to expand it out. We have:
\begin{align}
  p(\{\pi_r\}_{r\in R'}, \{\theta\}_{s\in S'})
  &= \sum\limits_{t}\card{\inner{F_t}{\overline{\tilde{g}\tilde{h}}}}^2\notag\\
  &= \sum\limits_{t}\card{\sum\limits_{r\in R',s\in S'}\kappa_r\rho_s\pi_r\theta_s\widehat{F}_t(r,s)}^2\notag\\
  &= \sum\limits_{t}\sum\limits_{r,r'\in R',s,s'\in S'}\kappa_r\overline{\kappa_{r'}}\rho_s\overline{\rho_{s'}}\pi_r\pi_{r'}\theta_s\theta_{s'}\widehat{F}_t(r,s)\overline{\widehat{F}_t(r',s')}\notag\\
  &= \sum\limits_{r,r'\in R',s,s'\in S'}\kappa_r\kappa_{r'}\rho_s\overline{\rho_{s'}}\pi_r\overline{\pi_{r'}}\theta_s\theta_{s'}\inner{V_{r,s}}{V_{r',s'}},\label{eq:delta_4}
\end{align}
where the vector $V_{r,s}\in\mathbb{C}^{\card{T'}}$ is defined as $V_{r,s}(t) = \widehat{F}_t(r,s)$. We are also going to need the following claim that gives us bounds on
the norms of the vectors $V_{r,s}$.
\begin{claim}\label{claim:bound_norm_V}
    For all $r\in R'$ and $s\in S'$,
    $(1-\sqrt{\zeta})\delta_{n-1,d'-1}^2\leq \norm{V_{r,s}}_2^2\leq \delta_{n-1,d'-1}^2$.
\end{claim}
\begin{proof}
    Note that we may write~\eqref{eq:assume_delta_3} as
    $\sum\limits_{r\in R', s\in S'}\kappa_r^2\rho_s^2\norm{V_{r,s}}_2^2\geq (1-\zeta)\delta_{n-1,d'-1}^2M_1M_2$,
    and so
    \[
        \sum\limits_{r\in R', s\in S'}\kappa_r^2\rho_s^2(\delta_{n-1,d'-1}^2 - \norm{V_{r,s}}_2^2)
        \leq \zeta \delta_{n-1,d'-1}^2.
    \]
    Thus, all of the terms on the left hand side are non-negative and it follows that for all $r\in R'$ and $s\in S'$
    \[
        \delta_{n-1,d'-1}^2 - \norm{V_{r,s}}_2^2\leq
        \frac{\zeta \delta_{n-1,d'-1}^2}{\kappa_r^2\rho_s^2}\leq \frac{\zeta \delta_{n-1,d'-1}^2}{\tau^4}
        \leq \sqrt{\zeta}\delta_{n-1,d'-1}^2,
    \]
    and re-arranging gives the lower bound. As for the upper bound, note that for all $r\in R'$ and $s\in S'$ we have that $\norm{V_{r,s}}_2^2\leq \delta_{n-1,d'-1}^2$,
    as this norm can be written as $\card{\Expect{(x,y,z)\sim \mu^{\otimes(n-1)}}{F'(x)g_r(y)h_s(z)}}^2$
    where $F'(x) = \frac{\sum\limits_{t\in T'}\widehat{F}_t(r,s)F_t}{\sqrt{\sum\limits_{t\in T'}\card{\widehat{F}_t(r,s)}^2}}$.
\end{proof}

\subsubsection{Lower Bounding the Expectation of $p$ and Upper Bounding the Variance of $p$}
The following claim bounds the expectation of $p$.
\begin{claim}\label{claim:bound_expectation}
    Consider the distribution over $\pi_r$ and $\theta_s$ as defined after~\eqref{eq:delta_3}. Then
    \[
        \Expect{}{p(\{\pi_r\}_{r\in R'}, \{\theta_s\}_{s\in S'}) }\geq \delta_{n-1,d'-1}^2(1-3\sqrt{\zeta}).
    \]
\end{claim}
\begin{proof}
  We use~\eqref{eq:delta_4} and linearity of expectation. If $(r,s)\neq (r',s')$, say $r\neq r'$, then
  \[
  \Expect{}{\pi_r\overline{\pi_{r'}}\theta_s\overline{\theta_{s'}}}
  =\Expect{}{\pi_r\overline{\pi_{r'}}}\Expect{}{\theta_s\overline{\theta_{s'}}}
  =0\Expect{}{\theta_s\theta_{s'}}=0,
  \]
  as the distribution of $\pi_r$ is invariant under multiplying any one of these numbers by a random sign. Thus,
  only $(r,s) = (r',s')$ contribute to the expectation, and so
  \begin{align*}
  \Expect{}{p(\{\pi_r\}_{r\in R'}, \{\theta_s\}_{s\in S'}) }
  &=\Expect{}{\sum\limits_{r\in R',s\in S'}\kappa_r^2\rho_s^2\card{\pi_r}^2\card{\theta_s}^2\norm{V_{r,s}}_2^2}\\
  &\geq \delta_{n-1,d'-1}^2(1-\sqrt{\zeta})\Expect{}{\sum\limits_{r\in R',s\in S'}\kappa_r^2\rho_s^2\card{\pi_r}^2\card{\theta_s}^2},
  \end{align*}
  where we used Claim~\ref{claim:bound_norm_V}. Using~\eqref{eq:delta_3} we get that
  the right hand side is at least $M_1M_2\delta_{n-1,d'-1}^2(1-\sqrt{\zeta})$, and
  lastly we note that
  \[
  M_1
  = 1-\sum\limits_{r\in T\setminus T'}\kappa_r^2
  \geq 1-m\eta^2
  \geq 1-\zeta,
  \]
  and similarly $M_2\geq 1-\zeta$, and the claim follows.
\end{proof}

Combining Claims~\ref{claim:bound_expectation} and~\ref{claim:pointwise_bd_p} give a strong upper bound on the variance of $p$.
\begin{claim}\label{claim:p_var_ub}
  ${\sf var}(p(\{\pi_r\}_{r\in R'}, \{\theta_s\}_{s\in S'}))\lll \sqrt{\zeta}\delta_{n-1,d'-1}^4$.
\end{claim}
\begin{proof}
    By definition of variance it is equal to
    \[
        \Expect{}{p(\{\pi_r\}_{r\in R'}, \{\theta_s\}_{s\in S'})^2}
        -\Expect{}{p(\{\pi_r\}_{r\in R'}, \{\theta_s\}_{s\in S'})}^2
        \leq \delta_{n-1,d'-1}^4 - \delta_{n-1,d'-1}^4(1-3\sqrt{\zeta})^2,
    \]
    where we used Claim~\ref{claim:pointwise_bd_p} to upper bound the first expectation and Claim~\ref{claim:bound_expectation} to lower bound the second expectation.
    Simplifying finishes the proof.
\end{proof}

\subsubsection{Lower Bounding the Variance of $p$}
The rest of the argument is devoted to lower bounding the variance of $p$.
Write $p = p_1+p_2+p_3+p_4$ where
\[
p_1 = \sum\limits_{r\in R', s\neq s'\in S'}\kappa_r^2\rho_s\rho_{s'}\card{\pi_r}^2\theta_s\overline{\theta_{s'}}\inner{V_{r,s}}{V_{r,s'}},
\qquad
p_2 = \sum\limits_{r\neq r'\in R', s\in S'}\kappa_r\kappa_{r'}\rho_s^2\pi_r\overline{\pi_{r'}}\card{\theta_s}^2\inner{V_{r',s}}{V_{r,s}},
\]
\[
p_3 = \sum\limits_{r\neq r'\in R', s\neq s'\in S'}\kappa_r\kappa_{r'}\rho_s\overline{\rho_{s'}}\pi_r\overline{\pi_{r'}}\theta_s\theta_{s'}\inner{V_{r',s}}{V_{r,s'}},
\qquad
p_4 = \sum\limits_{r\in R', s\in S'}\kappa_r^2\rho_s^2\card{\pi_r}^2\card{\theta_s}^2\norm{V_{r,s}}_2^2.
\]
Note that $p_4$ is almost constant, and close to $\E[p]$. Indeed, by Claim~\ref{claim:bound_norm_V}
\begin{align*}
\card{\E[p] - p_4}
\leq \sum\limits_{r\in R', s\in S'}\kappa_r^2\rho_s^2\card{\pi_r}^2\card{\theta_s}^2(\delta_{n-1,d'-1}^2 - \norm{V_{r,s}}_2^2)
&\leq \sqrt{\zeta}\delta_{n-1,d'-1}^2\sum\limits_{r\in R', s\in S'}\kappa_r^2\rho_s^2\card{\pi_r}^2\card{\theta_s}^2\\
&\leq \sqrt{\zeta}\delta_{n-1,d'-1}^2.
\end{align*}
Thus, we get that
\begin{align*}
{\sf var}(p)
&=
\E[\card{p-\E[p]}^2]\\
&=\E[\card{p_1+p_2+p_3+p_4-\E[p]}^2]\\
&\geq
\E[\card{p_1+p_2+p_3}^2]
- \E[(p_1+p_2+p_3)\overline{(p_4-\E[p])}]\\
&- \E[\overline{(p_1+p_2+p_3)}(p_4-\E[p])]+
\E[\card{p_4-\E[p]}^2].
\end{align*}
We have that $\E[\card{p_4-\E[p]}^2]\leq \zeta\delta_{n-1,d'-1}^4$, as well as
\[
\card{\E[(p_1+p_2+p_3)\overline{(p_4-\E[p])}]}
\lll \E[\zeta^{1/2}\card{p_1+p_2+p_3}^2+\zeta^{-1/2}\card{p_4-\E[p]}^2]
\]
and
\[
\card{\E[\overline{(p_1+p_2+p_3)}(p_4-\E[p])]}
\lll \E[\zeta^{1/2}\card{p_1+p_2+p_3}^2+\zeta^{-1/2}\card{p_4-\E[p]}^2].
\]
Therefore, we get that
\begin{equation}\label{eq:delta_5}
{\sf var}(p)
\geq \left(1-O(\sqrt{\zeta})\right)\E[\card{p_1+p_2+p_3}^2]
-O(\sqrt{\zeta}\delta_{n-1,d'-1}^4).
\end{equation}
Inspecting~\eqref{eq:delta_5}, we now calculate $\E[\card{p_1+p_2+p_3}^2]$.
We get diagonal terms $\E[\card{p_1}^2]$, $\E[\card{p_2}^2]$, $\E[\card{p_3}^2]$
as well as off diagonal terms such as $\E[p_1\overline{p_2}]$, and we claim the off diagonal terms are $0$.
Indeed, looking at $p_1\overline{p_2}$ for example, if we expand it out we will get that each term will be
of the a multiple of the pattern $\card{\theta_s}^2\overline{\pi_r}\pi_{r'}\theta_{s'}\theta_{s''}\card{\pi_{r''}}^2$ where $r\neq r'$
(and $s'\neq s''$), and in particular either $r$ or $r'$ (or both) are different from $r''$ --- say $r$ ---
in which case the expectation of this pattern is $0$ as the distribution is invariant under multiplying $\pi_r$
by a random sign. Therefore~\eqref{eq:delta_5} gives that
\begin{equation}\label{eq:delta_6}
  {\sf var}(p) \geq \left(1-O(\zeta^{1/2})\right)\E[\card{p_1}^2+\card{p_2}^2+\card{p_3}^2]
-O(\sqrt{\zeta}\delta_{n-1,d'-1}^4).
\end{equation}
To use~\eqref{eq:delta_6} effectively we must prove that at least one of the monomials $p_1,p_2$ or $p_3$ has a significant $2$-norm, and towards
this end we must show that some inner product $\inner{V_{r,s}}{V_{r',s'}}$ for $(r,s)\neq (r',s')$ is significant. This is the content
of the following lemma, whose proof is deferred to Section~\ref{sec:lem_V_not_perp}.
\begin{lemma}\label{lem:v_not_perp}
  There are $r,r'\in R'$ and $s,s'\in S$ such that $(r,s)\neq (r',s')$ and
  \[
  \card{\inner{V_{r,s}}{V_{r',s'}}} \geq c\delta_{n-1,d_1'-1}^2.
  \]
\end{lemma}
\begin{proof}
  The proof is deferred to Section~\ref{sec:lem_V_not_perp}. We remark that this is the place in the argument in which the relaxed base case is used.
\end{proof}

We can now lower bound the variance of $p$.
\begin{lemma}\label{lem:lb_var_p}
  ${\sf var}(p)\ggg_{m}\tau^9\delta_{n-1,d_1'-1}^4$.
\end{lemma}
\begin{proof}
  By Lemma~\ref{lem:v_not_perp} there are $(r^{\star},s^{\star})\neq ({r^{\star}}',{s^{\star}}')$ such that
  $\card{\inner{V_{r^{\star},s^{\star}}}{V_{{r^{\star}}',{s^{\star}}'}}}\ggg_{m} c\delta_{n-1,d'-1}^2$, and we split into cases.
  \paragraph{The case $r^{\star}\neq {r^{\star}}'$, $s^{\star}\neq {s^{\star}}'$.}
  In this case we have by~\eqref{eq:delta_6} that
  \[
  {\sf var}(p) \geq \left(1-O(\zeta^{1/2})\right)\E[\card{p_3}^2]-O(\sqrt{\zeta}\delta_{n-1,d'-1}^4),
  \]
  and we lower bound $\E[\card{p_3}^2]$ by expanding. Indeed, it is equal to
  \[
  \sum\limits_{\substack{r\neq r'\in R', s\neq s'\in S'\\ r''\neq r'''\in R', s''\neq s'''\in S'}}
  \kappa_r\kappa_{r'}\kappa_{r''}\kappa_{r'''}\rho_s\rho_{s'}\rho_{s''}\rho_{s'''}
  \inner{V_{r',s}}{V_{r,s'}}
  \overline{\inner{V_{r''',s''}}{V_{r'',s'''}}}
  \Expect{}{\pi_r\overline{\pi_{r'}}\overline{\pi_{r''}}\pi_{r'''}\theta_s\overline{\theta_{s'}}\overline{\theta_{s''}}\theta_{s'''}}.
  \]
  The expectation is $0$ unless $r'=r'''$, $r = r''$ and $s' = s'''$, $s = s''$; this is because the distribution of each $\pi_r$ is invariant under multiplying
  each one of them by a random complex number of absolute value $1$. We thus get
  \begin{align*}
  \E[\card{p_3}^2]
  &=
  \sum\limits_{r\neq r'\in R', s\neq s'\in S'}
  \kappa_r^2\kappa_{r'}^2\rho_s^2\rho_{s'}^2
  \card{\inner{V_{r',s}}{V_{r,s'}}}^2
  \Expect{}{\card{\pi_r}^2\card{\pi_{r'}}^2\card{\theta_s}^2\card{\theta_{s'}}^2}\\
  &\geq\kappa_{r^{\star}}^2\kappa_{{r^{\star}}'}^2\rho_{s^{\star}}^2\rho_{{s^{\star}}'}^2
  \card{\inner{V_{{r^{\star}}',s^{\star}}}{V_{r^{\star},{s^{\star}}'}}}^2
  \Expect{}{\card{\pi_{r^{\star}}}^2\card{\pi_{{r^{\star}}'}}^2\card{\theta_{s^{\star}}}^2\card{\theta_{{s^{\star}}'}}^2}\\
  &\geq c\tau^8\delta_{n-1,d'-1}^4\Expect{}{\card{\pi_{r^{\star}}}^2\card{\pi_{{r^{\star}}'}}^2\card{\theta_{s^{\star}}}^2\card{\theta_{{s^{\star}}'}}^2}\\
  &=c\tau^8\delta_{n-1,d'-1}^4
  \Expect{}{\card{\pi_{r^{\star}}}^2\card{\pi_{{r^{\star}}'}}^2}
  \Expect{}{\card{\theta_{s^{\star}}}^2\card{\theta_{{s^{\star}}'}}^2}.
  \end{align*}
  We now argue that
  \[
  \Expect{}{\card{\pi_{r^{\star}}}^2\card{\pi_{{r^{\star}}'}}^2}, \Expect{}{\card{\theta_{s^{\star}}}^2\card{\theta_{{s^{\star}}'}}^2}\ggg_{m} 1
  \]
  We show the argument for the first expectation and the second one is similar.
  Looking at $(\kappa_r\pi_r)_{r\in R'}$, we see that this is a random vector in $\mathbb{C}^{\card{R'}}$ of $2$-norm equal to
  $1$, and so $\Expect{}{\card{\kappa_r\pi_r}^2\card{\kappa_{r'}\pi_{r'}}^2}\ggg_{m} 1$ for all $r,r'$. Hence,
  $\Expect{}{\card{\pi_{r^{\star}}}^2\card{\pi_{{r^{\star}}'}}^2}\ggg 1$.

  Overall, we get that ${\sf var}(p)\geq \left(\Omega_{c,m}(\tau^8) - \sqrt{\zeta}\right)\delta_{n-1,d'-1}^4$,
  hence ${\sf var}(p)\geq \tau^9 \delta_{n-1,d'-1}^4$.

  \paragraph{The case $r^{\star}= {r^{\star}}'$, $s^{\star} \neq {s^{\star}}'$.}
  Let $E$ be the event that
  \[
  \card{\kappa_{r^{\star}}\pi_{r^{\star}}}\geq \sqrt{M_1-\frac{c^{50}}{m^{50}}},
  \qquad
  \card{\rho_{s^\star}\theta_{s^{\star}}}\geq \frac{1}{\sqrt{2}}\sqrt{M_2-\frac{c^{50}}{m^{50}}},
  \qquad
  \card{\rho_{{s^\star}'}\theta_{{s^{\star}}'}}\geq \frac{1}{\sqrt{2}}-\frac{c^{50}}{m^{50}}.
  \]
  Then $\Prob{}{E}\ggg_{c,m} 1$, and if $E$ happens then for every $r\neq r^{\star}$, $s\neq s^{\star}, {s^{\star}}'$
  it holds that
  \[
  \card{\kappa_r\pi_r}
  \leq \sqrt{M_1-\card{\kappa_{r^{\star}}\pi_{r^{\star}}}^2}
  \leq \frac{c^{50}}{m^{50}},
  \qquad
  \card{\rho_s\theta_s}
  \leq
  \sqrt{M_2 - \card{\rho_{s^\star}\theta_{s^{\star}}}^2 - \card{\rho_{{s^\star}'}\theta_{{s^{\star}}'}}^2}
  \leq \frac{c^{50}}{m^{50}},
  \]
  and it follows that
  \[
  \card{p_1}
  \geq
  \card{\kappa_{r^{\star}}\pi_{r^{\star}}}^2
  \card{\rho_{s^\star}\theta_{s^{\star}}}
  \card{\rho_{{s^\star}'}\theta_{{s^{\star}}'}}
  \card{\inner{V_{r^{\star}, s^{\star}}}{V_{r^{\star}, {s^{\star}}'}}}
  -m^4\frac{c^{50}}{m^{50}}\max_{r,s}\norm{V_{r,s}}_2^2.
  \]
  By the choice of $r^{\star}, s^{\star}, {s^{\star}}'$ and Claim~\ref{claim:bound_norm_V}, the last expression
  is at least $\ggg_{c,m}\delta_{n-1, d-1}^2$. It follows that
  $\Expect{}{\card{p_1}^2}\ggg_{c,m}\Prob{}{E}\delta_{n-1, d-1}^4\ggg_{c,m} \delta_{n-1, d-1}^4$,
  and the proof is concluded by plugging this into~\eqref{eq:delta_6}.

  \paragraph{The case $r^{\star} \neq {r^{\star}}'$, $s^{\star} = {s^{\star}}'$.}
  This case is very similar to the case above, except that we work with $p_2$ instead of $p_1$. We omit the details.
\end{proof}

Combining Claim~\ref{claim:p_var_ub} and Lemma~\ref{lem:lb_var_p} gives contradiction to~\eqref{eq:hier2}. The contradiction means that
not all of~\eqref{eq:assume_delta_1},~\eqref{eq:assume_delta_2},~\eqref{eq:assume_delta_3} can hold, and therefore as explained above
the proof of Lemma~\ref{lem:delta_inductive} is concluded.

\subsubsection{Proof of Lemma~\ref{lem:v_not_perp}: the Dimensionality Argument}\label{sec:lem_V_not_perp}
Assume towards contradiction that Lemma~\ref{lem:v_not_perp} is false. We start with the following general lemma,
which quantifying the fact that $\ell$ vectors in an $\ell-1$-dimensional space cannot be all orthogonal.
\begin{fact}\label{fact:dimensionality}
  Suppose that $u_1,\ldots,u_{\ell}$ are unit vectors in $\mathbb{C}^{k}$ where $k\leq \ell-1$. Then
  there are $i,j$ such that $\card{\inner{u_i}{u_j}}\geq \frac{1}{k\ell}$.
\end{fact}
\begin{proof}
  Construct the Gram matrix $M\in \mathbb{C}^{\ell\times \ell}$ defined as $M[i,j] = \inner{u_i}{u_j}$.
  Then $M$ is Hermitian and therefore it has non-negative real eigenvalues $a_1,\ldots,a_{\ell}$. As the
  rank of $M$ is at most $k$, it follows that only $k$ of these eigenvalues can be non-zero, and without
  loss of generality these are $a_1,\ldots,a_k$. Thus,
  we have that
  \[
  \sum\limits_{i,j}\card{\inner{u_i}{u_j}}^2
  = {\sf tr}(MM^{*})
  = \sum\limits_{i=1}^{k} a_i^2
  \geq \frac{1}{k}\left(\sum\limits_{i=1}^{k} a_i\right)^2
  = \frac{1}{k}{\sf tr}(M)^2
  = \frac{\ell^2}{k}.
  \]
  The summands where $i=j$ contribute $\ell$ to the left hand side, and so
  \[
  \sum\limits_{i\neq j}\card{\inner{u_i}{u_j}}^2\geq \frac{\ell(\ell-k)}{k}\geq \frac{\ell}{k}.
  \]
  It follows that there are $i\neq j$ such that $\card{\inner{u_i}{u_j}}\geq \frac{1}{k\ell}$.
\end{proof}

Next, we show that $F_t'$ has almost all of its mass on $r\in R'$ and $s\in S'$.
\begin{claim}\label{claim:almost_full_norm}
  For all $t\in T'$ it holds that $\sum\limits_{r\in R', s\in S'}\card{\widehat{F}_t'(r,s)}^2\geq 1-\sqrt{\zeta}$.
\end{claim}
\begin{proof}
  First, note that
  \[
  \sum\limits_{t\in T'}\psi_t^2 = 1-\sum\limits_{t\in T\setminus T'}\psi_t^2\geq 1-m\eta^2\geq 1-\eta.
  \]
  Combining with~\eqref{eq:assume_delta_2} we get that
  \[
  \sum\limits_{t\in T'}\psi_t^2\left(1-\sum\limits_{r\in R', s\in S'}\card{\widehat{F}_t'(r,s)}^2\right)\leq \zeta+\eta\leq 2\zeta.
  \]
  As each summand on the left hand side is non-negative, it follows that for each $t\in T'$ it holds that
  \[
  \sum\limits_{r\in R', s\in S'}\card{\widehat{F}_t'(r,s)}^2
  \geq 1-\frac{2\zeta}{\psi_t^2}\geq 1-\sqrt{\zeta}.
  \qedhere
  \]
\end{proof}

We can now show that for each $r\in R'$, $s\in S'$, the function $g_r'h_s'$ is close to be a linear combination of
$F_t'$ for $t\in T'$. Formally:
\begin{claim}\label{claim:almost_full_norm2}
  For all $r\in R'$ and $s\in S'$ it holds that $\sum\limits_{t\in T'}\card{\widehat{F}_t'(r,s)}^2\geq 1-\zeta^{1/4}$.
\end{claim}
\begin{proof}
  Assume this is not the case, so that for some $r^{\star}$ and $s^{\star}$ we have that
  $\sum\limits_{t\in T'}\card{\widehat{F}_t'(r^{\star},s^{\star})}^2\leq 1-\zeta^{1/4}$;
  we remark that for any other $r,s$ this sum is at most $1$. Summing Claim~\ref{claim:almost_full_norm}
  over $t\in T'$ gives:
  \[
  (1-\sqrt{\zeta})\card{T'}
  \leq
  \sum\limits_{r\in R', s\in S'}\sum\limits_{t\in T'}\card{\widehat{F}_t'(r,s)}^2
  \leq (\card{R'}\card{S'} - 1) + 1-\zeta^{1/4}.
  \]
  Thus,
  \[
  \card{T'}
  \leq
  (\card{R'}\card{S'} - 1)(1+O(\sqrt{\zeta}))
  +(1-\zeta^{1/4})(1+O(\sqrt{\zeta}))
  \leq
  \card{R'}\card{S'} + O_{m}(\sqrt{\zeta}) - \frac{1}{2}\zeta^{1/4}
  <\card{R'}\card{S'},
  \]
  and as $\card{T'}$ is an integer it follows that $\card{T'}\leq \card{R'}\card{S'} - 1$. From Fact~\ref{fact:dimensionality}
  it follows that there are $(r,s)\neq (r',s')$ such that
  \[
  \card{\inner{\frac{V_{r,s}}{\norm{V_{r,s}}_2}}{\frac{V_{r',s'}}{\norm{V_{r',s'}}_2}}}\geq \frac{1}{m^2},
  \]
  and by Claim~\ref{claim:bound_norm_V} we have $\card{\inner{V_{r,s}}{V_{r',s'}}}\geq c\delta_{n-1,d'-1}^2$,
  and contradiction.
\end{proof}

Motivated by Claim~\ref{claim:almost_full_norm2}, we define
\[
F_{r,s}' = \frac{\sum\limits_{t\in T'}\overline{\widehat{F_t}'(r,s)} F_t'}{\sum\limits_{t\in T'}\card{\widehat{F_t}'(r,s)}^2}.
\]
Thus, $F_{r,s}'$ has $2$-norm $1$ and by Claim~\ref{claim:almost_full_norm2}
\[
\Expect{(x,y,z)\sim \mu}{F_{r,s}'(y,z)g_r'(y)h_s'(z)}
=\sqrt{\sum\limits_{t\in T'}\card{\widehat{F_t}'(r,s)}^2}
\geq 1-\zeta^{1/4}.
\]
It follows that
\begin{equation}\label{eq:delta_7}
\norm{F_{r,s}' - g_r'h_s'}_2^2\leq 2\zeta^{1/4}.
\end{equation}
Also, as each $F_t'$ is constant on connected components it follows that
each $F_{r,s}'$ is constant on connected components and so it can be written as $F_{r,s}' = W f_{r,s}'$ for $f_{r,s}\colon \Sigma\to\mathbb{C}$.
Thus, as
\[
\card{\Expect{(x,y,z)\sim \mu}{f_{r,s}'(x)g_r'(y)h_s'(z)}}
\geq
\card{\Expect{(x,y,z)\sim \mu}{F_{r,s}'(y,z)g_r'(y)h_s'(z)}}
\geq 1-\zeta^{1/4}
\]
it follows by the relaxed base case (Definition~\ref{def:relaxed_base}) that
${\sf var}_{\Sigma_{{\sf modest}}}(f_{r,s}')\leq \zeta^{c}$ for all $r\in R'$ and $s\in S'$.
We are now going to use Claim~\ref{claim:almost_full_norm} and~\eqref{eq:delta_7} to argue
that each $F_t'$ is close to a linear combination of $F_{r,s}'$, and hence has small variance
on $\Sigma_{{\sf modest}}$.

\begin{claim}\label{claim:use_relaxed_base}
  For all $t\in T'$, write $F_t' = W f_t'$ for $f_t'\colon \Sigma\to\mathbb{C}$.
   Then for all $t\in T'$ we have
   \[
   {\sf var}_{\Sigma_{{\sf modest}}}(f_{t}')\lll_{m} \zeta^{c/2}.
   \]
\end{claim}
\begin{proof}
  By Claim~\ref{claim:almost_full_norm2}, taking $\tilde{F}_t' = \sum\limits_{r\in R', s\in S'}\overline{\widehat{F}_t'(r,s)}g_r'h_s'$
  we have
  \[
  \Expect{(x,y,z)\sim \mu}{\tilde{F}_t'(y,z) F_t'(y,z)}
  =\sqrt{\sum\limits_{r\in R', s\in S'}\card{\widehat{F}_t'(r,s)}^2}
  \geq 1-\sqrt{\zeta},
  \]
  hence $\norm{\tilde{F}_t' - F_t'}_2\lll \zeta^{1/4}$.
  Using~\eqref{eq:delta_7} we get that
  \[
  \norm{\tilde{F}_t' - \sum\limits_{r\in R', s\in S'}\overline{\widehat{F}_t'(r,s)}F_{r,s}'}_2
  \lll_{m}\zeta^{1/4},
  \]
  so by the triangle inequality
  \[
  \norm{F_t' - \sum\limits_{r\in R', s\in S'}\overline{\widehat{F}_t'(r,s)}F_{r,s}'}_2
  \lll_{m}\zeta^{1/4}.
  \]
  From the above it follows that
  \[
  \norm{f_t' - \sum\limits_{r\in R', s\in S'}\overline{\widehat{F}_t'(r,s)}f_{r,s}'}_2
  \lll_{m}\zeta^{1/4},
  \]
  and as ${\sf var}_{\Sigma_{{\sf modest}}}(f_{r,s}')\leq \zeta^{c}$ for all $r\in R'$, $s\in S'$
  it follows that ${\sf var}_{\Sigma_{{\sf modest}}}(f_{t}')\leq \zeta^{c/2}$, as desired.
\end{proof}

By Claim~\ref{claim:use_relaxed_base} we get that
\[
\sum\limits_{t\in T'}\psi_t^2 {\sf var}_{\Sigma_{{\sf modest}}}(f_{t}')
\leq \zeta^{c/2},
\]
which contradicts the choice of the partition $I,J$ as in~\eqref{eq:choose_j_in_delta} and~\eqref{eq:hier2}.
This contradiction completes the proof of Lemma~\ref{lem:v_not_perp}.\qed

\subsection{Finishing the Proof of Theorem~\ref{thm:nonembed_homogenous}}\label{sec:finish_pf_beta_delta}
We are now ready to prove Theorem~\ref{thm:nonembed_homogenous}. We will take the parameters
\[
0\ll \xi \ll C^{-1}\ll L^{-1}\ll \alpha, m^{-1}\leq 1.
\]
Let $n\geq d\geq d'$ and $\mu$ be as in Theorem~\ref{thm:nonembed_homogenous}. Applying Lemma~\ref{lem:beta_bound_iterated},
if the first item holds we are done, so assume otherwise. By the second item we get that there are $n'$, $e$ and $e'$
satisfying $e\geq d/2$, $n'\leq L e$ and $e'\geq d'/2$ such that
$\beta_{n,d,d'}\leq \beta_{n',e,e'}$. Clearly, $\beta_{n',e,e'}\leq \delta_{n',e,e'}$, and we now use
Lemma~\ref{lem:delta_inductive}. Now that $e'\geq \frac{e}{e^{\xi}}\geq \frac{n}{L e^{\xi}}$, and applying Lemma~\ref{lem:delta_inductive}
$e'/2$ times we get that
\[
\delta_{n',e,e'}
\leq \left(1-\frac{1}{L^{C}e^{\xi C}}\right)\delta_{n'-1,e-1,e'-1}
\leq\ldots\leq
\left(1-\frac{1}{2L^{C}e^{\xi C}}\right)^{e'/2}\delta_{n'-e'/2,e-e'/2,e'-e'/2},
\]
which is at most $\left(1-\frac{1}{2L^{C}e^{\xi C}}\right)^{e'/2}$ by the trivial bound $\delta_{n'-e'/2,e-e'/2,e'-e'/2}\leq 1$.
It follows that
\[
\delta_{n',e,e'}
\leq 2^{-\Omega\left(\frac{e'/2}{2L^{C}e^{\xi C}}\right)}
\leq 2^{-\Omega_{m,\alpha}(e'^{1-\sqrt{\xi}})},
\]
concluding the proof.

\subsection{Finishing the Proof of Theorem~\ref{thm:nonembed_deg_must_be_small}: the Chain of Implications}
We now note that Theorem~\ref{thm:nonembed_deg_must_be_small} has been proved, and below is the chain of implications:
\begin{enumerate}
  \item We completed the proof of Theorem~\ref{thm:nonembed_homogenous} in Section~\ref{sec:finish_pf_beta_delta}.
  \item By Section~\ref{sec:nonembed_hom_implies} it follows that Theorem~\ref{thm:nonembed_homogenous'} is true.
  \item By Claim~\ref{claim:beta_thm_implies} (namely, the content of Section~\ref{sec:truncate})
  it follows that Theorem~\ref{thm:nonembed_deg_must_be_small_rephrase_maximal_relaxed} is true.
  \item By Section~\ref{sec:relaxed_implies} it follows that Theorem~\ref{thm:nonembed_deg_must_be_small_rephrase_maximal2} is true.
  \item By Section~\ref{sec:rephrase_implies} it follows that Theorem~\ref{thm:nonembed_deg_must_be_small} is true.
\end{enumerate}

\section{A Structural Results for $f$ After Path Tricks: Proof of Theorem~\ref{thm:main_stab_3_saturated}}\label{sec:the_hastad_argument}
In this section, we prove Theorem~\ref{thm:main_stab_3_saturated}. We start by giving a high level overview of the argument.
Given a distribution $\mu$ and functions $f\colon\Sigma^n\to\mathbb{C}$, $g\colon\Gamma^n\to\mathbb{C}$ and $h\colon\Phi^n\to\mathbb{C}$
as in the statement of Theorem~\ref{thm:main_stab_3_saturated}, we perform the following steps:
\begin{enumerate}
  \item
  \textbf{Softly truncating the high non-embedding parts:}  First, consider the high non-embedding degree parts of $f$, $g$ and $h$, namely
  $(I - \mathrm{T}_{\text{non-embed}, 1-\delta})f$, $(I - \mathrm{T}_{\text{non-embed}, 1-\delta})g$
  and $(I - \mathrm{T}_{\text{non-embed}, 1-\delta})h$. We show that they give very little contribution
  to the $3$-wise correlation, and hence conclude that
  \begin{equation}\label{eq:hastad_argument_1}
  \card{
  \Expect{(x,y,z)\sim \mu^{\otimes n}}{\mathrm{T}_{\text{non-embed}, 1-\delta}f(x) \mathrm{T}_{\text{non-embed}, 1-\delta}g(y) \mathrm{T}_{\text{non-embed}, 1-\delta}h(z)}
  }\geq \frac{\eps}{2}.
  \end{equation}
  Set $f' = \mathrm{T}_{\text{non-embed}, 1-\delta}f$, $g' = \mathrm{T}_{\text{non-embed}, 1-\delta}g$ and
  $h' = \mathrm{T}_{\text{non-embed}, 1-\delta}h$, so that almost all of the mass of these functions lies on monomials
  with non-embedding degree which is constant.

  \item \textbf{Reducing the non-embedding degree to be $0$:}
  By random restriction, we show that with noticeable probability,
  after a random restriction almost all of the mass of $f'$, $g'$ and $h'$ collapses to non-embedding degrees $0$
  and~\eqref{eq:hastad_argument_1} continues to hold (with some loss in parameters). We then use an averaging operator
  that annihilates all of the mass of $f'$, $g'$ and $h'$ of non-zero non-embedding degree, apply it on $f'$, $g'$ and $h'$ and show that~\eqref{eq:hastad_argument_1}
  continues to hold (with an additional loss in parameters). We perform these steps (as opposed to harsher truncation-style operations) so as to arrive at functions
  that have embedding degree which is $0$ which are also bounded.

  \item \textbf{Shifting to the uniform distribution over an equation on $H$:}
  By another random restriction argument we switch from the distribution $\mu$ to a distribution $\nu$, in which sampling
  $(x,y,z)\sim \nu$, the distribution of $(\sigma(x),\gamma(y),\phi(z))$ is uniform over $\sett{(a,b,c)\in H^3}{a+b+c = 0}$.
  We argue that moving to the distribution $\nu$,~\eqref{eq:hastad_argument_1}
  continues to hold for the restrictions (again, with some loss in parameters).
  Once we have shifted the distribution to be $\nu$, we can use a standard Fourier-analytic computation
  showing that the restriction we chose for $f'$ has a significant Fourier coefficient.
  In our case, this last assertion translates to the fact that the restriction of $f'$ is correlated with a function of the form $\chi\circ \sigma$ for some character $\chi$ over $H$.

  \item \textbf{Unraveling restrictions and averaging:}
  The last part in the argument is to invoke a result that we refer to as the restriction inverse theorem. Morally speaking,
  up until now we have shown that after random restriction, with noticeable probability $f'$ is correlated with a function of the form $\chi\circ \sigma$.
  Now we would like to use this information to conclude a correlation result for the function $f'$ itself. The restriction
  inverse theorem is a result asserting precisely that: for restrictions of a function $f'$ to be correlated with a function
  of the form $\chi\circ \sigma$ with noticeable probability, it must be the case that $f'$ itself is already correlated with
  a function of the form $L\cdot \chi\circ \sigma$, where $L$ is a low-degree function.
\end{enumerate}

The restriction inverse theorem will be the subject of discussion in Section~\ref{sec:rest_inverse}, and for now we give a special case of it (which is
the version necessary for the current argument).
\begin{thm}\label{thm:restriction_inverse_special}
  For all $\alpha>0$, $r, m\in\mathbb{N}$, $\rho\in (0,1)$ and $\eps,\delta>0$ there are $d\in\mathbb{N}$ and $\eps'>0$ such that the following holds.
  Suppose that $\Sigma$ is an alphabet of size at most $m$ and $\mu$ is a distribution over $\Sigma$ in which the probability of each atom is at least
  $\alpha$. Suppose that $(H,+)$ is an Abelian group of size at most $r$ and $\sigma\colon \Sigma\to H$ is a map that has $0$ in its image.

  Suppose that $\mu = \rho \nu + (1-\rho)\nu'$ for distributions $\nu,\nu'$ in which the probability of each atom is at least $\alpha$.
  If $f\colon(\Sigma^n,\mu^{\otimes n})\to\mathbb{C}$ is a $1$-bounded function such that
  \[
  \Prob{I\subseteq_{\rho}[n], \tilde{x}\sim{\nu'}^{\overline{I}}}{\exists \chi\in \hat{H}^{I},
  \card{\inner{f_{\overline{I}\rightarrow \tilde{x}}}{\chi\circ\sigma}_{L_2(\Sigma^{I}; \nu^{I})}}\geq \eps}\geq \delta,
  \]
  then there exists $\chi'\in \hat{H}^{n}$ and $L\colon\Sigma^n\to\mathbb{C}$ a function of degree at most $d$ and $\norm{L}_2\leq 1$ such that
  \[
    \card{\inner{f}{L \chi'\circ\sigma}}\geq \eps'.
  \]
  Quantitatively,
  $d = {\sf poly}_{m,\alpha}\left(\frac{\log(1/\eps)\log(1/\delta)}{\rho}\right)$ and
  $\eps' = 2^{-{\sf poly}_{m,\alpha}\left(\frac{\log(1/\eps)\log(1/\delta)}{\rho}\right)}$.
\end{thm}
\begin{proof}
  Deferred to Section~\ref{sec:rest_inverse}.
\end{proof}

\subsection{Softly Truncating the High Non-embedding Parts: Applying Theorem~\ref{thm:nonembed_deg_must_be_small}}
We first need a slightly stronger variant of Theorem~\ref{thm:nonembed_deg_must_be_small} (which nevertheless follows from it
almost immediately), in which the condition that the support of $\mu_{y,z}$ is full is omitted.
\begin{thm}\label{thm:nonembed_deg_must_be_small_lesscond}
  For all $m\in\mathbb{N}$, $\alpha>0$ there are $M\in\mathbb{N}$, $\delta_0>0$ and $\eta>0$ such that the following holds for all $0< \delta\leq \delta_0$.
  Suppose that $\mu$ is a distribution over $\Sigma\times \Gamma\times \Phi$ satisfying:
  \begin{enumerate}
    \item The probability of each atom is at least $\alpha$.
    \item The size of each one of $\Sigma,\Gamma,\Phi$ is at most $m$.
    \item ${\sf supp}(\mu)$ is pairwise connected.
    \item There are master embeddings $\sigma,\gamma,\phi$ for $\mu$ into an Abelian group $(H,+)$ that are saturated,
    and the distribution of $(\sigma(x),\gamma(y),\phi(z))$ where $(x,y,z)\sim \mu$ has full support on $\{(a,b,c)\in H^3~|~a+b+c = 0\}$.
  \end{enumerate}
  Then, if $f\colon\Sigma^n\to \mathbb{C}$, $g\colon\Gamma^n\to \mathbb{C}$ and $h\colon\Phi^n\to \mathbb{C}$ are $1$-bounded functions such that
  ${\sf NEStab}_{1-\delta}(g;\mu_y^{\otimes n})\leq \delta$,
  then
  \[
  \card{\Expect{(x,y,z)\sim \mu^{\otimes n}}{f(x)g(y)h(z)}}\leq M\delta^{\eta}.
  \]
\end{thm}
\begin{proof}
    Take $t$ to be sufficiently large constant. By Lemma~\ref{lem:from_mu_to_path} we have that
    \[
    \card{\Expect{(x,y,z)\sim \mu^{\otimes n}}{f(x)g(y)h(z)}}^{2^t}
    \leq
    \card{\Expect{(x,y,z)\sim \mu_{2^t}^{\otimes n}}{F(x)g(y)h(z)}},
    \]
    where $\mu_{2^t}$ is the distribution $\mu$ after applying the path trick $t$ times.
    By Lemma~\ref{lem:master_stays}, we get that the distribution $\mu_{2^t}$ satisfies the fourth item in Theorem~\ref{thm:nonembed_deg_must_be_small}, and it is clear that it satisfies the first and second items
    therein as well (with a different $m'$ and $\alpha'$ that only depend on $m,\alpha$ and $t$). Finally,
    by Lemma~\ref{lem:path_keeps_connectedness} the marginal distribution of $\mu_{2^t}$ on $y,z$ has full
    support provided that $t$ is large enough, and
    \[
    {\sf NEStab}_{1-\delta}(g;{\mu_{2^t}}_y^{\otimes n})
    ={\sf NEStab}_{1-\delta}(g;{\mu}_y^{\otimes n})
    \leq \delta
    \]
    as the distributions $\mu_{2^t}$ and $\mu$ have the same marginal distribution over $y$.
    The proof is concluded by appealing to Theorem~\ref{thm:nonembed_deg_must_be_small}.
\end{proof}

The following lemma applies Theorem~\ref{thm:nonembed_deg_must_be_small_lesscond} along with a standard replacement argument to
truncate the high non-embedding degree parts of the functions $f,g$ and $h$ as explained above.
\begin{lemma}\label{lem:truncate_high_nonembed}
  For all $m\in\mathbb{N}$ and $\alpha>0$ there is $\xi>0$ such that the following holds for sufficiently small $\eps>0$.
  Let $\mu$ be a distribution as in Theorem~\ref{thm:main_stab_3_saturated}, and let
  $f\colon\Sigma^{n}\to\mathbb{C}$, $g\colon\Gamma^{n}\to\mathbb{C}$, $h\colon\Phi^n\to\mathbb{C}$
  be $1$-bounded functions such that
  \[
        \card{\Expect{(x,y,z)\sim \mu^{\otimes n}}{f(x)g(y)h(z)}}\geq \eps.
  \]
  Then for $\delta = (\eps/20)^{3/\xi}$ we have
  \[
        \card{\Expect{(x,y,z)\sim \mu^{\otimes n}}
        {\mathrm{T}_{\text{non-embed}, 1-\delta}f(x) \mathrm{T}_{\text{non-embed}, 1-\delta}g(y) \mathrm{T}_{\text{non-embed}, 1-\delta}h(z)}}
        \geq \frac{\eps}{2}.
  \]
\end{lemma}
\begin{proof}
  Denoting $f' = \mathrm{T}_{\text{non-embed}, 1-\delta}f$, $g' = \mathrm{T}_{\text{non-embed}, 1-\delta}g$ and
  $h' = \mathrm{T}_{\text{non-embed}, 1-\delta}h$, we which $f$, $g$ and $h$ to $f'$, $g'$ and $h'$ one step at
  a time. First, we claim that provided that $\xi$ is small enough it holds that
  \[
    \card{\Expect{(x,y,z)\sim \mu^{\otimes n}}{(f-f')(x)g(y)h(z)}}\leq \frac{\eps}{10}.
  \]
  Indeed, defining $f'' = \frac{1}{2}(I-\mathrm{T}_{\text{non-embed}, 1-\delta})f$, we get that $f''$ is $1$-bounded
  and
  \[
  {\sf NEStab}_{1-\delta^{1/3}}(f'')
  \leq \norm{\mathrm{T}_{\text{non-embed}, 1-\delta^{1/3}}(I-\mathrm{T}_{\text{non-embed}, 1-\delta})f}_2,
  \]
  which is at most the largest eigenvalue of $\mathrm{T}_{\text{non-embed}, 1-\delta^{1/3}}(I-\mathrm{T}_{\text{non-embed}, 1-\delta})$.
  This eigenvalue can be easily seen to be at most $\delta^{1/3}$, so ${\sf NEStab}_{1-\delta^{1/3}}(f'')\leq \delta^{1/3}$.
  Applying Theorem~\ref{thm:nonembed_deg_must_be_small_lesscond} we get that
  \[
    \card{\Expect{(x,y,z)\sim \mu^{\otimes n}}{f''(x)g(y)h(z)}}\leq \delta^{\xi/3}\leq \frac{\eps}{20},
  \]
  Thus $\card{\Expect{(x,y,z)\sim \mu^{\otimes n}}{(f-f')(x)g(y)h(z)}}\leq \frac{\eps}{10}$ and by the triangle
  inequality
  \[
  \card{\Expect{(x,y,z)\sim \mu^{\otimes n}}{f'(x)g(y)h(z)}}\geq \frac{9\eps}{10}.
  \]
  Continuing in this way, we apply the same argument to replace $g$ into $g'$ and $h$ into $h'$ to get that
  \[
  \card{\Expect{(x,y,z)\sim \mu^{\otimes n}}{f'(x)g'(y)h'(z)}}\geq \frac{7\eps}{10},
  \]
  and the proof is concluded.
\end{proof}

\subsection{Reducing the Non-embedding Degree to be $0$}
Let $f,g,h$ be as in the setting of Lemma~\ref{lem:truncate_high_nonembed}, and take
\[
f' = \mathrm{T}_{\text{non-embed}, 1-\delta}f,
\qquad
g' = \mathrm{T}_{\text{non-embed}, 1-\delta}g,
\qquad
h' = \mathrm{T}_{\text{non-embed}, 1-\delta}h.
\]
Choose $I\subseteq_{\kappa} [n]$ where $\kappa = \frac{\eps^3\delta}{100^3}$, sample
$(\tilde{x},\tilde{y},\tilde{z})\sim \mu^{\overline{I}}$ and define
\[
f'' = f'_{\overline{I}\rightarrow \tilde{x}},
\qquad
g'' = g'_{\overline{I}\rightarrow \tilde{y}},
\qquad
h'' = h'_{\overline{I}\rightarrow \tilde{z}}.
\]
Also, let
\[
f^{\sharp} = \mathrm{T}_{\text{non-embed}, 0}f'',
\qquad
g^{\sharp} = \mathrm{T}_{\text{non-embed}, 0}g'',
\qquad
h^{\sharp} = \mathrm{T}_{\text{non-embed}, 0}h''
\]
in other words, $f^{\sharp}$ is the part of $f''$ of non-embedding degree $0$,
$g^{\sharp}$ is the part of $g''$ of non-embedding degree $0$ and $h^{\sharp}$ is the part of $h''$ of non-embedding degree $0$.

\begin{claim}\label{claim:hastad_arg_rr1}
  In the above setting,
  \[
  \Expect{I,\tilde{x}}{\norm{f'' - f^{\sharp}}_2^2}\leq \frac{\kappa}{\delta},
  \qquad
  \Expect{I,\tilde{y}}{\norm{g'' - g^{\sharp}}_2^2}\leq \frac{\kappa}{\delta},
  \qquad
  \Expect{I,\tilde{z}}{\norm{h'' - h^{\sharp}}_2^2}\leq \frac{\kappa}{\delta}.
  \]
\end{claim}
\begin{proof}
  We prove the first inequality, and the other two are analogous.

  For a fixed $I$ and $\tilde{x}$, the norm $\norm{f'' - f^{\sharp}}_2^2$ is the mass of $f''$ on monomials of
  non-embedding degree at least $1$. Expanding
  \[
  f'(x) = \sum\limits_{v\in (B_{{\sf embed}}\cup B_{{\sf non-embed}}\cup B_{{\sf modest}})^{n}}\widehat{f'}(v)v(x),
  \]
  we have that
  \[
  (f'' - f^{\sharp})(x)
  =\sum\limits_{u\in (B_{{\sf embed}}\cup B_{{\sf non-embed}}\cup B_{{\sf modest}})^{I}}\left(\sum\limits_{v\in (B_{{\sf embed}}\cup B_{{\sf non-embed}}\cup B_{{\sf modest}})^{\overline{I}}}
  \widehat{f'}(v,u)v(\tilde{x})\right)u(x).
  \]
  Thus, for fixed $I$ we have
  \begin{align*}
  \Expect{\tilde{x}}{\norm{f'' - f^{\sharp}}_2^2}
  &=\sum\limits_{u\in (B_{{\sf embed}}\cup B_{{\sf non-embed}}\cup B_{{\sf modest}})^{I}}
  \Expect{\tilde{x}}{\left(\sum\limits_{v}\widehat{f'}(u,v)v(\tilde{x})\right)^2}1_{{\sf nedeg}(u)\geq 1}\\
  &=\sum\limits_{u\in (B_{{\sf embed}}\cup B_{{\sf non-embed}}\cup B_{{\sf modest}})^{I}}
  \sum\limits_{v}\widehat{f'}(v,u)^21_{{\sf nedeg}(u)\geq 1}\\
  &=\sum\limits_{v\in (B_{{\sf embed}}\cup B_{{\sf non-embed}}\cup B_{{\sf modest}})^{n}}
  \widehat{f'}(v)^21_{{\sf nedeg}(v|_{I})\geq 1}\\
  &=\sum\limits_{v\in (B_{{\sf embed}}\cup B_{{\sf non-embed}}\cup B_{{\sf modest}})^{n}}
  \widehat{f}(v)^2(1-\delta)^{2{\sf nedeg}(v)}1_{{\sf nedeg}(v|_{I})\geq 1}
  \end{align*}
  Taking expectation over $I$ gives
  \begin{align*}
  \Expect{I, \tilde{x}}{\norm{f'' - f^{\sharp}}_2^2}
  &\leq
  \sum\limits_{v\in (B_{{\sf embed}}\cup B_{{\sf non-embed}}\cup B_{{\sf modest}})^{n}}
  \widehat{f}(v)^2(1-\delta)^{2{\sf nedeg}(v)}\Expect{I}{1_{{\sf nedeg}(v|_{I})\geq 1}}\\
  &\leq
  \sum\limits_{v\in (B_{{\sf embed}}\cup B_{{\sf non-embed}}\cup B_{{\sf modest}})^{n}}
  \widehat{f}(v)^2(1-\delta)^{2{\sf nedeg}(v)}\kappa\cdot {\sf nedeg}(v),
  \end{align*}
  where we used the union bound. Looking at the function $P\colon [0,\infty)\to [0,\infty)$ defined by
  $P(s) = \kappa s(1-\delta)^{2s}$, its maximum is achieved at $s = \frac{1}{2\log(1/(1-\delta))}\leq \frac{1}{\delta}$
  and hence is at most $\frac{\kappa}{\delta}$.
\end{proof}

Next we analyze the triple correlation of $f''$, $g''$ and $h''$ and show that with noticeable probability it remains significant.
\begin{claim}\label{claim:Abelian_arg_2}
  In the above setting
  \[
        \Prob{I, \tilde{x}, \tilde{y}, \tilde{z}}{\card{\Expect{(x,y,z)\sim\mu^{I}}{f''(x)g''(y)h''(z)}}\geq \frac{\eps}{4}}\geq \frac{\eps}{4}.
  \]
\end{claim}
\begin{proof}
  Consider the real valued random variable $V(I,\tilde{x},\tilde{y},\tilde{z}) = \card{\Expect{(x,y,z)\sim\mu^{I}}{f''(x)g''(y)h''(z)}}$.
  Then by the triangle inequality
  \[
  \Expect{I,\tilde{x},\tilde{y},\tilde{z}}{\card{V}}
  \geq\card{\Expect{I,\tilde{x},\tilde{y},\tilde{z}}{\Expect{(x,y,z)\sim\mu^{I}}{f''(x)g''(y)h''(z)}}}
  = \card{\Expect{(x,y,z)\sim\mu^{n}}{f'(x)g'(y)h'(z)}}
  \geq \frac{\eps}{2},
  \]
  and as $0\leq V\leq 1$ it follows by an averaging argument that $\Prob{}{V\geq \eps/4}\geq \eps/4$, as required.
\end{proof}

We now combine the last two claims to conclude that the triple correlation of $f^{\sharp}$, $g^{\sharp}$ and $h^{\sharp}$ remains significant with noticeable probability.
\begin{lemma}\label{lem:move_to_gp_functions}
  In the above setting
  \[
        \Prob{I, \tilde{x}, \tilde{y}, \tilde{z}}{\card{\Expect{(x,y,z)\sim\mu^{I}}{f^{\sharp}(x)g^{\sharp}(y)h^{\sharp}(z)}}\geq \frac{\eps}{8}}\geq \frac{\eps}{8}.
  \]
\end{lemma}
\begin{proof}
  Let $E_1$ be the event that $\card{\Expect{(x,y,z)\sim\mu^{I}}{f''(x)g''(y)h''(z)}}\geq \frac{\eps}{4}$
  and let $E_2$ be the event that each one of
  $\norm{f'' - f^{\sharp}}_2^2$, $\norm{g'' - h^{\sharp}}_2^2$ and $\norm{g'' - h^{\sharp}}_2^2$ is at most $\frac{100\kappa}{\eps\delta}$.
  By the Union bound, Claim~\ref{claim:hastad_arg_rr1} and Markov's inequality we have that
  \[
  \Prob{}{E_2}
  =1-\Prob{}{\overline{E_2}}
  \geq 1-3\frac{\eps}{100},
  \]
  so $\Prob{}{E_1\cap E_2}\geq \frac{\eps}{4} - \frac{3\eps}{100}\geq \frac{\eps}{8}$. We argue that when $E_1\cap E_2$ holds, we have that
  \[
  \card{\Expect{(x,y,z)\sim\mu^{I}}{f^{\sharp}(x)g^{\sharp}(y )h^{\sharp}(z)}}\geq \frac{\eps}{8}.
  \]
  To do so, we start with $\card{\Expect{(x,y,z)\sim\mu^{I}}{f''(x)g''(y)h''(z)}}\geq \frac{\eps}{4}$ and replace $f''$, $g''$ and $h''$
  by $f^{\sharp}$, $g^{\sharp}$ and $h^{\sharp}$ one step at a time. Note that by the triangle inequality and $1$-boundedness
  \[
  \card{\Expect{(x,y,z)\sim\mu^{I}}{(f^{\sharp} - f'')(x)g''(y)h''(z)}}
  \leq \norm{f^{\sharp} - f''}_1\leq\norm{f^{\sharp} - f''}_2
  \leq \sqrt{\frac{100\kappa}{\eps\delta}}
  \leq \frac{\eps}{100},
  \]
  hence $\card{\Expect{(x,y,z)\sim\mu^{I}}{f^{\sharp}(x)g''(y)h''(z)}}\geq \frac{\eps}{4} - \frac{\eps}{100}$. Continuing in this way
  we get that
  \[
  \card{\Expect{(x,y,z)\sim\mu^{I}}{f^{\sharp}(x)g^{\sharp}(y )h^{\sharp}(z)}}\geq \frac{\eps}{4} - 3\frac{\eps}{100}\geq \frac{\eps}{8},
  \]
  as required.
\end{proof}

The functions $f^{\sharp}, g^{\sharp}$ and $h^{\sharp}$ are now only functions of $\sigma(x), \gamma(y)$ and $\phi(z)$, and abusing
notations we will also think of them as functions defined over $H^{I}$. As the distribution of $(\sigma(x),\gamma(y),\phi(z))$
is fully supported on $\sett{(a,b,c)\in H^3}{a+b+c=0}$ we expect the $3$-wise correlation
\[
\Expect{(x,y,z)\sim\mu^{I}}{f^{\sharp}(x)g^{\sharp}(y )h^{\sharp}(z)}
=\Expect{(x,y,z)\sim\mu^{I}}{f^{\sharp}(\sigma(x))g^{\sharp}(\gamma(y))h^{\sharp}(\phi(z))}
\]
to be related to the correlations of $f^{\sharp}$, $g^{\sharp}$ and $h^{\sharp}$. Indeed, examples for this include Roth's theorem~\cite{Roth,Meshulam}
as well as the analysis of the Blum-Luby-Rubinfeld linearity test~\cite{BLR,BCHKS,Has01}. However,
as the distribution of $(\sigma(x),\gamma(y),\phi(z))$ is not uniform over $\sett{(a,b,c)\in H^3}{a+b+c=0}$
our case is more closely related to the analysis of biased versions of this test and
the straightforward Fourier-analytic argument does not work. This point was part of our motivation in~\cite{BKMcsp3},
and the argument below generalizes that argument to our setting.

\subsection{Shifting to the Uniform Distribution over an Equation on $H$}
Fix $I$ and $\tilde{x},\tilde{y}$ and $\tilde{z}$ for which the event in Lemma~\ref{lem:move_to_gp_functions} holds.

We may write $\mu = \frac{\alpha}{2} U + \left(1-\frac{\alpha}{2}\right)\mu'$ where $U,\mu'$ are distributions over $\Sigma\times \Gamma\times \Phi$,
and sampling $(x,y,z)\sim U$, the distribution of $(\sigma(x),\gamma(y),\phi(z))$ is uniform over $\sett{(a,b,c)\in H^3}{a+b+c=0}$. We then take $J\subseteq_{\alpha/2} I$,
$(\tilde{x}',\tilde{y}',\tilde{z}')\sim \mu'^{I\setminus J}$ and set
\[
{f^{\sharp}}' = f^{\sharp}_{I\setminus J\rightarrow \tilde{x}'},
\qquad
{g^{\sharp}}' = g^{\sharp}_{I\setminus J\rightarrow \tilde{x}'},
\qquad
{h^{\sharp}}' = h^{\sharp}_{I\setminus J\rightarrow \tilde{x}'}.
\]
The following claim asserts that the triple correlation remains large with noticeable probability after passing to ${f^{\sharp}}'$,
${g^{\sharp}}'$ and ${h^{\sharp}}'$.
\begin{claim}\label{claim:move_to_gp_functions_shift}
  In the above setting
  \[
        \Prob{J, \tilde{x}', \tilde{y}', \tilde{z}'}
        {\card{\Expect{(x,y,z)\sim U^{J}}{{f^{\sharp}}'(x){g^{\sharp}}'(y){h^{\sharp}}'(z)}}\geq \frac{\eps}{16}}\geq \frac{\eps}{16}.
  \]
\end{claim}
\begin{proof}
  Consider the real valued random variable $V(J,\tilde{x}',\tilde{y}',\tilde{z}') = \card{\Expect{(x,y,z)\sim U^{J}}{{f^{\sharp}}'(x){g^{\sharp}}'(y){h^{\sharp}}'(z)}}$.
  Then by the triangle inequality
  \[
  \Expect{J,\tilde{x}',\tilde{y}',\tilde{z}'}{\card{V}}
  \geq\card{\Expect{J,\tilde{x}',\tilde{y}',\tilde{z}'}{\Expect{(x,y,z)\sim U^{J}}{{f^{\sharp}}'(x){g^{\sharp}}'(y){h^{\sharp}}'(z)}}}
  = \card{\Expect{(x,y,z)\sim\mu^{I}}{{f^{\sharp}}(x){g^{\sharp}}(y){h^{\sharp}}(z)}}
  \geq \frac{\eps}{8},
  \]
  and as $0\leq V\leq 1$ it follows by an averaging argument that $\Prob{}{V\geq \eps/16}\geq \eps/16$, as required.
\end{proof}

Next, we have the standard Fourier analytic computation that handles $3$-wise over $U$, showing that they can be significant only when
${f^{\sharp}}'$ has a significant Fourier coefficient.
\begin{claim}\label{claim:move_to_gp_functions_shift_fourier}
  If $J$, $\tilde{x}'$, $\tilde{y}'$ and $\tilde{z}'$ satisfy the event in Claim~\ref{claim:move_to_gp_functions_shift}, then
  there exists $\chi\in\hat{H}^{J}$ such that
  \[
        \card{\widehat{{f^{\sharp}}'}(\chi)}\geq \frac{\eps}{16}.
  \]
\end{claim}
\begin{proof}
  Expand ${f^{\sharp}}'$, ${g^{\sharp}}'$ and ${h^{\sharp}}'$ into Fourier basis over $H$:
  \[
  {f^{\sharp}}'(x) = \sum\limits_{\chi\in\hat{H}^{J}}\widehat{{f^{\sharp}}'}(\chi)\chi(\sigma(x)),
  \qquad
  {g^{\sharp}}'(y) = \sum\limits_{\chi'\in\hat{H}^{J}}\widehat{{g^{\sharp}}'}(\chi')\chi'(\gamma(y)),
  \qquad
  {h^{\sharp}}'(z) = \sum\limits_{\chi''\in\hat{H}^{J}}\widehat{{h^{\sharp}}'}(\chi'')\chi''(\phi(z)),
  \]
  and plug this into the expectation to get
  \[
  \Expect{(x,y,z)\sim U^{J}}{{f^{\sharp}}'(x){g^{\sharp}}'(y){h^{\sharp}}'(z)}\\
  =\sum\limits_{\chi,\chi',\chi''\in\hat{H}^{J}}
  \widehat{{f^{\sharp}}'}(\chi)\widehat{{g^{\sharp}}'}(\chi')\widehat{{h^{\sharp}}'}(\chi'')
  \Expect{(x,y,z)\sim U^{J}}{\chi(\sigma(x))\chi'(\gamma(y))\chi''(\phi(z))}.
  \]
  Using the fact that $\phi(z) = -\sigma(x)-\gamma(y)$ to write the expectation on the right hand side as
  \[
  \Expect{(x,y,z)\sim U^{J}}{(\chi\overline{\chi''})(\sigma(x))(\chi'\overline{\chi''})(\gamma(y))} = 1_{\chi = \chi' = \chi''},
  \]
  we conclude that
  \[
  \card{\Expect{(x,y,z)\sim U^{J}}{{f^{\sharp}}'(x){g^{\sharp}}'(y){h^{\sharp}}'(z)}}
  =\card{\sum\limits_{\chi\in\hat{H}^{J}}
  \widehat{{f^{\sharp}}'}(\chi)\widehat{{g^{\sharp}}'}(\chi)\widehat{{h^{\sharp}}'}(\chi)}
  \leq \max_{\chi}\card{\widehat{{f^{\sharp}}'}(\chi)}\sum\limits_{\chi\in\hat{H}^{J}}\card{\widehat{{g^{\sharp}}'}(\chi)}\card{\widehat{{h^{\sharp}}'}(\chi)},
  \]
  and using the fact that $\sum\limits_{\chi\in\hat{H}^{J}}\card{\widehat{{g^{\sharp}}'}(\chi)}\card{\widehat{{h^{\sharp}}'}(\chi)}\leq 1$ which follows by
  Cauchy-Schwarz and Parseval, we get that
  \[
  \frac{\eps}{16}\leq\card{\Expect{(x,y,z)\sim U^{J}}{{f^{\sharp}}'(x){g^{\sharp}}'(y){h^{\sharp}}'(z)}}\leq \max_{\chi}\card{\widehat{{f^{\sharp}}'}(\chi)},
  \]
  concluding the proof.
\end{proof}

\subsection{Unraveling Restrictions and Averaging: Applying the Restriction Inverse Theorem}
Summarizing, we have shown that after a sequence of restrictions, averaging and further restriction, our function $f$ has correlation with a function of the form
$\chi\circ \sigma$ with noticeable probability. We next unravel this operations to deduce a result about $f$ itself.
\begin{claim}\label{claim:reverse_F_1}
  Suppose that $f^{\sharp}$, $g^{\sharp}$ and $h^{\sharp}$ satisfy the event in Lemma~\ref{lem:move_to_gp_functions}.
  Then there is $\chi\in \hat{H}^{I}$ and an embedding function
  $L\colon \Sigma^{I}\to\mathbb{C}$ with $\norm{L}_2\leq 1$ and degree at most $d$,
  such that
  \[
  \card{\inner{f^{\sharp}}{L\cdot \chi\circ\sigma}}\geq \eps'
  \]
  where $d = {\sf poly}_{m,\alpha}\left(\log(1/\eps)\right)$ and
  $\eps' = 2^{-{\sf poly}_{m,\alpha}\left(\log(1/\eps)\right)}$.
\end{claim}
\begin{proof}
  By Claims~\ref{claim:move_to_gp_functions_shift}~\ref{claim:move_to_gp_functions_shift_fourier},
  we conclude that
  \[
  \Prob{J,\tilde{x}'}
  {\exists \chi\in\hat{H}, ~\card{\inner{{f^{\sharp}}_{I\setminus J\rightarrow\tilde{x}'}}{\chi\circ\sigma}}\geq \frac{\eps}{16}}\geq \frac{\eps}{16},
  \]
  where we think of functions as being defined over $H$. Applying Theorem~\ref{thm:restriction_inverse_special} the conclusion follows.
\end{proof}

Note that by definition of $f^{\sharp}$,
\[
\card{\inner{f^{\sharp}}{L\cdot \chi\circ\sigma}}
=\card{\inner{\mathrm{T}_{\text{non-embed}, 0}f''}{L\cdot \chi\circ\sigma}}
=\card{\inner{f''}{\mathrm{T}_{\text{non-embed}, 0}(L\cdot \chi\circ\sigma)}}
=\card{\inner{f''}{L'\cdot \chi\circ\sigma}},
\]
for $L' = \mathrm{T}_{\text{non-embed}, 0} L$. We the fact that $\mathrm{T}_{\text{non-embed}, 0}$ is Hermitian and that
$\chi\circ\sigma$ is constant on connected components of $\mathrm{T}_{\text{non-embed}, 0}$ (see Fact~\ref{fact:soft_nonbembed_op}).
Note that ${\sf deg}(L')\leq {\sf deg}(L)$ and that $\norm{L'}_2\leq \norm{L}_2\leq 1$.
Thus, we conclude from Claim~\ref{claim:reverse_F_1}
that if $f''$, $g''$ and $h''$ satisfy the event in Claim~\ref{claim:Abelian_arg_2}, then there are $L$ and $\chi$
as in Claim~\ref{claim:reverse_F_1} such that $\card{\inner{f''}{L\cdot \chi\circ\sigma}}\geq \eps'$.

Define $f''' = f'' \overline{\chi\circ \sigma}$, and observe that the above means that
$\card{\inner{f'''}{L}}\geq \eps'$ provided that $f''$, $g''$ and $h''$ satisfy the event in Claim~\ref{claim:Abelian_arg_2}.
Thus, by Cauchy-Schwarz $W_{\leq d}[f''']\geq \card{\inner{f'''}{L}}^2\geq \eps'^2$, and we now
use random restrictions to show that after a suitable restriction, $f'''$ has significant average.

\begin{lemma}\label{lem:rest_to_correlation}
  Suppose that $F\colon (\Sigma^I,\mu^{I})\to\mathbb{C}$ is a $1$-bounded function.
  Then choosing $(I',x')$ by including each $i\in I$ in $I'$ with probability $1/2d$ and sampling $x'\sim \mu^{I\setminus I'}$, we have that
  \[
        \Prob{I',x'}{\card{\E[F_{I\setminus I'\rightarrow x'}]}\geq \sqrt{\frac{W_{\leq d}[F]}{2e}}}\geq \frac{W_{\leq d}[F]}{2e}.
  \]
\end{lemma}
\begin{proof}
  Define the random variable $V_{I',x'} = \E[F_{I\setminus I'\rightarrow x'}]$, and note that
  \[
    \E_{I',x'}[\card{V_{I',x'}}^2] = \Expect{I'}{\Expect{x'\sim\mu^{I\setminus I'}}{\card{\sum\limits_{S\neq \emptyset} F^{=S}(x')1_{S\subseteq I\setminus I'}}^2}}
    =\Expect{I'}{\sum\limits_{S} \norm{F^{=S}}_2^2 1_{S\subseteq I\setminus I'}}
    \geq \frac{1}{e}W_{\leq d}[F].
  \]
  As $\card{V_{I',x'}}\leq 1$ always, it follows that with probability at least $\frac{\xi}{2e}$ we have $\card{V_{I',x'}}\geq \sqrt{\frac{\xi}{2e}}$.
\end{proof}

Applying Lemma~\ref{lem:rest_to_correlation} on $F = f'''$, we conclude that provided that $f''$, $g''$ and $h''$ satisfy the event in Claim~\ref{claim:Abelian_arg_2},
we have that
\[
\Prob{I',x'}{\card{\E[f'''_{I\setminus I'\rightarrow x'}]}\geq \frac{\eps'}{\sqrt{2e}}}\geq \frac{\eps'^2}{2e}.
\]
Noting that $\E[f'''_{I\setminus I'\rightarrow x'}] = \inner{f''_{I\setminus I'\rightarrow x'}}{\chi\circ \sigma|_{I\setminus I'\rightarrow x'}}$,
we conclude that after random restriction $f''$ is correlated with a function $\chi'\circ \sigma$ for $\chi'\in \hat{H}^{I'}$.
Thus, we are now in a position again to apply the restriction inverse theorem and conclude the proof of Theorem~\ref{thm:main_stab_3_saturated}.

\begin{claim}
  Theorem~\ref{thm:main_stab_3_saturated} is true.
\end{claim}
\begin{proof}
  Fixing $f,g,h$ and $\mu$ as in Theorem~\ref{thm:main_stab_3_saturated}, we conclude by Claim~\ref{claim:Abelian_arg_2}
  that
  \[
        \Prob{I, \tilde{x}, \tilde{y}, \tilde{z}}{\card{\Expect{(x,y,z)\sim\mu^{I}}{f''(x)g''(y)h''(z)}}\geq \frac{\eps}{4}}\geq \frac{\eps}{4},
  \]
  which by the above discussion implies that
  \[
        \Prob{\substack{I,  \tilde{x}, \tilde{y}, \tilde{z}\\ I', x',y',z'}}{\exists \chi'\in\hat{H}^{I'},
        ~\card{\inner{{f''}_{{I\setminus I'\rightarrow x'}}}{\chi'\circ \sigma}}\geq \frac{\eps'}{\sqrt{2e}}}\geq \frac{\eps}{4}\frac{\eps'^2}{2e}.
  \]
  Recalling the definition of $f''$, this means that
  \[
        \Prob{\substack{I,  \tilde{x}, \tilde{y}, \tilde{z}\\ I', x',y',z'}}{\exists \chi'\in\hat{H}^{I'},
        ~\card{\inner{{f'}_{\substack{\overline{I}\rightarrow\tilde{x}\\ I\setminus I'\rightarrow x'}}}{\chi'\circ \sigma}}\geq \frac{\eps'}{\sqrt{2e}}}\geq \frac{\eps\eps'^2}{8e},
  \]
  and we now apply Theorem~\ref{thm:restriction_inverse_special}. First, note that the restriction we are doing above fits the pattern therein;
  indeed, another way of viewing this restriction is as writing
  $\mu = \frac{\kappa}{2d} \mu + \left(1-\frac{\kappa}{2d}\right)\mu$.
  Theorem~\ref{thm:restriction_inverse_special} now implies that there is $\chi\in\hat{H}^{n}$ and $L\colon \Sigma^n\to\mathbb{C}$ with $\norm{L}_2\leq 1$
  and degree at most $d' = \eps^{-O_{m,\alpha}(1)}$ such that $\card{\inner{f'}{L\cdot \chi\circ \sigma}}\geq \eps''$
  where
  \[
  \eps''
  = 2^{-{\sf poly}_{m,\alpha}\left(\frac{\log(1/\eps')}{\kappa/2d}\right)}
  = 2^{-{\sf poly}_{m,\alpha}(1/\eps)}.
  \]
  Recalling that $f' = \mathrm{T}_{\text{non-embed}, 1-\delta} f$, we get that
  \[
  \eps''
  \leq
  \card{\inner{\mathrm{T}_{\text{non-embed}, 1-\delta} f}{L\cdot \chi\circ \sigma}}
  =
  \card{\inner{f}{\mathrm{T}_{\text{non-embed}, 1-\delta}(L\cdot \chi\circ \sigma})}
  =
  \card{\inner{f}{L' \cdot \chi\circ \sigma}}
  \]
  for $L' = \mathrm{T}_{\text{non-embed}, 1-\delta} L$. Here, we used the fact that $\chi\circ\sigma$ is
  constant on the connected components of $\mathrm{T}_{\text{non-embed}, 1-\delta}$. As ${\sf deg}(L')\leq {\sf deg}(L)$
  and $\norm{L'}_2\leq \norm{L}_2$, the proof is concluded.
\end{proof}

\section{Deducing the Structural Result for $f$: Proof of Theorem~\ref{thm:main_stab_3}}\label{sec:unraveling}
In this section we use Theorem~\ref{thm:main_stab_3_saturated} to prove Theorem~\ref{thm:main_stab_3}. At a high level, the argument proceeds as
follows:
\begin{enumerate}
  \item \textbf{Applying Theorem~\ref{thm:main_stab_3_saturated}.} Starting with a distribution $\mu$ and functions $f,g,h$ as in Theorem~\ref{thm:main_stab_3}, we use path tricks
  (and more specifically, Lemma~\ref{lem:iterate_path_trick_saturates}) to upper bound the $3$-wise correlation of $f,g$ and $h$ over $\mu$ by the
  $3$-wise correlation of functions $F, G$ and $H$ over a distribution $\mu'$, where values of $F$ correspond to product of values of $f$
  (as in Lemma~\ref{lem:from_mu_to_path}) and $G$ and $H$ are arbitrary (but $1$-bounded) functions. We are then in a position to apply Theorem~\ref{thm:main_stab_3_saturated}
  to conclude that the function $F$ is correlated with a function of the desired form, namely $L\cdot \chi\circ\sigma_{{\sf master}}'$ where $L$ is a low-degree function,
  $\chi\in\hat{H}$ is a character and $\sigma_{{\sf master}}'$ is part of the master embedding of $\mu'$.
  \item \textbf{Unraveling products.} Ignoring the low-degree part for a moment, we have that the function $F$ is a product of values of the function $f$,
  and as (by Lemma~\ref{lem:reverse_embed}) the master embedding $\sigma_{{\sf master}}'$ can be written as by an alternating sum of values of $\sigma_{{\sf master}}$
  (which is part of the master embedding of $\mu$), we have that $\chi\circ\sigma_{{\sf master}}'$ can also be written as product of values of $\chi\circ \sigma_{{\sf master}}$.
  Combining these facts, we conclude that $F\cdot\chi\circ\sigma_{{\sf master}}'$ can be written as product of values of $f\cdot\chi\circ\sigma_{{\sf master}}$, and we know
  that the expectation of this value over some distribution is significant. In applying the path tricks appropriately, we have made sure that the distribution of points
  on which we take product over is good enough (and more precisely, connected) so that this is only possible whenever $f\cdot\chi\circ\sigma_{{\sf master}}$
  have significant mass on the low-levels (see Lemma~\ref{lem:noticeable_to_lowdegwt}). Thus, we are able to conclude that
  $f\cdot\chi\circ\sigma_{{\sf master}}$ is correlated with a low-degree function, which is the result we are aiming for in Theorem~\ref{thm:main_stab_3}.
   \item \textbf{Applying the restriction inverse theorem.} To formalize this more precisely we must address the ``ignoring the low-degree function $L$'' part of the
   argument. For that, we apply random restrictions. Intuitively, a low-degree function becomes constant after random restrictions, and thus we would indeed be able to deduce that after random restriction,
   $F$ is correlated with a function of the form $\chi\circ \sigma_{{\sf master}}$. The rest of the argument proceeds in the same way, and we indeed manage to conclude that
   random restrictions of $f\cdot\chi\circ\sigma_{{\sf master}}$ are correlated with low-degree polynomials with noticeable probability. After further random restrictions
   we conclude that $f$ is correlated with a function of the form $\chi\circ\sigma_{{\sf master}}$, and to lift this information back to information about the function $f$ itself
   we use the restriction inverse theorem, namely Theorem~\ref{thm:restriction_inverse_special}.
\end{enumerate}

We now proceed to the formal argument.
\subsection{Applying Theorem~\ref{thm:main_stab_3_saturated}}
Let $f,g,h$ and $\mu$ be as in Theorem~\ref{thm:main_stab_3}, and let
$(\sigma_{{\sf master}},\gamma_{{\sf master}},\phi_{\sf master})$ be a master embedding
of $\mu$ into an Abelian group $(H,+)$ of size at most $r = r(m)\in\mathbb{N}$ as constructed in Lemma~\ref{lem:master_captures_all}. By Lemma~\ref{lem:iterate_path_trick_saturates} we may
apply the path trick on $\mu$ at most $T = T(m)$ times to get a distribution $\mu'$ over $\Sigma'\times\Gamma'\times \Phi'$,
where $\Sigma'=\Sigma^{T_1}$, $\Gamma'\subseteq \Gamma^{T_2}$, $\Phi'\subseteq \Phi^{T_3}$ and $T_1,T_2,T_3$ are odd numbers that depend
only on $m$. In addition, for future reference
we will look at the marginal distribution of $\mu'$ on its first coordinate as
$(x(1),\ldots,x(T_1))\sim \mu'$, and remark that by the construction in Lemma~\ref{lem:iterate_path_trick_saturates} it follows that the
marginal distribution of each $x(j)$ is $\mu_x$.

Following the evolution of the master embedding as in
Definition~\ref{def:master}, we see that a master embedding of $\mu'$
$(\sigma_{{\sf master}}',\gamma_{{\sf master}}', \phi_{{\sf master}}')$ is given by
\begin{align}\label{eq:unravel_0}
&\sigma_{{\sf master}}'(x(1),\ldots,x(T_1)) = \sum\limits_{j=1}^{T_1}(-1)^{j+1}\sigma_{{\sf master}}(x(j)),\notag\\
&\gamma_{{\sf master}}'(y(1),\ldots,y(T_2)) = \sum\limits_{j=1}^{T_2}(-1)^{j+1}\gamma_{{\sf master}}(y(j)),\notag\\
&\phi_{{\sf master}}'(z(1),\ldots,z(T_3)) = \sum\limits_{j=1}^{T_3}(-1)^{j+1}\phi_{{\sf master}}(z(j)).
\end{align}
Using Lemma~\ref{lem:from_mu_to_path} we that for $F\colon {\Sigma'}^{n}\to\mathbb{C}$ defined by
\begin{equation}\label{eq:unravel_00}
F(\vec{x}_1,\ldots,\vec{x}_n) = \prod\limits_{j=1}^{T_1}C^{j}\circ f(x(j)_1,\ldots,x(j)_n),
\end{equation}
where $C^{j}$ represents the operation of applying complex conjugate if $j$ is even and else applying the identity,
there are $1$-bounded functions $G\colon {\Gamma'}^n\to\mathbb{C}$ and $H\colon {\Phi'}^n\to\mathbb{C}$ such that
\[
\card{\Expect{(x,y,z)\sim \mu^{n}}{f(x)g(y)h(z)}}^{M}\leq \card{\Expect{(X,Y,Z)\sim {\mu'}^{n}}{F(X)G(Y)H(Z)}},
\]
where $M = M(m)\in\mathbb{N}$. By the premise of Theorem~\ref{thm:main_stab_3} we get
$\card{\Expect{(X,Y,Z)\sim {\mu'}^{n}}{F(X)G(Y)H(Z)}}\geq \eps^{M}$, hence we are in a position
to apply Theorem~\ref{thm:main_stab_3_saturated}. Using it, we get that there are $\chi\in\hat{H}^{n}$
and $L\colon{\Sigma'}^{n}\to\mathbb{C}$ such that:
\begin{enumerate}
  \item $\card{\inner{F}{L\cdot \chi\circ \sigma_{{\sf master}}'}}\geq \eps'$ where $\eps' = 2^{-{\sf poly}_{m,\alpha}\left(\frac{1}{\eps}\right)}$.
  \item $\norm{L}_2\leq 1$.
  \item $L$ has degree at most $d = {\sf poly}_{m,\alpha}\left(\frac{1}{\eps}\right)$.
\end{enumerate}
\subsection{Unraveling Products}
We now get rid of the low-degree part $L$ via random restrictions, and then use the product structure of $F$ and of $\chi\circ \sigma_{{\sf master}}'$
to convert the information regarding the correlation between them to information about the correlation between restrictions of $f$ and restrictions of
$\chi\circ \sigma_{{\sf master}}$. Let $\mathcal{D}$ be the marginal distribution of $\mu'$ on $\Sigma'$.
Choose $(I,\tilde{x})$ a random restriction by including in $I$ each $i\in [n]$ with probability $1/2d$ and taking $\tilde{x}\sim \mathcal{D}^{\overline{I}}$.
By Cauchy-Schwarz we have that
\[
W_{\leq d}[F\cdot \overline{\chi\circ \sigma_{{\sf master}}'}]
\geq
\card{\inner{F\cdot \overline{\chi\circ \sigma_{{\sf master}}'}}{L}}^2
=
\card{\inner{F}{L\cdot \chi\circ \sigma_{{\sf master}}'}}^2
\geq \eps'^2
\]
hence by Lemma~\ref{lem:rest_to_correlation} it follows that
\begin{equation}\label{eq:unravel_1}
\Prob{I,\tilde{x}\sim \mathcal{D}^{\overline{I}}}{\card{\E[F|_{\overline{I}\rightarrow \tilde{x}}\cdot \overline{\chi\circ \sigma_{{\sf master}}'|_{\overline{I}\rightarrow \tilde{x}}}]}\geq \frac{\eps'}{\sqrt{2e}}}
\geq \frac{\eps'^2}{2e}.
\end{equation}
Whenever this event holds we get that
\begin{equation}\label{eq:unravel_2}
\card{\inner{F|_{\overline{I}\rightarrow \tilde{x}}}{\chi\circ \sigma_{{\sf master}}'|_{\overline{I}\rightarrow \tilde{x}}}}\geq \frac{\eps'}{2e}.
\end{equation}
We denote $F' = F|_{\overline{I}\rightarrow \tilde{x}}$, and thus get that if $I,\tilde{x}$ are such that~\eqref{eq:unravel_1}
holds, then there is $\chi'\in\hat{H}^{I}$ such that $\card{\inner{F'}{\chi'\circ\sigma_{{\sf master}}'}}\geq \frac{\eps'}{\sqrt{2e}}$.
Each coordinate of the restriction of $\tilde{x}$, namely each $\tilde{x}_i$, is an element in $\Sigma'=\Sigma^{T_1}$ and we will view
it as a vector length $T_1$, $\tilde{x}_i = (\tilde{x}(1)_i,\ldots,\tilde{x}(T_1)_i)$. Thus, for each $j=1,\ldots,T_1$ we denote
\[
\tilde{x}(j) = (\tilde{x}(j)_i)_{i\in\overline{I}},
\]
With this notation, we get from~\eqref{eq:unravel_00} that
\begin{equation}\label{eq:unravel_5}
F'(x) = \prod\limits_{j=1}^{T_1} C^{j}\circ f_{\overline{I}\rightarrow \tilde{x}(j)}(x(j)).
\end{equation}

The following claim asserts that after further random restriction, with noticeable probability the function $f_{\overline{I}\rightarrow \tilde{x}(1)}$ is correlated
with a character.
\begin{claim}\label{claim:break_product}
Let $I$ and $\tilde{x}$ be such~\eqref{eq:unravel_2} holds. Then there is
$D = O_{\alpha,m}(\log(1/\eps'))$ such that choosing $J\subseteq I$ by including each element with probability $1/2D$ and sampling $\tilde{x}'\sim \mathcal{D}^{I\setminus J}$,
we have that
\[
\Prob{J, \tilde{x}'}{\exists \chi\in\hat{H}^{J},
~\card{\inner{f_{\substack{\overline{I}\rightarrow \tilde{x}(1)\\ I\setminus J\rightarrow \tilde{x}'(1)}}}{\chi\circ\sigma_{{\sf master}}}}\geq \frac{\eps'^2}{100}}
\geq \frac{\eps'^4}{1000}.
\]
\end{claim}
\begin{proof}
  Note that by~\eqref{eq:unravel_5}, we may write $F'$ as
  \begin{equation}\label{eq:unravel_3}
  F'(x) = \prod\limits_{j=1}^{T_1} C^{j}\circ f_{j}(x(j)),
  \end{equation}
  where $f_{j} = f_{\overline{I}\rightarrow \tilde{x}(j)}$. By~\eqref{eq:unravel_0} we have
  \[
  \sigma_{\sf master}'(x)
  =\sum\limits_{j=1}^{T_1}(-1)^{j+1}\sigma_{\sf master}'((x(j)),
  \]
  and so
  \begin{equation}\label{eq:unravel_4}
  \chi\circ\sigma_{\sf master}'(x)
  =\prod\limits_{j=1}^{T_1}C^{j+1}\circ \chi\circ \sigma_{\sf master}'(x(j)).
  \end{equation}
  Plugging~\eqref{eq:unravel_3} and~\eqref{eq:unravel_4} into~\eqref{eq:unravel_2} yields that
  \[
  \card{\Expect{(x,x')\sim\mathcal{D}^{I}}{\prod\limits_{j=1}^{T_1}
  C^{j}\circ f_{j}(x(j))\cdot C^{j+1}\circ \chi\circ \sigma_{\sf master}'(x(j))}}\geq \frac{\eps'}{\sqrt{2e}}.
  \]
  Let $\tilde{\Sigma} = \Sigma^{T_1-1}$ and define $Q\colon \tilde{\Sigma}^I\to\mathbb{C}$ by
  \[
  Q(w) = \prod\limits_{j=2}^{T_1}
  C^{j}\circ f_{j}(w(j))\cdot C^{j+1}\cdot \chi\circ \sigma_{\sf master}'(w(j)),
  \]
  then we get that
  $\card{\Expect{(x,w)\sim \nu^{I}}{f_{1}(x)\cdot C\circ \chi\circ \sigma_{\sf master}'(x) \cdot Q(w)}}\geq \frac{\eps'}{\sqrt{2e}}$, where
  $\nu$ is the marginal distribution of $\mathcal{D}$ on the first coordinate viewed as an element in $\Sigma\times\tilde{\Sigma}$.
  Consider the distribution $\nu'$ over $\Sigma^2$ where we sample $(x,w)\sim \nu$, then $(x',w')\sim \nu$ conditioned on $w' = w$,
  and then output $(x,x')$. By Cauchy-Schwarz it follows that
  \begin{align*}
    \frac{\eps'^2}{2e}
    &\leq  \card{\Expect{(x,w)\sim \nu^{I}}{f_{1}(x)\cdot C\circ\chi\circ \sigma_{\sf master}'(x) \cdot Q(w)}}^2\\
    &\leq  \card{\Expect{(x,x')\sim {\nu'}^{I}}{f_{1}(x)\cdot C\circ\chi\circ \sigma_{\sf master}'(x) \cdot \overline{f_{1}(x')\cdot C\circ\chi\circ \sigma_{\sf master}'(x')}}}\\
    &=\card{\inner{f_{1}\cdot C\circ\chi\circ \sigma_{\sf master}}{\mathrm{R}^{I}(f_{1} \cdot C\circ\chi\circ \sigma_{\sf master})}},
  \end{align*}
  where $\mathrm{R}\colon L_2(\Sigma; \nu_x)\to L_2(\Sigma; \nu_x)$ is the averaging operator corresponding to $\nu'$ defined by
  $\mathrm{R} p(a) = \cExpect{(x,x')\sim \nu}{x = a}{p(x')}$.

  As the support of $\mu'$ on the first coordinate is full we get that
  the support of $\nu'$ on $\Sigma^{2}$ is full, hence $\mathrm{R}$ is connected. Also, the probability of each atom is at least
  $\alpha' = \alpha'(\alpha,m)>0$, so by Lemma~\ref{lem:noticeable_to_lowdegwt} we conclude that for $D = O_{\alpha,m}(\log(1/\eps'))$
  we have that $W_{\leq D}[f_{1}\chi\circ \sigma_{\sf master}]\geq \frac{\eps'^4}{4e^2}$. The proof is now concluded by Lemma~\ref{lem:rest_to_correlation}.
\end{proof}

\subsection{Applying the Restriction Inverse Theorem}
Combining~\eqref{eq:unravel_1} and Claim~\ref{claim:break_product} gives that
\[
\Prob{\substack{I,\tilde{x}\sim \mathcal{D}^{I}\\ J, \tilde{x}'\sim\mathcal{D}^{I\setminus J}}}
{\exists \chi\in\hat{H}^{J},~\card{\inner{f|_{\substack{I\rightarrow \tilde{x}(1)\\ I\setminus J\rightarrow \tilde{x}'(1)}}}{\chi\circ \sigma_{{\sf master}}''}}
\geq\frac{\eps'^2}{100}}\geq \frac{\eps'^2}{2e}\cdot\frac{\eps'^4}{1000}=\frac{\eps'^6}{2000e}.
\]
We now appeal to the restriction inverse theorem, Theorem~\ref{thm:restriction_inverse_special}, to finish the proof. Towards this end, we note that the sequence of
restrictions can be viewed as a standard restriction: note that the distribution of $\tilde{x}(1)$ where $\tilde{x}\sim \mathcal{D}$ is $\mu_x$, so the above
restriction amounts to choosing $K\subseteq[n]$ by including each element with probability $\frac{1}{4dD}$, and then restricting the coordinates outside $K$
to be $x'\sim \mu_x^{\overline{K}}$. Thus, we may appeal Theorem~\ref{thm:restriction_inverse_special}, and the result follows.

\section{The Restriction Inverse Theorem}\label{sec:rest_inverse}
The main goal of this section is to prove the restriction inverse theorem, Theorem~\ref{thm:restriction_inverse_special}.
We begin by presenting a few notions that will be necessary for the statement of the theorem, give a formal
statement of a slight generalization of Theorem~\ref{thm:restriction_inverse_special} and then present some tools necessary for the proof.
Finally, in Section~\ref{sec:proof_of_rest_inverse} we give the formal proof.

\subsection{Product Functions and Classes of Product Functions}
The restriction inverse theorem is concerned with functions that, after random restriction, are correlated with product
functions, defined as follows.
\begin{definition}
  We say $f'\colon \Sigma^n\to\mathbb{C}$ is a product function if there are functions
  $p_1,\ldots,p_n\colon \Sigma\to\mathbb{C}$ that are $1$-bounded such that
  \[
  f(x_1,\ldots,x_n) = \prod\limits_{i=1}^{n} p_i(x_i).
  \]
  We denote by $\mathcal{P}(n,\Sigma)$ the collection of all product functions over $\Sigma^n$, and
  denote $\mathcal{P}(\Sigma) = \bigcup_{n\in\mathbb{N}}\mathcal{P}(n,\Sigma)$.
\end{definition}

We will need the notion of a class of product functions, which is a sub-collection of functions closed under restrictions.
For technical reasons, this closure will be up to multiplying by a complex number of absolute value $1$.
\begin{definition}
  We say $\mathcal{F}(\Sigma)\subseteq \mathcal{P}(\Sigma)$ is a class of product functions if
  for all $f'\in \mathcal{F}(\Sigma)$, say $f'\colon \Sigma^{n}\to\mathbb{C}$, for all $I\subseteq [n]$
  and for all $x\in\Sigma^{\overline{I}}$ it holds that there is $\theta\in\mathbb{C}$ of absolute value $1$
  such that $\theta f'_{\overline{I}\rightarrow x}\in \mathcal{F}(\Sigma)$.
\end{definition}
We also need the notion of separateness of product functions. Intuitively, this says that any two univariate functions in the class either have correlation $1$,
or else the correlation is bounded away from $1$.
\begin{definition}
  Suppose $\Sigma$ is a finite alphabet, $\mu$ is a distribution over $\Sigma$ and consider the inner product spaces $L_2(\Sigma^n, \mu^{\otimes n})$.
  For $\tau>0$,  we say a collection of product functions $\mathcal{F}(\Sigma)\subseteq \mathcal{P}(\Sigma)$ is $\tau$-separated if
  for any uni-variate functions $p,p'\in \mathcal{F}(\Sigma)$ it is either the case that $p=p'$, or else
  $\card{\inner{p}{p'}}\leq 1-\tau$.
\end{definition}

An important class of functions for us will be the class of functions arising from Abelian embeddings. Suppose
that we have an Abelian group $(H,+)$ and a map $\sigma\colon \Sigma\to H$. Then in this setting, we may define the
collection
\[
\mathcal{F}_{\sigma}(\Sigma,n) =
\sett{ f(x_1,\ldots,x_n) = \prod\limits_{i=1}^{n}\chi_i(\sigma(x_i))}{\chi_i\in\hat{H}~\forall i=1,\ldots,n}.
\]
The following fact asserts that the collection $\mathcal{F}_{\sigma}(\Sigma) = \bigcup_{n\in\mathbb{N}}\mathcal{F}(\Sigma,n)$ is a class of product functions
which is separated.
\begin{fact}\label{fact:F_are_class}
  For all $m\in\mathbb{N}$ and $\alpha>0$ there is $\tau>0$ such that the following holds.
  If $\Sigma$ is an alphabet of size at most $m$ and $\mu$ is a distribution over $\Sigma$ in which the probability of each atom is at least $\alpha$,
  then collection $\mathcal{F}_{\sigma}(\Sigma)$ is a class of product functions which is $\tau$-separated.
\end{fact}
\begin{proof}
  It is clear that each function in $\mathcal{F}_{\sigma}(\Sigma)$ is a product function. Also, when we restrict a set of variables, the corresponding terms give a constant
  factor $\theta$ with absolute value $1$, hence $\mathcal{F}_{\sigma}(\Sigma)$ is closed under restrictions.

  For the $\tau$-separatedness, fix univariate functions $p, p'\in \mathcal{F}_{\sigma}(\Sigma)$ and suppose that
  $p\neq p'$. We note that $\card{\inner{p}{p'}} < 1$: otherwise, Cauchy-Schwarz would be tight, hence $p$ and $p'$ would
  be proportional. However, as $0\in H$ is in the image of $\sigma_{{\sf master}}$ we may find $x$ such that $\sigma_{{\sf master}}(x) = 0$,
  and so $p(x) = 1 = p'(x)$, and in conjuction with the fact they are proportional we would get that $p \equiv p'$.

  Multiplying $p$ by a constant of absolute value $1$,
  we may assume that $\inner{p}{p'} \geq 0$, hence $\inner{p}{p'} = 1-\frac{1}{2}\norm{p-p'}_2^2$.
  By definition of $\hat{H}$, we may find $T = T(m)>0$ and $\lambda$ of absolute value $1$
  such that the values of $p$ and $p'$ take the forms $\lambda e^{\frac{L}{T}2\pi{\bm i}}$
  and $e^{\frac{L'}{T}2\pi{\bm i}}$ respectively, where $L,L'$ are integers.

  As $p$ and $p'$ are not proportional there are $a,b\in \Sigma$ such that
  $p(a) \neq p'(a)$ and $b$ such that $p(b)\neq e^{-\frac{1}{2T} 2\pi{\bm i}} p'(b)$.
  Writing $\lambda = e^{2\pi{\bm i} \theta}$ we have
  \[
  p(a) = e^{\frac{L_1+\theta T}{T}2\pi{\bm i}},
  \qquad
  p'(a) = e^{\frac{L_1'}{T}2\pi{\bm i}},
  \qquad
  p(b) = e^{\frac{L_2+\theta T}{T}2\pi{\bm i}},
  \qquad
  p'(b) = e^{\frac{L_2'+1/2}{T}2\pi{\bm i}}
  \]
  for integers $L_1, L_2, L_3, L_4$. Then
  \[
  \norm{p-p'}_2^2
  \geq \alpha\card{p(a) - p'(a)}^2+\alpha\card{p(b) - p'(b)}^2,
  \]
  and we argue that the right hand side is at least $\Omega(\alpha/T^2)$. Indeed,
  \[
  \card{p(a) - p'(a)} = \card{e^{\frac{L_1'-L_1 -\theta T}{T}2\pi{\bm i}} - 1},
  \qquad
  \card{p(b) - p'(b)} = \card{e^{\frac{L_2'-L_2 -\theta T + 1/2}{T}2\pi{\bm i}} - 1},
  \]
  so if both are at most $\frac{1}{100T}$ then each one of $-\theta T$ and $-\theta T + \frac{1}{2}$ is $\frac{1}{10}$-close
  to be an integer, but by the triangle inequality this is impossible.
\end{proof}

\subsection{Statement of the Restriction Inverse Theorem}
With the above set up, we are now ready to state the restriction inverse theorem.
\begin{thm}\label{thm:restriction_inverse}
  For all $\alpha,\tau>0$, $m\in\mathbb{N}$, $\rho\in (0,1)$ and $\eps>0$ there are $d\in\mathbb{N}$ and $\eps'>0$ such that the following holds.
  Suppose that $\Sigma$ is an alphabet of size at most $m$ and $\mu,\nu,\nu'$ are a distribution over $\Sigma$ in which the probability of each atom is at least
  $\alpha$ and $\mu = \rho\nu + (1-\rho)\nu'$. Suppose further that
  $\mathcal{F}(\Sigma)$ is a class of product functions that is $\tau$ separated.

  If $f\colon(\Sigma^n,\mu^{\otimes n})\to\mathbb{C}$ is a $1$-bounded function such that
  \[
  \Prob{I\subseteq_{\rho}[n], \tilde{x}\sim\mu^{\overline{I}}}{\exists f'\in \mathcal{F}(\Sigma), \card{\inner{f_{\overline{I}\rightarrow \tilde{x}}}{f'}}\geq \eps}\geq \eps,
  \]
  then there exist $f'\in\mathcal{F}(\Sigma)$, as well as $L\colon\Sigma^n\to\mathbb{C}$ a function of degree at most $d$ and $\norm{L}_2\leq 1$,
  such that
  \[
    \card{\inner{f}{Lf'}}\geq \eps'.
  \]
  Quantitatively, $d = {\sf poly}_{m,\alpha,\tau}\left(\frac{\log(1/\eps)}{\rho}\right)$ and
  $\eps' = 2^{-{\sf poly}_{m,\alpha,\tau}\left(\frac{\log(1/\eps)}{\rho}\right)}$.
\end{thm}

We note that Theorem~\ref{thm:restriction_inverse} immediately implies Theorem~\ref{thm:restriction_inverse_special}:
\begin{claim}
  Theorem~\ref{thm:restriction_inverse} implies Theorem~\ref{thm:restriction_inverse_special}.
\end{claim}
\begin{proof}
  Using Fact~\ref{fact:F_are_class}, the collection $\mathcal{F}_{\sigma}(\Sigma)$ satisfies the properties required by Theorem~\ref{thm:restriction_inverse},
  and so applying Theorem~\ref{thm:restriction_inverse} on it gives translates to the statement of Theorem~\ref{thm:restriction_inverse_special}.
\end{proof}

\subsection{Tools for the Proof of Theorem~\ref{thm:restriction_inverse}}
In this section we give the key ingredients for the proof of Theorem~\ref{thm:restriction_inverse}. Throughout this section, we denote by
$[R]^{\leq n}$ the set of vectors over $[R]$ of length at most $n$.
\subsubsection{The Direct Product Theorem}
First, we need a suitable direct product result. A function
$F\colon (\{0,1\}^n, \mu_{\rho}^{\otimes n}) \to [R]^{\leq n}$ is called a direct product function if there is $f\colon [n]\to[R]$ for which $F[A] = f|_{A}$ for all $A$.
Here and throughout, $\mu_{\rho}$ represents the $\rho$-biased distribution, in which $A\sim \mu_{\rho}^{\otimes n}$ is sampled by including each element
independently with probability $\rho$. The goal in direct product testing is to design a randomized test that queries a few locations at the function $F$, performs a test on them
and then accept or reject accordingly. The tester should have the following properties:
\begin{enumerate}
  \item \textbf{Completeness:} if $F$ is a direct product function, the tester must accept with probability $1$.
  \item \textbf{Soundness:} if the tester accepts with noticeable probability, then $F$ is somewhat correlated with a direct product function.
  By that, ideally one would like to say that if the tester accepts with probability $\eps$, then
  there is a direct product function $G$ such that $F[A] = G[A]$ with probability at least $\delta = \delta(\eps)>0$ over $A\sim\mu_{\rho}^{\otimes n}$.
  We will not be able to guarantee that type of soundness, and instead settle for something slightly weaker: with probability at least
  $\delta$ over the choice of $A$, we have that $\Delta(F[A],G[A])\leq r$, where $\Delta(x,y)$ measures the Hamming distance between two strings
  $x,y\in [R]^k$ and $r = r(\eps)\in\mathbb{N}$. In words, $F[A]$ and $G[A]$ agree on all but constantly many coordinates.\footnote{We remark
  that this is the best soundness one may hope to get in general. Indeed, for essentially all direct product testers and in particular the one
  we consider, it is the case that slight perturbations of legitimate direct product functions pass the test with noticeable probability. Namely,
  taking a direct product function $F$ and taking $F'$ such that $\Delta(F[A],F'[A])\leq r$ for at least $\eps$ fraction of the $A$'s, one typically
  has that $F'$ passes the direct product tester with probability at least $\eps^{2} 2^{-O(r)}$.}
\end{enumerate}

For our application we need to consider a particular test, which is also the most natural direct product tester one may think of.
Given parameters $\rho,\alpha,\beta\in (0,1)$ and oracle access to $F\colon (\{0,1\}^n, \mu_{\rho}^{\otimes n}) \to [R]^{\leq n}$
thought of as mapping $A$ to $F[A]\in [R]^A$, perform the following test, referred to as the ${\sf DP}(\rho,\alpha,\beta)$ test:
\begin{enumerate}
  \item Sample $C\sim \mu_{\rho\alpha}^{\otimes n}$ and sample $A, B \sim \mu_{\rho}^{\otimes n}$ independently conditioned on $A, B\supseteq C$.
  \item Sample $T\subseteq_{\beta} [n]$.
  \item Check that $F[A]|_{C\cap T} = F[B]|_{C\cap T}$.
\end{enumerate}

Roughly speaking, we sample $A$ and $B$ that have $\rho\alpha$ of their elements in common. Then, we choose a subset $T\subseteq_{\beta} [n]$ so
that $T\cap C$ contains roughly $\beta$ fraction of their shared elements, and check that the assignments $F[A]$ and $F[B]$ are consistent on $C\cap T$.
It is clear that the tester ${\sf DP}(\rho,\alpha,\beta)$ has the completeness property, and the following result addresses the soundness of ${\sf DP}(\rho,\alpha,\beta)$:
\begin{thm}\label{thm:DP_biased_version}
  For all $C>0$ there is $c>0$ such that the following holds for sufficiently large $n$ and $\eps\geq 2^{-n^{c}}$.
  Suppose that $\frac{1}{\log(1/\eps)^C}\leq \alpha,\beta\leq \frac{9}{10}$ and $\rho\in (0,1)$; then there are $r\in\mathbb{N}$ and $\eps'> 0$ such that if
  $F\colon (\{0,1\}^n, \mu_{\rho}^{\otimes n}) \to [R]^{\leq n}$ is an assignment as above that passes the ${\sf DP}(\rho,\alpha,\beta)$ test
  with probability at least $\eps$, then there exists $f\colon [n]\to[R]$ such that
  \[
  \Prob{A\sim \mu_{\rho}^{\otimes n}}{\Delta(F[A], f|_{A})\leq r}\geq \eps'.
  \]
  Quantitatively, we have $r = \rho^{-O_{C}(1)}\log(1/\eps)^{O_{C}(1)}$ and $\eps' = \eps^{O_{C}(\log(1/\rho)^2)}$.
\end{thm}
\begin{proof}
  Deferred to Section~\ref{sec:dp}.
\end{proof}

\subsubsection{Stability and Level $d$ Inequalities}
We record here a few basic notions from analysis of Boolean functions over product domains that we need; we refer the reader to~\cite{ODonnell}
for details. First is the notion of noise stability, for which we first define the standard noise operator:
\begin{definition}
  For a finite probability space $(\Sigma,\mu)$, $\rho\in [0,1]$ and $x\in \Sigma$, we define the distribution over $\rho$-correlated
  inputs with $x$, denoted by $\mathrm{T}_{\rho} x$, to be: take $y = x$ with probability $\rho$, otherwise sample $y\sim \mu$.
\end{definition}
As usual, we may associate with $\mathrm{T}_{\rho}$ an averaging operator $\mathrm{T}_{\rho} \colon L_2(\Sigma,\mu)\to L_2(\Sigma,\mu)$,
as well as tensor it to get an operator acting on $n$-variate functions.

\begin{definition}
  For a finite probability space $(\Sigma,\mu)$, $\rho\in [0,1]$ and $f\colon \Sigma^n\to\mathbb{C}$, we define the $\rho$-noise stability
  of $f$ as ${\sf Stab}_{\rho} = \inner{f}{\mathrm{T}_{\rho} f}$.
\end{definition}

The following fact asserts that if a function $f$ has significant noise stability, then it has significant weight on the low-levels.
\begin{fact}\label{fact:stability_to_weight}
  Suppose that $(\Sigma,\mu)$ is a finite probability space, $f\colon \Sigma^n\to\mathbb{C}$ is a function with $2$-norm at most $1$
  and ${\sf Stab}_{1-\eps}(f)\geq \delta$. Then $W_{\leq \frac{2\log(1/\delta)}{\eps}}[f]\geq \frac{\delta}{2}$.
\end{fact}
\begin{proof}
  Writing $f = \sum\limits_{S\subseteq [n]} f^{=S}$ according to the Efron Stein decomposition on $(\Sigma^n,\mu^{\otimes n})$ and noting
  that $\mathrm{T}_{\rho}^{\otimes n} f^{=S} = \rho^{\card{S}}f^{=S}$, we get that
  \[
  {\sf Stab}_{1-\eps}(f) = \sum\limits_{S\subseteq [n]} (1-\eps)^{\card{S}}\norm{f^{=S}}_2^2.
  \]
  The contribution from $\card{S} > 2\log(1/\delta)/\eps$ is at most
  \[
  (1-\eps)^{2\log(1/\delta)/\eps}\sum\limits_{\card{S} > 2\log(1/\delta)/\eps} \norm{f^{=S}}_2^2
  \leq e^{-2\log(1/\delta)} \norm{f}_2^2
  \leq \frac{\delta}{2},
  \]
  and the contribution from $\card{S}\leq 2\log(1/\delta)/\eps$ is at most $W_{\leq 2\log(1/\delta)/\eps}[f]$.
  It follows that $W_{\leq 2\log(1/\delta)/\eps}[f]+\frac{\delta}{2}\geq \delta$, and the proof is concluded by re-arranging.
\end{proof}

The following fact is known as the level $d$ inequality; it asserts that a Boolean function with small average may only have very small weight on low levels.
\begin{fact}\label{fact:level_d_inequality}
  Let $\Sigma$ be a finite alphabet and let $\mu$ be a distribution over $\Sigma$ in which the probability of each atom is at least $\alpha$.
  If $F\colon \Sigma^{n}\to\{0,1\}$ is a function with $\E_{\mu}[F] = s$, then $W_{\leq d}[F]\leq 2^{O_{\alpha}(d)}s^{3/2}$ for all $d$.
\end{fact}
\begin{proof}
  This is a standard consequence of the hypercontractive inequality, asserting that there is $C(\alpha)>0$ such that
  $\norm{g}_4\leq C(\alpha)^{d}\norm{g}_2$ for all functions $g$ of degree at most $d$; see~\cite[Theorem 10.21]{ODonnell}.
  Thus, by H\"{o}lder's inequality
  \[
  W_{\leq d}[F]
  =\inner{F^{\leq d}}{F}
  \leq \norm{F^{\leq d}}_{4}\norm{F}_{4/3}
  \leq C(\alpha)^{d} \norm{F^{\leq d}}_2\norm{F}_{4/3}
  =C(\alpha)^{d} \sqrt{W_{\leq d}[F]}\norm{F}_{4/3}
  \]
  and re-arranging gives
  $W_{\leq d}[F]\leq C(\alpha)^{2d} \norm{F}_{4/3}^2 = C(\alpha)^{2d}s^{3/2}$.
\end{proof}

\subsubsection{Some Averaging Arguments}
Our argument makes use of several standard probabilistic facts which we collect here. The first of which asserts that if we have functions $f$ and $g$
that are somewhat correlated and $g$ is $1$-bounded, then with noticeable probability they remain somewhat correlated after a random restriction.
\begin{fact}\label{fact:restrict_keeps_correlations}
  Let $\Sigma$ be a finite alphabet, $\mu$ be a distribution over $\Sigma$ and let $f,g\colon \Sigma^n\to\mathbb{C}$
  be functions such that $g$ is $1$-bounded and $\norm{f}_2\leq 1$. For all $I\subseteq [n]$, if $\card{\inner{f}{g}}\geq \eta$, then sampling $\tilde{x}\sim\mu^{\overline{I}}$,
  we have that
  \[
  \Prob{\tilde{x}\sim\mu^{\overline{I}}}{\card{\inner{f_{\overline{I}\rightarrow\tilde{x}}}{g_{\overline{I}\rightarrow \tilde{x}}}}\geq
  \frac{\eta}{2}}
  \geq \frac{\eta^2}{4}.
  \]
\end{fact}
\begin{proof}
  Denote the event in question by $E$, and note that
  \[
  \Expect{\tilde{x}}{\card{\inner{f_{\overline{I}\rightarrow\tilde{x}}}{g_{\overline{I}\rightarrow \tilde{x}}}}}
  \geq
  \card{\Expect{\tilde{x}}{\inner{f_{\overline{I}\rightarrow\tilde{x}}}{g_{\overline{I}\rightarrow \tilde{x}}}}}
  =\card{\inner{f}{g}}
  \geq \eta.
  \]
  On the other hand,
  \[
  \Expect{\tilde{x}}{\card{\inner{f_{\overline{I}\rightarrow\tilde{x}}}{g_{\overline{I}\rightarrow \tilde{x}}}}}
  =
  \underbrace{\Expect{\tilde{x}}{1_E\card{\inner{f_{\overline{I}\rightarrow\tilde{x}}}{g_{\overline{I}\rightarrow \tilde{x}}}}}}_{(\rom{1})}
  +
  \underbrace{\Expect{\tilde{x}}{1_{\overline{E}}\card{\inner{f_{\overline{I}\rightarrow\tilde{x}}}{g_{\overline{I}\rightarrow \tilde{x}}}}}}_{(\rom{2})},
  \]
  and we upper bound each term as follows. For $(\rom{1})$ we have by Cauchy-Schwarz that
  \[
    \card{(\rom{1})}
    \leq \sqrt{\Prob{}{E}}
    \sqrt{\Expect{\tilde{x}}{\card{\inner{f_{\overline{I}\rightarrow\tilde{x}}}{g_{\overline{I}\rightarrow \tilde{x}}}}^2}}
    \leq\sqrt{\Prob{}{E}}
    \sqrt{\Expect{\tilde{x}}{\norm{f_{\overline{I}\rightarrow\tilde{x}}}_2^2}}
    =\sqrt{\Prob{}{E}}\norm{f}_2
    \leq\sqrt{\Prob{}{E}}.
  \]
  Here, we used the fact that $g$ is $1$-bounded. For $(\rom{2})$ we have by definition of $\overline{E}$ we have
  $\card{(\rom{2})}\leq \frac{\eta}{2}$. Combining, we ge that $\eta\leq \sqrt{\Prob{}{E}} + \frac{\eta}{2}$, hence $\Prob{}{E}\geq \frac{\eta^2}{4}$.
\end{proof}

The second fact asserts that if $X,Y$ are independent random variables and $E(X,Y)$ is
an event with noticeable probability, then sampling $x_1,\ldots,x_k\sim X$ and $y_1,\ldots,y_{\ell}\sim Y$ independently we have that
all of $E(x_i,y_j)$ occur with noticeable probability.
\begin{fact}\label{fact:holder_trick}
  Let $X, Y$ be independent random variables and let $E$ be an event depending only on $X,Y$, and suppose that $\Prob{}{E}\geq \delta$.
  Then
  \[
  \Prob{\substack{x_1,\ldots,x_k\sim X\\ y_1,\ldots,y_{\ell}\sim Y}}{\bigcap_{i,j} E(x_i,y_j)}\geq \delta^{k\ell}.
  \]
\end{fact}
\begin{proof}
  We have that $\Expect{x\sim X, y\sim Y}{1_{E(x,y)}} = \Prob{}{E}\geq \delta$, so raising to the power $k$ and using H\"{o}lder's inequality gives
  \begin{align*}
  \delta^{k}
  \leq
  \Expect{x\sim X, y\sim Y}{1_{E(x,y)}}^k
  =
  \Expect{y\sim Y}{\Expect{x\sim X}{1_{E(x,y)}}}^k
  &\leq
  \Expect{y\sim Y}{\Expect{x\sim X}{1_{E(x,y)}}^k}\\
  &=
  \Expect{y\sim Y}{\Expect{x_1,\ldots,x_k\sim X}{1_{\bigcap_{i=1}^{k}E(x_i,y)}}}.
  \end{align*}
  Raising to the power $\ell$ and using H\"{o}lder's inequality again gives
  \begin{align*}
  \delta^{k\ell}
  \leq
  \Expect{y\sim Y}{\Expect{x_1,\ldots,x_k\sim X}{1_{\bigcap_{i=1}^{k}E(x_i,y)}}}^{\ell}
  &=
  \Expect{x_1,\ldots,x_k\sim X}{\Expect{y\sim Y}{1_{\bigcap_{i=1}^{k}E(x_i,y)}}}^{\ell}\\
  &\leq
  \Expect{x_1,\ldots,x_k\sim X}{\Expect{y\sim Y}{1_{\bigcap_{i=1}^{k}E(x_i,y)}}^{\ell}}\\
  &=\Expect{x_1,\ldots,x_k\sim X}{\Expect{y_1,\ldots,y_{\ell}\sim Y}{1_{\bigcap_{i,j}E(x_i,y_j)}}},
  \end{align*}
  as desired.
\end{proof}
\subsection{Proof of Theorem~\ref{thm:restriction_inverse}}\label{sec:proof_of_rest_inverse}
We now proceed to the proof of Theorem~\ref{thm:restriction_inverse}. The argument we present is similar to an argument from~\cite{BKMcsp3}
with some differences. In the setting therein, the class of product functions was the collection of all multiplicative characters over $\mathbb{F}_2$
and the underlying measure was uniform, hence any two distinct product functions were orthogonal. Such orthogonality properties were
used multiple times, and most importantly it implies that a given function $f\colon \Sigma^n\to\mathbb{C}$ with $2$-norm at most $1$ could
be $\eta$-correlated with at most $\frac{1}{\eta^2}$ product functions. In the current setting we do not have these orthogonality properties.
To circumvent that, we consider nets, which are small collections of product functions that in some sense capture all of the product functions
correlated with $f$. Stated simply, while there could be many product functions that are correlated with $f$ (their number
could depend on the dimension $n$ for example), we argue that one could choose a short list of product functions that are correlated with $f$,
so that any other product function that is correlated with $f$ must be close to a product function from the list.
With this change, we can use the main ideas from argument in~\cite{BKMcsp3} (which still requires some non-trivial but more minor adaptations).

\skipi
Throughout this section we will use the following notation:
\begin{definition}
  For a product function $p\colon\Sigma^{n}\to\mathbb{C}$ given as $p(x_1,\ldots,x_n) = \prod\limits_{i=1}^{n}p_i(x_i)$ and $T\subseteq [n]$,
  we define $p|_T\colon\Sigma^{T}\to\mathbb{C}$ by $p|_{T}(y) = \prod\limits_{i\in T}p_i(y_i)$.
\end{definition}
Next, we define the action of this operation on a collection of functions in the natural way.
\begin{definition}
  For a collection of product functions $W\subseteq\sett{p\colon\Sigma^{n}\to\mathbb{C}}{p\text{ is a product function}}$
  and $T\subseteq [n]$ we define $W|_{T} = \sett{p|_T}{p\in W}$.
\end{definition}

We also note that by the $\tau$-separatedness of $\mathcal{F}$, it follows that the number of univariate functions in $\mathcal{F}(\Sigma)$ is
at most some finite number $R = R(\tau)$, and we fix this $R$ henceforth. We will identify between $[R]$ and univariate functions in $\mathcal{F}(\Sigma)$,
and thus also between $n$-variate product functions in $\mathcal{F}(\Sigma)$ and $[R]^n$.

\subsubsection{The Net of Product Functions and the Symbolic Distance}
Fix $\Sigma$, a distribution $\mu$ over $\Sigma$ and a class $\mathcal{F} = \mathcal{F}(\Sigma)$ as in Theorem~\ref{thm:restriction_inverse}.
We will consider various domains, and to make the notations more precise we will denote by $\mathcal{F}(\Sigma,I)$ the subset of $\mathcal{F}$
consisting of functions whose domain is $\Sigma^{I}$.

For a function $f\colon (\Sigma^n,\mu^{n})\to\mathbb{C}$ and parameter $\eps>0$, denote
\[
{\sf List}_{\eps}[f] = \sett{p\in\mathcal{F}}{\card{\inner{f}{p}}\geq \eps}.
\]
As discussed earlier, there need not be a bound on the size of ${\sf List}_{\eps}[f]$ in terms of $\eps$; it may well be the case that
its size grows with the dimension $n$. Nevertheless, the following lemma states that within ${\sf List}_{\eps}[f]$, we may find a short list whose size
is bounded in terms of $\eps$ that essentially captures all of ${\sf List}_{\eps}[f]$.
\begin{lemma}\label{lem:net_of_prods}
  Fix $\Sigma$, a distribution $\mu$ over $\Sigma$ and a class $\mathcal{F} = \mathcal{F}(\Sigma)$ as in Theorem~\ref{thm:restriction_inverse},
  let $f\colon (\Sigma^n,\mu^{n})\to\mathbb{C}$ be such that $\norm{f}_2\leq 1$, and let $\eps,\delta>0$ be parameters satisfying that $\delta<\eps^2$.
  Then one may find ${\sf ShortList}_{\eps,\delta}[f]\subseteq {\sf List}_{\eps}[f]$ such that:
  \begin{enumerate}
    \item $\card{{\sf ShortList}_{\eps,\delta}[f]}\leq \frac{1}{\eps^2-\delta}$.
    \item For all $p\in {\sf List}_{\eps}[f]$ there is $p'\in {\sf ShortList}_{\eps,\delta}[f]$
    such that $\card{\inner{p}{p'}}\geq \delta$.
  \end{enumerate}
\end{lemma}
\begin{proof}
  The proof is by a greedy algorithm. Starting with ${\sf ShortList}_{\eps,\delta}[f] = \emptyset$, so long as there is $p\in {\sf List}_{\eps}[f]$
  that is at most $\delta$-correlated with all functions in ${\sf ShortList}_{\eps,\delta}[f]$, we add it to ${\sf ShortList}_{\eps,\delta}[f]$.

  We show that the above process terminates after less than $k = \frac{1}{\eps^2-\delta}$ steps. Indeed, otherwise we would be able to find
  $p_1,\ldots,p_{k+1}\in {\sf List}_{\eps}[f]$
  whose pairwise correlations are at most $\delta$. Write $\inner{f}{p_j} = \theta_j \rho_j$ where $\rho_j \geq 0$ and $\theta_j$ has absolute value $1$;
  then we have that $\rho_j\geq\eps$. We get that
  \[
  k\eps
  \leq \sum\limits_{j=1}^{k}\theta_j^{-1}\inner{f}{p_j}
  =\inner{f}{\sum\limits_{j=1}^{k}\theta_j^{-1} p_j}
  \leq\norm{f}_2\norm{\sum\limits_{j=1}^{k}\theta_j^{-1} p_j}_2.
  \]
  Using $\norm{f}_2\leq 1$ and the fact that
  \[
  \norm{\sum\limits_{j=1}^{k+1}\theta_j^{-1} p_j}_2^2
  =\sum\limits_{j=1}^{k}\norm{p_j}_2
  +\sum\limits_{j\neq j'}\theta_{j}^{-1}\theta_{j'}^{-1}\inner{p_j}{p_{j'}}
  \leq k + k(k-1)\delta
  \]
  gives $k\eps\leq \sqrt{k+k^2\delta}$, and simplifying finishes the proof.
\end{proof}
Note that the short-list found in Lemma~\ref{lem:net_of_prods} is not unique, and indeed there may be several choices for it. We will want to
think of some canonical short-list that is associated with a given function $f$. This may be achieved in several ways; for instance, we may
fix a total ordering among all product functions, and then consider some induced ordering (say, lexicographic ordering) induced on collections of functions,
and define the canonical short list of $f$ to be the short list which is first according to this ordering. Thus, henceforth when we write
${\sf ShortList}_{\eps,\delta}[f]$, we refer to the canonical short list of $f$.

Lemma~\ref{lem:net_of_prods} gives us a rather satisfactory answer in the sense that we get a short list of product functions that, in a sense, encapsulates
within it the entire list of product functions correlated with $f$. We will want to imagine this short list as the center of Hamming balls, and of this property
as saying that any product function correlated with $f$ is inside a ball of small radius around some product function from the short list. To facilitate that,
we define the symbolic distance between product functions.

\begin{definition}
  Fix $\Sigma$, a distribution $\mu$ over $\Sigma$ and a class $\mathcal{F} = \mathcal{F}(\Sigma)$ as in Theorem~\ref{thm:restriction_inverse}.
  For two functions $p,p'\in\mathcal{F}$ over $n$-variables $p,p'\colon \Sigma^n\to\mathbb{C}$ written as $p(x) = \prod\limits_{j=1}^{n}p_j(x_j)$
  and $p'(x) = \prod\limits_{j=1}^{n}p'_j(x_j)$, we define the symbolic distance between $p$ and $p'$ as:
  \[
  \Delta_{{\sf symbolic}}(p,p') = \card{\sett{j\in [n]}{p_j\neq p_j'}}.
  \]
\end{definition}

The following lemma asserts that products functions that are correlated are close in symbolic distance.
\begin{lemma}\label{lem:corr_to_symbolic}
  Fix $\Sigma$, a distribution $\mu$ over $\Sigma$ and a class $\mathcal{F} = \mathcal{F}(\Sigma)$ as in Theorem~\ref{thm:restriction_inverse}
  which is $\tau$-separated. Then for any $n\in\mathbb{N}$ and any $n$-variable functions $p,p'\in\mathcal{F}$ we have that
  \[
  \card{\inner{p}{p'}}\leq (1-\tau)^{\Delta_{{\sf symbolic}}(p,p')}.
  \]
\end{lemma}
\begin{proof}
  Writing $p(x) = \prod\limits_{j=1}^{n}p_j(x_j)$ and $p'(x) = \prod\limits_{j=1}^{n}p'_j(x_j)$, we have that
  \[
  \card{\inner{p}{p'}} = \card{\prod\limits_{j=1}^{n}\inner{p_j}{p_j'}}=\prod\limits_{j=1}^{n}\card{\inner{p_j}{p_j'}}.
  \]
  For $j$ such that $p_j = p_j'$ the inner product is $1$, and for any other $j$ the absolute value of the inner product is at most $1-\tau$.
  The result follows.
\end{proof}

\subsubsection{Local Structure}
Fix $\Sigma$, $\mu$ distribution over $\Sigma$, $\mathcal{F} = \mathcal{F}(\Sigma)$, $f$ and the parameters in Theorem~\ref{thm:restriction_inverse}.
We will use the parameters
\[
0\ll c_4\ll c_3\ll c_2 \ll c_1\ll s,C^{-1}\ll \tau, m^{-1},\alpha\leq 1,
\]
take
\[
0\leq \eta = \eps^{\frac{1}{c_3\rho}} \leq \zeta = \eps^{\frac{1}{c_2\rho}}\leq\xi = \eps^{1/c_1} \leq 1,
\]
and
\begin{equation}\label{eq:rest_inverse_eps'}
  \eps' = \zeta^{-\log(1/\rho)^2/c_4}.
\end{equation}

Throughout, we denote by $I\subseteq [n]$ a set which is picked by including each $i\in [n]$ in it with probability $\rho$. For $I\subseteq [n]$
and $z\in \Sigma^I$ define
\[
W_{I,z}= {\sf List}_{\eps}[f_{\overline{I}\rightarrow z}],
\qquad\qquad
\tilde{W}_{I,z}= {\sf List}_{\eps/2}[f_{\overline{I}\rightarrow z}],
\]
and
\[
SW_{I,z}= {\sf ShortList}_{\eps,\eps^2/10}[f_{\overline{I}\rightarrow z}],
\qquad\qquad
\tilde{SW}_{I,z}= {\sf ShortList}_{\eps/2, \eps^2/100}[f_{\overline{I}\rightarrow z}].
\]

We will now be interested in looking at $I'\subseteq_{\rho/2}[n]$. An equivalent way of sampling such $I'$ is by first taking
$I\subseteq_{\rho}[n]$, and then taking $I'\subseteq_{1/2}I$. As for restrictions, we can also first take $z\sim {\nu'}^{\overline{I}}$,
then $z'\sim \nu^{I\setminus I'}$ and concatenate them to get $f_{\overline{I'}\rightarrow z\circ z'}$.

For fixed $I$ and $z$, if we have that $p(x) = \prod\limits_{i\in I}p_i(x_i)$ is in ${SW}_{I,z}$, then
by Fact~\ref{fact:restrict_keeps_correlations} we get that
\[
\card{\inner{f_{\overline{I'}\rightarrow z\circ z'}}{p|_{I\setminus I'\rightarrow z'}}}\geq \frac{\eps}{2}
\]
with probability at least $\frac{\eps^2}{4}$, in which case we get that $p|_{I'}\in \tilde{W}_{I',z\circ z'}$,
so that in a formula the above says that
\begin{equation}\label{eq:restrict_inverse_1}
\cProb{\substack{I, I'\\z, z'}}{p\in {SW}_{I,z}}{p|_{I'}\in \tilde{W}_{I',z\circ z'}}\geq\frac{\eps^2}{4}.
\end{equation}
This gives us some interesting information: while the identity of the original product function $p$ may depend on both $I$ and $z$,
the function $p|_{I'}$ depends only on $I, I'$ and $z$ and not on $z'$. This suggests that on average, the identity of $p$ itself should
also only depend on $I$ (the set of live variables) and not really on the value that we restrict outside them. To formalize this we use~\eqref{eq:restrict_inverse_1}
in conjunction with Fact~\ref{fact:holder_trick} to get that
\begin{equation}\label{eq:restrict_inverse_2}
\cProb{\substack{I, z \\ I',I''\subseteq_{1/2} I\\ z_1', z_1'',z_2',z_2''}}
{p\in {SW}_{I,z}}
{p|_{I'}\in \tilde{W}_{I',z\circ z_1'}\cap \tilde{W}_{I',z\circ z_2'},
~~~~~
p|_{I''}\in \tilde{W}_{I'',z\circ z_1''}\cap  \tilde{W}_{I'',z\circ z_2''}}\geq\frac{\eps^{16}}{4^4}.
\end{equation}
We note that in particular, this inequality means that with noticeable probability,
the short-lists of $I',z\circ z_1'$ and $I',z\circ z_2'$ contain two elements that have symbolic distance at most $C\cdot \log(1/\eps)$.
Indeed, as $p|_{I''}$ is in the list of these two restrictions, it follows by definition of the short-lists that for each one of this restrictions,
the corresponding short list contains a product function that has correlation at least $\eps^2/100$ with it, at which point we can use Lemma~\ref{lem:corr_to_symbolic}
to say that $p|_{I''}$ is close in symbolic distance to an element in each short list.

Take $T\subseteq [n]$ randomly by including each element in it with probability $\beta=\frac{1}{C\log(1/\eps)}$, and
define the following events:
\begin{enumerate}
  \item Let $E_1(p)$ be the event that $p|_{I'\cap T}\in \tilde{SW}_{I',z\circ z_1'}|_{T}\cap \tilde{SW}_{I',z\circ z_2'}|_{T}$.
  \item Let $E_2(p)$ be the event that $p|_{I''\cap T}\in \tilde{SW}_{I'',z\circ z_1''}|_{T}\cap \tilde{SW}_{I'',z\circ z_2''}|_{T}$.
  \item Let $E(p)$ be the event that $p\in {SW}_{I,z}$.
\end{enumerate}
With these notations, the inequality~\eqref{eq:restrict_inverse_2} gives that
\begin{equation}\label{eq:restrict_inverse_3}
\cProb{\substack{I, z \\ I',I''\subseteq_{1/2} I\\ z_1', z_1'',z_2',z_2''\\T\subseteq_{\beta}[n]}}
{E(p)}
{E_1(p)\cap E_2(p)}
\geq\frac{\eps^{16}}{4^4e}.
\end{equation}

Define $\mathcal{D} = \frac{1-\rho}{1-/\rho/2}\nu' + \frac{\rho/2}{1-\rho/2}\nu$, and note that choosing $I,I'$ and $z, z'$ as above, the
distribution of $z\circ z'$ is exactly $\mathcal{D}$. Let
\[
W_{I',T} = \sett{p\in \mathcal{F}(\Sigma,I'\cap T)}{\Prob{z\sim \mathcal{D}^{\overline{I'}}}{p|_{I'\cap T}\in \tilde{SW}_{I',z}|_{T}}\geq \zeta}.
\]
\begin{claim}\label{claim:use_sse_inverse_restriction}
  For all $I'$ and $T$ we have that
  \[
  \Prob{\substack{\overline{I'} = \overline{I}\cup (I\setminus I')\\ z,z',z''}}
  {\exists p\in \tilde{W}_{I\cup I'\setminus I,z\circ z'}|_{T}\cap \tilde{W}_{I\cup I'\setminus I,z\circ z''}|_{T},
  p\not\in W_{I',T}}\leq \sqrt{\xi}.
  \]
\end{claim}
\begin{proof}
  For each $p\in\mathcal{F}(\Sigma,I'\cap T)$, define $X_p = \sett{z\in\Sigma^{\overline{I'}}}{p\in \tilde{SW}_{I',z}|_{T}}$. Note that
  the condition that $p\not\in W_{I',T}$ is equivalent to $\mathcal{D}(X_p)< \zeta$, and also that
  \begin{equation}\label{eq:restrict_inverse_4}
  \sum\limits_{p} \mathcal{D}(X_p)
  \leq
  \sum\limits_{p\in\mathcal{F}(\Sigma,I'\cap T)}\sum\limits_{z}\mathcal{D}(z) 1_{p\in \tilde{SW}_{I',z}|_{T}}
  \leq \sum\limits_{z}\mathcal{D}(z) \card{\tilde{SW}_{I',z}|_{T}}
  \leq O\left(\frac{1}{\eps^2}\right),
  \end{equation}
  where we used the fact that the short-lists as defined above have size at most $O(1/\eps^2)$.
  Now consider the distribution over $z\circ z'$ and $z\circ z''$ as in the claim, and note that it is a product distribution $\tilde{\mu}^{I'}$
  which has full support and the probability of each atom is at least $\Omega(\rho)$. Thus, denoting by $\tilde{\mathrm{T}}$ the averaging operator
  corresponding to $\tilde{\mu}$, we get by Lemma~\ref{lem:mossel_MC} that $\lambda_2(\tilde{\mathrm{T}})\leq 1-\Omega_{\alpha,m}(\rho)$. Thus, noting that
  $\mathcal{D}$ is a stationary distribution for $\tilde{\mathrm{T}}$ and using the Efron-Stein decomposition we get that the probability of the
  left hand side of the claim can be written as
  \[
  \sum\limits_{p\not\in W_{I',T}}{\inner{1_{X_p}}{\tilde{\mathrm{T}}^{\otimes I'} 1_{X_p}}}
  =
  \sum\limits_{p\not\in W_{I',T}} \sum\limits_{d=0}^{n}{\inner{1_{X_p}^{=d}}{\tilde{\mathrm{T}}^{\otimes I'} 1_{X_p}^{=d}}}
  \leq
  \sum\limits_{p\not\in W_{I',T}} \sum\limits_{d=0}^{n}\lambda_2(\tilde{\mathrm{T}})^d\norm{1_{X_p}^{=d}}_2^2.
  \]
  Setting $D = \frac{1}{\rho}\log(1/\xi)$, we get that for $d\leq D$ the contribution is at most
  \[
  \sum\limits_{p\not\in W_{I',T}} \sum\limits_{d=0}^{D}\norm{1_{X_p}^{=d}}_2^2
  =
  \sum\limits_{p\not\in W_{I',T}} \norm{1_{X_p}^{\leq D}}_2^2
  \leq
  \sum\limits_{p\not\in W_{I',T}} 2^{O_{\alpha}(D)}\mu(X_p)^{3/2}
  \leq 2^{O_{\alpha}(D)}\sqrt{\zeta}\sum\limits_{p\not\in W_{I',T}} \mu(X_p)
  \leq \frac{\sqrt{\xi}}{2},
  \]
  where we used Fact~\ref{fact:level_d_inequality}, Parseval and~\eqref{eq:restrict_inverse_4}. For $d>D$, the contribution is at most
  \[
  (1-\Omega_{\alpha,m}(\rho))^{D}\sum\limits_{p\not\in W_{I',T}} \sum\limits_{d=D}^{n}\norm{1_{X_p}^{=d}}_2^2
  \leq
  (1-\Omega_{\alpha,m}(\rho))^{D}\sum\limits_{p\not\in W_{I',T}} \norm{1_{X_p}}_2^2
  \leq
  (1-\Omega_{\alpha,m}(\rho))^{D}O\left(\frac{1}{\eps^2}\right)
  \leq \frac{\sqrt{\xi}}{2},
  \]
  where we used Parseval and~\eqref{eq:restrict_inverse_4} again. The proof is thus concluded by summing up the two inequalities.
\end{proof}

Note that the premise of the theorem says that with probability at least $\eps$ over the choice of $I$ and $z$ we have that
$W_{I,z}$ is non-empty, and hence by the definition of short lists we get that $\Prob{}{\exists p, E(p)}\geq \eps$.
In conjunction with~\eqref{eq:restrict_inverse_3} we conclude that
\[
\Prob{\substack{I\subseteq_{\rho}[n]\\I',I''\subseteq_{1/2}I\\ T\subseteq_{\beta} [n]}}{\exists p, E_1(p)\cap E_2(p)\cap E(p)}\geq \frac{\eps^{17}}{4^4e}.
\]
By Claim~\ref{claim:use_sse_inverse_restriction}, sampling $T\subseteq_{\beta}[n]$, the probability that $E_1(p)$ holds but
$p|_{I'\cap T}\not\in W_{I',T}$ is at most $\sqrt{\xi}$ and similarly for $E_2(p)$, so we get that
\begin{equation}\label{eq:restrict_inverse_5}
\Prob{\substack{I\subseteq_{\rho}[n], z\sim {\nu'}^{\overline{I}}\\I',I''\subseteq_{1/2}I\\ T\subseteq_{\beta} [n]}}
{\exists p, E(p), p|_{I'\cap T}\in W_{I',T}, p|_{I''\cap T}\in W_{I'',T}}\geq \frac{\eps^{17}}{4^4e} - 2\sqrt{\xi}\geq \eps^{18}.
\end{equation}

\subsubsection{Designing the Direct Product Function}
We are going to use~\eqref{eq:restrict_inverse_5} to define a strategy for the direct product test. Towards this end, we first define the
collection of $I'$'s on which we are going to assign a value which will give us a decent acceptance probability.
\begin{definition}
  We say $I'$ is good if there is $p\in \mathcal{F}(\Sigma,I')$ such that $\Prob{T\subseteq_{\beta} [n]}{p|_{I'\cap T}\in W_{I',T}}
    \geq \xi$.
\end{definition}
The following claim asserts that there is a significant fraction of good $I'$'s.
\begin{claim}\label{claim:restriction_inverse_good}
  $\Prob{I'\subseteq_{\rho/2}[n]}{I'\text{ is good}}\geq \frac{\eps^{18}}{2}$.
\end{claim}
\begin{proof}
Note that by~\eqref{eq:restrict_inverse_5}, with probability at least $\eps^{18}/2$ over the choice of $I, I'$ and $z$ we get that
\[
\Prob{T\subseteq_{\beta} [n]}
{\exists p, E(p), p|_{I'\cap T}\in W_{I',T}}\geq \frac{\eps^{18}}{2},
\]
and as there are at most $O(1/\eps^2)$ product functions $p$ for which $E(p)$ holds we get that there is $p$ such that
\[
\Prob{T\subseteq_{\beta} [n]}
{p|_{I'\cap T}\in W_{I',T}}\geq \Omega(\eps^{20}) \geq \xi,
\]
hence $I'$ is good.
\end{proof}

Define the following randomized strategy $F$ for the direct product test. For each good $I'$:
\begin{itemize}
  \item Choose $z'\sim \mathcal{D}^{I'}$.
  \item Sample $p\in \tilde{SW}_{I',z'}$ uniformly among the $p$'s satisfying that $\Prob{T\subseteq_{\beta} [n]}{p|_{I'\cap T}\in W_{I',T}}
    \geq \xi$ and set $F[I']= p$.
\end{itemize}
For $I'$ which is not good, we choose $F[I']$ as a random string in $[R]^{I'}$. We now show two claims. The first of which asserts that in expectation,
the above randomized strategy passes the direct product test with significant probability.

\begin{claim}\label{claim:dp_passes_inverse}
  $\Expect{}{\Prob{}{F\text{ passes }{\sf DP}(\rho,1/2,\beta)}}\geq \Omega(\eps^{20} \zeta^2)\geq \zeta^3$.
\end{claim}
\begin{proof}
  Note that the left hand side of the claim can be written as
  \[
  \Expect{}
  {\Expect{\substack{I\subseteq_{\rho}[n], T\subseteq_{\beta}[n]\\I',I''\subseteq_{1/2}I\\z'\sim \mathcal{D}^{I'}, z''\sim\mathcal{D}^{I''}}}
  {\Expect{\substack{p'\in  \tilde{SW}_{I',z'}\\ p'' \in \tilde{SW}_{I'',z'}}}{1_{I', I''\text{ are good} }1_{p'|_{T} = p''|_{T}}}}}.
  \]
  Looking at~\eqref{eq:restrict_inverse_5}, we get that with probability at least $\eps^{18}/2$ over the choice of $I,z$ and $I',I''$ it holds that
  \[
    \Prob{T\subseteq_{\beta} [n]}
    {\exists p, E(p), p|_{I'\cap T}\in W_{I',T}, p|_{I''\cap T}\in W_{I'',T}}\geq \frac{\eps^{18}}{2},
  \]
  so as in the proof of Claim~\ref{claim:restriction_inverse_good} there is $p$ such that
  \[
    \Prob{T\subseteq_{\beta} [n]}
    {p|_{I'\cap T}\in W_{I',T}, p|_{I''\cap T}\in W_{I'',T}}\geq \Omega(\eps^{20}),
  \]
  and in particular $I'$ and $I''$ are good. Also, it follows that $p$ is a candidate for both $F[I']$ and $F[I'']$ in the
  above randomized strategy, and by definition of $W_{I',T}$ we get that $p'$ chosen in the randomized strategy satisfies that $p'|_{T} = p|_{I'\cap T}$
  with probability at least $\zeta$; the same goes for $p''$. Hence under the randomness of the choice of $F$ the test passes with
  probability at least $\zeta^2$. The conclusion follows.
\end{proof}

The second claim asserts that the contribution of the assignment $F$ on non-good parts to the agreement of $F$ with any direct product function is small.
\begin{claim}\label{claim:dp_not_get_from_not_good}
  For all $r\leq n/2$,
  \[
    \Prob{F}{\exists g\colon [n]\to[R],~\Prob{I'}{I'\text{ not good and } \Delta(g|_{I'}, F[I'])\leq r}\geq \frac{1}{2}\eps'}\leq 2^{-n}.
  \]
\end{claim}
\begin{proof}
  Fix $g\colon [n]\to[R]$. Note that for each $I'$ which is not good, the probability that $\Delta(g|_{I'}, F[I'])\leq r$
  is at most $n^r R^{r-n}$. Let $Z_{I'}$ be the indicator random variable of this event, so that $Z_{I'}$ are independent and
  each of them has expectation at most $n^r m^{r-n}$. Thus, we get by Chernoff's inequality that
  \[
        \Prob{}{\sum\limits_{I'\text{ not good}} Z_{I'}\geq \frac{1}{2}\eps' N}\leq 2^{-\Omega(\eps' N)}
  \]
  where $N$ is the number of $I'$ which are not good. Taking a union bound over $g$, we get that the probability in question is at most
  $R^n 2^{-\Omega(\eps' N)}\leq 2^{-\Omega(\eps' N)}\leq 2^{-n}$ as required.
\end{proof}

\subsubsection{Applying the Direct Product Theorem}
We are now in a position to invoke the direct product theorem to finish the proof.

By Claim~\ref{claim:dp_passes_inverse} and an averaging argument, with probability at least $\zeta^3/2$ over the choice of $F$ we
have that $F$ passes the direct product test with probability at least $\zeta^3/2$. Hence, by Claim~\ref{claim:dp_not_get_from_not_good}
and the union bound we have that with probability at least $\zeta^3/2 - 2^{-n}\geq \zeta^3/4$ it holds that $F$ passes the direct
product test with probability at least $\zeta^3/2$ and the contribution of the not-good $I'$'s to the agreement of $F$ with any
direct product function is at most $\frac{1}{2}\eps'$. We fix such choice of $F$.

Using Theorem~\ref{thm:DP_biased_version} we conclude that there is $g\colon [n]\to [R]$ such that
\[
    \Prob{I'}{\Delta(g|_{I'}, F[I'])\leq r}\geq \eps'
\]
for $r = {\sf poly}\left(\frac{\log(1/\zeta)}{\rho}\right)$ and $\eps'$ as in~\eqref{eq:rest_inverse_eps'}. Thus, we get that $\Prob{I'}{I'\text{ is good}, \Delta(g|_{I'}, F[I'])\leq r}\geq \eps'/2$.

Let $p_g\in\mathcal{F}(\Sigma,[n])$ be the product function corresponding to $g$. We note that for any good $I'$ such that
$\Delta(g|_{I'}, F[I'])\leq r$, by the definition of the randomized strategy we have that choosing $T\subseteq_{\beta}[n]$
and $T'\subseteq_{1/r} T$ it follows that $F[I']\in W_{I',T}$ with probability at least $\xi$,
and that conditioned on $T$ we have that $g|_{I'\cap T'} = F[I']|_{T'}$ with probability at least $\Omega(1)$.
Hence both events hold with probability at least $\Omega(\xi)$, and we fix such choices of $T$ and $T'$ henceforth.

By definition of $W_{I', T}$ we get that as $F[I']\in W_{I', T}$, choosing $z\sim \mathcal{D}^{\overline{I'}}$
we have $F[I']|_{T} \equiv {p_z}|_{T}$ for some $p_z\in \tilde{SW}_{I', z}$ with probability at least $\zeta$. For such $z$ it holds that
$\card{\inner{{p_z}}{f_{\overline{I'}\rightarrow z}}}\geq \eps/2$, and so by Fact~\ref{fact:restrict_keeps_correlations} we get that choosing $z'\sim \nu^{I'\setminus T}$
it holds that
\[
\card{\inner{{p_z}|_{I'\cap T}}{f_{\substack{\overline{I'}\rightarrow z\\ I'\setminus T\rightarrow z'}}}}\geq \frac{\eps}{4}
\]
with probability at least $\eps^2/16$. In this case we get
\[
\card{\inner{F[I']|_{T}}{f_{\substack{\overline{I'}\rightarrow z\\ I'\setminus T\rightarrow z'}}}}\geq \frac{\eps}{4}.
\]
Choosing $z''\sim \mu^{T\setminus T'}$, we get by Fact~\ref{fact:restrict_keeps_correlations} that
\[
\card{\inner{{F[I']|_{T}}_{T\setminus T'\rightarrow z''}}{f_{\substack{\overline{I'}\rightarrow z\\ I'\setminus T\rightarrow z'\\T\setminus T'\rightarrow z''}}}}\geq \frac{\eps}{8}
\]
with probability at least $\eps^2/64$, in which case it follows that
\[
\card{\inner{{{p_g}|_{T'}}}{f_{\substack{\overline{I'}\rightarrow z\\ I'\setminus T\rightarrow z'\\T\setminus T'\rightarrow z''}}}}\geq \frac{\eps}{8}.
\]
In conclusion, we get that
\[
\Expect{\substack{I', T, T'\\ z,z',z''}}
{\card{\inner{{{p_g}|_{T'}}}{f_{\substack{\overline{I'}\rightarrow z\\ I'\setminus T\rightarrow z'\\T\setminus T'\rightarrow z''}}}}^2}\ggg \eps'\xi\zeta\eps^2.
\]
Thus, looking at $G = p_g f$ we get that
\[
\Expect{\substack{I', T, T'\\ z,z',z''}}{
{\sf Stab}_{1/2}\left(G_{\substack{\overline{I'}\rightarrow z\\ I'\setminus T\rightarrow z'\\T\setminus T'\rightarrow z''}}\right)
}
\geq
\Expect{\substack{I', T, T'\\ z,z',z''}}
{\card{\E[G_{\substack{\overline{I'}\rightarrow z\\ I'\setminus T\rightarrow z'\\T\setminus T'\rightarrow z''}}]}^2}\ggg \eps'\xi\zeta\eps^2.
\]
We now consider the distribution over the restriction, and note that it is equivalent to a restriction that chooses $J\subseteq_{\frac{\rho}{2}\cdot\beta\cdot\frac{1}{r}} [n]$
and fixes $\overline{J}$ according to $(\mathcal{D}')^{\overline{J}}$, where $\mathcal{D}'$ is some mixture of the distributions $\mathcal{D}$ and $\nu$. Thus, re-writing
the above inequality we get that
\[
\Expect{\substack{J\subseteq_{\frac{\rho}{2}\cdot\beta\cdot\frac{1}{r}} [n] \\ y\sim (\mathcal{D}')^{\overline{J}}}}
{{\sf Stab}_{1/2}\left(G_{\overline{J}\rightarrow y}\right)}
\ggg \eps'\xi\zeta\eps^2.
\]

Using Lemma~\ref{lem:op_comparison_lemma} we get that the left hand side is at most
${\sf Stab}_{1-A\cdot \rho\beta/r}(G)$ where $A$ depends only on $\alpha$ and $m=\card{\Sigma}$, so
${\sf Stab}_{1-A\cdot \rho\beta/r}(G)\ggg \eps'\xi\zeta\eps^2$. It follows from Fact~\ref{fact:stability_to_weight}
that $W_{\leq A'\cdot \frac{r\log(1/\eps'\xi\zeta\eps)}{\rho\beta}}[G]\ggg \eps'\xi\zeta\eps^2$, where $A'$ depends only on $\alpha$ and $m$.
Thus, $L = (p_g f)^{\leq A'\cdot \frac{r\log(1/\eps'\xi\zeta\eps)}{\rho\beta}}$ is a function of degree at most $A'\cdot \frac{r\log(1/\eps'\xi\zeta\eps)}{\rho\beta}$
and $2$-norm $1$ such that
\[
\inner{f}{\overline{p_g}L}=\inner{G}{L} = W_{\leq A'\cdot \frac{r\log(1/\eps'\xi\zeta\eps)}{\rho\beta}}[G]\ggg \eps'\xi\zeta\eps^2,
\]
and we are done.\qed

\section{Proof of the Direct Product Theorem}\label{sec:dp}

\newcommand{\poly}{\mathsf{poly}}
\newcommand{\maintest}{{\em \textrm{Agreement-Test}}\xspace}
\newcommand{\test}{{\em \textrm{Modified-Test}}\xspace}
\newcommand{\subsetmaintest}{{\em \textrm{Subset-Agreement-Test}}\xspace}
	
In this section, we prove Theorem~\ref{thm:DP_biased_version}. This result improves can be seen as a quantitative improvement over the direct product theorem from~\cite{BKMcsp3}.  To prove this result, we will first present and analyze a uniform version of this direct product tester; by that, we mean that we are going to be given an assignment to sets of
size precisely $k = \rho n$, and that the various intersection sizes in the set are all replaced with exact intersection sizes.
In Sections~\ref{sec:product_1} and ~\ref{sec:product_2}, we will show how to derive the direct product theorem when the underlying distribution is a product distribution as required for Theorem~\ref{thm:DP_biased_version}, using the well known trick of ``going to infinity and back''.
	
\paragraph{Notation.}Fix an alphabet $[m]=\{0,1, 2, \ldots, m-1\}$. Given a string $x\in [m]^n$ and a subset $S\subseteq [n]$, we use the notation $x|_{S}$ to denote the part of the string $x$ restricted to the set $S$.  Given $x, y\in [m]^n$, we use the notation $x \overset{\geq t}{\neq} y$ to denote that the set $\{ i\in [n] \mid x_i \neq y_i\}$ is of size at least $t$. Similarly, we use $x \overset{\leq t}{\neq} y$ to denote that the strings $x$ and $y$ differ in at most $t$ locations.

	\paragraph{The set up for the direct product testing.} Fix $q, q'\in (0,1)$ such that $q'<q$ and an integer $t\in\mathbb{N}$. Let $q''= q-q'$. Suppose we are given a table $F: {[n]\choose qn} \rightarrow [m]^{qn}$ where $F[S] \in [m]^{qn}$ can be thought of as assigning a symbol from $[m]$ to every element in $S$ (by associating some fixed ordering on the elements of $[n]$). Consider the  agreement test (\maintest) parameterized by $(q, q', t)$ given in Figure~\ref{fig:maintest}. Let $\mathcal{D}_{q,q'}$ be the distribution associated with the pair $(A_0\cup B_0, A_0\cup B_1)$ in the test.
	\begin{figure}[!h]
		\label{fig:maintest}
		\fbox{
			
			\parbox{450pt}{
				\vspace{10pt}
				Given $F: {[n]\choose qn} \rightarrow [m]^{qn}$,
				\begin{itemize}
					\item Pick a random set $A_0\cup B_0$ of size $qn$ where $|A_0| = q'n$
					\item Select a random set $B_1\subseteq [n]\setminus A_0$ of size $q''n$
					\item Check if $F[A_0\cup B_0]|_{A_0} \overset{\leq t}{\neq}  F[A_0\cup B_1]|_{A_0}$
				\end{itemize}
			}
		}
		\caption{\maintest with parameters $(q, q', t)$.}
	\end{figure}
	
	It is clear that if the table $F$ comes from a global string $a\in [m]^n$, in the sense that there is a vector $a$ such that $F[S] = a|_{S}$ for all $S$, then
    the test accepts with probability $1$ (even when $t =0$). The following result is an inverse type result to this statement in the small soundness regime,
    and as discussed in the introduction in this case there are several challenging examples. In the following theorem, we prove that the type of assignments discussed
    in the introduction are essentially the only assignments that pass the test with non-negligible probability.
	\begin{thm}
		\label{thm:DPtest}
        There exists $c>0$ such that the following holds for sufficiently large $n\in\mathbb{N}$ and $\eps\geq 2^{-n^{c}}$.

		Fix an alphabet $[m]$ and  $c_0 > 0$ be any constant. For all $\eps>0$, $0<q'< q<1$ such that $q'\leq \frac{9q}{10}$ and
        $T = \left(\frac{\log(1/\eps)}{q'}\right)^{c_0}$,  suppose that $F: {[n]\choose qn}\rightarrow [m]^{qn}$ satisfies
		\[
		\Prob{(S_1,S_2)\sim \mathcal{D}_{q,q'}}{F[S_1]|_{S_1\cap S_2} \overset{\leq T}{\neq} F[S_2]|_{S_1\cap S_2}}\geq \eps.
		\]
		Then there exists a function $g: [n]\rightarrow[m]$ such that for at least an $\eps^{O(\log (1/q')^2)}$ fraction of $S\in {[n] \choose qn}$, we have
        $|\{ i\in S\mid F[S]_i\neq  g(i)\}| \leq  \left(\frac{T\log{(1/\eps)}}{q'}\right)^{O(1)}$.\end{thm}
    We start by giving a high level overview of the proof of Theorem~\ref{thm:DPtest}.  The overall argument is similar to the one from~\cite{BKMcsp3}, and in order to
    improve upon the quantitative bounds therein we require a more careful analysis, as well as more explicit quantitative bounds small set expansion result on
     the multi-cube graph (which plays an important role in our proof)
	
	\begin{enumerate}
		\item {\bf Getting a local structure:} Suppose the table $F$ passes the \maintest with probability at least $\eps$. For a pair $(A,B)$, we define a function $g_{A,B} : [n]\rightarrow [m]$ by taking the plurality vote among $F[A\cup B']|_{i}$ where $(A, B')$ is such that $F[A\cup B]|_{A}$ agrees with $F[A\cup B']|_{A}$ on all but $t$ coordinates. We then show that for a typical $(A_0,B_0)$, the function $g_{A_0,B_0}$ agrees with the table $F[A_0\cup B']$, on all but $\poly(t)$ coordinates, for at least $\eps^{O(1)}$ fraction of the $B'$s. Since in this step we can only show that $g_{A_0,B_0}$ only agrees sets of the form $A_0\cup B'$, which only constitutes $\exp(-n)$ fraction of the sets from ${[n]\choose qn}$, we somehow need to make sure that these different $g_{A_0, B_0}$s indeed agree with each other.
		
\item {\bf Establishing consistency between various local views:} In this step, we show that for a typical pairs $(A_0, B_0)$ and $(A'_0, B'_0)$, the functions $g_{A_0, B_0}$ and $g_{A'_0, B'_0}$ agree with each other on all but $t^{O(1)}$ many coordinates. We show this in two steps. In the first step, we study a slightly different test that we call the \test. The purpose of  analyzing this test is to conclude that for a typical pair $(A_0, B_0)$, if we select a pair $(\tilde{A}_0, B'_0)$, where $\tilde{A}_0$ is a correlated copy of $A_0$ and $B'_0$ in independent of $(A_0, B_0)$, then with noticeable probability the functions $g_{A_0, B_0}$ and $g_{\tilde{A}_0, B'_0}$ agree on all but $t^{O(1)}$ many coordinates. In the second step, we show how to break the correlation between $\tilde{A}_0$ and $A_0$. Here, we appeal to the small set expansion property of a graph defined over a multi-slice in $[3]^n$. Towards this end,  we think of $(A,B)$ as an element $x \in [3]^n$ in the following way: $x_i=1$ if $i\in A$, $x_i=2$ if $i\in B$ and $x_i=0$ otherwise. The edges of the graph are given by the distribution on the pairs $(A_0, B_0)$ and $(\tilde{A}_0, B'_0)$. We show that this graph is a small-set expander using techniques from~\cite{BravermanKLM22}. As every small set in this graph expands, we use this to show that for typical pairs $(A_0, B_0)$ and $(A'_0, B'_0)$, the functions $g_{A_0, B_0}$ and $g_{A'_0, B'_0}$ agree with each other on all but $\poly(\eta)$ many coordinates. Finally, this means that there exists a pair $(A^\star, B^\star)$ such that $g_{A^\star, B^\star}$ satisfies the conclusion of Theorem~\ref{thm:DPtest}.
		
\item {\bf From uniform setting to a product distribution:} We take $N$ which is significantly larger than $n$, and attempt to simulate the $q$-biased distribution
over $\{0,1\}^n$ by the uniform distribution over $qN$ sized subsets of $[N]$. More precisely, given a function $G: (P[n], \mu^{\otimes n}) \rightarrow [m]^{\leq n}$, we define a map $\tilde{G} : {[N] \choose qN} \rightarrow [m]^{qN}$ as follows: for $S\in {[N] \choose qN}$, define $\tilde{G}(S)|_{S\cap [n]} = G(S\cap [n])$ and $\tilde{G}(S)|_{S\setminus [n]} =  0^{|S\setminus [n]|}$. We note that taking a random $qN$ sized set $S$ from $[N]$, the distribution of $A = S\cap [n]$ is very close to being of a random subset of $[n]$ in which each coordinate is included with probability $q$. Thus, we are able to relate the performance of the (uniform sized set) direct product tester of $\tilde{G}$,
and the performance of the ($q$-biased) direct product tester of $G$. Applying the uniform sized direct product testing result on $\tilde{G}$, we are quick able then
to conclude Theorem~\ref{thm:DP_biased_version}.
	\end{enumerate}

	\subsection{Preliminaries}
	
\subsubsection{A Sampling Lemma}
	Consider the bipartite inclusion graph $G(n, \ell) = G([n]\cup Y, E)$ between $[n]$ and $Y = {[n]\choose \ell}$ for some $1\leq \ell< n$, in which the edge set consists of pairs
$(i,A)$ such that $i\in A$. The following sampling lemma from ~\cite{ImpagliazzoKW12} will be useful for our analysis.
	
	\begin{lemma}
		\label{lemma:sampler}
		Let $G(n,\ell)$ be the inclusion graph for $1\leq \ell< n$. Let $Y' \subseteq Y$ be any subset of measure $\rho <1/2$. For any constant $0 < \nu< 1$, we have that for all but at most $O_\nu\left(\frac{\log 1/\rho}{\ell}\right)$ fraction of vertices $x\in [n]$,
		$$\left| \Pr_{y\in N(x)}[y\in Y'] - \rho\right|\leq \nu\rho.$$
		Here, $N(x)$ is the neighbors of the vertex $x$ in $G$.
	\end{lemma}

\subsubsection{A Small Set Expansion Result}
	For a graph $G(V,E)$, let $\mathcal{T}(G)$ be the Markov operator associated with $G$. Also, let $$\phi_{G}(\mu) := \min_{\substack{S\subseteq V({G}),\\ |S|\leq \mu |V({G})|}} \Pr_{(u,v)\in E({G})}{[v\notin S \mid u\in S]}.$$ Note that if every subset of size at most $\mu$ in $G$ expands, then $\phi_{{G}}(\mu)$ is large.  For any linear operator $T$, its $p\rightarrow q$ norm is defined as $\|T\|_{p\rightarrow q} := \max_{v\neq 0}\frac{\|Tv\|_q}{\|v\|_p}$. We will need the small set expansion property of a graph $\mathcal{G}_n$ defined below.

	\paragraph{The graph $\mathcal{G}_n$.}Consider the graph $\mathcal{G}_n$ induced on the set of vertices $\{(A, B) \mid A, B\subseteq [n], A\cap B = \emptyset, |A| = q'n, |B|=(q-q')n\}$ as follows, defined using a parameter $c$ (that is to be thought of as an absolute but small constant).
A random neighbor $(A', B')$ of $(A, B)$ in this graph is a pair where $A'$ is distributed uniformly conditioned on $|A'| = q'n$ , $|A\cap A'| = cq'n$ and $B'\subseteq[n]\setminus A'$ is a uniformly random set of size $(q-q')n$. An alternative view of this game (which will be crucial for us in order to derive small-set expansion results) proceeds by
viewing it as a graph over the multi-slice -- namely the set of strings in $\{0,1,2\}^n$ with prescribed number of coordinates equal to $0$, $1$ and $2$.
Indeed, we map the vertex $(A, B)$ to the string $\x\in \{0,1,2\}^n$ where $x_i = 1$ if $i\in A$, $x_i = 2$ if $i\in B$ and $x_i=0$ if $i\in [n]\setminus (A\cup B)$. The
edges of the graph $\mathcal{G}_n$ then naturally translate to edges over strings: a random neighbour of $x\in \{0,1,2\}^n$ is a string resulting from choosing a subset of
size $cq'n$ of the $1$'s in $x$, and re-sampling the rest of the coordinates so that the number of coordinates equal to $0$, $1$ and $2$ is as required. By abuse of notation we shall denote this graph also as $\mathcal{G}_n$.

The following result from~\cite{BravermanKLM22} in order to give the following bounds on the expansion of $\mathcal{G}_n$.
	\begin{lemma}
		\label{lemma:sse_prem}
		For every $c>0$,  $0<q'< q<1$ such that $q'\leq \frac{9q}{10}$, and $\mu>0$, the graph $\mathcal{G}_n$ defined above has
		$$\phi_{\mathcal{G}_n}(\mu) \geq 1 - \mu^{\Omega\left(\frac{1}{\log(1/cq')^2}\right)}.$$
	\end{lemma}
We give a proof of this lemma in Section~\ref{section:sse}.
	
	\subsubsection{From Uniform Size Distributions to Biased Distributions}
	We need the following two claims from~\cite{BKMcsp3} to move from the uniform setting to the product setting.
	\begin{claim}
		\label{claim:binom}
		Fix $q\in (0,1)$, $t\leq n$ and $N= \omega(n^4)$. We have,
		$$ \frac{{N-n \choose qN-t}}{{N \choose qN}}  = q^t (1-q)^{n-t} (1\pm o(1)).$$
		where the $o(1)$ factor can be taken as $2^{-\Omega(n)}$.
	\end{claim}
	\begin{proof}
		The proof proceeds by a direct calculation and is deferred to Section~\ref{sec:missing_dp}.
	\end{proof}

    With Claim~\ref{claim:binom} in hand, we can now show a coupling between the uniform distribution over subsets of $[N]$ of fixed size, and a distribution
    close to $q$-biased subsets of $[n]$.
	\begin{claim}
		\label{claim:sd_prod_set}
		Fix $q\in (0,1)$ and $N= \omega(n^4)$. Consider the following two distributions on $P([n])$:
		\begin{itemize}
			\item $\mathcal{D}_1$: Select a subset $A\subseteq [n]$ by including $i\in [n]$ to $A$ with probability $q$ for each $i$ independently.
			\item $\mathcal{D}_2$: Select a random subset $S\subseteq [N]$ of size $qN$ and output $S|_{[n]}$.
		\end{itemize}
		Then, the statistical distance between $\mathcal{D}_1$ and $\mathcal{D}_2$ is at most $e^{-n}$.
	\end{claim}
	\begin{proof}
		We will compare the point-wise probabilities $p_1, p_2: P([n])\rightarrow \mathbb{R}$ assigned by the two distributions $\mathcal{D}_1$ and $\mathcal{D}_2$, respectively. Fix any set $A\subseteq [n]$ of size $t$. We have $p_1(A) = q^t (1-q)^{n-t}$. Now, in order to sample $A$ from $\mathcal{D}_2$, it must be the case that $S|_{[n]}=A$ and therefore, we have
		\begin{align*}
			p_2(A ) &= \frac{{N-n \choose qN-t}}{{N \choose qN}},
		\end{align*}
		which is $q^t (1-q)^{n-t} (1\pm e^{-n})$ using Claim~\ref{claim:binom}.
	\end{proof}

	\subsection{Direct Product in the Uniform Setting: Proof of Theorem~\ref{thm:DPtest}}
	
	Throughout this section, we use $\eps$ to denote the passing probability of the \maintest. Since there is almost a black-box reduction from non-binary alphabet setting to the binary alphabet setting, we first focus on the binary alphabet for simplicity and prove Theorem~\ref{thm:DPtest} when $m=2$. In Section~\ref{sec:general_dp}, we show how to generalize the result for non-binary alphabet by keeping all the parameters asymptotically the same.

\subsubsection{Parameters}
Throughout this section, we are going to use the parameters
\[
        1\leq c_0\ll C_1\ll C_2\ll C_3\ll C_4
\]
\begin{equation}\label{eq:hier_dp}
        0< \eps\leq \nu = \eps^{3}\leq \gamma = \eps^{10},
\end{equation}

\begin{equation}\label{eq:hier_dp2}
        T = \left(\frac{\log(1/\eps)}{q'}\right)^{c_0}\leq h = \frac{C_1T\log(1/\eps)}{q'}\leq R = W = \frac{C_2 h}{q'}
        \leq \tilde{h} = \frac{C_3W^2\log(1/\eps)}{q'}
        \leq \tilde{R} = \tilde{W} = \frac{C_4\tilde{h}}{q'}
\end{equation}
and $q'' = q-q'$.

	\subsubsection{Local structure}\label{sec:dp_lst}
	\label{sec:locals}
	In this section, we prove the local structure stated in Lemma~\ref{lemma:localg} below. We need a few definitions to state the lemma.

	Consider selecting a random set of size $qn$ as follow. First sample a subset $A\subseteq [n]$ of size $q'n$ and then select a set $B\subseteq [n]\setminus A$ of size $(q-q')n$ uniformly at random. Output $(A,B)$. We need the following few definitions that are similar to the definitions from~\cite{ImpagliazzoKW12}, adapted towards analyzing the \maintest.
	
	\begin{definition}(consistency)
	Fix a set $(A,B)$. A subset $B'\subseteq [n]\setminus A$ is said to be $t$-consistent with $(A,B)$ if $F[A, B]|_{A}\overset{\leq t}{\neq} F[A, B']|_{A}$. Let
    $\Cons_t(A,B)$ be the set of all the sets that are $t$-consistent with  $(A,B)$.
	\end{definition}
	
    We say that $(A,B)$ is good if $\Cons_t(A,B)$  has a significant size. More precisely,
	\begin{definition}(goodness)
		A set $(A,B)$ is called $(\eps/2,t)$-good if
		\[
        \Pr_{B'\subseteq [n]\setminus A}{[B'\in \Cons_t(A,B)]} \geq \eps/2.
        \]
	\end{definition}
	
	\newcommand{\rr}{R}
    We next define the notion of excellence, for which we need to describe an auxiliary distribution over sets $(E,D_1,D_2)$.
    Given $A\subseteq [n]$, select two random subsets $B_1, B_2\subseteq [n]\setminus A$ independently, each of size $q''n$ and let $E$ be a random subset of $B_1\cap B_2$ of size $q''n/\rr$. We take $D_1 = B_1\setminus E$ and $D_2 = B_2\setminus E$.
    We remark that the probability that $|B_1\cap B_2|$ is smaller than $q''n/\rr$ is at most $\exp(-q''^2n)$, hence so long as $\gamma = \omega(\exp(-q''^2n))$ it will
    be absorbed into $\gamma$ in the definition below.
	\begin{definition}(excellence)
		A set $(A,B)$ is called $(\eps/2, t,\rr,h, \gamma)$-excellent if it is $(\eps/2, t)$-good and
		$$\Pr_{E, D_1, D_2}{[ (E, D_i)\in \Cons_t(A,B) \mbox{ for i=1,2 } \ \land \ F[A,  E,  D_1]|_{E} \overset{>h}{\neq} F[A,  E,  D_2]|_{E}]} \leq \gamma.$$
	\end{definition}

	Fix any $(\eps/2, T, \rr, h, \gamma)$-excellent pair $(A_0,B_0)$. We define a function $g_{A_0,B_0}:[n]\rightarrow \{0,1\}$ based on the majority vote of the table $F$ restricted to the sets in $\Cons_t(A_0,B_0)$. More formally, for $x\in [n]\setminus A_0$, we set
	$$ g_{A_0,B_0}(x) := \mathop{\mathsf{Majority}}_{B\in \Cons_t(A_0, B_0) \mid B\ni x} F[A_0,B]|_{x}.$$
	If there is no such $B$ that contains $x$ then we set $g_{A_0,B_0}(x) = 0$. We also set $ g_{A_0,B_0}(A_0) = F[A_0,B_0]|_{A_0}$.

	Based on these definitions, we prove the following local structure, which is the main lemma from this subsection. This is called a local structure as the functions $g_{A_0,B_0}$ enjoy strong consistency (similar to what we need for the global function in Theorem~\ref{thm:DPtest}) but it is weaker: the consistent is
only guaranteed to be {\em local}, namely within $\Cons_t(A_0,B_0)$.
	
	\newcommand{\fillin}{\mathbf{\boxplus}}
	\begin{lemma}
		\label{lemma:localg}
		For all $\eps>0$, if $(A_0,B_0)$ is $(\eps/2, T,\rr,h, \gamma)$-excellent then
		\[
        \Pr_{B\in \Cons_T(A_0,B_0)}{[F[A_0, B]|_{B} \overset{>W^{2}}{\neq} g_{A_0,B_0}(B)]} \leq \nu.
        \]
		Furthermore, a random pair $(A_0, B_0)$ is $(\eps/2, T, \rr, h, \gamma)$-excellent with probability at least $\frac{\eps}{2} -\frac{2^{-\Omega(h)}}{\gamma}\geq \frac{\eps}{4}$.
	\end{lemma}

%
%
For notational convenience, in the remaining part of this subsection, we call a pair $(A,B)$ good if it is $(\eps/2,T)$-good. Similarly, we call a pair $(A, B)$ excellent if it is $(\eps/2, T, \rr, h, \gamma)$-excellent. We start by showing that a random pair is good and excellent with noticeable probability.
	\begin{claim}
		\label{claim:good}
		If $\Pr_{A_0, B_0, B_1}{[F[A_0, B_0]|_{A_0} \overset{\leq T}{\neq} F[A_0, B_1]|_{A_0}]}\geq \eps$, then a random $(A_0, B_0)$ is $(\eps/2,T)$-good with probability at least $\eps/2$.
	\end{claim}
	\begin{proof}
		The proof is by a simple averaging argument.
	\end{proof}
	
	The next claim shows that almost all the good pairs are excellent.
	\begin{claim}
		\label{claim:excellent}
		It holds that
		\[
        \Pr_{A_0, B_0}{[(A_0, B_0)\text{ is $(\eps/2, T)$-good but not $(\eps/2, T, \rr,h, \gamma)$-excellent}]} \leq \frac{2^{-\Omega(h)}}{\gamma}.
        \]
	\end{claim}
	\begin{proof}
		Consider the following two events.
		\begin{enumerate}
			\item Event $Z_1$: $(A_0, B_0)$ is good but
			\[
        \Pr_{E, D_1, D_2}{[ (E, D_i)\in \Cons_T(A_0,B_0) \text{ for $i=1,2$ } \land F[A_0, E, D_1]|_{E} \overset{>h}{\neq} F[A_0, E, D_2]|_{E}]} > \gamma.
        \]
			
			\item Event $Z_2$: $(A_0, B_0)$ is good, $(E, D_i)\in \Cons_T(A_0, B_0)$ for $i=1,2$ and
			\[
        F[A_0, E, D_1]|_{A_0\cup E} \overset{> h/2}{\neq} F[A_0, E, D_2]|_{A_0\cup E}.
        \]
		\end{enumerate}
		We wish to upper bound $\Pr[Z_1]$, and towards this end we write:
        \[
        \Pr[Z_1] = \Pr[Z_1\land Z_2]/\Pr[Z_2 \mid Z_1]\leq  \Pr[Z_2]/\Pr[Z_2 \mid Z_1],
        \]
        and we give an upper bound on the numerator as well as a lower bound on the denominator.
        Note that $\Pr[Z_2\mid Z_1]\geq \gamma$, and we now upper bound $\Pr[Z_2]$. The sets from the event $Z_2$ can be equivalently sampled as follows. First, sample a random subset $A'$ of $[n]$ of size $q'n+\frac{q''n}{\rr}$ (which is to be thought of as $A_0\cup E$), and then pick random sets $D_1, D_2\subseteq [n]\setminus A'$ of size $q''n-\frac{q''n}{\rr}$, conditioned on the event $F[A', D_1]|_{A'} \overset{> h/2}{\neq} F[A',D_2]|_{A'}$. Let $A''\subseteq A'$ be the set of coordinates $x\in A'$ where $F[A',D_1](x)\neq F[A',D_2](x)$,
so that $\card{A''}\geq h/2$. Taking $A_0$ to be a random subset of $A'$ of size $q'n$ and setting $E = A'\setminus A_0$. If $(E, D_i)\in \Cons_T(A_0, B_0)$ for $i=1,2$,
we get that $F[A_0, E, D_1]|_{A_0} \overset{\leq  2T}{\neq} F[A_0, E, D_2]|_{A_0}$ and hence $\card{A''\cap A_0}\leq  2T$. On the other hand, the
expected size of $|A''\cap A_0|$ is at least $h/4$, and so by Chernoff's inequality the probability that $|A''\cap A_0|\leq  2T$ is at most $2^{-\Omega(h)}$.
	\end{proof}
	Next, the following claim asserts that for an excellent pairs $(A_0, B_0)$, the function $g_{A_0,B_0}$ enjoys strong agreement with $F$ inside $\Cons_T(A_0, B_0)$.
	
	\begin{claim}
		\label{claim:localg}
		If $(A_0,B_0)$ is excellent then $\Pr_{B\in \Cons_T(A_0,B_0)}{[F[A_0, B]|_{B} \overset{>W^{2}}{\neq} g_{A_0,B_0}(B)]} \leq \nu$.
	\end{claim}
    Before proving Claim~\ref{claim:localg}, we first note that it implies Lemma~\ref{lemma:localg}.
	\begin{proof}[Proof of Lemma~\ref{lemma:localg}]
    The lemma is immediate by Claims~\ref{claim:good},~\ref{claim:excellent} and~\ref{claim:localg}.
    \end{proof}

    The rest of this section is devoted to the proof Claim~\ref{claim:localg},
    and for that we introduce additional notations and prove some auxiliary claims claims.  Fix an excellent pair $(A_0,B_0)$,
    and let
    \[
    \Cons_T^x = \{B \in \Cons_T(A_0, B_0) \mid x\in B\};
    \]
    we note that this set is used to define $g_{A_0, B_0}(x)$
    in the majority voting step. Also, for $x\in [n]$ we denote by $\mathcal{B}_x$ the collection of all $B\subseteq [n]\setminus A_0$ such that $|B| = (q-q')n$ and $x\in B$
	\begin{claim}
		\label{claim:xE_sampler}
		For at least $1-O\left(\frac{\ln 1/\eps}{q'' n}\right)$ fraction of $x\in [n]\setminus A_0$, we have $2\eps |\mathcal{B}_x|\geq |\Cons_T^x|\geq \frac{\eps}{6} |\mathcal{B}_x|$.
	\end{claim}
	\begin{proof}
		This claim follows from the sampler property of the inclusion graph by invoking Lemma~\ref{lemma:sampler} for the graphs $G((1-q')n, q''n)$.
	\end{proof}

	We are now ready to prove Claim~\ref{claim:localg}.
	
	\begin{proof}[Proof of Claim~\ref{claim:localg}]
    Fix $A_0,B_0$ as in the claim and assume towards contradiction the statement is false; then as $\Cons_T(A_0,B_0)$ consists of at least $\eps/2$ fraction of $B\subseteq [n]\setminus A_0$ of size $(q-q')n$, we get that
	\begin{equation}
		\label{eq:localg_contradiction}
		\Pr_{B\subseteq [n]\setminus A_0}{[B\in \Cons_T(A_0, B_0)\ \land \ F[A_0, B]|_{B} \overset{>\zeta^{2}}{\neq} g_{A_0,B_0}(B)]} > \frac{\nu\eps}{2}.
	\end{equation}
     Denote the above event by $Z$, and sample $B$ conditioned on $Z$. Let $X_B\subseteq B$ be the set of coordinates $x\in B$ for which $F[A_0, B]|_{x} \neq g_{A_0,B_0}(x)$.  We have the following claim:
	
	\begin{claim}\label{claim:localg_c}
		For all but $O\left(\frac{\ln 1/\eps}{q''}\right)$ many $x\in [n]\setminus A_0$, we have
\[\Pr_{B'\in \Cons_T(A_0, B_0)}[ x\in B' \ \land \ g_{A_0, B_0}(x) = F[A_0, B']|_{x}]\geq \frac{q''}{10}.
\]
	\end{claim}
\begin{proof}
  Immediate by Claim~\ref{claim:xE_sampler}.
\end{proof}
	
	Let $X'_B\subseteq X_B$ be the set of coordinates for the inequality in Claim~\ref{claim:localg_c} holds. Note that given the event $Z$,
$|X'_B|\geq \card{X_B} - O\left(\frac{\ln 1/\eps}{q''}\right) \geq W^{2}/2$, so we may pick a subset of $X'_B$ of size $W^2/2$; without loss of generality
we assume $X'_B$ is already of that size. Picking $B'\in \Cons_T(A_0, B_0)$ randomly, by Claim~\ref{claim:localg_c} the expected size of $Y_{B'}:=\{ x \in B'\mid x\in X'_B\ \&\  g_{A_0, B_0}(x) = F[A_0, B']|_{x}\}$ is at least $\frac{q''}{10}W^{2}$. Therefore , by an averaging argument, as $\card{Y_{B'}}\leq \frac{W^2}{2}$ always,
with probability at least $\frac{q''}{20}$ we have $|Y_{B'}|\geq \frac{q''}{20}W^{2}$.
	
	Also, picking $B'\in \Cons_T(A_0, B_0)$ at random, we have by Chernoff's bound that $|B\cap B'| \geq \frac{q''^2 n}{10^{10}}$ with probability at least $1-\exp(-\Omega(q''^2n))$. Therefore, by union bound, with probability at least $\frac{q''}{30}$, we have
\begin{enumerate}
  \item $\card{\{ x \in B'\mid x\in X'_B\ \&\ g_{A_0, B_0}(x) = F[A_0, B']|_{x}\}}|\geq \frac{q''}{100}W^{2}$.
  \item $|B\cap B'| \geq \frac{q''^2 n}{10^{10}}$.
\end{enumerate}
 Taking a random $E\subseteq B\cap B'$ of size $q''n/\rr$, with probability at least $1/2$, we get that $|E\cap Y_{B'}| \geq \frac{q''}{200\rr}W^{2}$. Removing the conditioning on $B'\in \Cons_T(A_0, B_0)$ and the above three conditions, we get
	\[
\Pr_{\substack{B, B'\subseteq [n]\setminus A_0\\ E\subseteq B\cap B'}}{[B, B'\in \Cons_T(A_0, B_0)\ \land\ F[A_0, B]|_E \overset{>\frac{q''}{200\rr}W^{2}}{\neq} F[A_0, B']|_E]}\geq \frac{\nu\eps}{2}\cdot\frac{q''}{30}\cdot \frac{\eps}{2}\cdot\frac{1}{2} > \gamma.
\]
	Since $\frac{q''}{200R}W^{2} \geq h$, this is a contradiction to the fact that $(A_0, B_0)$ is $(\eps/2, T, \rr,h, \gamma)$-excellent.
	\end{proof}

	\subsubsection{Global Structure: the Modified Test}
	\label{sec:globals}
	
	Now that we have a  function $g_{A_0, B_0}$ for every excellent pair $(A_0, B_0)$, the last step is to show that these functions are similar to each other and hence there is a global function $g$ that (almost) agrees with at least $\delta(\eps)$ fraction of the entries from the table $F[\cdot]$. We follow the same proof strategy as appeared in~\cite{BKMcsp3}, however, we will use an explicit bound on the small-set expansion property of a certain graph.
	
	For $c\in (0,1)$, consider selecting a random set of size $qn$ as follows. First, sample a subset $A\subseteq [n]$ of size $q'n$ and then select a set $B\subseteq [n]\setminus A$ of size $(q-q')n$ uniformly at random. Select a random subset of $A$ of size $cq'n$ and call it $D$. Let $E=A\setminus D$. Output $(D, E, B)$, where $A= D\cup E$. Consider the modified agreement test (\test) given in Figure~\ref{fig:test2}.
	
	\begin{figure}[!h]
		\label{fig:test2}
		\fbox{
			\parbox{450pt}{
				\vspace{10pt}
				Given $F: {[n]\choose qn} \rightarrow \{0,1\}^{qn}$,
				\begin{itemize}
					\item Pick a random set $(D\cup E\cup B)$ of size $qn$ where $|D| = cq'n$, $|E|= (1-c)q'n$ and $|B| =q''n =  (q-q')n$.
					\item Select a random subset $E'\subseteq [n]\setminus (D\cup E)$ of size $(1-c)q'n$.
					\item Select a random set $ B'\subseteq [n]\setminus (D\cup E')$ u.a.r. where $|B'|=(q-q')n$.
					\item Check if $F[D\cup E\cup B]|_{D} \overset{\leq 2 T}{=} F[D\cup E'\cup B']|_{D}$.
				\end{itemize}
			}
		}
		\caption{Modified agreement test \test}
	\end{figure}
	
	Note that \test\ is similar to the agreement test \maintest that we wish to analyze, except that the we change the parameters from $(q', q, T)$ to $(cq'n, q, 2T)$ for $c\in (0,1)$. Another (minor) difference is that we require the sets $E$ and $E'$ to be disjoint in the above distribution, whereas in the \maintest, the sets $B_0$ and $B_1$ are uncorrelated. As the distribution of $(E',B')$ depends on $(D,E)$, for notational convenience, we denote this marginal distribution by $\mathcal{D}(D,E)$.
	
	We now relate the two tests \maintest and \test\ in order to show the consistency between the functions $g_{A,B}$. Towards this, we define the notion of consistency, goodness, and excellence tailored to \test.
	
	\begin{definition}(consistency)
		Fix a set $(D, E, B)$ where $A=D\cup E$. A subset $(E',B')$ in the support of $\mathcal{D}(D,E)$ is said to be consistent with $(D, E, B)$ if $F[D, E, B]|_{D} \overset{\leq 2T}{=} F[D, E', B']|_{D}$. Also, we let $\newCons_{2T}(D, E, B)$ be the set of all the sets $(E',B')$ that are consistent with  $(D, E, B)$.
	\end{definition}
	
	\begin{definition}(goodness)
		A set $(D,E,B)$ is called $(\eps^2/2, 2T)$-good if
		$$\Pr_{ (E',B')\sim \mathcal{D}(D,E)}{[(E',B')\in \newCons_{2T}(D,E,B)]} \geq \eps^2/2.$$
	\end{definition}
	
	\begin{definition}(excellence)
		A set $(D,E,B)$ is called $(\eps^2/2, 2 T, R, h, \gamma)$-excellent if it is $(\eps^2/2, 2T)$-good and
		\[
\Pr_{\substack{(E_1,B_1), (E_2,B_2)\sim  \mathcal{D}(D,E)\\ E'\subseteq (E_1\cup B_1)\cap (E_2\cup B_2)\\ |E'|= q''n/\rr }}
{\left[ \substack{F[D\cup  E_1\cup  B_1]|_{E'} \overset{\geq h}{\neq} F[D\cup  E_2\cup B_2]|_{E'}\text{ and }\\ (E_i, B_i)\in \newCons_{2T}(D,E,B) \text{ for $i=1,2$ }}\right]} \leq \gamma.
\]
	\end{definition}
	
	The following claim shows that if the \maintest passes with probability at least $\eps$, then the test \test\ passes with probability at least $0.99\eps^2$.
	\begin{claim}
		\label{claim:newgood}
		If $\Pr_{A_0, B_0, B_1}{[F[A_0, B_0]|_{A_0} \overset{\leq T}{=} F[A_0, B_1]|_{A_0}]}\geq \eps$, then there is $c\in \left[ \frac{q'}{2q}, \frac{2q'}{q} \right]$ such that the test \test passes with probability at least $0.99\eps^2$. Consequently, a random triple $(D, E, B)$, with $|D| = cq'n$, is $(\eps^2/2, 2T)$-good with probability at least $0.49\eps^2$.
	\end{claim}
	\begin{proof}
		Consider the following distribution.
		\begin{itemize}
			\item Select $S\subseteq [n]$ of size $qn$ u.a.r.
			\item Select $A, A'\subseteq S$ each of size $q'n$ u.a.r.
			\item Select $B\subseteq [n]\setminus A$ of size $(q-q')n$ u.a.r.
			\item Select $B'\subseteq [n]\setminus A'$ of size $(q-q')n$ u.a.r.
		\end{itemize}
		We observe the following properties of the above distribution.
		\begin{enumerate}
			\item The pairs $(A, S\setminus A)$ and $(A, B)$  are distributed according to the test distribution \maintest. The same holds for the pairs $(A', S\setminus A')$ and $(A', B')$
			\item For a fixed $S$, the pairs $(A, B)$ and $(A', B')$ are independent.
		\end{enumerate}
		Note that
		\begin{align*}
			&\Expect{S, (A, B), (A',B')}{\mathbf{1}_{ F[S]|_{A}  \overset{\leq T}{\neq} F[A,B]|_A }\ \cdot \ \mathbf{1}_{ F[S]|_{A'}  \overset{\leq T}{\neq} F[A',B']|_{A'}}}\\ &= 		\Expect{S}{\Expect{(A, B), (A',B')}{ \mathbf{1}_{ F[S]|_{A}  \overset{\leq T}{\neq} F[A,B]|_A }\ \cdot\ \mathbf{1}_{ F[S]|_{A'}  \overset{\leq T}{\neq} F[A',B']|_{A'}}}}\\
			&= 		\Expect{S}{\Expect{(A, B)}{\mathbf{1}_{ F[S]|_{A}  \overset{\leq T}{\neq} F[A,B]|_A} }^2} \tag*{(Property 2. above)}\\
			&\geq 	\left(	\Expect{S}{\Expect{(A, B)}{\mathbf{1}_{ F[S]|_{A}  \overset{\leq T}{\neq} F[A,B]|_A} }}\right)^2\tag*{(Jensen's inequality)}\\
			& = \eps^2.
		\end{align*}
		Note that the events $F[S]|_{A} \overset{\leq T}{\neq} F[A,B]|_A$ and $F[S]|_{A'}  \overset{\leq T}{\neq} F[A',B']|_{A'}$ together imply that $F[A,B]|_{A\cap A'}  \overset{\leq 2T}{=} F[A',B']|_{A\cap A'}$. Therefore, we have
		\[
        \Pr_{(A, B), (A', B') }[F[A,B]|_{A\cap A'} \overset{\leq 2T}{\neq} F[A',B']|_{A\cap A'}]\geq \eps^2.
        \]
		
		Based on how the sets $A$ and $A'$ are distributed, we have with $1-\exp(-n)$ probability, the size of $A\cap A'$  lies in $\left[ \frac{q'^2n}{2q}, \frac{2q'^2n}{q} \right]$. Thus, there exist $c\in \left[ \frac{q'}{2q}, \frac{2q'}{q} \right]$ such that
		\[
        \Pr_{(A, B), (A', B') | |A\cap A'| = cq'n}[F[A,B]|_{A\cap A'} \overset{\leq 2T}{=} F[A',B']|_{A\cap A'}]\geq \eps^2-\exp(-n)\geq 0.99\eps^2.
        \]
		
		Now, if we let $D = A\cap A'$, $E = A\setminus D$ and $E' = A'\setminus D$, then the pairs $(D, E, B)$ and $(D, E', B')$ are distributed according to the distribution in \test. Hence, the acceptance probability of \test\ is at least $0.99\eps^2$. The claim now follows from the averaging argument similar to the one in the proof of Claim~\ref{claim:good}.
	\end{proof}
	
	The following claim shows that a random good set is excellent with high probability.
	\begin{claim}
		\label{claim:excellentnew}
		A random $(\eps^2/2, 2 T)$-good set $(D, E, B)$ is $(\eps^2/2, 2 T,\rr, h, \gamma)$-excellent with probability at least $1-\frac{2^{-\Omega(h)}}{\gamma}$.
	\end{claim}
	\begin{proof}
		The is along the same lines as the proof of Claim~\ref{claim:excellent}, and we prove it here again for completeness. Consider the following two events.
		\begin{enumerate}
			\item Event $Z_1$: $(D_0, E_0, B_0)$ is good but
			\[
			\Pr_{\substack{(E_1,B_1), (E_2,B_2)\sim  \mathcal{D}(D,E)\\ E'\subseteq (E_1\cup B_1)\cap (E_2\cup B_2)\\ |E'|= q''n/\rr }}{\left[ \substack{(E_i, B_i)\in \newCons_{2T}(D_0, E_0, B_0) \text{ for $i=1,2$ } \\\land F[D_0\cup  E_1\cup  B_1]|_{E'} \overset{\geq h}{\neq} F[D_0\cup  E_2\cup B_2]|_{E'}}\right]} > \gamma.
			\]
			
			\item Event $Z_2$: $(D_0, E_0, B_0)$ is good, $(E_i, B_i)\in \newCons_{2T}(D_0, E_0, B_0)$ for $i=1,2$ such that $|(E_1\cup B_1)\cap (E_2\cup B_2)|\geq q''n/R$, $E'\subseteq (E_1\cup B_1)\cap (E_2\cup B_2)$ of size $q''n/R$  and
			\[
			F[D_0, E_1, B_1]|_{D_0\cup E'} \overset{> h/2}{\neq} F[D_0, E_2, B_2]|_{D_0\cup E'}.
			\]
		\end{enumerate}
		We wish to upper bound $\Pr[Z_1]$, and towards this end we write:
		\[
		\Pr[Z_1] = \Pr[Z_1\land Z_2]/\Pr[Z_2 \mid Z_1]\leq  \Pr[Z_2]/\Pr[Z_2 \mid Z_1],
		\]
		and we give an upper bound on the numerator as well as a lower bound on the denominator.
		Note that $\Pr[Z_2\mid Z_1]\geq \gamma$, and we now upper bound $\Pr[Z_2]$. The sets from the event $Z_2$ can be equivalently sampled as follows. First, sample a random subset $D'$ of $[n]$ of size $cq'n+\frac{q''n}{\rr}$ (which is to be thought of as $D_0\cup E'$), and then pick random sets $H_1, H_2\subseteq [n]\setminus D'$ of size $qn - cq'n-\frac{q''n}{\rr}$, conditioned on the event $F[D'\cup H_1]|_{D'} \overset{> h/2}{\neq} F[D'\cup H_2]|_{D'}$ and the event that the distribution of $(D'\cup H_1), (D'\cup H_2)$ is consistent with the distribution of $(D_0\cup E_1\cup B_1), (D_0\cup E_2\cup B_2)$ from event $Z_2$ . Let $D''\subseteq D'$ be the set of coordinates $x\in D'$ where $F[D',H_1]|_x\neq F[D',H_2]|_x$,
		so that $\card{D''}\geq h/2$. Take $D_0$ to be a random subset of $D'$ of size $cq'n$ and set $E' = D'\setminus D_0$. If $(E_i, B_i)\in \newCons_{2T}(D_0, E_0, B_0)$ for $i=1,2$,
		we get that $F[D_0, E_1, B_1]|_{D_0} \overset{\leq  4T}{\neq} F[D_0, E_2, B_2]|_{D_0}$ and hence $\card{D''\cap D_0}\leq  4T$. On the other hand, the
		expected size of $|D''\cap D_0|$ is at least $h/4$, and so by Chernoff's inequality the probability that $|D''\cap D_0|\leq  4T$ is at most $2^{-\Omega(h)}$.
	\end{proof}
	
	Similar to the previous analysis, for an $(\eps^2/2, 2 T, \rr,  h, \gamma)$-excellent  triple $(D_0, E_0, B_0)$, we define a function $g_{D_0,E_0,B_0} : [n]\rightarrow \{0,1\}$ based on the majority vote of the table $F$ restricted to the sets in $\newCons_{2T}(D_0,E_0,B_0)$. More formally, for $x\in [n]\setminus D_0$, we set
	\[
        g_{D_0,E_0,B_0}(x) := \mathop{\mathsf{Majority}}_{\substack{(E, B)\in  \newCons_{2T}(D_0,E_0,B_0) \mid\\ E\cup B\ni x}} F[D_0,E,B]|_{x}.
        \]
	If there is no such $E\cup B$ that contains $x$ then we set $ g_{D_0,E_0,B_0}(x) := 0$. We also set $  g_{D_0,E_0,B_0}(x) = F[D_0, E_0, B_0]|_{x}$ for $x\in D_0$.
	The following claim is analogous to Claim~\ref{claim:localg}, saying that the local function $g_{D_0,E_0,B_0}$ strongly agrees with the table $F$ on $\newCons_{2T}(D_0,E_0,B_0)$.
	\begin{claim}
		\label{claim:localgnew}
		If $(D_0, E_0, B_0)$ is $(\eps^2/2, 2 T, \rr, h, \gamma)$-excellent then
		\[
        \Pr_{(E,B)\in \newCons_{2T}(D_0,E_0,B_0)}{[F[D_0,E, B]|_{E,B} \overset{>W^{2}}{\neq} g_{D_0,E_0,B_0}(E,B)]} \leq \nu.
        \]
	\end{claim}
\begin{proof}
  The proof of this claim is analogous to the proof of Claim~\ref{claim:localg}.  Fix an excellent pair $(D_0, E_0, B_0)$. We write $\newCons_{2T}$ for $\newCons_{2T}(D_0,E_0,B_0)$ for convenience. We begin with the following claim.
  \begin{claim}
  	\label{claim:xE_sampler_2}
  	For at least $1-O\left(\frac{\ln 1/\eps}{q'' n}\right)$ fraction of $x\in [n]\setminus D_0$,
  	$$\Pr_{\substack{(E, B)\sim \mathcal{D}(D_0, E_0), \\x\in E\cup B}}[(E, B)\in \newCons_{2T}]\geq \frac{1}{100}  \Pr_{(E, B)\sim \mathcal{D}(D_0, E_0)}[ (E, B)\in \newCons_{2T}]. $$
  \end{claim}
  \begin{proof}
  	The proof follows essentially from the sampler property of an inclusion graph. In this case, the distribution of $(E, B)$ is not uniform among all the sets of size $(q-cq')n$  of $[n]\setminus D_0$. Therefore we cannot apply Lemma~\ref{lemma:sampler} directly. However, similar proof works here and we include it below.
  	
  	For succinctness, let $\mathcal{D}$ be the distribution $\mathcal{D}(D_0, E_0)$ and let $\mathcal{D}_x$ be the distribution $\mathcal{D}$ conditioned on $x\in E\cup B$.  The way we sample $(E, B)\sim \mathcal{D}$, the marginal distribution on $x\in E\cup B$ is not the same. However, for every $x\in [n]\setminus (D_0\cup E_0)$, the marginal distribution is identical.  Same is true for every $x\in E_0$. We will show that at most $O\left(\frac{\ln 1/\tilde{\eps}}{q''}\right)$ many $x\in [n]\setminus (D_0\cup E_0)$, the claimed inequality dos not hold.  Similar poof shows that for at most $O\left(\frac{\ln 1/\tilde{\eps}}{q''}\right)$ many $x\in E_0$, the inequality does not hold.
  	
  	We know that  $\Pr_{(E, B)\sim \mathcal{D}}[ (E, B)\in \newCons_{2T}] = \tilde{\eps} \geq \eps^2/2$. Let $n' = n-|D_0|-|E_0| = (1-q')n$ and  $Z = \{x\in [n]\setminus (D_0\cup E_0) \mid \Pr_{(E, B)\sim \mathcal{D}_x}[S\in \newCons_{2T}] < \tilde{\eps}/100\}$. We need to show that $|Z| = O\left(\frac{\ln 1/\tilde{\eps}}{q''}\right)$.  For $x\in [n]\setminus (D_0\cup E_0)$, let $p$ be the probability that $x\in (E\cup B)$ when $(E, B)$ is sampled according to $\mathcal{D}$. Note that $p = \Omega(q'')$.
  	Suppose for contradiction,  $|Z| = C\frac{\ln 1/\tilde{\eps}}{p}$ for a large  $C>0$. By the definition of the set $Z$ we have,
  	\begin{equation}
  		\label{eq:sampler_lhs}
\Pr_{\substack{x\in [n]\setminus (D_0\cup E_0)\\ (E, B)\sim \mathcal{D}_x}}[x\in Z \land (E, B) \in \newCons_{2T} ] < \frac{|Z|}{n'}\cdot \frac{\tilde{\eps}}{100} = \frac{C\ln(1/\tilde{\eps})\cdot \tilde{\eps}}{100p(1-q')n}.
\end{equation}
  	If we sample a random $(E, B)\sim \mathcal{D}$, then by Chernoff's inequality, $E\cup B$ contains at least  $\frac{C}{2}\ln(1/\tilde{\eps})$ elements from $Z$ with probability at least
  	$1-\exp(-\Omega(C\ln 1/\tilde{\eps})) \geq 1-\tilde{\eps}/10$. Therefore, we have
 	$$\Pr_{\substack{(E, B)\sim \mathcal{D}}}\left[(E, B) \in \newCons_{2T} \land |(E\cup B)\cap Z|\geq \frac{C}{2}\ln(1/\tilde{\eps}) \right] \geq  \frac{\tilde{9\eps}}{10}.$$
 	Conditioned on the above event, if we sample a random $x\in(E\cup B)\setminus E_0$, then $x\in Z$ with probability at least $\frac{1}{2}\frac{C/2\ln(1/\tilde{\eps})}{pn}$.  Therefore, combining with the above event, we have
 	$$\Pr_{\substack{(E, B)\sim \mathcal{D}\\ x\in (E\cup B)\setminus E_0}}[x\in Z \land (E, B) \in \newCons_{2T} ] \geq \frac{9C\ln(1/\tilde{\eps})\tilde{\eps}}{40pn}.$$
 	Assuming $q'\leq 9/10$, this contradicts \eqref{eq:sampler_lhs}
  \end{proof}

  	We now proceed to proving the claim. Assume towards contradiction the statement of the claim  is false; then as $(E, B)\in \newCons_{2T}(D_0,E_0, B_0)$ with probability at least $\eps^2/2$ when sampled according to $\mathcal{D}(D_0, E_0)$, we get that
  	\begin{equation}
  		\label{eq:localg_contradiction}
  		\Pr_{(E,B)\sim \mathcal{D}(D_0, E_0)}{[(E,B)\in \newCons_{2T}(D_0, E_0, B_0)\ \land \ F[D_0, E, B]|_{E,B} \overset{>W^{2}}{\neq} g_{D_0,E_0,B_0}(E, B)]} > \frac{\nu\eps^2}{2}.
  	\end{equation}
  	Denote the above event by $Z$, and sample $(E,B)$ conditioned on $Z$. Let $X_{E,B}\subseteq E\cup B$ be the set of coordinates $x\in E\cup  B$ for which $F[D_0, E, B]|_{x} \neq g_{D_0, E_0, B_0}(x)$.  We have the following claim:
  	
  	\begin{claim}\label{claim:localg_c_2}
  		For all but $O\left(\frac{\ln 1/\eps}{q''}\right)$ many $x\in [n]\setminus D_0$, we have
  		\[\Pr_{(E',B')\in \newCons_{2T}(D_0, E_0, B_0)}[ x\in E'\cup B' \ \land \ g_{D_0, E_0, B_0}(x) = F[D_0, E', B']|_{x}]\geq \frac{q''}{2000}.
  		\]
  	\end{claim}
  	\begin{proof}
  		Immediate by Claim~\ref{claim:xE_sampler_2} and the fact that $g_{D_0, E_0, B_0}(x)$ was set according to the majority vote.
  	\end{proof}
  	
  	Let $X'_{E,B}\subseteq X_{E,B}$ be the set of coordinates for the inequality in Claim~\ref{claim:localg_c_2} holds. Note that given the event $Z$,
  	$|X'_{E,B}|\geq \card{X_{E,B}} - O\left(\frac{\ln 1/\eps}{q''}\right) \geq W^{2}/2$, so we may pick a subset of $X'_{E,B}$ of size $W^2/2$; without loss of generality
  	we assume $X'_{E,B}$ is already of that size. Picking $(E',B')\in \newCons_{2T}(D_0, E_0, B_0)$ randomly, by Claim~\ref{claim:localg_c_2} the expected size of $Y_{E', B'}:=\{ x \in E'\cup B'\mid x\in X'_{E,B}\ \&\  g_{D_0, E_0, B_0}(x) = F[D_0, E', B']|_{x}\}$ is at least $\frac{q''}{2000}W^{2}$. Therefore , by an averaging argument, as $\card{Y_{E',B'}}\leq \frac{W^2}{2}$ always,
  	with probability at least $\frac{q''}{4000}$ we have $|Y_{E',B'}|\geq \frac{q''}{4000}W^{2}$.
  	
  	Also, picking $(E',B')\sim \mathcal{D}(D_0, E_0)$ at random conditioned on $(E',B')\in \newCons_{2T}(D_0, E_0, B_0)$, we have by Chernoff's bound that $|(E\cup B) \cap (E'\cup B')| \geq \frac{q''^2 n}{10^{10}}$ with probability at least $1-\exp(-\Omega(q''^2n))$, as $\eps^2/2 \ll \exp(-\Omega(q''^2n))$. Therefore, by union bound, with probability at least $\frac{q''}{5000}$, we have
  	\begin{enumerate}
  		\item $\card{\{ x \in E'\cup B'\mid x\in X'_{E,B}\ \&\ g_{D_0, E_0, B_0}(x) = F[D_0, E', B']|_{x}\}}|\geq \frac{q''}{4000}W^{2}$.
  		\item $|B\cap B'| \geq \frac{q''^2 n}{10^{10}}$.
  	\end{enumerate}
  	Taking a random $E\subseteq (E\cup B)\cap (E'\cup B')$ of size $q''n/\rr$, with probability at least $1/2$, we get that $|E\cap Y_{E',B'}| \geq \frac{q''}{8000\rr}W^{2}$. Removing the conditioning on $(E', B')\in \newCons_{2T}(D_0, E_0, B_0)$ and the above three conditions, we get
  	\[
  	\Pr_{\substack{(E, B), (E',B')\sim \mathcal{D}(D_0, E_0)\\ E\subseteq  (E\cup B)\cap (E'\cup B')}}{\left[\substack{(E,B),  (E',B')\in \newCons_{2T}(D_0,E_0, B_0)\ \\ \land\ F[D_0, E, B]|_E \overset{>\frac{q''}{8000\rr}W^{2}}{\neq} F[D_0, E', B']|_E }\right]}\geq \frac{\nu\eps^2}{2}\cdot\frac{q''}{5000}\cdot \frac{\eps^2}{2}\cdot\frac{1}{2} > \gamma.
  	\]

  	Since $\frac{q''}{8000R}W^{2} \geq h$, this is a contradiction to the fact that $(D_0, E_0, B_0)$ is $(\eps^2/2, 2T, \rr,h, \gamma)$-excellent.
  \end{proof}

	\subsubsection{Consistency Between the Functions $g_{A,B}$ via newCons}
    In this section, we use the local structure result proved in the previous section (but with different parameters) and show a relationship between these local functions.
    To do so, we introduce a variant of the ``Cons'' set which allows us to compare the local functions $g_{A,B}$ and $g_{A',B'}$ for $A$ and $A'$ that are correlated (but
    not the same).

	Let $c_0$ be the constant from Theorem~\ref{thm:DPtest} and take $c\in \left[ \frac{q'}{2q}, \frac{2q'}{q} \right]$ from Claim~\ref{claim:newgood}.
	\begin{definition}
		For an $(\eps^2/2, 2 T,\rr,h, \gamma)$-excellent triple $(D_0, E_0, B_0)$, let $\newCons_{W^2}^\star(D_0, E_0,B_0)\subseteq \newCons_{2T}(D_0, E_0,B_0)$ consist
        of all $(E,B)$'s such that $F[D_0,E, B]|_{E,B} \overset{\leq W^{2}}{\neq} g_{D_0,E_0,B_0}(E,B)]$. Similarly, for an $(\eps^2/8, 2W^{2}, \tilde{\rr}, \tilde{h},  \gamma)$-excellent pair $(A_0, B_0)$, let $\Cons^\star_{2W^{2}}(A_0, B_0)\subseteq \Cons_{2W^{2}}(A_0, B_0)$ consists of $B$'s such that $F[A_0, B]|_{B} \overset{\leq W^{2}}{\neq} g_{A_0,B_0}(B)$.
	\end{definition}
    \begin{remark}
      Note that in the definition above we look at $\Cons_{2W^{2}}(A_0, B_0)$, which are the same as the ``Cons'' sets considered in previous sections except
      that the parameter $T$ is replaced by the parameter $2W^2$. To compensate for that, we have adjusted the excellence parameters accordingly and
      have taken $\tilde{\rr}$ and $\tilde{h}$ instead of $\rr$ and $h$ from the previous section.
    \end{remark}
	Note that by Claims~\ref{claim:newgood},~\ref{claim:localgnew}, for every excellent triple $(D_0, E_0, B_0)$, we have
	\begin{equation}
		\label{eq:consstar}
		\Pr_{ (E',B')\sim \mathcal{D}(D_0,E_0)}{[(E',B')\in \newCons_{W^2}^\star(D_0,E_0,B_0)]} \geq \eps^2/4.
	\end{equation}
	
	\begin{definition}	Fix $(D_0, E_0, B_0)$ that is $(\eps^2/2, 2 T)$-good. A pair $(E,B)\in \newCons_{W^2}^\star(D_0, E_0, B_0)$ is called a {\em dense pair}, if $\Pr_{B'\subseteq [n]\setminus (D_0\cup E)}{[(E,B')\in \newCons_{W^2}^\star(D_0, E_0,B_0)]} \geq \eps^2/8$.
	\end{definition}
	
	The next claim shows that many sets $(E, B)$ in $\newCons_{W^2}^\star(D_0, E_0, B_0)$ when we pair it with $D_0$ as $(D_0\cup E, B)$, then the pair $(D_0\cup E, B)$ is excellent for \maintest.
	
	\begin{claim}
		\label{claim:old_new_excellence}
		Let $(D_0, E_0, B_0)$ be an $(\eps^2/2, 2 T, \rr, h, \gamma)$-excellent triple; then,
		\[
\Pr_{(E, B)\in \newCons_{W^2}^\star(D_0, E_0, B_0)}{[ (E,B) \text{ is dense } \land\ (D_0\cup E, B)\ is\ (\eps^2/8, 4W^{2}, \tilde{\rr}, \tilde{h},  \gamma)\text{-excellent}]} \geq \eps^2/16.
\]
	\end{claim}
	\begin{proof}
		Since $(D_0, E_0, B_0)$ is excellent, by using \eqref{eq:consstar} and an averaging argument we have that at least $\eps^2/8$ fraction of
$(E,B)\in \newCons_{W^2}^\star(D_0, E_0, B_0)$ are dense pairs. 	We will show that for such pair $(E,B)$, $(D_0\cup E, B)$ is $(\eps^2/8, 2W^{2})$-good.

Fix any such $(E,B)$. By the definition of a dense pair, we have
		\[
\Pr_{\substack{B'\subseteq [n]\setminus D_0\cup E}}{\left[\substack{g_{D_0,E_0,B_0}(D_0\cup E) \overset{\leq 2W^{2}}{\neq}  F[D_0, E, B]|_{D_0 \cup E}\text{ and}\\  \ g_{D_0,E_0,B_0}(D_0\cup E) \overset{\leq 2W^{2}}{\neq}  F[D_0, E, B']_{D_0\cup E}}\right]}\geq \eps^2/8.
\]
		Note that the first event in the above probability follows from the fact that $(E, B) \in \newCons_{W^2}^\star(D_0, E_0, B_0)$ ($W^2$ accounts for the possible disagreement on $E$ and another $W^2\geq 2T$ accounts for the disagreement on $D_0$) and the second event follows from the fact that for a random $B'$,
        $(E,B')\in \newCons_{W^2}^\star(D_0, E_0, B_0)$ with probability at least $\eps^2/8$, as $(E, B)$ is a dense pair. It follows that
		\[
        \Pr_{\substack{B'\subseteq [n]\setminus D_0\cup E}}{\left[F[D_0, E, B]|_{D_0 \cup E}\overset{\leq 4W^{2}}{\neq}  F[D_0, E, B']_{D_0\cup E}\right]}\geq \eps^2/8,
        \]
		and hence $(D_0\cup E, B)$ is $(\eps^2/8, 4W^{2})$-good.
        By Claim~\ref{claim:excellent} at least $1-\nu$ fraction of
        the $(\eps^2/8, 4W^{2})$-good pairs are  $(\eps^2/8, 4W^{2}, \tilde{\rr}, \tilde{h}, \gamma)$-excellent,
        it follows that at least $\eps^2/8 - O(\nu/\eps^2)\geq \eps^2/16$ of the pairs $(E,B)\in \newCons_{W^2}^\star(D_0,E_0,B_0)$ it holds that $(D_0\cup E, B)$ is excellent,
        concluding the proof.
	\end{proof}

	Claim~\ref{claim:old_new_excellence} will help us to relate the functions $g_{D_0, E_0, B_0}$ and $g_{D_0\cup E, B}$ where $(D_0, E_0, B_0)$ is excellent
and $(E,B)$ is a typical pair from $\newCons_{W^2}^\star(D_0, E_0, B_0)$. Since the functions $g_{D_0, E_0, B_0}$ and  $g_{D_0\cup E, B}$ are defined based on their respective
``Cons" set, we need to make sure that  the sets in $\Cons^\star_{2W^2}$ and $\newCons_{W^2}^\star$ are correlated.  The following claim (and in particular the third item)
shows that this is indeed true.
	
	\begin{claim}
		\label{claim:commonCons}
		Fix any $(\eps^2/2, 2 T, \rr, W, \gamma)$-excellent triple $(D_0,E_0, B_0)$. When $(E,B)$ is sampled according to the distribution $\mathcal{D}(D_0, E_0)$,
then with probability at least $\Omega(\eps^4)$, we have:
		\begin{enumerate}
			\item $(E,B)\in \newCons_{W^2}^\star(D_0, E_0, B_0)$.
			\item  $(D_0\cup E, B)$ is $(\eps^2/8, 4W^{2}, \tilde{\rr}, \tilde{h}, \gamma)$-excellent.
			\item There is at least an $\Omega(\eps^2)$ fraction of $B'\subseteq [n]\setminus (D_0\cup E)$ such that
        $(E, B')\in \newCons_{W^2}^\star(D_0, E_0, B_0)$ and $B'\in \Cons^\star_{4W^2}(D_0\cup E, B)$.
		\end{enumerate}
	\end{claim}
	\begin{proof}
		Using Claim~\ref{claim:old_new_excellence} and \eqref{eq:consstar}, for a random $(E,B)\sim \mathcal{D}(D_0, E_0)$, with probability at least $\Omega(\eps^4)$, we have
		\begin{enumerate}
			\item[(i)] $(E,B)\in \newCons_{W^2}^\star(D_0, E_0, B_0)$,
			\item[(ii)] $(D_0\cup E,B)$ is $(\eps^2/8, 4W^{2}, \tilde{\rr}, \tilde{h}, \gamma)$-excellent, and
			\item[(iii)] $\Pr_{B'\subseteq [n]\setminus (D_0\cup E)}{[(E,B')\in \newCons_{W^2}^\star(D_0, E_0, B_0)]} = \Omega(\eps^2).$
		\end{enumerate}
		The condition (i) above implies
		\begin{equation}
			\label{eq:cond_a_newcons_excellent}
			g_{D_0, E_0, B_0}(D_0\cup E\cup B) \overset{\leq 2W^{2}}{\neq} F[D_0\cup E\cup B],
		\end{equation}
		where the additional $W^2\geq 2T$ accounts for the possible disagreement on $D_0$. Using the condition (iii) above, we also have
		\begin{equation}
			\label{eq:cond_b_newcons_excellent}
			g_{D_0, E_0, B_0}(D_0\cup E\cup B')  \overset{\leq 2W^{2}}{\neq} F[D_0\cup E\cup B']
		\end{equation}
	for at least $\Omega(\eps^2)$ fraction of $B'\subseteq [n]\setminus (D_0\cup E)$.
    Using \eqref{eq:cond_a_newcons_excellent} and \eqref{eq:cond_b_newcons_excellent}, for at least $\Omega(\eps^2)$ fraction of $B'$
    we get that $ F[D_0\cup E\cup B]|_{D_0\cup E} \overset{\leq 4W^{2}}{\neq} F[D_0\cup E\cup B']|_{D_0\cup E}$, in
    which case $B'\in \Cons_{4 W^2}(D_0\cup E, B)$ and the claim follows.
	\end{proof}

	Using the above claim, we show that for every $(D_0, E_0, B_0)$ that is excellent,  the functions $g_{D_0, E_0, B_0}$ and $g_{D_0\cup E, B}$ are very close to each other in hamming distance for many pairs $(E, B)$.
	\begin{claim}
		\label{claim:oldnew_g_consistency}
		Fix any  $(\eps^2/2, 2 T, \rr, h, \gamma)$-excellent triple $(D_0,E_0, B_0)$. Then
		\[
        \Pr_{(E,B)\sim\mathcal{D}(D_0, E_0)}{[g_{D_0, E_0, B_0} \overset{\leq \tilde{W}^3}{\neq} g_{D_0\cup E, B}]}\ggg \eps^4.
        \]
	\end{claim}
	\begin{proof}
		Select a pair $(E,B)$ according to the distribution $\mathcal{D}(D_0, E_0)$.
		Using Claim~\ref{claim:commonCons} with probability at least $\Omega(\eps^4)$, we have
        $g_{D_0,E_0,B_0}(D_0\cup E\cup B) \overset{\leq 2W^2}{\neq} F[D_0\cup E\cup B]$ ,
        $(D_0\cup E, B)$ is  $(\eps^2/8, 4W^2, \tilde{\rr}, \tilde{h}, \gamma)$-excellent and
        $g_{D_0\cup E, B}(D_0\cup E\cup B)  \overset{\leq 3\tilde{W}^2}{\neq} F[D_0\cup E\cup B]$.
        Note that in the last condition, we used the fact that  if $(D_0\cup E, B)$ is excellent and if
        $B'\in \Cons^\star_{2W^{2}}(D_0\cup E, B)$, then $(D\cup E, B')$ is also excellent and furthermore
        the functions $g_{D_0\cup E, B}$ and  $g_{D_0\cup E, B'}$ are the same as they are defined using the set
        $\Cons^\star_{4W^{2}}(D_0\cup E, B) = \Cons^\star_{4W^{2}}(D_0\cup E, B')$.
		
		Combining these events and Claim~\ref{claim:commonCons}, we get that over the randomness of $(E,B)$, with probability at least $\Omega(\eps^4)$, the following happens:
		\begin{enumerate}
			\item $(D_0\cup E, B)$ is  $(\eps^2/8, 4W^2, \tilde{\rr}, \tilde{h}, \gamma)$-excellent,
			\item $g_{D_0,E_0,B_0}(D_0\cup E\cup B)\overset{\leq 4\tilde{W}^2}{\neq}  g_{D_0\cup E, B}(D_0\cup E\cup B)$, and
			\item $\Pr_{B'\subseteq [n]\setminus (D_0\cup E)}{[(E, B')\in \newCons_{W^2}^\star(D_0, E_0, B_0)\ \land \ B'\in \Cons^\star_{4W^2}(D_0\cup E, B)]} \ggg \eps^2$.
		\end{enumerate}
		From the third point we can conclude that $\Pr_{B'\subseteq [n]\setminus (D_0\cup E)}{[g_{D_0, E_0, B_0}(B') \overset{\leq 4W^2}{\neq} g_{D_0\cup E, B}(B')]} \ggg \eps^2$. Using Claim~\ref{claim:the_following_claim}, we get that
		$g_{D_0, E_0, B_0}|_{[n]\setminus (D_0\cup E)} \overset{\leq 16 W^2/q''}{\neq} g_{D_0\cup E, B}|_{[n]\setminus (D_0\cup E)}$.
        Combined with the second item above saying that $g_{D_0,E_0,B_0}$ and $g_{D_0\cup E, B}$ differ on at most $4\tilde{W}^2$ elements from $D_0\cup E$,
        we get that $g_{D_0, E_0, B_0}$ and $g_{D_0\cup E, B}$ differ on at most $\tilde{W}^3$ elements.
	\end{proof}

	\begin{claim}\label{claim:the_following_claim}
		For any $q\in (0,1)$, $a\geq 1$ and $\delta >  e^{-qa/8}$, given two functions
$f, g : [n] \rightarrow \{0,1\}$ such that $\mathop{\Pr}_{S\subseteq [n] \mid |S| = qn}[f(S) \overset{\leq a}{\neq} g(S)]\geq \delta$, then
$| \{i\in [n] \mid f(i) \neq  g(i)\}|\leq  4a/q$.
	\end{claim}
	\begin{proof}
		Let $\mathcal{E} = \{ i \in [n]\mid f(i)\neq g(i)\}$, and suppose towards  contradiction $|\mathcal{E}| \geq 4a/q$.
        Sample a subset $S\subseteq_{q/2} [n]$; by the Chernoff's bound $|S| \leq qn$ with probability at least $1-\exp(-\Omega(q^2n))$.
        Furthermore, the expected size of $S\cap \mathcal{E}$ is $2a$. Now, again by the Chernoff's bound, the probability that $|S\cap \mathcal{E}|$ is at most
        $a$ is at most $e^{-qa/4}$. This implies $\mathop{\Pr}_{S\subseteq [n] \mid |S| = qn}[f(S) \overset{\leq a}{\neq} g(S)]\leq  e^{-qa/8}$ which is a contradiction.
	\end{proof}
	
	Claim~\ref{claim:oldnew_g_consistency} allows us to relate functions defined used ``Cons'' and functions defined using ``newCons'',
    and using it along with an application of the Cauchy-Schwarz inequality we conclude relations between functions defined using
    ``Cons'' (but with different $A_0$'s).
	\begin{claim}
		\label{claim:correlated_g_closeness}
		There exists a constant $d_0\in (0,1)$ such that for $c' = d_0c$, we have the following. For at least $\Omega(\eps^{12})$ fraction of excellent pairs $(A_0,B_0)$,
		\[
        \Pr_{(\tilde{A_0}, B)}{[g_{A_0, B_0} \overset{\leq 2\tilde{W}^3}{\neq} g_{\tilde{A_0}, B}]}\ggg \eps^{12}.
        \]
		Here, $\tilde{A_0}$ is distributed uniformly conditioned on $|\tilde{A_0}| = q'n$ , $|A_0\cap \tilde{A_0}| = c'q'n$ and $B\subseteq[n]\setminus \tilde{A_0}$ of size $(q-q')n$.
	\end{claim}
	\begin{proof}
		Since $\Omega(\eps^2)$ fraction of the triples $(D_0, E_0, B_0)$ are $(\eps^2/2, 2T, \rr, h, \gamma)$-excellent, we have the following from Claim~\ref{claim:oldnew_g_consistency}.
		\begin{align*}
			\Omega(\eps^6) \leq \E_{(D_0, E_0, B_0)}\E_{(E,B)\sim\mathcal{D}(D_0, E_0)}{\left[\mathbf{1}_{g_{D_0, E_0, B_0} \overset{\leq \tilde{W}^3}{\neq} g_{D_0\cup E, B}}\right]}.
		\end{align*}
		Note that $(D_0\cup E, B)$ being $(\eps^2/8, 4W^{2}, \tilde{\rr}, \tilde{h}, \gamma)$-excellent is implicit in the event.
        Squaring and applying Cauchy-Schwarz we get that
		\begin{align*}
			\Omega(\eps^{12}) &\leq \E_{(D_0, E_0, B_0)}\E_{\substack{(E,B)\sim\mathcal{D}(D_0, E_0)\\ (E',B')\sim\mathcal{D}(D_0, E_0)}}
{\left[\mathbf{1}_{g_{D_0, E_0, B_0} \overset{\leq \tilde{W}^3}{\neq} g_{D_0\cup E, B}} \ \ \mathbf{1}_{g_{D_0, E_0, B_0} \overset{\leq \tilde{W}^3}{\neq} g_{D_0\cup E', B'}}\right]}\\
			& \leq  \E_{(D_0, E_0, B_0)}\E_{\substack{(E,B)\sim\mathcal{D}(D_0, E_0)\\ (E',B')\sim\mathcal{D}(D_0, E_0)}}
                {\left[ \mathbf{1}_{g_{D_0\cup E, B} \overset{\leq 2\tilde{W}^3}{\neq} g_{D_0\cup E', B'}}\right]}
		\end{align*}
		Now letting $A_0 = D_0\cup E'$, $B_0 = B'$, $\tilde{A}_0 = D_0\cup E$ and $B = B$, we see that with probability at least $1-\exp(-\Omega((cq')^2n))$, we have $|E\cap E'| =O(q'^2n)$. Therefore, there exists a constant $d_0$ such that when we condition the above distribution on $|A_0\cap \tilde{A}_0| = d_0cq'n$, the expectation at still least $\Omega(\eps^{12})$. The claim now follows from an averaging argument.
	\end{proof}
	
    \subsubsection{Applying Small-set Expansion}
	Intuitively, Claim~\ref{claim:correlated_g_closeness} asserts that after re-randomizing $B_0$ and perturbing $A_0$,
    the functions $g_{A_0,B_0}$ remain basically the same with noticeable probability. In the next lemma we appeal to small
    set expansion type results to show that there is a single function $g$ such that $g_{A_0,B_0}$ is very close to  $g$ with noticeable
    probability, hence concluding the proof of Theorem~\ref{thm:DPtest} for $m=2$.
	\begin{lemma}
		\label{lemma:global_SSE_arg}
		If $\Pr_{(A_0\cup B_0, A_0\cup B_1)\sim  \mathcal{D}_{q,q'}}{[F[A_0, B_0]|_{A_0} \overset{\leq T}{\neq} F[A_0, B_1]|_{A_0}]}\geq \eps$,
		then there exists a pair $(A^\star_0,B^\star_0)$ which is $(\eps^2/8, 2W^{2}, \tilde{\rr}, \tilde{h}, \gamma)$-excellent such that
		\[
        \Pr_{T\in{[n]\choose qn}}{\left[g_{A^\star_0,B^\star_0}(T) \overset{\leq \tilde{W}^{O(1)}}{\neq} F[T]\right]} \geq \delta,
        \]
		where $\delta =  \eps^{O(\log (1/q')^2)}$.
	\end{lemma}
	\begin{proof}
		We first note that if two functions $f, g:[n]\rightarrow \{0,1\}$ differ at $O(a)$ locations and $a\leq \sqrt{n}$,
        then for a random subset $S\subseteq [n]$ of size $n/a$ the probability that $f(S) = g(S)$ is at least $\Omega(1)$.
        Using this fact and Claim~\ref{claim:correlated_g_closeness}, and by letting $W' = 2\tilde{W}^3$, we have the following
		\[
        \E_{S\subseteq[n] \mid |S|=\frac{n}{W'}}
        \left[ \E_{(A_0, B_0), (\tilde{A_0}, B)}\left[ \mathbf{1}_{g_{A_0, B_0}(S) = g_{\tilde{A_0}, B}(S)}\right] \right]
        \ggg \eps^{\Theta(1)}.
        \]
		where $(\tilde{A_0}, B)$ is distributed as in Claim~\ref{claim:correlated_g_closeness}.
Note that the event $ \mathbf{1}_{g_{A_0, B_0}(S) = g_{\tilde{A_0}, B}(S)}$ implicitly implies that the pairs $(A_0, B_0), (\tilde{A_0}, B)$ are excellent.
By an averaging argument, at least $\Omega(\eps^{\Theta(1)})$ fraction of the sets $S$ are such that
		\[
        \E_{(A_0, B_0), (\tilde{A_0}, B)}\left[ \mathbf{1}_{g_{A_0, B_0}(S) = g_{\tilde{A_0}, B}(S)}\right] \ggg \eps^{\Theta(1)}.
        \]
		Given such a set $S$, we define a partition of ${[n] \choose qn}$ based on the value $g_{A,B}(S)$. In other words, we have parts identified by strings in
$\{0,1\}^{n/W'^2}$;  $(A,B)$ and $(A',B')$ belong to the part $\mathcal{S}_\tau$ if $g_{A,B}(S) = g_{A',B'}(S) = \tau$.
We claim that there exists a $\tau$ such that $|\mathcal{S}_\tau| = \eta {n \choose qn}$, for some $\eta = \eta(\eps, q') = \eps^{O(\log (1/q')^2)}$.
		
		Consider the graph $\mathcal{G}_n$ induced on the set of vertices $\{(A, B) \mid |A| = q'n, |B|=(q-q')n\}$ as follows: a random neighbor
$(A', B')$ of $(A, B)$ in this graph is sampled conditioned on the fact that $A'$ is distributed uniformly conditioned on $|A'| = q'n$ , $|A\cap A'| = c'q'n$ and
$B'\subseteq[n]\setminus A'$ is a uniformly random set of size $(q-q')n$. Using the small set expansion property of $\mathcal{G}_n$ from Lemma~\ref{lemma:sse_prem},
if all the parts were of size at most $\delta {n \choose qn}$, then
\[\Pr_{(A_0, B_0), (\tilde{A_0}, B)}{[(A_0, B_0), (\tilde{A_0}, B) \in \mathcal{S}_\tau\ \text{ for some } \tau ]}\leq \delta^{O(1/\log(1/q')^2)}.
\]
		We conclude that there must be a part $\tau^{\star}(S)$ whose fractional size is at least $\eta$.
        Therefore, for any such $S$ we get that
        \[
        \E_{(A_0, B_0), (A'_0, B'_0)}\left[ \mathbf{1}_{g_{A_0, B_0}(S) = g_{{A'_0}, B'_0}(S)}\right]
        \geq \left(\E_{(A_0, B_0), (A'_0, B'_0)}\left[ \mathbf{1}_{g_{A_0, B_0}(S) =\tau^{\star}(S)}\right]\right)^2
        \geq \eta^2.
        \]
        Taking expectation over $S$ yields that
		\[
        \E_{S\subseteq[n] \mid |S|=\frac{n}{\zeta'^2}}\left[ \E_{(A_0, B_0), (A'_0, B'_0)}\left[ \mathbf{1}_{g_{A_0, B_0}(S) = g_{{A'_0}, B'_0}(S)}\right] \right]
        \ggg
        \eps^{\Theta(1)}\eta^2.
        \]
		This implies,
		\[
        \E_{(A_0, B_0), (A'_0, B'_0)}\left[ \E_{S\subseteq[n] \mid |S|=\frac{n}{\zeta'^2}}\left[\mathbf{1}_{g_{A_0, B_0}(S) = g_{{A'_0}, B'_0}(S)}\right] \right]
        \ggg
        \eps^{\Theta(1)}\eta^2,
        \]
        and so by an averaging argument for at least $\Omega(\eps^{\Theta(1)} \eta^2)$ of pairs $(A_0, B_0), (A'_0, B'_0)$ we have that
        \[
        \E_{S\subseteq[n] \mid |S|=\frac{n}{\zeta'^2}}\left[\mathbf{1}_{g_{A_0, B_0}(S) = g_{{A'_0}, B'_0}(S)}\right]\ggg
        \eps^{\Theta(1)}\eta^2.
        \]
		Using Claim~\ref{claim:the_following_claim} we get that
		\[
    \Pr_{(A_0, B_0), (A'_0, B'_0)}\left[ {g_{A_0, B_0} \overset{\leq W'^{O(1)}}{\neq} g_{{A'_0}, B'_0}}\right] \ggg
        \eps^{\Theta(1)}\eta^2.
    \]
		This means that there exists an excellent pair $(A^\star_0, B^\star_0)$ such that
		\[
\Pr_{(A'_0, B'_0)}\left[ {g_{A^\star_0, B^\star_0}\overset{\leq W'^{O(1)}}{\neq} g_{{A'_0}, B'_0}}\right] \ggg
        \eps^{\Theta(1)}\eta^2.
\]
		Now consider selecting a random set $T$ of size $qn$ and select a random subset $A$ of $T$ of size $q'n$ and let $B= T\setminus A$.
Select a random set $B'\subseteq [n]\setminus A$. Using the above inequality, with probability $\Omega(\eps^{\Theta(1)} \eta^2)$,
$g_{A^\star_0, B^\star_0}\overset{\leq W'^{O(1)}}{\neq} g_{A, B'}$ and with probability at least $\Omega(\eps^2)$,
$B\in \Cons^\star_{2W^2}(A, B')$ and hence  $g_{A, B'}(T) \overset{\leq W'}{\neq}F[T]$. Combining all these events, we have
		\[
\Pr_{T\in{[n]\choose qn}}{\left[g_{A^\star_0,B^\star_0}(T) \overset{\leq W'^{O(1)}}{\neq} F[T]\right]}
\ggg \eps^{\Theta(1)}\eta^2\cdot \eps^2 = \eps^{O(\log (1/q')^2)},
\]
		and the claim follows.
	\end{proof}

	\subsubsection{Extending the Result to Large Alphabets}
	\label{sec:general_dp}
	In this section, we prove the direct product result with similar conclusion but for larger alphabet.

	\begin{thm}
		[Restatement of Theorem~\ref{thm:DPtest}]
         There exists $c>0$ such that the following holds for sufficiently large $n\in\mathbb{N}$ and $\eps\geq 2^{-n^{c}}$. Fix an alphabet $[m]$ and  $c_0>0$ be any constant. For all $\eps>0$, $0<q'< q<1$ and $T = \left(\frac{\log(1/\eps)}{q'}\right)^{c_0}$,  suppose that $F: {[n]\choose qn}\rightarrow [m]^{qn}$ satisfies
		\[
		\Prob{(S_1,S_2)\sim \mathcal{D}_{q,q'}}{F[S_1]|_{S_1\cap S_2} \overset{\leq T}{\neq } F[S_2]|_{S_1\cap S_2}}\geq \eps.
		\]
		Then there exists a function $g: [n]\rightarrow[m]$ such that for at least an $\eps^{O(\log (1/q')^2)}$ fraction of $S\in {[n] \choose qn}$, we have $|\{ i\in S\mid F[S]_i\neq  g(i)\}| \leq  \left(\frac{T\log{(1/\eps)}}{q'}\right)^{O(1)}$.\end{thm}
	
	The proof of this theorem follows the same lines as the proof of the binary case. In the binary case,
    the only place we used the fact that the alphabet is binary is when deriving the local structure, i.e.\,  Lemma~\ref{lemma:localg}.
    Given such a local structure for larger alphabet the rest of the proof goes through as is.
    Thus, to prove the case $m\geq 3$ it suffices to prove a local structure lemma in this setting, and we show it via a reduction to the binary case (Lemma~\ref{lemma:localg}).
    We remark that this type of argument already appeared in ~\cite{ImpagliazzoKW12}, and we reproduce it here for the sake of completeness.

	The definitions of goodness and excellence still makes sense, and we use the same parameters as in the binary case. It is also easy to observe that Claim~\ref{claim:good} and Claim~\ref{claim:excellent} follow in this setting as is. Fix an excellent pair $(A_0, B_0)$, we can define the function $g_{A_0,B_0}: [n] \rightarrow [m]$ as follows:
	\[
        g_{A_0,B_0}(x) := \mathop{\mathsf{Plurality}}_{B\in \Cons_{T}(A_0, B_0) \mid B\ni x} F[A_0,B]|_{x}.
    \]
	We have the following claim, which is a substitution for Lemma~\ref{lemma:localg}:	
	\begin{claim}
		\label{claim:localg_general}
		Fix an alphabet $[m]$ and $c_0>0$ be any constant. For every $\eps>0$,  and $q, q'\in (0,1)$, suppose a table $F: {[n]\choose qn} \rightarrow [m]^{qn}$ passes the \maintest with parameters $(q, q', \eta)$ with probability at least $\eps$. 	If $(A_0,B_0)$ is excellent then
		\[
            \Pr_{B\in \Cons_T(A_0,B_0)}{[F[A_0, B]|_{B} \overset{>10W^2}{\neq} g_{A_0,B_0}(B)]} \leq 2\nu.
        \]
	\end{claim}
	\begin{proof}
		Let $Z\subseteq \Cons_T(A_0,B_0)$ be the such sets $B$ such that $g_{(A_0, B_0)}(B) \overset{>10W^2}{\neq} F[A_0, B]|_{B}$. Select $n$ random functions $f_1, f_2, \ldots, f_n: [m]\rightarrow \{0,1\}$. Define a table $G_f:{[n]\choose qn}\rightarrow \{0,1\}^{qn}$ as follows: For a given tuple $S = (i_1, i_2, \ldots, i_{qn}) \in {[n]\choose qn}$, if $F(S) = (a_1, a_2, \ldots, a_{qn})$ then set
\[
G_f(S) = (f_{i_1}(a_1), f_{i_2}(a_2),, \ldots, f_{i_{qn}}(a_{qn})).
\]
The reason for doing this is that whenever sets $(A, B)$ and $(A, B')$ agree on $A$ w.r.t. the table $F$, then they also agree on $A$ w.r.t. the map $G_f$. Therefore, for any fixed set of functions $f_i$'s, if $(A_0, B_0)$ is excellent with respect to the table $F$ then, the same pair $(A_0, B_0)$ is also excellent with respect to the map $G_f$. Therefore, by  Lemma~\ref{lemma:localg}, there is a function $g_f$ defined by the majority votes from $\Cons_T$ such that for at most a $\nu$ fraction of the sets in $B\in \Cons_T$, we have $G_f[A_0, B] \overset{\geq W^2}{\neq} g_f(B)$.
Let $Z_f$ be thecollection of such sets in  $B\in \Cons_T$. We will show that if $B\in Z$, then almost surely $B\in Z_f$. Since we know that the density of $Z_f$ in $\Cons_T$ is at most $\nu$, this will show that the density of $Z$ in $\Cons_T$ is also roughly $\nu$ and this will conclude the proof.
		
		Towards showing this, fix a $B$, fix a $x\in B$ and consider a $u\in [m]$ such that $u\neq g_{A_0, B_0}(x)$.
        We first show that with probability at least $1/4$ it holds that $g_f(x) \neq f_x(u)$. Indeed, let us change the value of
         $f_x$ at $u$ and at $g_{A_0, B_0}(x)$ for a moment to both be $\perp$, and  define
		\[
        b =  \mathop{\mathsf{Majority}}_{\substack{B\in \Cons_T(A_0, B_0) \mid B\ni x\\ G_f[A_0,B]|_{x}\neq \perp}} G_f[A_0,B]|_{x}.
        \]
		Now consider defining $f_x$ once again on $u$ and $g_{A_0, B_0}(x)$ by uniform $\{0, 1\}$ independently chosen assignments, and the affect of that
        on the value of $b$. With probability $1/4$, $f_x$ maps $g_{A_0, B_0}(x)$ to $b$ and $u$ to $1 -b$, in which case as $g_{A_0, B_0}(x)$ is more popular than $u$,
        the majority value above will be unchanged and equal to $g_f (x)$. In this case we get that $g_f(x) \neq f_x(u)$.
		
		Fix $B\in Z$ so that $g_{(A_0, B_0)}(B) \overset{>10W^2}{\neq} F[A_0, B]|_{B}$. On each coordinate $x\in B$ where the disagreement occurs, letting
        $u =  F[A_0, B]|_{B}(x)$ we have $g_f(x) \neq f_x(u)$ with probability at least $1/4$. However, $G_f[A_0, B]|_x = f_x(u)$ by definition and hence
        $g_f(x) \neq G_f[A_0, B]|_x$ with probability at least $1/4$. For each coordinate $x\in B$ of disagreement these events are independent (as $f_x$ are independently chosen), so by Chernoff's bound with probability $1-2^{-\Omega(W^2)}$ the function $g_f$ disagrees with $ G_f[A_0, B]|_B$ on at least $2W^2$ coordinates, in which case
        $B\in Z_f$.

        By linearity of expectation we conclude that the expected density of $Z_f$ inside $\Cons_T$ is at least half the density of $Z$;
        thus, as the density of $Z_f$ in $\Cons_T$ is at most $\nu$ it follows that the density of $Z$ in $\Cons_T$ is at most $2\nu$ and the proof is concluded.
	\end{proof}

	\subsection{Direct Product Testing: From Sets to Product Distributions}
	\label{sec:product_1}
	
	In this section, we move from the uniform setting to a product setting. Consider the $q$-biased measure over $P([n])$, i.e. $\mu_q^{\otimes n}(A) = q^{\card{A}}(1-q)^{n-\card{A}}$, and let
	$G\colon (P[n], \mu_q^{\otimes n})\to[m]^{\leq n}$ be an assignment that to each $A\in P([n])$ assigns a string $G[A]\in [m]^{|A|}$
	in a {\it locally consistent manner}. Namely, for $\alpha\in (0,1)$, consider the distribution $\mathcal{D}_{q,\alpha}$
	over $A, A'\subseteq [n]$ that results from by taking, for each $i\in [n]$
	independently, $i$ to in $A\cap A'$ with probability $\alpha q$, to be in $A\setminus A'$ with probability $(1-\alpha)q$ and to be in $A'\setminus A$ with
	probability $(1-\alpha)q$. Suppose that
	\[
	\Prob{(A,A')\sim \mathcal{D}_{q,\alpha}}{G[A]|_{(A\cap A')}  \overset{\leq T}{\neq}  G[A']|_{(A\cap A')}}\geq \eps.
	\]
	Using Theorem~\ref{thm:DPtest}, we show that $G$ must be correlated with a global $S\in [m]^n$.

	\begin{corollary}\label{cor:DP}
        For all $C, c_0>0$ there is $c>0$ such that the following holds for sufficiently large $n\in\mathbb{N}$ and $\eps\geq 2^{-n^{c}}$.

		Fix the alphabet $[m]$, $q\in (0,1)$, $T=\left(\frac{\log(1/\eps)}{q}\right)^{c_0}$ and suppose that $0.9>\alpha\geq\frac{1}{\log^C(1/\eps)}$.
		If $G\colon (P[n], \mu_q^{\otimes n})\to [m]^{\leq n}$ satisfies
		\[
		\Prob{(A,A')\sim \mathcal{D}_{q,\alpha}}{G[A]|_{(A\cap A')}  \overset{\leq T}{\neq} G[A']|_{(A\cap A')}}\geq \eps.
		\]
		Then there exists a string $S\in [m]^n$ such that $\Prob{A\sim_{q}[n]}{\Delta(G[A], S|_{A}) \leq r}\geq \delta$, where $\delta = \eps^{O(\log (1/q)^2)}$
        and $r=\left(\frac{\log{(1/\eps)}}{q}\right)^{O(1)}$.
	\end{corollary}
	\begin{proof}
		Fix $N= \omega(n^4)$. Given a function $G: (P[n], \mu^{\otimes n}) \rightarrow [m]^{\leq n}$ where $G(A)$ can be thought of as a string in $[m]^{|A|}$ by specifying a fixed arbitrary ordering on $[n]$, we define a map $\tilde{G} : {[N] \choose qN} \rightarrow [m]^{qN}$ as follows. For a set $S\in {[N] \choose qN}$, define $\tilde{G}(S)|_{S\cap [n]} = G(S\cap [n])$ and $\tilde{G}(S)|_{S\setminus [n]} =  0^{|S\setminus [n]|}$.
		We know that
		$$\Pr_{(A,A')\sim \mathcal{D}_{q,\alpha}}[G[A]|_{(A \cap A')} \overset{\leq T}{\neq} G[A']|_{(A \cap A')}]\geq \eps.$$
		Instead of checking consistency on $A \cap A'$, we select a set $A'' \subseteq A \cap A'$ by independently including $i\in A\cap A'$ to $A''$ with probability $q'/\alpha q$. This can only increase the acceptance probability, and hence
		
		\begin{equation}\label{eq:indep_acc}
			\Pr_{\substack{(A,A')\sim \mathcal{D}_{q,\alpha}\\ A''\sim_{q'/\alpha q} A\cap A'}}[G[A]|_{A''} \overset{\leq T}{\neq} G[A']|_{A''}]\geq \eps.
		\end{equation}
		
		We denote the overall distribution on $({A}, {A}', {A}'')$ from the above probability by $\mathcal{D}$. Now consider selecting the pairs $(A_0, B_0)$ and $(A_0, B_1)$ according to the \maintest with parameters $(q, q', \eta)$ for the table $\tilde{G}$. Let $\tilde{A} = A_0\cup B_0 \cap [n]$,  $\tilde{A}' = A_0\cup B_1 \cap [n]$ and $\tilde{A}'' = A_0\cap [n]$. We use $\tilde{\mathcal{D}}$ to denote the distribution on $(\tilde{{A}}, \tilde{{A}}', \tilde{{A}}'')$. We show that the statistical distance between the distributions ${\mathcal{D}}$  and $\tilde{\mathcal{D}}$  is at most $o(1)$.
		\begin{claim}\label{claim:bd_sd_dp}
			The statistical distance between the distributions ${\mathcal{D}}$  and $\tilde{\mathcal{D}}$  is at most $e^{-\Omega(n)}$ when $\alpha q = q'+\frac{(q-q')^2}{(1-q')}$.
		\end{claim}
		\begin{proof}
			Deferred to Section~\ref{sec:missing_dp}.
        \end{proof}

        Using Claim~\ref{claim:bd_sd_dp} and~\eqref{eq:indep_acc}, we conclude	$\Pr[G[\tilde{A}]|_{\tilde{A}''} \overset{\leq T}{\neq}  G[\tilde{A}']|_{\tilde{A}''}]\geq \eps -2^{-\Omega(n)}$. Since $\tilde{G}(S)|_i = 0$ for every $i>n$ and $S\ni i$, we have	
	\[
        \Pr_{(A_0, B_0), (A_0, B_1)}[\tilde{G}[A_0, B_0] |_{A_0}  \overset{\leq T}{\neq}  \tilde{G}[A_0, B_1]|_{A_0}]\geq \eps - 2^{-\Omega(n)},
    \]
	Therefore, using Theorem~\ref{thm:DPtest}, we conclude that there exists a global function $\tilde{g}: [N] \rightarrow [m]$ such that	
	\[
        \Pr_{S\in {[N]\choose qN}}\left[\tilde{G}[S] \overset{\leq r}{\neq} \tilde{g}(S)\right] \geq \delta,
    \]
    where $r$ and $\delta$ are as in Theorem~\ref{thm:DPtest}. Furthermore, based on how we came up with the global function $\tilde{g}$, we have $\tilde{g}(i) = 0$ for all $i\in (n, N]$ as $\tilde{G}(S)|_i = 0$ for all $S\in {[N]\choose qN}$ and $i\in S$. If we let $g: [n] \rightarrow [m]$ be the function $\tilde{g}$ restricted to the domain $[n]$, we have	
	\[
    \Pr_{A \sim \mu_q^{\otimes n}}\left[{G}[A] \overset{\leq r}{\neq} {g}(A)\right] \geq \delta - 2^{-\Omega(n)}.
    \]
	Here,  we used Claim~\ref{claim:sd_prod_set} that shows the statistical distance between the distribution $\mu_q^{\otimes n}$ and the distribution on $S|_{[n]}$ where $S$ is a uniformly random set from ${[N] \choose qN}$ is at most $2^{-\Omega(n)}$. This concludes the proof.
\end{proof}

	\subsection{Getting the Final Direct Product Theorem: Proof of Theorem~\ref{thm:DP_biased_version}}
	\label{sec:product_2}

	In this section we analyze the main direct product test from Theorem~\ref{thm:DP_biased_version}, which is a slight variant of the test given in Corollary~\ref{cor:DP}. Let $G\colon (P[n], \mu_q^{\otimes n})\to[m]^{\leq n}$ be an assignment that to each $A\in P([n])$ assigns a string $G[A]\in [m]^{|A|}$. We show that if $G$ passes this test from Theorem~\ref{thm:DP_biased_version} with probability at least $\eps>0$, it also passes the test from Corollary~\ref{cor:DP} with probability at least $\eps/2$ and hence we get the same global structure.
	
	\begin{figure}[!h]
		\label{fig:maintest_subsetagreement}
		\fbox{
			
			\parbox{450pt}{
				\vspace{10pt}
				Given $G\colon (P[n], \mu_q^{\otimes n})\to[m]^{\leq n}$,
				\begin{itemize}
					\item Sample $(A, A')$ from the distribution $\mathcal{D}_{q,\alpha}$.
					\item Select a random subset $\tilde{A}_0\subseteq A\cap A'$ by adding  each $i\in  A\cap A'$ to $\tilde{A}_0$ with probability $\beta$ independently.
					\item Check if $G[A]|_{\tilde{A}_0} =  G[A']|_{\tilde{A}_0}$
				\end{itemize}
			}
		}
		\caption{\subsetmaintest with parameters $(q,\alpha, \beta)$.}
	\end{figure}

	We first analyze the \subsetmaintest  given in Figure~\ref{fig:maintest_subsetagreement}. In comparison to Corollary~\ref{cor:DP}, in this test we check for complete agreement (i.e., $\eta=0$), but only on a subset of coordinates from $A\cap A'$.
	
	\begin{thm}
		\label{thm:DPtest_subsetagreement}
         For all $C, c_0>0$ there is $c>0$ such that the following holds for sufficiently large $n\in\mathbb{N}$ and $\eps\geq 2^{-n^{c}}$.

		Suppose that $\alpha, \beta \geq \frac{1}{\log^C(1/\eps)}$, $\alpha\leq \frac{9}{10}$ and
        further suppose that $G\colon (P[n], \mu_q^{\otimes n})\to[m]^{\leq n}$ passes the \subsetmaintest with parameters $(q,\alpha, \beta)$ with probability at least $\eps$,
        then there exists a string $S\in [m]^n$ such that
        \[
        \Prob{A\sim_{q}[n]}{\Delta(G[A], S|_{A}) \leq r}\geq \delta,
        \]
		where $\delta = \eps^{O(\log (1/q)^2)}$ and $r=\left(\frac{\log{(1/\eps)}}{q}\right)^{O(1)}$.
\end{thm}
	\begin{proof}
		Let $c_0\gg C$. We show that if $G$ passes the \subsetmaintest with parameters $(q, \alpha, \beta)$ with probability $\eps$, then it also passes the the test from Corollary~\ref{cor:DP}, with respect to the distribution $\mathcal{D}_{q,\alpha}$ and $\eta = \left(\frac{1}{q\cdot \beta}\right)^{c_0}$, with probability at least $\eps/2$. Let $E_{>\eta}$ be the event that $G[A]|_{A\cap A'} \overset{>\eta}{\neq}   G[A']|_{A\cap A'}$, and $E_{\leq \eta}$ be the event that $G[A]|_{A\cap A'}\overset{\leq \eta}{\neq} G[A']|_{A\cap A'}$. Let $p_{>\eta}$ and $p_{\leq \eta}$ be the probability of  events $E_{>\eta}$ and $E_{\leq\eta}$ respectively. We have
		\begin{align*}
			\eps\leq \Pr_{\substack{(A, A')\sim\mathcal{D}_{q,\alpha},\\ \tilde{A_0}\sim_\beta A\cap A'}}[G[A]|_{\tilde{A}_0} =  G[A']|_{\tilde{A}_0}]
			 &= p_{>\eta} \cdot \Pr_{\substack{(A, A')\sim\mathcal{D}_{q,\alpha},\\ \tilde{A_0}\sim_\beta A\cap A'}}\left[G[A]|_{\tilde{A}_0} =  G[A']|_{\tilde{A}_0}\mid E_{>\eta}\right]
\\ &+ p_{\leq \eta} \cdot \Pr_{\substack{(A, A')\sim\mathcal{D}_{q,\alpha},\\ \tilde{A_0}\sim_\beta A\cap A'}}\left[G[A]|_{\tilde{A}_0} =  G[A']|_{\tilde{A}_0}\mid E_{\leq \eta}\right]\\
			& \leq  1\cdot (1-\beta)^\eta + p_{\leq \eta}\cdot 1.\\
			& \leq \eps/2 +  p_{\leq \eta}.
		\end{align*}
		This shows that $ p_{\leq \eta}\geq \eps/2$ and hence $G$ passes the test from Corollary~\ref{cor:DP} with probability at least $\eps/2$. We can now apply  Corollary~\ref{cor:DP} on $G$ to get the conclusion.

	We are now ready to prove Theorem~\ref{thm:DP_biased_version}.
	
	\paragraph{Proof of Theorem~\ref{thm:DP_biased_version}:} The distribution in the test $\textsf{DP}(\rho, \alpha, \beta)$ can be simulated by the \subsetmaintest with parameters $(q, \tilde{\alpha}, \tilde{\beta})$ where $q=\rho$, $\tilde{\alpha}q = \rho \alpha + (\rho-\alpha\rho)^2$  and $\tilde{\beta}\tilde{\alpha} q = \beta \alpha\rho$. Therefore, the proof follows from Theorem~\ref{thm:DPtest_subsetagreement}.

	\subsection{Small-set Expansion Property of the Graphs over a Multi-slice}
	\label{section:sse}
	In this section, we show the small set expansion property of the graph as stated in Lemma~\ref{lemma:sse_prem}.
	
	Recall the graph $\mathcal{G}_n$ that we defined on  the set of vertices$\{(A, B) \mid A, B\subseteq [n], A\cap B = \emptyset, |A| = q'n, |B|=(q-q')n\}$. A random neighbor $(A', B')$ of $(A, B)$ in this graph is sampled conditioned on the fact that $A'$ is distributed uniformly conditioned on $|A'| = q'n$ , $|A\cap A'| = cq'n$ and $B'\subseteq[n]\setminus A'$ is a uniformly random set of size $(q-q')n$. In this section, we deduce the small-set expansion property of the graph $\mathcal{G}_n$. We can view the vertices of the above graph $\mathcal{G}_n$ as the multi-slice of $\{0,1,2\}^n$ -- map the vertex $(A, B)$ to $\x\in \{0,1,2\}^n$ where $x_i = 1$ if $i\in A$, $x_i = 2$ if $i\in B$ and $x_i=0$ if $i\in [n]\setminus (A\cup B)$.
	Let us denote the multi-slice by $\mathcal{U}_n$.
	
	One can view the multi-slice $\mathcal{U}_n$ as a quotient space $S_n/(S_{(1-q)n} \times S_{q'n} \times S_{(q-q')n})$, which is useful in lifting the standard representation-theoretic decomposition of functions over $S_n$ to decompositions of functions over $\mathcal{U}_n$. In order to state the relevant lemmas from~\cite{BravermanKLM22}, we need the following few definitions.
	\begin{definition}
		A function $f: S_n \rightarrow \mathbb{R}$ is called a $d$-junta if there exists a set of coordinate $A\subseteq [n]$ of size at most $d$ such that $f(\pi) = g(\pi(A))$  for some function $g : [n]^A \rightarrow \mathbb{R}$.
	\end{definition}
	For two function $f, g: S_n \rightarrow \mathbb{R}$, define the inner product $\langle f, g\rangle$ as $\E_{\pi}{[f(\pi)g(\pi)]}$.
	
	\begin{definition}
		For $d = 0,1,\ldots, n$ we denote by $V_d(S_n) \subseteq \{f : S_n \rightarrow R\}$ the span of $d$-juntas. Also, define $V_{=d}(S_n) = V_d(S_n) \cap V_{d-1}(S_n)^\perp$.
	\end{definition}
	Therefore, we can write the space of real-valued functions as $V_{=0}(S_n)\oplus V_{=1}(S_n) \oplus \ldots \oplus V_{=n-1}(S_n)$, and
	thus write any $f : S_n \rightarrow  \mathbb{R}$ uniquely as $f  = \sum_{i=0}^{n-1} f^{=i}$ where  $f^{=i}\in V_{=i}(S_n)$. Let $V_{\geq d}(\mathcal{U}_n)$ ($V_{\leq d}(\mathcal{U}_n)$) be the span of functions over $\mathcal{U}_n$ whose degree is at least (at most) $d$.\\

	We say a distribution $\mu$ over $([3] \times  [3])^n$ commutes with the action of $S_n$ if the following distributions are the same for all $\pi\in S_n$ and $x \in [3]^n$: a) $\x'$, where $(\x, \x')\sim \mu$ conditioned on $\x = \pi(x)$, and b) $\pi(\x')$, where $(\x, \x')\sim \mu$ conditioned on $\x = x$. The following claim shows that the operator $\mathcal{T}$ that commutes with $S_n$ preserves the degree of the functions.
	
	\begin{claim}
		\label{claim:spaces_preserved}(Claim 3.6 from ~\cite{BravermanKLM22})
		Suppose $\mathcal{T}$ is an operator that commutes with the action of $S_n$ on functions over the multi-slice $\mathcal{U}_n$. Then for each $0\leq d < n$, we have that $\mathcal{T}(V_{=d}(\mathcal{U}_n)) \subseteq V_{=d}(\mathcal{U}_n)$.
	\end{claim}
	We observe the following few properties of the multi-slice $\mathcal{U}_n$ and the operator $\mathcal{T}(\mathcal{G}_n)$.
	\begin{enumerate}
		\item In a multi-slice, every symbol appears the same number of times in every element of the multi-slice. If we let $k_i$ be the number of times the symbol $i$ appears, then the multi-slice is called {\em $\alpha$-balanced} if $k_i\geq \alpha n$ for every $i$. The multi-slice $\mathcal{U}_n$ is $\alpha$-balanced for $\alpha = \min\{q', (q-q'), (1-q)\}$.
		\item The edge distribution of the graph $\mathcal{G}_n$ is a distribution on the multi-slices $\mathcal{U}_n\times \mathcal{U}_n$. A distribution $\mu$ on $\mathcal{U}_n\times \mathcal{U}_n$ is called $\alpha$-admissible if a) the distribution is symmetric under $S_n$, and b) for all $(a,b) \in \{0,1,2\} \times \{0,1,2\}$, the quantity $\Pr_{(\x, \y)\sim \mu, i\in [n]}{[x_i = a\ \&\ y_i = b]}$ is either at least $\alpha$ or $0$. It can be easily observed that the edge distribution of $\mathcal{G}_n$ is $\alpha$-admissible for $\alpha =\Omega((cq')^2)$.
		\item A distribution $\mu$ on $\mathcal{U}_n\times \mathcal{U}_n$ also called {\em connected} iff the bipartite graph $(V_1\cup V_2, E)$ where a) $V_i$ is the corresponding support of the marginal distribution of $\mu$, and  b) $(\x, \y)\in E$ iff $(\x, \y)$ is in the support of $\mu$, is connected. Here again, it is easy to observe the connectedness property of the edge distribution of the graph $\mathcal{G}_n$.
		\item Finally, the operator $\mathcal{T}(\mathcal{G}_n)$ commutes with  the action of $S_n$.
	\end{enumerate}
	
	One of the important characteristic of $\alpha$-admissible and connected distributions, as shown in ~\cite{BravermanKLM22}, is that it can be replaced by a certain product distribution $\nu^{\otimes n}$ as far as low-degree functions are concerned. Thus, this gives a way to prove analytical results for a multi-slice by invoking the corresponding results over a product distribution.

	The following lemma from~\cite{FOW,FOWitcs} gives an upper bound on $\|\mathcal{P}_{\leq d}\|_{2\rightarrow 4}$, where $\mathcal{P}_{\leq d}$ is the projector operator into the space $V_{\leq d}(\mathcal{U}_n)$. It crucially uses the fact that $\mathcal{U}_n$ is $\alpha$-balanced.
	
	\begin{lemma}[Lemma 29 from~\cite{FOW}]
		\label{lemma:hc}
		For all $c>0$ and  $0<q'< q< 1$ and $d\in \mathbb{N}$, if $f : \mathcal{U}_n \rightarrow \mathbb{R}$ is a function of degree at most $d$, then $\|f\|_4 \leq \left(\frac{1}{cq'}\right)^{O(d)} \|f\|_2$.
	\end{lemma}
	
	Finally, we need the following lemma that bounds the eigenvalues of the operator $\mathcal{T}(\mathcal{G}_n)$. This lemma uses the fact the the edge distribution is $\alpha$-admissible and connected.
	
	\begin{lemma}[Lemma 3.11 from~\cite{BravermanKLM22}]
		\label{lemma:evs_multislice graph}
		There are constants $C>0$ and $\delta>0$ such that for all $c>0$ and  $0<q'< q< 1$ such that $q' = \Omega(q)$,
        for all $d\in \mathbb{N}$, if $f\in V_{>d}(\mathcal{U}_n)$, we have $\|\mathcal{T}(\mathcal{G}_n) f\|_2 \leq  C(1+\delta)^{-\frac{d}{\log(1/cq')}}\|f\|_2$.
	\end{lemma}
	
	We note that the Lemma 3.11 from~\cite{BravermanKLM22} gives a bound of the form $\|\mathcal{T}(\mathcal{G}_n) f\|_2 \leq  C(1+\delta)^{-d}\|f\|_2$, where the parameters $C$ and $\delta$ depend on the parameters $q', q$, or more generally the parameter $\alpha$ from the $\alpha$-balancedness and $\alpha$-admissible property. The dependence of $\alpha$ on $\delta$ can be as bad as $\delta \sim \alpha^{-1}$ and such a bound will give us
	$\|\mathcal{T}(\mathcal{G}_n) f\|_2 \leq  C(1+\delta)^{- d\cdot (cq')^2}\|f\|_2$ which is not sufficient for our purpose.
	We observe that a slight adjustment of their proof gives a better quantitative bound as stated in the lemma in our current setting. We first describe the difference between our setting and the setting in  ~\cite{BravermanKLM22} and then briefly sketch the proof of Lemma~\ref{lemma:evs_multislice graph}.
	
	The $\alpha$-admissible property is used to get bound the second eigenvalue of a graph $H$ defined as follows: The vertex set is $\{0,1,2\}$ and for $a, b \in \{0,1,2\}$, the weight of the edge $(a, b)$ is
	$$w(a, b) = \Pr_{(x,x') \sim \mathcal{G}_n, i\in [n]}[ x'_i = b \mid x_i = a].$$
	Therefore, if the edge distribution of $\mathcal{G}_n$ is $\alpha$-admissible, then the non-zero weight of an edge in the graph $H$ is at least $\alpha$, and hence $\lambda_2(H)\leq 1-\alpha^2/2$. In our case, the stationary distribution of $H$ is given by $\nu$ where $\nu(0) = q'$, $\nu(1) = (q-q')$ and $\nu(2)= (1-q)$. As $q'= \Omega(q)$, in our case, we have $\lambda_2(H) = 1-\Omega(1)$. Therefore, we do not lose much in the final bound. We now sketch the proof and assume readers to be familiar with the concepts from ~\cite{BravermanKLM22}.
	
	\paragraph{Proof sketch of Lemma~\ref{lemma:evs_multislice graph}:} For convenience we write $T:= \mathcal{T}(\mathcal{G}_n)$ and let $\alpha = \Omega((cq')^2)$.  Suppose $f\in V_{>d}(\mathcal{U}_n)$. We can write the function $f$ in as the sum of eigenvectors $f_1, f_2, \ldots$ of the operator $T$. Therefore, it is enough to given an upper bound the  eigenvalue of $T$ corresponding to the functions $f_i$, $\|T f_i \|_2\leq C(1+\delta)^{-d} \|f_i\|_2$. Let $\theta$ be the eigenvalue corresponding to the function $f_i$. In ~\cite{BravermanKLM22}, the authors show $\theta \leq C(1+\delta)^{-d}$ using the trace method as follows: suppose the multiplicity of $\theta$ is $m$, then using the fact that the trace of an operator is the sum of its eigenvalues, we have
	$$m\theta \leq \mathrm{Tr}(T).$$
	This gives a bound $\theta \leq \frac{\mathrm{Tr}(T)}{m}$. In order to get a reasonable bound, instead of bounding the trace of $T$, in ~\cite{BravermanKLM22}, the authors work with $T^h$ for some $h\geq 1$ and  get
	$$\theta^h \leq \frac{\mathrm{Tr}(T^h)}{m}.$$
	
	Since the degree of $f$ was at least $d$ to begin with, using this fact one can show that $m \geq c_0^d$, for some absolute constant $c_0>1$. Therefore, it is enough to bound $\mathrm{Tr}(T^h)$. Note that if we show $\mathrm{Tr}(T^h)\leq (1+\xi)^{d}$, $\xi$ is an absolute constant and $1+\xi < c_0$ , then we get $\theta = (1+\delta)^{-d/h}$. Therefore, if we can show that $\mathrm{Tr}(T^h)\leq (1+\xi)^{d}$ when $h = \Theta(\log(1/cq'))$, then this will be enough to prove  Lemma~\ref{lemma:evs_multislice graph}. We show this is indeed the case when $\lambda_2(H) \leq 1-\Omega(1)$.
	
	The crucial point is that as $\lambda_2(H) \leq 1-\Omega(1)$, we have the following claim which gives a stronger quantitative bound on $h$ with respect to $\eps$.
	
	\begin{claim}
		\label{claim:walk_on_H}
		For all $\eps>0$ and $h = O(\log(1/\eps))$ such that for all $v_1\in [3]$
		$$\left|\Pr_{\substack{v_2, v_3, \ldots, v_h\\\text{ random walk on $H$ from $v_i$}}}[v_h = a] - \nu(a) \right| \leq \eps.$$
	\end{claim}
	\begin{proof}
		Using the standard spectral argument, the left-hand side quantity can be upper bounded by $\lambda_2(H)^h$ and since $\lambda_2(H) \leq 1-\Omega(1)$, the claim follows.
	\end{proof}
	
	In order to give a concrete bound on the trace of $T^h$ using combinatorial analysis, the proof goes by writing the operator $T^h$ as a sum of operators of the form $R_{\overrightarrow{r}}$ where $\vec{r}= (r_{a,b})_{a,b\in [3]}$ that sum to $n$. The operator $R_{\vec{r}}$ corresponds to the distribution on $(\mathbf{x}, \mathbf{y})$ where $\mathbf{x}\in \mathcal{U}_n$ is uniform and $\mathbf{y}$ is sampled according to $T^h$ conditioned on the statistics if $(\mathbf{x}, \mathbf{y})$ being $\vec{r}$. Thus, we can write
	
	$$T^h = \sum_{\vec{r}} p_{\vec{r}} R_{\vec{r}}.$$
	
	Call $\vec{r}$  $\eps$-reasonable $\left| r_{a,b} - n\nu(a)\nu(b)\right|\leq 3\eps n$ for all $a,b\in [3]$; otherwise call $\vec{r}$ unreasonable.
	Claim 3.15 from~\cite{BravermanKLM22} shows that $\sum_{\mbox{unresonable }  \vec{r}} p_{\vec{r}} \leq 100\alpha^{-1}e^{-\alpha\eps^2 n}$. Therefore, we can focus on reasonable $\vec{r}$.  The actual argument is based on weather $d \geq \gamma n$ for some constant $\gamma$ of $d\leq \gamma n$. Since we do not care about the constants in the exponent, we can focus on the latter case when $d <\gamma n$.
	
	As shown in~\cite{BravermanKLM22}, when $d\leq \gamma n$, one can assume that $n = 3d$ in order to compute the trace of the operator $R_{\vec{r}}$ (and hence the trace of $T^h$). Unraveling the quantitative bounds, the Claim 3.24 from the paper shows that
	
	$$\mathrm{Tr}(R_{\vec{r}}) \leq \left( 1 + O\left(\frac{\eps}{\alpha^2} \right)\right)^d.$$
	
	Therefore if we set $\eps = O_{c_0}(\alpha^2)$, then using Claim~\ref{claim:walk_on_H}  we get $\mathrm{Tr}(T^h)\leq (1+\xi)^d$ for $h = \Theta(\log(1/cq'))$ as required, ignoring the unreasonable $\vec{r}$ as their total weight is $\ll (1+\xi)^d$.
	\qed\\

	We are now ready to prove the small-set expansion property of the graph $\mathcal{G}_n$.
	\begin{lemma}
		[Restatement of Lemma~\ref{lemma:sse_prem}]
		For every $c>0$,  $0<q'< q< 1$ and $\mu>0$, the graph $\mathcal{G}_n$ defined above has
		$$\phi_{\mathcal{G}_n}(\mu) \geq 1 - \mu^{\Omega\left(\frac{1}{\log(1/cq')^2}\right)}.$$
	\end{lemma}
	\begin{proof}
		Fix any set $S\subseteq V(\mathcal{G}_n)$ of density at most $\mu$. Let $f : \mathcal{U}_n \rightarrow \{0,1\}$ be the indicator function of $S$. As stated at the begining of the section, we can write $f$ as $f = \sum_{i=0}^{n-1} f^{=i}$ where $f^{=i} \in V_{=i}(\mathcal{U}_n)$. Let $f = f_1 + f_2$, where the component $f_1 =  \sum_{i=0}^{d} f^{=i}$ and $f_2  =  \sum_{i=d+1}^{n-1} f^{=i}$ for some $d$ to be fixed later. By letting  $\mathcal{T}:= \mathcal{T}(\mathcal{G}_n)$, we have
		\begin{align*}
			\phi_{\mathcal{G}_n}(S) = 1 - \frac{\langle f, \mathcal{T} f\rangle }{\mu}.
		\end{align*}
		By letting $\tau : =  C(1+\delta)^{-\frac{d}{\log(1/cq')}}$ from Lemma~\ref{lemma:evs_multislice graph}, we can bound
		\begin{align*}
			\langle f, \mathcal{T} f\rangle &= \langle f_1, \mathcal{T} f_1\rangle  + \langle f_2,  \mathcal{T} f_2\rangle \tag*{(Using Claim~\ref{claim:spaces_preserved} and orthogonality of spaces $V_{=i}(\mathcal{U}_n)$)}\\
			&\leq \|f_1\|_2^2 + \tau \|f_2\|_2^2 \tag*{(Using Claim~\ref{lemma:evs_multislice graph})}\\
			&\leq \|f_1\|_2^2 + \tau \mu.
		\end{align*}
		Therefore,
		\begin{equation}
			\label{eq:expansion_lb1}
			\phi_{\mathcal{G}_n}(S) \geq 1 - \tau - \frac{\|f_1\|_2^2}{\mu}.
		\end{equation}
		If we let $\mathcal{P}_{\leq d}$ be the projector operator into the subspace $V_{\leq d}(\mathcal{U}_n)$, then we have
		\begin{align*}
			\| \mathcal{P}_{\leq d}\|_{4/3 \rightarrow 2} = \max_{g\neq 0} \frac{ \|\mathcal{P}_{\leq d} g\|_{2}}{\|g\|_{4/3}} \geq \frac{ \|\mathcal{P}_{\leq d} f\|_{2}}{\|f\|_{4/3} }=  \frac{ \|f_1\|_{2}}{\|f\|_{4/3}}.
		\end{align*}
		Since $\|f\|_{4/3} = \mu^{3/4}$, we have $\|f_1\|_2^2 \leq \| \mathcal{P}_{\leq d}\|^2_{4/3 \rightarrow 2}\cdot \mu^{3/2}$. Using the fact that $\| \mathcal{P}_{\leq d}\|^2_{4/3 \rightarrow 2} = \| \mathcal{P}_{\leq d}\|^2_{2 \rightarrow 4}$, we get $\|f_1\|_2^2 \leq \| \mathcal{P}_{\leq d}\|^2_{2 \rightarrow 4}\cdot \mu^{3/2}$. Therefore,
		\begin{equation}
			\label{eq:expansion_lb2}
			\phi_{\mathcal{G}_n}(S) \geq 1 - \tau -\| \mathcal{P}_{\leq d}\|^2_{2 \rightarrow 4}\cdot \mu^{1/2}.
		\end{equation}
		Finally, using Lemma~\ref{lemma:hc}, we can bound $\| \mathcal{P}_{\leq d}\|^2_{2 \rightarrow 4}$ as follows.For any function $g$, let $g_1$ be the component of $g$ from $V_{\leq d}$ and $g_2:=g - g_1$ be orthogonal to $g_1$.
		$$\| \mathcal{P}_{\leq d}\|^2_{2 \rightarrow 4} = \max_{g:=g_1+g_2\neq 0}\frac{\|\mathcal{P}_{\leq d} g\|^2_4}{\|g\|^2_2} = \max_{g:=g_1+g_2}\frac{\|g_1\|^2_4}{\|g_1\|_2^2 + \|g_2\|_2^2} \leq \max_{g_1\in V_{\leq d}}\frac{\|g_1\|_4^2}{\|g_1\|_2^2}\leq (1/cq')^{O(d)},$$
		where the last inequality uses Lemma~\ref{lemma:hc}. Plugging this into \eqref{eq:expansion_lb2}, we get
		\begin{align*}
			\phi_{\mathcal{G}_n}(\mu) \geq 1 - C(1+\delta)^{-\frac{d}{\log(1/cq')}} - (1/cq')^{O(d)}\cdot \mu^{1/2}.	
		\end{align*}
		If we choose $d = O\left(\frac{\log(1/\mu)}{\log(1/cq')}\right)$, then it is easy to observe that $\phi_{\mathcal{G}_n}(\mu) \geq 1 - \mu^{\Omega\left(\frac{1}{\log(1/cq')^2}\right)}$.
	\end{proof}

\section*{Acknowledgements}
We thank Yang P. Liu for his careful reading of the manuscript and for spotting several errors in earlier versions.
\bibliography{ref}

\begin{thebibliography}{10}

\bibitem{BCHKS}
Mihir Bellare, Don Coppersmith, Johan H{\aa}stad, Marcos~A. Kiwi, and Madhu
  Sudan.
\newblock Linearity testing in characteristic two.
\newblock {\em {IEEE} Trans. Inf. Theory}, 42(6):1781--1795, 1996.

\bibitem{BKMinv}
Amey Bhangale, Subhash Khot, and Dor Minzer.
\newblock A mixed invariance principle and applications to csps.

\bibitem{BKMcsp1}
Amey Bhangale, Subhash Khot, and Dor Minzer.
\newblock On approximability of satisfiable \emph{k}-csps: {I}.
\newblock In {\em {STOC} '22: 54th Annual {ACM} {SIGACT} Symposium on Theory of
  Computing, Rome, Italy, June 20 - 24, 2022}, pages 976--988, 2022.

\bibitem{BKLMroth}
Amey Bhangale, Subhash Khot, and Dor Minzer.
\newblock Effective bounds for restricted 3-arithmetic progressions in
  $\mathbb{F}_p^n$.
\newblock {\em CoRR}, abs/2308.06600, 2023.

\bibitem{BKMcsp2}
Amey Bhangale, Subhash Khot, and Dor Minzer.
\newblock On approximability of satisfiable k-csps: {II}.
\newblock In {\em Proceedings of the 55th Annual {ACM} Symposium on Theory of
  Computing, {STOC} 2023, Orlando, FL, USA, June 20-23, 2023}, pages 632--642,
  2023.

\bibitem{BKMcsp3}
Amey Bhangale, Subhash Khot, and Dor Minzer.
\newblock On approximability of satisfiable k-csps: {III}.
\newblock In {\em Proceedings of the 55th Annual {ACM} Symposium on Theory of
  Computing, {STOC} 2023, Orlando, FL, USA, June 20-23, 2023}, pages 643--655,
  2023.

\bibitem{BLR}
Manuel Blum, Michael Luby, and Ronitt Rubinfeld.
\newblock Self-testing/correcting with applications to numerical problems.
\newblock {\em J. Comput. Syst. Sci.}, 47(3):549--595, 1993.

\bibitem{BravermanGarg}
Mark Braverman and Ankit Garg.
\newblock Small value parallel repetition for general games.
\newblock In {\em Proceedings of the Forty-Seventh Annual {ACM} on Symposium on
  Theory of Computing, {STOC} 2015, Portland, OR, USA, June 14-17, 2015}, pages
  335--340. {ACM}, 2015.

\bibitem{BravermanKLM22}
Mark Braverman, Subhash Khot, Noam Lifshitz, and Dor Minzer.
\newblock An invariance principle for the multi-slice, with applications.
\newblock In {\em 2021 IEEE 62nd Annual Symposium on Foundations of Computer
  Science (FOCS)}, pages 228--236. IEEE, 2022.

\bibitem{GHZgameBKM}
Mark Braverman, Subhash Khot, and Dor Minzer.
\newblock Parallel repetition for the {GHZ} game: Exponential decay.
\newblock {\em CoRR}, abs/2211.13741, 2022.

\bibitem{DinurFH19}
Irit Dinur, Yuval Filmus, and Prahladh Harsha.
\newblock Analyzing boolean functions on the biased hypercube via
  higher-dimensional agreement tests.
\newblock In {\em Proceedings of the Thirtieth Annual ACM-SIAM Symposium on
  Discrete Algorithms}, pages 2124--2133. SIAM, 2019.

\bibitem{DinurG08}
Irit Dinur and Elazar Goldenberg.
\newblock Locally testing direct product in the low error range.
\newblock In {\em 2008 49th Annual IEEE Symposium on Foundations of Computer
  Science}, pages 613--622. IEEE, 2008.

\bibitem{DHVY}
Irit Dinur, Prahladh Harsha, Rakesh Venkat, and Henry Yuen.
\newblock Multiplayer parallel repetition for expanding games.
\newblock In {\em 8th Innovations in Theoretical Computer Science Conference,
  {ITCS} 2017, January 9-11, 2017, Berkeley, CA, {USA}}, volume~67 of {\em
  LIPIcs}, pages 37:1--37:16, 2017.

\bibitem{DinurR06}
Irit Dinur and Omer Reingold.
\newblock Assignment testers: Towards a combinatorial proof of the pcp theorem.
\newblock {\em SIAM Journal on Computing}, 36(4):975--1024, 2006.

\bibitem{DinurSteurer}
Irit Dinur and David Steurer.
\newblock Analytical approach to parallel repetition.
\newblock In {\em Symposium on Theory of Computing, {STOC} 2014, New York, NY,
  USA, May 31 - June 03, 2014}, pages 624--633. {ACM}, 2014.

\bibitem{DinurS14}
Irit Dinur and David Steurer.
\newblock Direct product testing.
\newblock In {\em 2014 IEEE 29th Conference on Computational Complexity (CCC)},
  pages 188--196. IEEE, 2014.

\bibitem{FOWitcs}
Yuval Filmus, Ryan O'Donnell, and Xinyu Wu.
\newblock A log-sobolev inequality for the multislice, with applications.
\newblock In Avrim Blum, editor, {\em 10th Innovations in Theoretical Computer
  Science Conference, {ITCS} 2019, January 10-12, 2019, San Diego, California,
  {USA}}, volume 124 of {\em LIPIcs}, pages 34:1--34:12. Schloss Dagstuhl -
  Leibniz-Zentrum f{\"{u}}r Informatik, 2019.

\bibitem{FOW}
Yuval Filmus, Ryan O’Donnell, and Xinyu Wu.
\newblock Log-sobolev inequality for the multislice, with applications.
\newblock {\em Electronic Journal of Probability}, 27:1--30, 2022.

\bibitem{GHMRZ}
Uma Girish, Justin Holmgren, Kunal Mittal, Ran Raz, and Wei Zhan.
\newblock Parallel repetition for the {GHZ} game: {A} simpler proof.
\newblock In {\em Approximation, Randomization, and Combinatorial Optimization.
  Algorithms and Techniques, {APPROX/RANDOM} 2021, August 16-18, 2021,
  University of Washington, Seattle, Washington, {USA} (Virtual Conference)},
  volume 207 of {\em LIPIcs}, pages 62:1--62:19, 2021.

\bibitem{GHMRZ2}
Uma Girish, Justin Holmgren, Kunal Mittal, Ran Raz, and Wei Zhan.
\newblock Parallel repetition for all 3-player games over binary alphabet.
\newblock In {\em {STOC} '22: 54th Annual {ACM} {SIGACT} Symposium on Theory of
  Computing, Rome, Italy, June 20 - 24, 2022}, pages 998--1009. {ACM}, 2022.

\bibitem{GMRZ}
Uma Girish, Kunal Mittal, Ran Raz, and Wei Zhan.
\newblock Polynomial bounds on parallel repetition for all 3-player games with
  binary inputs.
\newblock In {\em Approximation, Randomization, and Combinatorial Optimization.
  Algorithms and Techniques, {APPROX/RANDOM} 2022, September 19-21, 2022,
  University of Illinois, Urbana-Champaign, {USA} (Virtual Conference)}, volume
  245 of {\em LIPIcs}, pages 6:1--6:17, 2022.

\bibitem{Gowers}
William~T Gowers.
\newblock A new proof of {S}zemer{\'e}di's theorem.
\newblock {\em Geometric \& Functional Analysis GAFA}, 11(3):465--588, 2001.

\bibitem{Green}
Ben Green.
\newblock 100 open problems.
\newblock {\em manuscript}.

\bibitem{Has01}
Johan H{\aa}stad.
\newblock Some optimal inapproximability results.
\newblock {\em J. {ACM}}, 48(4):798--859, 2001.

\bibitem{Holenstein}
Thomas Holenstein.
\newblock Parallel repetition: Simplification and the no-signaling case.
\newblock {\em Theory Comput.}, 5(1):141--172, 2009.

\bibitem{HR}
Justin Holmgren and Ran Raz.
\newblock A parallel repetition theorem for the {GHZ} game.
\newblock {\em CoRR}, abs/2008.05059, 2020.

\bibitem{IJKW}
Russell Impagliazzo, Ragesh Jaiswal, Valentine Kabanets, and Avi Wigderson.
\newblock Uniform direct product theorems: Simplified, optimized, and
  derandomized.
\newblock {\em {SIAM} J. Comput.}, 39(4):1637--1665, 2010.

\bibitem{ImpagliazzoKW12}
Russell Impagliazzo, Valentine Kabanets, and Avi Wigderson.
\newblock New direct-product testers and 2-query pcps.
\newblock {\em SIAM Journal on Computing}, 41(6):1722--1768, 2012.

\bibitem{Meshulam}
Roy Meshulam.
\newblock On subsets of finite abelian groups with no 3-term arithmetic
  progressions.
\newblock {\em Journal of Combinatorial Theory, Series A}, 71(1):168--172,
  1995.

\bibitem{Mossel}
Elchanan Mossel.
\newblock Gaussian bounds for noise correlation of functions.
\newblock {\em Geometric and Functional Analysis}, 19(6):1713--1756, 2010.

\bibitem{MOO}
Elchanan Mossel, Ryan O’Donnell, and Krzysztof Oleszkiewicz.
\newblock Noise stability of functions with low influences: Invariance and
  optimality.
\newblock {\em Annals of Mathematics}, 171(1):295--341, 2010.

\bibitem{ODonnell}
Ryan O'Donnell.
\newblock {\em Analysis of boolean functions}.
\newblock Cambridge University Press, 2014.

\bibitem{polymath2012new}
DHJ Polymath.
\newblock A new proof of the density {H}ales-{J}ewett theorem.
\newblock {\em Annals of Mathematics}, pages 1283--1327, 2012.

\bibitem{Rag08}
Prasad Raghavendra.
\newblock Optimal algorithms and inapproximability results for every csp?
\newblock In {\em Proceedings of the fortieth annual ACM symposium on Theory of
  computing (STOC)}, pages 245--254, 2008.

\bibitem{Rao}
Anup Rao.
\newblock Parallel repetition in projection games and a concentration bound.
\newblock {\em {SIAM} J. Comput.}, 40(6):1871--1891, 2011.

\bibitem{Raz}
Ran Raz.
\newblock A parallel repetition theorem.
\newblock {\em {SIAM} J. Comput.}, 27(3):763--803, 1998.

\bibitem{Roth}
Klaus~F Roth.
\newblock On certain sets of integers.
\newblock {\em J. London Math. Soc}, 28(104-109):3, 1953.

\bibitem{szemeredi1975sets}
Endre Szemer{\'e}di.
\newblock On sets of integers containing no k elements in arithmetic
  progression.
\newblock {\em Acta Arith}, 27(299-345):21, 1975.

\bibitem{Verbitsky}
Oleg Verbitsky.
\newblock Towards the parallel repetition conjecture.
\newblock {\em Theoretical Computer Science}, 157(2):277--282, 1996.

\end{thebibliography}
\bibliographystyle{plain}
\appendix
\section{Missing Proofs}
\subsection{Merging Symbols: Proof of Lemma~\ref{lem:merge}}\label{sec:merge}
Consider $\mathrm{T}\colon L_2(\Sigma;\mu_x)\to L_2(\Gamma\times \Phi;\mu_{y,z})$ defined as
\[
\mathrm{T}f(y,z) = \cExpect{(x',y',z')\sim \mu}{y'=y,z'=z}{f(x')},
\]
and let $\mathrm{S} = \mathrm{T}^{*}\mathrm{T}$. Note that $\mathrm{S}$ can also be viewed as a Markov chain over $\Sigma$,
and for $a\in \Sigma$ we denote by $\mathrm{S}a$ the distribution over $\Sigma$ of the neighbours of $a$ according to this
Markov chain. Let $G = (\Sigma, E)$ be the graph in which the edges are the support of the Markov chain $\mathrm{S}$, and
note that $y,z$ imply $z$ if and only if the number of connected components of $G$ is exactly $\card{\Sigma}$. Thus, we assume
that the number of connected components is $\ell < \card{\Sigma}$, and we write them as $C_1\cup\ldots\cup C_{\ell}$.

Note that $\mathrm{S}$ is a symmetric operator, hence we may diagonalize it on each connected components separately. Namely,
we may find functions $\chi_{i,j}\colon \Sigma\to \mathbb{R}$ for $i=1,\ldots,\ell$ and $j=0,\ldots,\card{C_i} -1$
that are an orthogonal basis of $L_2(\Sigma;\mu_x)$ and
furthermore:
\begin{enumerate}
  \item $\chi_{i,j}$ is only supported on $C_i$.
  \item $\chi_{i,0}$ is constant on $C_i$.
  \item $\chi_{i,j}$ are eigenfunctions of $\mathrm{S}$ with eigenvalue $\lambda_{i,j}$.
\end{enumerate}
Thus, for all $i$ we have $\lambda_{i,0} = 1$, and for $j>0$ we have as in Lemma~\ref{lem:mossel_MC} that
$\lambda_{i,j}\leq 1 - \Omega_{\alpha,m}(1)$, and $\lambda_{i,j}\geq 0$ as $\mathrm{S}$ is positive semi-definite.

We consider the orthonormal basis $\chi_{\vec{i},\vec{j}}$ now over $L_2(\Sigma^{n}, \mu_x^{\otimes n})$ defined
by $\chi_{\vec{i},\vec{j}}(x) = \prod\limits_{k=1}^{n} \chi_{i_k,j_k}(x_k)$. The non-merged degree of a monomial
$\chi_{\vec{i},\vec{j}}$ is defined to be the number of $k$'s such that $j_k>0$. Thus, we can write
\[
f(x) = \sum\limits_{\vec{i},\vec{j}}\widehat{f}(\vec{i},\vec{j}) \chi_{\vec{i},\vec{j}}(x),
\qquad\qquad \text{ where }\qquad\widehat{f}(\vec{i},\vec{j}) = \inner{f}{\chi_{\vec{i},\vec{j}}}_{L_2(\Sigma^n; \mu_x^{n})}.
\]
Take $d = W\log(1/\delta)$ for a parameter $W = W(\alpha,m)>0$ to be chosen later,
and write $f = f_1 + f_2$ where $f_1$ is the part of $f$ with non-merged degree less
than $d$, and $f_2$ is the part of $f$ with non-merged at least $d$. Then
\begin{equation}\label{eq:merge}
\Expect{(x,y,z)\sim \mu}{f(x)g(y)h(z)}
=
\underbrace{\Expect{(x,y,z)\sim \mu}{f_1(x)g(y)h(z)}}_{(\rom{1})}
+
\underbrace{\Expect{(x,y,z)\sim \mu}{f_2(x)g(y)h(z)}}_{(\rom{2})},
\end{equation}
and we upper bound each term on the right hand side separately.

\paragraph{Bounding $(\rom{2})$.} For $f_2$, we have by Cauchy-Schwarz that
\begin{align*}
\card{\Expect{(x,y,z)\sim \mu^{\otimes n}}{f_2(x)g(y)h(z)}}^2
=\card{\inner{\overline{gh}}{\mathrm{T} f_2}_{\mu_{y,z}}}^2
\leq \norm{\mathrm{T} f_2}_{2; \mu_{y,z}}^2
=\inner{\mathrm{T} f_2}{\mathrm{T} f_2}_{\mu_{y,z}}
=\inner{\mathrm{S} f_2}{f_2}_{\mu_x}.
\end{align*}
To upper bound the last expression, we note that
\[
\inner{\mathrm{S} f_2}{f_2}_{\mu_x}
=\sum\limits_{\substack{\vec{i},\vec{j}\\\text{non-merge degree at least $d$}}}\card{\widehat{f}(\vec{i},\vec{j})}^2\prod\limits_{k=1}^{n} \lambda_{i,j}
\leq (1-\Omega_{\alpha,m}(1))^{d}\leq \delta,
\]
for appropriately chosen $W$. Hence, we get that $(\rom{2})\leq \sqrt{\delta}$.

\paragraph{Bounding $(\rom{1})$.}
Here, we are going to use random restrictions so that almost all of the mass of $f_1$ will collapse to monomials of non-merged degree
$0$, at which point we could truncate off the part of non-merged degree exceeding $0$ and get a function that does not distinguish between
the distribution $\mu$ and the distribution $\mu'$. More precisely, let $s = \delta^{1/3}/d$ and choose $J\subseteq [n]$ by including
each element in it with probability $s$, sample $(\tilde{x},\tilde{y},\tilde{z})\sim \mu^{J}$ and define
\[
\tilde{f_1} = (f_1)_{\overline{J}\rightarrow \tilde{x}},
\qquad
\tilde{g} = (g)_{\overline{J}\rightarrow \tilde{y}},
\qquad
\tilde{h} = (h)_{\overline{J}\rightarrow \tilde{z}}.
\]
Let $\eta>0$ be from Theorem~\ref{thm:nonembed_deg_must_be_small} for $\mu'$, and
define the events:
\begin{enumerate}
  \item $E_1$: $\norm{\tilde{f_1}}_2\geq \delta^{-\eta/10}$.
  \item $E_2$: ${\sf NEStab}_{1-K\delta/s}(\tilde{g};\mu_y)\leq \delta^{2/3}$, where $K = K(m,\alpha)>0$.
  \item $E_3$: The mass of $\tilde{f_1}$ on monomials of non-merged degree more than $0$ exceeds $\delta^{1/6}$.
\end{enumerate}
Let $E = \overline{E_1} \cap E_2\cap \overline{E_3}$; we show that $\Prob{}{E}\geq 1-3\delta^{\eta'}$ for $\eta'=\min(\eta/5,1/6)$.
We do so using the union bound. For $E_1$, we have that the expected value of $\norm{\tilde{f_1}}_2^2$ is $\norm{f_1}_2^2\leq 1$,
hence by Markov's inequality $\Prob{}{E_1}\leq \delta^{\eta/5}$. For $E_2$, we have by Lemma~\ref{lem:op_comparison_lemma} that
$\Expect{}{{\sf NEStab}_{1-K\delta/s}(\tilde{g};\mu_y)}\leq {\sf NEStab}_{1-\delta}(g;\mu_y)\leq \delta$ for suitably chosen $K$,
hence by Markov's inequality $\Prob{}{E_2}\geq 1-\delta^{1/3}$. For $E_3$, we have
\begin{align*}
\Expect{J,\tilde{x}}{\sum\limits_{\vec{i}\in [\ell]^J,\vec{j}}\widehat{\tilde{f_1}}(\vec{i},\vec{j})^2 1_{\chi_{\vec{i}, \vec{j}}\text{ has non-merge degree at least $1$}}}
&=\Expect{J}{\sum\limits_{\vec{i}\in [\ell]^n,\vec{j}}\widehat{\tilde{f_1}}(\vec{i},\vec{j})^2 1_{\sett{k\in [n]}{j_k>1}\cap J \neq\emptyset}}\\
&=\sum\limits_{\vec{i}\in [\ell]^n,\vec{j}}\widehat{\tilde{f_1}}(\vec{i},\vec{j})^2\Expect{J}{1_{\sett{k\in [n]}{j_k>1}\cap J \neq\emptyset}}\\
&\leq \sum\limits_{\vec{i}\in [\ell]^n,\vec{j}}\widehat{\tilde{f_1}}(\vec{i},\vec{j})^2 ds\\
&\leq \delta^{1/3},
\end{align*}
and so by Markov's inequality $\Prob{}{E_3}\leq \delta^{1/6}$. We write
\[
(\rom{1})=
\underbrace{\Expect{J,\tilde{x},\tilde{y},\tilde{z}}{1_{\overline{E}} \phi_{\mu}(\tilde{f_1},\tilde{g},\tilde{h})}}_{(\rom{3})}
+
\underbrace{\Expect{J,\tilde{x},\tilde{y},\tilde{z}}{1_{E} \phi_{\mu}(\tilde{f_1},\tilde{g},\tilde{h})}}_{(\rom{4})}
\]
where $\phi_{\mu}(\tilde{f_1},\tilde{g},\tilde{h}) = \Expect{(x,y,z)\sim \mu^{J}}{\tilde{f_1}(x)\tilde{g}(y)\tilde{h}(z)}$.

\paragraph{Bounding $(\rom{3})$.}
For $(\rom{3})$ we have that
\[
\card{(\rom{3})}\leq
\sqrt{\Prob{}{\overline{E}}}
\sqrt{\Expect{J,\tilde{x},\tilde{y},\tilde{z}}{\phi_{\mu}(\tilde{f_1},\tilde{g},\tilde{h})^2}}
\leq
\sqrt{\Prob{}{\overline{E}}}
\sqrt{\Expect{J,\tilde{x},\tilde{y},\tilde{z}}{\norm{\tilde{f_1}}_2^2}}
=\sqrt{\Prob{}{\overline{E}}}\norm{f_1}_2,
\]
which is at most $\sqrt{3}\delta^{\eta'/2}$.

\paragraph{Bounding $(\rom{4})$.}
For $(\rom{4})$, the point is that after random restriction the function $\tilde{f_1}$ barely notices the difference
between the distributions $\mu$ and $\mu'$, and so we can try to upper bound $\phi(\tilde{f_1},\tilde{g},\tilde{h})$
by appealing to Theorem~\ref{thm:nonembed_deg_must_be_small} over the distribution $\mu'$. The only two
issues is that $\tilde{f_1}$ still has a slight mass on monomials with non-merge degree greater than $0$, and secondly
that $\tilde{f_1}$ is not bounded, and we next address these issues.

Let $\tilde{f}_0$ be the part of $\tilde{f_1}$ of non-merge degree $0$. Then
\[
(\rom{4})
=
\underbrace{\Expect{J,\tilde{x},\tilde{y},\tilde{z}}{1_{E} \phi_{\mu}(\tilde{f_1}-\tilde{f}_0,\tilde{g},\tilde{h})}}_{(\rom{5})}
+
\underbrace{\Expect{J,\tilde{x},\tilde{y},\tilde{z}}{1_{E} \phi_{\mu}(\tilde{f}_0,\tilde{g},\tilde{h})}}_{(\rom{6})},
\]

\paragraph{Bounding $(\rom{5})$.}
For $(\rom{5})$ we clearly have that
\[
\card{(\rom{5})}
\leq \Expect{J,\tilde{x},\tilde{y},\tilde{z}}{1_{E} \norm{\tilde{f_1}-\tilde{f}_0}_2}
\leq \delta^{1/12}
\]
as the event $E_3$ fails.

\paragraph{Bounding $(\rom{6})$.}
For $(\rom{6})$, we first note that $\phi_{\mu}(\tilde{f}_0,\tilde{g},\tilde{h}) = \phi_{\mu'}(\tilde{f}_0,\tilde{g},\tilde{h})$, namely we can switch from
the distribution $\mu$ to its $x$-merge $\mu'$.
Secondly, note that for fixed $J,\tilde{x},\tilde{y}$ and $\tilde{z}$, defining
\[
f''(x) = \cExpect{(x',y',z')\sim \mu'^{J}}{x' = x}{\tilde{g}(y')\tilde{h}(z')}
\]
it holds that
\[
\card{\phi_{\mu'}(\tilde{f}_0,\tilde{g},\tilde{h})}^2
=\card{\inner{\tilde{f}_0}{f''}}^2
\leq \norm{\tilde{f}_0}_2^2\norm{f''}_2^2
\leq 2\delta^{-\eta/10}\phi_{\mu'}(\overline{f''},\tilde{g},\tilde{h}),
\]
Thus, if the event $E$ holds we may use Theorem~\ref{thm:nonembed_deg_must_be_small} on $\mu'$ to get that
$\phi_{\mu'}(\overline{f''},\tilde{g},\tilde{h})\leq M\delta^{2\eta/3}$, hence
$\card{\phi_{\mu'}(\tilde{f}_0,\tilde{g},\tilde{h})}^2\leq 2M\delta^{17\eta/30}$
and so $\card{(\rom{6})}\leq \sqrt{2M}\delta^{17\eta/60}$.

Therefore, we get that $\card{(\rom{4})}\leq M'\delta^{\eta'}$ for $M'>0$ depending only on $M$ and $\eta'>0$ depending only on $\eta$,
hence $\card{(\rom{1})}$ is upper bounded by the same type of bound, and using~\eqref{eq:merge} the proof is concluded.

\subsection{Merging Symbols: Proof of Lemma~\ref{lem:merge2}}\label{sec:merge2}
The proof here is very similar to the proof of Lemma~\ref{lem:merge}, but as the roles of $x$ and $y$ in Theorem~\ref{thm:nonembed_deg_must_be_small}
are not symmetric we give it in detail.

Consider $\mathrm{T}\colon L_2(\Gamma;\mu_y)\to L_2(\Sigma\times \Phi;\mu_{x,z})$ defined as
\[
\mathrm{T}g(x,z) = \cExpect{(x',y',z')\sim \mu}{x'=x,z'=z}{g(y')},
\]
and let $\mathrm{S} = \mathrm{T}^{*}\mathrm{T}$. Note that $\mathrm{S}$ can also be viewed as a Markov chain over $\Gamma$,
and for $a\in \Gamma$ we denote by $\mathrm{S}a$ the distribution over $\Gamma$ of the neighbours of $a$ according to this
Markov chain. Let $G = (\Gamma, E)$ be the graph in which the edges are the support of the Markov chain $\mathrm{S}$, and
note that $x,z$ imply $y$ if and only if the number of connected components of $G$ is exactly $\card{\Gamma}$. Thus, we assume
that the number of connected components is $\ell < \card{\Gamma}$, and we write them as $C_1\cup\ldots\cup C_{\ell}$.

Note that $\mathrm{S}$ is a symmetric operator, and we may diagonalize it on each connected components separately. Namely,
we may find functions $\chi_{i,j}\colon \Sigma\to \mathbb{R}$ for $i=1,\ldots,\ell$ and $j=0,\ldots,\card{C_i} -1$
that are an orthogonal basis of $L_2(\Gamma;\mu_y)$ and
furthermore:
\begin{enumerate}
  \item $\chi_{i,j}$ is only supported on $C_i$.
  \item $\chi_{i,0}$ is constant on $C_i$.
  \item $\chi_{i,j}$ are eigenfunctions of $\mathrm{S}$ with eigenvalue $\lambda_{i,j}$.
\end{enumerate}
Thus, for all $i$ we have $\lambda_{i,0} = 1$, and for $j>0$ we have as in Lemma~\ref{lem:mossel_MC} that
$\lambda_{i,j}\leq 1 - \Omega_{\alpha,m}(1)$, and $\lambda_{i,j}\geq 0$ as $\mathrm{S}$ is positive semi-definite.

We consider the orthonormal basis $\chi_{\vec{i},\vec{j}}$ now over $L_2(\Gamma^{n}, \mu_y^{\otimes n})$ defined
by $\chi_{\vec{i},\vec{j}}(y) = \prod\limits_{k=1}^{n} \chi_{i_k,j_k}(y_k)$. The non-merged degree of a monomial
$\chi_{\vec{i},\vec{j}}$ is defined to be the number of $k$'s such that $j_k>0$ and denoted by $\text{non-merge-deg}(\chi_{\vec{i},\vec{j}})$. Thus, we can write
\[
g(y) = \sum\limits_{\vec{i},\vec{j}}\widehat{g}(\vec{i},\vec{j}) \chi_{\vec{i},\vec{j}}(y),
\qquad\qquad \text{ where }\qquad\widehat{g}(\vec{i},\vec{j}) = \inner{g}{\chi_{\vec{i},\vec{j}}}_{\mu_y}.
\]
Take $d = W\delta^{-1/2}$ for a parameter $W = W(\alpha,m)>0$ to be chosen later. We now would like to carry out the argument
which splits $g$ into high and low non-merged degrees, however we need to be more careful now so as to preserve boundedness.
\footnote{In Section~\ref{sec:merge}, we didn't really care about boundedness as we were able to gain it back later on by changing the function $f$.
However, in this case we cannot afford to change the function $g$ because we need to preserve its non-embedding stability to be small.}
Indeed, the argument below is morally the same, except that we apply a softer type of such split.

Consider the
Markov chain $\mathrm{R}$ on $\Gamma$ that on $y\in \Gamma$, takes $y' = y$ with probability $1-1/d$,
and otherwise we take the connected component $C_i$ in which $y$ lies, and then sample $y'\sim \mu_y$ conditioned on
$y'\in C_i$. Then $R \chi_{i,0} = \chi_{i,0}$ for all $i$, and for $j>0$ we have that $R\chi_{i,j} = \left(1-\frac{1}{d}\right) \chi_{i,j}$.
Write $g = g_1 + g_2$ where $g_1 = \mathrm{R}^{\otimes n} g$ and $g_2 = (\mathrm{I} - \mathrm{R}^{\otimes n})g$.
Then
\begin{equation}\label{eq:merge2}
\Expect{(x,y,z)\sim \mu}{f(x)g(y)h(z)}
=
\underbrace{\Expect{(x,y,z)\sim \mu}{f(x)g_1(y)h(z)}}_{(\rom{1})}
+
\underbrace{\Expect{(x,y,z)\sim \mu}{f(x)g_2(y)h(z)}}_{(\rom{2})},
\end{equation}
and we upper bound each term on the right hand side separately.

\paragraph{Bounding $(\rom{2})$.} For $g_2$, we have by Cauchy-Schwarz that
\begin{align*}
\card{\Expect{(x,y,z)\sim \mu^{\otimes n}}{f(x)g_2(y)h(z)}}^2
=\card{\inner{\overline{fh}}{\mathrm{T} g_2}_{\mu_{x,z}}}^2
\leq \norm{\mathrm{T} g_2}_{2; \mu_{x,z}}^2
=\inner{\mathrm{T} g_2}{\mathrm{T} g_2}_{\mu_{x,z}}
=\inner{\mathrm{S} g_2}{g_2}_{\mu_y}.
\end{align*}
To upper bound the last expression, we note that
\begin{align*}
\inner{\mathrm{S} g_2}{g_2}_{\mu_y}
&=\sum\limits_{\substack{\vec{i},\vec{j}}}\card{\widehat{g}(\vec{i},\vec{j})}^2\left(1-\left(1-\frac{1}{d}\right)^{{\sf non-merged-deg}(\chi_{\vec{i},\vec{j}})}\right)\prod\limits_{k=1}^{n} \lambda_{i,j}\\
&\leq \max_{r}\left(1-\left(1-\frac{1}{d}\right)^r\right)(1-\Omega_{\alpha,m}(1))^{r}.
\end{align*}
For $r > W \log(1/\delta)$, the second term is at most $\delta$ (for appropriately chosen $W$), and
for $r \leq W\log(1/\delta)$ the first term is at most $r/d\leq \delta^{1/4}$ for sufficiently small $\delta_0$,
hence $\card{(\rom{2})}\leq \delta^{1/4}$.

\paragraph{Bounding $(\rom{1})$.}
Here, we are going to use random restrictions so that almost all of the mass of $g_1$ will collapse to monomials of non-merged degree
$0$, at which point we could truncate off the part of non-merged degree exceeding $0$ and get a function that does not distinguish between
the distribution $\mu$ and the distribution $\mu'$. More precisely, let $s = \delta^{1/3}/d$ and choose $J\subseteq [n]$ by including
each element in it with probability $s$, sample $(\tilde{x},\tilde{y},\tilde{z})\sim \mu^{J}$ and define
\[
\tilde{f} = (f)_{\overline{J}\rightarrow \tilde{x}},
\qquad
\tilde{g_1} = (g_1)_{\overline{J}\rightarrow \tilde{y}},
\qquad
\tilde{h} = (h)_{\overline{J}\rightarrow \tilde{z}}.
\]
Let $M\in\mathbb{N}$ and $\eta>0$ be from Theorem~\ref{thm:nonembed_deg_must_be_small} for $\mu'$.
We will prove the statement for $\eta' = \eta/20$ and $M' = M+3$. Define the events:
\begin{enumerate}
  \item $E_1$: ${\sf NEStab}_{1-K\delta/s}(\tilde{g_1};\mu_y)\leq \delta^{2/3}$, where $K = K(m,\alpha)>0$.
  \item $E_2$: The mass of $\tilde{g_1}$ on monomials of non-merged degree more than $0$ exceeds $\delta^{1/6}$.
\end{enumerate}
We show that the probability of $E = E_1\cap \overline{E_2}$ is at least $1-2\delta^{1/6}\log(1/\delta)$; to do that, we use the union bound and
bound the probability of $\overline{E_1}$ and of $E_2$.
\paragraph{The event $E_1$.}
For the event $E_1$, let $V_{t}$ be the space of functions spanned by monomial of non-embedding degree
exactly $t$. We claim that $V_t$ is an invariant space of $\mathrm{R}^{\otimes n}$. Indeed, to see that
it suffices to show that if $\chi\colon \Gamma\to\mathbb{C}$ is orthogonal to all embedding functions
(that is, univariate functions in ${\sf Embed}_{\gamma}(\mu)$), then $\mathrm{R} \chi$ is orthogonal to all embedding
functions. Indeed, if $\chi'$ is an embedding funciton then
\[
\inner{\mathrm{R} \chi}{\chi'} = \inner{\chi}{\mathrm{R}^{*} \chi'} = \inner{\chi}{\mathrm{R}\chi'} = \inner{\chi}{\chi'} = 0,
\]
where we used the fact that $\mathrm{R}$ is self adjoint as it is an averaging operator corresponding to a reversible Markov chain.
Thus, writing $g = \sum\limits_{t} g_{=t}$ where $g_{=t}$ is in $V_{t}$, we get that
\begin{align*}
{\sf NEStab}_{1-\delta}(g_1)
=
{\sf NEStab}_{1-\delta}(\sum\limits_{t} \mathrm{R}^{\otimes n} g_{=t})
&=
\sum\limits_{t,t'}\inner{\mathrm{R}^{\otimes n} g_{=t'}}{\mathrm{T}_{\text{non-embed}, 1-\delta}\mathrm{R}^{\otimes n} g_{=t}}\\
&=
\sum\limits_{t,t'}\inner{\mathrm{R}^{\otimes n} g_{=t'}}{(1-\delta)^t\mathrm{R}^{\otimes n} g_{=t}},
\end{align*}
where we used the fact that $\mathrm{R}^{\otimes n} g_{=t}\in V_{t}$ and the fact that $V_{t}$ is an eigenspace of $\mathrm{T}_{\text{non-embed}, 1-\delta}$
of eigenvalue $(1-\delta)^t$. The inner product is $0$ for all $t\neq t'$, hence we get that this is equal to
\[
\sum\limits_{t}(1-\delta)^{t}\norm{\mathrm{R}^{\otimes n} g_{=t}}_2^2
\leq
\sum\limits_{t}(1-\delta)^{t}\norm{g_{=t}}_2^2
={\sf NEStab}_{1-\delta}(g)
\leq \delta,
\]
so ${\sf NEStab}_{1-\delta}(g_1)\leq \delta$. Thus, by Lemma~\ref{lem:op_comparison_lemma} it holds that
$\Expect{}{{\sf NEStab}_{1-K\delta/s}(\tilde{g_1};\mu_y)}\leq {\sf NEStab}_{1-\delta}(g_1;\mu_y)\leq \delta$ for suitably chosen $K$,
hence by Markov's inequality $\Prob{}{E_2}\geq 1-\delta^{1/3}$.

\paragraph{The event $E_2$.}
The expected weight of $\tilde{g_1}$ on monomials of non-merged degree more than $0$ is at most
\begin{align*}
&\Expect{J,\tilde{y}}{\sum\limits_{\vec{i}\in [\ell]^J,\vec{j}}\widehat{\tilde{g_1}}(\vec{i},\vec{j})^2 1_{\chi_{\vec{i}, \vec{j}}\text{ has non-merge degree at least $1$}}}\\
&\qquad=\Expect{J}{\sum\limits_{\vec{i}\in [\ell]^n,\vec{j}}\widehat{g_1}(\vec{i},\vec{j})^2 1_{\sett{k\in [n]}{j_k>1}\cap J \neq\emptyset}}\\
&\qquad=\sum\limits_{\vec{i}\in [\ell]^n,\vec{j}}\widehat{g}(\vec{i},\vec{j})^2\left(1-\frac{1}{d}\right)^{\card{\sett{k\in [n]}{j_k>1}}}\Expect{J}{1_{\sett{k\in [n]}{j_k>1}\cap J \neq\emptyset}}.
\end{align*}
Denoting $S = \card{\sett{k\in [n]}{j_k>1}}$, we have that the expectation in consideration is equal to $1-\left(1-s\right)^S$ hence
the last sum is at most
\[
\max_{S}\left(1-\frac{1}{d}\right)^{S}\left(1-\left(1-s\right)^S\right).
\]
For $S\geq d\log(1/\delta)$, the first term is at most $\delta$ hence it is at most $\delta$, and for
$S<d\log(1/\delta)$ the second term is at most $sS\leq \log(1/\delta)\delta^{1/3}$. Overall we get that
the expected weight of $\tilde{g_1}$ on monomials of non-merged degree more than $0$ is at most $\log(1/\delta)\delta^{1/3}$.
Hence by Markov's inequality $\Prob{}{E_2}\leq \delta^{1/6}\log(1/\delta)$.

\skipi
We write
\[
(\rom{1})=
\underbrace{\Expect{J,\tilde{x},\tilde{y},\tilde{z}}{1_{\overline{E}} \phi_{\mu}(\tilde{f},\tilde{g_1},\tilde{h})}}_{(\rom{3})}
+
\underbrace{\Expect{J,\tilde{x},\tilde{y},\tilde{z}}{1_{E} \phi_{\mu}(\tilde{f},\tilde{g_1},\tilde{h})}}_{(\rom{4})}
\]
where $\phi_{\mu}(\tilde{f},\tilde{g_1},\tilde{h}) = \Expect{(x,y,z)\sim \mu^{J}}{\tilde{f}(x)\tilde{g_1}(y)\tilde{h}(z)}$.

\paragraph{Bounding $(\rom{3})$.}
For $(\rom{3})$ we have that $\card{(\rom{3})}\leq\Prob{}{\overline{E}}$
which is at most $\delta^{\eta'}$.

\paragraph{Bounding $(\rom{4})$.}
For $(\rom{4})$, let $\tilde{g_1}'\colon \Gamma^{J}\to\mathbb{C}$ be defined by
\[
\tilde{g_1}'(y) = \Expect{y'\sim y}{\tilde{g_1}(y')},
\]
where $y'$ is distributed as: for each coordinate $i$, sample $y_i'$ according to $\mu_y$ conditioned on it being in the connected component of $y_i$.
In words, we average over the connected components. Then the value of $\tilde{g_1}'$ is constant on all connected components, and we also have that
$\norm{\tilde{g_1}-\tilde{g_1}'}_2^2$ is the mass of $\tilde{g_1}$ on monomials of non-merge degree at least $1$, hence it is at most $\delta^{1/6}$.
It follows that
\[
{\sf NEStab}_{1-K\delta/s}(\tilde{g_1}';\mu_y)\leq {\sf NEStab}_{1-K\delta/s}(\tilde{g_1};\mu_y) + O(\norm{\tilde{g_1}-\tilde{g_1}'}_2)
\leq O(\delta^{1/12})
\leq \delta^{1/20},
\]
hence by Claim~\ref{claim:increase_noise_decrease_stab} we have ${\sf NEStab}_{1-\delta^{1/20}}(\tilde{g_1}';\mu_y)\leq \delta^{1/20}$.
We get that
\begin{align*}
\card{(\rom{4})}\leq
\Expect{J,\tilde{x},\tilde{y},\tilde{z}}{1_{E} \card{\phi_{\mu}(\tilde{f},\tilde{g_1},\tilde{h})}}
&\leq
\Expect{J,\tilde{x},\tilde{y},\tilde{z}}{1_{E} \left(\card{\phi_{\mu}(\tilde{f},\tilde{g_1}',\tilde{h})} + O(\norm{\tilde{g_1}-\tilde{g_1}'}_2)\right)}\\
&=
\Expect{J,\tilde{x},\tilde{y},\tilde{z}}{1_{E} \left(\card{\phi_{\mu'}(\tilde{f},\tilde{g_1}',\tilde{h})} + O(\norm{\tilde{g_1}-\tilde{g_1}'}_2)\right)},
\end{align*}
where in the last transition we used the fact that $\tilde{g_1}'$ is constant on connected components hence
$\phi_{\mu}(\tilde{f},\tilde{g_1}',\tilde{h}) = \phi_{\mu'}(\tilde{f},\tilde{g_1}',\tilde{h})$. As
$\card{\phi_{\mu'}(\tilde{f},\tilde{g_1}',\tilde{h})}\leq M\delta^{\eta/20}$ and $\norm{\tilde{g_1}-\tilde{g_1}'}_2\leq \delta^{1/12}$, we get
that $\card{(\rom{4})}\leq (M+1)\delta^{\eta/20}$.

\paragraph{Finishing the proof.}
Together, we get that $\card{(\rom{1})}\leq (M+2)\delta^{\eta'}$ ; combining with the bound on $(\rom{2})$ and plugging into~\eqref{eq:merge2}
finishes the proof.

\section{Theorem~\ref{thm:nonembed_homogenous'} Implies Theorem~\ref{thm:nonembed_deg_must_be_small_rephrase_maximal_relaxed}}\label{sec:truncate}
In this section we show that Theorem~\ref{thm:nonembed_homogenous'} Implies Theorem~\ref{thm:nonembed_deg_must_be_small_rephrase_maximal_relaxed}.
The argument here is analogous to the argument in~\cite[Section A]{BKMcsp2}, with one important difference. Therein, we used the fact that
$\mu$ has some sub-distribution that has no Abelian embedding, and this was used in the ``soft truncation'' steps
(to argue that terms with mis-match of degrees have negligible contribution). Here, we no longer have this property,
and instead we appeal to the maximality property of the same sub-distribution of $\mu$ (which we have ensured in our case). This is ultimately
the reason we have to introduce the additional complication that comes to accommodate the notion of maximality.

The argument proceeds by first splitting $g$ into soft non-embedding homogenous parts, which are functions that have most of their $\ell_2$
mass of characters of roughly the same non-embedding degree and are still bounded. Then, we softly truncate the non-embedding degree of $f$
to be of at most roughly of the same order (while keeping boundedness), and truncate the effective non-embedding to be roughly at least
that same order (while keeping boundedness). Finally, once all of this is done we will be able to apply harsh truncations on the resulting
replacements of $f$ and $g$ function and consider the homogenous parts, then do the same for $h$, and then appeal to
Theorem~\ref{thm:nonembed_homogenous'} to upper bound each one of these homogenous terms.

\skipi
Let $f,g,h$ be as in Theorem~\ref{thm:nonembed_deg_must_be_small_rephrase_maximal_relaxed}, and let $\rho_{j} = 1-2^{-j}\delta$. We take $\xi = \xi(m,\alpha)>0$
to be small enough.

\subsection{Softly Splitting $g$ into Softly (non-embedding) Homogenous Parts}
By the triangle inequality
\begin{align}
\card{\Expect{(x,y,z)\sim \mu^{\otimes n}}{f(x)g(y)h(z)}}
&\leq
\card{\Expect{(x,y,z)\sim \mu^{\otimes n}}{f(x)\mathrm{T}_{\text{non-embed}, \rho_0} g(y)h(z)}}\notag\\\label{eq:trunate1}
&+
\sum\limits_{j=0}^{\infty}
\card{\Expect{(x,y,z)\sim \mu^{\otimes n}}{f(x)(\mathrm{T}_{\text{non-embed}, \rho_{j+1}}-\mathrm{T}_{\text{non-embed}, \rho_j}) g(y)h(z)}}.
\end{align}
For the first term, as $f,h$ are $1$-bounded we have that it is at most
\[
\Expect{(x,y,z)\sim \mu^{\otimes n}}{\card{\mathrm{T}_{\text{non-embed}, \rho_0} g(y)}}
\leq \sqrt{\Expect{(x,y,z)\sim \mu^{\otimes n}}{\card{\mathrm{T}_{\text{non-embed}, \rho_0} g(y)}^2}}
=\sqrt{{\sf NEStab}_{\rho_0^2}(g)}
\leq
\sqrt{{\sf NEStab}_{\rho_0}(g)},
\]
which is at most $\sqrt{\delta}$ by assumption; in the last inequality we used Claim~\ref{claim:increase_noise_decrease_stab}.
For the second term on the right hand side of~\eqref{eq:trunate1} we show:
\begin{claim}\label{claim:soft_truncate_main}
There are $M$, $\delta_0>0$ and $\eta>0$
depending only on $m$ and $\alpha$ such that
for all $0<\delta\leq \delta_0$
\[
\card{\Expect{(x,y,z)\sim \mu^{\otimes n}}{f(x)(\mathrm{T}_{\text{non-embed}, \rho_{j+1}}-\mathrm{T}_{\text{non-embed}, \rho_j}) g(y)h(z)}}
\leq M(1-\rho_j)^{\eta}.
\]
\end{claim}
Given Claim~\ref{claim:soft_truncate_main}, the proof is quickly concluded. Summing up over $j$ yields that the second term above is at most
\[
M\sum\limits_{j=0}^{\infty} 2^{-j\eta}\delta^{\eta}
\leq M\frac{1}{1-2^{-\eta}} \delta^{\eta},
\]
hence $\card{\Expect{(x,y,z)\sim \mu^{\otimes n}}{f(x)g(y)h(z)}}\leq M'\delta^{\eta'}$ as desired.

\skipi
The rest of this section is devoted to the proof of Claim~\ref{claim:soft_truncate_main}, and we fix some $j\geq 0$.

\subsection{Softly Truncating the (non-embedding) Degree of $f$ from Above}
For the simplicity of notation
we write $\rho = \rho_j$ and $\mathrm{T}' = \mathrm{T}_{\text{non-embed}, \rho_{j+1}}-\mathrm{T}_{\text{non-embed}, \rho_j}$. Thus,
\begin{align}\label{eq:trunate5}
  &\card{\Expect{(x,y,z)\sim \mu^{\otimes n}}{f(x)\mathrm{T}' g(y)h(z)}}\leq \notag\\
  &
  \underbrace{\card{\Expect{(x,y,z)\sim \mu^{\otimes n}}{(I-\mathrm{T}_{\text{non-embed}, 1-\rho^{1+\xi}})f(x)\mathrm{T}' g(y)h(z)}}}_{(\rom{1})}
  +
  \underbrace{\card{\Expect{(x,y,z)\sim \mu^{\otimes n}}{\mathrm{T}_{\text{non-embed}, 1-\rho^{1+\xi}}f(x)\mathrm{T}' g(y)h(z)}}}_{(\rom{2})}.
\end{align}

\begin{claim}\label{claim:soft_trunc1}
  There are $M$, $\delta_0>0$ and $\eta>0$ depending only on $m,\alpha$ and $\xi$ such that
  for $0<\delta\leq \delta_0$ it holds that $(\rom{1})\leq M\rho^{\eta}$.
\end{claim}
\begin{proof}
We may upper bound $\card{(\rom{1})}$ by
\begin{align*}
&\card{\Expect{(x,y,z)\sim \mu^{\otimes n}}{(I-\mathrm{T}_{\text{non-embed}, 1-\rho^{1+\xi}})f(x) \mathrm{T}_{\text{non-embed}, 1-\rho} g(y)h(z)}}\\
&+
\card{\Expect{(x,y,z)\sim \mu^{\otimes n}}{(I-\mathrm{T}_{\text{non-embed}, 1-\rho^{1+\xi}})f(x) \mathrm{T}_{\text{non-embed}, 1-\rho/2} g(y)h(z)}}.
\end{align*}
Each one of these terms is upper bounded in the same way, and we demonstrate it on the first. We re-interpret it as
\[
\card{\Expect{(x,y,z)\sim \mu'^{\otimes n}}{(I-\mathrm{T}_{\text{non-embed}, 1-\rho^{1+\xi}})f(x) g(y)h(z)}},
\]
where the distribution $\mu'$ is defined by first sampling $(x',y',z')\sim \mu$, then $y\sim \mathrm{T}_{\text{non-embed}, 1-\rho} y'$ and
outputting $(x',y,z')$. A somewhat annoying feature of $\mu'$ is that there are atoms whose probability is of the order of $\rho$, and to
circumvent it we write $\mu' = \frac{\rho}{2}\mu'' + \left(1-\frac{\rho}{2}\right)\mu'''$ where $\mu''$ and $\mu'''$ are distributions as follows
:
(1) Recalling property 6b from Theorem~\ref{thm:nonembed_deg_must_be_small_rephrase_maximal_relaxed}, we take $\Sigma'$, $\Gamma'$ and $\Phi'$ from there;
the support of $\mu''$ is $\sett{(x,y,z)\in {\sf supp}(\mu)}{x\in\Sigma', y\in\Gamma', z\in\Phi'}$, (2) each atom in $\mu''$ has probability at least $\alpha'(\alpha) > 0$.
Write $f' = (I-\mathrm{T}_{\text{non-embed}, 1-\rho^{1+\xi}})f$.

Pick $J\subseteq_{\rho/2} [n]$, let $(\tilde{x},\tilde{y},\tilde{z})\sim \mu'''^{\overline{J}}$
and define
\[
f'' = (f')_{\overline{J}\rightarrow \tilde{x}},
\qquad
g' = (g)_{\overline{J}\rightarrow \tilde{y}},
\qquad
h' = (h)_{\overline{J}\rightarrow \tilde{z}}.
\]
Then sampling $(x',y',z')\sim \mu''^{J}$ we have
\[
\Expect{(x,y,z)\sim \mu'^{\otimes n}}{f'(x) g(y)h(z)}
=
\Expect{J, (\tilde{x},\tilde{y},\tilde{z})\sim \mu''^{\overline{J}}}
{\Expect{(x',y',z')\sim \mu''^{J}}{f''(x')g'(y)h'(z)}},
\]
so
\[
\card{\Expect{(x,y,z)\sim \mu'^{\otimes n}}{f'(x) g(y)h(z)}}
\leq
\Expect{J, (\tilde{x},\tilde{y},\tilde{z})\sim \mu''^{\overline{J}}}
{\card{\Expect{(x',y',z')\sim \mu''^{J}}{f''(x')g'(y)h'(z)}}}.
\]
Define the event $E$ that ${\sf NEStab}_{1-\rho^{\xi/4}}(f'')\leq \rho^{\xi/8}$, and write
\begin{align*}
&\Expect{J, (\tilde{x},\tilde{y},\tilde{z})\sim \mu''^{\overline{J}}}
{\card{\Expect{(x',y',z')\sim \mu''^{J}}{f''(x')g'(y)h'(z)}}}
=\\
&\underbrace{\Expect{J, (\tilde{x},\tilde{y},\tilde{z})\sim \mu''^{\overline{J}}}
{1_E\card{\Expect{(x',y',z')\sim \mu''^{J}}{f''(x')g'(y)h'(z)}}}}_{(\rom{3})}
+
\underbrace{\Expect{J, (\tilde{x},\tilde{y},\tilde{z})\sim \mu''^{\overline{J}}}
{1_{\overline{E}}\card{\Expect{(x',y',z')\sim \mu''^{J}}{f''(x')g'(y)h'(z)}}}}_{(\rom{4})}.
\end{align*}
\paragraph{Upper bounding $(\rom{4})$.}
We will show that the probability of
$\overline{E}$ is close to $0$. First, we estimate the non-embedding stability of $f$:
\[
{\sf NEStab}_{1-\rho^{1+\xi/2}}(f')
\leq \norm{\mathrm{T}_{\text{non-embed}, 1-\rho^{1+\xi/2}}(I-\mathrm{T}_{\text{non-embed}, 1-\rho^{1+\xi}})f}_2,
\]
and as the eigenvalues of $\mathrm{T}_{\text{non-embed}, 1-\rho^{1+\xi/2}}(I-\mathrm{T}_{\text{non-embed}, 1-\rho^{1+\xi}})$
are $(1-\rho^{1+\xi/2})^d(1-(1-\rho^{1+\xi})^d)$, it follows that the last $2$-norm is at most the maximum of that over $d\in\mathbb{N}$.
If $d\geq \rho^{-(1+0.75\xi)}$, then $(1-\rho^{1+\xi/2})^d\leq e^{-\Omega(\rho^{-\xi/4})}\leq \rho^{\xi/4}$ for small enough $\delta_0$.
If $d\leq \rho^{-(1+0.75\xi)}$, then by Bernouli's inequality $1-(1-\rho^{1+\xi})^d\leq \rho^{1+\xi}d\leq \rho^{\xi/4}$. Thus we
get that ${\sf NEStab}_{1-\rho^{1+\xi/2}}(f')\leq \rho^{\xi/4}$.

Appealing to Claim~\ref{claim:rr_nestab} we get that
\[
\Expect{J,\tilde{x}}{{\sf NEStab}_{1-\rho^{\xi/4}}(f'')}\leq {\sf NEStab}_{1-c\rho^{1+\xi/4}}(f')
\]
where $c>0$ is a constant depending only on $m$ and $\alpha$, and as $c\rho^{1+\xi/4}\geq \rho^{1+\xi/2}$ for sufficiently small
$\delta_0$ we get from Claim~\ref{claim:increase_noise_decrease_stab} that
${\sf NEStab}_{1-c\rho^{1+\xi/4}}(f')\leq {\sf NEStab}_{1-\rho^{1+\xi/2}}(f')\leq \rho^{\xi/4}$. Hence, by Markov's inequality
$\Prob{}{\overline{E}}\leq \rho^{\xi/8}$, and by $1$-boundedness of the functions
\[
\card{(\rom{4})}\leq \Prob{}{\overline{E}}\leq \rho^{\xi/8}.
\]

\paragraph{Upper bounding $(\rom{3})$.}
Fix $J$ and $(\tilde{x},\tilde{y},\tilde{z})$ for which $E$ holds.
Here, we are going to appeal to maximality conditions. Recall property 6b in Theorem~\ref{thm:nonembed_deg_must_be_small_rephrase_maximal_relaxed},
the condition there roughly said that we can embed in a subset of
\[
\sett{(x,y,z)\in {\sf supp}(\mu)}{x\in \Sigma', y\in \Gamma', z\in \Phi'}
\]
a maximal
distribution $\nu$, and the point is that in our $\mu''$ we will be able to embed a distribution $\nu'$ on $\Sigma'\times \Gamma\times \Phi'$
whose support strictly contains $\nu$ (and in which the probability of each atom is at least $\alpha''(\alpha,m)>0$).

Indeed, take the mapping $a$ and the distribution $\tilde{\nu}$ as in property 6b in Theorem~\ref{thm:nonembed_deg_must_be_small_rephrase_maximal_relaxed},
which we denote by $\nu$ for notational convenience. Let $\nu'$ be the distribution of $(x'',y'',z'')$ sampled as follows: first sample $(x',y',z')\sim \mu''$, sample $x''\sim a^{-1}(x')$
uniformly and take $y'' = y'$ and $z'' = z'$. Also, define $f'''(x'') = f''(a(x''))$; then
\begin{equation}\label{eq:trunate2}
(\rom{3}) = \Expect{(x'',y'',z'')\sim \nu'^{J}}{f'''(x'')g'(y'')h'(z'')}.
\end{equation}
Thus, the right hand side is an expectation as in Theorem~\ref{thm:nonembed_deg_must_be_small} over a distribution $\nu'$. It is clear
that $\nu'$ satisfies all of the conditions of Theorem~\ref{thm:nonembed_deg_must_be_small}. Also, if $(u,v,w)\in {\sf supp}(\nu)$,
then $(a(u),v,w)\in {\sf supp}(\mu'')$, and so by definition of $\nu'$ we have that $(u,v,w)\in {\sf supp}(\nu')$, hence
${\sf supp}(\nu)\subseteq {\sf supp}(\nu')$. Next, we argue that this is a strict containment. Indeed, by property 6b(vi)
there are distinct $v,v'\in \Gamma$ such that $\gamma(v) = \gamma(v')$
(here $\gamma$ is the master embedding of $y$).
We take $u,u'\in \Sigma''$ and $w,w'\in \Phi'$ such that $(u,v,w), (u',v',w')\in {\sf supp}(\nu)$, and note that as in $\nu$ the value of any two
coordinates implies the last, it follows that $(u,v',w)\not\in {\sf supp}(\nu)$. However, by definition of $\mu''$ we have that
as $(a(u), v, w)$ is in the support of $\mu$, and as $\gamma(v) = \gamma(v')$ it follows that $(a(u),v',w)$ is in the support of $\mu$,
hence $(a(u), v', w)$ is in the support of $\mu''$ and so $(u,v',w)$ is in the support of $\nu'$.

It follows that the expectation on the right hand side of~\eqref{eq:trunate2} is an expectation with respect to a distribution $\nu'$ whose support
strictly contains the support of a maximal distribution. Also, by Claim~\ref{claim:increase_noise_decrease_stab}
\[
{\sf NEStab}_{1-\rho^{\xi/8}, \nu'}(f'''; \nu'_x)
\leq
{\sf NEStab}_{1-\rho^{\xi/4}, \nu'}(f'''; \nu'_x)
=
{\sf NEStab}_{1-\rho^{\xi/4}, \mu''}(f''; \mu''_x)
\leq \rho^{\xi/8},
\]
hence by maximality we get that $\card{(\rom{3})}\leq M\rho^{\eta\xi/8}$ for some $M$ and $\eta>0$
depending only on $m$ and $\alpha$.
\end{proof}

\subsection{Softly Truncating the Effective (non-embedding) Degree of $f$ from Below}
Looking at $(\rom{2})$, we write
$f' = \mathrm{T}_{\text{non-embed}, 1-\rho^{1+\xi}}f$
and
$g' = \mathrm{T}' g$. By the triangle inequality we get
\begin{align}\label{eq:trunate6}
(\rom{2})
&=\underbrace{\card{\Expect{(x,y,z)\sim \mu^{\otimes n}}{\mathrm{E}_{\text{non-embed}, 1-\rho^{1-\xi}}f'(x) g'(y)h(z)}}}_{(\rom{5})}\notag\\
&+
\underbrace{\card{\Expect{(x,y,z)\sim \mu^{\otimes n}}{(I-\mathrm{E}_{\text{non-embed}, 1-\rho^{1-\xi}})f'(x)g'(y)h(z)}}}_{(\rom{6})}.
\end{align}
\begin{claim}\label{claim:soft_trunc2}
  There are $M$, $\delta_0>0$ and $\eta>0$ depending only on $m,\alpha$ and $\xi$
  such that for $0<\delta\leq \delta_0$ it holds that $(\rom{5})\leq M\rho^{\eta}$.
\end{claim}
\begin{proof}
  We re-interpret the expectation in $(\rom{5})$ as
  \[
  \Expect{(x',y',z')\sim \mu'^{\otimes n}}{f'(x') g'(y')h(z')},
  \]
  where the distribution $\mu'$ is defined as first sampling $(x,y,z)\sim \mu$, then $x'\sim \mathrm{E}_{\text{non-embed}, 1-\rho^{1-\xi}} x$
  and outputting $(x',y',z')$ where $y' = y$ and $z' = z$. Again, a somewhat annoying feature of $\mu'$ is that there are atoms with probability
  $\rho^{1-\xi}$, and to circumvent it we again use random restrictions.

  More precisely, we write $\mu' = \frac{1}{2}\rho^{1-\xi} \mu'' + \left(1-\frac{1}{2}\rho^{1-\xi}\right)\mu'''$ where $\mu''$ and
  $\mu'''$ are distributions as follows:
  (1) Recalling property 6b from Theorem~\ref{thm:nonembed_deg_must_be_small_rephrase_maximal_relaxed}, we take $\Sigma'$ and $\Phi'$ from there;
  the support of $\mu''$ is $\sett{(x,y,z)\in {\sf supp}(\mu)}{x\in\Sigma', z\in\Phi'}$, (2) each atom in $\mu''$ has probability at least $\alpha'(\alpha) > 0$.

  Pick $J\subseteq_{\rho^{1-\xi}/2} [n]$, let $(\tilde{x},\tilde{y},\tilde{z})\sim \mu'''^{\overline{J}}$
  and define
\[
f'' = (f')_{\overline{J}\rightarrow \tilde{x}},
\qquad
g'' = (g')_{\overline{J}\rightarrow \tilde{y}},
\qquad
h' = (h)_{\overline{J}\rightarrow \tilde{z}}.
\]
Then sampling $(x',y',z')\sim \mu''^{J}$ we have
\[
\Expect{(x,y,z)\sim \mu'^{\otimes n}}{f'(x) g'(y)h(z)}
=
\Expect{J, (\tilde{x},\tilde{y},\tilde{z})\sim \mu''^{\overline{J}}}
{\Expect{(x',y',z')\sim \mu''^{J}}{f''(x')g''(y)h'(z)}},
\]
so
\begin{equation}\label{eq:trunate3}
\card{\Expect{(x,y,z)\sim \mu'^{\otimes n}}{f'(x) g'(y)h(z)}}
\leq
\Expect{J, (\tilde{x},\tilde{y},\tilde{z})\sim \mu''^{\overline{J}}}
{\card{\Expect{(x',y',z')\sim \mu''^{J}}{f''(x')g''(y)h'(z)}}}.
\end{equation}

Our goal is to use the maximality property from 6b in Theorem~\ref{thm:nonembed_deg_must_be_small_rephrase_maximal_relaxed} to bound this. Towards this end, first note
that
\[
{\sf NEStab}_{1-\rho^{1-\xi/2}}(g')
=\inner{g'}{\mathrm{T}_{\text{non-embed}, 1-\rho^{1-\xi/2}}g'}
\leq \norm{\mathrm{T}_{\text{non-embed}, 1-\rho^{1-\xi/2}}g'}_2,
\]
where the last inequality is Cauchy-Schwarz (and using the fact that $g'$ has $2$-norm at most $1$). Recalling the definition of
$g'$,
\[
\norm{\mathrm{T}_{\text{non-embed}, 1-\rho^{1-\xi/2}}g'}_2
=\norm{\mathrm{T}_{\text{non-embed}, 1-\rho^{1-\xi/2}}(\mathrm{T}_{\text{non-embed}, 1-\rho/2}-\mathrm{T}_{\text{non-embed}, 1-\rho})g}_2,
\]
which is at most the largest eigenvalue of $\mathrm{T}_{\text{non-embed}, 1-\rho^{1-\xi/2}}(\mathrm{T}_{\text{non-embed}, 1-\rho/2}-\mathrm{T}_{\text{non-embed}, 1-\rho})$.
The eigenvalues of this operator are $(1-\rho^{1-\xi/2})^r\left((1-\rho/2)^r - (1-\rho)^r\right)$ for $r\in\mathbb{N}$, and we argue they are all at most $\rho^{\xi/4}$.
Indeed, if $r\geq \rho^{1-\xi/4}$ then the first term is at most that (and the second one is at most $1$), and otherwise the first term is at most $1$ and
the second one is at most $(1-(1-r\rho))\leq r\rho = \rho^{\xi/4}$.

Hence we get that ${\sf NEStab}_{1-\rho^{1-\xi/2}}(g')\leq \rho^{\xi/4}$. Let $E$ be the event that
${\sf NEStab}_{1-\rho^{\xi/4}}(g'')\leq \rho^{\xi/8}$; by Claim~\ref{claim:rr_nestab} we get that
\[
\Expect{J,\tilde{y}}{{\sf NEStab}_{1-\rho^{\xi/4}}(g'')}\leq {\sf NEStab}_{1-c\rho^{1-3\xi/4}}(g')
\leq
{\sf NEStab}_{1-\rho^{1-\xi/2}}(g')
\leq \rho^{\xi/4}
\]
(where we used Claim~\ref{claim:increase_noise_decrease_stab} and also that $\rho$ is sufficiently small as $\delta_0$ is sufficiently small),
so by Markov's inequality $\Prob{}{E}\geq 1-\rho^{\xi/8}$. Hence, splitting the right hand side of~\eqref{eq:trunate3} according to whether $E$ holds or not,
we get that it is at most
\[
\underbrace{\Expect{\substack{J\\ (\tilde{x},\tilde{y},\tilde{z})\sim \mu''^{\overline{J}}}}
{1_E\card{\Expect{(x',y',z')\sim \mu''^{J}}{f''(x')g''(y)h'(z)}}}}_{(\rom{7})}
+
\underbrace{
\Expect{\substack{J\\ (\tilde{x},\tilde{y},\tilde{z})\sim \mu''^{\overline{J}}}}
{1_{\overline{E}}\card{\Expect{(x',y',z')\sim \mu''^{J}}{f''(x')g''(y)h'(z)}}}}_{(\rom{8})}.
\]
We have $\card{(\rom{8})}\leq 4\Prob{}{\overline{E}}\leq 4\rho^{\xi/8}$ (as $f''$ and $g''$ are $2$-bounded and $h'$ is $1$-bounded).

For $(\rom{7})$ we use maximality.
Take the mapping $a$ and the distribution $\nu$ as in property 6b in Theorem~\ref{thm:nonembed_deg_must_be_small_rephrase_maximal_relaxed},
and let $\nu'$ be the distribution of $(x'',y'',z'')$ sampled as follows: first sample $(x',y',z')\sim \mu''$, sample $x''\sim a^{-1}(x')$
uniformly and take $y'' = y'$ and $z'' = z'$. Also, define $f'''(x'') = f''(a(x''))$; then
\begin{equation}\label{eq:trunate4}
(\rom{7}) = \Expect{(x'',y'',z'')\sim \nu'^{J}}{f'''(x'')g'(y'')h'(z'')}.
\end{equation}
Thus, the right hand side is an expectation as in Theorem~\ref{thm:nonembed_deg_must_be_small} over a distribution $\nu'$. It is clear
that $\nu'$ satisfies all of the conditions of Theorem~\ref{thm:nonembed_deg_must_be_small}. Also, if $(u,v,w)\in {\sf supp}(\nu)$,
then $(a(u),v,w)\in {\sf supp}(\mu'')$, and so by definition of $\nu'$ we have that $(u,v,w)\in {\sf supp}(\nu')$, hence
${\sf supp}(\nu)\subseteq {\sf supp}(\nu')$. Next, we argue that this is a strict containment.

Take distinct
$u,u'\in\Sigma_{{\sf modest}}$ and $u_{{\sf pre}} \in a^{-1}(u)$, $u_{{\sf pre}}'\in a^{-1}(u')$,
and take $v,v'\in\Gamma$ and $w,w'\in \Phi'$ such that $(u_{{\sf pre}},v,w)$ and $(u_{{\sf pre}}',v',w')$ are in ${\sf supp}(\nu)$.
Then $(u_{{\sf pre}}',v,w)\not\in {\sf supp}(\nu)$ as in $\nu$ it is the case that the value of two coordinates implies the third.
On the other hand, it holds that $(u',v,w)\in {\sf supp}(\mu'')$ and so $(u_{{\sf pre}}',v,w)\in {\sf supp}(\nu')$.

Thus, for every $J$ and $\tilde{x},\tilde{y}$ and $\tilde{z}$ such that $E$ hold we have that $(\rom{7})$ is an expectation
with respect to a distribution satisfying the conditions of Theorem~\ref{thm:nonembed_deg_must_be_small} whose support
strictly contains the support of a maximal distribution, hence $\card{(\rom{7})}\leq M' \rho^{\eta\xi/8}$ for $\eta = \eta(\alpha,m)>0$.

\end{proof}

\subsection{Harsh Truncations}
\begin{claim}\label{claim:soft_trunc3}
  There are $M$ and $\xi_0,\delta_0>0$ depending only on $m$ and $\alpha$ such that if $0 < \xi\leq \xi_0$ and $0<\delta\leq \delta_0$, then
  $\card{(\rom{6})}\leq M 2^{-\frac{1}{\rho^{1-10\xi}}}$.
\end{claim}
\begin{proof}
Set $f'' = (I-\mathrm{E}_{\text{non-embed}, 1-\rho^{1-\xi}})f'$.
Let $g_1'$ be the part of $g'$ of non-embedding degree at most $d_1 = \frac{1}{\rho^{1-\xi}}$,
$g_2'$ be the part of $g'$ of non-embedding degree at least $d_2 = \frac{1}{\rho^{1+\xi}}$
and $g_3'$ be the part of $g'$ of non-embedding degree between $d_1$ and $d_2$.
Then $\norm{g_1'}_2\leq (1-\rho/2)^{d_1} - (1-\rho)^{d_1}\leq d_1\rho\leq \rho^{\xi}$
and $\norm{g_2'}_2\leq (1-\rho/2)^{d_2} - (1-\rho)^{d_2}\leq \rho^{\xi}$, hence
\[
\card{(\rom{6})}\leq
2\rho^{\xi} +
\card{\Expect{(x,y,z)\sim \mu^{\otimes n}}{f''(x)g_3'(y)h(z)}}.
\]
Similarly, let $f_1''$ be the part of $f''$ of effective non-embedding degree at most
$d_3 = \frac{1}{\rho^{1-2\xi}}$, $f_2''$ be the part of $f''$ of non-embedding degree more
than $d_4 = \frac{1}{\rho^{1+2\xi}}$ and $f_3''$ be the part of $f''$ of effective non-embedding
degree at least $d_3$ and non-embedding degree at most $d_4$. We have $\norm{f_1''}_2,\norm{f_2''}_2\leq \rho^{\xi}$
hence
\begin{align*}
\card{\Expect{(x,y,z)\sim \mu^{\otimes n}}{f''(x)g_3'(y)h(z)}}
&\leq
\norm{f''_1}_2\norm{g_3'}_2
+\norm{f''_2}_2\norm{g_3'}_2
+\card{\Expect{(x,y,z)\sim \mu^{\otimes n}}{f_3''(x)g_3'(y)h(z)}}\\
&\leq 2\rho^{\xi}
+
\card{\Expect{(x,y,z)\sim \mu^{\otimes n}}{f_3''(x)g_3'(y)h(z)}},
\end{align*}
hence
\[
\card{(\rom{6})}\leq
4\rho^{\xi}
+
\card{\Expect{(x,y,z)\sim \mu^{\otimes n}}{f_3''(x)g_3'(y)h(z)}}.
\]
We now partition $f_3''$, $g_3'$ and $h$ according to embedding degrees, non-embedding degrees and effective degrees.
Denoting by $\vec{D} = (D_a)_{a\in \hat{H}}$ a sequence representing embedding degrees,
we write $f_{3}'' = \sum\limits_{\vec{D} ,i,i'} f_{3,\vec{D},i,i'}''$ where $f_{3,\vec{D},i,i'}''$ is the part of $f''$ of
$a$-embedding degree equal to $D_a$ for all $a\in\hat{H}$, effective non-embedding of degree exactly $i$ and non-embedding degree exactly $i'$. We also write
$g_{3}' = \sum\limits_{\vec{D}',j} g_{3,J,j}'$ where $g_{3,\vec{D}',j}'$ is the part of $g_{3}'$ of $a$-embedding degree exactly $D_a$ for all $a$ and
non-embedding degree exactly $j$. Finally, we write $h = \sum\limits_{\vec{D}'',k} h_{\vec{D}'',k}$ where $h_{\vec{D}'',k}$ is the part of $h$ of
$a$-embedding degree equal to $D_a''$ for all $a\in \hat{H}$ and non-embedding degree exactly $k$. Then
\[
\Expect{(x,y,z)\sim \mu^{\otimes n}}{f_3''(x)g_3'(y)h(z)}
=
\sum\limits_{\vec{D},\vec{D}',\vec{D}''}
\sum\limits_{i,i'=d_3}^{d_4}
\sum\limits_{j=d_1}^{d_2}
\sum\limits_{k=0}^{n}
\Expect{(x,y,z)\sim \mu^{\otimes n}}{f_{3,\vec{D},i,i'}''(x)g_{3,\vec{D}',j}'(y)h_{\vec{D}'',k}(z)}.
\]
In the rest of the argument, we upper bound the absolute value of the sum on the right hand side.
\paragraph{The contribution is zero unless embedding degrees nearly match.}
We first claim that unless $\card{D_a-D_a'}$ and $\card{D_a-D_a''}$ are each at most $10d_4$ for all $a$, this expectation is $0$.
To do so, we are going to expand each function according to the monomials, multiply this out and consider a term in this resulting expression.

Suppose that there is an $a\in\hat{H}$ such that $\card{D_a-D_a'}\geq 10d_4$, say $D_a\geq D_a'+10d_4$, then there must be at least
$D_a' - (D_a'+d_2) - d_4 \geq 1$ coordinates $j$ such that the term contains
$\chi_j(x_j)\chi_j'(y_j)\chi_j''(z_j)$ where
$\chi_j = a$,
$\chi_j' = a'$ for some $a'\in \hat{H}$ different from $a$, and
$\chi_{j}''$ is an embedding function (may be constant).
By independence it suffices to argue that
\[
\Expect{(x_j,y_j,z_j)\sim \mu}{\chi_j(x_j)\chi_{j}'(y_j)\chi_{j}''(z_j)} = 0,
\]
and this is clear as it is equal to
\[
\Expect{(x_j,y_j,z_j)}{a(\sigma(x_j))a'(\gamma(y_j))\chi_{j}''(z_j)}
= \Expect{(x_j,y_j,z_j)}{\overline{a(\gamma(y))}a'(\gamma(y_j))\overline{a(\phi(z))}\chi_{j}''(z_j)}
=\inner{a'}{a}\inner{\chi_{j}''}{a},
\]
where in the last equality we used independence of $y$ and $z$.
The case $D_a'\geq D_a+10d_4$ is nearly identical as well as the case that $\card{D_a-D_a''}\geq 10d_4$.

Next, note that if $k>i'+j$, then
$\Expect{(x,y,z)\sim \mu^{\otimes n}}{f_{3,I,i,i'}''(x)g_{3,J,j}'(y)h_{K,k}(z)} = 0$; this is
true because looking at $h'(z) = \cExpect{(x',y',z')\sim \mu^{\otimes n}}{z' = z}{f_{3,i,i'}''(x)g_{3,j}'(y)}$,
we get that $h'$ is a function of non-embedding degree at most $i'+j$, and the expectation is
$\inner{h}{\overline{h'}}$.

Denoting by $\mathcal{D}$ the set of tuples of degree sequences $(\vec{D}, \vec{D}', \vec{D}'')$
such that $\card{D_a-D_a'},\card{D_a-D_a''}\leq 10d^4$, we get that
\begin{align*}
&\card{\Expect{(x,y,z)\sim \mu^{\otimes n}}{f_3''(x)g_3'(y)h(z)}}\\
&\qquad\qquad\qquad\leq
\sum\limits_{(\vec{D}, \vec{D}', \vec{D}'') \in \mathcal{D}}
\sum\limits_{i,i'=d_3}^{d_4}
\sum\limits_{j=d_1}^{d_2}
\sum\limits_{k=0}^{d_2+d_4}
\card{\Expect{(x,y,z)\sim \mu^{\otimes n}}{f_{3,\vec{D},i,i'}''(x)g_{3,\vec{D}',j}'(y)h_{\vec{D}'',k}(z)}}\\
&\qquad\qquad\qquad\leq
\sum\limits_{(\vec{D}, \vec{D}', \vec{D}'') \in \mathcal{D}}
\sum\limits_{i,i'=d_3}^{d_4}
\sum\limits_{j=d_1}^{d_2}
\sum\limits_{k=0}^{d_2+d_4}
\beta_{n,i,i'}[\mu]\norm{f_{3,\vec{D},i,i'}''}_2\norm{g_{3,\vec{D}',j}'}_2\norm{h_{\vec{D}'',k}}_2,
\end{align*}
where by abuse of notation we denote by $D, D'$ and $D''$ the sum of entries in $\vec{D}, \vec{D}'$ and $\vec{D}''$ respectively.
By Theorem~\ref{thm:nonembed_homogenous'} we get that provided that $\xi$ is small enough it holds that
$\beta_{n,i,i'}[\mu]\leq (1+c)^{-i^{1-2\xi}}$, and plugging that above yields that
\begin{align*}
&\card{\Expect{(x,y,z)\sim \mu^{\otimes n}}{f_3''(x)g_3'(y)h(z)}}\\
&\qquad\leq
\sum\limits_{(\vec{D}, \vec{D}', \vec{D}'') \in \mathcal{D}}
\sum\limits_{i,i'=d_3}^{d_4}
\sum\limits_{j=d_1}^{d_2}
\sum\limits_{k=0}^{d_2+d_4}
(1+c)^{-i^{1-2\xi}}\norm{f_{3,\vec{D},i,i'}''}_2\norm{g_{3,\vec{D}',j}'}_2\norm{h_{\vec{D}'',k}}_2\\
&\qquad\leq
\sum\limits_{(\vec{D}, \vec{D}', \vec{D}'') \in \mathcal{D}}
\sum\limits_{i,i'=d_3}^{d_4}
\sum\limits_{j=d_1}^{d_2}
\sum\limits_{k=0}^{d_2+d_4}
(1+c)^{-i^{1-2\xi}}\norm{f_{3,\vec{D}}''}_2\norm{g_{3,\vec{D}'}'}_2\norm{h_{\vec{D}''}}_2,
\end{align*}
where $f_{3,\vec{D}}''$ is the part of $f_3''$ of embedding degrees $\vec{D}$, and similarly
$g_{3,\vec{D}'}'$ and $h_{\vec{D}''}$; here used Parseval.
We get that the above is at most
\[
\card{\Expect{(x,y,z)\sim \mu^{\otimes n}}{f_3''(x)g_3'(y)h(z)}}
\leq
\sum\limits_{(\vec{D}, \vec{D}', \vec{D}'') \in \mathcal{D}}
\norm{f_{3,\vec{D}}''}_2\norm{g_{3,\vec{D}'}'}_2\norm{h_{\vec{D}''}}_2
\sum\limits_{i,i'=d_3}^{d_4}
\sum\limits_{j=d_1}^{d_2}
\sum\limits_{k=0}^{d_2+d_4}
(1+c)^{-i^{1-2\xi}};
\]
the inner sum is at most $M(1+c)^{-d_3^{1-\xi}}$ for $M$ depending only on $c$, and
\begin{align*}
\sum\limits_{(\vec{D}, \vec{D}', \vec{D}'') \in \mathcal{D}}
\norm{f_{3,\vec{D}}}_2\norm{g_{3,\vec{D}'}'}_2\norm{h_{\vec{D}''}}_2
&\leq
\sum\limits_{\vec{D}} \norm{f_{3,\vec{D}}''}_2
\sum\limits_{\vec{D}', \vec{D}'': (\vec{D}, \vec{D}', \vec{D}'') \in \mathcal{D}}\norm{g_{3,\vec{D}'}'}_2\\
&\leq
\sqrt{\sum\limits_{\vec{D}} \norm{f_{3,\vec{D}}}_2^2}
\sqrt{\sum\limits_{\vec{D}}\left(\sum\limits_{\vec{D}', \vec{D}'': (\vec{D}, \vec{D}', \vec{D}'') \in \mathcal{D}}\norm{g_{3,\vec{D}'}'}_2\right)^2}\\
&\leq
(20d_4)^{\card{H}}\sqrt{\sum\limits_{\vec{D}} \norm{f_{3,\vec{D}}''}_2^2}\sqrt{\sum\limits_{(\vec{D}, \vec{D}', \vec{D}'') \in \mathcal{D}}\norm{{g_{3,\vec{D}'}'}^{=J}}_2^2}\\
&\leq
(20d_4)^{2\card{H}}\sqrt{\sum\limits_{\vec{D}} \norm{f_{3,\vec{D}}''}_2^2}\sqrt{\sum\limits_{\vec{D}}\norm{g_{3,\vec{D}}'}_2^2}\\
&=(20d_4)^{2\card{H}} \norm{f_3''}_2\norm{g_3'}_2,
\end{align*}
which is at most $(20d_4)^{2\card{H}}$. We used Cauchy-Schwarz multiple times, the fact that once we fix $\vec{D}$ there are at most $(20d_4)^{2\card{H}}$
degree vectors $\vec{D}', \vec{D}''$ such that $(\vec{D},\vec{D}',\vec{D}'')\in\mathcal{D}$; in the end we also used Parseval.
Plugging above, we conclude that
\[
\card{
\Expect{(x,y,z)\sim \mu^{\otimes n}}{f_3''(x)g_3'(y)h(z)}}
\leq (20d_4)^{2m}\cdot M'' (1+c)^{-d_3^{1-\xi}}
\leq M'''(1+c')^{-d_3^{1-\xi}}
\]
where $M'''$, $\xi$ and $c$ depend only on $m$ and $\alpha$, and we are done.
\end{proof}

\subsection{Proof of Claim~\ref{claim:soft_truncate_main}}
We are now ready to prove Claim~\ref{claim:soft_truncate_main}.
Fixing $m$ and $\alpha$, we take $M_1,\xi_1$ and $\delta_1$ from Claim~\ref{claim:soft_trunc3}.
We take $\xi = \min(\xi_1,1/100)$, then $M_2, \delta_2$ and $\eta_1$ from Claim~\ref{claim:soft_trunc1}
and $M_3,\delta_3, \eta_2$ from Claim~\ref{claim:soft_trunc2}. Finally we pick $M = M_1+M_2+M_3$,
$\eta = \min(\eta_1,\eta_2, 1/100)$ and $\delta_0 = \min(\delta_1,\delta_2,\delta_3,1/100)$.

Using~\eqref{eq:trunate5}, we upper bound the absolute value of left hand side of Claim~\ref{claim:soft_truncate_main} by
$(\rom{1}) + (\rom{2})$, and by Claim~\ref{claim:soft_trunc2} we have that $(\rom{1})\leq M_3\rho^{\eta}$. Using~\eqref{eq:trunate6}
we bound $(\rom{2})\leq (\rom{5}) + (\rom{6})$ and by Claim~\ref{claim:soft_trunc2} we have
$(\rom{5})\leq M_2\rho^{\eta}$. Finally, by Claim~\ref{claim:soft_trunc3} we have that
$(\rom{6})\leq M_1\rho^{\eta}$ and the proof is complete.\qed

\section{SVD Decompositions Proofs}\label{sec:svd_apx}
	In this section, we prove the claims establishing the SVD decompositions. We remind the reader that throughout,
$I,J$ is a partition of $[n]$ where $\card{I} = n-1$ and $\card{J} = 1$.

\subsection{Homogeneity and Singular Value Decompositions}
The proof of our singular value decomposition will require a basic connection between such decompositions and
the various notions of homogeneity. This will be used multiple times, and therefore we abstract below.
\begin{definition}
    Let $B_1,\ldots, B_{\ell}\subseteq \set{g\colon \Gamma\to\mathbb{C}}$ be orthonormal sets and suppose that
    $B_1\cup\ldots\cup B_{\ell}$ is an orthonormal basis. Then we may write any $g\colon \Gamma^n\to\mathbb{C}$
    as
    \[
    g = \sum\limits_{v_{i_1},\ldots,v_{i_n}\in B_1\cup\ldots\cup B_{\ell}} \alpha_{i_1,\ldots,i_n} v_{i_1}\otimes v_{i_2}\otimes \ldots \otimes v_{i_n}
    \]
    Then the degree of $v_{i_1}\otimes v_{i_2}\otimes \ldots \otimes v_{i_n}$ with respect to $B_j$ is defined to be the number of
    $i_k$'s such that $v_{i_k}\in B_j$. Furthermore, we say that $g$ is degree $d$ homogenous with respect to $B_j$ if for each monomial
    $v_{i_1}\otimes v_{i_2}\otimes \ldots \otimes v_{i_n}$ in it with non-zero coefficient it holds that the degree of that monomial
    with respect to $B_j$ is $d$.
\end{definition}
\begin{claim}\label{claim:svd_invariant_spaces}
  Let $B_1,\ldots, B_{\ell}\subseteq \set{g\colon \Gamma^J\to\mathbb{C}}$ be orthonormal sets and suppose that
  $B_1\cup\ldots\cup B_{\ell}$ is an orthonormal basis. Given a function $g\colon \Gamma^{I}\times \Gamma^{J}\to\mathbb{C}$
  define the matrix $M\in \mathbb{C}^{\Gamma^{J}\times \Gamma^{I}}$ as $M(a,b) = g(a,b)$, and consider the matrix
  $MM^{*}\in \mathbb{C}^{\Gamma^{J}\times \Gamma^{J}}$.

  If $g$ is degree $d$ homogenous with respect to $B_j$, then ${\sf Span}(B_j)$ is an invariant space of $MM^{*}$.
\end{claim}
\begin{proof}
        For notational convenience, we assume that $J = \{n\}$ and $m = \card{\Gamma}$. Let $v\in B_{j}$, and write
		\begin{align*}
			(MM^{*} v)_{a}
			=\sum\limits_{b} (M M^{*})_{a,b}v(b)
			=\sum\limits_{b} \sum\limits_{y_I\in \Gamma^{I}}M[a,y_I] M^{*}[y_I,b]v(b)
			&=\sum\limits_{b} \sum\limits_{y_I\in \Gamma^{I}}\overline{g(y_I,b)} g(y_I,a)v(b)\\
			&=\sum\limits_{y_I\in \Gamma^{I}} g(y_I,a)\sum\limits_{b}\overline{g(y_I,b)}v(b).
		\end{align*}
        Expanding $g$ according to the basis $(B_1\cup\ldots\cup B_{\ell})^{\otimes n}$ we may write
        \[
        g(y_I,b) = \sum\limits_{v'\in B_j}\tilde{g}_{v'}(y_I)v'(b) + \sum\limits_{v''\not\in B_j}\tilde{g}_{v''}(y_I)v''(b),
        \]
        where $\tilde{g}_{v'}$ is the part of $g$ in which the coordinate in $J$ contributes the vector $v'$. Plugging
        this above, we get that
        \[
        \sum\limits_{b}\overline{g(y_I,b)}v(b)
        =m\sum\limits_{v'\in B_j}\overline{\tilde{g}_{v'}(y_I)}\inner{v'}{v} + \sum\limits_{v''\not\in B_j}\overline{\tilde{g}_{v''}(y_I)}\inner{v''}{v}
        =m\sum\limits_{v'\in B_j}\overline{\tilde{g}_{v'}(y_I)}\inner{v'}{v},
        \]
        as $v$ and $v''$ are orthogonal. Plugging this above further, we conclude that
        \[
            (MM^{*} v)_{a}
            =
            m\sum\limits_{y_I\in \Gamma^{I}} g(y_I,a)
            \sum\limits_{v'\in B_j}\overline{\tilde{g}_{v'}(y_I)}\inner{v'}{v}
            =m^{n}
            \sum\limits_{v'\in B_j}\inner{g_{J\rightarrow a}}{\tilde{g}_{v'}}\inner{v'}{v}.
        \]
        Expanding $g$ again, we have
        $g_{J\rightarrow a} = \sum\limits_{\tilde{v}'\in B_j}\tilde{g}_{\tilde{v}'}\tilde{v}'(a) + \sum\limits_{\tilde{v}''\not\in B_j}\tilde{g}_{\tilde{v}''}\tilde{v}''(a)$,
        and so
        \[
        \inner{g_{J\rightarrow a}}{\tilde{g}_{v'}}
        =\sum\limits_{\tilde{v}'\in B_j}\inner{\tilde{g}_{\tilde{v}'}}{\tilde{g}_{v'}}\tilde{v}'(a)
        + \sum\limits_{\tilde{v}''\not\in B_j}\inner{\tilde{g}_{\tilde{v}''}}{\tilde{g}_{v'}}\tilde{v}''(a).
        \]
        Note that for $v'\in B_j$, the function $\tilde{g}_{v'}$ is degree $d-1$ homogenous with respect to $B_j$,
        whereas for $\tilde{v''}\not\in B_j$ the function $\tilde{g}_{\tilde{v}''}$ is degree $d$ homogenous with
        respect to $B_j$, and so $\inner{\tilde{g}_{\tilde{v}''}}{\tilde{g}_{v'}} = 0$. Thus,
        $\inner{g_{J\rightarrow a}}{\tilde{g}_{v'}} = \sum\limits_{\tilde{v}'\in B_j}\inner{\tilde{g}_{\tilde{v}'}}{\tilde{g}_{v'}}\tilde{v}'(a)$, and
        plugging that above (and re-arranging) we get that
        \[
        MM^{*} v  =
        \sum\limits_{\tilde{v}'\in B_j}\left(m^{n}\sum\limits_{v'\in B_j}\inner{\tilde{g}_{\tilde{v}'}}{\tilde{g}_{v'}}\inner{v'}{v}\right) \tilde{v}'
        \in {\sf Span}(B_j)
        \qedhere
        \]
\end{proof}
	\subsection{Proof of Claim~\ref{claim:svd}}
	\begin{proof}
        Recall that marginal distribution of $\mu$ on $y$ is uniform, and we denote $m = \card{\Gamma}$.
        We think of $g$ as a matrix $M$ in $\mathbb{C}^{\Gamma^{J}\times \Gamma^{I}}$, whose $(a,b)$ entry is
		$g(y_J = a, y_I = b)$. The decomposition stated by the claim is an appropriately chosen singular-value
		decomposition of $M$; below are the details for completeness
		
		Looking at $M M^{*} \in\mathbb{C}^{\Gamma^{J}\times \Gamma^{J}}$, we see that it is an $m\times m$ Hermitian matrix,
		hence we may find an eigenbasis $g_1',\ldots,g_m'$ of $\mathbb{R}^{\Gamma^{J}}$ with real non-negative eigenvalues $\lambda_1,\ldots,\lambda_m\geq 0$.
        We will use Claim~\ref{claim:svd_invariant_spaces} to find a more structured eigenbasis.

		By Claim~\ref{claim:svd_invariant_spaces} we conclude that each one of ${\sf span}(\{\chi\circ \gamma\})$ for $\chi\in \hat{H}$
        , as well as ${\sf span}(B_{{\sf non-embed}})$, are invariant spaces of $MM^{*}$. Therefore we may choose an
        an orthonormal eigenbasis $g_1',\ldots,g_m'$ in which each $g_i'$ is from one of these spaces, and we choose this.
        We also let $\kappa_r$ be the eigenvalue corresponding to $g_r'$. To choose $g_1,\ldots, g_m$, define
        $\tilde{g}_r = M^{*} g_r'$; first we note that $\tilde{g}_r$ are orthogonal:
		\[
		\inner{\tilde{g}_r}{\tilde{g}_{r'}} = \inner{M^{*} g_r'}{M^{*} g_{r'}'} = \inner{M M^{*} g_r'}{g_{r'}'} = \kappa_r\inner{g_r'}{g_{r'}'} = \kappa_r 1_{r\neq r'}
		\]
		This means that if we look only at the set $R$ of $r$'s such that $\kappa_r\neq 0$, then we get that $\{\tilde{g}_r\}_{r\in R}$ is
		orthogonal, and we choose $g_r = \tilde{g}_r/\sqrt{\kappa_r}$ (which has $2$-norm equal to $1$).
		
		We prove that
		\[
		M = \sum\limits_{r\in R} \sqrt{\kappa_r} g_r(y_I) \overline{g_r'(y_J)}.
		\]
		Define $M' = \sum\limits_{r\in R} \sqrt{\kappa_r} g_r(y_I) \overline{g_r'(y_J)}$, and note that
        for all $r\in R$ we have $M g_r' = M' g_r' = \sqrt{\kappa_r} g_r\neq 0$.
        As for $r\not\in R$, $g_r'$ is in the kernel of $M$ and also in the kernel of $M'$, and so
        $Mg_r' = 0 = M'g_r'$. This implies that $M = M'$.
		
		Next, we observe that
		\[
		\sum\limits_{r}\sqrt{\kappa_r}^{2}
		=\sum\limits_{r}{\kappa_r}
		=\frac{{\sf Tr}(M M^{*})}{m^n}
        =\frac{1}{m^n}\sum\limits_{b}(M^t M)_{b,b}
		=\norm{g}_2^2
		=1.
		\]
		
        Now we argue about the homogeneity properties of the $g_r$. All of these arguments are basically the same.
        For $\chi\in \hat{H}$, if $r$ is such that $g_r'$ is in ${\sf span}(\chi\circ \gamma)$, then by definition again
        \[
        g_r(y_I) = \frac{m}{\sqrt{\kappa_r}}\Expect{y_J}{\overline{g(y)}g_r'(y_J)},
        \]
		and expanding $g$ we note that on the right hand side only monomials in $g$ in which the variable from $J$ gives the embedding function $\chi\circ \gamma$ give non-zero
        contribution, and we see that $g_r$ is completely embedding homogenous and non-embedding homogenous.

        If $r$ is such that $g_r'$ is in ${\sf span}(B_{{\sf non-embed}})$, then by definition again
        \[
        g_r(y_I) = \frac{m}{\sqrt{\kappa_r}}\Expect{y_J}{\overline{g(y)}g_r'(y_J)},
        \]
		and expanding $g$ we note that on the right hand side only monomials in $g$ on which the variable from $J$ give a non-embedding function can contribute
        (the rest give $0$), hence $g_r$ is completely embedding homogenous and non-embedding homogenous.
	\end{proof}

	\subsection{Proof of Claim~\ref{claim:svd_effective}}
		We run the same argument as in the proof of Claim~\ref{claim:svd}, and we sketch it below.
        Let $P = \Gamma\times \Phi$.
        Recall that marginal distribution of $\mu$ on $y,z$ is uniform, and we denote $m = \card{P}$.
        We think of $F$ as a matrix $M$ in $\mathbb{C}^{P^{J}\times P^{I}}$, whose $(a,b)$ entry is
		$F((y_J, z_J) = a, (y_I, z_I) = b)$.

        Looking at $M M^{*} \in\mathbb{C}^{P^{J}\times P^{J}}$, we see that it is an $m\times m$ Hermitian matrix,
        and by Claim~\ref{claim:svd_invariant_spaces} we conclude that each one of ${\sf span}(\{W\circ \chi\circ \sigma\})$
        for $\chi\in\hat{H}$, as well as ${\sf span}(B_{{\sf non-embed}})$, are invariant spaces of $MM^{*}$. Hence we may choose an
        an orthonormal eigenbasis $F_1',\ldots,F_m'$ in which each $F_i'$ is from one of these spaces, and we choose this.
        All of the items follow exactly in the same way, except that we also argue about the non-embedding degrees.

        Indeed,
        defining $F_t$ analogously to there, we get that
        \[
        F_t(y_I,z_I) = \frac{m}{\sqrt{\psi_t}}\Expect{y_J,z_J}{\overline{F_t(y_I,z_I)}F_t'(y_J,z_J)}.
        \]
        Expanding $F$ we note that if $F_t\in {\sf span}(\chi\circ \sigma)$ for some $\chi\in\hat{H}$,
        then only monomials in $F$ in which the variable from $J$ gives the embedding function $\chi\circ \sigma$ give non-zero
        contribution, so $F_t$ is completely embedding homogenous and non-embedding homogenous of degree $d$.
        Otherwise, if $F_t$ is a non-embedding function, then only monomials in $F$ in which the variable from $J$ gives a non-embedding
        function give non-zero contribution, so $F_t$ is completely embedding homogenous and non-embedding homogenous of degree $d-1$.
        Also, the effective non-embedding degree drops by at most $1$, hence the effective non-embedding degree of $F_t$ is at least
        $d'-1$.
\qed

\section{Missing Proofs: the Direct Product Theorem}\label{sec:missing_dp}
\subsection{Proof of Claim~\ref{claim:binom}}
Expanding the left-hand side,
		\begin{align*}
			\frac{{N-n \choose qN-t}}{{N \choose qN}}
			& = \frac{ (N-n)! (qN)! (N-qN)!}{N! (qN-t)! (N-qN - n +t)!} \\
			& =\frac{qN\cdot(qN-1)\cdots (qN-t+1)\cdot (N-qN)\cdot (N-qN-1)\cdots(N-qN-n+t+1)}{N\cdot (N-1)\cdots (N-n+1)}\\
			&= q^t (1-q)^{n-t} \frac{N (N-1/q)\cdots (N-(t-1)/q) \cdot N(N-1/(1-q))\cdots (N-(n-t-1)/(1-q))}{N(N-1)\cdots (N-(n-1))}\\
			& =  q^t (1-q)^{n-t}  \frac{(1-1/qN)\cdots (1-(t-1)/qN) \cdot (1-1/(1-q)N)\cdots (1-(n-t-1)/(1-q)N)}{(1-1/N)\cdots (1-(n-1)/N)}\\
			& = q^t (1-q)^{n-t} \cdot \frac{e^{-\Theta(t^2/qN)} \cdot e^{-\Theta((n-t)^2/(1-q)N)} }{e^{-\Theta(n^2/N)} }\\
			& = q^t (1-q)^{n-t} (1\pm o(1)),
		\end{align*}
		and the claim follows.

\subsection{Proof of Claim~\ref{claim:bd_sd_dp}}
For simplicity, consider the following distribution  which is a refinement of the distribution ${\mathcal{D}}$. Fix $0\leq p_1, p_2, p_3, p_4\leq 1$ such that $\sum_i p_i \leq 1$. For each $i\in [n]$ independently, $i\in A\setminus A'$ with probability $p_1$ and $i\in A'\setminus A$ with probability $p_2$, $i \in (A\cap A')\setminus A''$ with probability $p_3$,  and  $i \in  A''$ with probability $p_4$. Note that by choosing $p_1 = p_2 = (1-\alpha)q$, $p_4 = q'$ and $p_3= \alpha q - q'$, we recover the given distribution $\mathcal{D}$ and hence we will fix these values of $p_i$ throughout the claim.
			
			Fix a triple $(A, A', A'')$. Let $p$ and $\tilde{p}$ be the probability masses given to $(A, A', A'')$ by the distributions ${\mathcal{D}}$  and $\tilde{\mathcal{D}}$, respectively. Let $a= |A\setminus A'|$, $b = |A'\setminus A|$, $c = |(A\cap A')\setminus A''|$ and $d = |A''|$.  We have
			$$p = p_1^a\cdot p_2^b\cdot p_3^c\cdot p_4^d\cdot (1-(p_1 +p_2 + p_3+p_4))^{n-(a+b+c+d)}$$
			We can compute $\tilde{p}$ as follows:
			\begin{align*}
				\tilde{p} &= \Pr_{A_0, B_0, B_1}\left[ A_0|_{[n]} = A'', B_0|_{[n]} = A\setminus A'', B_1|_{[n]} = A'\setminus A''\right]\\
				& = \Pr_{A_0}\left[ A_0|_{[n]} = A''\right]\cdot  \Pr_{A_0, B_0}\left[  B_0|_{[n]} = A\setminus A'' \mid A_0|_{[n]} = A'' \right]  \Pr_{A_0, B_1}\left[  B_1|_{[n]} = A'\setminus A''\mid  A_0|_{[n]} = A'' \right]\\
				&= \frac{{ N-n\choose q'N-d}}{{ N \choose q'N}}\cdot \frac{{N-q'N - (n-d) \choose (q-q')N - (a+c)}}{{ N-q'N \choose (q-q')N }}\cdot \frac{{N-q'N - (n - d) \choose (q-q')N - (b+c)}}{{ N-q'N \choose (q-q')N }}\\
				& = (1\pm 2^{-\Omega(n)})\cdot q'^d (1-q')^{n-d} \cdot \left( \frac{q-q'}{1-q'}\right)^{a+c} \left( \frac{1-q}{1-q'}\right)^{n - (a+c+d)}\cdot  \left( \frac{q-q'}{1-q'}\right)^{b+c} \left( \frac{1-q}{1-q'}\right)^{n - (b+c+d)},
			\end{align*}
			where the last equality follows from Claim~\ref{claim:binom}. It can be shown that $p = \tilde{p}$ with the following setting of $p_is$
			$$ p_1 = p_2 =  \frac{(q-q')(1-q)}{(1-q')}, p_3 = \frac{(q-q')^2}{(1-q')}, \mbox{ and }p_4 = q'.$$
			Therefore, when $\alpha q = q'+\frac{(q-q')^2}{(1-q')}$, the above identities hold along with  $p_1 = p_2 = (1-\alpha)q$, $p_4 = q'$ and $p_3= \alpha q - q'$. This finishes the proof of this claim.
		\end{proof}
\end{document}